\documentclass[aos]{imsart}

\RequirePackage{amsthm,amsmath,amsfonts,amssymb}
\RequirePackage[authoryear]{natbib}
\RequirePackage[colorlinks,citecolor=blue,urlcolor=blue]{hyperref}
\RequirePackage{graphicx}

\usepackage{algorithm}
\usepackage{algorithmic}
\usepackage{array}
\usepackage{bm}
\usepackage{bbm}
\usepackage{color}
\usepackage{float}
\usepackage{booktabs}
\usepackage{bigints}
\usepackage{extarrows}
\usepackage{multirow}
\usepackage{subcaption}
\usepackage[bbgreekl]{mathbbol}
\usepackage{mathabx}
\usepackage{mathtools}
\usepackage[section]{placeins}
\usepackage{verbatim}
\usepackage{enumitem}
\usepackage{xcolor}
\usepackage{optidef}
\DeclareSymbolFontAlphabet{\mathbbm}{bbold}
\DeclareSymbolFontAlphabet{\mathbb}{AMSb}%

\startlocaldefs
\theoremstyle{plain}
\newtheorem{theorem}{Theorem}[section]
\newtheorem{lemma}{Lemma}[section]
\newtheorem{proposition}{Proposition}[section]
\newtheorem{corollary}{Corollary}[section]
\newtheorem{definition}{Definition}

\newtheorem{remark}{Remark}[section]

\numberwithin{equation}{section}
\numberwithin{theorem}{section}
\numberwithin{lemma}{section}
\numberwithin{proposition}{section}
\numberwithin{corollary}{section}
\numberwithin{definition}{section}
\numberwithin{cons}{section}
\numberwithin{remark}{section}
\numberwithin{exa}{section}
\numberwithin{table}{section}
\numberwithin{figure}{section}

\endlocaldefs

\newcommand{\iz}{z}

\newcommand{\MP}{\operatorname{MP}}

\renewcommand{\i}{\mathrm{i}}

\newcommand{\bZ}{\mathbf{{Z}}}

\newcommand{\bv}{\mathbf{v}}

\newcommand{\bY}{\mathbf{Y}}

\newcommand{\bF}{\mathbf{F}}
\newcommand{\bG}{\mathbf{G}}
\newcommand{\bH}{\mathbf{H}}
\newcommand{\bL}{\mathbf{L}}

\newcommand{\bJ}{\mathbf{J}}
\newcommand{\bW}{\mathbf{W}}

\newcommand{\calD}{\mathcal{D}}

\newcommand{\calF}{\mathcal{F}}
\newcommand{\calG}{\mathcal{G}}
\newcommand{\calH}{\mathcal{H}}
\newcommand{\calI}{\mathcal{I}}

\newcommand{\calK}{\mathcal{K}}

\newcommand{\calN}{\mathcal{N}}

\newcommand{\calW}{\mathcal{W}}

\newcommand{\calZ}{\mathcal{Z}}

\newcommand{\frakB}{\mathfrak{B}}

\newcommand{\frakD}{\mathfrak{D}}

\newcommand{\frakM}{\mathfrak{M}}
\newcommand{\frakN}{\mathfrak{N}}

\newcommand{\mE}{\mathbb{E}}

\newcommand{\mP}{\mathbb{P}}

\newcommand{\tr}{\mathrm{tr}}

\newcommand{\bw}{\mathbf{w}}

\newcommand{\TW}{\mathbb{TW}}
\newcommand{\supp}{\operatorname{supp}}

\newcommand{\tbZ}{\tilde{\bZ}}
\newcommand{\bK}{\mathbf{K}}

\newcommand{\SNR}{{\rm SNR}}

\begin{document}

\begin{frontmatter}
\title{Ridge-Regularized Largest Root Test For High-Dimensional General Linear Hypotheses}
\runtitle{Ridge-regularized largest root test}

\begin{aug}
\author[A]{\fnms{Haoran}~\snm{Li}\ead[label=e1]{hzl0152@auburn.edu}}
\address[A]{Department of Mathematics and Statistics, Auburn
University\printead[presep={,~}]{e1}}
\end{aug}

\begin{abstract}
A fundamental problem in multivariate analysis is testing general linear hypotheses for regression coefficients in a multivariate linear model. This framework includes many important tasks such as MANOVA, joint significance testing, and detecting trends or seasonal effects. Roy’s largest root test is a classical and powerful method for detecting concentrated signals, based on the largest eigenvalue of an $F$-matrix constructed from residual covariance matrices. However, in high-dimensional settings, these matrices often become ill-conditioned or singular, making the test infeasible. To address this, we propose a ridge-regularized Roy’s test that stabilizes the covariance estimation via a ridge term. We establish that the largest eigenvalue of the regularized $F$-matrix follows an asymptotic Tracy–Widom distribution under a high-dimensional regime where the number of variables and hypotheses are comparable to the sample size, assuming only finite-moment conditions. We develop an efficient procedure to estimate the centering and scaling parameters, and analyze the power of the test under a class of low-rank alternatives, with attention to the effect of the regularization parameter. Data-driven strategies for selecting the regularization parameter are developed from both the Bayesian decision-theoretic and minimax perspectives. The proposed method shows strong performance in simulations and is applied to Human Connectome Project data to assess associations between brain measurements and behavioral outcomes.
\end{abstract}



\end{frontmatter}

\section{Introduction}\label{sec:introduction}
In multivariate analysis, one of the fundamental inferential problems is to test a group of hypotheses jointly involving linear transformations of regression coefficients under a multivariate linear model. We consider a setting with \(p\)-variate responses and \(m\)-variate predictors observed from \(n_0\) independent subjects. The relationship between the responses and predictors is modeled as
\[
\bY = B X + \Sigma^{1/2}_p \bZ,
\]
where \(\bY\) is a \(p \times n_0\) response matrix, \(X\) is a \(m \times n_0\) deterministic predictor matrix, \(B\) is a \(p \times m\) coefficient matrix, and \(\bZ\) is a \(p \times n_0\) error matrix with i.i.d. entries of mean zero and variance one. \(\Sigma_p\) is the \(p \times p\) positive definite population covariance matrix, and \(\Sigma^{1/2}_p\) is its  symmetric ``square-root'', satisfying \((\Sigma^{1/2}_p)^2 = \Sigma_p\). The primary objective is to test linear hypotheses about the regression coefficients:
\[
H_0: BC = 0 \quad \text{vs.} \quad H_a: BC \neq 0,
\]
where \(C\) is a \(m \times n_1\) constraint matrix, with each column specifying a linear constraint on the coefficients. Without loss of generality, we assume that \(X\) and \(C\) are of full rank, and $BC$ is estimable in the sense that the matrix $C^T(XX^T)^{-1}C$ is positive definite.  

With various choices of $X$ and $C$, this formulation is broadly applicable across various fields, encompassing many classical problems such as MANOVA, testing the joint significance of predictors, and detecting trends or seasonal patterns. When the sample size $n_0$ is substantially larger than $p$, $m$ and $n_1$, the problem is well-studied. \cite{anderson1958introduction} and \cite{muirhead2009aspects} are among standard references.

Classical inferential procedures rely on the eigen-analysis of an F-matrix of the form 
\[\bF = \bW_1\bW_2^{-1}\] 
where $\bW_1 = n_1^{-1} \bY P_1 \bY^T $ and $\bW_2 = n_2^{-1} \bY P_2 \bY^T$ with 
\begin{equation}\label{eq:W1_W2}
\begin{aligned}
&P_1 = X^T\left(X X^T\right)^{-1} C\left[C^T\left(X X^T\right)^{-1} C\right]^{-1} C^T\left(X X^T\right)^{-1} X, \\
&P_2 = I_{n_0} -X^T\left(X X^T\right)^{-1} X.
\end{aligned}
\end{equation}
Here, $n_1$ and $n_2 = n_0 - m$ are the respective degrees of freedom. The matrix $\bW_2$ is the residual covariance of the full model and serves as an estimator of $\Sigma_p$, while $\bW_1$ is the hypothesis sum-of-squares and cross-products matrix, scaled by $n_1^{-1}$. In a one-way MANOVA set-up, $\bW_1$ and $\bW_2$ reduce to the between-group and within-group sum-of-squares and cross-products matrices, each normalized by its associated degrees of freedom. Both matrices involve the projection of $\bY$ onto two orthogonal subspaces via the projection matrices \(P_1\) and \(P_2\), defined in \eqref{eq:W1_W2}, of rank \(n_1\) and \(n_2\). Under the null hypothesis and assuming Gaussian errors, \(\bW_1\) and \(\bW_2\) are independent Wishart matrices: \(n_1 \bW_1 \sim \operatorname{Wishart}(\Sigma_p, n_1)\) and \(n_2 \bW_2 \sim \operatorname{Wishart}(\Sigma_p, n_2)\). For this reason, $\bF$ is often referred to as a double Wishart matrix.

To test the hypothesis, classical literature commonly employs two main types of techniques, both relying on the eigenvalues of $\bF$. The first approach involves aggregating all eigenvalues of \(\bF\) after applying a specific transformation, as in the classical likelihood ratio test statistic, also known as Wilk's lambda. This method is useful for assessing the overall deviation from the null hypothesis.  

In this paper, we focus on the second approach, known as {Roy's largest root test}, which relies exclusively on the largest eigenvalue of \(\bF\), denoted by \(\ell_{\max}\). This technique is particularly powerful for detecting concentrated alternatives, where the signal is primarily aligned with the leading eigen-direction. Specifically, the test rejects the null hypothesis \(H_0: BC = 0\) if \(\ell_{\max} > \zeta_{\alpha}\) at a significance level \(\alpha\), where \(\zeta_{\alpha}\) is the critical value corresponding to the \((1-\alpha) \times 100\%\) quantile of the distribution of \(\ell_{\max}\) under the null. Traditionally, \(\zeta_{\alpha}\) is approximated using the critical value of an \(F\) or \(\chi^2\) distribution after appropriate scaling; see Chapter 10 of \citet{muirhead2009aspects} for details. 

In contemporary statistical research, it is increasingly common for the dimension \(p\) to be at least comparable to, or even exceed, the sample size $n_0$. While Roy's largest root test performs well when \(n_0\) are much larger than \(p\) and $m$, its reliability diminishes significantly in high-dimensional settings. In particular, when \(p > n_2\), the singularity of \(\bW_2\) renders the F-matrix ill-defined. Even when \(p\) is smaller than \(n_2\), the presence of small eigenvalues in \(\bW_2\) can cause instability in \(\bW_2^{-1}\), thereby compromising the power of the test. Addressing these issues is essential to ensure the robustness and validity of inferential procedures in high-dimensional regimes.

To correct the asymptotic null distribution of the largest root, \citet{johnstone2008multivariate} showed that, under normality, the largest root of $\bF$ converges to the Tracy--Widom distribution of type one \citep{tracy1994level,tracy1996orthogonal}, after appropriate scaling, provided that \(p\), \(n_1\), and \(n_2\) grow proportionally with \(p < n_2\). This result was later extended to non-Gaussian settings by \citet{han2016tracy} and \citet{han2018unified}.

However, when \(p > n_2\), existing remedy methods have primarily focused on problems corresponding to the case when \(n_1\) is small, such as two-sample mean tests and MANOVA. As summarized in \citet{huang2022overview}, these methods can be roughly grouped into three categories. The first category involves quadratic-form tests, which construct modified statistics by replacing \(\bW_2^{-1}\) with regularized substitutes. For example, \citet{bai1996effect} and \citet{chen2010two} proposed replacing \(\bW_2^{-1}\) with the identity matrix in two-sample mean tests. This idea was extended to MANOVA settings by \citet{srivastava2006multivariate}, \citet{yamada2015testing}, and \citet{hu2017testing}. \citet{chen2014two} introduced a ridge regularization in the context of mean testing, further developed in \citet{li2020adaptable} and extended to MANOVA by \citet{li2020high}. The second category includes extreme-type tests, which typically rely on the maximum value among a sequence of test statistics. Notable examples include \citet{tony2014two}, \citet{xu2016adaptive}, and \citet{chang2017simulation}. The third category comprises projection-based tests, which reduce dimensionality by projecting the high-dimensional observations onto a lower-dimensional subspace, followed by a classical likelihood ratio test. Prominent works include \citet{lopes2011more}, \citet{runze2022linear}, and \citet{liu2024projection}. Relatedly, \citet{he2021likelihood} proposed a screening-based approach that first reduces the response dimension and then applies the likelihood ratio test to the reduced data.

\textcolor{black}{On the technical side, there is by now a substantial body of work on the Tracy–Widom law and edge universality for a wide range of random matrix models. Among these developments, results for sample covariance matrices, spiked covariance models, and signal-plus-noise models are particularly relevant to the present paper and form an important part of the technical foundation. Representative references include \cite{Johnstone2001,BaikBenArousPeche2005, el2007tracy,Onatski2008, BaoPanZhou2013LocalEdge,bao2015, PillaiYin2012,PillaiYin2014, bloemendal2014isotropic,lee2016tracy, knowles2017anisotropic,DingYang2018AAP, Yang2019EJP,ding2021spiked,dingYang2021spiked,ding2022edge,ding2023extreme,ding2024eigenvector,ZhangLiuPan2024TWSignalNoise,ding2025tracy}.}

This paper aims to develop a high-dimensional largest root test that remains applicable when $p$, $n_1$ and $n_2$ are comparable. To address the challenge posed by the singularity of \(\bW_2\) when $p>n_2$, we introduce a ridge regularization framework. Specifically, inspired by the work of \citet{li2020adaptable}, we propose a family of ridge-regularized F-matrices:
\begin{equation}\label{eq:reg_F_matrix}
\bF_\lambda = \bW_1 \Big(\bW_2+\lambda I_p\Big)^{-1},  
\end{equation}
where \(\lambda > 0\) is a regularization parameter. The largest eigenvalue of \(\bF_\lambda\), denoted by $\ell_{\max}(\bF_\lambda)$, serves as the test statistic. 

Importantly, the regularized F-matrix and its largest root are \emph{rotation-invariant}, meaning they remain unchanged under arbitrary orthogonal transformations of the data. This invariance arises from the fact that the regularized matrix \(\bW_2 + \lambda I_p\) preserves the eigen-structure of \(\bW_2\). This property is particularly advantageous in settings where additional structural knowledge, such as sparsity, is limited. It ensures that the test’s performance does not depend on the specific coordinate system in which the data are observed. Consequently, the test remains robust across different data representations, making it broadly applicable in various high-dimensional scenarios.

In this paper, we focus on the regime where \( p \), \( n_1 \), and \( n_2 \) are all comparable (see \ref{enum:high_dimensional_regime} in Section \ref{sec:asymptotic_theory}). Under this regime, we establish the asymptotic Tracy-Widom distribution of \(\ell_{\max}(\bF_\lambda)\) under the null hypothesis after appropriate scaling. We propose consistent estimators for the scaling parameters, which rely exclusively on the eigenvalues of \(\bW_2\). We further conduct a power analysis of the proposed test under a class of low-rank alternatives and examine the influence of the regularization parameter $\lambda$. Finally, we propose data-driven strategies for selecting the regularization parameter $\lambda$.

The main contributions of this work are as follows. First, we extend the existing analysis of high dimensional general linear hypotheses to the regime where $n_1$ is comparable to $n_2$ while also allowing $p$ exceed $n_0$. Second, although the regularization scheme is not new, to the best of our knowledge, this is the first study to examine the largest root-type tests under statistical regularization. Third, from an \emph{Random Matrix Theory} (RMT) perspective, this work is the first to establish local laws, universality, edge rigidity, and the asymptotic Tracy-Widom distribution for the extreme eigenvalue of a class of regularized sample covariance matrices and $F$-matrices. Fourth, we propose an estimation procedure for the complicated parameters involved in the asymptotic Tracy-Widom distribution of extreme eigenvalues of a random matrix. These types of parameters frequently appear in the RMT literature, yet their estimation has remained challenging.

The paper is organized as follows. Section \ref{sec:asymptotic_theory} introduces the asymptotic Tracy–Widom distribution for the ridge-regularized largest root test statistic under the null hypothesis, following a brief review of necessary preliminaries from random matrix theory. Section \ref{sec:estimation} presents the estimation procedure for the scaling parameters, which relies solely on the eigenvalues of $\bW_2$, and establishes the consistency of the proposed method. A power analysis under a class of low-rank alternatives is conducted in Section \ref{sec:power_analysis}. A strategy for data-driven selection of $\lambda$ is introduced in Section \ref{sec:selection_lambda}. Section \ref{sec:simulation} reports the results of simulation studies. Section \ref{sec:discussion} concludes with a summary of the main findings and a discussion of future research directions. The Supplementary Material contains an application to \emph{Human Connectome Project} data, additional discussion of the estimation procedure, further numerical results, complete proofs of the main theoretical results, and possible extensions of the current framework both technically and in terms of the form of the linear hypotheses. 

\section{Asymptotic theory under the null hypothesis}\label{sec:asymptotic_theory}
After giving necessary preliminaries on Random Matrix Theory, the asymptotic theory of the proposed ridge-regularized largest root is presented in this section. First of all, note that \(\bF_\lambda\) is asymmetric. Under the null hypothesis $H_0:BC=0$, $\bF_{\lambda}$ shares the same nonzero eigenvalues as a symmetric dual matrix of the form  
$\tilde{\bF}_\lambda = U_1^T \bZ^T \bG_{\lambda}^{-1} \bZ U_1$, { where} 
\[
\bG_\lambda = \bZ U_2 U_2^T \bZ^T + \lambda \Sigma^{-1}_p.
\]
Here, \(U_1\) and \(U_2\) are \(p \times n_1\) and \(p \times n_2\) scaled orthogonal matrices satisfying \(U_1 U_1^T = n_1^{-1} P_1\) and \(U_2 U_2^T = n_2^{-1} P_2\).   Thus, for technical convenience, we focus our analysis on \(\tilde{\bF}_\lambda\).

\begin{definition}[Empirical Spectral Distribution]
 For any $p\times p$ matrix $A$ with real-valued eigenvalues, denote $\ell_j(A)$ to be the $j$-th largest eigenvalue of $A$ and specifically $\ell_{\max}(A) = \ell_1(A)$ and $\ell_{\min}(A) = \ell_p(A)$. Define the \emph{Empirical Spectral Distribution} (ESD) $F^A$ of $A$ by 
\[F^A(\tau) = \frac{1}{p}\sum_{j=1}^p \mathbb{I}(\tau> \ell_j(A)).\]
\end{definition}
The following basic assumptions are employed.
\begin{enumerate}[label ={\bf C\arabic*}]\itemsep0.2em
\item \label{enum:high_dimensional_regime} (High-dimensional regime) \textcolor{black}{We assume $p, n_0, n_1, n_2$ grow to infinity simultaneously in the sense that 
$p/n_0 \to \gamma_0, ~~ p/n_1 \to \gamma_1 ~~ p/n_2 \to \gamma_2$, where $\gamma_0, \gamma_1, \gamma_2 \in (0,\infty)$ are constants.} 
\item \label{enum:moments_conditions} (Moments conditions) $\mE |z_{ij}|^k<\infty$ for any $k\geq 2$.
\item \label{enum:boundedness_spectral_norm} (Boundedness of spectrum) $\lim\inf_{p\to\infty} \ell_{\min}(\Sigma_p) > 0$ and $\lim\sup_{p\to\infty} \ell_{\max}(\Sigma_p) <\infty$.
\item \label{enum:wasserstein_convergence} (Asymptotic stability of population ESD) \textcolor{black}{There exists a deterministic distribution $F^{\Sigma_\infty}$ with compact support in $(0,\infty)$, such that $F^{\Sigma_p}$ converges weakly to $F^{\Sigma_\infty}$.}
\item \label{enum:edge_stability} (Asymptotic stability of edge eigenvalues) Denote the rightmost edge point of the support of $F^{\Sigma_\infty}$ to be $\sigma_{\max}$. That is, $\sigma_{\max} = \inf\{x\in \mathbb{R}: F^{\Sigma_\infty}(x) =1\}$. \textcolor{black}{We assume that $\ell_{\max}(\Sigma_p) \to \sigma_{\max}$.}
\item \label{enum:regular_edge} (Regular edge) We assume the sequence $\{F^{\Sigma_p}, ~p=1,2,\dots\}$ is regular near $\sigma_{\max}$ for all sufficiently large $p$ in the sense of Definition \ref{def:regular_edge} (see Section \ref{sec:asymptotic_TW}). 
\end{enumerate}
\ref{enum:high_dimensional_regime} defines the asymptotic regime where dimensionality $p$, the sample size $n_0$, degrees of freedom $n_1$, and $n_2$ grow proportionately. As stated in \ref{enum:moments_conditions}, our asymptotic results only require moment conditions. \ref{enum:boundedness_spectral_norm} is a standard assumption in the RMT literature, ensuring that eigenvalue bounds are obtainable.  \ref{enum:wasserstein_convergence} ensures that the bulk spectrum of $\Sigma_p$ characterized by its ESD stablizes at a deterministic limit as $p$ increases. While \ref{enum:wasserstein_convergence} offers limited insight into the edge eigenvalue, \ref{enum:edge_stability} ensures that the largest eigenvalue of $\Sigma_p$ also stablizes at the rightmost support edge of $F^{\Sigma_\infty}$. The details of \ref{enum:regular_edge} are given in Section \ref{sec:asymptotic_TW}.

\textcolor{black}{It is worth noting that \ref{enum:edge_stability} rules out the possibility that an $o(p)$ number of eigenvalues of $\Sigma_p$ remain outside the support of $F^{\Sigma_\infty}$. Such a situation arises in classical factor models. For instance, consider $\Sigma_p = A_p + I_p$, where $A_p$ is a sequence of matrices of fixed rank with $\liminf_{p\to\infty} \ell_{\max}(A_p) > 0$. This model, commonly referred to as the \emph{spiked covariance model} \citep{Johnstone2001}, represents a special case of factor models. In this setting, we have $\liminf_{p\to\infty} \ell_{\max}(\Sigma_p) = \liminf_{p\to\infty} \ell_{\max}(A_p) + 1 > 1$. In contrast, the weak limit of $F^{\Sigma_p}$ is the Dirac measure at $1$, so that $\sigma_{\max}=1$. Theoretical analysis in such scenarios requires different techniques and is beyond the scope of the present work. Instead, to assess the performance of the proposed method when \ref{enum:edge_stability} is violated, we include a factor model setting in our numerical study in Section~\ref{sec:simulation}.
} 

\subsection{Preliminary results}\label{sec:preliminary}
In this section, we present key preliminary results in RMT, focusing on the asymptotic Mar\v{c}enko-Pastur (M-P) law for the ESD of $\bG_\lambda$ and the matrix
\[
\widetilde{\bW}_2 = U_2^T \bZ^T \Sigma_p \bZ U_2. 
\]
The matrix \(\widetilde{\bW}_2\) is the companion matrix of \(\bW_2\) in the sense that \(\widetilde{\bW}_2\) shares the same nonzero eigenvalues as \(\bW_2\), differing only by \(|p - n_2|\) additional zeros in their spectra. Consequently, the ESD of \(\widetilde{\bW}_2\) and \(\bW_2\) can be easily reconstructed from one another. 
For mathematical convenience, we focus on the asymptotics of \(F^{\widetilde{\bW}_2}\) rather than \(F^{\bW_2}\) in the subsequent analysis, as it results in cleaner and more tractable expressions. Results in this section are based on \citet{silverstein1995empirical} and \citet{bai2004clt}.

Recall that the \emph{Stieltjes transform} $\varphi(\cdot)$ of any function $F$ of bounded variation on $\mathbb{R}$ is defined as:
\[ \varphi(\iz) = \int_{-\infty}^{\infty} \frac{dF(x)}{x-\iz}, \quad \iz\in\mathbb{C}^+ \coloneqq \{E+ \i \eta : \eta >0\}.\]
Note that $\varphi$ is a mapping from $\mathbb{C}^+$ to $\mathbb{C}^+$. 

Given a compactly supported probability measure $F$ and any $\gamma>0$, the renowned {M-P equation} in RMT defined on $z\in\mathbb{C}^+$ is  expressed as
\begin{equation}\label{eq:MP_W2}
z = -\frac{1}{\varphi} + \gamma \int\frac{\tau d F(\tau) }{ \tau \varphi +1}.  
\end{equation}
There exists a unique solution \(\varphi \equiv \varphi(z; \gamma, F)  \in \mathbb{C}^+\) to Eq. \eqref{eq:MP_W2}. The function \(\varphi(z) =\varphi(z;\gamma, F) \) is the Stieltjes transform of a probability measure with bounded support in $[0, \infty)$. 
\noindent Unfortunately, except in extreme cases where \(F\) has a simple structure, there is no explicit closed-form solution for $\varphi$ or the associated probability distribution. For further details of {M-P equation}, see Chapter 3 of \citet{bai2010spectral}. 

In the following, we shall write $\varphi(\cdot) =  \varphi(~\cdot~; \gamma_2, F^{\Sigma_\infty})$ and $\varphi_p(\cdot) = \varphi(~\cdot~; p/n_2, F^{\Sigma_p})$.
\begin{lemma}[Marčenko-Pastur Theorem]\label{lemma:M_P_theorem}
Assume that conditions \ref{enum:high_dimensional_regime}--\ref{enum:wasserstein_convergence} hold. As \(p \to \infty \), the ESD \( F^{\widetilde{\bW}_2} \) converges pointwise almost surely to \( \calW \) at all points of continuity of \( \calW \), where $\calW$ is the probability measure associated with $\varphi$.  

Moreover, for any $\iz\in\mathbb{C}^+$, let $\hat{\varphi}(\iz)$ be the Stieltjes transform of $F^{\widetilde{\bW}_2}$. That is,
\[ \hat{\varphi}(\iz) = \int \frac{dF^{\widetilde{W}_2}(\tau) }{\tau - \iz} = \frac{1}{n_2} \sum_{j=1}^{n_2} \frac{1}{\ell_j(\widetilde{\bW}_2) - \iz}.\]
For any closed and bounded region $\calZ \subset \mathbb{C}^+$ and any fixed $k\in\mathbb{N}$, we have 
\[ \sup_{\iz \in \calZ} p^{2/3}|\hat\varphi^{(k)}(\iz)- \varphi_p^{(k)}(\iz) | \stackrel{P}{\longrightarrow}0.\]  
\end{lemma}

Following immediately from Lemma \ref{lemma:M_P_theorem}, the M-P equation \eqref{eq:MP_W2} holds approximately at \(\hat{\varphi}(\iz)\) in the sense that 
\begin{equation}\label{eq:MP_W2_approximate}
 z  + \frac{1}{\hat{\varphi}(\iz)}  - \frac{p}{n_2} \int \frac{\tau \, dF^{\Sigma_p}(\tau)}{\tau \hat{\varphi}(\iz) + 1} \stackrel{P}{\longrightarrow} 0,
\end{equation}
uniformly for \(\iz \in \calZ\). This result serves as the foundation for the estimation method proposed in Section \ref{sec:estimation}.

Next, we analyze the limiting ESD of $\bG_\lambda$. For any fixed $\lambda>0$, $\gamma>0$, and a compactly supported probability measure $F$ on $[0,\infty)$ that is nondegenerate to zero, consider the following equation on $z\in \mathbb{C}^+$ 
\begin{equation}\label{eq:MP_G}
\phi  =  \int \frac{\tau \, dF(\tau)}{ \tau \{ [1+ \gamma \phi ]^{-1} - \iz\} + \lambda }. 
\end{equation}
Then, there exists a unique solution $\phi \equiv \phi(z; \lambda, \gamma, F)\in \mathbb{C}^+$ to Eq. \eqref{eq:MP_G}. The function $\phi(z) = \phi(z; \lambda, \gamma, F)$ is the Stieltjes transform of a probability measure with bounded support in $(0, \infty)$. We shall write $\phi_\lambda (\cdot) = \phi(~\cdot~;\lambda, \gamma_2, F^{\Sigma_\infty})$ and $\phi_{p\lambda}(\cdot) = \phi(~\cdot~;\lambda, p/n_2 , F^{\Sigma_p})$.

\begin{theorem}[Generalized Mar\v{c}enko-Pastur Theorem of $\bG_\lambda$]\label{thm:generalized_MP_G}
Suppose \ref{enum:high_dimensional_regime}--\ref{enum:wasserstein_convergence} hold. Fix $\lambda > 0$. As $p\to\infty$, the ESD $F^{\bG_\lambda}$ of $\bG_\lambda$ converges pointwise almost surely to $\calG_\lambda$ at all points of continuity of $\calG_\lambda$, where $\calG_\lambda$ is the probability measure associated with $\phi_{\lambda}$.
\end{theorem}
While Eq. \eqref{eq:MP_G} plays a fundamental role in the subsequent analysis, we present an equivalent reformulation that is more convenient to work with in certain cases.  Define  
\[
h = h(z; \lambda, \gamma, F) = z - [1 + \gamma \phi(z; \lambda, \gamma, F)]^{-1}, \quad z\in\mathbb{C}^+.
\]
Then, Eq. \eqref{eq:MP_G} can be rewritten as  
\begin{equation}\label{eq:MP_G_rewrite}
z  = h + \left(1 + \gamma \int \frac{\tau \, dF(\tau)}{\lambda -\tau h}\right)^{-1}.
\end{equation}
It follows that for any \( \iz \in \mathbb{C}^+ \), there exists a unique solution \( h \in \mathbb{C}^+ \) to Eq. \eqref{eq:MP_G_rewrite}. We shall write $h_\lambda(z) = h(z;\lambda, \gamma_2, F^{\Sigma_\infty})$ and $h_{p\lambda}(z) = h(z;\lambda, p/n_2, F^{\Sigma_p})$ in the following.

Lastly, we consider the extension of \( \phi_{p\lambda}(\iz) \) to the real line.  
\begin{lemma}\label{lemma:extension_phi_to_s} 
Suppose that Conditions \ref{enum:boundedness_spectral_norm}--\ref{enum:edge_stability} hold. Consider any fixed $p$ and $\lambda$. Denote the probability measure associated with $\phi_{p\lambda}$ to be $\calG_{p\lambda}$. Then, 
\( \calG_{p\lambda} \) is compactly supported on \( (0, \infty) \). Denote the leftmost edge of the support of \( \calG_{p\lambda} \) by \(\rho_{p\lambda} \), i.e., $\rho_{p\lambda} = \sup \{x \in \mathbb{R} : \calG_{p\lambda}(x) = 0\}$. The following results hold:  
\begin{itemize}  
\item[(i)] We have \( \rho_{p\lambda} > 0 \), and for any \( x \in [0, \rho_{p\lambda}) \),  
\[
\lim_{\iz \in \mathbb{C}^+ \to x} \phi_{p\lambda}(\iz) \coloneqq s_{p\lambda}(x) \quad \text{exists and is real-valued}.
\]
Moreover, the following relationship holds:  
\begin{equation}\label{eq:s_as_ST_G}
s_{p\lambda}(x) = \int \frac{d\calG_{p\lambda}(\tau)}{\tau - x}, \quad x \in [0, \rho_{p\lambda}).
\end{equation}
\item[(ii)] From Eq~\eqref{eq:s_as_ST_G}, \( s_{p\lambda}(x) \) is analytic on \( [0, \rho_{p\lambda})\), and its \( k \)th derivative satisfies \( s^{(k)}_{p\lambda}(x) > 0 \) for any \( k \in \mathbb{N} \) and \( x \in [0, \rho_{p\lambda}) \). This implies that \( s^{(k)}_{p\lambda}(x) \) is monotonically increasing on \( [0, \rho_{p\lambda}) \) for all \( k \).
\end{itemize}  
\end{lemma}

\subsection{Asymptotic Tracy-Widom distribution}
\label{sec:asymptotic_TW}

In this section, we present the asymptotic Tracy-Widom distribution for the proposed ridge-regularized largest root statistic 
after appropriate scaling. Throughout the analysis, $\lambda$ is treated as fixed. 


We first impose an edge regularity condition on \(F^{\Sigma_p}\). Denote $\ell_{\max}(\Sigma_p)$ to be $\sigma_{1p}$ for convenience. Recall that under Condition \ref{enum:edge_stability}, \(\sigma_{1p} \to \sigma_{\max}\), which is the rightmost edge point of the support of \(F^{\Sigma_\infty}\). For \(h \in (-\infty, \lambda / \sigma_{1p} )\), define
\begin{equation}\label{eq:x_h}
x_p(h) = h + \left[1 + \frac{p}{n_2} \int \frac{\tau \, dF^{\Sigma_p}(\tau)}{\lambda - \tau h}\right]^{-1}.
\end{equation}
It can be viewed as an extension of Eq. \eqref{eq:MP_G_rewrite} as $\Im(h)\downarrow0$ evaluated at $F=F^{\Sigma_p}$ and $\gamma = p/n_2$.
\textcolor{black}{
\begin{definition}\label{def:regular_edge}  
We say that \( \{F^{\Sigma_p}\}\) is \emph{regular near} \( \sigma_{\max} \) if any of the following holds:  
\begin{itemize}  
\item[(a)] There exists a constant $\epsilon>0$ such that the equation $x'_p(h) = 0$ admits a root in $h \leq \lambda/\sigma_{\max} -\epsilon$ for all sufficiently large $p$;
\item[(b)] There exists a constant $\epsilon>0$ such that $(p/n_2) F^{\Sigma_p}(\{\sigma_{1p}\}) >1 + \epsilon$ for all sufficiently large $p$, where $F^{\Sigma_p}(\{\sigma_{1p}\})$ denotes the mass of the distribution on $\sigma_{1p}$.
\end{itemize}  
\end{definition}
\begin{lemma}\label{lemma:when_discrete_but_limiting_continuous}
A sufficient condition for part (a) of Definition \ref{def:regular_edge} to hold is that there exists a constant $\epsilon>0$ such that $\epsilon< (p/n_2) F^{\Sigma_p}(\{\sigma_{1p}\})< 1 - \epsilon$, for all sufficiently large $p$.
\end{lemma}
}

\textcolor{black}{Part (a) of Definition \ref{def:regular_edge} is a standard regularity condition in the RMT literature for establishing Tracy--Widom fluctuations of extreme eigenvalues, although it appears in various equivalent forms. See, for example, Definition~2.7 of \citet{knowles2017anisotropic}, Condition~(1) of \citet{el2007tracy}, Condition~(2.12) of \citet{lee2016tracy}, and Assumption~2.5 of \citet{DingYang2018AAP}. Under this condition, the limiting spectral distribution of $\bG_\lambda$ exhibits a \emph{square-root} behavior at its leftmost support edge. Precisely,  recall the distribution $\calG_{p\lambda}$ and its leftmost support edge $\rho_{p\lambda}$ as in Lemma \ref{lemma:extension_phi_to_s}. For sufficiently large $p$, we have (i) $\rho_{p\lambda} > \lambda/\sigma_{1p} + c$ for some small constant $c>0$; (ii) $\calG_{p\lambda}$ admits a density $f_{\calG}$ in a neighborhood of $\rho_{p\lambda}$; and (iii) there exist positive constants $C_1$, $C_2$, and small $\epsilon>0$ such that 
\[
C_1 \sqrt{x - \rho_{p\lambda}} \leq f_{\calG}(x) \leq C_2 \sqrt{x - \rho_{p\lambda}}, 
\qquad \text{for all } x \in (\rho_{p\lambda}, ~\rho_{p\lambda} + \epsilon).
\]
The results are shown in Theorem 5.1 of \cite{li2024analysis}. 
}

\textcolor{black}{
On the other hand, Part (b) of Definition \ref{def:regular_edge} leads to a different behavior of $\calG_{p}$ near $\rho_{p\lambda}$. In this case, we have $\rho_{p\lambda} = \lambda / \sigma_{1p}$, and $\rho_{p\lambda}$ is an isolated point mass of $\calG_\lambda$. More details are given in Section S.5.3 
of the Supplementary Material.
}

The following lemma follows from the behavior of \( \calG_{p} \) near \( \rho_p \). 
\begin{lemma}\label{lemma:s_prime_infinity}  
Suppose the conditions in Lemma \ref{lemma:extension_phi_to_s} hold, and further assume that \( \{F^{\Sigma_p}\} \) is regular near \( \sigma_{\max} \). Then, for any fixed $\lambda>0$ and sufficiently large $p$, $s'_{p\lambda}(x) \to \infty$ as $x \uparrow \rho_p$.   
\end{lemma}


Our main theorem is presented as follows.

\begin{theorem}\label{thm:main}  
Suppose that \ref{enum:high_dimensional_regime}--\ref{enum:regular_edge} hold.  
For any fixed \( \lambda > 0 \), as \( p \to \infty \), the asymptotic distribution of the largest eigenvalue of \( \bF_\lambda \) (or equivalently \( \tilde{\bF}_\lambda \)) is given by  
\begin{equation}
\frac{p^{2/3}}{\Theta_{2p}(\lambda)} \Big(\ell_{\max}(\bF_\lambda) - \Theta_{1p}(\lambda)\Big) \Longrightarrow \TW_1,  
\label{eq:Tracy_Widom}
\end{equation}
where \( \TW_1 \) denotes the Tracy-Widom distribution of type 1, and \( {\Longrightarrow} \) indicates convergence in distribution.  
The centering and scaling parameters \( \Theta_{1p}(\lambda)\) and \( \Theta_{2p}(\lambda) \) are determined as follows. Define \( \beta = \beta_{p\lambda}\) to be the solution in \( (0, \rho_{p\lambda}) \) to  
\begin{equation}\label{eq:def_beta}
\beta^2 {s}'_{p\lambda}(\beta) = \frac{n_1}{p}.
\end{equation}  
By Lemma~\ref{lemma:s_prime_infinity}, such a solution \( \beta \) exists and is unique. Then,  
\begin{align}
&\Theta_{1p}(\lambda) = \frac{1}{\beta} \Big[1 + \frac{p}{n_1} \beta {s}_{p\lambda}(\beta)\Big], \label{eq:def_Theta1} \\
&\Theta_{2p}(\lambda) = \Bigg[\frac{(p/n_1)^3}{2} {s}''_{p\lambda}(\beta) + \frac{(p/n_1)^2}{\beta^3} \Bigg]^{1/3}. \label{eq:def_Theta2}
\end{align}  
\end{theorem}

\begin{remark}\label{remark:Tracy-Widom}
For detailed information on Tracy-Widom distributions, we refer readers to \citet{tracy1994level}, \citet{tracy1996orthogonal}, and \citet{johnstone2008multivariate}. Available software for working with Tracy-Widom laws includes the \texttt{R} package \texttt{RMTstat} \citep{johnstone2022package} and the \texttt{MATLAB} package \texttt{RMLab} \citep{dieng2006matlab}.
\end{remark}

\begin{remark}
Unlike the \( O(p) \) scaling observed in classical extreme value theory for i.i.d. random variables with appropriate tail behavior, the normalized largest eigenvalue of \( \bF_\lambda \) fluctuates on the \( O(p^{2/3}) \) scale. An analogous scaling has been established for the largest eigenvalue of various random matrices, such as those in the \emph{Jacobi orthogonal ensemble} \citep{johnstone2008multivariate}. This \( O(p^{2/3}) \) fluctuation reflects the repulsion between eigenvalues, causing them to be more regularly spaced than i.i.d. random variables.
\end{remark}

\begin{remark}\label{remark:link_to_literature}
\textcolor{black}{When $p < n_1+n_2$, \citet{han2016tracy} propose a test based on the largest root of the equation $\det(\ell \bW_2 - \bW_1)=0$, denoted by $\ell_H$. This statistic can be interpreted as the limiting case of $\ell_{\max}(\bF_\lambda)$ as $\lambda \to 0$, since $\ell_{\max}(\bF_\lambda)$ is the largest root of $\det(\ell (\bW_2+\lambda I_p) - \bW_1)=0$. Consequently, Theorem~\ref{thm:main} may be regarded as an extension of Theorem~2.1 in \citet{han2018unified} to the regularized setting $\lambda>0$. When $p<n_1$, $\ell_H = \ell_{\max}(\bF_0)$. When $p>n_2$, there is no simple functional relationship. Notably, under the assumed model, $\ell_H$ has a limiting distribution that is independent of $\Sigma_p$, whereas the distribution of $\ell_{\max}(\bF_\lambda)$ depends on $\Sigma_p$. In Sections~\ref{sec:power_analysis} and \ref{sec:simulation}, we compare the power of the proposed test with that of $\ell_H$ both analytically and numerically under a class of low-rank alternatives.}
\end{remark}

\begin{remark}\label{remark:generalization_k_largest_eigenvalues}  
Theorem \ref{thm:main} can be generalized to the joint asymptotic distribution of the first \( k \) largest eigenvalues of \( \bF_\lambda \). Specifically, under \ref{enum:high_dimensional_regime}--\ref{enum:regular_edge}, for any fixed integer \( k \) and $\lambda>0$, we have  
\begin{equation}\label{eq:generalization_k_largest_eigenvalues}
\begin{aligned}
\lim_{p \to \infty} &\mP\left( \frac{p^{2/3}}{\Theta_{2p}(\lambda)} \Big(\ell_1(\bF_\lambda) - \Theta_{1p}(\lambda)\Big) \leq t_1, \dots, \frac{p^{2/3}}{\Theta_{2p}(\lambda)} \Big(\ell_k(\bF_\lambda) - \Theta_{1p}(\lambda)\Big) \leq t_k \right) \\ 
& = \lim_{p \to \infty} \mP\left( {p^{2/3}} \Big(\ell_1^{\rm GOE} - 2\Big) \leq t_1, \dots, \Big(p^{-1/2}\ell_k^{\rm GOE} - 2\Big) \leq t_k \right),  
\end{aligned}
\end{equation}
for any \( t_1, \dots, t_k \in \mathbb{R}\). Here, \( \ell_k^{\rm GOE} \) denotes the \( k \)-th largest eigenvalue of a \( p \times p \) matrix from the Gaussian Orthogonal Ensemble (GOE), scaled by $p^{-1/2}$.
\end{remark}

\textcolor{black}{
While the formal proof of Theorem \ref{thm:main} is presented in the Supplementary Material, we briefly outline the main ideas and highlight the main technical contributions. We begin by working under Gaussianity, which has the additional advantage that $\bW_1$ and $\bW_2$ are independent. This independence enables us to treat $\bG^{-1}_{\lambda}$ as an effective population covariance matrix, so that $\tilde{\bF}_\lambda$ can be viewed as a companion sample covariance matrix (arising from the structure of $\widetilde{\bW}_2$) with population covariance $\bG^{-1}_{\lambda}$. The proof then proceeds in three stages. We first analyze the asymptotic behavior of $\bG_{\lambda}$. Secondly, conditional on $\bG_{\lambda}$, we show the asymptotic behavior of $\tilde{\bF}_{\lambda}$. Thirdly, we extend the results to non-Gaussianity.} 


\textcolor{black}{Note that $\bG_\lambda$ has the form of a sample covariance matrix $\bZ U_2U_2^T \bZ^T$ plus a full-rank shift matrix $\lambda \Sigma^{-1}_p$, a structure considered in the pioneering works of \cite{marchenko1967distribution} and \cite{silverstein1995analysis}, where convergence of the ESD of $\bG_\lambda$ is established (see Theorem~\ref{thm:generalized_MP_G}). However, subsequent developments on local laws, universality, edge rigidity, and $O_P(1/p)$-fluctuation of linear spectral statistics (LSS) \citep{bai1998no,bai2004clt,bloemendal2014isotropic,bao2015,knowles2017anisotropic} have primarily focused on the case of a zero shift matrix. Extending these results to the presence of a nonzero shift matrix is nontrivial and plays a central role in the technical arguments. The present work contributes in this direction by establishing these properties for general $\Sigma_p$.}

\textcolor{black}{
\textbf{Step 1:} We establish the aforementioned properties of $\bG_\lambda$ by closely following the strategy of \cite{knowles2017anisotropic}, with suitable modifications to accommodate the presence of the deterministic shift term $\lambda \Sigma_p^{-1}$, which alters the associated M-P equation from \eqref{eq:MP_W2} to \eqref{eq:MP_G}. The key idea is to construct a suitable linearization block matrix (see $\bK$ in Eq.~S.6.1 
in the Supplementary Material) and prove convergence of the entries of $\bK^{-1}$. This step is carried out in Section~S.6 
of the Supplementary Material.}

\textcolor{black}{
\textbf{Step 2:} Under Gaussianity, treating $\bG_{\lambda}^{-1}$ as an effective population covariance matrix, we verify that the regularity conditions of \citet{knowles2017anisotropic} and \citet{lee2016tracy} are satisfied with high probability. This allows us to establish the Tracy--Widom limit for $\ell_{\max}(\tilde{\bF}_\lambda)$ under Gaussianity by applying the main theorem of \cite{lee2016tracy}. This step is carried out in Section S.7.1 
of the Supplementary Material.
}

\textcolor{black}{
\textbf{Step 3:} To establish universality of the largest root and extend the results to the general setting of Condition~\ref{enum:moments_conditions}, we outline a proof strategy based on a Green function comparison theorem. In particular, we show that universality can be reduced to proving local laws for the inverse of an appropriate linearization block matrix $\bH_\lambda$ (see Eq. S.7.4 
in the Supplementary Material). This step is carried out in Section S.7.2 
of the Supplementary Material. 
}

\textcolor{black}{
\textbf{Step 4:} We show the local laws for $\bH^{-1}_\lambda$ and completing the extension to non-Gaussian distributions. Indeed, $\bH_\lambda$ differs from the linearization matrix studied in \cite{han2016tracy} and \cite{han2018unified} only by a deterministic block. Owing to this structural similarity, we build upon the arguments in \cite{han2016tracy} and \cite{knowles2017anisotropic}. To avoid unnecessary repetition, we only explicitly highlight the similarities, distinctions, and required modifications in our setting. It is carried out in Section~S.7.3 
of the Supplementary Material. 
}

\section{Estimation}\label{sec:estimation}

In this section, we develop an estimation procedure for the centering and scaling parameters \(\Theta_{1p}(\lambda)\) and \(\Theta_{2p}(\lambda)\). The proposed estimators leverage the eigenvalues of \(\bW_2\) and are designed to be computationally efficient while maintaining robustness in high-dimensional settings. This section is organized as follows. Section \ref{subsec:reformulation} introduces a reformulation of the problem as an optimization task. The proposed algorithm is described in detail in Section \ref{subsec:algorithm}. In Section \ref{subsec:consistency}, we establish the consistency of the proposed estimators. Additional discussion, implementation details, and practical considerations are provided in Section S.2 
of the Supplementary Material.

Throughout this section, $\lambda$ is treated as fixed. For notational simplicity, we suppress the explicit dependence of certain quantities on $\lambda$ and $p$ whenever no confusion may arise. We denote \(\hat{\gamma}_1 = p / n_1\) and \(\hat{\gamma}_2 = p / n_2\).


\subsection{Problem formulation} \label{subsec:reformulation}

The estimation of \(\Theta_1\) and \(\Theta_2\) necessitates the precise estimation of the edge \(\rho\) and the function \(s(x)\) along with its derivatives on the interval \([0, \rho)\). Although \(s(x)\) is formally defined as the extension of the Stieltjes transform \(\phi(z)\) from \(\mathbb{C}^+\) to the real line as in Lemma \ref{lemma:extension_phi_to_s}, this definition provides limited practical utility due to the absence of an explicit analytical expression. In this section, we rigorously demonstrate that both $\rho$ and \(s(x)\) can be characterized by Eq. \eqref{eq:x_h}, thereby establishing a foundation for their accurate estimation.

\begin{lemma}\label{lemma:determine_rho}
Assume that \(\{F^{\Sigma_p}\}\) is regular near \(\sigma_{\max}\) as defined in Definition \ref{def:regular_edge}.  If $\hat{\gamma}_2 F^{\Sigma_p}(\{\sigma_{1p}\}) \geq 1$, then $\rho = \lambda/\sigma_{1p}$. Otherwise, \(\rho = x(h_0) > \lambda / \sigma_{1p}\), where $h_0$ is the unique solution in $(-\infty, \lambda / \sigma_{1p})$ to \(x'(h_0) = 0\).  
\end{lemma}


Define the auxiliary functions on \( h \in (-\infty, \lambda / \sigma_{1p}) \) as
\[
\calH_j(h) = \int \frac{\tau^j \, dF^{\Sigma_p}(\tau)}{( \lambda - \tau h)^j}, \quad j = 1, 2, 3.
\]
With the definition, $x(h) = h + [1+ \hat{\gamma}_2 \calH_1(h)]^{-1}$ and $x'(h) = 1- \hat{\gamma}_2 [1+\hat{\gamma}_2 \calH_1(h)]^{-2} \calH_2(h)$. 
\begin{theorem}\label{thm:determine_s}
Suppose that Conditions \ref{enum:boundedness_spectral_norm}--\ref{enum:edge_stability} hold. 
Fix any $x_0 \in [0, \rho)$. Then, $s(x_0)$ is determined as  
\[ s(x_0) = \frac{1}{\hat{\gamma}_2} \left( \frac{1}{x_0 - h_0} -1 \right),\]
where $h_0$ is the unique solution in $(-\infty, \lambda/\sigma_{1p})$ to 
\begin{equation}\label{eq:determine_s}
\begin{cases}
&x_0  =  h_0 + [1+ \hat{\gamma}_2 \calH_1(h_0)]^{-1},\\   
&1 - \hat{\gamma}_2 [1+\hat{\gamma}_2 \calH_1(h_0)]^{-2} \calH_2(h_0) >0. 
{}\end{cases}
\end{equation}
Furthermore, the derivatives $s'(x_0)$ and $s''(x_0)$ are determined as:  
{\small
\begin{equation}\label{eq:characterization_s2s3}
\begin{aligned}
&g(x_0) = x_0 - \frac{1}{1 + \hat{\gamma}_2 {s}(x_0)},\\
&\zeta(x_0) = \frac{\hat{\gamma}_2}{( 1+ \hat{\gamma}_2 s(x_0))^{2}},\\
&{s'}(x_0) =  \calH_2(g(x_0))\left(\zeta(x_0) s'(x_0) + 1 \right),\\
&{s''}(x_0) = 2 \left(\zeta(x_0) s'(x_0) + 1 \right)^2 \calH_3(g(x_0))+ \calH_2(g(x_0)) \Big[ \frac{- 2 (\hat{\gamma}_2 s'(x_0))^2 }{(1+\hat{\gamma}_2s(x_0)))^3} + \zeta(x_0)s''(x_0) \Big].
\end{aligned}
\end{equation}
}
\end{theorem}

\noindent It is worth noting that the third equation in \eqref{eq:characterization_s2s3} is linear in \(s'(x_0)\), given \(s(x_0)\) and \(\calH_2(\cdot)\). Similarly, the last equation in \eqref{eq:characterization_s2s3} is linear in \(s''(x_0)\), given \(s(x_0)\), \(s'(x_0)\), \(\calH_2(\cdot)\), and \(\calH_3(\cdot)\). 

The following lemma is useful for the numerical computation of $s(x_0)$.
\begin{lemma}\label{lemma:monotonicity_x_h}
Under the conditions of Theorem \ref{thm:determine_s}, the derivative \(x'(h) = 1 - \hat{\gamma}_2 [1 + \hat{\gamma}_2 \calH_1(h)]^{-2} \calH_2(h)\) is monotonically decreasing on \((-\infty, \lambda / \sigma_{1p})\).
\end{lemma}
\noindent This result implies that \(x(h)\) is either strictly increasing on \((-\infty, \lambda / \sigma_{1p})\) or exhibits a single concave-down peak within this interval. If \(x_0 = h + [1 + \hat{\gamma}_2 \calH_1(h)]^{-1}\) has two distinct roots, Eq. \eqref{eq:determine_s} indicates that \(h_0\) corresponds to the smaller root.

Through the results in Theorem \ref{thm:determine_s}, we can transform the problem of estimating \(s(x)\), \(s'(x)\), and \(s''(x)\) into the task of estimating \(\calH_j(h)\) for \(j = 1, 2, 3\) over \(h \in (-\infty, \lambda / \sigma_{1p})\). To this end, we revisit Lemma \ref{lemma:M_P_theorem} and Eq. \eqref{eq:MP_W2_approximate}, which together suggest that:
\begin{align*}
&\calH_1(-\lambda \hat{\varphi}(\iz)) = \int \frac{\tau \, dF^{\Sigma_p}(\tau)}{\lambda + \tau \lambda \hat{\varphi}(\iz)} \approx \frac{\iz}{\lambda {\hat{\gamma}}_2} + \frac{1}{\lambda {\hat{\gamma}}_2 \hat{\varphi}(\iz)},\\
&\calH_2(-\lambda \hat{\varphi}(\iz)) = \int \frac{\tau^2 \, dF^{\Sigma_p}(\tau)}{(\lambda + \tau \lambda \hat{\varphi}(\iz))^2} \approx \frac{1}{\lambda^2 {\hat{\gamma}}_2 \hat{\varphi}^2(z)} - \frac{1}{\lambda^2{\hat{\gamma}}_2 \hat{\varphi}'(z)}. 
\end{align*}
Here, we are extending the definition of $\calH_j(h)$ to the complex domain. The second result is obtained by differentiating both sides of Eq. \eqref{eq:MP_W2_approximate}. Notably, $\hat{\varphi}(\iz)$ solely relies on the eigenvalues of $\bW_2$ as 
 \[\hat{\varphi}(\iz) = \frac{1}{n_2} \sum_{j=1}^{n_2} \frac{1}{\ell_j(\bW_2) - \iz}, \quad  \text{ with } \ell_j(\bW_2) =0 \text{ if } j>p.\]  This observation motivates the following strategy for estimating \(\calH_j(y)\) for \(j = 1, 2, 3\).

Note that \(F^{\Sigma_p}\) can be approximated by a weighted sum of point masses:
\[
F^{\Sigma_p}(\tau) \simeq \sum_{k=1}^K w_k \delta_{\sigma_k}(\tau),
\]
where \(\{\sigma_k\}_{k=1}^K\) is a grid of points chosen to densely cover the range from $\liminf\nolimits_{p} \ell_{\min}(\Sigma_p)$ to $\limsup\nolimits_p \ell_{\max}(\Sigma_p)$,
and \(w_k\) are weights satisfying \(w_k \geq 0\) and \(\sum_{k=1}^K w_k = 1\). The grid \(\{\sigma_k\}_{k=1}^K\) is user-selected and fixed throughout the estimation procedure, while the weights \(w_k\) are model parameters. This approach was introduced in \citet{el2008spectrum} and later generalized in \citet{ledoit2012nonlinear}. 
Accordingly, the integrals can be approximated as:
\[
\calH_j(h) = \int \frac{\tau^j \, dF^{\Sigma_p}(\tau)}{(\lambda - \tau h)^j} \simeq \sum_{k=1}^K \frac{w_k \sigma_k^j}{(\lambda - \sigma_k h)^j} \coloneqq \tilde{\calH}_j(h, w_1, \dots, w_K).
\]
For \(\iz \in \mathbb{C}^+\), define:
\[
\hat{Q}_1(\iz) = \frac{\iz}{\lambda \hat{\gamma}_2} + \frac{1}{\lambda \hat{\gamma}_2 \hat{\varphi}(\iz)}, \quad \mbox{and}\quad \hat{Q}_2(\iz) =  \frac{1}{\lambda^2 \hat{\gamma}_2 \hat{\varphi}^2(z)} - \frac{1}{\lambda^2\hat{\gamma}_2 \hat{\varphi}'(z)}. 
\]
Consider a grid of points \(\iz\) in $\mathbb{C}^+$, denoted by \(\{\iz_i\}_{i=1}^I\). We select the weights \(w_k\) such that \(\hat{Q}_j(\iz_i)\) approximately matches \(\tilde{\calH}_j(-\lambda \hat{\varphi}(\iz_i), w_1, \dots, w_K)\) for all \(i = 1, 2, \dots, I\), and $j=1,2$.
 Further details are presented in the following sections. The recommended choice of the grids \(\{\sigma_k\}_{k=1}^K\) and \(\{\iz_i\}_{i=1}^I\) is provided in Section S.2 
 of the Supplementary Material. 
\subsection{The algorithm}\label{subsec:algorithm}
Given the grids $\{\sigma_k\}_{k=1}^K$ and $\{\iz_i\}_{i=1}^I$,  define the loss function as 
\[ L_\infty(w_1,\dots, w_K) = \max_{i=1,\dots,I} \max_{j=1,2} \{ |\Re(e_{ij} )|, |\Im(e_{ij})|\},~~\mbox{where} \]
\[e_{ij} = \frac{\hat{Q}_j(\iz_i) - \tilde{\calH}_j(-\lambda\hat{\varphi}(\iz_i), w_1,\dots, w_K)}{|\hat{Q}_j(\iz_i)|}, \quad i = 1,\dots, I; j=1,2.\] 
\begin{algorithm}
\caption{Linear Programming Formulation for Selecting Weights}
\label{algo:linear_program}
\begin{algorithmic}
\STATE \textbf{Input:} A grid of masses \(\sigma_1, \dots, \sigma_K\); a grid of points \(\iz_1, \dots, \iz_I\); eigenvalues $\ell_j$'s of $\bW_2$; $\lambda$. 
\STATE \textbf{Calculation:}  Compute
\[e_{ij} =  \frac{\hat{Q}_j(\iz_i)}{|\hat{Q}_j(\iz_i)| }- \frac{1}{|\hat{Q}_j(\iz_i)|}\sum_{k=1}^K \frac{\sigma_k^jw_k}{(\sigma_k  \lambda\hat{\varphi}(\iz_i) +\lambda)^j},~ i =1,\dots, I; j=1,2,\]

\STATE \textbf{Optimization:} 
Find the optimal $(w_1,\dots, w_K, \theta)$ through the linear programming
\renewcommand{\arraystretch}{0.3}  
\[
\begin{array}{ll}
\underset{\theta,w_1,\dots,w_K}{\text{minimize}} & \phantom{-\theta\leq}\theta\\[1pt]
\text{subject to} &  -\theta\leq \Re(e_{ij}) \leq \theta, \quad \forall~ i = 1,\dots, I; j=1,2, \\[1pt]
&  -\theta\leq \Im(e_{ij}) \leq \theta, \quad \forall~ i = 1,\dots, I; j=1,2, \\[1pt]
& w_k \geq 0, \quad \forall k = 1,\dots, K, \\[1pt]
&  \mbox{and } w_1 + w_2 + \dots + w_K = 1.
\end{array}
\]
\STATE \textbf{Truncation}: To mitigate the effect of floating-point errors, truncate the optimal weight $w_k$ to $w_k \times \mathbbm{1}(w_k>K^{-1}10^{-d})$ and rescale the truncated weights to be of sum one. The recommended choice of $d$ is $d =2$.  
\STATE \textbf{Output:} All postive weights \(\hat{w}_1, \dots, \hat{w}_B\), the associated masses $\hat{\sigma}_1, \dots, \hat{\sigma}_B$ and loss \(\theta\).
\end{algorithmic}
\end{algorithm}

We choose the weights as
\[ (w^*_1, \dots, w^*_K) = \arg\min_{w_j\geq 0,\sum w_j =1} L_\infty(w_1,\dots, w_K).\]
The main motivation for choosing the $L_\infty$ type loss is that given the point masses $\{\sigma_k\}_{k=1}^K$, the optimization can be formulated into a linear programming problem as in Algorithm \ref{algo:linear_program}.

Suppose the optimal positive weights given by Algorithm \ref{algo:linear_program} are $\hat{w}_1,\dots, \hat{w}_B$ and the associated point masses are $\hat{\sigma}_1>\hat{\sigma}_2>\cdots > \hat{\sigma}_B$. We estimate $\calH_j(h)$ by 
\[\hat{\calH}_j(h) = \sum_{k=1}^B\frac{\hat{w}_k \hat{\sigma}_k^j}{(\lambda - \hat{\sigma}_k h)^j},\quad j=1,2,3,  \quad 0< h< \lambda/\hat\sigma_{1}.\]
Given $\hat{\calH}_j(\cdot)$, $j=1,2,3$, the estimation of $s(x)$ and its derivatives over the domain $[0,\rho)$ can be performed efficiently by leveraging its inherent smoothness.  
First of all, following Lemma \ref{lemma:determine_rho}, the estimator $\hat{\rho}$ of $\rho$ is obtained in Algorithm \ref{algo:estimation_rho}. 
 Secondly, note that Eq. \eqref{eq:characterization_s2s3} defines an ordinary differential equation (ODE) system. We propose to estimate $s(x)$, $s'(x)$, $s''(x)$ by solving this ODE, as detailed in Algorithm \ref{algo:ODE}.

\begin{algorithm}[htbp]
\caption{Estimation of $\rho$}
\label{algo:estimation_rho}
\begin{algorithmic}
\STATE \textbf{Input:} $\lambda$, $\hat\calH_1(\cdot)$, $\hat\calH_2(\cdot)$, largest point-mass $\hat{\sigma}_1$, the associated weight $\hat{w}_1$;
\STATE \textbf{Case 1:} If $\hat{\gamma}_2\hat{w}_1 \geq 1$, $\hat{\rho} = \lambda/\hat{\sigma}_1$; 
\STATE \textbf{Case 2:} Otherwise, solve for the unique root $\hat{h}$ in $(-\infty, \lambda/\hat{\sigma}_1)$ to 
\[\frac{\hat{\gamma}_2 \hat{\calH}_2(h)}{(1+\hat{\gamma}_2 \hat{\calH}_1(h))^2}  =1.\]
The function on the left-hand side is monotonically increasing on $(-\infty, \lambda/\hat{\sigma}_1)$. The equation can be solved using Newton's method with an initial value pre-selected through grid search.  Then, $\hat{\rho} = \hat{h} + [1+\hat{\gamma}_2 \hat{\calH}_1(\hat{h})]^{-1}$.
\end{algorithmic}
\end{algorithm}

Following Algorithm \ref{algo:estimation_rho} and Algorithm \ref{algo:ODE}, we estimate $\beta$ by $\hat{\beta}$, which is the unique solution on $(0, \hat{\rho})$ to 
$\beta^2 \hat{s}'(\beta) = 1/\hat{\gamma}_1$.
Then, $\Theta_1$ and $\Theta_2$ are estimated as:
\begin{equation}\label{eq:estimators_Theta1Theta2}
\hat\Theta_{1} = \frac{1}{\hat\beta} \Big[1 + \hat\gamma_1 \hat\beta {\hat{s}} (\hat\beta)\Big] ~~\mbox{and}~~ \hat\Theta_{2} = \Big[\frac{\hat\gamma_1^3}{2} {\hat{s}''}(\hat\beta) + \frac{\hat\gamma_1^2}{\hat\beta^3} \Big]^{1/3}.
\end{equation}
\begin{algorithm}[htbp]
\caption{ODE Formulation for Estimation of $s(x)$, $s'(x)$, and $s''(x)$}
\label{algo:ODE}
\begin{algorithmic}
\STATE \textbf{Input:} $\hat{\rho}$, $\hat{\calH}_1(\cdot)$, $\hat{\calH}_2(\cdot)$, $\hat{\calH}_3(\cdot)$;
\STATE \textbf{Initial:} Find  $s$ as the unique solution to 
$\hat{\calH}_1 \Big( \frac{-1}{1+\hat{\gamma}_2 s}\Big)  = s$ {and} $\hat{\gamma}_2 \hat{\calH}_2\Big(\frac{-1}{1+\hat{\gamma}_2 s}\Big) < (1+\hat{\gamma}_2 s)^2$.
\STATE \textbf{ODE}: On $x \in [0,\hat{\rho})$, solve the following ODE with the initial $\hat{s}(0) =s$:
\begin{equation*}
\begin{aligned}
&\hat{s}'(x) =  \frac{\hat{\calH}_2( \hat{g}(x) )}{ 1 - \hat{\calH}_2(\hat{g}(x)) \hat{\zeta}(x) },\\
&\hat{s}''(x) = \frac{2[\hat\zeta(x)\hat{s}'(x) +1]^2 \hat{\calH}_3(\hat{g}(x)) - 2{[1+\hat{\gamma}_2 \hat{s}(x)]^{-3} } {\hat{\calH}_2(\hat{g}(x)) (\hat{\gamma}_2\hat{s}'(x))^2}  }{1 -  \hat{\zeta}(x) \hat{\calH}_2(\hat{g}(x))},\\
&\hat{g}(x) = x - \frac{1}{1+\hat{\gamma}_2 \hat{s}(x)}, ~~\mbox{and } ~~ \hat\zeta(x) = \frac{\hat{\gamma}_2}{ (1+\hat{\gamma}_2\hat{s}(x))^2} .
\end{aligned}
\end{equation*}
\STATE \textbf{Output:} Estimated functions: $\hat{s}(x)$, $\hat{s}'(x)$, $\hat{s}''(x)$ on $[0,\hat{\rho})$.
\end{algorithmic}
\end{algorithm}


In Algorithm \ref{algo:linear_program}, the choice of the $L_\infty$ type loss function leads to a linear programming problem. Alternative convex loss functions, such as the $L_2$-norm, can also be considered, transforming the optimization into a general convex problem. However, our numerical studies indicate that the $L_2$-loss provides no significant advantage over the $L_\infty$-loss. Given that linear programming is computationally more efficient, we recommend the use of the $L_\infty$-loss. A detailed comparison between different loss functions is beyond the scope of the current work.

Additional discussion, including a comparison between Algorithm~\ref{algo:linear_program} and the methods of \cite{el2008spectrum} and \cite{ledoit2012nonlinear}, is provided in Section~S.2 
of the Supplementary Material.

\subsection{Consistency}\label{subsec:consistency}
In this section, we show the consistency of the proposed estimators.

\begin{definition}\label{def:grid_size}
For a grid of complex values $J = \{t_1, t_{2}, t_3,\dots\}$, define the corresponding grid size as $\operatorname{size}(J) = \sup_{k} |t_k-t_{k-1}|$.
\end{definition}
\begin{definition}\label{def:grid_cover}
A grid $R$ of real numbers is said to cover a compact interval $[a,b]$ if there exists at least one $t_k \in R$ with $t_k\leq a$ and at least another $t_{k'}$ with $t_{k'} \geq b$. A sequence of grids $\{R_m, m =1,2,3,\dots \}$ is said to eventually cover a compact interval $[a,b]$ if for every $\epsilon >0$ there exists $M$ such that $R_m$ covers the compact interval $[a+\epsilon, b-\epsilon]$ for all $m\geq M$. 
\end{definition}

\begin{theorem}\label{thm:consistency_estimators} 
Suppose that \ref{enum:high_dimensional_regime}--\ref{enum:regular_edge} hold. Assume that \(R_p = \{\sigma_{p1}, \sigma_{p2}, \dots\}\), for \(p = 1, 2, \dots\), is a sequence of grids that eventually covers \([0, \sigma_{\max}]\), with \(\text{size}(R_p) = o(p^{-2/3})\). Suppose that \(J_p = \{\iz_{p1}, \iz_{p2}, \dots\}\), for \(p = 1, 2, \dots\), is a sequence of grids.

\begin{itemize}
    \item[(i)] Let \(\hat{F}_p\) denote the estimated measure obtained from Algorithm \ref{algo:linear_program} using the inputs \(\{\sigma_k\} = R_p\) and \(\{\iz_i\} = J_p\). If $J_p$ is  such that for some $\nu>0$,
\begin{equation}\label{eq:consistency_condition1}
\sup_{\iz \in D(\nu)} \inf_{y \in J_p} |\iz - \varphi(y)| = o(p^{-2/3}),
\end{equation}
where \(D(\nu)\) is the closed half-disk in \(\mathbb{C}^+\) centered at 0 with radius \(\nu\). Then, there exists a constant $C$ such that for any $T\geq 3$
    \[
    \calD_W(\hat{F}_p, F^{\Sigma_p}) \leq \frac{C(1+\nu)}{\nu^2} p^{-2/3} \exp(T)\delta_p + \frac{C}{T},
    \]
    where $\mathcal{D}_W(F_1, F_2)$ denotes the \emph{Wasserstein distance} between distributions $F_1$ and $F_2$, defined as   
\[ \mathcal{D}_W(F_1, F_2) = \sup_{f} \left\{ \left|\int f dF_1 - \int fdF_2  \right|: f \text{ is 1- Lipschitz}\right\},\]
and $\delta_p$ is such that $\delta_p \stackrel{P}{\longrightarrow}0$, as $p\to\infty$. 
    \item[(ii)] Let \(h_\beta\) be the solution to Eq. \eqref{eq:determine_s} when \(x_0 = \beta\). If $J_p$ is such that there exists $\varepsilon>0$ such that $\{\varphi(y), ~y\in J_p\}$ eventually covers $[-h_\beta/\lambda - \varepsilon, -h_\beta/\lambda + \varepsilon]$. 
    Then, the proposed estimators \(\hat{\Theta}_1\) and \(\hat{\Theta}_2\), obtained using the inputs \(\{\sigma_k\} = R_p\) and \(\{\iz_i\} = J_p\), satisfy:
    \begin{align*}
    p^{2/3} |\hat{\Theta}_1 - \Theta_1| \stackrel{P}{\longrightarrow} 0~~\mbox{and}~~
    |\hat{\Theta}_2 - \Theta_2| \stackrel{P}{\longrightarrow} 0.
    \end{align*}
\end{itemize}
\end{theorem}

\begin{lemma}\label{lemma:consistency_condition1}
For any \(\gamma_1\), \(\gamma_2\), \(\lambda\), and \(F^{\Sigma_p}\) satisfying \ref{enum:high_dimensional_regime}--\ref{enum:regular_edge}, there exists \(\nu > 0\) and a sequence of grids \(\{J_p\}\) such that Condition \eqref{eq:consistency_condition1} is satisfied.  
\end{lemma}
Theorem \ref{thm:consistency_estimators} establishes two fundamental results. First, it provides a Berry-Esseen type bound on the Wasserstein distance between the estimated population ESD and its limiting counterpart. Second, it demonstrates that the proposed estimator achieves \(o_P(p^{-2/3})\) consistency, provided that $\{ \varphi(y), y\in J_p\}$ is dense around $-h_{\beta}/\lambda$. We can find a sequence of $J_p$ satisfying the condition when \(\beta\) is not excessively close to \(\rho\). The positioning of \(\beta\) is influenced by the parameters \(\gamma_1\) and \(\gamma_2\). Specifically, $\beta$ approaches $0$, as $\gamma_1$ grows. Reversely, for small $\gamma_1$, $\beta$ is close to $\rho$.   


\begin{lemma}\label{lemma:consistency_condition3}
Fix any \(\gamma_2\), \(\lambda\), and \(F^{\Sigma_p}\) satisfying \ref{enum:high_dimensional_regime}--\ref{enum:regular_edge}. There exists \(\gamma_1^{(0)}\) such that for any \(\gamma_1 > \gamma_1^{(0)}\), we can find a sequence of grids \(\{J_p\}\) such that the condition in (ii) of Theorem \ref{thm:consistency_estimators} is satisfied.
\end{lemma}
Lemma \ref{lemma:consistency_condition3} indicates that the condition in (ii) of Theorem \ref{thm:consistency_estimators} can be satisfied when \(\gamma_1\) is large. This scenario corresponds to the regime where \(n_2\) is large and \(n_1\) is relatively small compared to \(p\). 
The numerical simulation results reported in Tables S.4.1 and S.4.2 
of the Supplementary Material support these findings, showing that the estimation precision decreases as \(p/n_2\) increases or as \(p/n_1\) decreases under the given simulation settings. See Section \ref{subsec:estimation_precision} for details. Nonetheless, the overall estimation precision remains well-controlled within a reasonable range for all settings under consideration.

\section{Power analysis under low-rank alternatives}\label{sec:power_analysis}
In this section, we study the power of the proposed family of tests under a class of low-rank alternatives, with particular emphasis on the role of the regularization parameter $\lambda$. We begin with a deterministic alternative framework, which is commonly used in the literature for analyzing asymptotic power. We then adopt a Bayesian perspective by introducing a class of priors on the alternatives that capture structural information. Within this framework, we study the effect of $\lambda$ and compare the proposed method with the limiting cases when $\lambda\to 0$ \citep{han2016tracy} and $\lambda\to\infty$ both analytically and numerically. 

\subsection{Deterministic alternatives}\label{subsec:deterministic_alternatives}

Consider a sequence of rank-one deterministic alternatives as follows. 
\begin{description}
\item[\textbf{DA}] Under $H_a: BC \neq 0$, assume that $BC$ satisfies 
\begin{equation}\label{eq:alternatives_deterministic}
\frac{1}{n_1} BC ( C^T(XX^T)^{-1} C)^{-1} C^T B^T = p^{s} d_p q_pq_p^T, 
\end{equation}
where $s>0$ is a constant controlling the signal magnitude; $d_p$ is a sequence of positive scalars such that $d_p \to d>0$; and $q_p\in \mathbb{R}^p$ is a sequence of unit vectors representing the direction in which the signal resides.
\end{description}

\begin{lemma}\label{lemma:alternatives_deterministic}
Suppose that \ref{enum:high_dimensional_regime}-\ref{enum:regular_edge} hold and $BC$ is as in  \textbf{DA}. For any fixed $\lambda>0$, we have as $p\to\infty$, 
\[ \frac{1}{p^s}\ell_{\max}(\bF_\lambda) =  d_p q_p^T \left( \lambda \varphi_p(-\lambda) \Sigma_p + \lambda I_p \right)^{-1} q_p   + o_P(1).\]
\end{lemma}

Given the estimators $\hat\Theta_1(\lambda)$ and $\hat{\Theta}_2(\lambda)$ of $\Theta_{1p}(\lambda)$ and $\Theta_{2p}(\lambda)$, respectively, define the standardized test statistic as 
\begin{equation}\label{eq:std_stat}
\tilde{\ell}(\lambda) = \frac{p^{2/3}}{\hat{\Theta}_2(\lambda)}\{ \ell_{\max}(\bF_\lambda) - \hat{\Theta}_1(\lambda)\}.
\end{equation}
\begin{corollary}\label{corollary:power_consistency_deterministic}
Suppose that \ref{enum:high_dimensional_regime}-\ref{enum:regular_edge} hold and $BC$ is as in  \textbf{DA}. Fix $\lambda>0$. If $\hat\Theta_1(\lambda) - \Theta_{1p}(\lambda) \stackrel{P}{\longrightarrow}0$ and  $\hat\Theta_2(\lambda) - \Theta_{2p}(\lambda) \stackrel{P}{\longrightarrow}0$, then the power of the test is such that, as $p\to\infty$,
\[\mP\left( \tilde{\ell}(\lambda) > \TW_1(1-\alpha)\mid BC\right) \to 1,\]
where $\TW_1(1-\alpha)$ denotes the $(1-\alpha)$-quantile of $\TW_1$.
\end{corollary}

\subsection{Probabilistic alternatives} \label{subsec:probabilistic_alternatives}

While deterministic alternatives offer valuable insight, they can be somewhat restrictive for a systematic investigation of the power properties of the test. To address this limitation, we consider probabilistic alternatives in the form of a sequence of prior distributions on $BC$. This approach offers greater flexibility, allowing structural information about the regression coefficient matrix $B$ and the constraint matrix $C$ to be incorporated naturally through the choice of prior distributions. Such a probabilistic approach has been previously proposed and explored in \citet{li2020adaptable} and \citet{li2020high}.

\begin{description}
	\item[\textbf{PA}] Assume that under $H_a$, $BC$ is such that 
	\begin{equation}\label{eq:alternatives_prob}
	\frac{1}{n_1} BC ( C^T(XX^T)^{-1} C)^{-1} C^T B^T = p^{s-1} D^{1/2} \nu \nu^T D^{1/2},   
	\end{equation} 
	where $s>0$ is a constant controlling the signal magnitude; $D$ is the $p\times p$ symmetric covariance matrix satisfying $\|D\|_2 <\infty$; and $\nu$ is a random vector with independent entries having mean zero, variance one, and finite fourth moments.  
\end{description}
Under the \textbf{PA} model, the signal is aligned with the eigen-direction defined by the random vector $D^{1/2}\nu$. Structural information can be incorporated through appropriate choices of the covariance matrix $D$ and the distribution of $\nu$. For instance, $D$ may encode community-based structure or stationary auto-covariance patterns. The vector $\nu$ can be selected to control signal sparsity: a Gaussian prior yields a dense signal across all coordinates, while a Bernoulli prior induces sparsity by concentrating the signal on a subset of coordinates.

\begin{lemma}\label{lemma:alternatives_prob}
Suppose that \ref{enum:high_dimensional_regime}-\ref{enum:regular_edge} hold and $BC$ is as in  \textbf{PA}. For any fixed $\lambda>0$, we have as $p\to\infty$, 
\[ \frac{1}{p^s}\ell_{\max}(\bF_\lambda) =  \frac{1}{p}  \tr \left[\left( \lambda \varphi_p(-\lambda) \Sigma_p + \lambda I_p \right)^{-1} D\right]   + o_P(1).\]
\end{lemma}

\begin{corollary}\label{corollary:power_consistency_prob}
Suppose that \ref{enum:high_dimensional_regime}-\ref{enum:regular_edge} hold and $BC$ is as in  \textbf{PA}. Fix $\lambda>0$. If $\hat\Theta_1(\lambda) - \Theta_{1p}(\lambda) \stackrel{P}{\longrightarrow}0$ and  $\hat\Theta_2(\lambda) - \Theta_{2p}(\lambda) \stackrel{P}{\longrightarrow}0$, then the power of the test is such that, as $p\to\infty$,
\[\mP\left( \tilde{\ell}(\lambda) > \TW_1(1-\alpha)\mid BC \right) \xrightarrow{\hspace*{0.2cm}\mP_{BC}\hspace*{0.2cm}} 1,\]
where $\xrightarrow{\hspace*{0.2cm}\mP_{BC}\hspace*{0.2cm}}$ denotes convergence in probability with respect to the prior measure of $BC$. 
\end{corollary}

Lemma \ref{lemma:alternatives_prob} shows that, under the rank-one alternative as in \textbf{PA}, the power of the proposed test is primarily governed—up to a scaling factor—by the term 

\begin{equation}\label{eq:def_Xi}
\begin{split}
&\SNR_p(\lambda, D) = \xi_p(\lambda, D)/\Theta_{2p}(\lambda),\quad\mbox{where} \\
&\xi_p(\lambda, D) = p^{-1}\tr\left[\left(\lambda \varphi_p(-\lambda) \Sigma_p + \lambda I_p \right)^{-1} D \right].    
\end{split}
\end{equation}
We interpret $\SNR_p(\lambda,D)$ as the scaled asymptotic signal-to-noise ratio (SNR).

\subsection{Power comparison}\label{subsec:power_comparison}
\textcolor{black}{Recall from Remark~\ref{remark:link_to_literature} that the statistic $\ell_H$ in \citet{han2016tracy} can be interpreted as the limiting case of $\ell_{\max}(\bF_\lambda)$ as $\lambda \to 0$. On the other hand, as $\lambda \to \infty$, the proposed test reduces to one based on the largest eigenvalue of $\bW_1$, since $\lambda \bW_1 (\bW_2+\lambda I_p)^{-1} \to \bW_1$ pointwise. In this subsection, we derive the asymptotic behavior in these two limiting regimes and compare them with the proposed method. 
\begin{lemma}\label{lemma:alternative_Han2016}
Suppose that Conditions~\ref{enum:high_dimensional_regime}--\ref{enum:regular_edge} hold and that $BC$ is in \textbf{PA}. (i) If $\gamma_2 <1$, as $p\to\infty$, we have $p^{-s}\ell_{H} =  (1-\gamma_2)^{-1}\, p^{-1}\tr(D \Sigma_p^{-1}) + o_P(1)$;
(ii) If $\gamma_2 >1$ and $\gamma_1^{-1} + \gamma_2^{-1}>1$, then for any constant $\varepsilon>0$, we have 
\[\mP\Big(p^{-s} \ell_H  \leq \gamma_2^{-1} (\gamma_2-1)^{-1} p^{-1} \tr(D \Sigma_p^{-1}) + \varepsilon \Big) \longrightarrow1, \quad \text{as }p\to\infty.\]
\end{lemma}
We denote the asymptotic SNR of the test using $\ell_H$ to be $\SNR_p(0,D)$, after the same scaling as that in $\SNR_{p}(\lambda, D)$. Using the normalization parameter $\sigma_p$ defined in (2.7) of \citet{han2016tracy}, we obtain an upper bound on $\SNR_p(0,D)$ as
\begin{equation}\label{eq:SNR_0}
{\rm SNR}_{p}(0,D)  \leq \Theta_{H}^{-1} \,(\gamma_2^{-1} \wedge 1)\, |\gamma_2 - 1|^{-1}\, p^{-1}\tr(D \Sigma_p^{-1}).
\end{equation}
Here, $\Theta_{H}$ is the counterpart of $\Theta_{2p}(\lambda)$ derived using $\sigma_p$ after transformation. We give an explicit expression for $\Theta_{H}$ in Section~S.3 
of the Supplementary Material. Notably, we can show that $\Theta_{H}^{-1}|\gamma_2 -1|^{-1} \leq \gamma_1^{-2/3} |\gamma_2-1|^{-1} (1-(\gamma_2 \wedge \gamma_2^{-1} )^{1/2})^2$. It then follows that when $\gamma_2$ is close to $1$,
\[{\rm SNR}_p(0,D)\lesssim \frac{1}{2}\gamma_1^{-2/3} (1- \sqrt{\gamma_2 \wedge \gamma_2^{-1}} ) \tr(D \Sigma_p^{-1}).\] 
This bound shows that the SNR deteriorates as $\gamma_2 \to 1\pm$ for any sequence of $D$ whose spectral norms remain uniformly bounded. By contrast, when $\gamma_2$ is away from $1$, ${\rm SNR}_p(0,D)$ becomes relatively large. See Section S.3 for the proof of the technical results.} 


\textcolor{black}{
In the limiting regime when $\lambda \to \infty$, the test based on $\ell_{\max}(\bW_1)$ can be viewed as an instance of replacing $\bW_2$ by $I_p$, an idea originally proposed in \cite{bai1996effect} in the context of two-sample mean testing. Accordingly, the associated scaled asymptotic SNR is given by 
\begin{equation}\label{eq:SNR_infty}
\SNR_p(\infty,D) \coloneqq \lim_{\lambda \to \infty} \SNR_p(\lambda,D) =  \frac{\lim_{\lambda \to \infty}\lambda \xi_p(\lambda,D)}{\lim_{\lambda \to \infty}\lambda\Theta_{2p}(\lambda) }  =\frac{p^{-1}\tr(D)}{\Theta_{2p}(\infty)}.
\end{equation}
We give an explicit expression for $\Theta_{2p}(\infty)$ in Section~5.2 of the Supplementary Material. Notably, this limit is independent of $\gamma_2$, indicating that the resulting test is particularly advantageous when $\gamma_2$ is excessively large.} 

\textcolor{black}{The proposed method can be viewed as interpolating between $\ell_H$ and $\ell_{\max}(\bW_1)$: small values of $\lambda$ produce behavior close to the former, whereas large values of $\lambda$ approach the latter. However, a direct analytical characterization of how $\lambda$ influences the SNR is intractable, due to the implicit and nonlinear dependence of $\Theta_{2p}(\lambda)$ on $\lambda$. In the next section, we therefore investigate this relationship numerically under several representative settings.}

\subsection{Numerical study: Effect of the regularization parameter}\label{subsec:effect_lambda}

\textcolor{black}{We examine four spectral structures of $\Sigma_p$: the identity model (Identity), a polynomially decaying spectrum (Poly-Decay), an AR(1) autocovariance structure (AR-ACF), and a three–point mixture model (Point-Mix). For each choice of $\Sigma_p$, we consider two signal structures, $D \propto I_p$ and $D \propto \Sigma_p$. The detailed specifications are provided in Section~S.3 
of the Supplementary Material.}

\begin{figure}[htbp]
  \centering 
  \begin{subfigure}{0.24\textwidth}
    \includegraphics[width=\linewidth,height=0.6\linewidth]{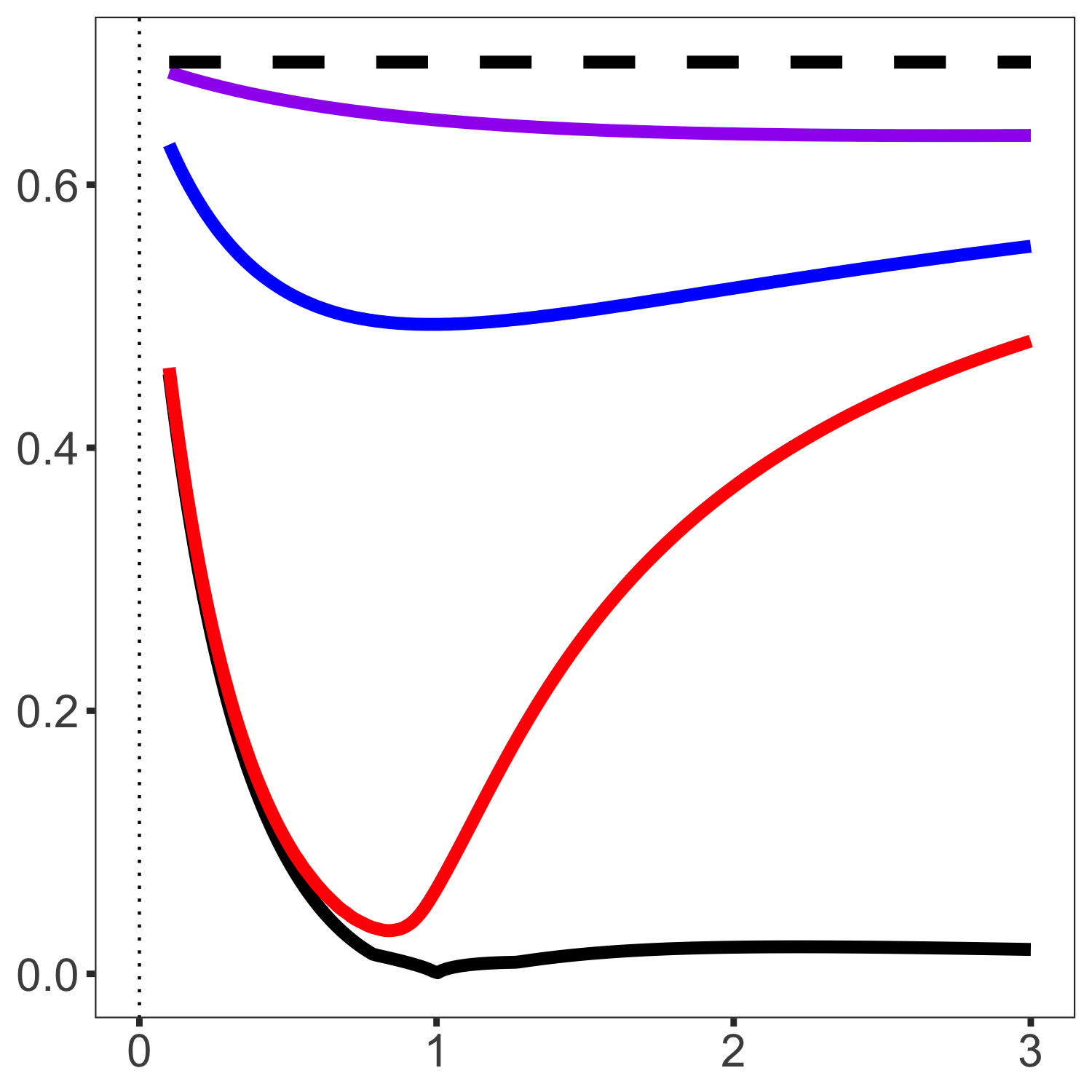}
   \end{subfigure}
  \hfill
  \begin{subfigure}{0.24\textwidth}
    \includegraphics[width=\linewidth,height=0.6\linewidth]{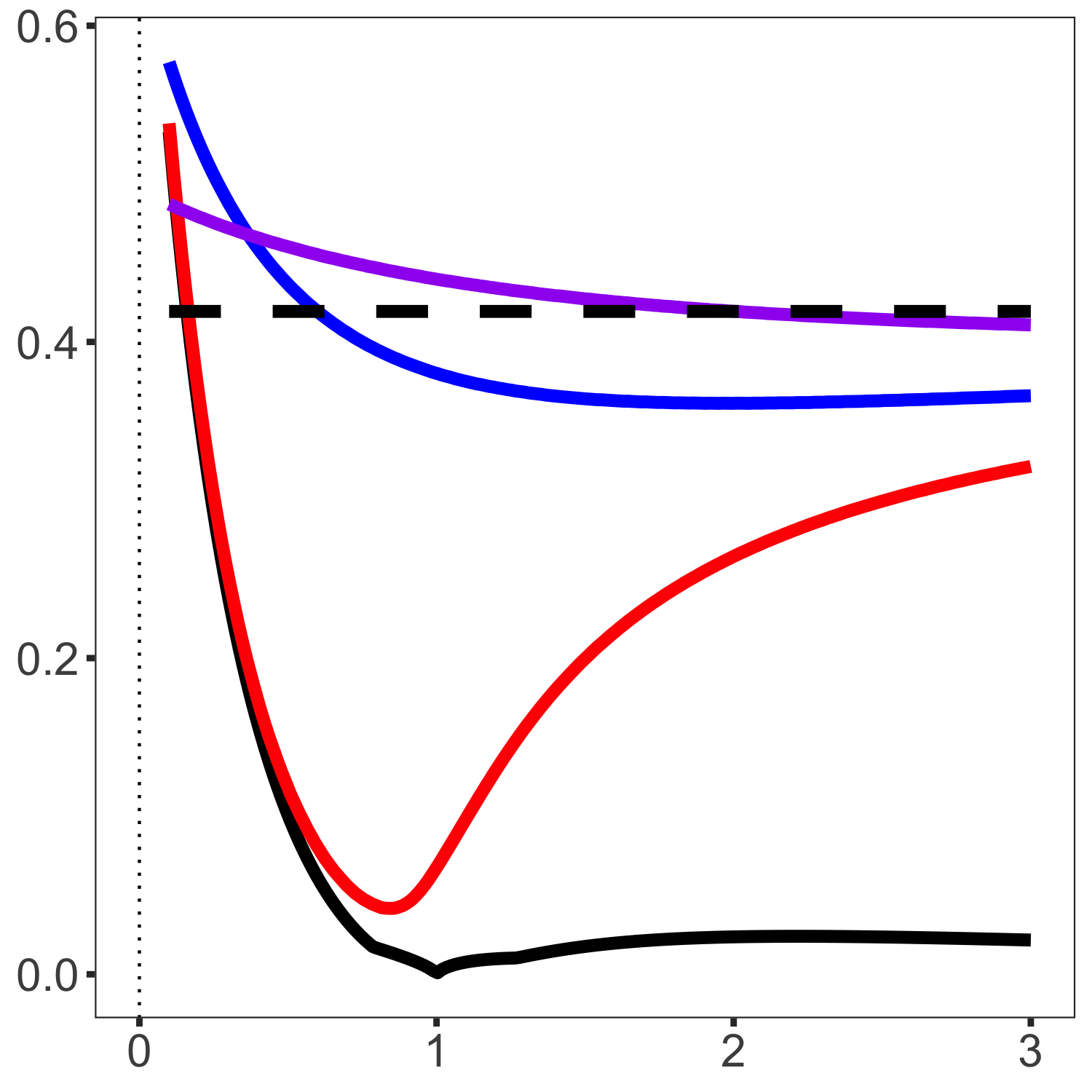}
   \end{subfigure}
  \hfill
  \begin{subfigure}{0.24\textwidth}
    \includegraphics[width=\linewidth,height=0.6\linewidth]{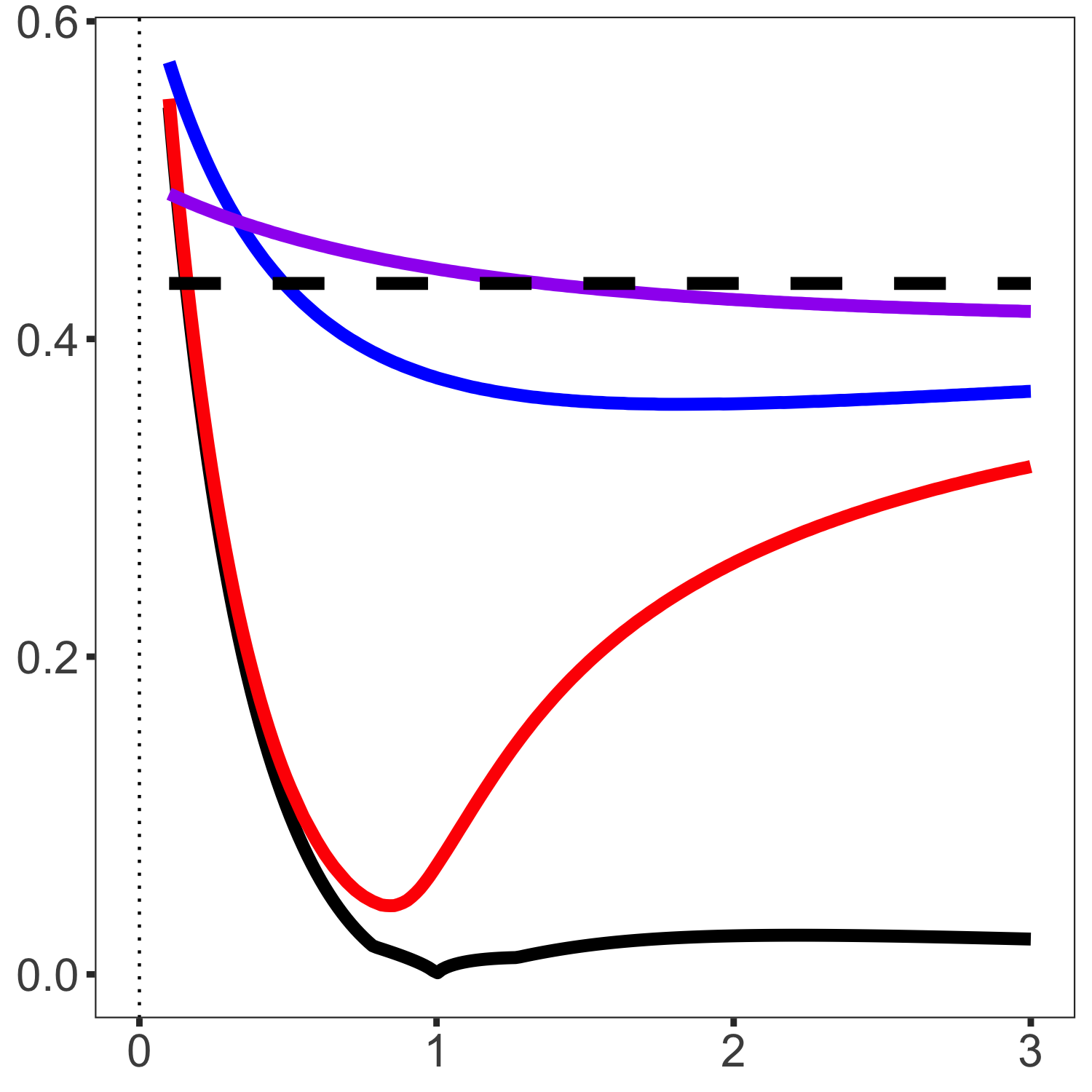}
   \end{subfigure}
  \begin{subfigure}{0.24\textwidth}
    \includegraphics[width=\linewidth,height=0.6\linewidth]{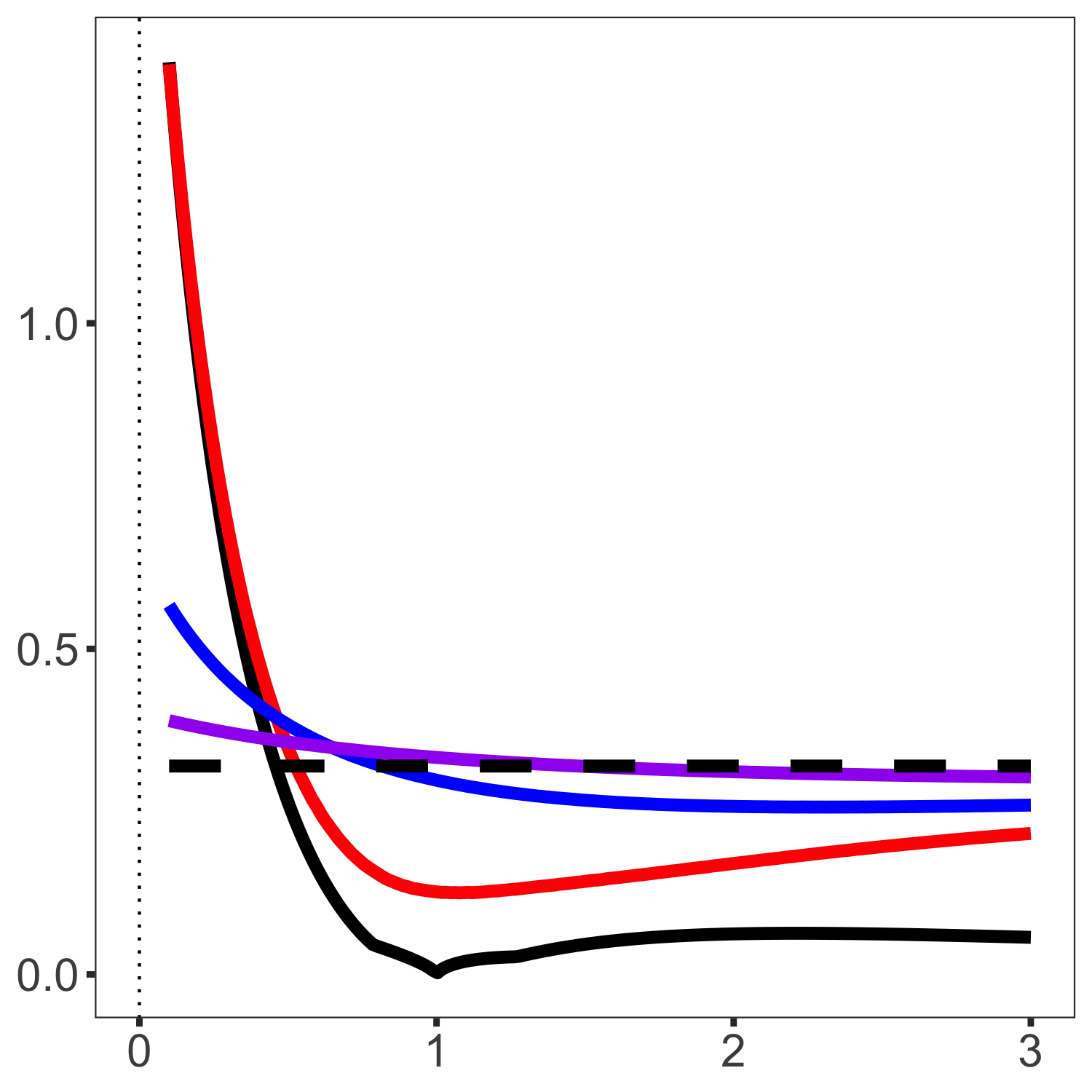}
   \end{subfigure}
  \vfill
  \begin{subfigure}{0.24\textwidth}
    \includegraphics[width=\linewidth,height=0.6\linewidth]{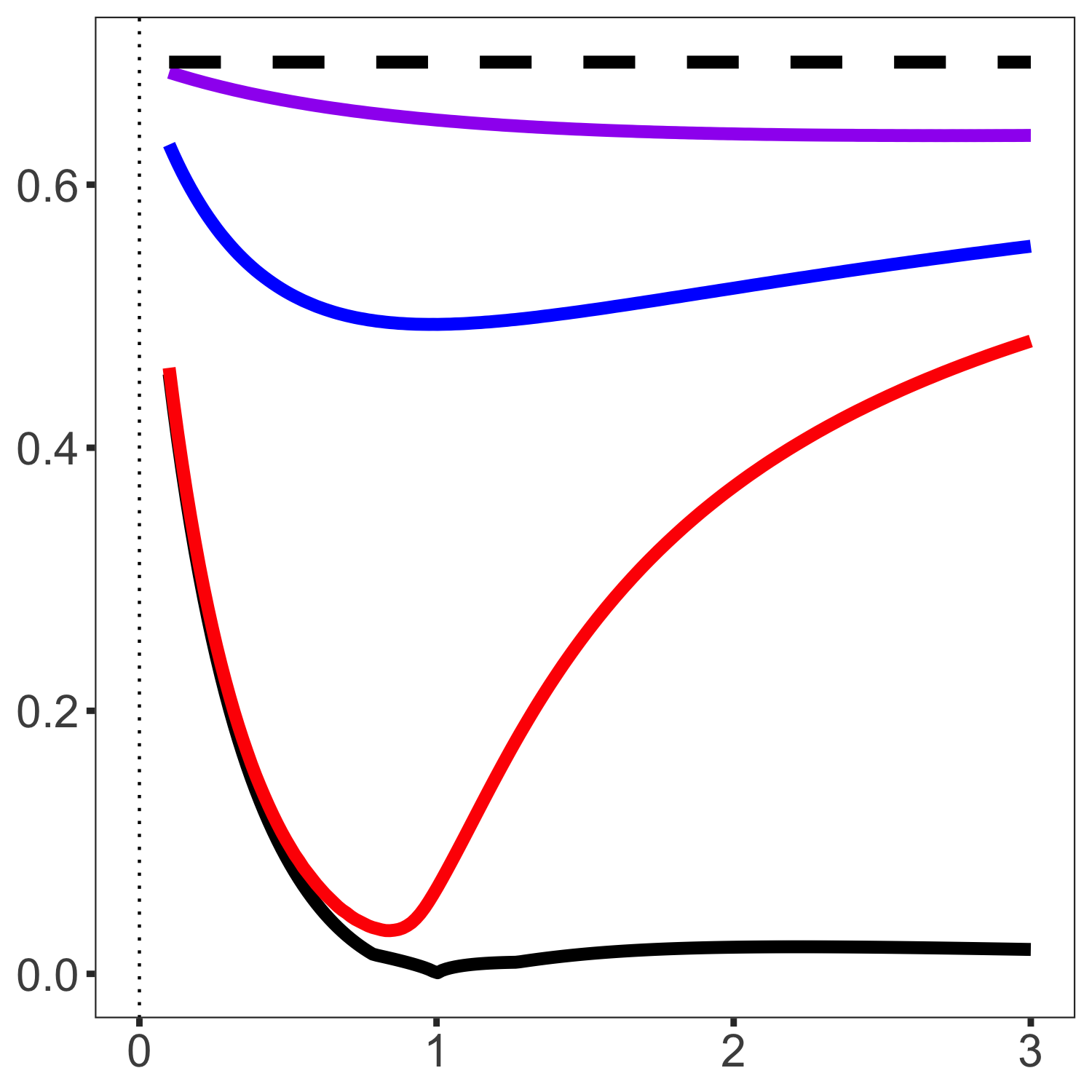}
   \end{subfigure}
  \hfill
  \begin{subfigure}{0.24\textwidth}
    \includegraphics[width=\linewidth,height=0.6\linewidth]{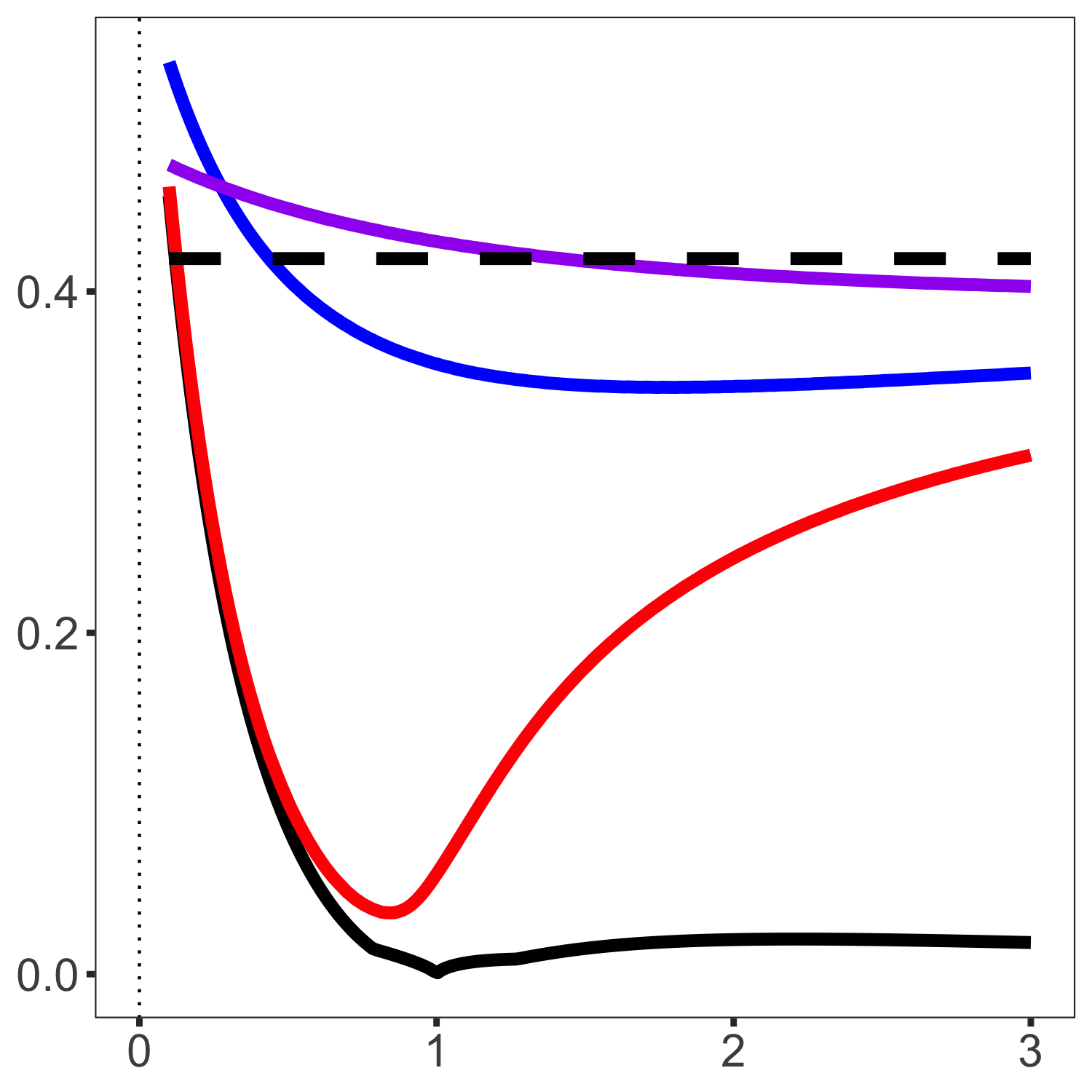}
   \end{subfigure}
   \hfill
   \begin{subfigure}{0.24\textwidth}
    \includegraphics[width=\linewidth,height=0.6\linewidth]{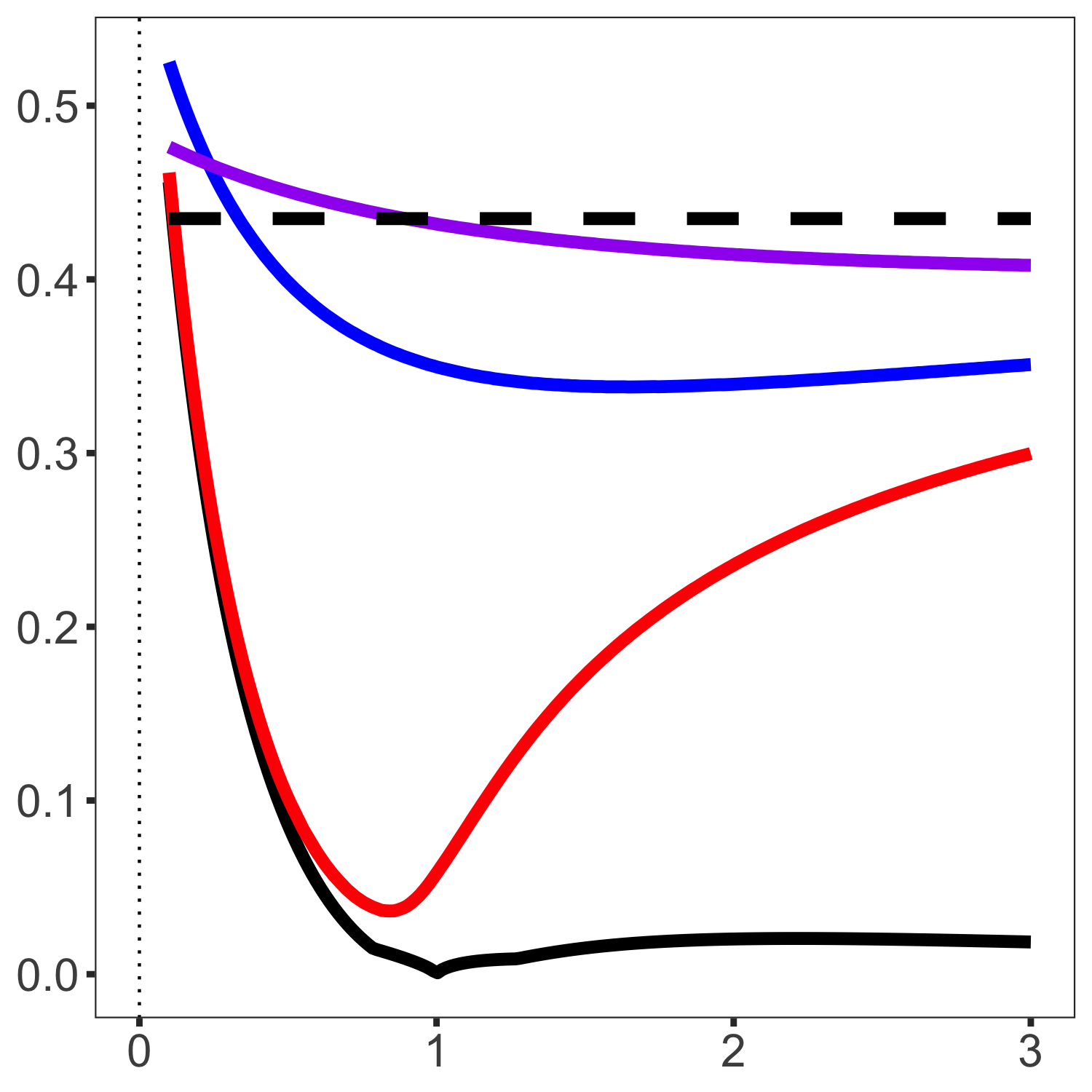}
   \end{subfigure}
  \hfill
  \begin{subfigure}{0.24\textwidth}
    \includegraphics[width=\linewidth,height=0.6\linewidth]{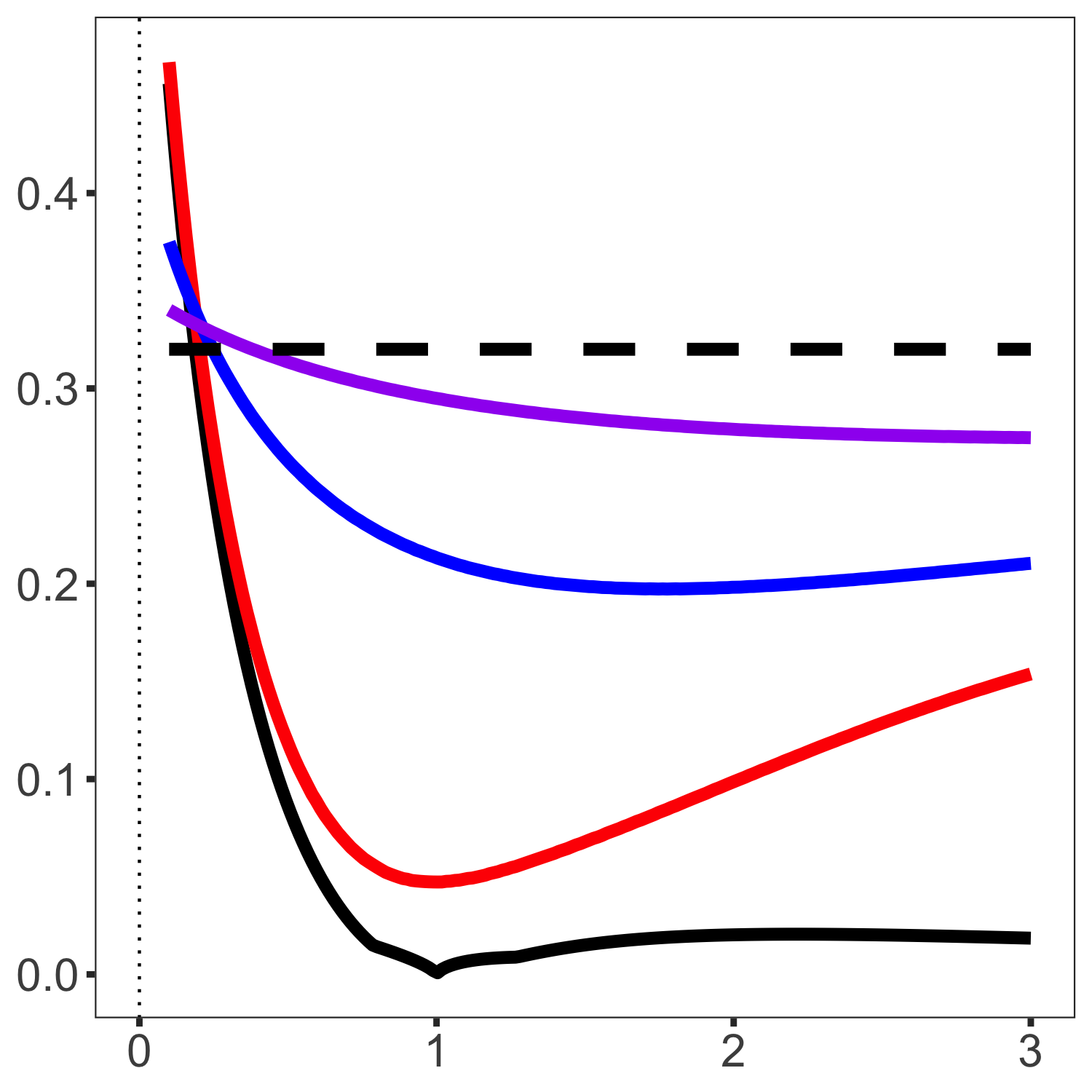}
   \end{subfigure}
\caption{$\SNR$ as a function of $\gamma_2$ when $\gamma_1 =0.5$. Columns: $\Sigma_p$ = Identity, Poly-Decay, AR-ACF, Point-Mix. Rows: $D \propto I_p$ (top), $D \propto \Sigma_p$ (bottom). Red/blue/purple: $\lambda =0.01$, $1$, and $5$. Black solid: upper bound of $\SNR_p(0)$ \citep{han2016tracy}. Black dashed: $\SNR_p(\infty)$.}
  \label{fig:SNRvsGamma2_Gamma1_05}
\end{figure}

\begin{figure}[htbp]
    \centering
    \begin{subfigure}{0.24\textwidth}
    \includegraphics[width=\linewidth,height=0.6\linewidth]{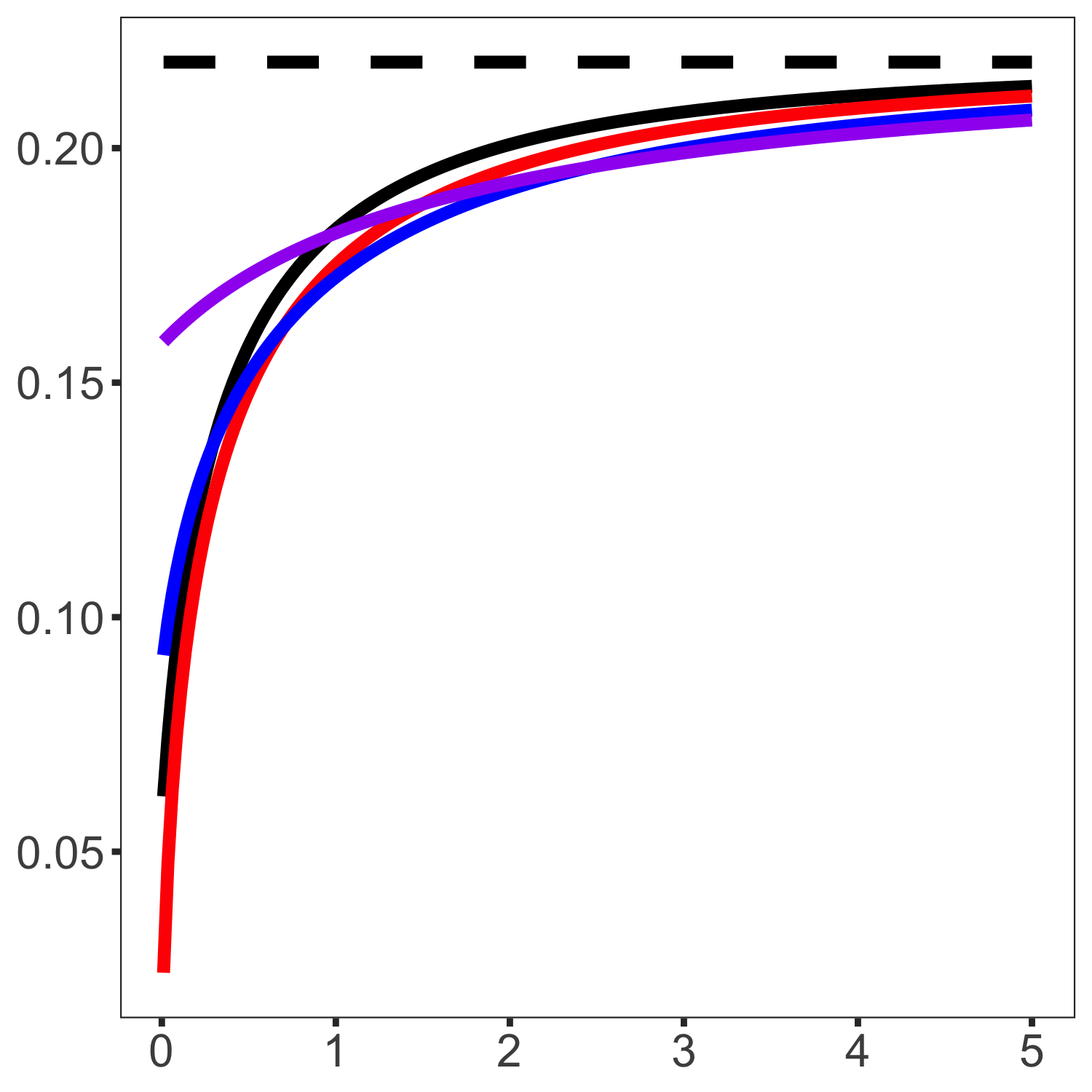}    
    \end{subfigure}
    \hfill
    \begin{subfigure}{0.24\textwidth}
    \includegraphics[width=\linewidth, height=0.6\linewidth]{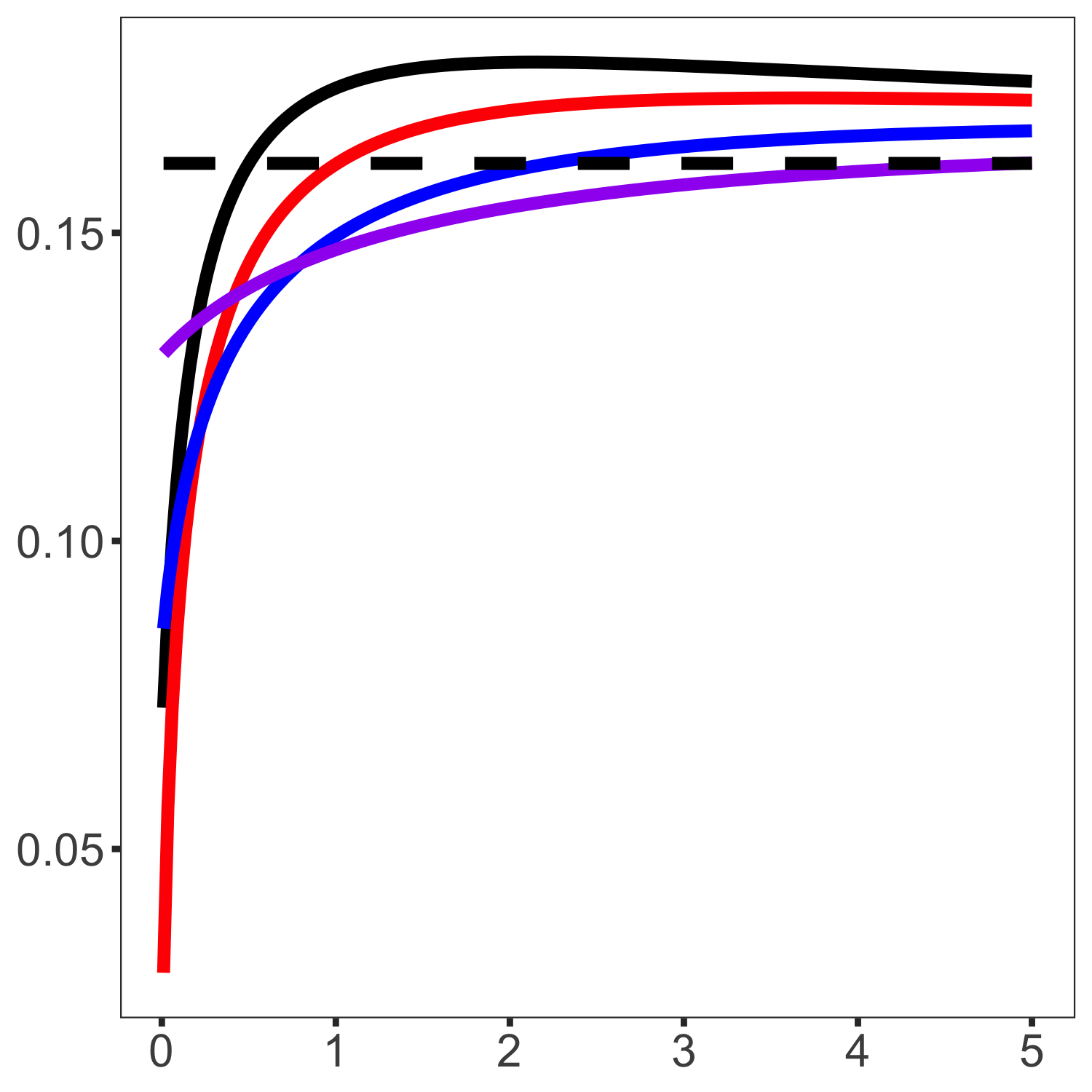}    
    \end{subfigure}
    \hfill
    \begin{subfigure}{0.24\textwidth}
    \includegraphics[width=\linewidth, height=0.6\linewidth]{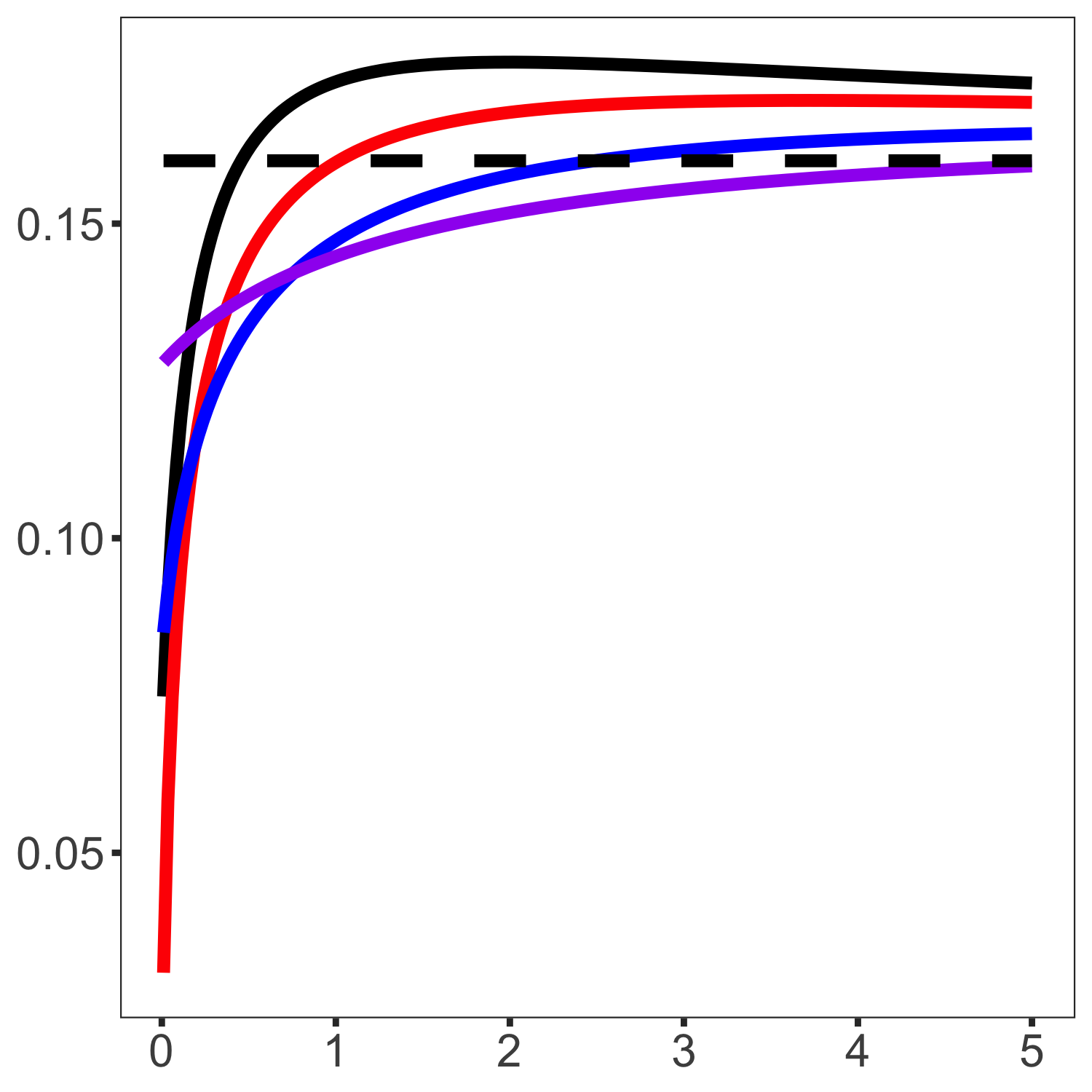}    
    \end{subfigure}
    \hfill
    \begin{subfigure}{0.24\textwidth}
    \includegraphics[width=\linewidth, height=0.6\linewidth]{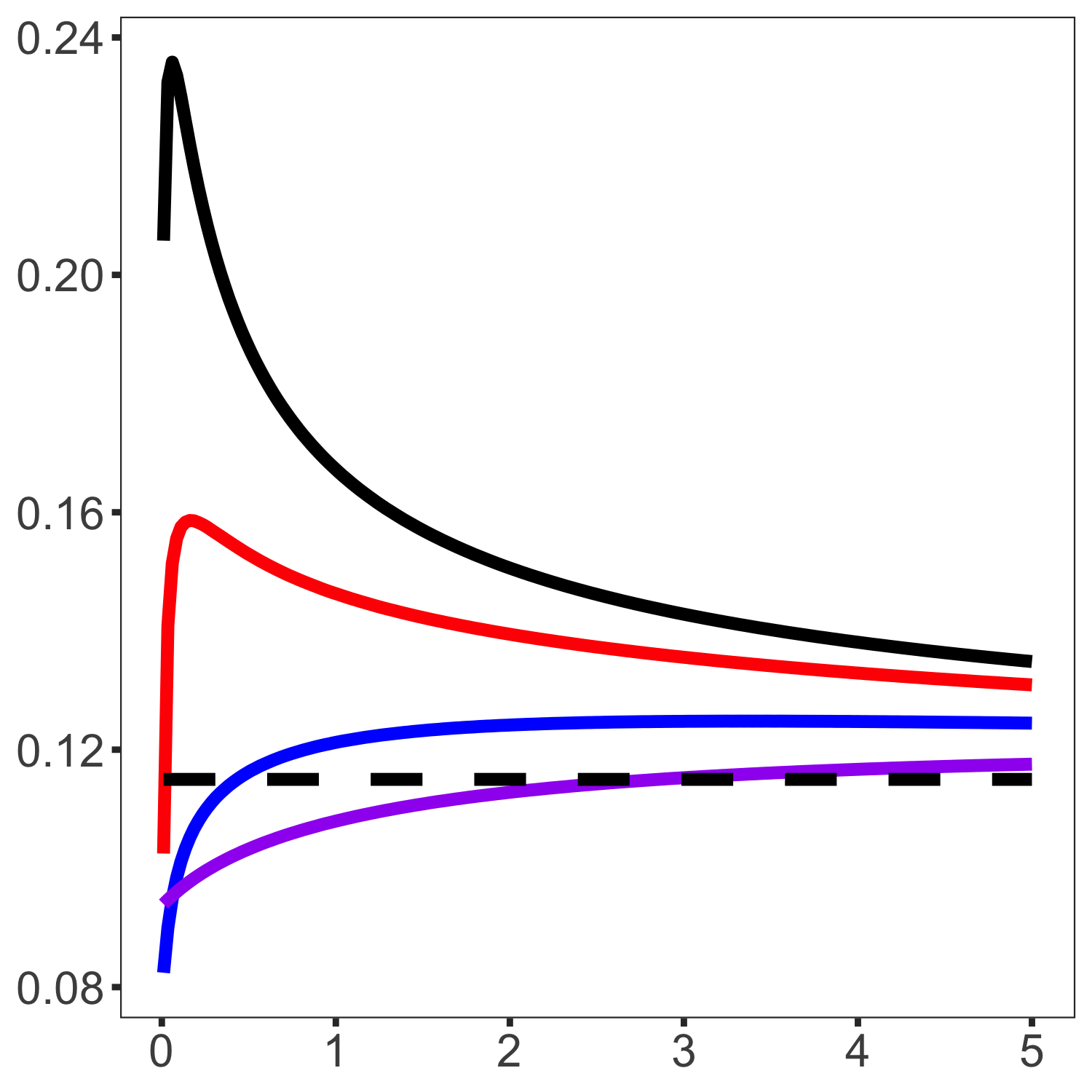}    
    \end{subfigure}
    \vfill
    \begin{subfigure}{0.24\textwidth}
    \includegraphics[width=\linewidth, height=0.6\linewidth]{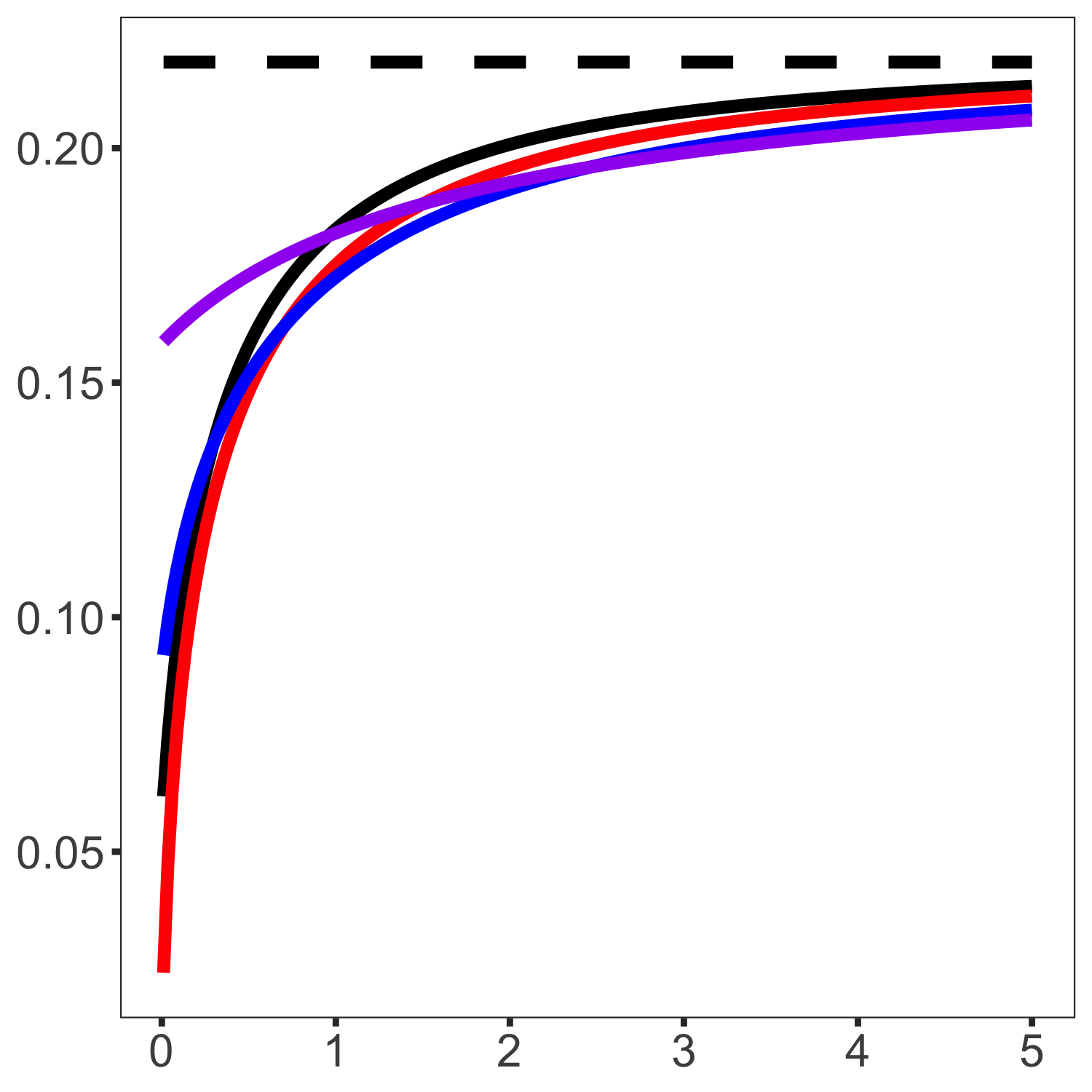}    
    \end{subfigure}
    \hfill
    \begin{subfigure}{0.24\textwidth}
    \includegraphics[width=\linewidth, height=0.6\linewidth]{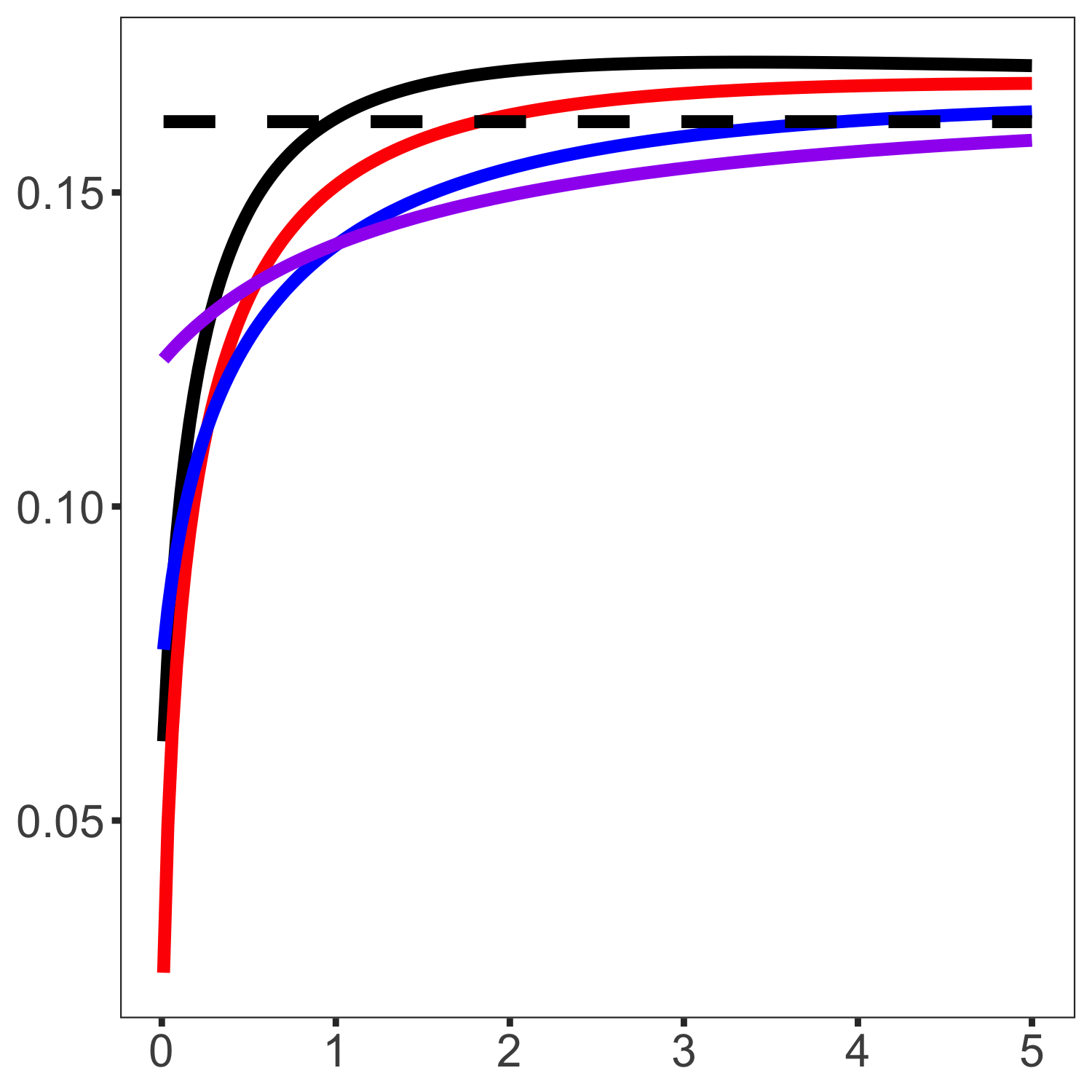}    
    \end{subfigure}
    \begin{subfigure}{0.24\textwidth}
    \includegraphics[width=\linewidth, height=0.6\linewidth]{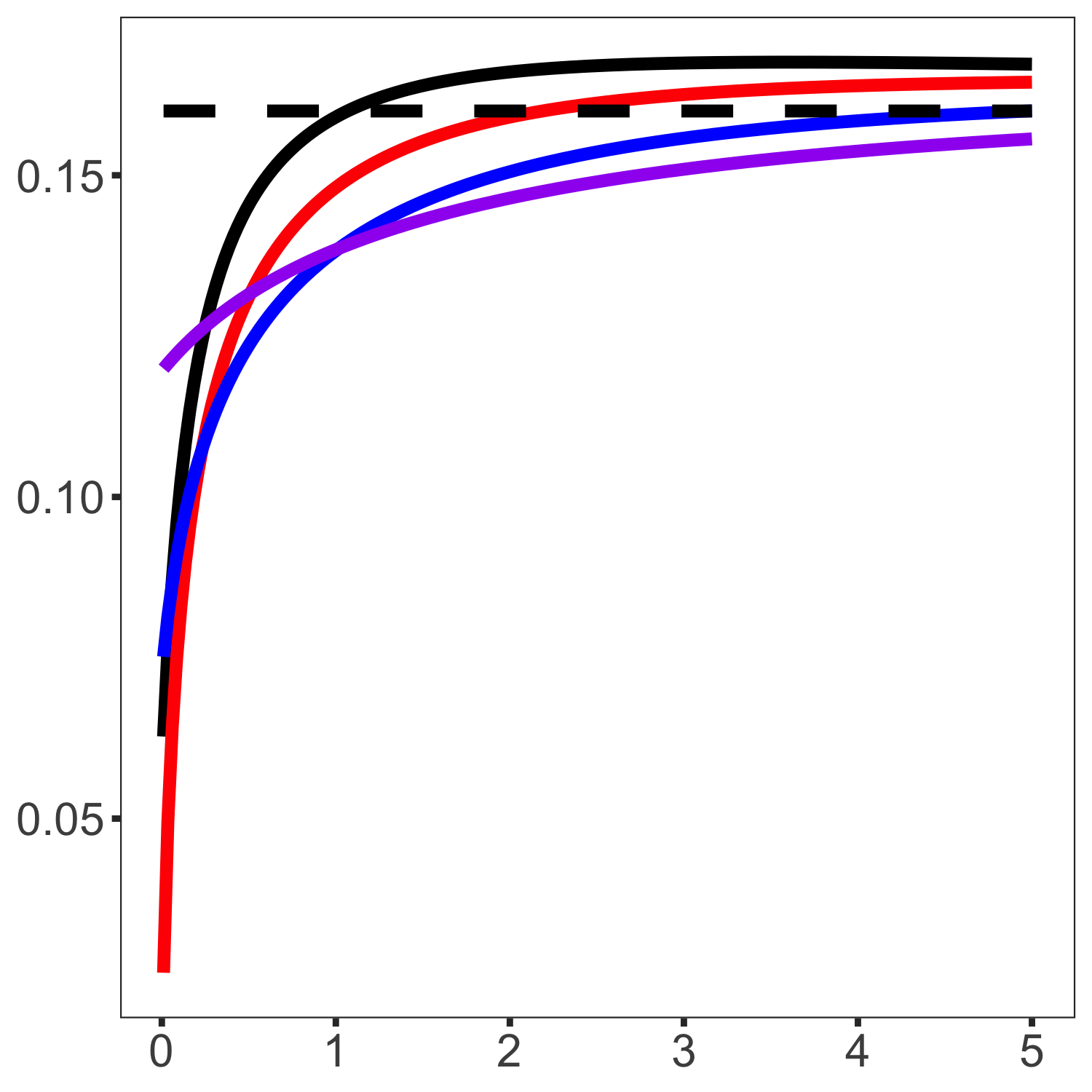}    
    \end{subfigure}
    \begin{subfigure}{0.24\textwidth}
    \includegraphics[width=\linewidth, height=0.6\linewidth]{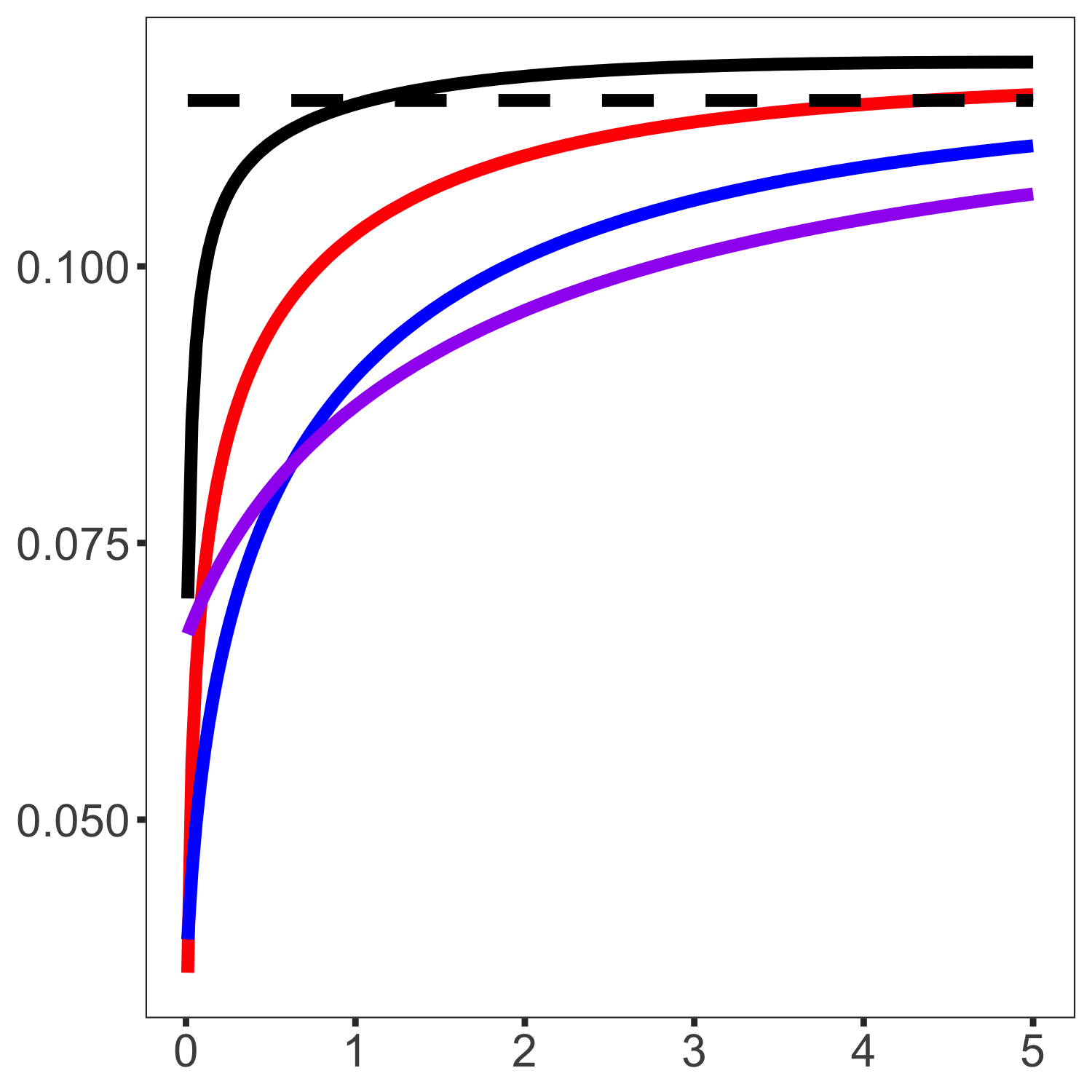}    
    \end{subfigure}
        \caption{$\SNR_p(\lambda)$ on $\lambda\in[0.01, 5]$ when $\gamma_1=2$. Columns: $\Sigma_p$ = Identity, Poly-Decay, AR-ACF, Point-Mix. Rows: $D \propto I_p$ (top), $D \propto \Sigma_p$ (bottom). Black/red/blue/purple: $\gamma_2 = 0.5$, $0.8$, $1.5$, $3$. Dashed: $\SNR_p(\infty)$ (independent of $\gamma_2$ and $\lambda$). }
    \label{fig:SNRvsLamda_Gamma1_2}
\end{figure}

\textcolor{black}{Figure~\ref{fig:SNRvsGamma2_Gamma1_05} shows the SNR as $\gamma_2$ increases when $\gamma_1=0.5$. The corresponding results for $\gamma_1=2$ are presented in Figure~S.3.1 
in the Supplementary Material. We observe that $\ell_H$ achieves comparable SNR when $\gamma_2$ is relatively small (e.g., $\gamma_2 \lesssim 0.2$), whereas the regularized tests exhibit superior performance when $\gamma_2$ is relatively large (e.g., $\gamma_2 \gtrsim 0.5$), demonstrating the advantage of spectral regularization in high dimension. When $\gamma_2$ is large, larger values of $\lambda$ tend to yield higher SNR. Overall, our results suggest that choosing $\lambda \approx 5\, p^{-1}\tr(\bW_2)$ yields a test that closely approximates the behavior of the $\lambda \to \infty$ limit under the investigated settings. For moderately large $\gamma_2$, however, tests with finite $\lambda$ may outperform the limiting case $\lambda \to \infty$.  The results highlight the benefit of flexibly selecting $\lambda$.}

\textcolor{black}{Figure~\ref{fig:SNRvsLamda_Gamma1_2} shows the SNR as a function of $\lambda$ when $\gamma_1=2$. The corresponding results for $\gamma_1=0.5$ are presented in Figure~S.3.2 
in the Supplementary Material. Overall, the results indicate that the SNR is not overly sensitive to $\lambda$ when it lies within a reasonable range (e.g., $\lambda \in [1,5]$) under the considered settings. When the spectrum of $\Sigma_p$ is highly dispersed, such as under the Point-Mix model, the optimal $\lambda$ tends to occur at a relatively small value. Except in extreme cases, the optimal choice of $\lambda$ is attained at a finite value. When $\Sigma_p = I_p$, the limiting regime $\lambda \to \infty$, which corresponds to replacing $\bW_2^{-1}$ with the identity matrix, outperforms the alternatives; in this setting, such a test appears at least intuitively the most natural choice since it utilizes the true covariance matrix $\Sigma_p=I_p$.}


\section{Data-driven selection of the regularization parameter}\label{sec:selection_lambda}

In this section, we discuss data-driven strategies for selecting $\lambda$, guided by the principles of Bayesian decision theory and minimaxity. Throughout, we focus on the \textbf{PA} framework of the alternatives under which the asymptotic behavior of the proposed test is primarily governed by $\SNR_p(\lambda, D)$. Notably, $\SNR_p(\lambda,D)$ depends on the \textbf{PA} model only through the covariance matrix $D$. Accordingly, we regard two \textbf{PA} priors as equivalent whenever they share the same $D$, and we identify a prior model with its associated matrix $D$.

\subsection{Bayes choice of $\lambda$}\label{subsec:Bayes_selection}
Suppose first that a prior model on the alternatives, as specified under \textbf{PA}, is given. Within the Bayesian decision-theoretic framework, we choose $\lambda$ by maximizing the power, or equivalently, maximizing $\SNR_p(\lambda,D)$. 
Implementing this strategy requires a consistent estimator of $\xi_p(\lambda,D)$. If the large matrix $D$ can be fully specified, we can estimate $\xi_p(\lambda,D)$ consistently by 
\begin{equation}\label{eq:xi_hat_D_specified}
\hat\xi_p(\lambda) = \frac{1}{p} \tr [(\bW_2 + \lambda I_p)^{-1} D ].
\end{equation}
The consistency of $\hat{\xi}_p(\lambda)$ is well known in the RMT literature and is therefore omitted. Such a setting typically arises when additional structural information about the data is available. For example, if it is reasonable to assume a stationary autocovariance structure, $D$ can be taken as a Toeplitz matrix with $(i,j)$-th entry $a^{|i-j|}$ for some $a>0$. However, when no additional information is available, specifying a high-dimensional positive definite matrix $D$ can be challenging. 

In the following, we take a different approach. Motivated by \citet{li2020high}, we consider a polynomial alternative framework, by which we mean the following model: $BC$ satisfies \textbf{PA} with $D$ given as a matrix polynomial of $\Sigma_p$, namely 
\begin{equation}\label{eq:polynomial_alternatives} 
D = \sum_{i=0}^r \pi_i \Sigma_p^i, 
\end{equation}
for coefficients $\pi_0$, $\pi_1$, $\dots$, $\pi_r$ such that $D$ is positive definite. In this sense, $D$ is specified by the coefficients up to the unknown covariance $\Sigma_p$. Since any arbitrary smooth function can be approximated by polynomials, this formulation is quite useful and fairly general. Moreover, the choice of $D$ as a matrix polynomial in $\Sigma_p$ also allows for an easier interpretation of the signal structure under the alternative. Note that unless $\Sigma_p = I_p$, such a prior implies a certain distribution of the rank-one alternative in the spectral coordinate system. Specifically, larger values of $\pi_i$ for high powers $i$ imply that the signal has a larger contribution from the leading eigenvectors of $\Sigma_p$. 

Under the polynomial alternatives Eq. \eqref{eq:polynomial_alternatives}, $\xi_p(\lambda,D)$ can be consistently estimated by 
\begin{equation}\label{eq:xi_hat_poly_alter}
\hat{\xi}_p(\lambda) = \sum_{i=0}^r \pi_i \Upsilon_i(\lambda)
\end{equation}
with $\Upsilon_i(\lambda)$ satisfying the recursive formula
\[ \Upsilon_{i+1}(\lambda)  = \frac{1}{\lambda \hat{\varphi}_p(-\lambda) }  \Big\{ \frakM_i  -\lambda \Upsilon_i(\lambda)\Big\}, \quad i=0,1,2, \dots, \]
and $\Upsilon_0(\lambda) = (n_2/p) (\hat{\varphi}_p(-\lambda) - (1-p/n_2)/\lambda)$, where $\frakM_i$ is a consistent estimator of the $i$th population spectral moment $\int \tau^i dF^{\Sigma_p}(\tau)$, given in Lemma 1 of \cite{bai2010estimation}. The above formula was derived in \citet{li2020high} using Lemma~3 of \citet{ledoit2011eigenvectors}, where consistency of the estimator $\hat{\xi}_p(\lambda)$ is also established. 
In this paper, we restrict to the case $r=2$ which requires only $\frakM_0 = 1$ and $\frakM_1 = p^{-1}\tr(\bW_2)$. There are several considerations that guides this choice of $r$. First, for $r=2$, all quantities involved in estimating $\hat{\xi}_p(\lambda)$ can be computed explicitly without requiring knowledge of higher order moments of the observations. Also, the corresponding estimating equations are more stable as they do not involve higher-order spectral moments. Secondly, the choice of $r=2$ yields a significant, yet nontrivial, concentration of the prior covariance in the directions of the leading eigenvectors of $\Sigma_p$. Finally, the choice $r=2$  leads to a quadratic polynomial, which accommodates both convex and concave shapes. For extensions to $r>2$, we refer to Lemma 1 of \citet{bai2010estimation} for the construction of $\frakM_i$ for $i\geq 2$.  

Given a \textbf{PA} prior such that $\xi_p(\lambda,D)$ can be consistently estimated, the resulting data-driven Bayes choice of the regularization parameter $\hat{\lambda}_{\rm B}$ is defined as
\[
\hat{\lambda}_{\mathrm{B}} = \arg\max_{\lambda \in \mathfrak{X} } \frac{\hat{\xi}_p(\lambda)}{\hat{\Theta}_{2}(\lambda)},
\]
where $\mathfrak{X}\subset \mathbb{R}_+ $ is a user-specified parameter space.

\textcolor{black}{
A potential concern arises from the “double-dipping” issue, since the same dataset is used both to select $\lambda$ and to perform the test. Although Theorem~\ref{thm:main} shows that, for a fixed $\lambda$, the proposed test statistic follows an asymptotic Tracy–Widom distribution after appropriate scaling, the data-driven selection of $\lambda$ introduces additional variability that may affect the limiting behavior. To mitigate the fluctuation of $\hat{\lambda}_{\rm B}$, we propose restricting the search to a finite discrete parameter set $\mathfrak{X}$ rather than a dense interval. This discretization is not expected to incur a substantial loss of efficiency, owing to the smooth dependence of the test statistic on $\lambda$. Indeed, the numerical results in Section~\ref{subsec:effect_lambda} indicate that $\SNR_p(\lambda, D)$ is not overly sensitive to $\lambda$ under the considered settings. A practical suggestion for $\mathfrak{X}$ is a set of 10 evenly spaced points ranging from $(0.1\gamma_2\wedge 1) p^{-1}\tr(\bW_2)$ to $5p^{-1}\tr(\bW_2)$.}
\textcolor{black}{
\begin{lemma}\label{lemma:converge_hat_lambda}
Suppose that Conditions~\ref{enum:high_dimensional_regime}--\ref{enum:regular_edge} hold, and let $\mathfrak{X}\subset\mathbb{R}_+$ be any finite parameter set. Assume that under \textbf{PA}, the matrix $D$ is such that there exist $\lambda_\infty\in\mathfrak{X}$ and $\varepsilon>0$ satisfying $\max_{\lambda\in\mathfrak{X}\setminus\{\lambda_\infty\}}
\SNR_p(\lambda,D)
<
\SNR_p(\lambda_\infty,D)-\varepsilon$, for all sufficiently large $p$. Suppose further that $\hat{\xi}_p(\lambda) - \xi_p(\lambda, D) \stackrel{P}{\longrightarrow}0$ and $\hat{\Theta}_2(\lambda) -\Theta_{2p}(\lambda) \stackrel{P}{\longrightarrow}0$ for any $\lambda\in\mathfrak{X}$. Then $\mathbb{P}\!(\hat{\lambda}_{\mathrm{B}}=\lambda_\infty)\to1$, as $p\to\infty$. Consequently, the conclusions of Theorem~\ref{thm:main} remain valid when $\lambda$ is replaced by $\hat{\lambda}_{\rm B}$.
\end{lemma}}

\subsection{Minimax selection of $\lambda$}\label{subsec:minimax_selection}
\textcolor{black}{In practice, rather than a particular choice of the prior, we may consider a collection of such priors. In such situations, a strategy is needed to synthesize the information they provide and enhance the robustness against prior misspecification. In this subsection, we consider a procedure for selecting the regularization parameter based on the principle of minimaxity.}

\textcolor{black}{Consider a family $\{D_\theta:\theta\in \mathfrak{D}\}$ indexed by $\theta$ in a compact set $\mathfrak{D}$, satisfying the requirements in \textbf{PA}, such that
\begin{equation}\label{eq:constrainst_on_B}
p^{-1}\tr(D_\theta)=\calK,\qquad \text{for all }\theta\in \mathfrak{D},
\end{equation}
for some constant $\calK>0$. This constraint ensures that the expected total signal strength, $\mE n_1^{-1} \tr[BC (C^T (XX^T)^{-1}C)^{-1}C^TB^T] = p^{s-1} \tr(D_\theta)$, remains the same across all priors identified by $D_\theta$.}

\textcolor{black}{
Under \textbf{PA}, we consider the asymptotic type-II error rate and define the risk function 
\[ \text{\rm Risk}_p(\lambda, D_\theta) = F_{\TW} \Big( \TW_1(1-\alpha) - p^{s} \SNR_p(\lambda, D_\theta)\Big), \]
where $F_{\TW}$ denotes the distribution function of $\TW_1$ and $\alpha \in (0,1)$ is the level of significance. We restrict attention to the group of all decision rules $\mathbb{I}\big(\tilde{\ell}(\lambda) > \TW_1(1-\alpha)\big)$ when $\lambda \in \mathfrak{X}$. We say that the decision rule associated with $\lambda_*$ is \emph{asymptotically minimax} with respect to the prior family identified by $\{D_\theta$, $\theta \in \mathfrak{D}\}$ if 
\[\lambda_* = \arg\min_{\lambda\in \mathfrak{X}} \sup_{\theta\in \mathfrak{D}} \text{\rm Risk}_p(\lambda, D_\theta) = \arg\max_{\lambda\in\mathfrak{X}} \inf_{\theta\in\mathfrak{D}} \frac{ \xi_p(\lambda, D_\theta)}{\Theta_{2p}(\lambda)}.\]
The parameter $\lambda_*$ is referred to as the \emph{minimax selection} of $\lambda$. In practice, this procedure is implemented by replacing the SNR with its estimator, and the resulting estimated minimax choice is denoted by $\hat{\lambda}_*$.
\begin{lemma}\label{lemma:converge_lambda_*}
    Suppose that Conditions~\ref{enum:high_dimensional_regime}--\ref{enum:regular_edge} hold, and let $\mathfrak{X}\subset\mathbb{R}_+$ be any finite parameter set. Assume that the family $\{ D_\theta, \theta\in\mathfrak{D}\}$ satisfies:
    \begin{itemize}
        \item[(i)] There exist $\lambda_\infty \in \mathfrak{X}$ and $\varepsilon>0$ such that 
        \[\max_{\lambda \in \mathfrak{X}\setminus \{\lambda_\infty\}} \inf_{\theta \in \mathfrak{D}} \SNR_p(\lambda, D_{\theta}) < \inf_{\theta \in\mathfrak{D}} \SNR_p(\lambda_\infty, D_\theta) -\varepsilon.\]
        \item[(ii)] $\sup_{\theta\in\mathfrak{D}} |\hat{\xi}_p(\lambda, D_\theta) - \xi_p(\lambda, D_\theta)| \stackrel{P}{\longrightarrow}0$, for any $\lambda \in \mathfrak{X}$. 
    \end{itemize}
    Additionally, assume that $\hat{\Theta}_2(\lambda)-\Theta_{2p}(\lambda) \stackrel{P}{\longrightarrow}0$, for any $\lambda\in\mathfrak{X}$. Then, $\mP(\hat{\lambda}_*  = \lambda_*) \to 1$.  Consequently, the conclusions of Theorem~\ref{thm:main} remain valid when $\lambda$ is replaced by $\hat{\lambda}_*$.
\end{lemma}
Under the two scenarios described in Eq. \eqref{eq:xi_hat_D_specified} and \eqref{eq:xi_hat_poly_alter}, Condition (ii) is satisfied if $\sup_{\theta\in\mathfrak{D}}\|D_\theta\|_2<\infty$.
}


\section{Simulation studies}\label{sec:simulation}

In this section, we evaluate the performance of the proposed ridge-regularized F-matrix framework and the associated estimation methods through Monte Carlo experiments. 

Due to space limitations, detailed configurations and full descriptions of all competing methods are deferred to Section~S.4 
of the Supplementary Material. Briefly, we consider four spectral structures for $\Sigma_p$: the identity model (Identity), a polynomial decaying spectrum (Poly-Decay), an AR(1) auto-correlation structure (AR-ACF), and a factor model (Factor). We evaluate the proposed procedures with different choices of $\lambda$ and compare them with representative methods from the literature. Specifically, we include the test of \citet{han2016tracy}, corresponding to $\lambda=0$; the ridge-regularized likelihood ratio test of \citet{li2020high}, denoted Ridge-LRT; and a projection-based test, denoted Proj-LRT, in which the data are randomly projected to a lower-dimensional space and a classical likelihood ratio test is applied to the projected data.

\subsection{Estimation precision}\label{subsec:estimation_precision}
Figure~\ref{fig:example_s_fun} displays the estimated curves $\hat{s}(x)$, $\hat{s}'(x)$, and $\hat{s}''(x)$ under a representative setting. Two additional examples are provided in Section~S.3.1 of the Supplementary Material; further results are omitted, as they exhibit similar patterns. The estimation accuracy of $\hat{\Theta}_1$ and $\hat{\Theta}_2$ is also reported in Tables~S.4.1 and S.4.2 
of the Supplementary Material.

The results indicate that the overall estimation accuracy is high, although it deteriorates as $x$ approaches $\rho_\lambda$ and as $\hat\gamma_2$ increases. This suggests that the estimators $\hat{\Theta}_1$ and $\hat{\Theta}_2$ remain reliable unless $\beta$ is close to $\rho_\lambda$ and $\hat\gamma_2$ is large. Since $\beta$ satisfies $\beta^2 s'_{p\lambda}(\beta)=1/\hat\gamma_1$, it approaches $\rho_\lambda$ only when $\hat\gamma_1$ is small. Hence, estimation becomes more challenging in regimes where $n_1 \gg p \gg n_2$, which in practice corresponds to situations in which a large number of hypotheses ($n_1$) are tested with a relatively small effective sample size ($n_2$). 
Figure~S.4.3 
in the Supplementary Material shows the empirical density of the estimation error $\hat{\Theta}_1-\Theta_1$ under representative settings; the distribution appears approximately bell-shaped. 
\begin{figure}[t]
\centering
\includegraphics[width=\textwidth]{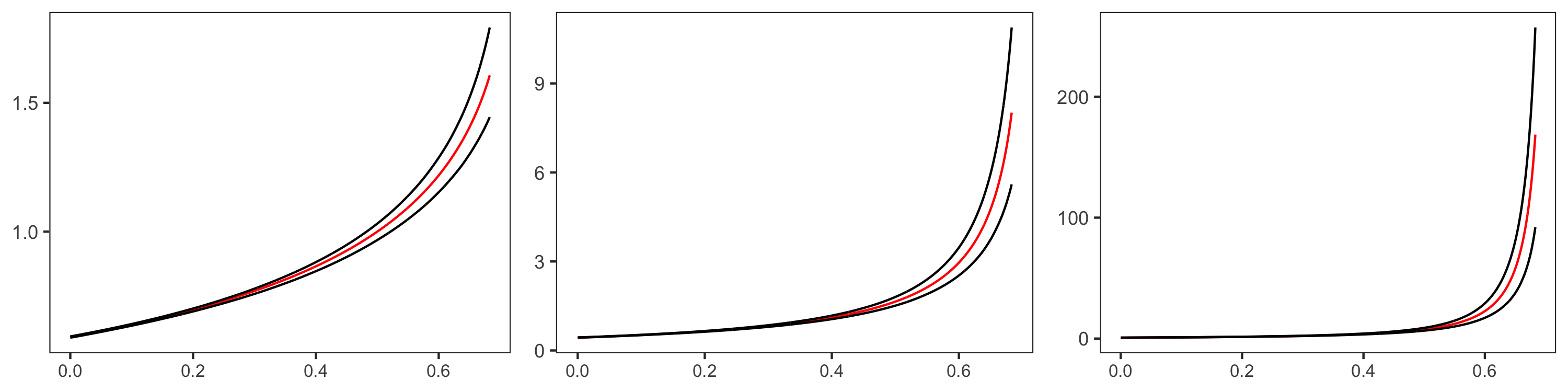}
\caption{The estimated $s(x)$, $s'(x)$ and $s''(x)$ (from left to right) when $\Sigma_p$ is {AR-ACF}, $\hat{\gamma}_2 = 1$, $\lambda/\hat{\gamma}_2 = 1$. Red: the true functions; Black: 5\% and 95\% pointwise percentile bands of the estimated functions.}
\label{fig:example_s_fun}
\end{figure}

\begin{table}[htbp]
    \centering
    \resizebox{\textwidth}{!}{
    \begin{tabular}{clc ccc ccc ccc ccc}
        \toprule
        \multicolumn{3}{c}{Size$\times100\%$, $n_2=500$} & \multicolumn{3}{c}{$\hat\gamma_2 = 0.3$} & \multicolumn{3}{c}{$\hat\gamma_2 = 0.5$} & \multicolumn{3}{c}{$\hat\gamma_2 = 0.9$} & \multicolumn{3}{c}{$\hat\gamma_2 = 2$} \\
        \cmidrule(lr){4-6} \cmidrule(lr){7-9} \cmidrule(lr){10-12} \cmidrule(lr){13-15}
        $\Sigma$ & $\lambda$ & $n_1 =$ & $ 50 $ & $ 100$ & $ 250$ & $ 50$ & $100$ & $250$ & $50$ & $ 100$ & $250$ & $50$ & $100$ & $250$ \\
        \midrule
        \multirow{5}{*}{Poly-Decay} 
        & $0$ && 4.12 & 4.35 & 4.53 & 4.01 & 4.18 & 4.95 & 4.76 & 4.83 & 5.08 & -- & -- & -- \\
        & $0.5$ &  & 3.80 & 3.92 & 4.30 & 4.00 & 4.40 & 4.60 & 4.58 & 4.08 & 4.85 & 4.83 & 4.72 & 4.85 \\
         & $1$ &  & 3.35 & 3.90 & 4.10 & 3.98 & 4.23 & 4.00 & 4.42 & 4.15 & 4.58 & 4.32 & 4.45 & 4.28 \\
         & $\hat{\lambda}_{I_p}$ &  & 3.45 & 3.88 & 3.98 & 4.15 & 3.92 & 4.00 & 4.23 & 4.08 & 3.85 & 3.62 & 4.45 & 3.77 \\
         & $\hat{\lambda}_{\Sigma_p}$ &  & 3.62 & 4.10 & 4.05 & 3.95 & 3.85 & 4.08 & 4.12 & 3.88 & 3.67 & 3.62 & 4.45 & 3.77 \\
         & $\hat{\lambda}_*$ &  & 3.43 & 4.12 & 4.05 & 4.65 & 4.88 & 5.62 & 5.15 & 5.65 & 6.35 & 5.85 & 5.20 & 5.25 \\
        \midrule
        \multirow{5}{*}{AR-ACF}
        & $0$ && 4.12 & 4.35 & 4.53 & 4.01 & 4.18 & 4.95 & 4.76 & 4.83 & 5.08 & -- & -- & -- \\
        & $0.5$ &  & 3.90 & 4.62 & 4.50 & 3.77 & 4.30 & 4.50 & 4.42 & 4.60 & 4.58 & 5.17 & 4.72 & 4.45 \\
         & $1$ &  & 3.65 & 4.30 & 4.10 & 3.65 & 4.15 & 4.30 & 4.30 & 4.35 & 4.17 & 4.83 & 4.30 & 4.15 \\
         & $\hat{\lambda}_{I_p}$ &  & 3.52 & 4.35 & 4.10 & 3.70 & 3.90 & 4.17 & 4.10 & 3.80 & 3.98 & 4.25 & 3.98 & 4.00 \\
         & $\hat{\lambda}_{\Sigma_p}$ &  & 3.67 & 4.10 & 4.03 & 3.70 & 3.88 & 4.32 & 3.95 & 3.65 & 3.75 & 4.25 & 3.98 & 3.98 \\
         & $\hat{\lambda}_*$ &  & 4.28 & 4.92 & 4.45 & 4.40 & 4.67 & 5.80 & 5.20 & 5.33 & 5.92 & 5.73 & 5.15 & 5.03 \\
        \midrule
        \multirow{5}{*}{Factor} 
        & $0$ && 4.12 & 4.35 & 4.53 & 4.01 & 4.18 & 4.95 & 4.76 & 4.83 & 5.08 & -- & -- & -- \\
        & $0.5$ &  & 3.88 & 4.15 & 4.45 & 4.52 & 4.15 & 3.72 & 4.52 & 4.38 & 4.58 & 4.83 & 5.40 & 4.92 \\
         & $1$ &  & 3.50 & 3.57 & 4.32 & 4.70 & 4.10 & 3.75 & 4.25 & 4.20 & 4.23 & 4.80 & 4.72 & 4.92 \\
         & $\hat{\lambda}_{I_p}$ &  & 3.80 & 3.95 & 4.40 & 4.67 & 4.05 & 4.20 & 3.90 & 4.12 & 4.30 & 4.25 & 4.67 & 4.72 \\
         & $\hat{\lambda}_{\Sigma_p}$ &  & 3.38 & 3.57 & 4.40 & 4.50 & 4.15 & 3.90 & 4.05 & 3.80 & 4.32 & 4.17 & 4.40 & 4.78 \\
         & $\hat{\lambda}_*$ &  & 4.75 & 5.40 & 4.90 & 4.65 & 4.78 & 5.03 & 5.05 & 5.60 & 5.88 & 5.53 & 6.17 & 5.05 \\
        \midrule
        \multirow{5}{*}{Identity} 
        & $0$ && 4.12 & 4.35 & 4.53 & 4.01 & 4.18 & 4.95 & 4.76 & 4.83 & 5.08 & -- & -- & -- \\
        & $0.5$ &  & 4.58 & 4.32 & 4.92 & 4.20 & 4.67 & 5.17 & 4.52 & 4.17 & 4.78 & 4.72 & 5.03 & 4.65 \\
         & $1$ &  & 4.30 & 4.20 & 4.83 & 3.88 & 4.38 & 4.52 & 4.30 & 4.00 & 4.58 & 4.60 & 4.88 & 4.23 \\
         & $\hat{\lambda}_{I_p}$ &  & 4.28 & 4.20 & 4.17 & 3.60 & 4.50 & 4.05 & 3.98 & 3.82 & 4.17 & 4.03 & 3.98 & 3.08 \\
         & $\hat{\lambda}_{\Sigma_p}$ &  & 4.28 & 4.20 & 4.17 & 3.60 & 4.50 & 4.05 & 3.98 & 3.82 & 4.17 & 4.03 & 3.98 & 3.08 \\
         & $\hat{\lambda}_*$ &  & 4.05 & 4.47 & 4.83 & 4.83 & 5.60 & 5.88 & 4.92 & 5.53 & 5.40 & 4.98 & 5.53 & 5.08 \\
        \bottomrule
    \end{tabular}
    }
    \caption{Empirical sizes of empirically normalized $\ell_{\max}(\bF_\lambda)$ at asymptotic level $5\%$ under different settings.}
    \label{table:empirical_sizes}
\end{table}

\subsection{Empirical null distribution}\label{subsec:empirical_null_distribution}
We examine the empirical distribution of $\ell_{\max}(\bF_\lambda)$ under the null hypothesis. We consider five choices of the regularization parameter $\lambda$: two fixed values, $\lambda=0.5$ and $\lambda=1$, and three data-driven values, namely the Bayes selections corresponding to $D\propto I_p$ and $D\propto \Sigma_p$, and the minimax selection when the prior family consists of all linear combinations of $I_p$ and $\Sigma_p$ with nonnegative coefficients satisfying \eqref{eq:constrainst_on_B}. See Section S.4.1 
of the Supplementary Material for details. The latter three are denoted by $\hat{\lambda}_{I_p}$, $\hat{\lambda}_{\Sigma_p}$, and $\hat{\lambda}_*$. In addition, we report the test size of the procedure in \citet{han2016tracy}, when $p<n_1+n_2$, corresponding to $\lambda=0$. Table~\ref{table:empirical_sizes} reports the empirical sizes at the asymptotic $5\%$ level. The results indicate that empirical sizes are consistently close to the nominal level, generally falling within the range of $3\%\text{--}6\%$, even with data-driven choices of $\lambda$. 

In Section S.4.3 
of the Supplementary Material, we report the empirical density of $\ell_{\max}(\bF_\lambda)$. The results show that the empirical distributions closely follow the Tracy–Widom law of type~1, both when normalized using the true values $\Theta_{1p}$ and $\Theta_{2p}$ and when using their estimators, thereby supporting the validity of the asymptotic theory and the accuracy of the proposed estimators. This pattern is consistent across all five choices of $\lambda$. 

\begin{figure}[htbp]
    \centering
    \begin{subfigure}[t]{0.23\textwidth}
        \centering
         \includegraphics[width=\linewidth, height=0.7\linewidth]{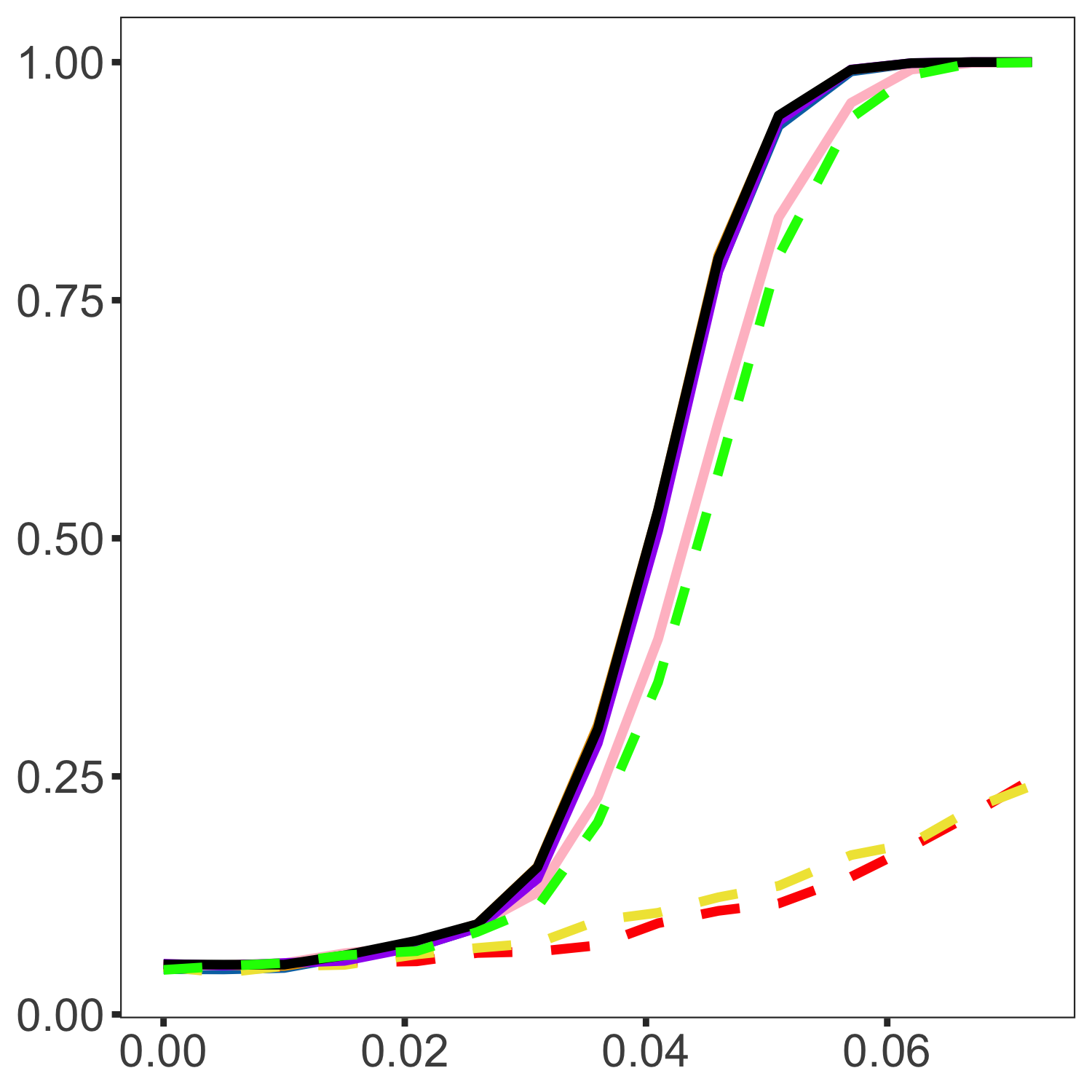}
    \end{subfigure}%
    \begin{subfigure}[t]{0.23\textwidth}
        \centering
        \includegraphics[width=\linewidth, height=0.7\linewidth]{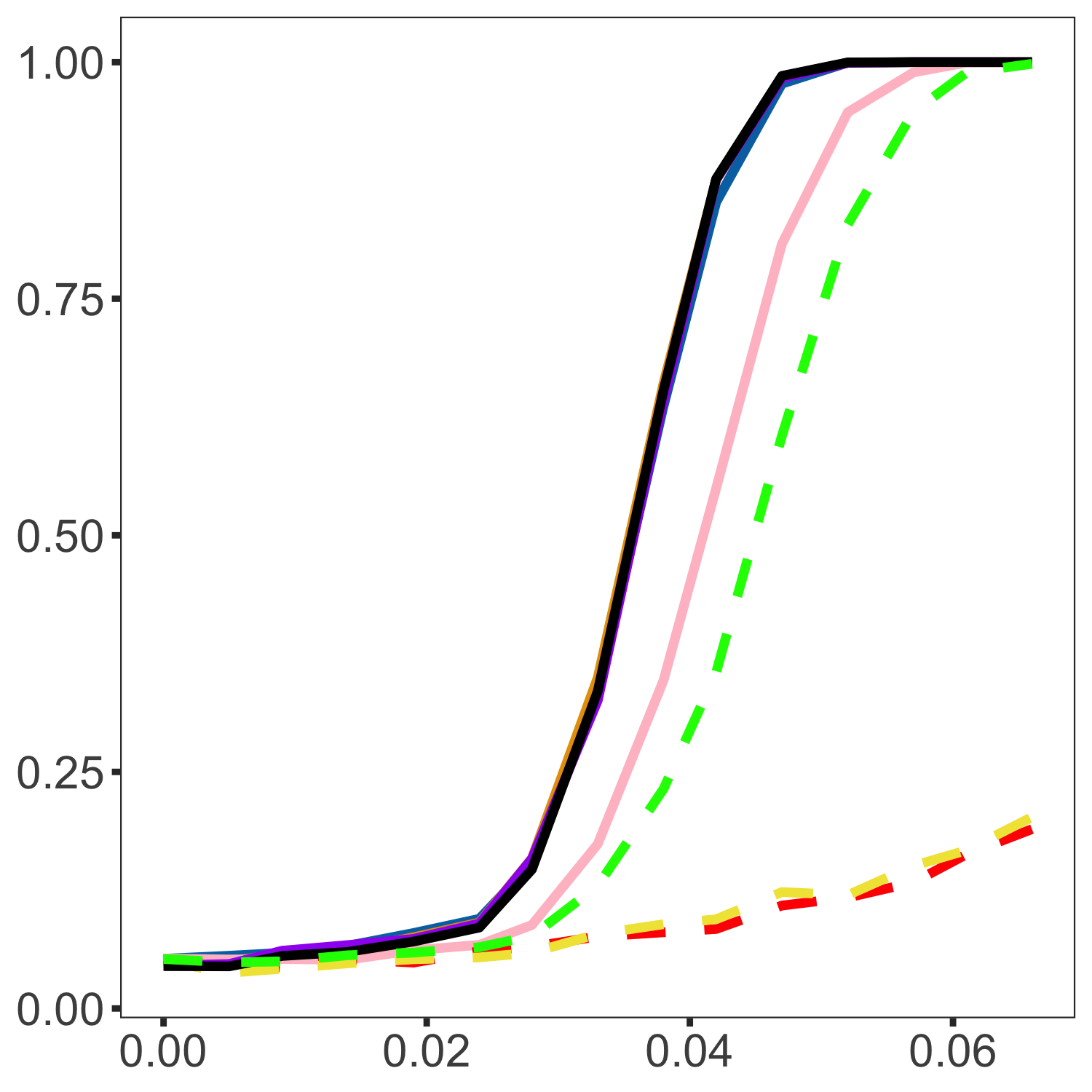}
    \end{subfigure}
     \begin{subfigure}[t]{0.23\textwidth}
        \centering
        \includegraphics[width=\linewidth, height=0.7\linewidth]{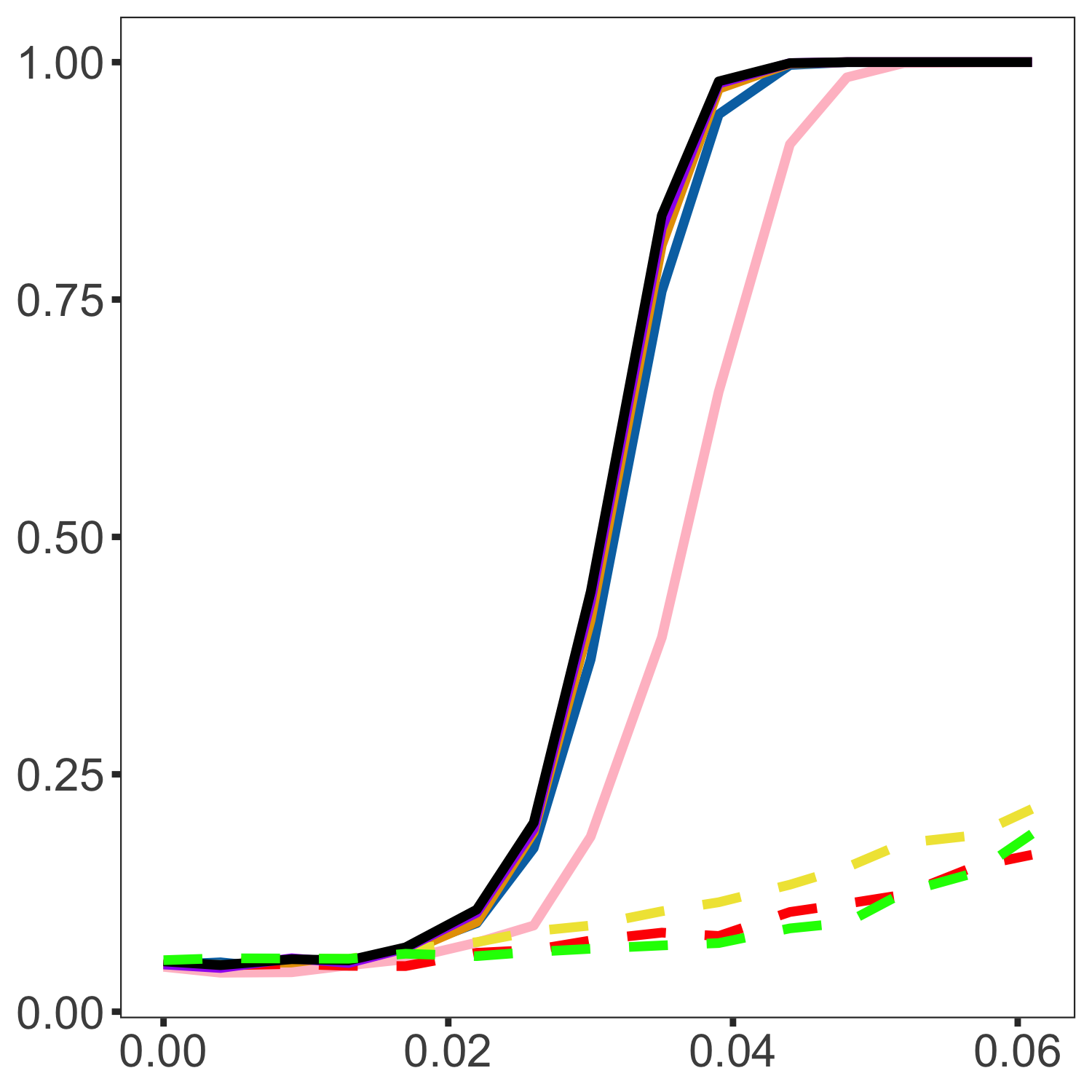}
    \end{subfigure}
    \begin{subfigure}[t]{0.23\textwidth}
        \centering
        \includegraphics[width=\linewidth, height=0.7\linewidth]{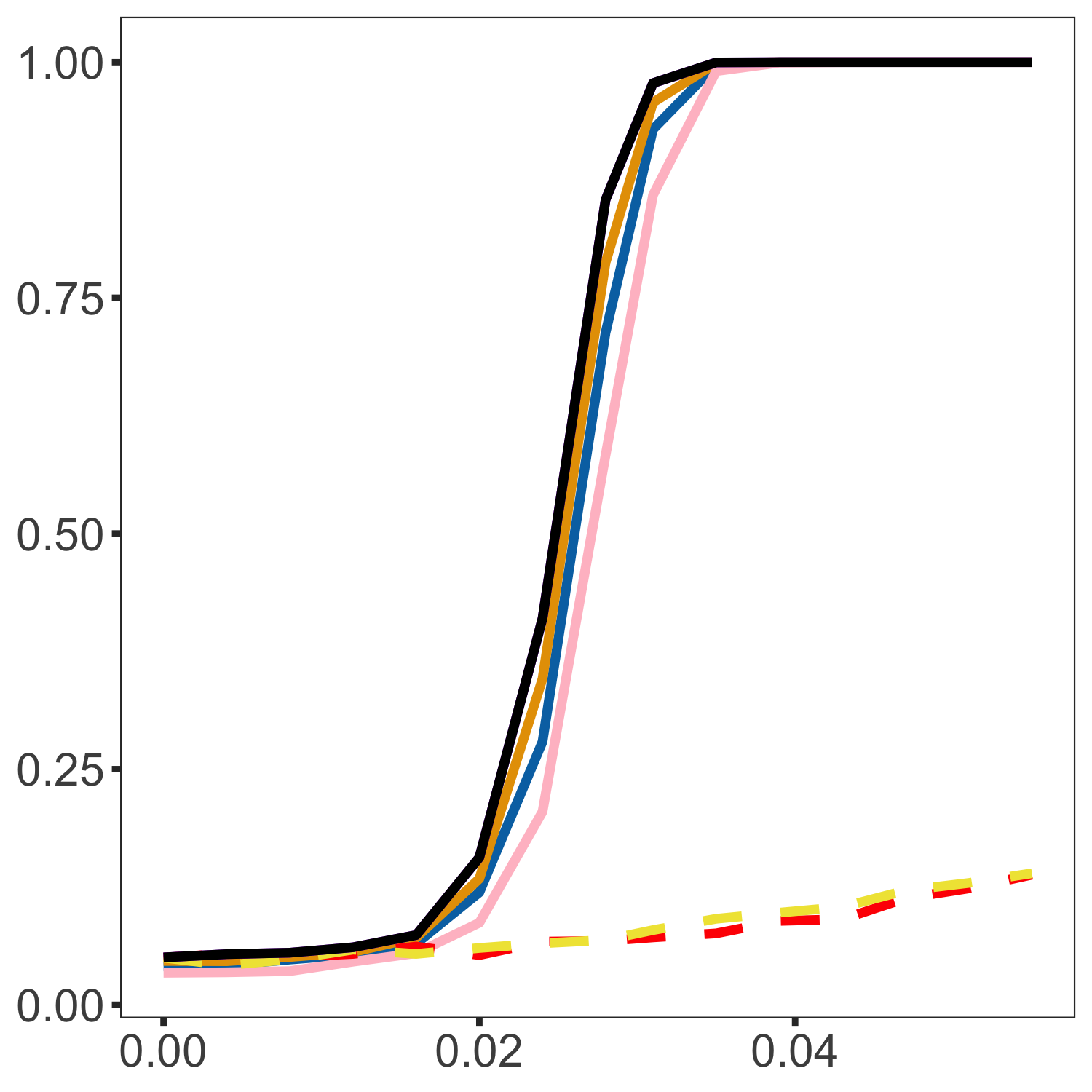}
    \end{subfigure}
    \vfill
    \begin{subfigure}[t]{0.23\textwidth}
        \centering
         \includegraphics[width=\linewidth, height=0.7\linewidth]{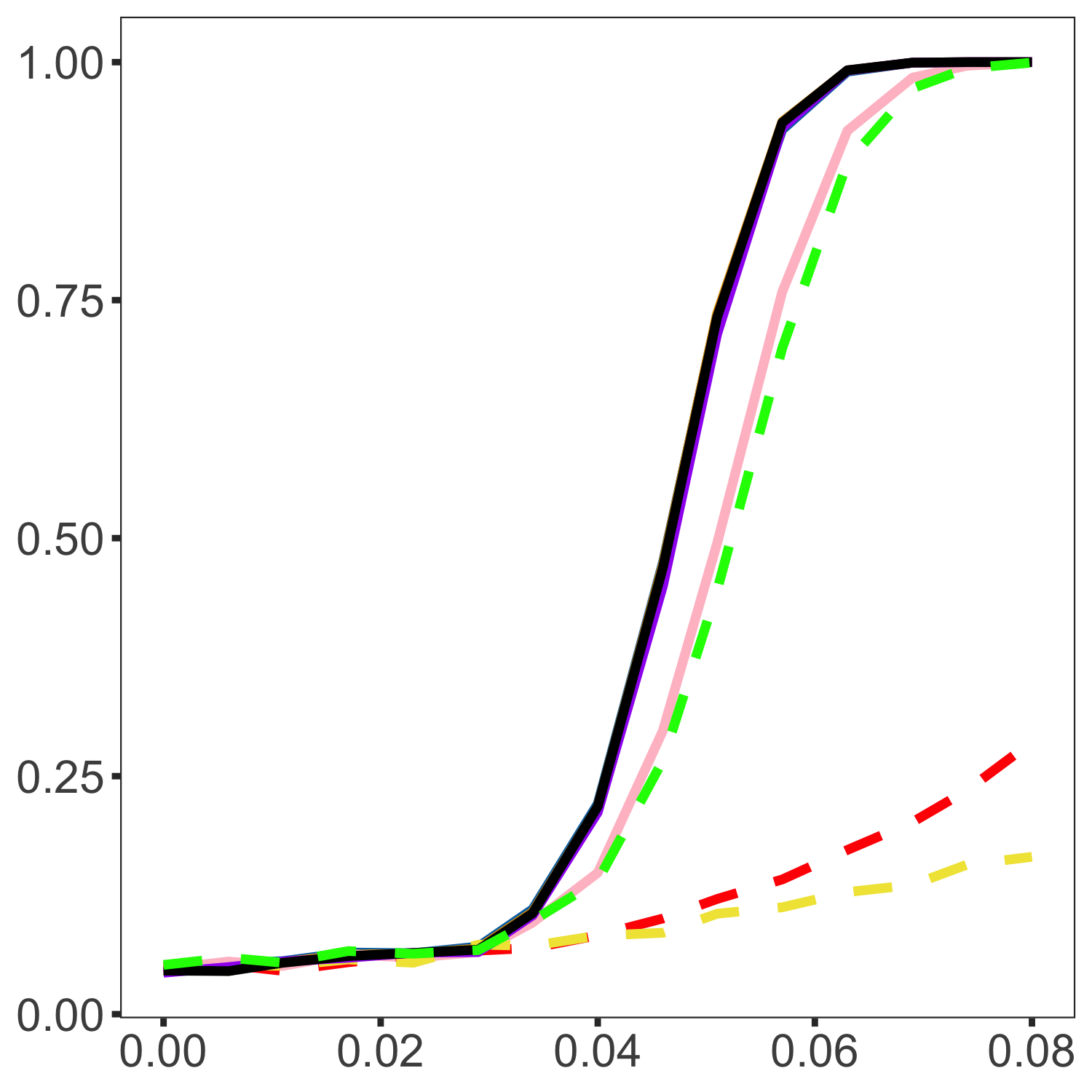}
    \end{subfigure}%
    \begin{subfigure}[t]{0.23\textwidth}
        \centering
        \includegraphics[width=\linewidth, height=0.7\linewidth]{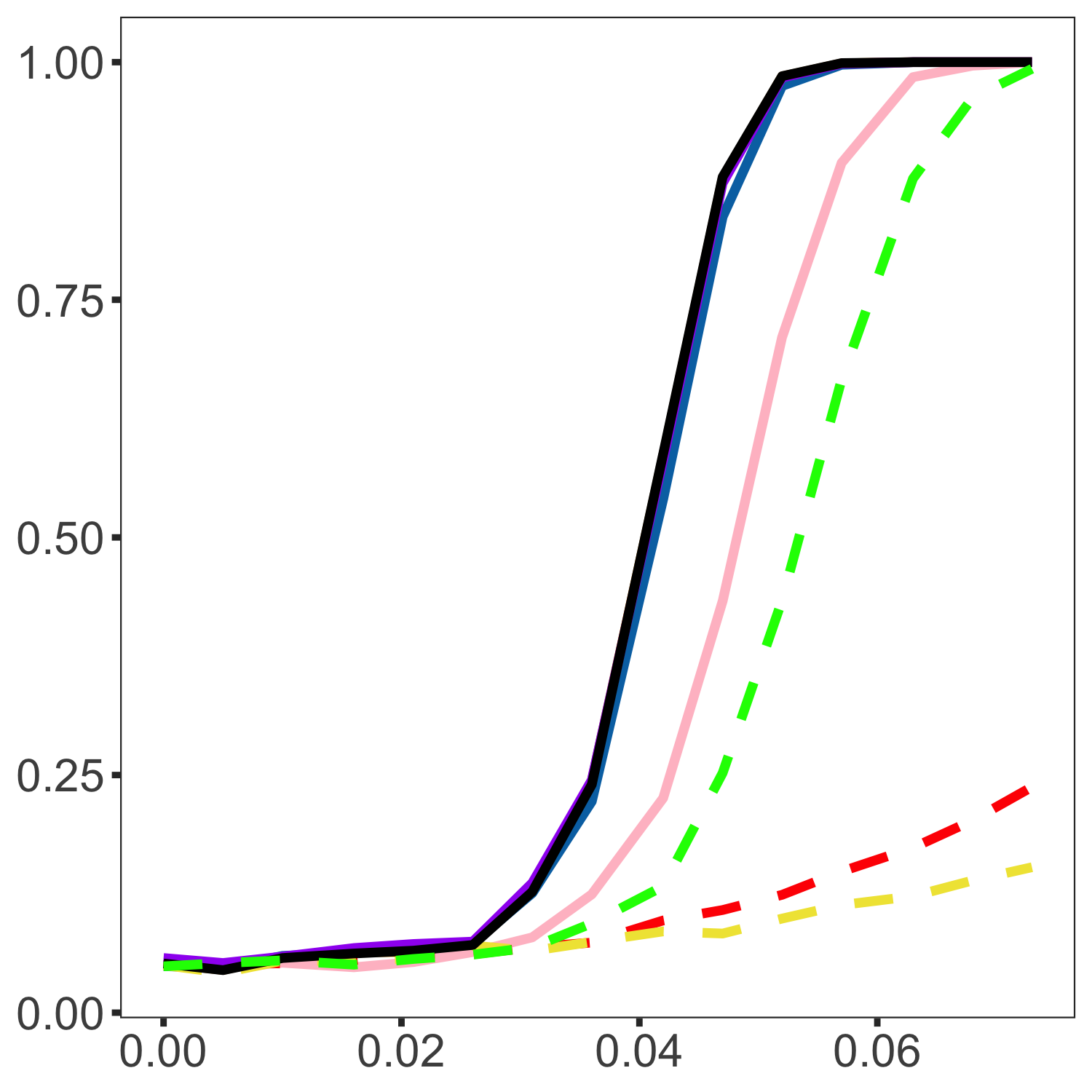}
    \end{subfigure}
     \begin{subfigure}[t]{0.23\textwidth}
        \centering
        \includegraphics[width=\linewidth, height=0.7\linewidth]{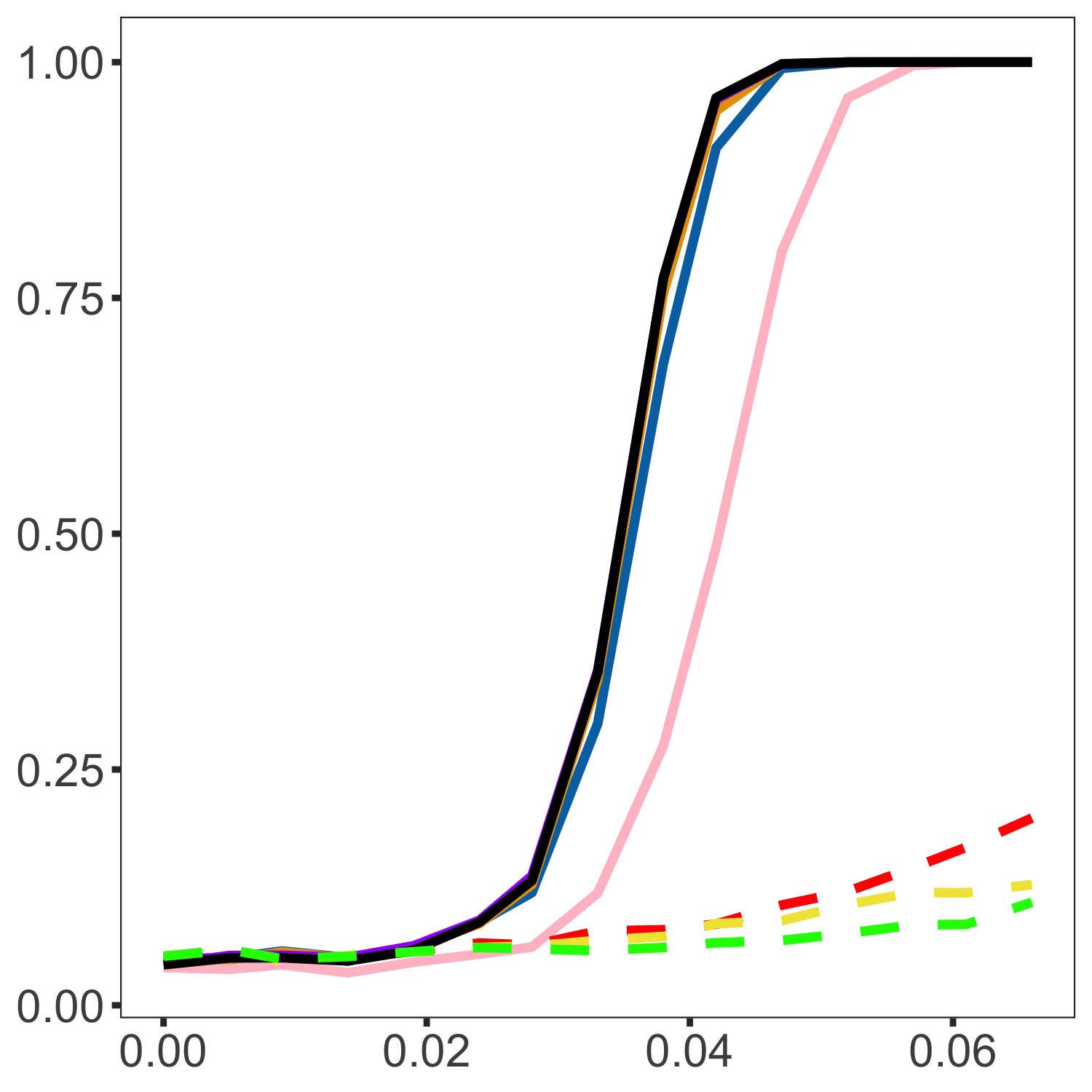}
    \end{subfigure}
    \begin{subfigure}[t]{0.23\textwidth}
        \centering
        \includegraphics[width=\linewidth, height=0.7\linewidth]{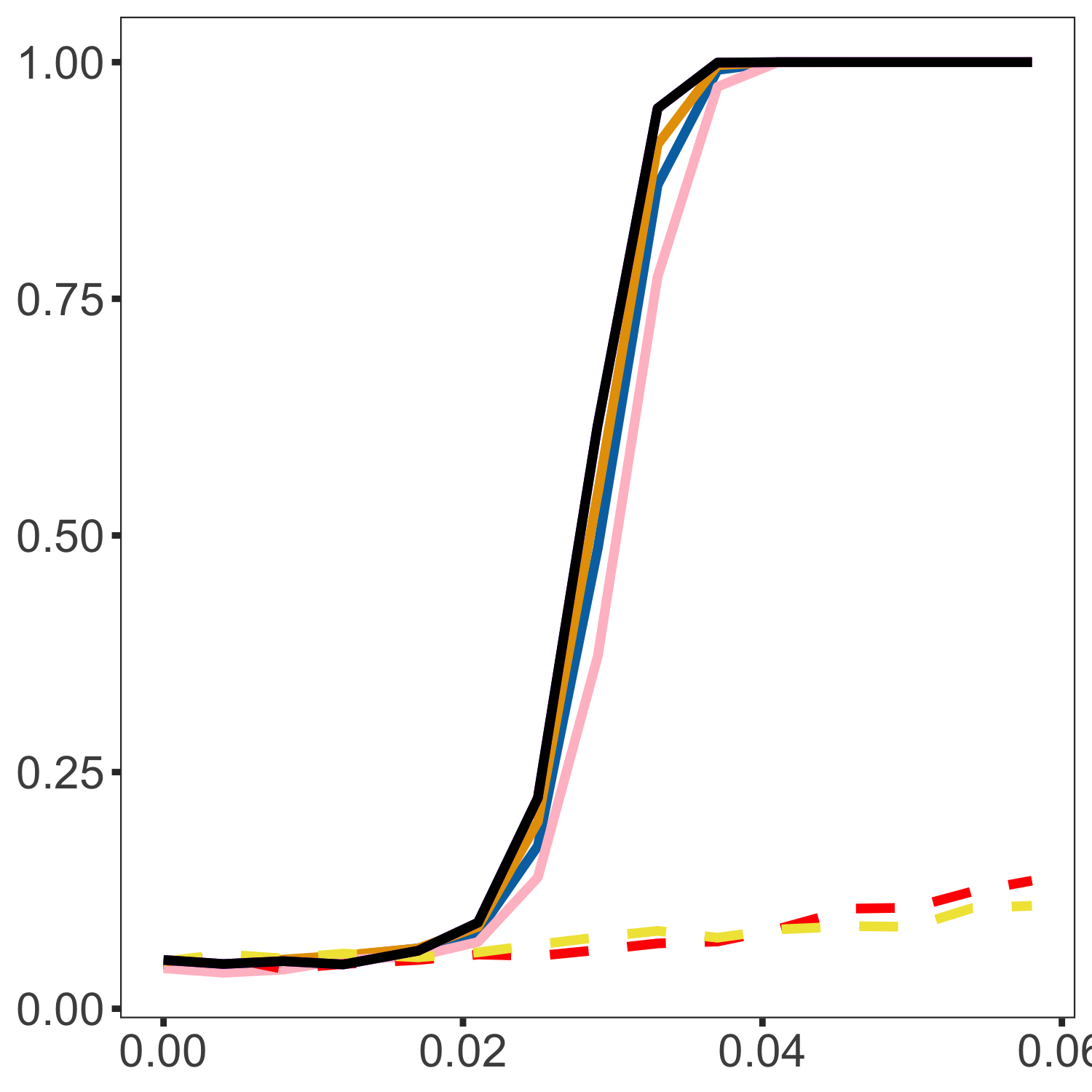}
    \end{subfigure}
    \vfill
    \begin{subfigure}[t]{0.23\textwidth}
        \centering
         \includegraphics[width=\linewidth, height=0.7\linewidth]{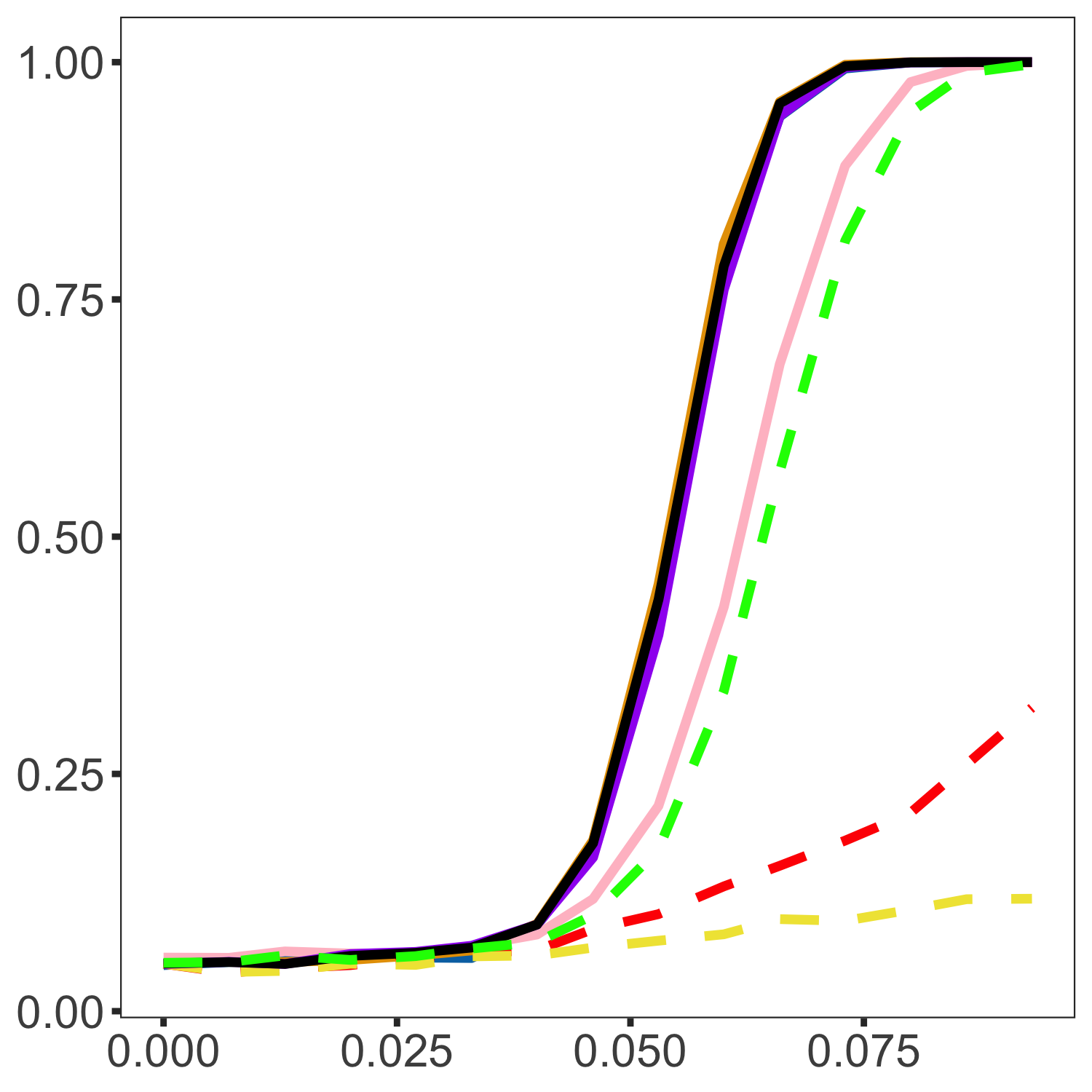}
    \end{subfigure}%
    \begin{subfigure}[t]{0.23\textwidth}
        \centering
        \includegraphics[width=\linewidth, height=0.7\linewidth]{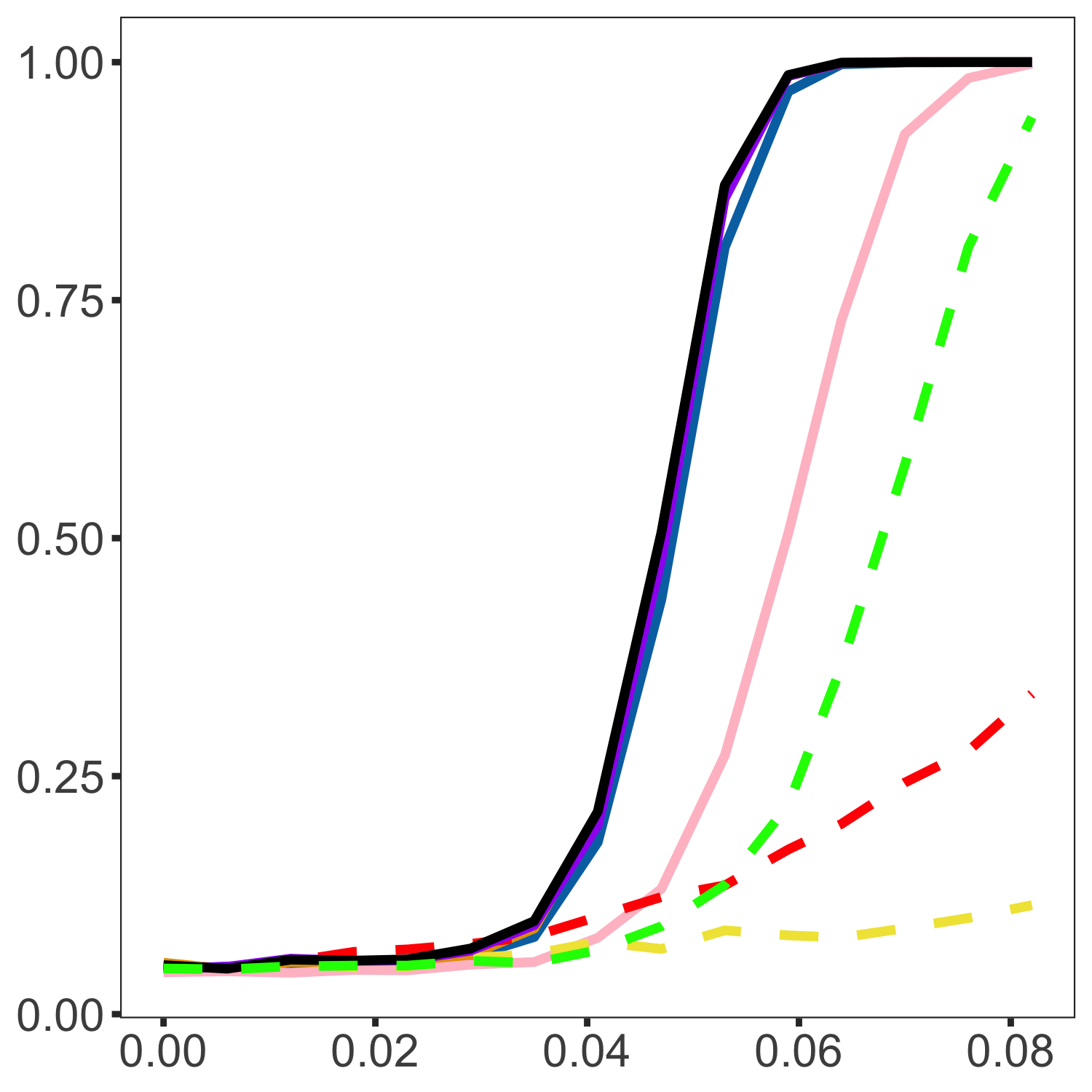}
    \end{subfigure}
     \begin{subfigure}[t]{0.23\textwidth}
        \centering
        \includegraphics[width=\linewidth, height=0.7\linewidth]{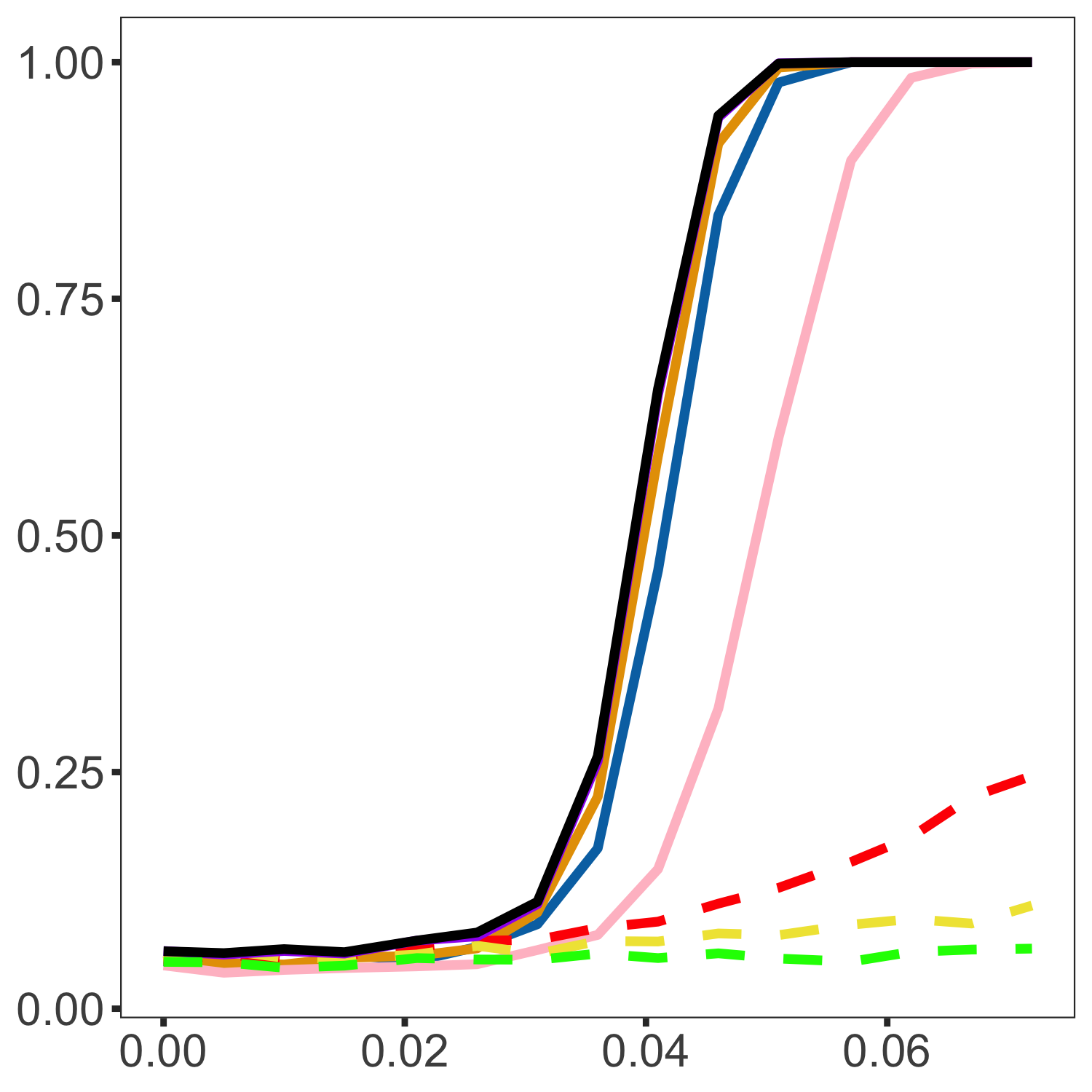}
    \end{subfigure}
    \begin{subfigure}[t]{0.23\textwidth}
        \centering
        \includegraphics[width=\linewidth, height=0.7\linewidth]{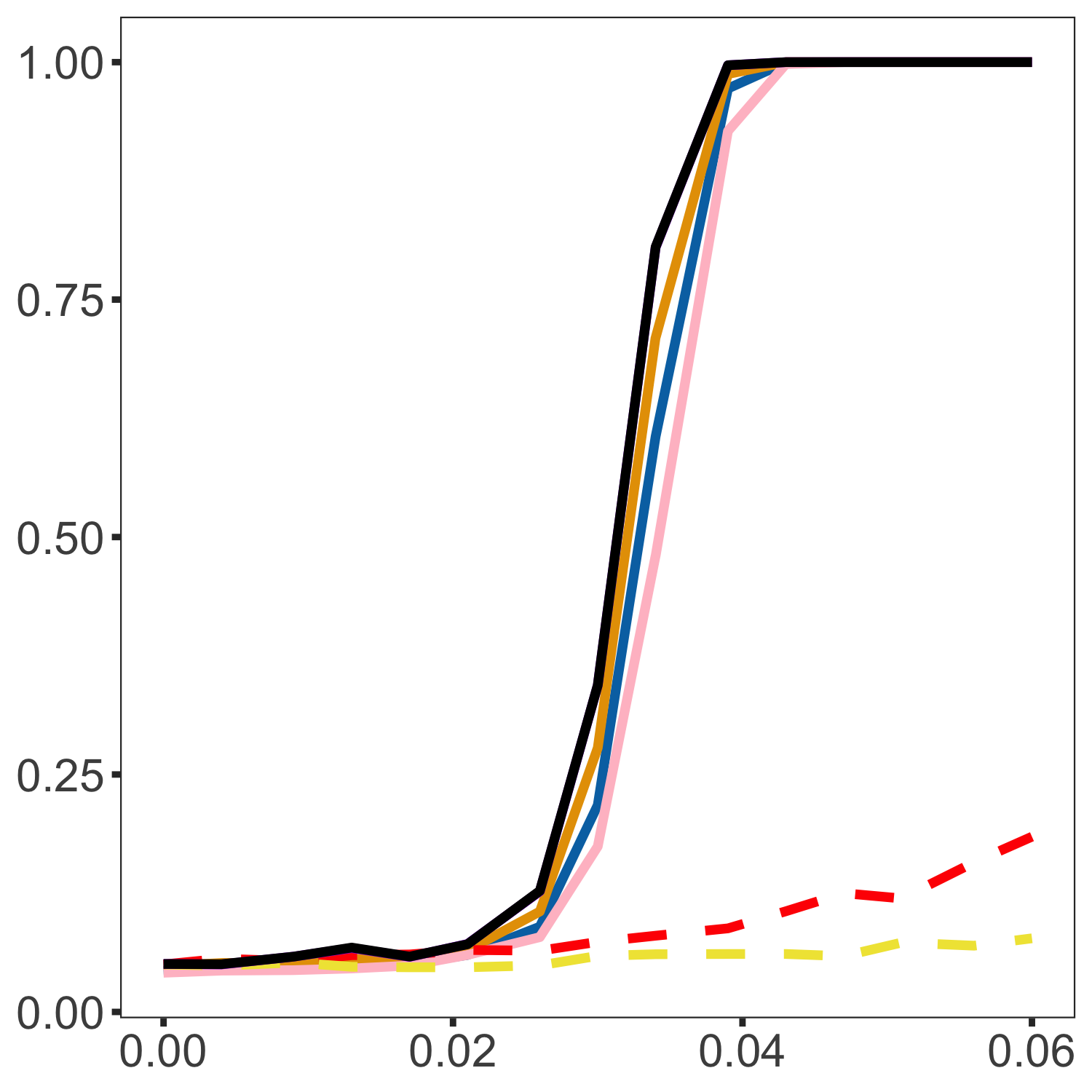}
    \end{subfigure}
   \caption{Size-adjusted empirical power \(\Sigma\) is Poly-Decay. Columns (left to right) correspond to \(\hat{\gamma}_2 = 0.3, 0.5, 0.9, 2\); rows (top to bottom) correspond to \(n_1 = 50, 100, 250\). Solid curves: blue (\(\lambda = 0.5\)), orange (\(\lambda = 1\)), black (\(\lambda=\hat{\lambda}_{I_p}\)), purple (\(\lambda=\hat{\lambda}_{\Sigma_p}\)), and pink (\(\lambda=\hat{\lambda}_*\)). Dashed curves: red (Proj-LRT), yellow (Ridge-LRT), and green (\cite{han2016tracy}, \(\lambda=0\)), the latter available only when \(p<n_1+n_2\).}
    \label{fig:emp_power_Sigma2}
\end{figure}

\subsection{Empirical power}\label{sub:empirical_power}
In this section, we evaluate the empirical power of the proposed tests and compare them with Ridge-LRT, Proj-LRT, and \cite{han2016tracy} ($\lambda \to 0$). Figure~\ref{fig:emp_power_Sigma2} presents power curves against the signal strength parameter $\zeta$ (See Section S.4.1 
in the Supplementary Material) under the Poly-Decay setting; results for other settings are provided in Section~S.3.3 of the Supplementary Material. To ensure a fair comparison, all power curves are adjusted using size-corrected cutoffs derived from each test’s empirical null distribution via simulation. Results of the proposed methods based on asymptotic cutoffs are similar and therefore omitted for brevity.

\textcolor{black}{The proposed methods demonstrate substantially higher power across all simulation settings than the likelihood-based methods, Ridge-LRT and Proj-LRT, underscoring their effectiveness in detecting concentrated alternatives. The test of \cite{han2016tracy} exhibits comparable performance when $\hat{\gamma}_2 =0.3$, but becomes noticeably less powerful as $\hat{\gamma}_2$ increases, particularly when $\hat{\gamma}_2$ approaches 1. This behavior is consistent with the theoretical power analysis in Sections~\ref{subsec:power_comparison} and~\ref{subsec:effect_lambda}.} 

\textcolor{black}{Under the alternative model in Setting (v) (See Section S.4.1 
in the Supplementary Material), the matrix $BC$ can be viewed as a generalization of the \textbf{PA} model \eqref{eq:polynomial_alternatives} with $D = I_p$. Accordingly, $\lambda = \hat{\lambda}_{I_p}$ corresponds to the data-driven Bayesian choice under a correctly specified prior, whereas $\hat{\lambda}_{\Sigma_p}$ represents the choice under a mildly misspecified prior. The choice $\lambda = \hat{\lambda}_{I_p}$ consistently attains the highest power among all considered settings as expected. Nevertheless, across all settings, the loss in power is relatively modest when $\lambda$ is fixed at $0.5$, $1$ or when the prior is mildly misspecified. In contrast, the minimax selection exhibits a noticeable efficiency loss, which is typical of minimax procedures that trade optimality under specific alternatives for robustness across a broader class of scenarios.}

Lastly, the power of the largest root tests remains essentially stable as $n_1$ varies, whereas increasing $n_1$ leads to a noticeable decline in power for Proj-LRT and Ridge-LRT. This behavior arises because the variance of LRT–type statistics inflates when $p$ is large, thereby degrading the signal-to-noise ratio. This highlights the advantage of the proposed methods in settings that involve the simultaneous testing of many linear hypotheses.

\section{Discussion}\label{sec:discussion}

In this paper, we proposed a powerful and computationally efficient largest-root test for high-dimensional general linear hypotheses, based on a ridge-regularized $F$-matrix. Leveraging tools from Random Matrix Theory, we established that, under the null hypothesis, the largest eigenvalue follows an asymptotic Tracy–Widom distribution after appropriate scaling, in regimes where the dimension is comparable to the effective sample size and the number of hypotheses. We developed a consistent and efficient procedure for estimating the scaling parameters, relying solely on the empirical eigenvalues of the test matrices and without requiring structural assumptions. The power of the test and the influence of the regularization parameter were analyzed under a class of rank-one alternatives. Methods for data-driven selection of the regularization parameter are developed, guided by Bayesian decision theory and minimaxity. Extensive simulation studies demonstrate that the proposed test maintains well-controlled type I error rates and exhibits strong power across diverse settings. Furthermore, the method is robust to variations in covariance structure and data distributions, underscoring its practical utility in high-dimensional hypothesis testing.

Several promising research directions can be pursued to extend the current work. From a technical standpoint, our analysis assumes regular stability properties for the leading eigenvalues of the population covariance matrix. This excludes important models such as classical factor models, where a finite number of eigenvalues are well-separated from the bulk. Investigation under a class of spiked covariance models is a valuable direction for future research. Another avenue is the integration of the proposed test with variable screening techniques to enable application in ultra-high-dimensional regimes, where $p$ grows substantially faster than the sample sizes. By reducing dimensionality prior to testing, such an approach could improve computational efficiency while preserving statistical power.

\section*{Supplementary Material}
The Supplementary Material includes real data analysis, additional discussions and implementation details for Algorithms~\ref{algo:linear_program}, \ref{algo:estimation_rho}, and \ref{algo:ODE}, additional simulation results, detailed proofs of the technical results presented in the paper, and possible extensions of the current framework both
technically and in terms of the form of the linear hypotheses.

\bibliographystyle{apalike}
\bibliography{Ridge_reg_Fmatrix}

\newpage\clearpage


\setcounter{page}{1}
\setcounter{section}{0}
\renewcommand{\thesection}{S.\arabic{section}}
\setcounter{subsection}{0}
\renewcommand{\thesubsection}{S.\arabic{section}.\arabic{subsection}}
\setcounter{equation}{0}
\renewcommand{\theequation}{S.\arabic{section}.\arabic{equation}}
\setcounter{figure}{0}
\setcounter{table}{0}
\begin{center}
\textbf{ \large Supplementary Material to \\ ``Ridge-regularized Largest Root Test for High-Dimensional \\ General Linear Hypotheses''}\\
{Haoran Li}\\
{Department of Mathematics and Statistics}\\
{Auburn University}
\end{center}

This Supplementary Material is organized as follows.
\begin{enumerate}
\item[] Section~\ref{sec:real_data_analysis} presents an application of the proposed method to Human Connectome Project data, assessing associations between brain measurements and behavioral outcomes.
\item[] Section~\ref{sec:implement_details} provides additional discussion and implementation details for the estimation algorithms presented in Section \ref{sec:estimation}.
\item[] Section~\ref{sec:supp_power_analysis} gives further details on the power comparison reported in Section~\ref{subsec:power_comparison}, along with additional numerical results complementing Section~\ref{sec:power_analysis}.
\item[] Section~\ref{sec:additional_simulation_results} reports supplementary simulation results.
\end{enumerate}

The remaining sections are devoted to proofs of the technical results in the manuscript.  
\begin{enumerate}
\item[] Section~\ref{sec:basic_definitions} introduces the notation, definitions, and technical tools used in the proofs.
\item[] Section~\ref{sec:properties_of_bg_lambda} establishes local laws and rigidity for the matrix \(\bG_\lambda = \bZ U_2U_2^T \bZ^T + \lambda \Sigma_p^{-1}\).
\item[] Section~\ref{sec:proof_theorem_main} contains the proof of Theorem~\ref{thm:main}.
\item[] Section~\ref{sec:proof_of_theorem_ref_thm_consistency_estimators} provides proofs of Theorem~\ref{thm:consistency_estimators}, Lemma~\ref{lemma:consistency_condition1}, and Lemma~\ref{lemma:consistency_condition3}.
\item[] Section~\ref{sec:proof_power} contains proofs of Lemma~\ref{lemma:alternatives_deterministic}, Corollary~\ref{corollary:power_consistency_deterministic}, Lemma~\ref{lemma:alternatives_prob}, and Corollary~\ref{corollary:power_consistency_prob}.
\item[] Section~\ref{sec:relaxtion_conditions} discusses potential relaxation of the technical assumptions. 
\end{enumerate}

Other technical results, including Lemmas~\ref{lemma:converge_hat_lambda} and~\ref{lemma:converge_lambda_*}, are straightforward and their proofs are omitted.

Lastly, in Section \ref{sec:extension_to_more_general_hypotheses}, we extend the proposed procedure to accommodate a more general family of linear hypotheses of the form
\[ H_0: ABC = \Gamma \quad \text{against} \quad H_a: ABC \neq \Gamma, \]
for some given matrices $A$, $C$ and $\Gamma$.

\section{Real data analysis}\label{sec:real_data_analysis}

The Human Connectome Project (HCP), funded by the National Institutes of Health (NIH), aims to map the neural pathways that underlie human brain function and behavior. A consortium led by Washington University in St. Louis and the University of Minnesota has collected extensive neuroimaging and behavioral data from  $n_0 = 1,113$ healthy young adults aged $22$ to $35$. This dataset facilitates detailed analyses of brain circuits and their relationships to behavior and genetics. Among the publicly available data are cerebral volumetric measurements and various behavioral evaluation scores. In this section, we utilize the proposed tests to investigate the associations between cerebral measurements and behavioral variables using the HCP dataset.

The behavioral evaluation scores in the HCP dataset encompass several domains, including cognition, emotion, sensation, and motor functions. The cognition domain assesses various cognitive abilities such as episodic memory, cognitive flexibility, and attention control. The emotion domain includes measures of psychological well-being, social relationships, stress, and self-efficacy. The sensation domain comprises tests related to audition, olfaction, taste, and vision. The motor domain evaluates aspects like cardiovascular endurance, manual dexterity, grip strength, and gait speed. After a pre-screening process to filter out highly correlated variables, we selected $p=127$ representative behavioral variables to study their associations with cerebral measurements.

Cerebral measurements were obtained for 68 cortical surface regions, assessing surface area, cortical thickness, and gray matter volume. These regions are distributed across 14 cerebral lobes symmetrically located in both hemispheres, resulting in a total of $204$ cortical surface variables ($68$ regions $\times$ $3$ measurement types). Additionally, each voxel in the brain's subcortical structures was assigned one of 44 labels, with the volume of each labeled region measured, yielding 44 subcortical variables. Four labels were excluded from the analysis due to a high proportion of missing values.

Available demographic information includes the age and gender of subjects. Specifically, the subjects are divided into four age groups, namely 22–25, 26–30, 31–35, and 36+. The data set is roughly balanced with respect to gender, as the 1113 subjects consist of 606 females and 507 males.

Specifically, we consider the following multivariate regression model: 
\begin{equation}\label{eq:regression_model_real_data}
Y_i = \beta_0 + \beta_1^T \mathbf{L}_{1i} + \dots +\beta_{14}^T \mathbf{L}_{14,i} + \beta_{15}^T \mathbf{SC}_{i} + \beta_{16}^T \mathbf{D}_i +  \Sigma_{p}^{1/2}Z_i.  
\end{equation}
where (i) $Y_i$ is the vector of $127$ behavioral scores of subject $i$; (ii) $\mathbf{L}_{ji}$ ($j=1,\dots, 14$) represents the cerebral measurements vector of the surface regions belonging to lobe $j$ for subject $i$, with dimensions varying from $3$ to $33$ depending on the lobe. The variables are of 3 types: surface area, cortical thickness, and gray matter volume; (iii) $\mathbf{SC}_{i}$ contains the volume measurements of $40$ subcortical regions of subject $i$; (iv) $\mathbf{D}_i$ are age and gender group dummy variables of subject $i$. This model allows us to assess the associations between various cerebral measurements and behavioral outcomes while accounting for demographic variables. The dimension of the explanatory variables (including the intercept) is $m = 249$, the dimension of the response $Y_i$ is $p=127$, and the sample size is $n_0 = 1113$. We are interested in testing the joint significance of all variables associated with each cerebral lobe. Namely, we test $H_0: \beta_j =0$ against $H_a: \beta_j \neq 0$, for each $j=1,\dots, 14$.  The proposed test procedure is applied with five choices of $\lambda$: $0.5 \bar{\tau}$, $\bar{\tau}$, $1.5 \bar{\tau}$, the data-driven choice of $\lambda$ described in Section \ref{sec:selection_lambda} with $D=I_p$ and when $D= \Sigma_p$. The latter two are referred to as $\hat{\lambda}_{I_p}$ and $\hat{\lambda}_{\Sigma_p}$, respectively. Here, $\bar{\tau}$ is the average of eigenvalues of $\bW_2$. The p-values are reported in Table \ref{table:pvalues}
\begin{table}[ht]
\centering
\begin{tabular}{lccccc}
\toprule
\textbf{Lobe name} &  $\lambda =0.5\bar{\tau}$   & $\lambda =\bar{\tau}$ & $\lambda =1.5\bar{\tau}$ & $\lambda =\hat{\lambda}_{I_p}$ & $\lambda=\hat{\lambda}_{\Sigma_p}$  \\
\midrule
Left Temporal Lobe Medial Aspect  & 0.070  & 0.020 & 0.009 & 0.070 & 0.070  \\
Right Temporal Lobe Medial Aspect  & 0.435&0.565&0.637&0.435&0.435 \\
Left Temporal Lobe Lateral Aspect & 0.723&0.766&0.760&0.723&0.723 \\
Right Temporal Lobe Lateral Aspect & 0.085&0.035&0.022&0.085&0.085\\
Left Frontal Lobe                   & 0.658&0.585&0.579&0.658&0.658\\
Right Frontal Lobe                   & 0.010&0.001&0.001&0.010&0.010 \\
Left Parietal Lobe                  & 0.108&0.056&0.045&0.108&0.108\\
Right Parietal Lobe                  & 0.365&0.337&0.338&0.365&0.365\\
Left Occipital Lobe                 & 0.786&0.759&0.725&0.786&0.786\\
Right Occipital Lobe                 & 0.204&0.357&0.357&0.204&0.204 \\
Left Cingulate Cortex               & 0.252&0.252&0.285&0.252&0.252 \\
Right Cingulate Cortex               & 0.153&0.094&0.086&0.153&0.153\\
Left Insula                         &0.164&0.210&0.266&0.164&0.164\\
Right Insula                         & 0.185&0.244&0.291&0.185&0.185\\
\bottomrule
\end{tabular}
\caption{P-values: Joint significance tests of 14 lobes.}
\label{table:pvalues}
\end{table}


Among the significant coefficients at the 90\% level detected by the proposed tests with at least one choice of $\lambda$ are left medial temporal lobe, right frontal lobe, and left parietal lobe. Similar results are reported in the literature, for example \cite{li2024testing}. The left medial temporal lobe, encompassing structures such as the hippocampus and parahippocampal gyrus, plays a pivotal role in memory formation and retrieval. Damage to the lobe has been associated with impairments in verbal memory tasks, indicating its importance in processing verbal material. Functional connectivity within the lobe is crucial for cognitive functions like episodic memory. 

The frontal lobes, located directly behind the forehead, are the largest lobes in the human brain and are essential for voluntary movement, expressive language, and higher-level executive functions. These executive functions encompass a range of cognitive abilities, including planning, organizing, initiating, self-monitoring, and controlling responses to achieve specific goals. Damage to the frontal lobe can result in impairments in these areas, affecting one's ability to perform tasks that require these skills. 

The parietal lobe, situated behind the frontal lobe, plays a pivotal role in integrating sensory information from various parts of the body. It is responsible for processing touch, temperature, pressure, and pain sensations. Additionally, the left parietal lobe is involved in controlling skilled motor actions and is associated with numerical processing and language functions. Damage to this area can lead to difficulties in these domains, affecting one's ability to perform tasks that require these skills.

\clearpage

\section{Discussion and implementation details of the proposed estimation method}\label{sec:implement_details}
~\\

\emph{Discussion}. Additional constraints can be incorporated into Algorithm \ref{algo:linear_program} to further enhance precision. For instance, to improve the accuracy of  $\hat{\calH}_3(y)$, we may want to ensure that the second-order derivative of $\hat{Q}_1(\iz)$ matches those of $\tilde{\calH}_1(-\lambda \hat{\varphi}(\iz))$. Since the derivatives are also linear in the weights $w_k$, it is straightforward to regularize the problem and add such constraints. Our numerical studies suggest that when $\hat{\gamma}_2$ is relatively small (e.g., $\hat{\gamma}_2 \leq 1$), incorporating constraints on the second-order derivative of $\tilde{\calH}_1(-\lambda \hat{\varphi}(\iz))$ can slightly improve the estimation accuracy of $\hat{\Theta}_1$ and $\hat{\Theta}_2$. However, when $\hat{\gamma}_2$ is relatively large (e.g., $\hat{\gamma}_2 \geq 5$), such constraints tend to reduce estimation precision. This is mainly because the accuracy of $\hat{\varphi}''(\iz)$, which serves as the estimator of $\varphi''(\iz)$, diminishes when $p$ becomes much larger than $n_2$.

 Algorithm \ref{algo:linear_program} is inspired by \citet{el2008spectrum} and \citet{ledoit2012nonlinear}, but it diverges in the choice of basis functions and the design of the loss function. While \citet{el2008spectrum} and \citet{ledoit2012nonlinear} prioritize achieving smoothness in the estimator of \(F^{\Sigma_\infty}\) by employing a combination of smooth basis functions and point masses, our focus is fundamentally different. In the context of estimating \(s(x)\), the smoothness of \(F^{\Sigma_\infty}\) is not a primary concern. 
Consequently, we adopt a simpler approach by approximating \(F^{\Sigma_\infty}\) solely using a mixture of point masses. Our numerical experiments indicate that this method performs comparably to, if not better than, the smooth basis approach, provided the grid of \(\sigma_k\) points is sufficiently dense. This finding suggests that the additional complexity introduced by smooth basis functions offers no substantial benefit in this specific setting. As a result, we focus exclusively on point-mass mixture models in our estimation procedure. In addition, the loss functions used by \citet{el2008spectrum} and \citet{ledoit2012nonlinear} control only the discrepancy between $\hat{Q}_1$ and $\tilde{\calH}_1$. In contrast, our method explicitly penalizes the discrepancy between $\hat{Q}_2$ and $\tilde{\calH}_2$, as accurate estimation of $\calH_2$ is of direct interest.  

\emph{Implementation details for Algorithm \ref{algo:linear_program}}. The recommended choice of the grid \(\{\sigma_k\}_{k=1}^K\) is equally spaced on \([\tilde{\ell}_{n_2} , \tilde{\ell}_1]\), with \(K \approx 500\). The choice of \(K\) strikes a balance between computational efficiency and estimation accuracy.
The recommended choice of \(\{\iz_i\}_{i=1}^I\) is such that \(\Re(\hat{\varphi}(\iz_i))\) are evenly spaced on \([\hat{\varphi}(1.05\tilde{\ell}_1), \hat{\varphi}(-\lambda)]\), and $\Im(\hat{\varphi}(\iz_i)) = 10^{-2}\times \tilde{\ell}_1^{-1}$. We can first select $\hat{\varphi}(\iz_i)$ and numerically find the corresponding $\iz_i$ using optimization methods like Newton-Raphson. While more points improve accuracy, we recommend setting \(I \approx 500\) to balance computational efficiency and estimation precision. 
Traditional linear programming solvers can be used for implementation. In our simulation, we utilize the \texttt{R} package \texttt{Rglpk}, which implements the GNU Linear Programming Kit. With the recommended settings, the problem can be efficiently handled on a typical PC (\(\leq 10\) seconds).

\emph{Implementation details for Algorithm \ref{algo:ODE}}. The proposed ODE requires high accuracy on the initial condition $\hat{s}(0)$. We propose using Newton-Raphson method to find $\hat{s}(0)$, starting from the initial point $1/(\lambda\hat{\gamma}_2\hat{\varphi}(-\lambda)) -1/ \hat{\gamma}_2$. The ODE can then be solved by traditional numerical methods such as the fourth-order Runge-Kutta method (RK4). In our simulation studies, the \texttt{R} package \texttt{deSolve} is used for implementation.

\newpage
\clearpage

\section{Effect of the regularization parameter: additional details}\label{sec:supp_power_analysis}
First, we give the explicit expressions for $\Theta_H$ and $\Theta_{2p}(\infty)$ involved in the SNR of the tests in the two limiting regimes when $\lambda\to0$ and $\lambda\to\infty$. Let $\MP_p(\tau)$ be the Mar\v{c}enko-Pastur law with parameter $\Breve{\gamma}_2 = \hat{\gamma}_2\wedge \hat{\gamma}_2^{-1}$, that is 
\[ d\MP_p(\tau) = \frac{1}{2\pi\breve{\gamma}_2} \frac{\sqrt{[b_+ - \tau ][\tau - b_- ]}}{\tau }d\tau,  \quad \tau \in [b_-, b_+],\]
where $b_{\pm} = (1\pm \sqrt{\breve{\gamma}_2})^2$. Denote 
\[ s_{p0}(x) = \int\frac{d\MP_p(\tau)}{\tau - x},  \quad \text{ for }~~ x \in [0, b_-).\]
An explicit form of $s_{p0}(x)$ is available; see, for example, \cite{bai2010spectral}. Further, call 
\[ \breve{p} = p\wedge n_2 \quad \text{and} \quad \breve{n}_1 = n_1 \wedge (n_1+n_2 -p).\]

Define $\beta = \beta_{p0}$ to be the solution in $(0, b_-)$ to 
\[ \beta^2 s'_{p0}(\beta) =  \frac{\breve{n}_1 }{\breve{p}}.\]
Then, 
\[\Theta_H = \left[ \frac{p^2 \breve{p} }{2 \breve{n}_1^3  }  s''_{p0}(\beta) +  \frac{p^2}{\breve{n}_1^2\beta^3}  \right]^{1/3}. \]
It can be verified that $\Theta_{H} = p^{2/3} \breve{n}_1 ^{-2/3} \sigma_p^{-1}$, where $\sigma_p$ is the quantity defined in (2.7) of \cite{han2016tracy}. 

To obtain a bound on $\Theta_H^{-1}$, notice that 
\[  \Theta_H \geq  \frac{p^{2/3}}{ (n_1\wedge (n_1+n_2 -p))^{2/3}} \frac{1}{\beta} \geq \frac{p^{2/3}}{n_1^{2/3}} \frac{1}{b_-} = \gamma_1^{2/3} (1- \sqrt{\breve{\gamma}_2})^{-2}.\]
Therefore, 
\[ \Theta_H^{-1}|\gamma_2 -1|^{-1} \leq  \gamma_1^{-2/3} (1-\sqrt{\breve{\gamma}_2})^2 |1-\gamma_2|^{-1}.\]
If $\gamma_2<1$, 
\[\Theta_H^{-1}|\gamma_2 -1|^{-1} \leq \gamma_1^{-2/3} \frac{(1-\sqrt{\gamma_2})^2}{ 1-\gamma_2} = \gamma_1^{-2/3} \frac{1-\sqrt{\gamma_2}}{1+\sqrt{\gamma_2}}\sim \frac{1}{2\gamma_1^{2/3}} (1-\sqrt{\gamma_2}), \quad \text{if }\gamma_2 \sim 1-. \]
If $\gamma_2>1$,
\[ \Theta_H^{-1}|\gamma_2-1|^{-1} \leq \gamma_1^{-2/3} \gamma_2^{-1} \frac{ (1-\sqrt{\breve{\gamma}_2})^2}{1-\breve{\gamma}_2} \sim \frac{1}{2\gamma_1^{2/3}} (1-\sqrt{\breve{\gamma}_2}), \quad \text{ if }\gamma_2\sim 1+.\]

\medskip
Next, we consider the limiting case when $\lambda \to \infty$. The definition of $\Theta_{2p}(\infty)$ is given as follows. Let 
\[s_{p\infty}(x) = \int \frac{dF^{\Sigma_p}(\tau)}{1/\tau - x}, \quad \text{for }~~x\in [0, 1/\sigma_{1p}). \]
Define $\beta = \beta_{p\infty}$ to be the solution in $(0, 1/\sigma_{1p})$ to 
\[ \beta^2 s'_{p\infty}(\beta) = \frac{n_1}{p}.\]
Then, 
\[ \Theta_{2p}(\infty) = \left[\frac{(p/n_1)^3}{2} s''_{p\infty}(\beta) + \frac{(p/n_1)^2}{\beta^3}\right]^{-1/3}.\]
Indeed, $\Theta_{2p}(\infty) = \lim_{\lambda \to \infty}\lambda \Theta_{2p}(\lambda)$. Also, it can be verified that $\Theta_{2p}(\infty)= (n_1/p)^{2/3} \gamma_0$, where $\gamma_0$ is the quantity defined in (2.17) of \citet{lee2016tracy}. 

\begin{proof}[Proof of Lemma \ref{lemma:alternative_Han2016}]
    \label{proof:Lemma_Han2016}
    When $\gamma_2<1$, we have that $\ell_H = \ell_{\max}(\tilde{\bF}_0)$. Therefore, the limit of $p^{-s}\ell_H$ follows directly from that of $p^{-s}\ell_{\max}(\tilde{\bF}_\lambda )$ as in Lemma \ref{lemma:alternatives_prob} when we take $\lambda \to 0$. In particular,
    \[ \lambda \varphi_p(-\lambda) \to (1-\gamma_2), \quad \text{if }\gamma_2<1.\]
    
    When $\gamma_2>1$, $\ell_H$ is characterized as follows. Set $\hat{\bW}_1 = \Sigma_p^{-1/2} \bW_1 \Sigma_p^{-1/2}$ and $\hat{\bW}_2 = \Sigma_p^{-1/2} \bW_2\Sigma_p^{-1/2}$. Note that $\hat\bW_2$ is of rank $n_2$ with probability 1. We denote the eigen-decomposition of $\hat{\bW}_2$ as 
    \[  \hat{\bW}_2 = [Q_1, Q_2] \begin{pmatrix}
        D ~&~ 0 \\0 ~&~0  
    \end{pmatrix}  [Q_1,Q_2]^T.\]
    Here, $D$ is $n_2\times n_2$ and $Q_1$ is $p\times n_2$. Then, following the analysis in (5.4)--(5.9) of \cite{han2016tracy}, $\ell_H$ is the largest root to 
    \[ \det\left( -\ell D + Q_1^T \hat{\bW}_1^{1/2}  [ I_p - \hat{\bW}_1^{1/2}Q_2 (Q_2^T \hat\bW_{1}  Q_2)^{-1}Q_2^T \hat{\bW}_1^{1/2}]  \hat{\bW}_1^{1/2} Q_1  \right) =0.\]
    Equivalently, $\ell_H$ is the largest eigenvalue of 
    \[  D^{-1}Q_1^T \hat{\bW}_1^{1/2}  [ I_p - \hat{\bW}_1^{1/2}Q_2 (Q_2^T \hat\bW_{1}  Q_2)^{-1}Q_2^T \hat{\bW}_1^{1/2}]  \hat{\bW}_1^{1/2} Q_1.\]
    Notice that the matrix in the square bracket above is a projection matrix. We then have that $\ell_H$ is bounded by 
    \[\ell_H  \leq  \text{the largest eigenvalue of } D^{-1}Q_1^T \hat{\bW}_1 Q_1, \text{or equivalently, that of } Q_1 D^{-1} Q_1^T \hat{\bW}_1.\]
    Since $Q_1 D^{-1} Q_1^T$ is the Moore-Penrose inverse of $\hat{\bW}_2$, denoted by $\hat{\bW}_2^+$,  we have that 
    \[  \ell_H \leq \ell_{\max}(\hat{\bW}_2^+ \hat{\bW}_1).\]
    
    Under \textbf{PA}, $\hat{\bW}_1$ is decomposed into 
    \[ \hat{\bW}_1^{(0)} + p^{s-1} \Sigma_p^{-1/2} D^{1/2} \nu \nu^T D^{1/2}\Sigma_p^{-1/2} + n_1^{-1}\Sigma^{-1/2}_p BX P_1 \bZ^T + n_1^{-1}\bZ P_1 X^TB^T\Sigma^{-1/2}_p.\]
    Here, $\hat{\bW}_1^{(0)}= n_1^{-1} \bZ P_1 \bZ^T$.  Following a well-known result in RMT, the largest (smallest) singular value of $n_1^{-1/2}\bZ P_2$ is bounded away from infinity (zero). We conclude that 
    \begin{align*}
    p^{-s} \ell_{\max}(\hat{\bW}_2^+ \hat{\bW}_1) &= p^{-1}  \nu ^TD^{1/2}  \Sigma^{-1/2}_p \hat{\bW}_2^+ \Sigma_p^{-1/2} D^{1/2}  \nu+ o_P(1)\\
    & = p^{-1} \tr(\hat{\bW}^+_2 \Sigma_p^{-1/2}D\Sigma^{-1/2}_p) + o_P(1).
    \end{align*}
    
    Recently, the asymptotic properties of the Moore-Penrose inverse of a large-dimensional sample covariance matrix have been studied in \cite{bodnar2024reviving}. In particular, Theorem 2.1 of \cite{bodnar2024reviving} implies that 
    \[p^{-1} \tr(\hat{\bW}^+_2 \Sigma_p^{-1/2}D\Sigma^{-1/2}_p) = \gamma_2^{-1}(\gamma_2-1)^{-1}p^{-1}\tr(D\Sigma_p^{-1}) + o_P(1).\]
    The bound on $\ell_H$ when $\gamma_2>1$ follows. 
\end{proof}

\clearpage

The specifications in the numerical study presented in Section \ref{subsec:effect_lambda} are as follows:
\begin{itemize}
    \item[(1)] We consider four spectral structures for $\Sigma_p$:
    \begin{itemize}
        \item[(i)] (Identity) $\Sigma_p = I_p$.
        \item[(ii)] (Poly-Decay) The eigenvalues of $\Sigma_p$ are: $(1+ j/p)^{-2}$ when $j=1,2,\dots, p$.
        \item[(iii)] (AR-ACF) The entries of $\Sigma_p$ are  $0.3^{|i-j|}$, $1\leq i,j\leq p$.
        \item[(iv)] (Point-Mix) The eigenvalues are 1, 5, 15 of proportion $40\%$, $30\%$, and $30\%$.  
    \end{itemize}
    Under all settings, the spectrum is further normalized so that $\tr(\Sigma_p) = p$.
    \item[(2)] We consider two settings for $D$. Namely, $D=I_p$ and $D=\Sigma_p$. 
\end{itemize}

Figure \ref{fig:SNRvsGamma2_Gamma1_2} displays the SNR as $\gamma_2$ increases when $\gamma_1 =2$. Figure \ref{fig:SNRvsLamda_Gamma1_05} displays $\SNR_p(\lambda)$ as a function of $\lambda$ when $\gamma_1 = 0.5$.

\begin{figure}[H]
  \centering 
  \begin{subfigure}{0.24\textwidth}
    \includegraphics[width=\linewidth,height=0.6\linewidth]{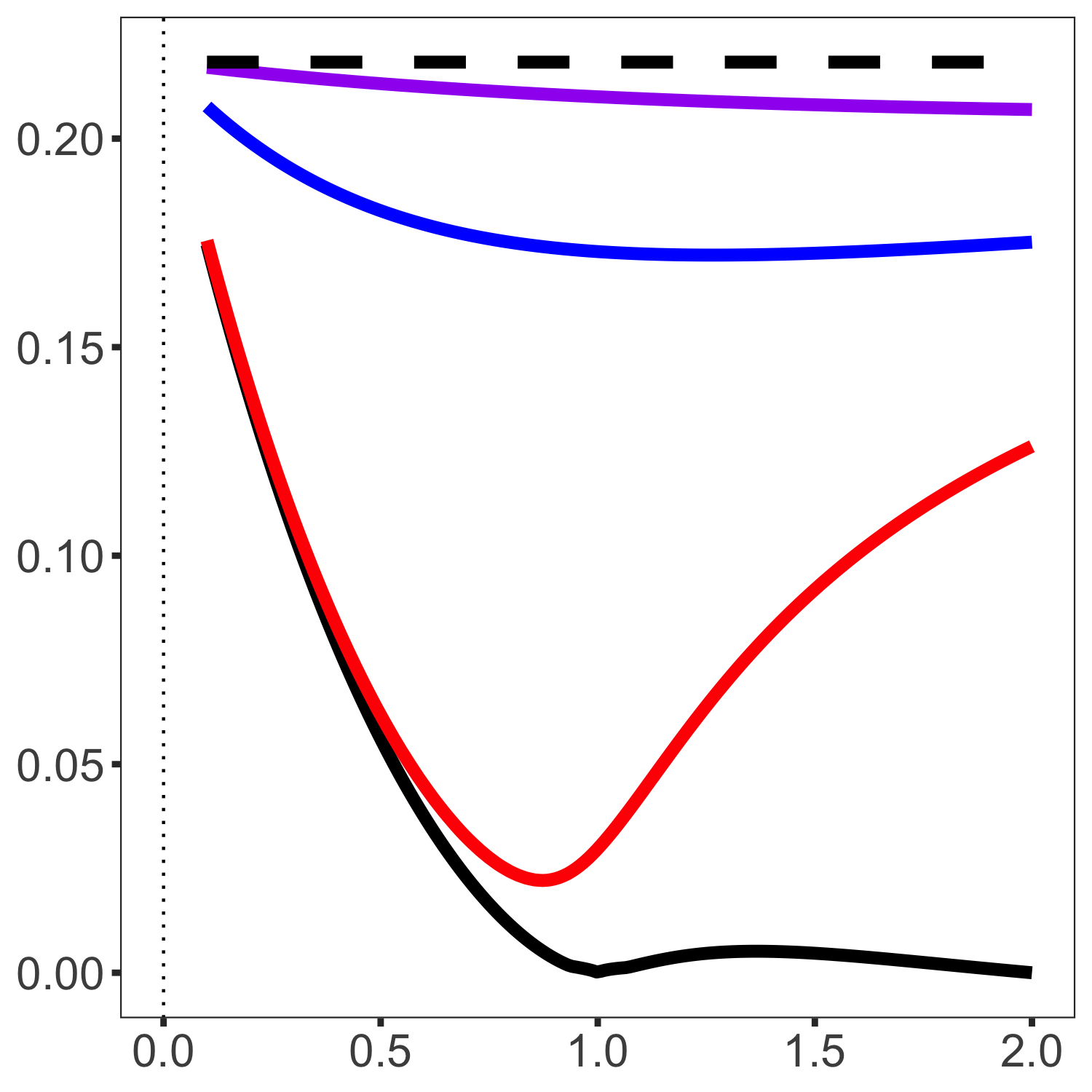}
   \end{subfigure}
  \hfill
  \begin{subfigure}{0.24\textwidth}
    \includegraphics[width=\linewidth,height=0.6\linewidth]{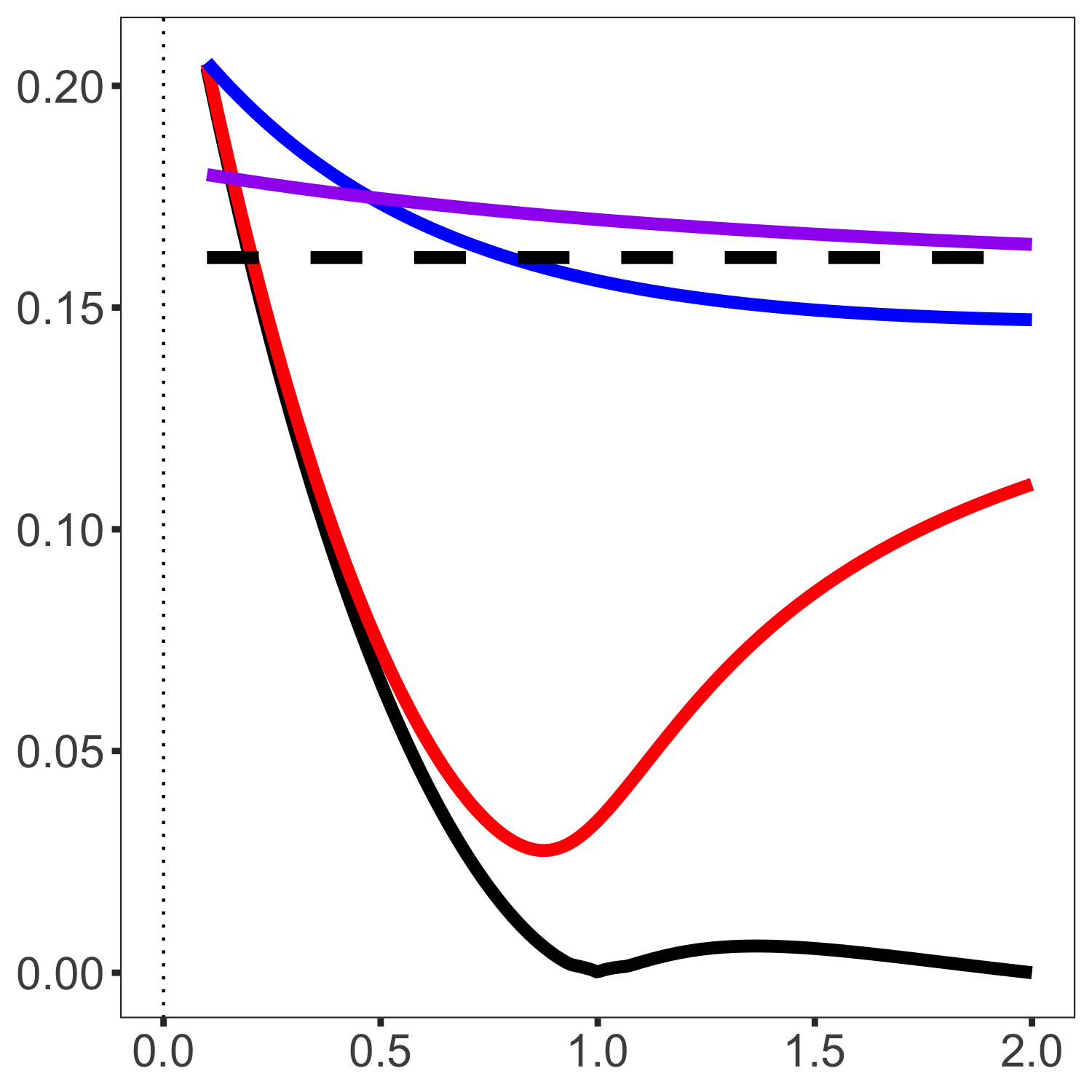}
   \end{subfigure}
  \hfill
  \begin{subfigure}{0.24\textwidth}
    \includegraphics[width=\linewidth,height=0.6\linewidth]{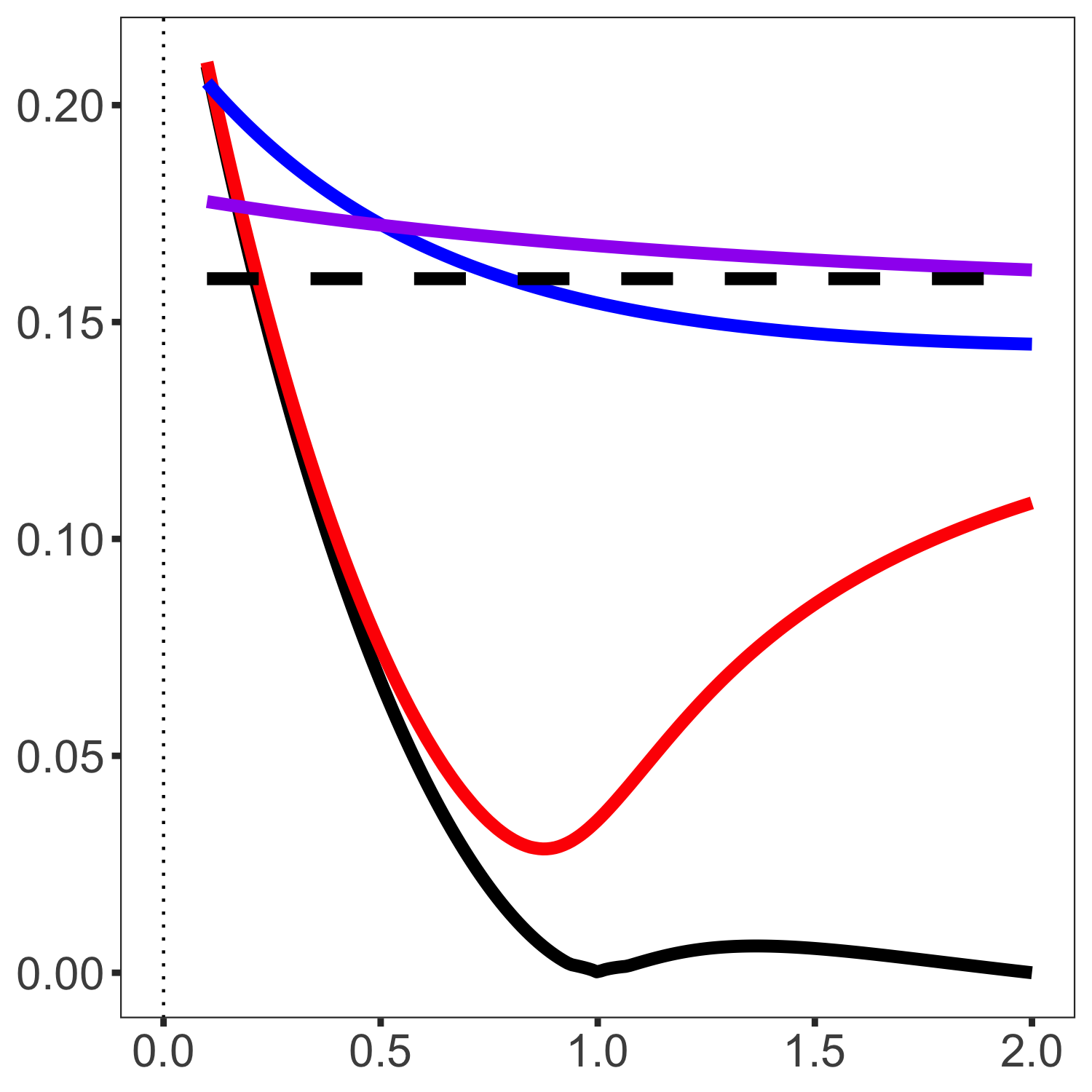}
   \end{subfigure}
  \begin{subfigure}{0.24\textwidth}
    \includegraphics[width=\linewidth,height=0.6\linewidth]{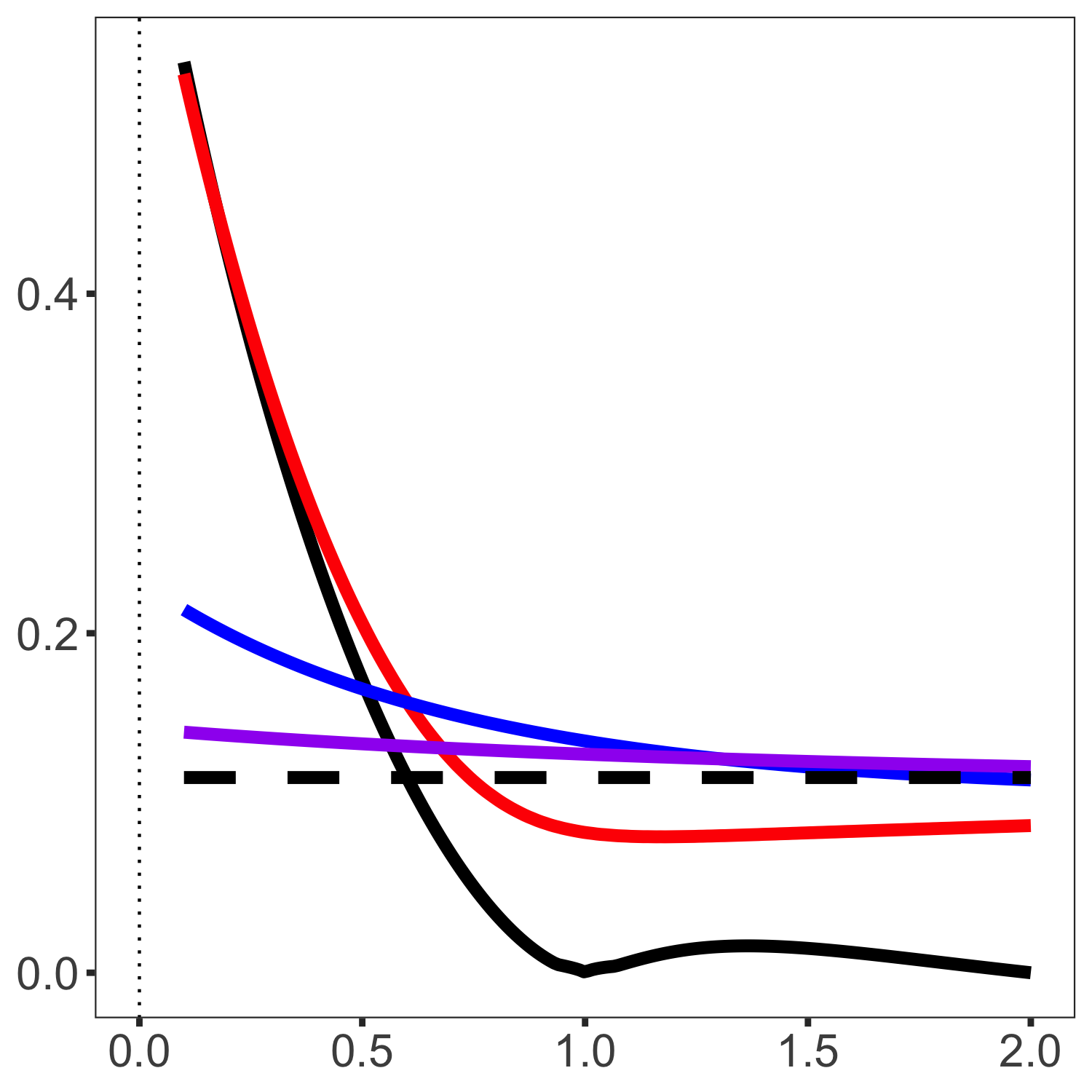}
   \end{subfigure}
  \vfill
  \begin{subfigure}{0.24\textwidth}
    \includegraphics[width=\linewidth,height=0.6\linewidth]{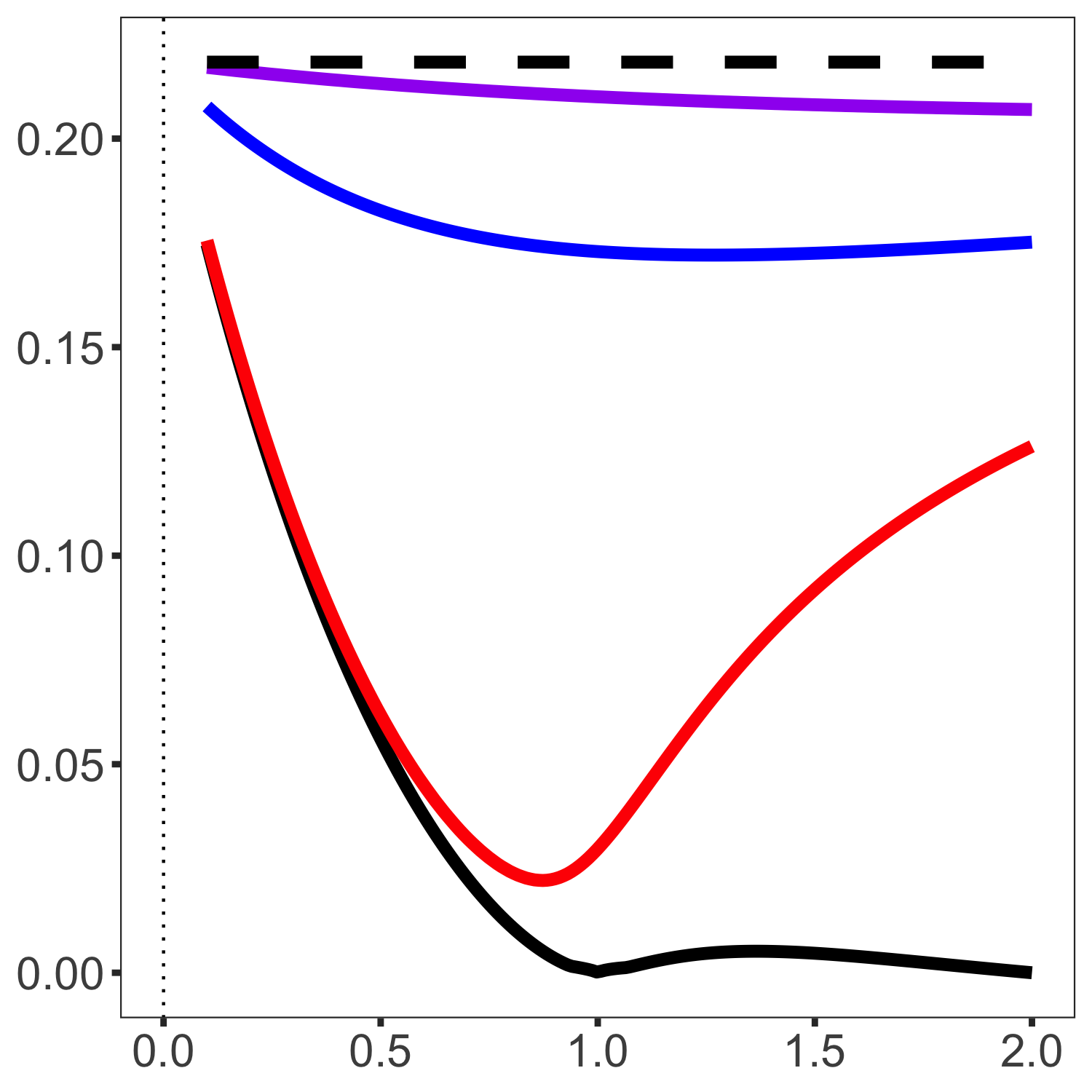}
   \end{subfigure}
  \hfill
  \begin{subfigure}{0.24\textwidth}
    \includegraphics[width=\linewidth,height=0.6\linewidth]{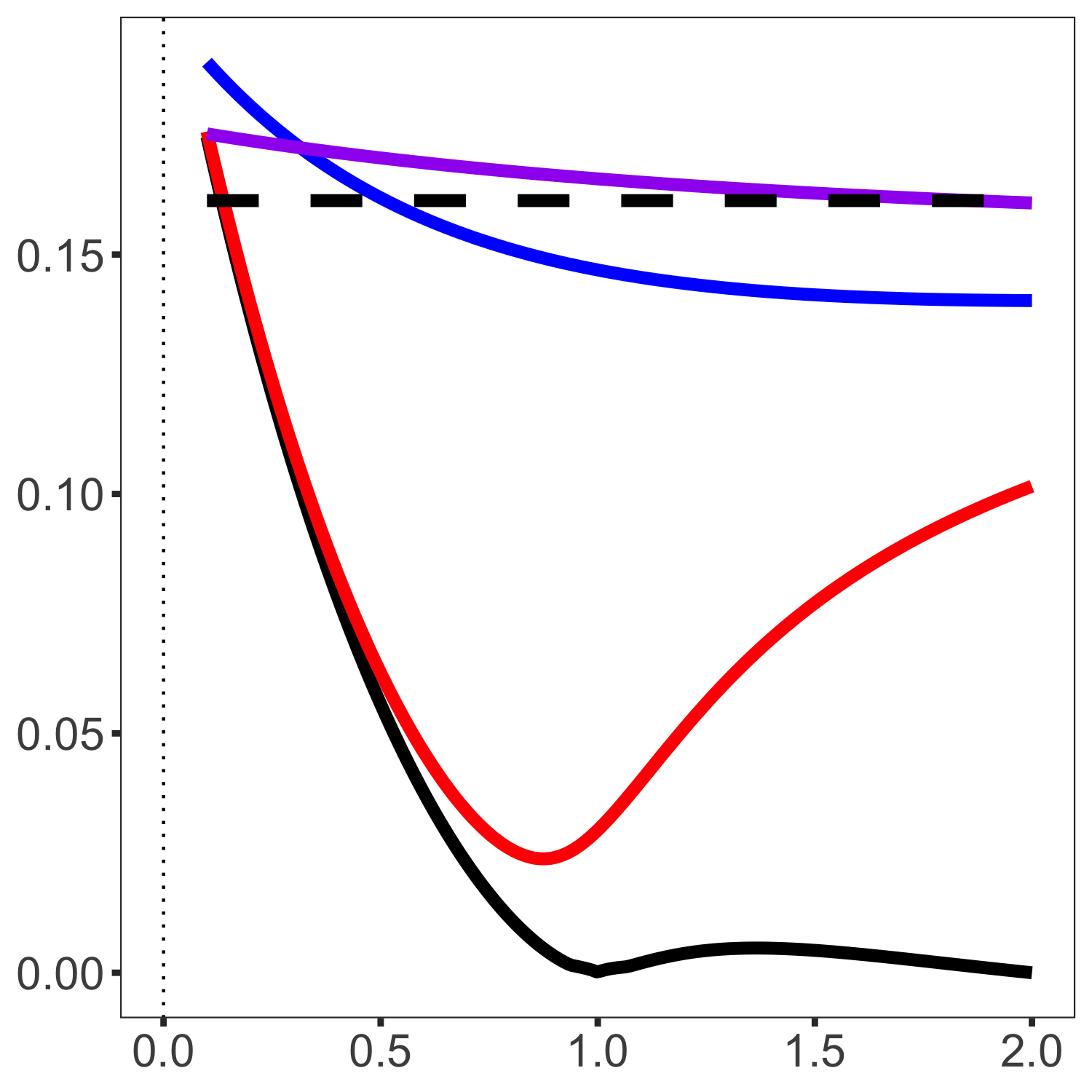}
   \end{subfigure}
   \hfill
   \begin{subfigure}{0.24\textwidth}
    \includegraphics[width=\linewidth,height=0.6\linewidth]{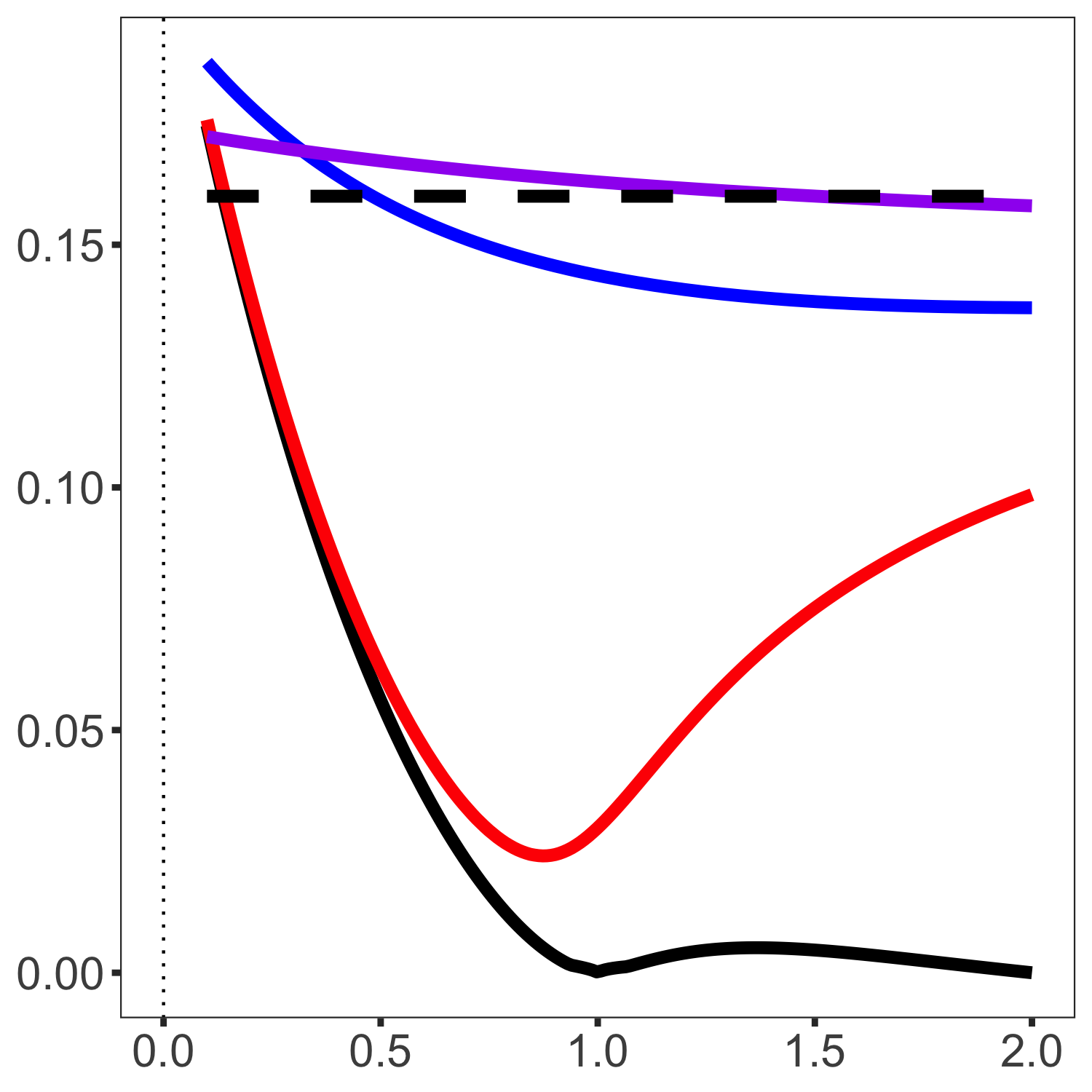}
   \end{subfigure}
  \hfill
  \begin{subfigure}{0.24\textwidth}
    \includegraphics[width=\linewidth,height=0.6\linewidth]{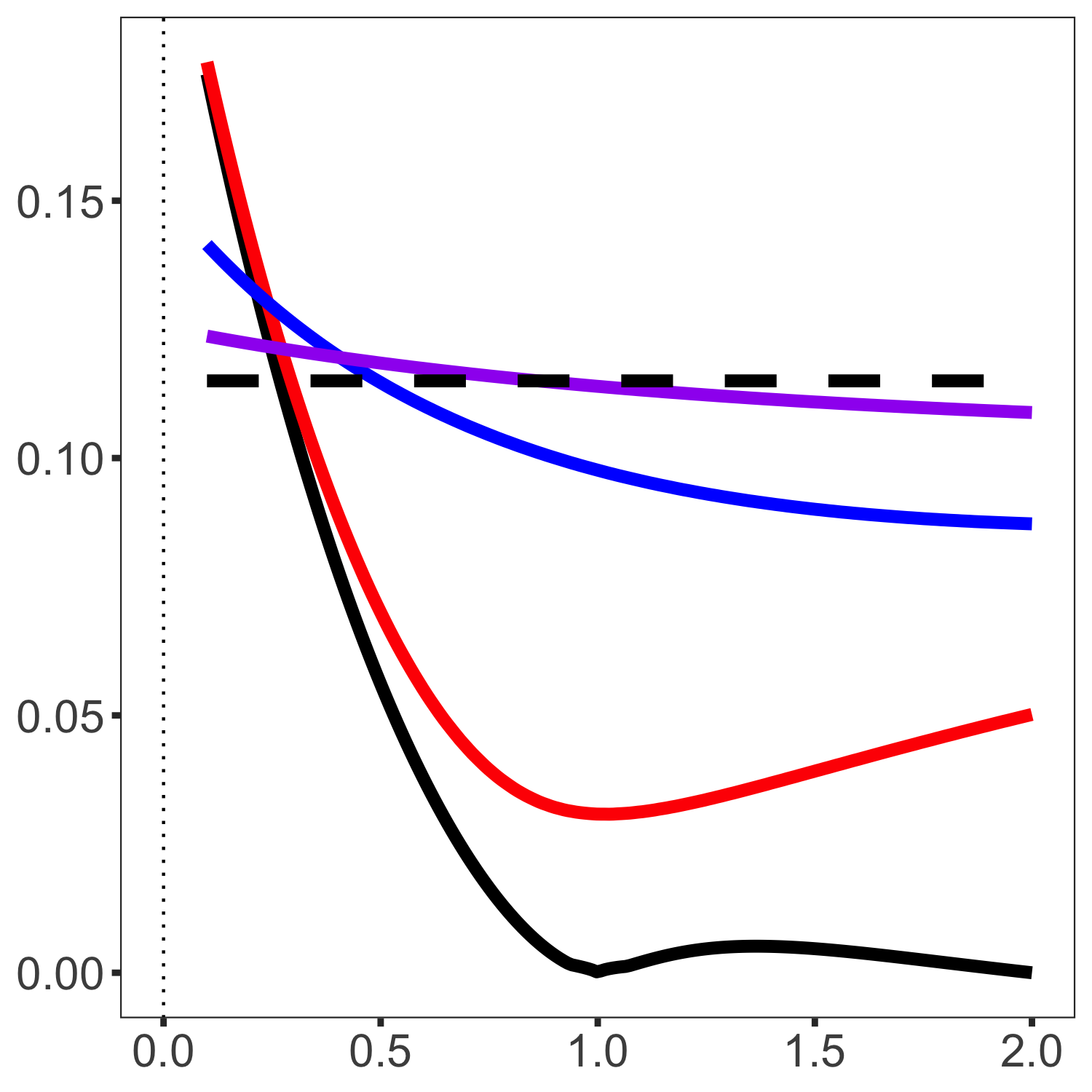}
   \end{subfigure}
\caption{$\SNR$ as a function of $\gamma_2$ when $\gamma_1 =2$. Columns: $\Sigma_p$ = Identity, Poly-Decay, AR-ACF, Point-Mix. Rows: $D \propto I_p$ (top), $D \propto \Sigma_p$ (bottom). Red/blue/purple: $\lambda =0.01$, $1$, and $5$. Black solid: upper bound of $\SNR_p(0)$ \citep{han2016tracy} in \eqref{eq:SNR_0}. Black dashed: $\SNR_p(\infty)$ in \eqref{eq:SNR_infty}.}
  \label{fig:SNRvsGamma2_Gamma1_2}
\end{figure}

\begin{figure}[h]
    \centering
    \begin{subfigure}{0.24\textwidth}
    \includegraphics[width=\linewidth,height=0.6\linewidth]{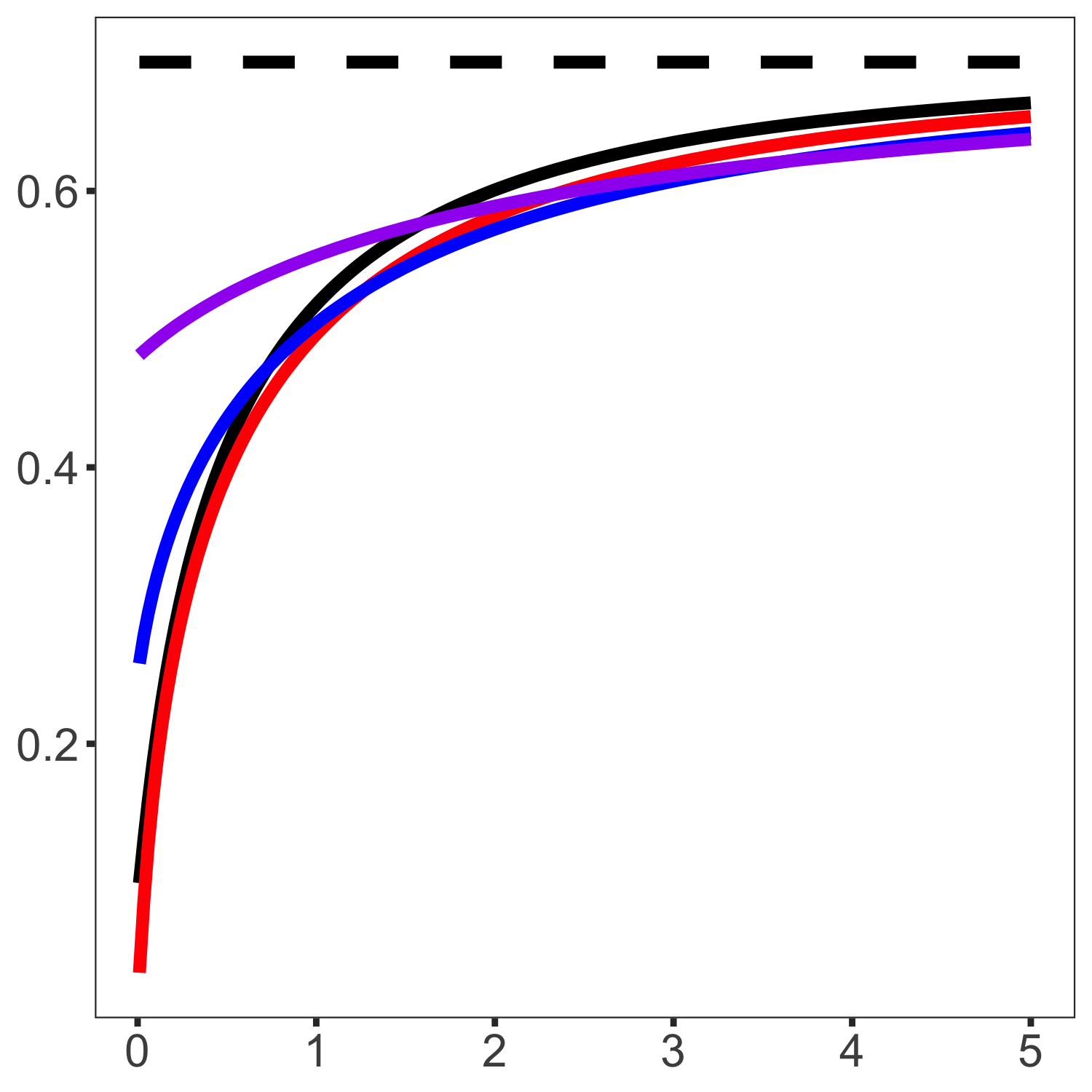}    
    \end{subfigure}
    \hfill
    \begin{subfigure}{0.24\textwidth}
    \includegraphics[width=\linewidth, height=0.6\linewidth]{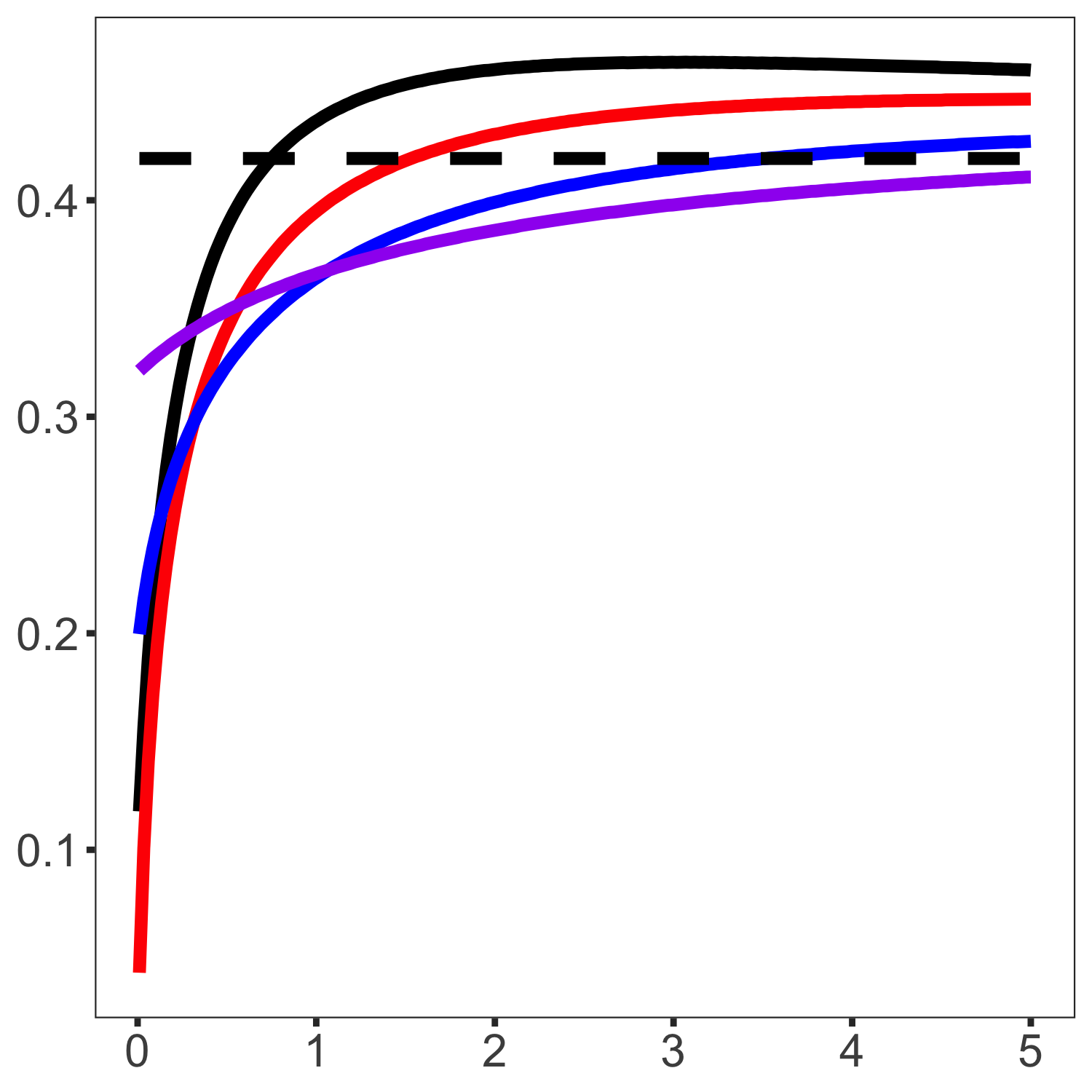}    
    \end{subfigure}
    \hfill
    \begin{subfigure}{0.24\textwidth}
    \includegraphics[width=\linewidth, height=0.6\linewidth]{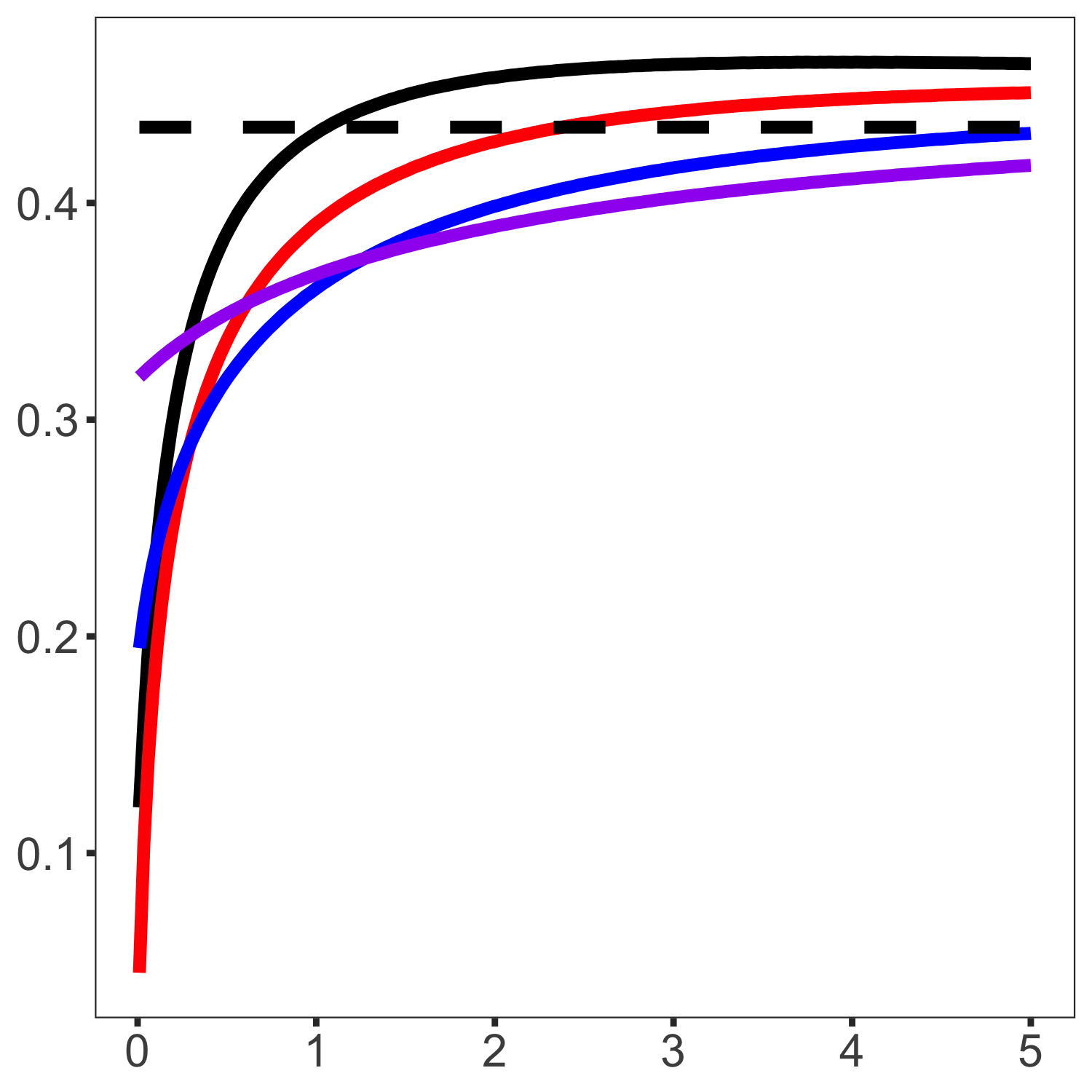}    
    \end{subfigure}
    \hfill
    \begin{subfigure}{0.24\textwidth}
    \includegraphics[width=\linewidth, height=0.6\linewidth]{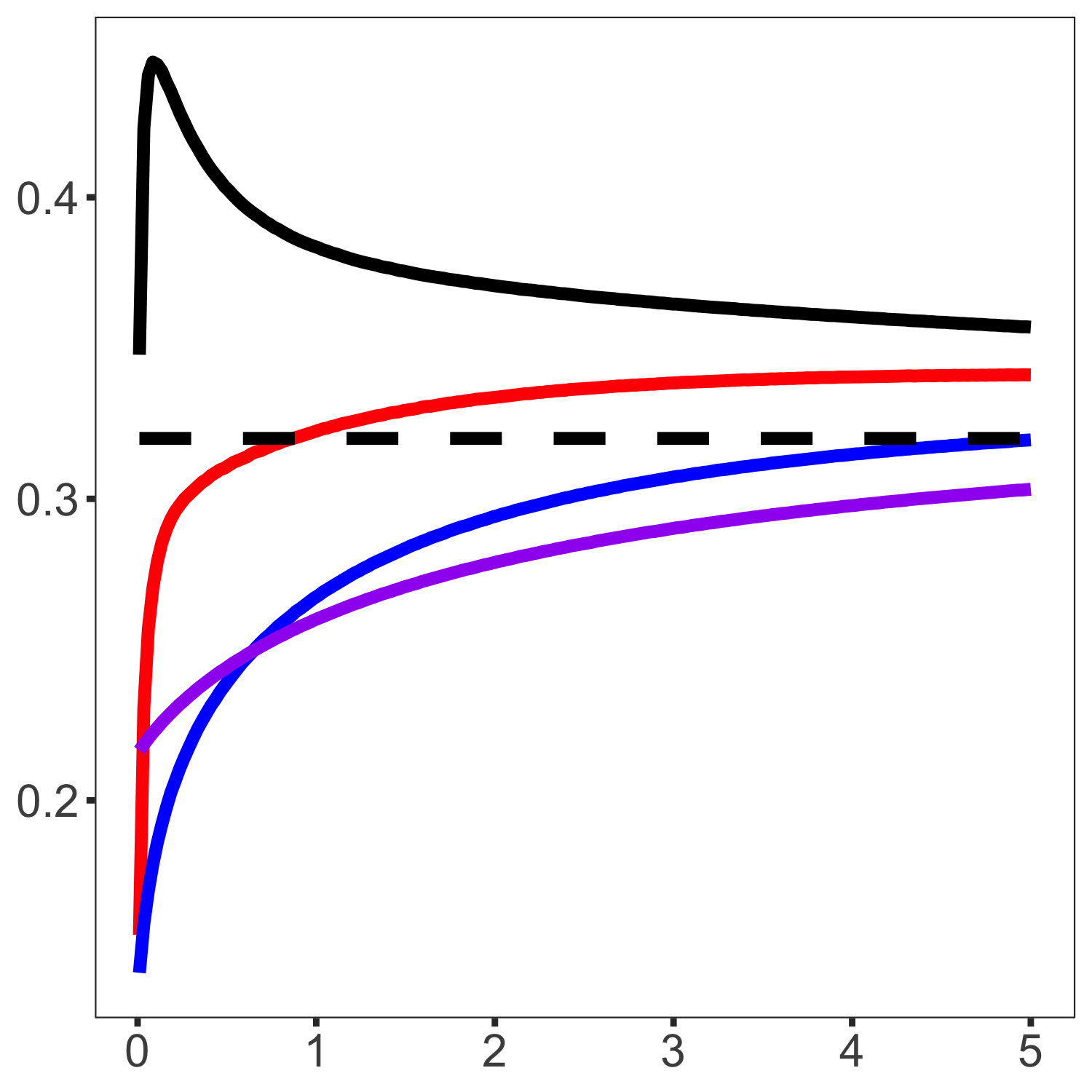}    
    \end{subfigure}
    \vfill
    \begin{subfigure}{0.24\textwidth}
    \includegraphics[width=\linewidth, height=0.6\linewidth]{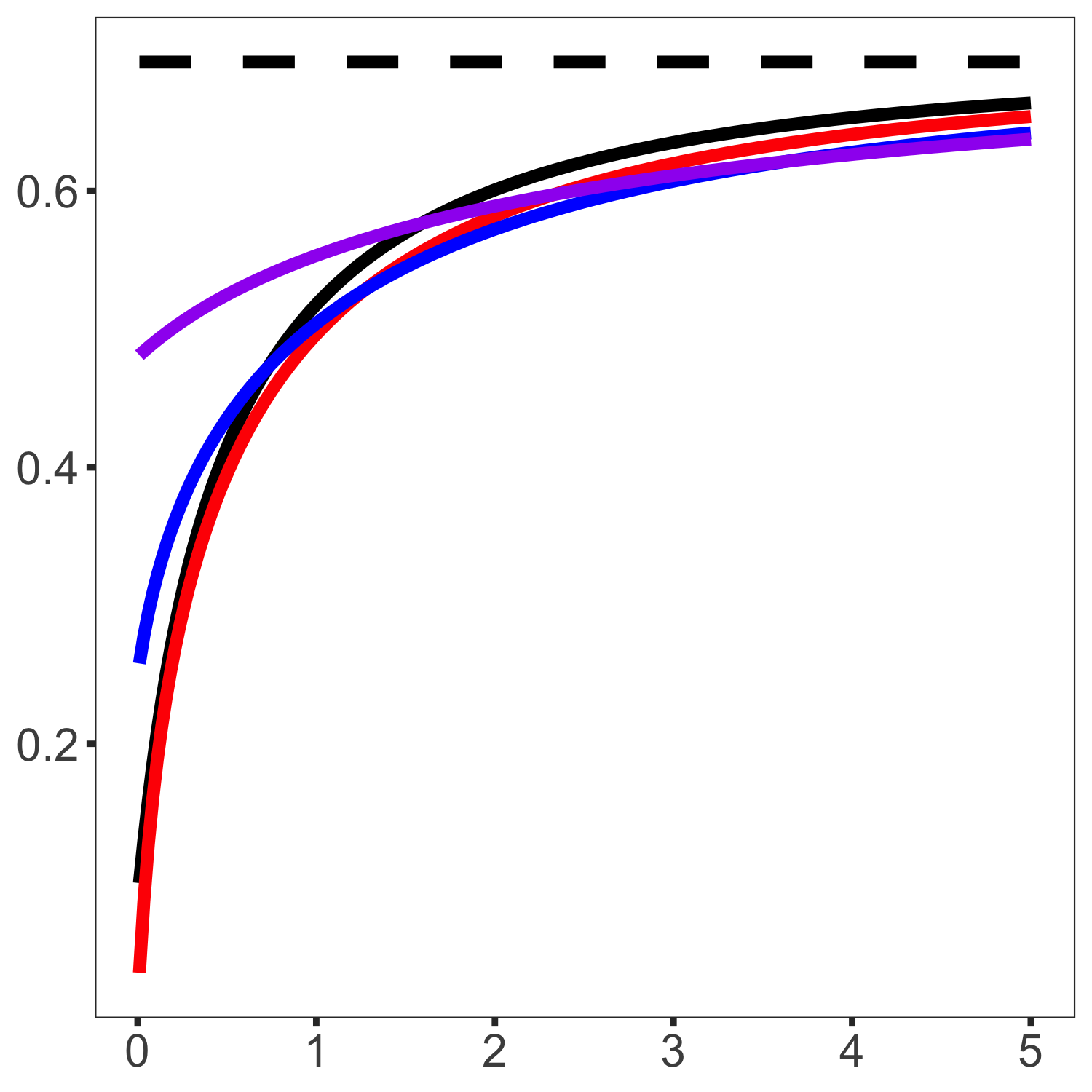}    
    \end{subfigure}
    \hfill
    \begin{subfigure}{0.24\textwidth}
    \includegraphics[width=\linewidth, height=0.6\linewidth]{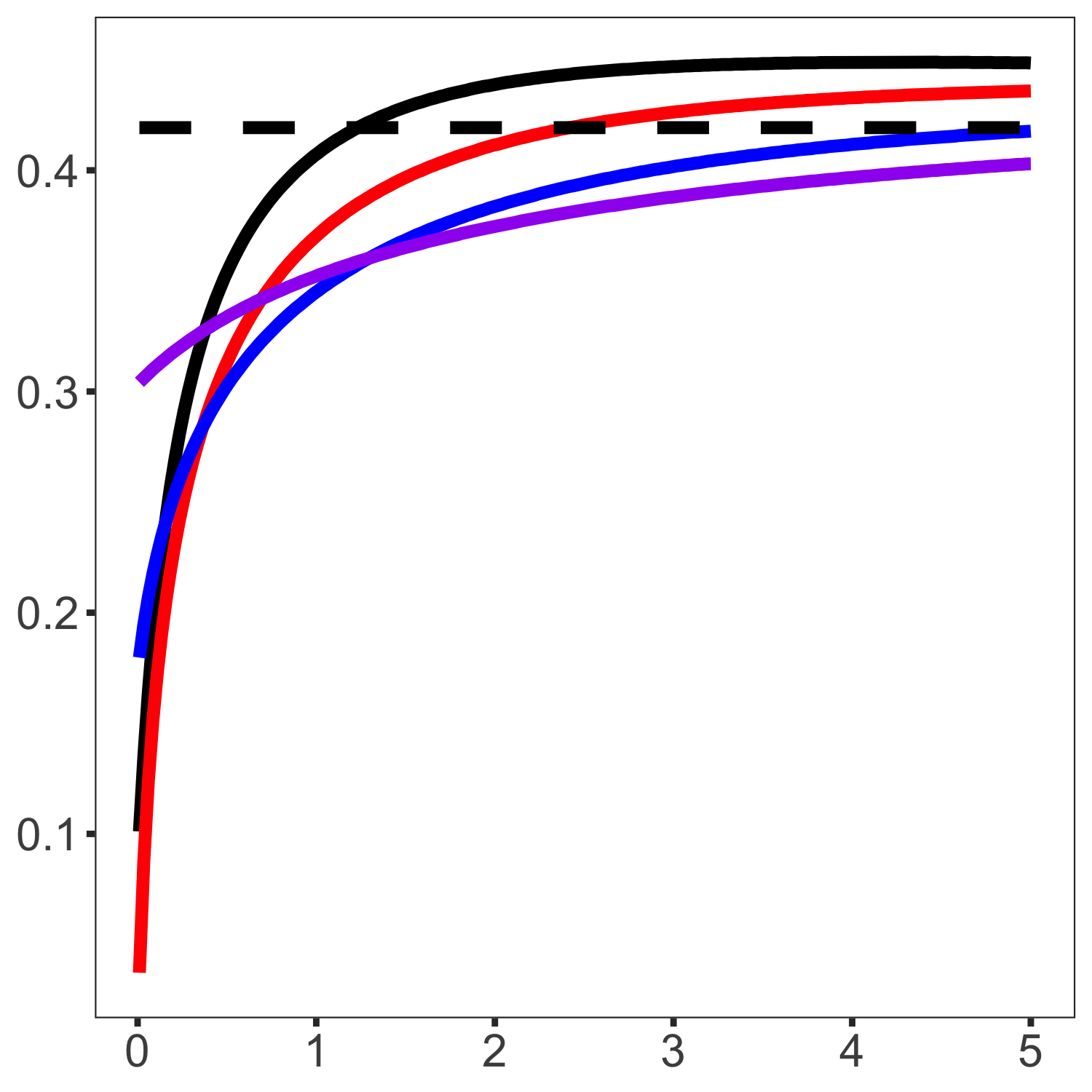}    
    \end{subfigure}
    \begin{subfigure}{0.24\textwidth}
    \includegraphics[width=\linewidth, height=0.6\linewidth]{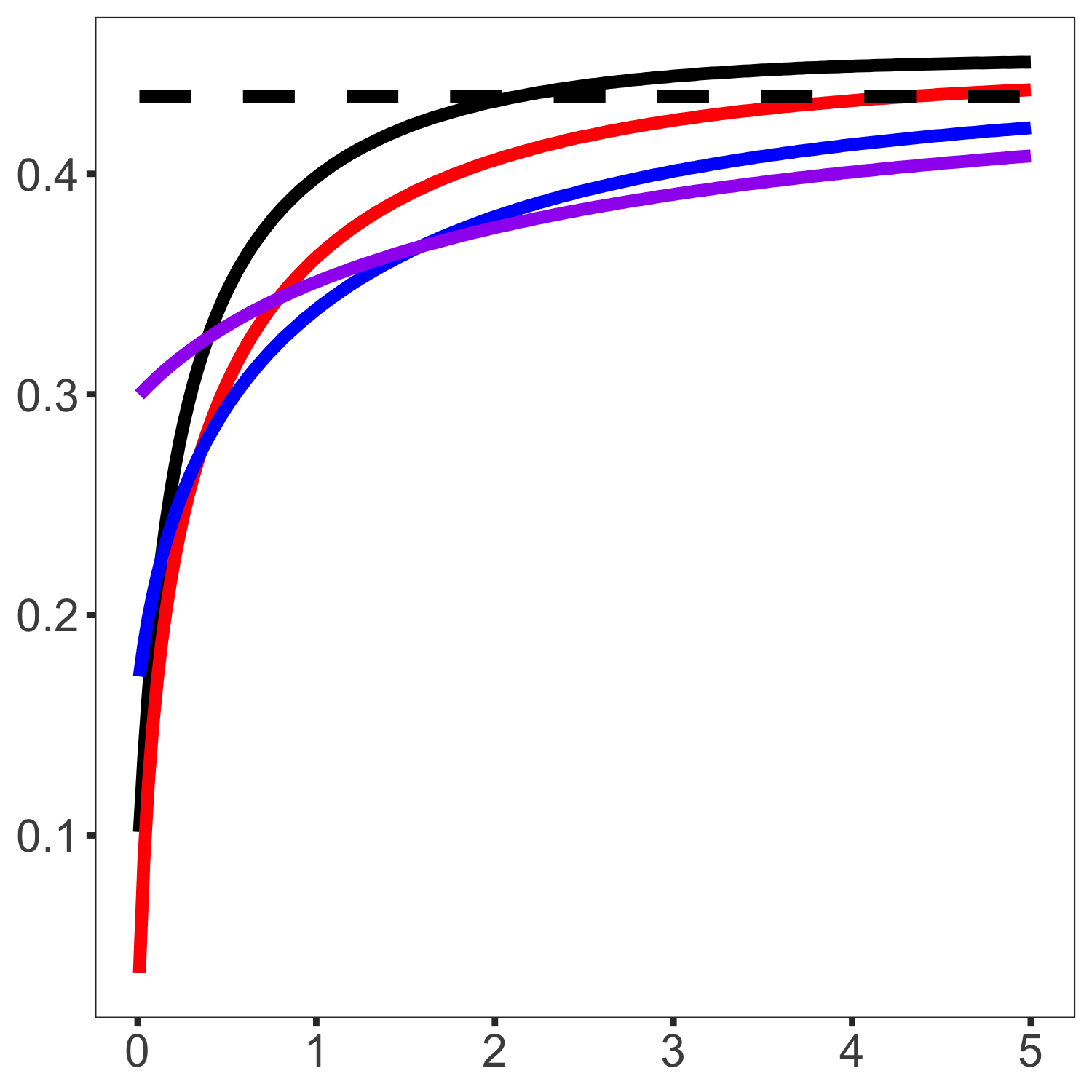}    
    \end{subfigure}
    \begin{subfigure}{0.24\textwidth}
    \includegraphics[width=\linewidth, height=0.6\linewidth]{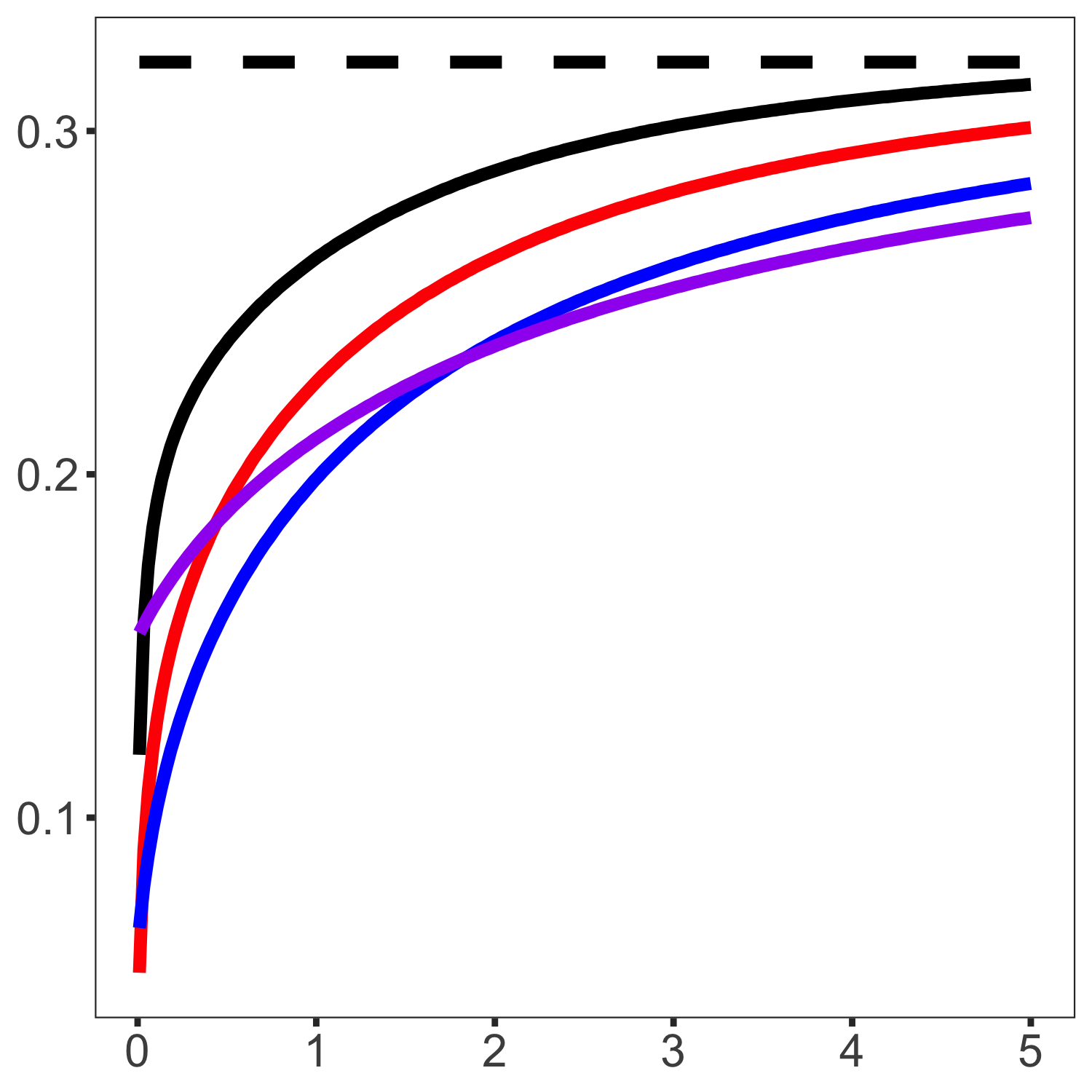}    
    \end{subfigure}
        \caption{$\SNR_p(\lambda)$ on $\lambda\in[0.01, 5]$ when $\gamma_1=0.5$. Columns: $\Sigma_p$ = Identity, Poly-Decay, AR-ACF, Point-Mix. Rows: $D \propto I_p$ (top), $D \propto \Sigma_p$ (bottom). Black/red/blue/purple: $\gamma_2 = 0.5$, $0.8$, $1.5$, $3$. Dashed: $\SNR_p(\infty)$ (independent of $\gamma_2$ and $\lambda$). }
    \label{fig:SNRvsLamda_Gamma1_05}
\end{figure}

\clearpage
\section{Simulation study: additional details}\label{sec:additional_simulation_results}

\subsection{Simulation settings and competing methods}\label{sec:simulation_settings}
The study considers the following configurations:
\begin{itemize}
    \item[(i)] The entries of the error matrix $\bZ$ are drawn from normal, t-distributions with 4 degrees of freedom, or Poisson distributions. However, as the results are consistent across these settings, only those for normal errors are reported. 
    \item[(ii)] With $n_2 = 500$, the dimension $p$ is set to $150$, $250$, $450$ or $1000$, corresponding to $\hat{\gamma}_2 = 0.3$, $0.5$, $0.9$, and $2$. 
    \item[(iii)] Three cases of $n_1$ are considered: $n_1\in\{50, 100, 250 \}$.  
    \item[(iv)] The eigenvalues $\tau_j$ of $\Sigma_p$ are set according to three models:
    \begin{itemize}
        \item[(a)] \emph{Identity}: $\Sigma_p = I_p$.
        \item[(b)] \emph{Poly-Decay}: $\tau_j = (1+j/p)^{-2}$, $j = 1, \dots, p$, representing polynomial decay.
        \item[(c)] \emph{AR-ACF}: $\Sigma_p$ is a Toeplitz matrix with $(i,j)$-th entry $0.3^{|i-j|}$, for $i,j = 1, \dots, p$.
        \item[(d)] \emph{Factor}: For $j \geq 6$, the eigenvalues $\tau_j$ match those from the \emph{Poly-Decay} model. The first five eigenvalues are spikes: $\tau_j = (2.2 - 0.2j) \tau_6$, for $j = 1, \dots, 5$. This represents a setting in which Condition~\ref{enum:edge_stability} is violated.
    \end{itemize}
    All models are further normalized so that $\tr(\Sigma_p) = p$. The eigenvector matrix of $\Sigma_p$ is randomly generated from the Haar distribution on the group of orthogonal matrices.
    \item[(v)] $X$ is a $n_1\times (n_1+n_2)$ matrix whose entries are iid $N(0,1)$. The coefficient matrix $B$ ($p\times n_1$) is assumed to be of low rank with its first column drawn from $N(0, \zeta^2)$ and all other columns set to zero. The parameter $\zeta$ controls the signal strength. We test $H_0: B = 0$ against $H_a: B\neq 0$. That is, $C = I_{n_1}$. 
\end{itemize}

Additionally, while the technical power analysis in Section~\ref{sec:power_analysis} focuses on a rank-one probabilistic alternative model, it is also of interest to examine the numerical performance of the proposed methods when the hypothesized matrix $BC$ has high rank, so that the signal is distributed across a collection of eigendirections. Empirical power results for additional settings in which $BC$ has rank $(n_1/2)$ are reported in Section~\ref{subsec:full_rank_simulation}.

Across all simulation settings, the estimators $\hat{\Theta}_1(\lambda)$ and $\hat{\Theta}_2(\lambda)$ are computed using Algorithms~\ref{algo:linear_program}, \ref{algo:estimation_rho}, and \ref{algo:ODE}, implemented with the recommended configurations detailed in Section~\ref{sec:implement_details}.

\medskip

All competing methods are summarized as follows. For the proposed method, we consider five choices of $\lambda$:
\begin{itemize}
\item[(1)] $\lambda=0.5$;
\item[(2)] $\lambda=1$;
\item[(3)] the Bayes selection when $D\propto I_p$, with $\mathfrak{X}$ chosen as recommended in Section~\ref{sec:selection_lambda};
\item[(4)] the Bayes selection when $D\propto \Sigma_p$, with $\mathfrak{X}$ chosen as recommended in Section~\ref{sec:selection_lambda}; and
\item[(5)] the estimated minimax selection when the prior family is
\[
\left\{
D_\theta=\theta_1 I_p+\theta_2 \Sigma_p
:\;
\theta_1\ge0,\;
\theta_2\ge0,\;
\theta_1+\theta_2\, p^{-1}\tr(\bW_2)=1
\right\},
\]
with $\mathfrak{X}$ again taken as in Section~\ref{sec:selection_lambda}.
\end{itemize}
The proposed method is compared with representative approaches from the literature. As discussed in the Introduction, most related methods fall into three broad categories. Since extreme-type tests—such as those of \citet{tony2014two} and \citet{chang2017simulation}—are not readily applicable to general linear hypotheses, we restrict attention to quadratic-form and projection-based procedures.

\begin{itemize}
\item[\textbf{Ridge-LRT}:] We consider the ridge-regularized likelihood ratio test of \citet{li2020high}, referred to as Ridge-LRT. The tests employs the same ridge regularization but aggregate all eigenvalues of $\bF_\lambda$ through a $\log(1+x)$ transformation. Specifically, we consider the Bayesian choice of the regularization parameter with the true prior presented in Section 3 of \cite{li2020high}.  

\item[\textbf{Proj-LRT}:] To benchmark projection-based methods, we note that most existing projection tests assume $n_1 \ll p$ and are not directly applicable to the general linear hypothesis setting considered here. We therefore construct a baseline procedure, denoted Proj-LRT, by projecting the observations onto random subspaces of dimension $\min(n_1,n_2,p)/10$, and then applying the likelihood ratio test to the projected data. A composite statistic is formed by taking the supremum over 20 such projections, with projection matrices drawn according to the Haar measure.

\item[\textbf{Han et. al. (2016)}:] We also include the test proposed by \citet{han2016tracy}, treated as the case of the proposed method when $\lambda\to0$. 
\end{itemize}

\subsection{Estimation accuracy: additional details}\label{sub:estimation_accuracy_additional_details}

Figure \ref{fig:example_s_fun2} and \ref{fig:example_s_fun3} show the estimated $s(x)$, $s'(x)$ and $s''(x)$ under two additional settings. Figure \ref{fig:density_curve_Theta1} shows the density plot of $p^{2/3}(\hat{\Theta}_1 -\Theta_1)/\Theta_2$ when $\Sigma$ is Poly-Decay. Table \ref{tab:Theta1_means_sd} and Table \ref{tab:Theta2_means_sd} present the estimation error of $\hat{\Theta}_1$ and $\hat{\Theta}_2$. 

 \begin{figure}[H]
\centering
\includegraphics[width=\textwidth, height = 0.25\textwidth]{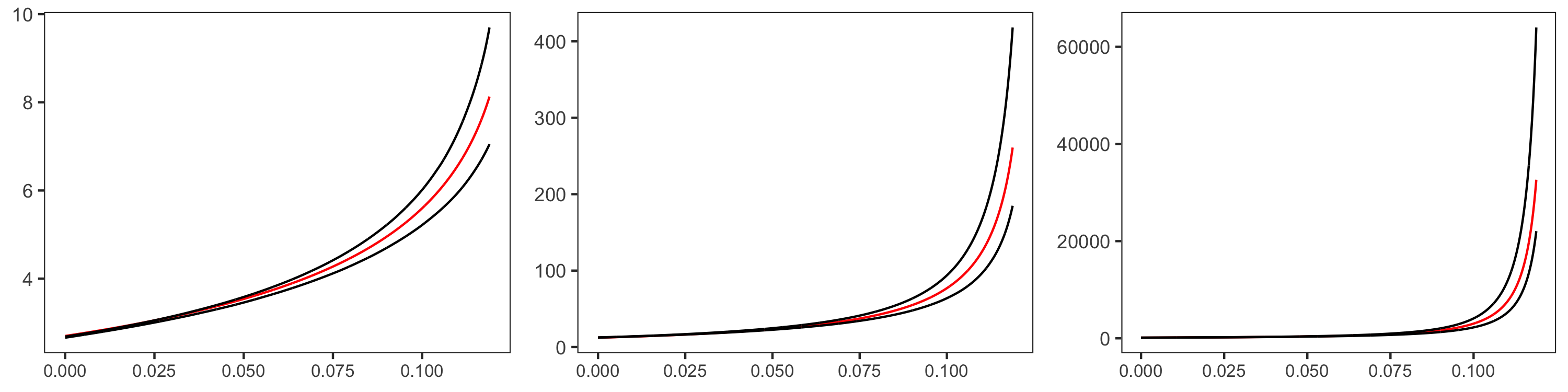}
\caption{The estimated $s(x)$, $s'(x)$ and $s''(x)$ (from left to right) when $\Sigma_p$ is Poly-Decay, $\hat{\gamma}_2 = 2$, $\lambda/\hat{\gamma}_2 = 0.1$. Red: the true functions; Black: 5\% and 95\% pointwise percentile bands of the estimated functions.}
\label{fig:example_s_fun2}
\end{figure}

\begin{figure}[H]
\centering
\includegraphics[width=\textwidth, height =0.25 \textwidth]{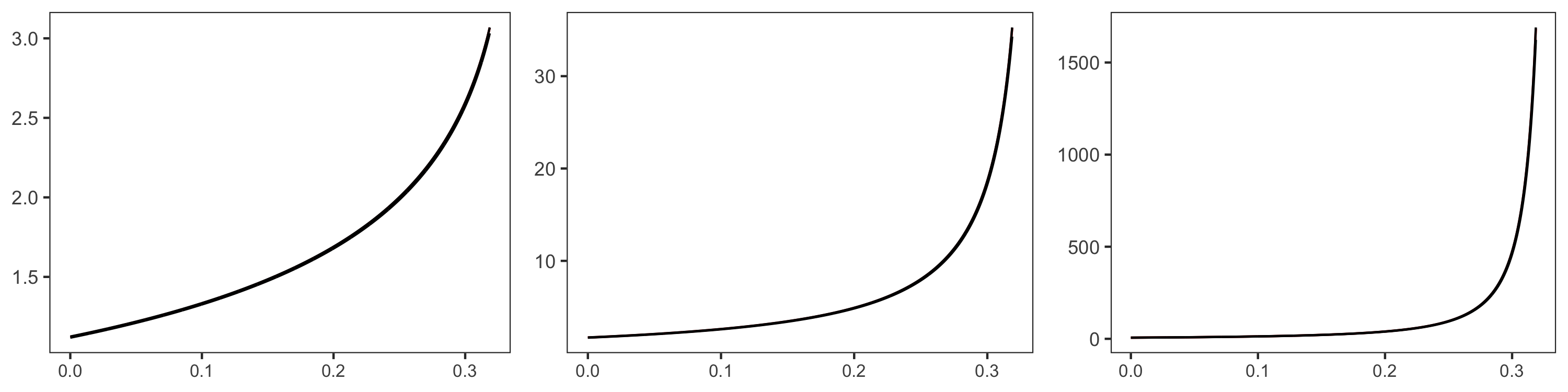}
\caption{The estimated $s(x)$, $s'(x)$ and $s''(x)$ (from left to right) when $\Sigma_p$ is Identity, $\hat{\gamma}_2 = 0.5$, $\lambda/\hat{\gamma}_2 = 0.5$. Red: the true functions; Black: 5\% and 95\% pointwise percentile bands of the estimated functions.}
\label{fig:example_s_fun3}
\end{figure}

\begin{figure}[H]
\centering
\includegraphics[width=\textwidth, height = 0.2\textwidth]{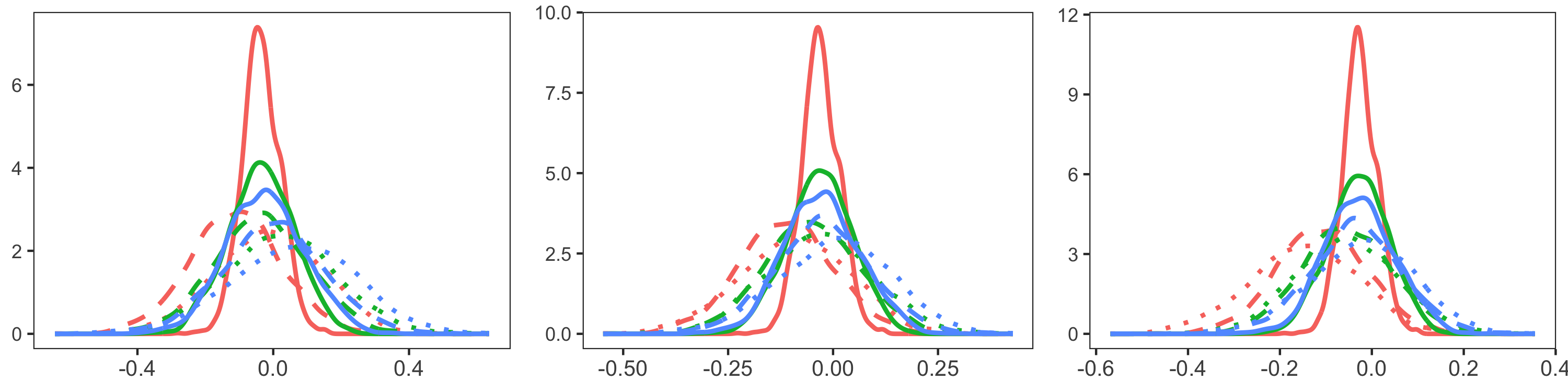}
\caption{Density of $p^{2/3}(\hat{\Theta}_1 -\Theta_1)/\Theta_2$ when $\Sigma_p$ is Poly-Decay. From left to right: $n_1 = 250$, $n_1 = 100$, $n_1 = 50$. Line type: $\hat\gamma_2 = 0.5$ (solid), $1$ (dashed), $2$ (dotted). Color: $\lambda/\gamma_2 = 0.5$ (red), $1$ (green), $1.5$ (blue).  }
\label{fig:density_curve_Theta1}
\end{figure}

\begin{table}[htbp]
\centering
\begin{tabular}{c c c| c c c}
\toprule
\textbf{$\Sigma_p$} & $\hat\gamma_2$ & $\lambda/\hat\gamma_2$ & $n_1 = 250$ & $ n_1 =100$ & $n_1 = 50$ \\ 
\midrule
\multirow{9}{*}{Identity} 
& \multirow{3}{*}{0.5} & 0.1 & 0.05 (0.04) & 0.04 (0.03) & 0.03 (0.03) \\
& & 0.5 & 0.09 (0.07) & 0.07 (0.06) & 0.06 (0.05) \\
& & 1.0 & 0.12 (0.09) & 0.09 (0.07) & 0.07 (0.06) \\
\cline{2-6}
& \multirow{3}{*}{1.0} & 0.1 & 0.15 (0.11) & 0.13 (0.09) & 0.12 (0.08) \\
& & 0.5 & 0.13 (0.09) & 0.10 (0.08) & 0.09 (0.07) \\
& & 1.0 & 0.14 (0.11) & 0.11 (0.08) & 0.09 (0.07) \\
\cline{2-6}
& \multirow{3}{*}{2.0} & 0.1 & 0.18 (0.13) & 0.24 (0.13) & 0.27 (0.12) \\
& & 0.5 & 0.14 (0.11) & 0.11 (0.08) & 0.10 (0.08) \\
& & 1.0 & 0.17 (0.13) & 0.12 (0.09) & 0.10 (0.07) \\
\midrule
\multirow{9}{*}{Poly-Decay} 
& \multirow{3}{*}{0.5} & 0.1 & 0.05 (0.04) & 0.04 (0.03) & 0.04 (0.03) \\
& & 0.5 & 0.08 (0.06) & 0.07 (0.05) & 0.06 (0.04) \\
& & 1.0 & 0.09 (0.07) & 0.08 (0.06) & 0.07 (0.05) \\
\cline{2-6}
& \multirow{3}{*}{1.0} & 0.1 & 0.14 (0.10) & 0.13 (0.09) & 0.12 (0.08) \\
& & 0.5 & 0.11 (0.09) & 0.10 (0.07) & 0.09 (0.07) \\
& & 1.0 & 0.12 (0.09) & 0.09 (0.07) & 0.08 (0.06) \\
\cline{2-6}
& \multirow{3}{*}{2.0} & 0.1 & 0.12 (0.09) & 0.13 (0.09) & 0.16 (0.10) \\
& & 0.5 & 0.13 (0.10) & 0.10 (0.08) & 0.10 (0.07) \\
& & 1.0 & 0.16 (0.11) & 0.11 (0.08) & 0.09 (0.07) \\
\midrule
\multirow{9}{*}{AR-ACF} 
& \multirow{3}{*}{0.5} & 0.1 & 0.05 (0.04) & 0.04 (0.03) & 0.04 (0.03) \\
& & 0.5 & 0.08 (0.06) & 0.07 (0.05) & 0.06 (0.04) \\
& & 1.0 & 0.09 (0.07) & 0.08 (0.06) & 0.07 (0.05) \\
\cline{2-6}
& \multirow{3}{*}{1.0} & 0.1 & 0.14 (0.10) & 0.13 (0.09) & 0.12 (0.08) \\
& & 0.5 & 0.11 (0.08) & 0.10 (0.07) & 0.09 (0.06) \\
& & 1.0 & 0.12 (0.09) & 0.10 (0.07) & 0.08 (0.06) \\
\cline{2-6}
& \multirow{3}{*}{2.0} & 0.1 & 0.12 (0.09) & 0.13 (0.09) & 0.16 (0.10) \\
& & 0.5 & 0.13 (0.10) & 0.10 (0.08) & 0.10 (0.07) \\
& & 1.0 & 0.16 (0.12) & 0.11 (0.08) & 0.09 (0.07) \\
\midrule
\multirow{9}{*}{Factor} 
& \multirow{3}{*}{0.5} & 0.1 & 0.06 (0.04) & 0.05 (0.04) & 0.04 (0.03) \\
& & 0.5 & 0.08 (0.06) & 0.07 (0.05) & 0.06 (0.04) \\
& & 1.0 & 0.09 (0.07) & 0.07 (0.06) & 0.07 (0.05) \\
\cline{2-6}
& \multirow{3}{*}{1.0} & 0.1 & 0.15 (0.10) & 0.13 (0.09) & 0.12 (0.08) \\
& & 0.5 & 0.11 (0.08) & 0.09 (0.07) & 0.09 (0.06) \\
& & 1.0 & 0.12 (0.09) & 0.09 (0.07) & 0.08 (0.06) \\
\cline{2-6}
& \multirow{3}{*}{2.0} & 0.1 & 0.13 (0.10) & 0.13 (0.09) & 0.15 (0.10) \\
& & 0.5 & 0.13 (0.10) & 0.10 (0.08) & 0.10 (0.07) \\
& & 1.0 & 0.15 (0.11) & 0.10 (0.08) & 0.09 (0.07) \\
\bottomrule
\end{tabular}
\caption{Mean (standard deviation) of $p^{2/3} |\hat{\Theta}_1 -\Theta_1|/\Theta_2$ under different settings. }
\label{tab:Theta1_means_sd}
\end{table}

\begin{table}[!h]
\centering
\begin{tabular}{c c c| c c c}
\toprule
\textbf{$\Sigma_p$} & $\hat\gamma_2$ & $\lambda/\hat\gamma_2$ & $n_1 = 250$ & $ n_1 =100$ & $n_1 = 50$ \\ 
\midrule
\multirow{9}{*}{Identity} & \multirow{3}{*}{0.5} & 0.1 & 0.19 (0.16) & 0.20 (0.15) & 0.18 (0.14) \\
& & 0.5 & 0.28 (0.22) & 0.24 (0.19) & 0.23 (0.18) \\
& & 1.0 & 0.33 (0.25) & 0.27 (0.21) & 0.25 (0.19) \\
\cline{2-6}
& \multirow{3}{*}{1.0} & 0.1 & 0.24 (0.19) & 0.24 (0.18) & 0.27 (0.19) \\
& & 0.5 & 0.32 (0.24) & 0.24 (0.18) & 0.21 (0.16) \\
& & 1.0 & 0.41 (0.32) & 0.29 (0.21) & 0.25 (0.19) \\
\cline{2-6}
& \multirow{3}{*}{2.0} & 0.1 & 0.60 (0.33) & 0.46 (0.27) & 0.34 (0.22) \\
& & 0.5 & 0.46 (0.33) & 0.29 (0.21) & 0.21 (0.16) \\
& & 1.0 & 0.50 (0.37) & 0.29 (0.22) & 0.20 (0.16) \\
\midrule
\multirow{9}{*}{Poly-Decay} & \multirow{3}{*}{0.5} & 0.1 & 0.17 (0.14) & 0.18 (0.14) & 0.17 (0.13) \\
& & 0.5 & 0.21 (0.17) & 0.19 (0.15) & 0.19 (0.14) \\
& & 1.0 & 0.25 (0.19) & 0.20 (0.16) & 0.18 (0.14) \\
\cline{2-6}
& \multirow{3}{*}{1.0} & 0.1 & 0.20 (0.16) & 0.20 (0.15) & 0.24 (0.18) \\
& & 0.5 & 0.26 (0.20) & 0.20 (0.15) & 0.18 (0.14) \\
& & 1.0 & 0.32 (0.25) & 0.23 (0.17) & 0.19 (0.15) \\
\cline{2-6}
& \multirow{3}{*}{2.0} & 0.1 & 0.64 (0.36) & 0.43 (0.26) & 0.29 (0.19) \\
& & 0.5 & 0.37 (0.27) & 0.25 (0.19) & 0.18 (0.15) \\
& & 1.0 & 0.37 (0.27) & 0.24 (0.19) & 0.17 (0.14) \\
\midrule
\multirow{9}{*}{AR-ACF} & \multirow{3}{*}{0.5} & 0.1 & 0.18 (0.15) & 0.19 (0.15) & 0.19 (0.14) \\
& & 0.5 & 0.22 (0.18) & 0.20 (0.16) & 0.20 (0.15) \\
& & 1.0 & 0.26 (0.20) & 0.21 (0.17) & 0.19 (0.15) \\
\cline{2-6}
& \multirow{3}{*}{1.0} & 0.1 & 0.22 (0.18) & 0.22 (0.17) & 0.25 (0.19) \\
& & 0.5 & 0.27 (0.22) & 0.21 (0.17) & 0.19 (0.15) \\
& & 1.0 & 0.34 (0.27) & 0.24 (0.18) & 0.20 (0.16) \\
\cline{2-6}
& \multirow{3}{*}{2.0} & 0.1 & 0.65 (0.36) & 0.43 (0.27) & 0.29 (0.20) \\
& & 0.5 & 0.38 (0.28) & 0.25 (0.19) & 0.19 (0.15) \\
& & 1.0 & 0.38 (0.28) & 0.25 (0.19) & 0.17 (0.14) \\
\midrule
\multirow{9}{*}{Factor} & \multirow{3}{*}{0.5} & 0.1 & 0.17 (0.14) & 0.19 (0.14) & 0.19 (0.14) \\
& & 0.5 & 0.22 (0.17) & 0.20 (0.16) & 0.19 (0.15) \\
& & 1.0 & 0.26 (0.19) & 0.21 (0.17) & 0.19 (0.15) \\
\cline{2-6}
& \multirow{3}{*}{1.0} & 0.1 & 0.20 (0.16) & 0.20 (0.15) & 0.24 (0.18) \\
& & 0.5 & 0.27 (0.20) & 0.20 (0.15) & 0.18 (0.14) \\
& & 1.0 & 0.33 (0.25) & 0.23 (0.17) & 0.19 (0.15) \\
\cline{2-6}
& \multirow{3}{*}{2.0} & 0.1 & 0.63 (0.36) & 0.43 (0.26) & 0.29 (0.19) \\
& & 0.5 & 0.37 (0.27) & 0.25 (0.19) & 0.18 (0.15) \\
& & 1.0 & 0.37 (0.27) & 0.24 (0.18) & 0.17 (0.14) \\
\bottomrule
\end{tabular}
\caption{Mean (standard deviation) of $p^{2/3} |\hat{\Theta}_2 -\Theta_2|/\Theta_2$ under different settings. }
\label{tab:Theta2_means_sd}
\end{table}

\newpage
\clearpage

\subsection{Empirical null distribution: additional details}\label{sub:empirical_null_distribution_additional_details}
In Figures~\ref{fig:emp_density_lambda1}--\ref{fig:emp_density_lambda5}, we display the empirical density of $\ell_{\max}(\bF_\lambda)$ under representative settings. The remaining configurations are omitted, as they exhibit similar patterns. The reported cases include both fixed and data-driven choices of $\lambda$ and span all considered settings of $\Sigma_p$, $\gamma_1$, and $\gamma_2$.

\begin{figure}[H]
\centering
\includegraphics[width =\textwidth]{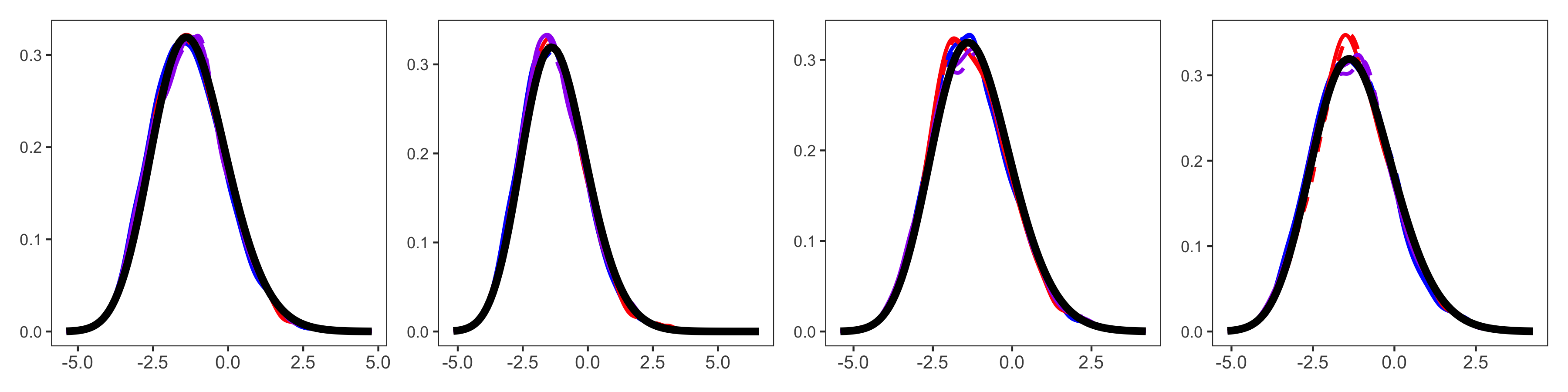}
\caption{Empirical density of $\ell_{\max}(\bF_\lambda)$ with $\lambda = 0.5$ and $\Sigma_p$ is Identity. Solid curves use true parameters for normalization; dashed curves use estimated ones. Colors indicate $n_1$: blue = 50, red = 100, purple = 250. Columns (left to right) correspond to $\hat{\gamma}_2=0.3,\,0.5,\,0.9,$ and $2$. Black solid: density of $\TW_1$.}
\label{fig:emp_density_lambda1}
\end{figure}

\begin{figure}[H]
\centering
\includegraphics[width =\textwidth]{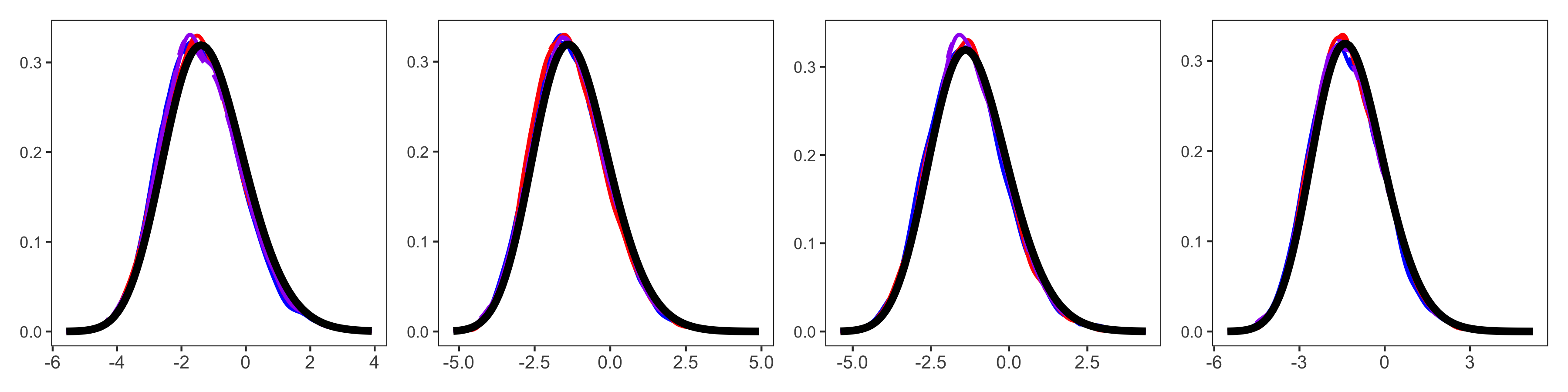}
\caption{Empirical density of $\ell_{\max}(\bF_\lambda)$ with $\lambda = 1$ and $\Sigma_p$ is Poly-Decay. Solid curves use true parameters for normalization; dashed curves use estimated ones. Colors indicate $n_1$: blue = 50, red = 100, purple = 250. Columns (left to right) correspond to $\hat{\gamma}_2=0.3,\,0.5,\,0.9,$ and $2$. Black solid: density of $\TW_1$.}
\label{fig:emp_density_lambda2}
\end{figure}

\begin{figure}[H]
\centering
\includegraphics[width =\textwidth]{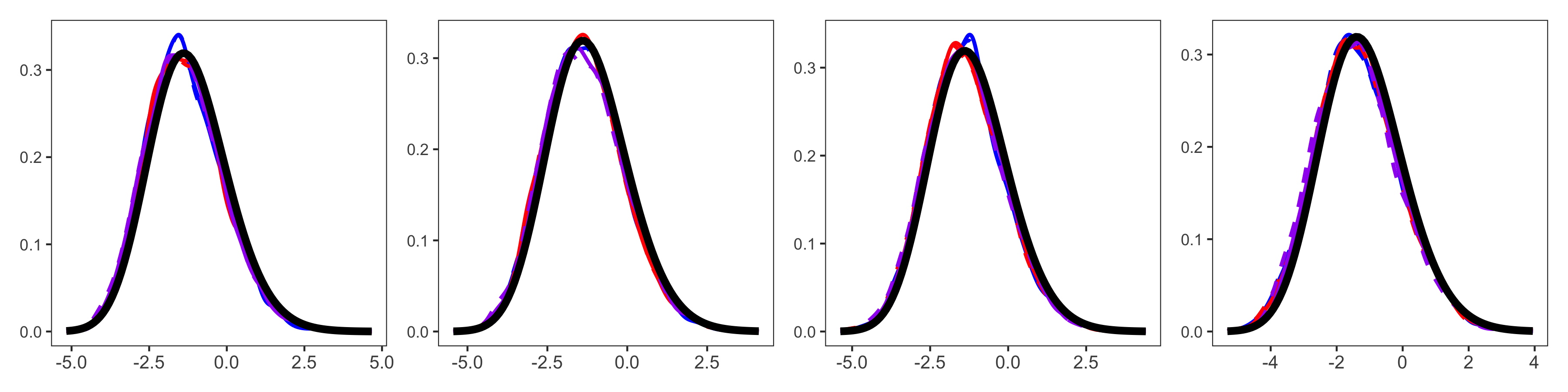}
\caption{Empirical density of $\ell_{\max}(\bF_\lambda)$ with $\lambda = \hat{\lambda}_{I_p}$ and $\Sigma_p$ is AR-ACF. Solid curves use true parameters for normalization; dashed curves use estimated ones. Colors indicate $n_1$: blue = 50, red = 100, purple = 250. Columns (left to right) correspond to $\hat{\gamma}_2=0.3,\,0.5,\,0.9,$ and $2$. Black solid: density of $\TW_1$.}
\label{fig:emp_density_lambda3}
\end{figure}

\begin{figure}[H]
\centering
\includegraphics[width =\textwidth]{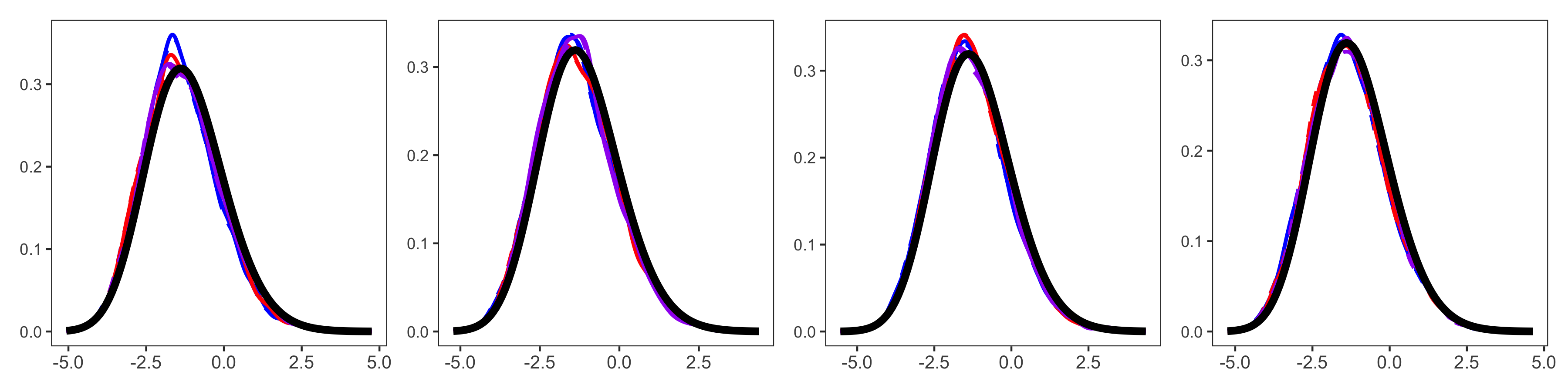}
\caption{Empirical density of $\ell_{\max}(\bF_\lambda)$ with $\lambda = \hat{\lambda}_{\Sigma_p}$ and $\Sigma_p$ is Factor. Solid curves use true parameters for normalization; dashed curves use estimated ones. Colors indicate $n_1$: blue = 50, red = 100, purple = 250. Columns (left to right) correspond to $\hat{\gamma}_2=0.3,\,0.5,\,0.9,$ and $2$. Black solid: density of $\TW_1$.}
\label{fig:emp_density_lambda4}
\end{figure}

\begin{figure}[H]
\centering
\includegraphics[width =\textwidth]{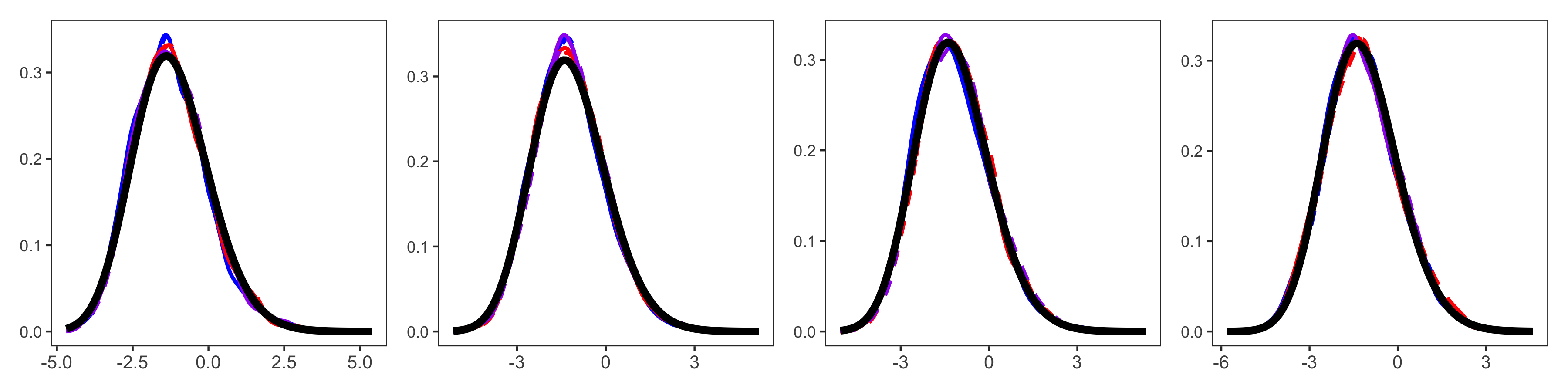}
\caption{Empirical density of $\ell_{\max}(\bF_\lambda)$ with $\lambda = \hat{\lambda}_*$ and $\Sigma_p$ is Factor. Solid curves use true parameters for normalization; dashed curves use estimated ones. Colors indicate $n_1$: blue = 50, red = 100, purple = 250. Columns (left to right) correspond to $\hat{\gamma}_2=0.3,\,0.5,\,0.9,$ and $2$. Black solid: density of $\TW_1$.}
\label{fig:emp_density_lambda5}
\end{figure}

\newpage
\clearpage

\subsection{Empirical power: additional details}\label{sub:empirical_power_additional_details}

Figures \ref{fig:emp_power_Sigma1}, \ref{fig:emp_power_Sigma3}, and \ref{fig:emp_power_Sigma4} show the size-adjusted empirical power of the proposed tests when $\Sigma$ is Identity, AR-ACF, and Factor. 
\begin{figure}[htbp]
    \centering
    \begin{subfigure}[t]{0.23\textwidth}
        \centering
         \includegraphics[width=\linewidth, height=0.7\linewidth]{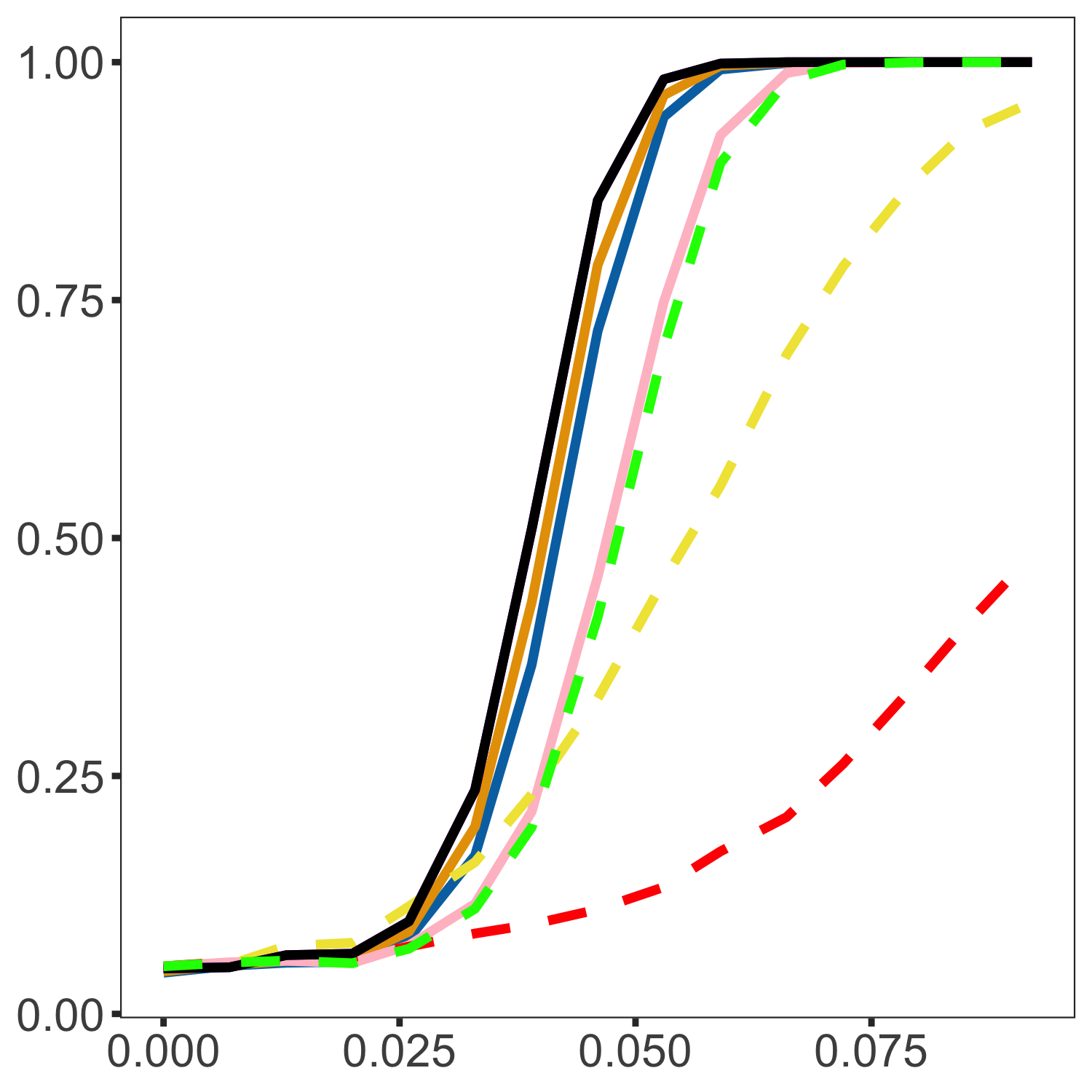}
    \end{subfigure}%
    \begin{subfigure}[t]{0.23\textwidth}
        \centering
        \includegraphics[width=\linewidth, height=0.7\linewidth]{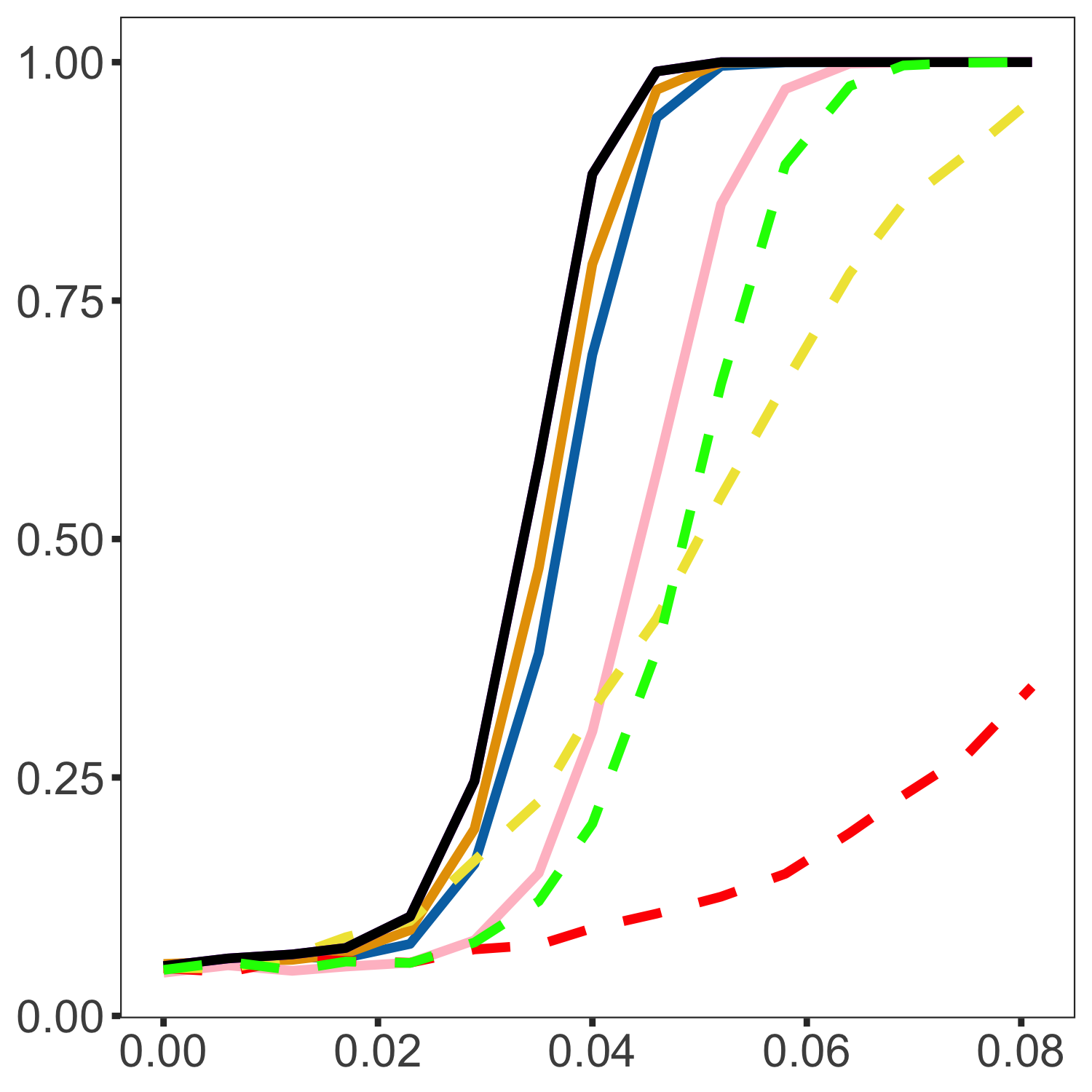}
    \end{subfigure}
     \begin{subfigure}[t]{0.23\textwidth}
        \centering
        \includegraphics[width=\linewidth, height=0.7\linewidth]{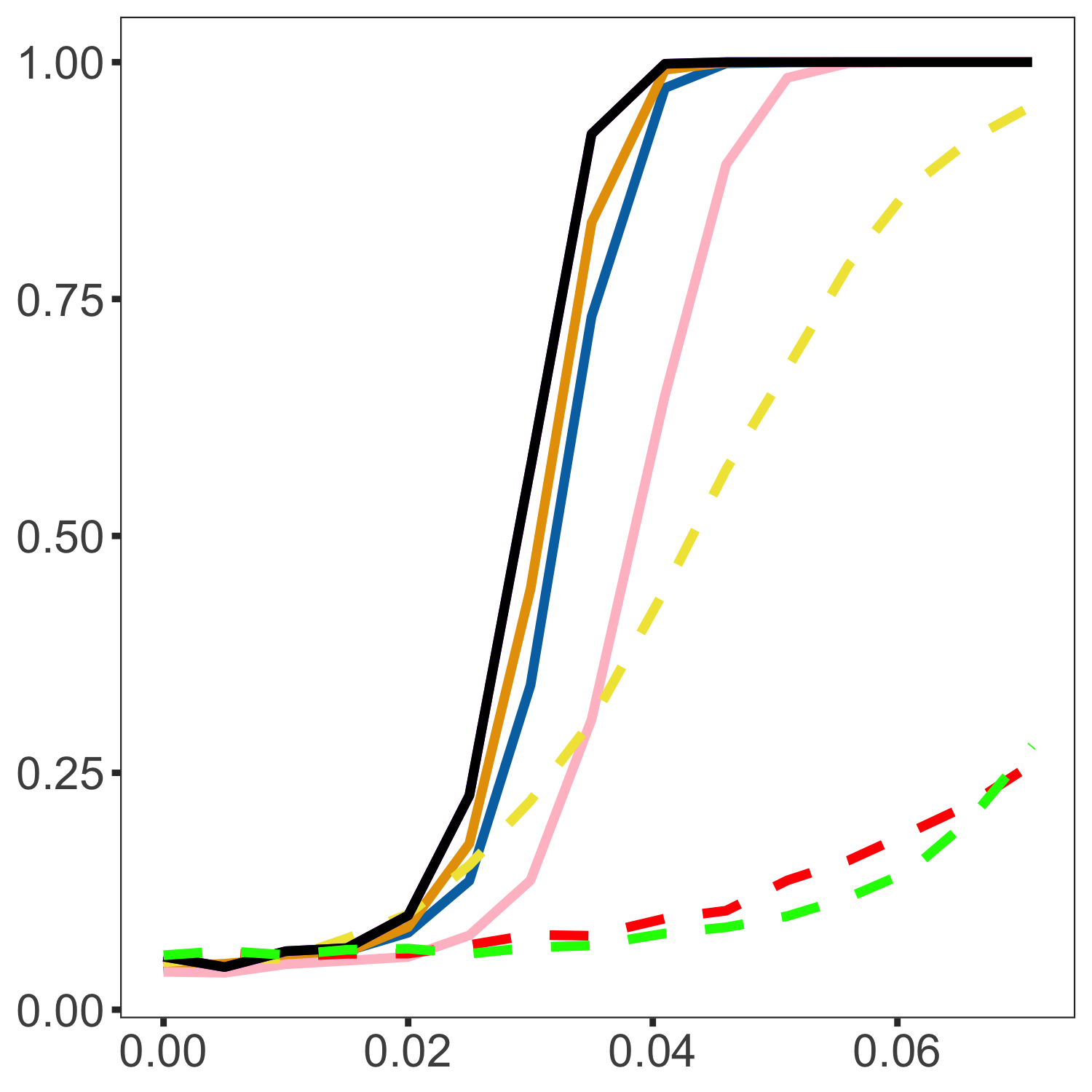}
    \end{subfigure}
    \begin{subfigure}[t]{0.23\textwidth}
        \centering
        \includegraphics[width=\linewidth, height=0.7\linewidth]{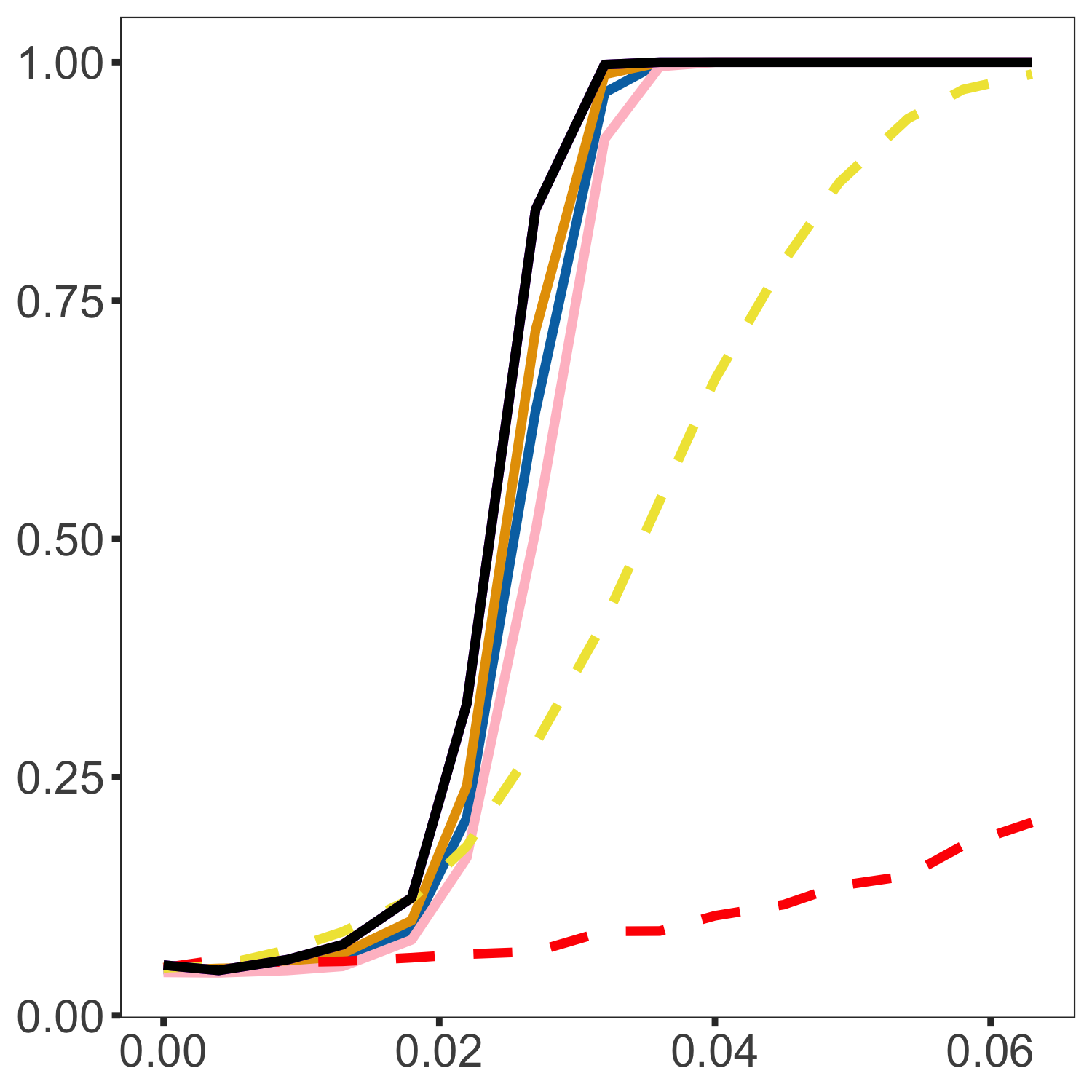}
    \end{subfigure}
    \vfill
    \begin{subfigure}[t]{0.23\textwidth}
        \centering
         \includegraphics[width=\linewidth, height=0.7\linewidth]{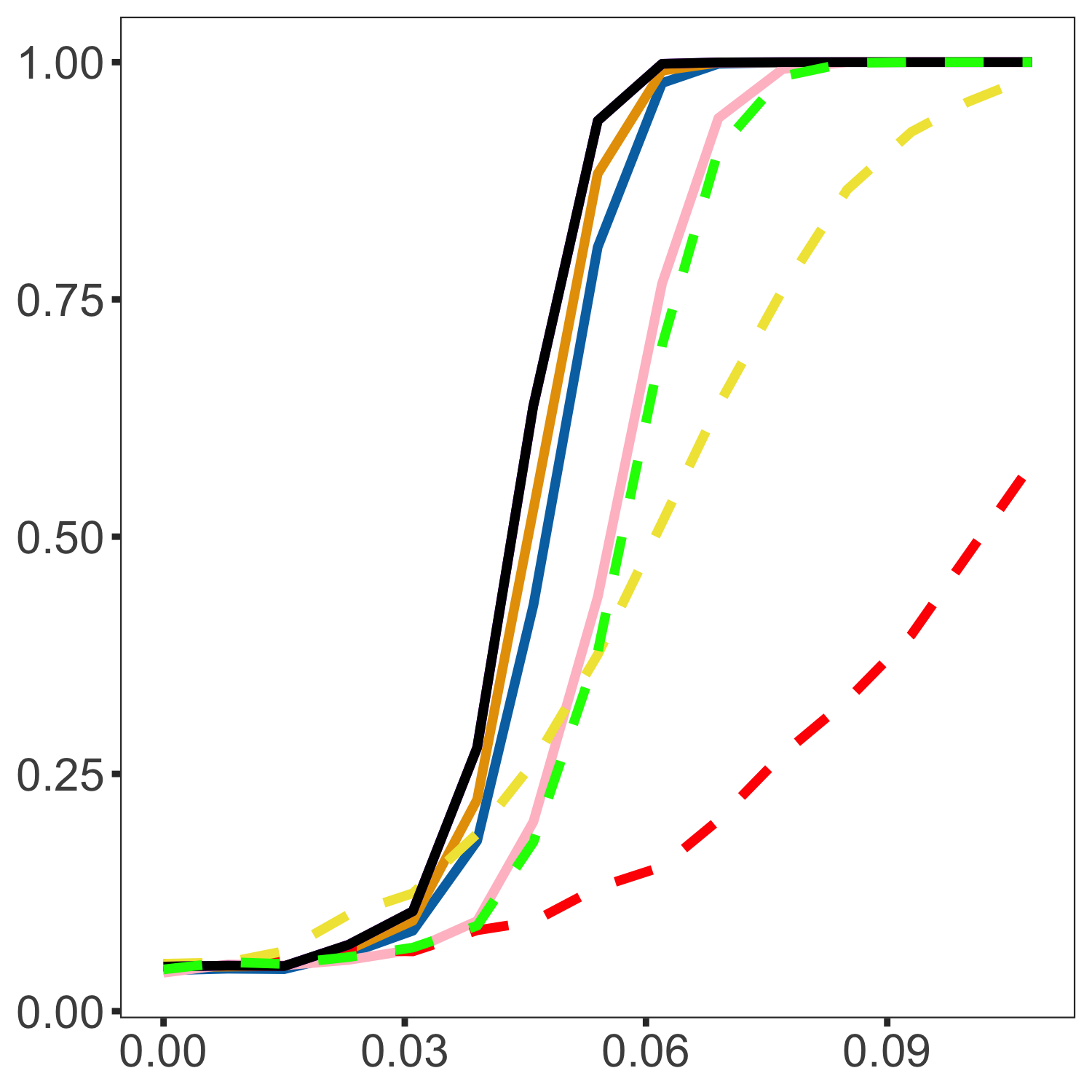}
    \end{subfigure}%
    \begin{subfigure}[t]{0.23\textwidth}
        \centering
        \includegraphics[width=\linewidth, height=0.7\linewidth]{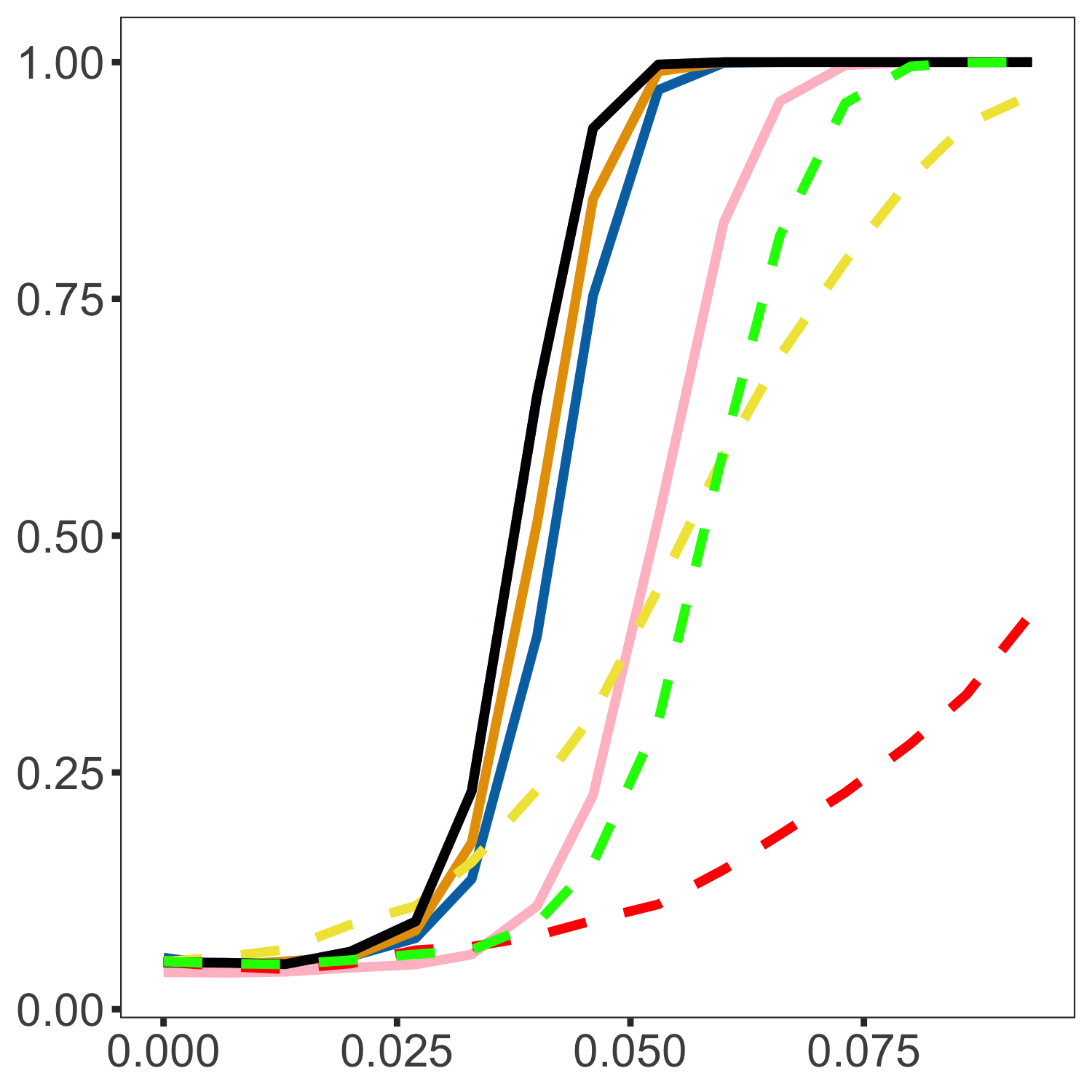}
    \end{subfigure}
     \begin{subfigure}[t]{0.23\textwidth}
        \centering
        \includegraphics[width=\linewidth, height=0.7\linewidth]{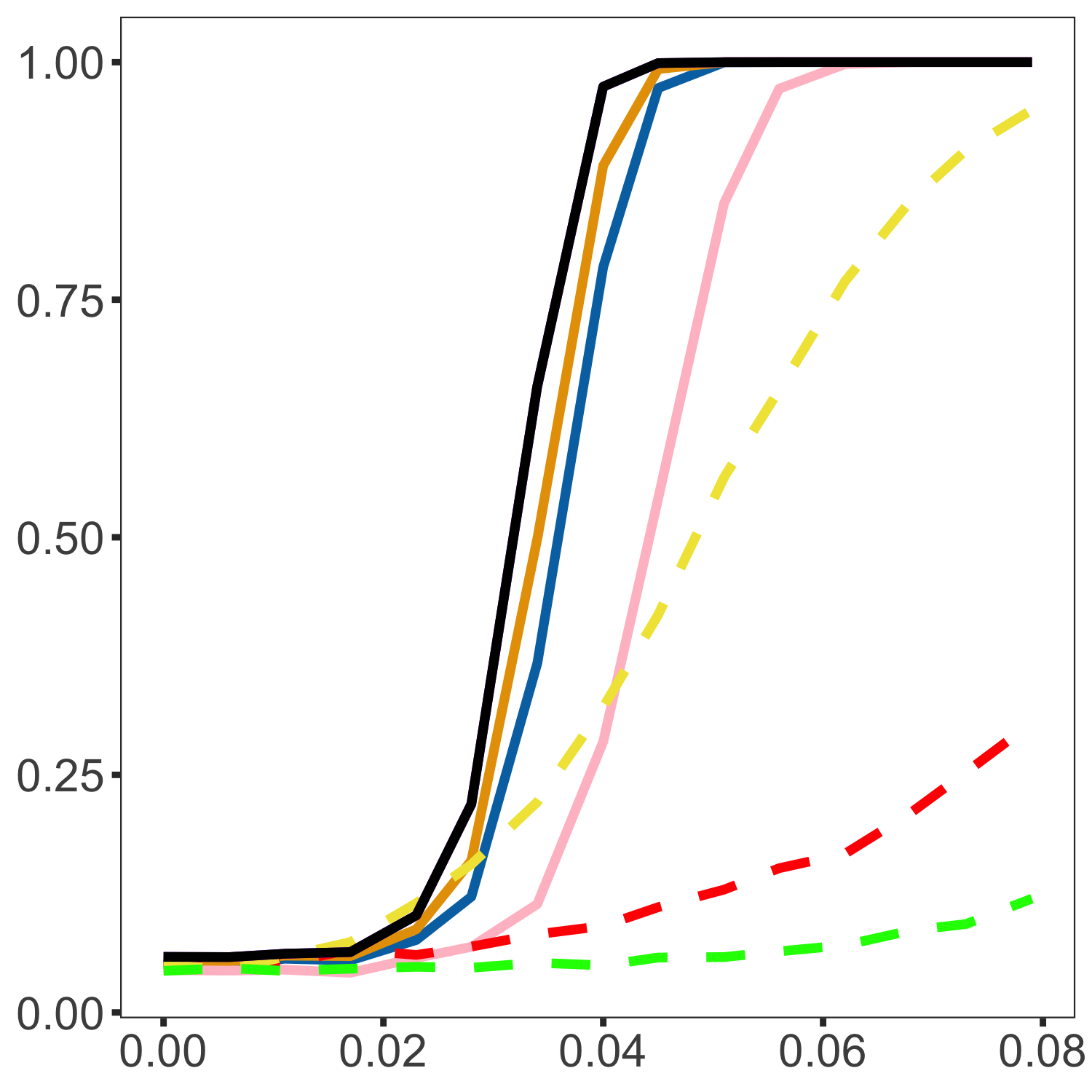}
    \end{subfigure}
    \begin{subfigure}[t]{0.23\textwidth}
        \centering
        \includegraphics[width=\linewidth, height=0.7\linewidth]{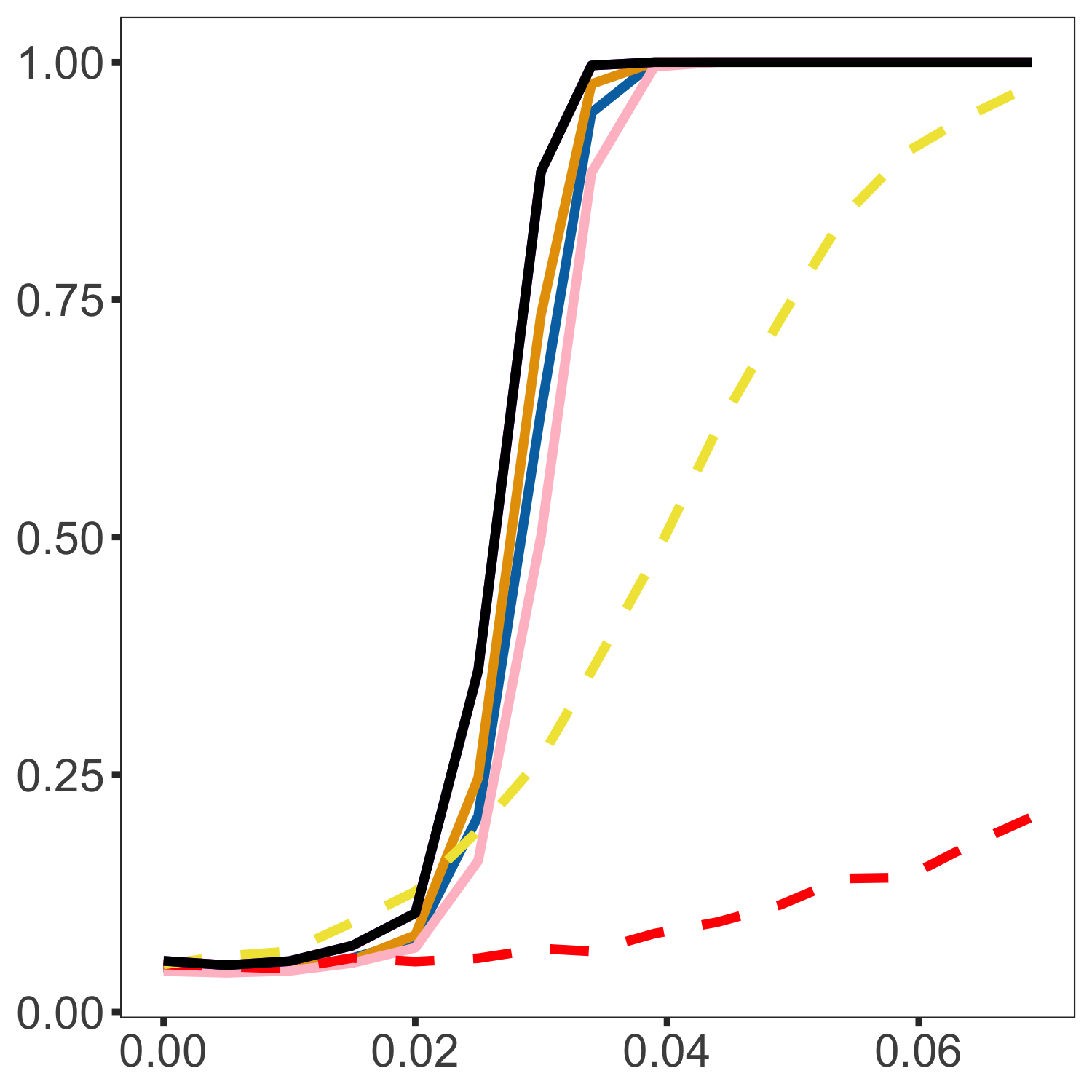}
    \end{subfigure}
    \vfill
    \begin{subfigure}[t]{0.23\textwidth}
        \centering
         \includegraphics[width=\linewidth, height=0.7\linewidth]{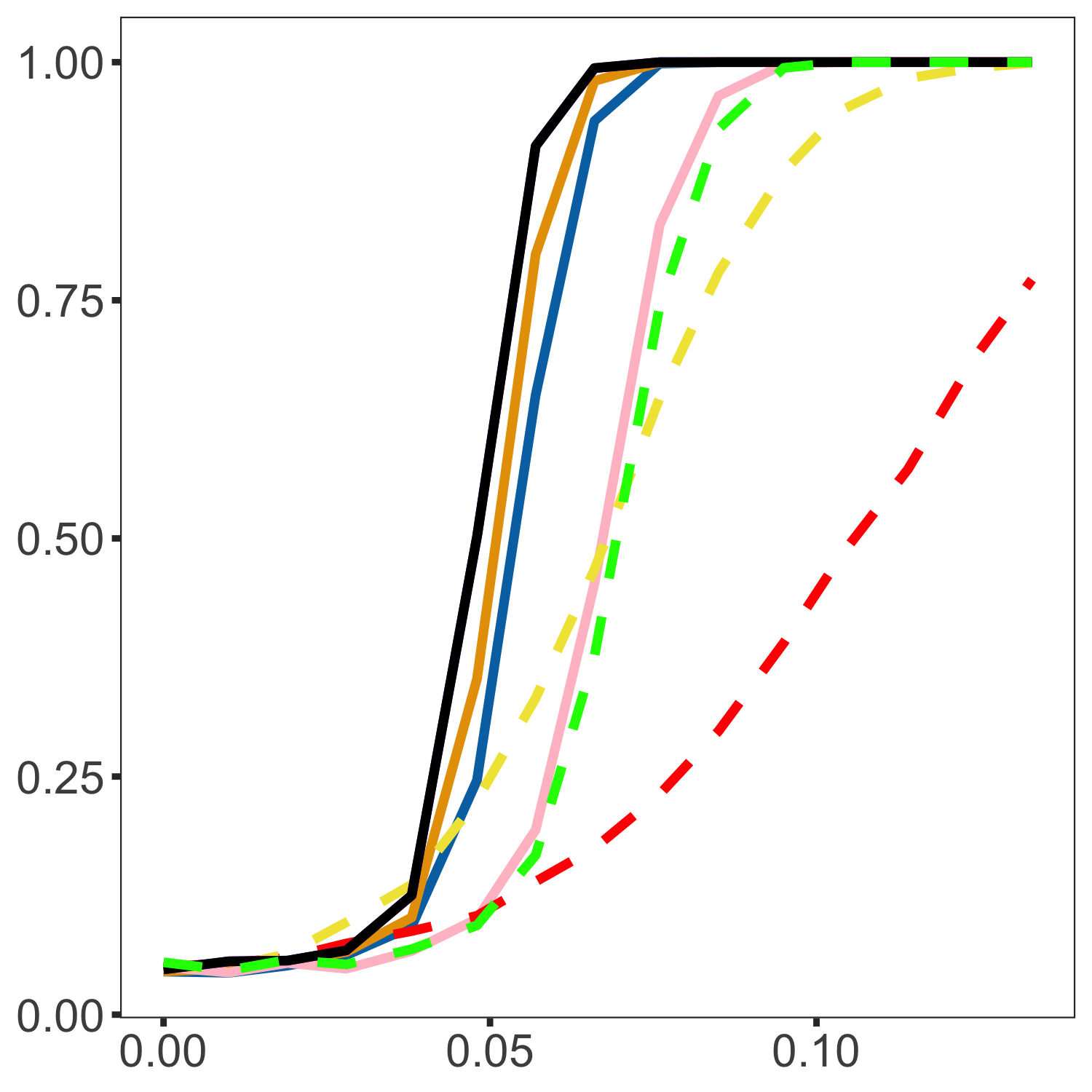}
    \end{subfigure}%
    \begin{subfigure}[t]{0.23\textwidth}
        \centering
        \includegraphics[width=\linewidth, height=0.7\linewidth]{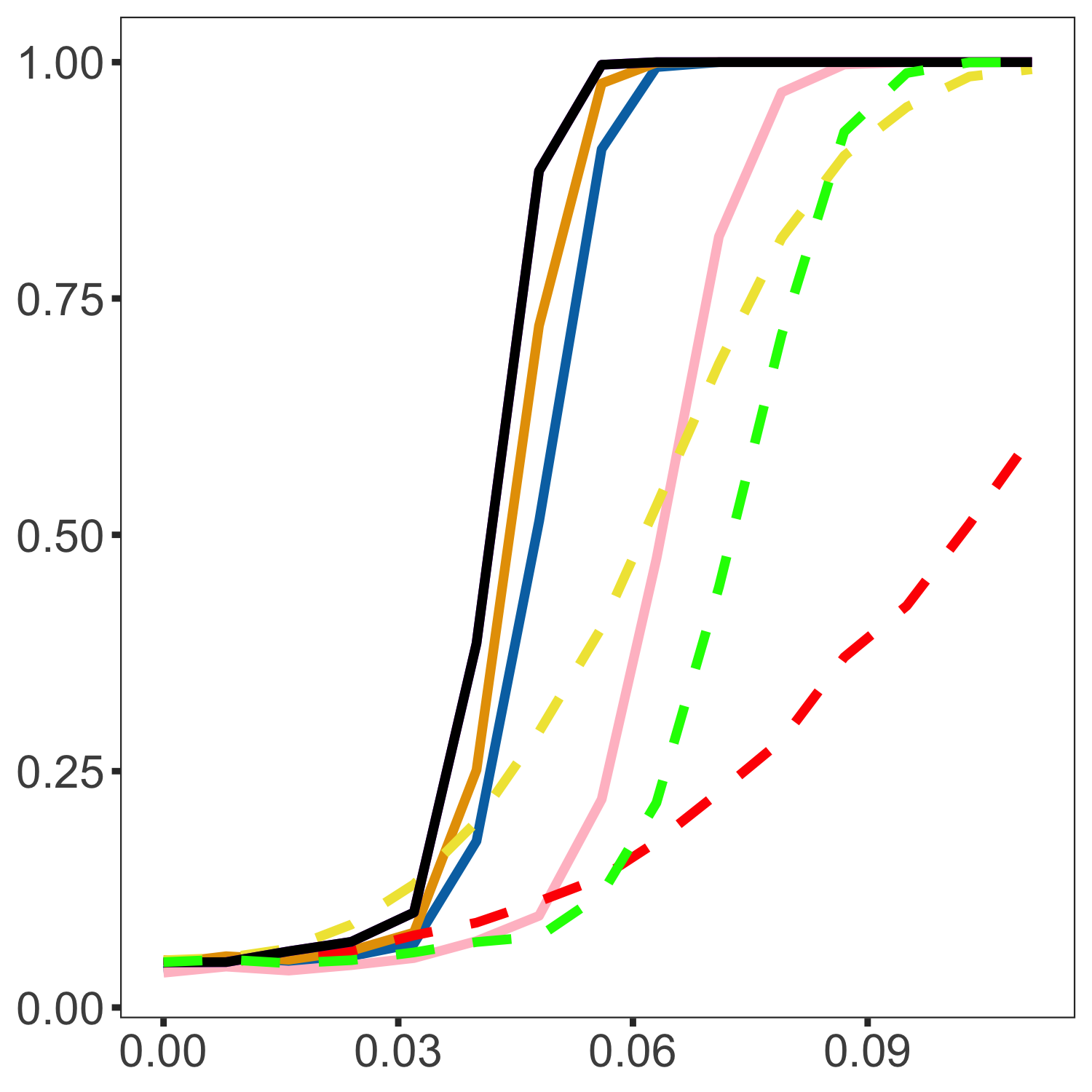}
    \end{subfigure}
     \begin{subfigure}[t]{0.23\textwidth}
        \centering
        \includegraphics[width=\linewidth, height=0.7\linewidth]{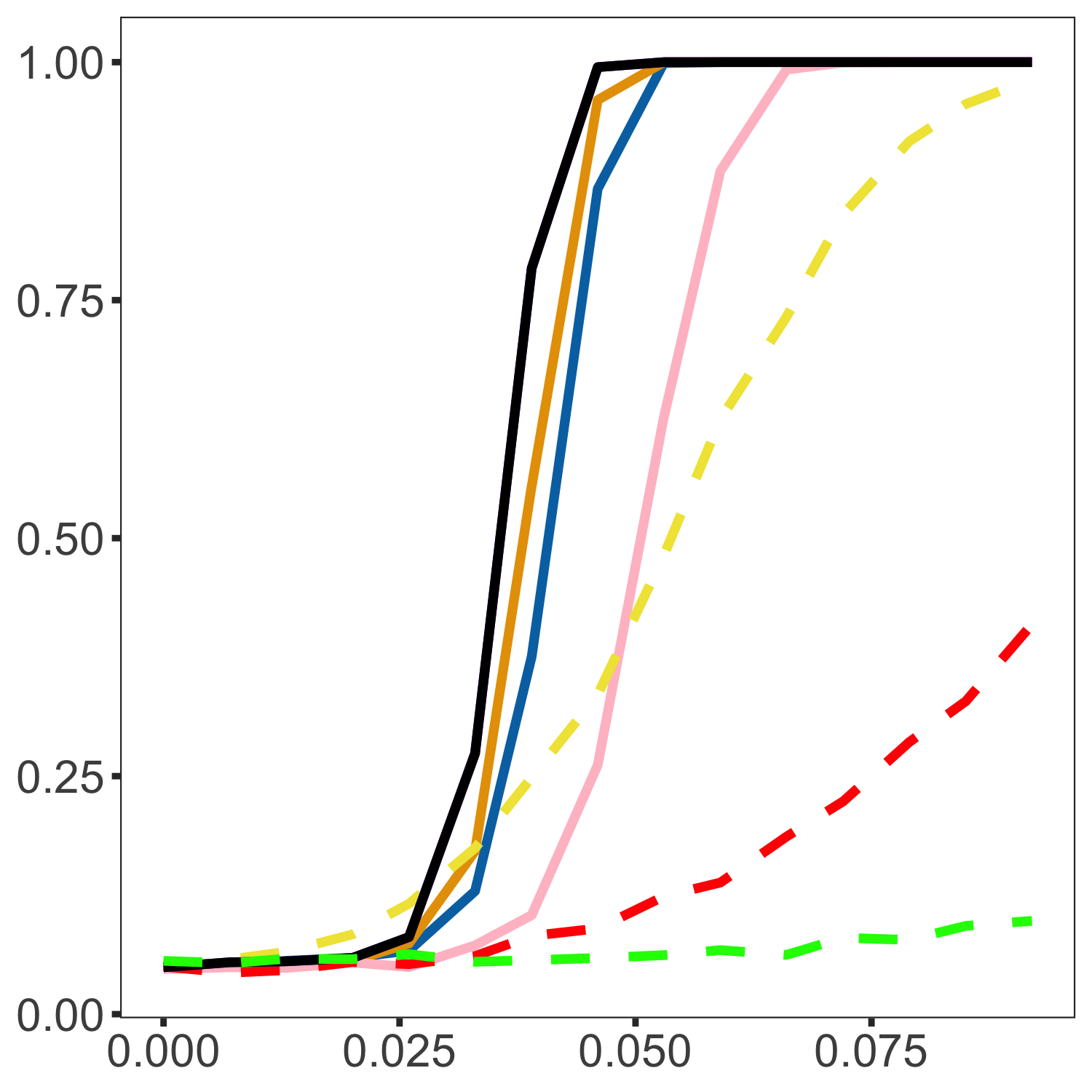}
    \end{subfigure}
    \begin{subfigure}[t]{0.23\textwidth}
        \centering
        \includegraphics[width=\linewidth, height=0.7\linewidth]{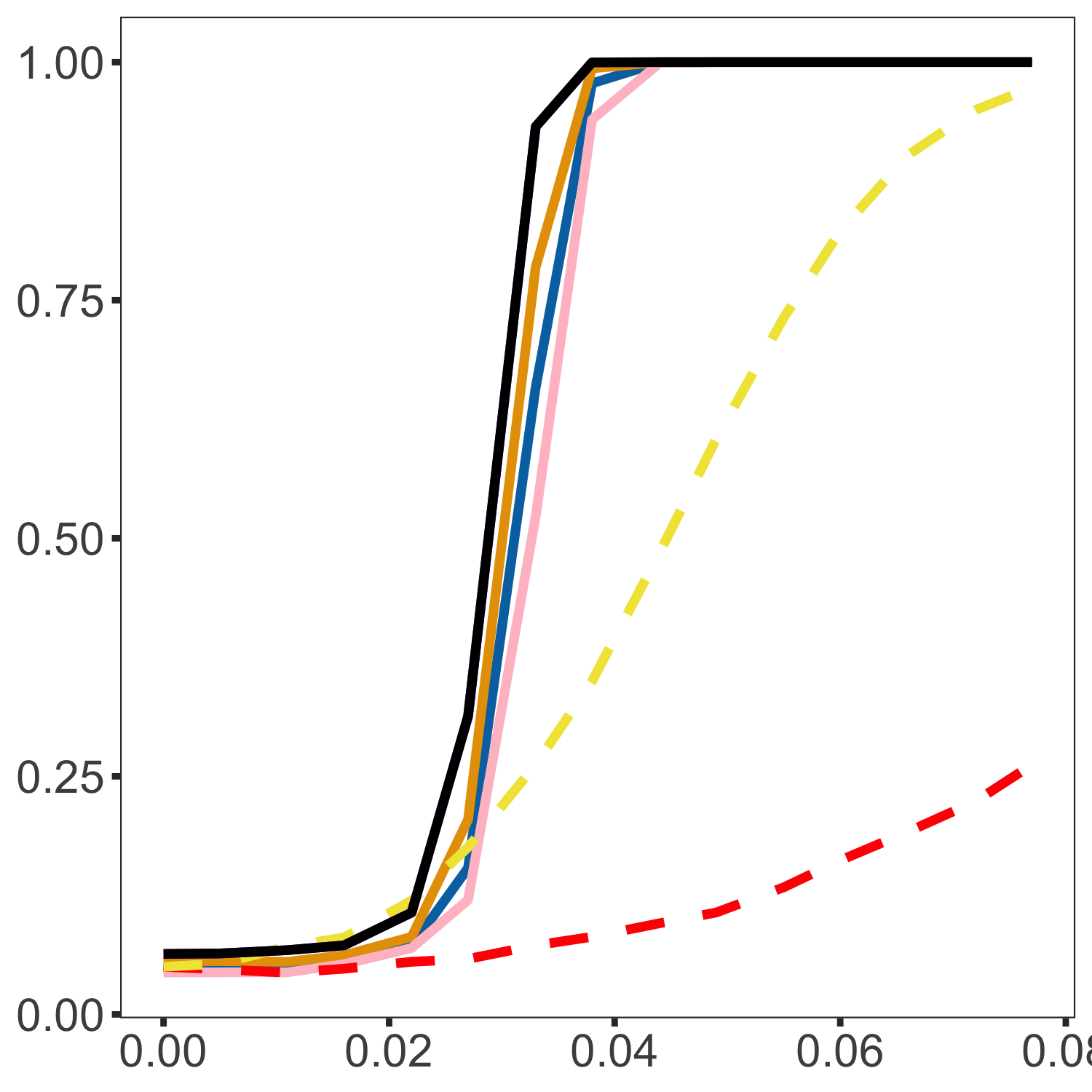}
    \end{subfigure}
   \caption{Size-adjusted empirical power when \(\Sigma\) is Identity. Columns (left to right) correspond to \(\hat{\gamma}_2 = 0.3, 0.5, 0.9, 2\); rows (top to bottom) correspond to \(n_1 = 50, 100, 250\). Solid curves: blue (\(\lambda = 0.5\)), orange (\(\lambda = 1\)), black (\(\lambda=\hat{\lambda}_{I_p}\)), purple (\(\lambda=\hat{\lambda}_{\Sigma_p}\)), and pink (\(\lambda=\hat{\lambda}_*\)). Dashed curves: red (Proj-LRT), yellow (Ridge-LRT), and green (\cite{han2016tracy}, \(\lambda=0\)), the latter available only when \(p<n_1+n_2\).}
    \label{fig:emp_power_Sigma1}
\end{figure}

\begin{figure}[htbp]
    \centering
    \begin{subfigure}[t]{0.23\textwidth}
        \centering
         \includegraphics[width=\linewidth, height=0.7\linewidth]{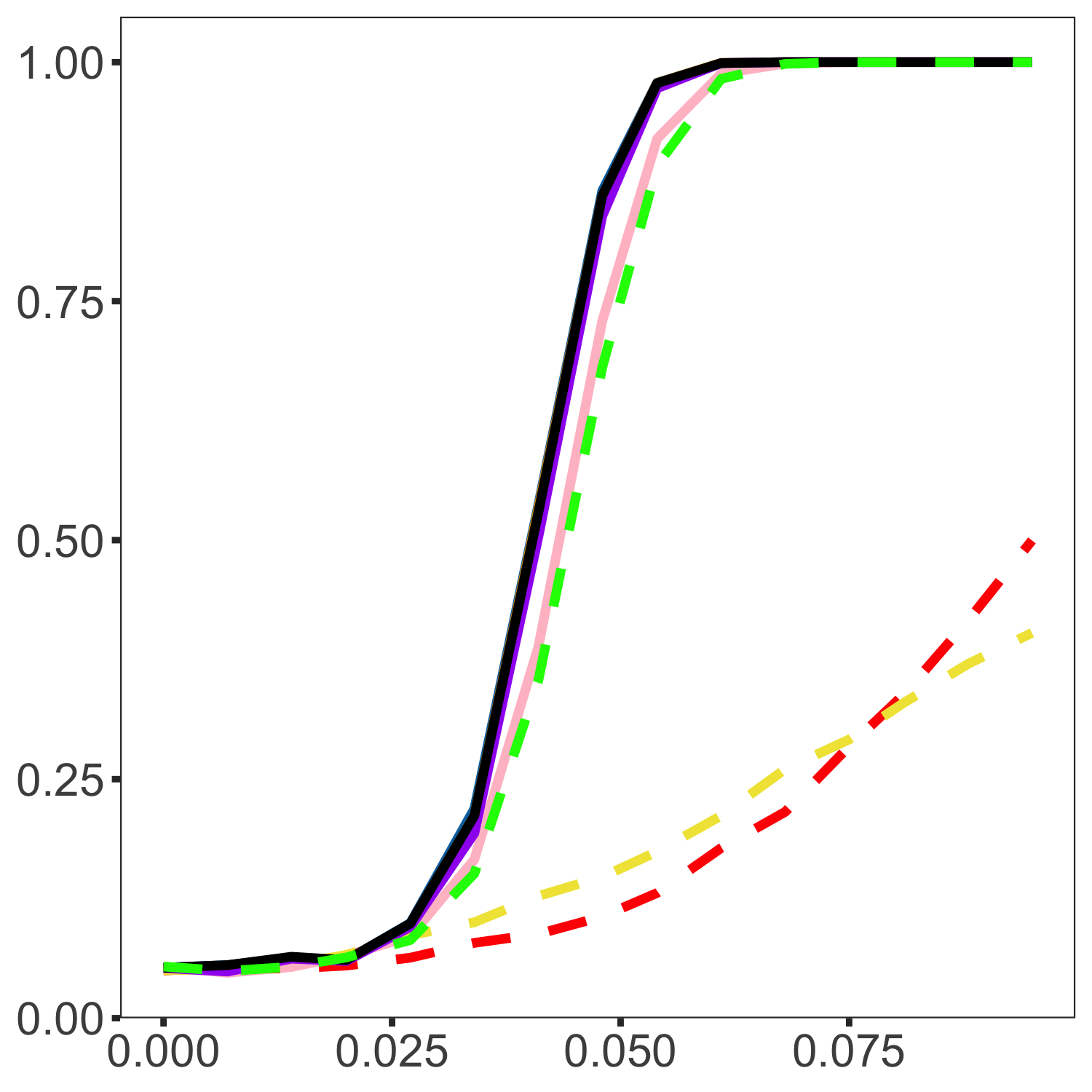}
    \end{subfigure}%
    \begin{subfigure}[t]{0.23\textwidth}
        \centering
        \includegraphics[width=\linewidth, height=0.7\linewidth]{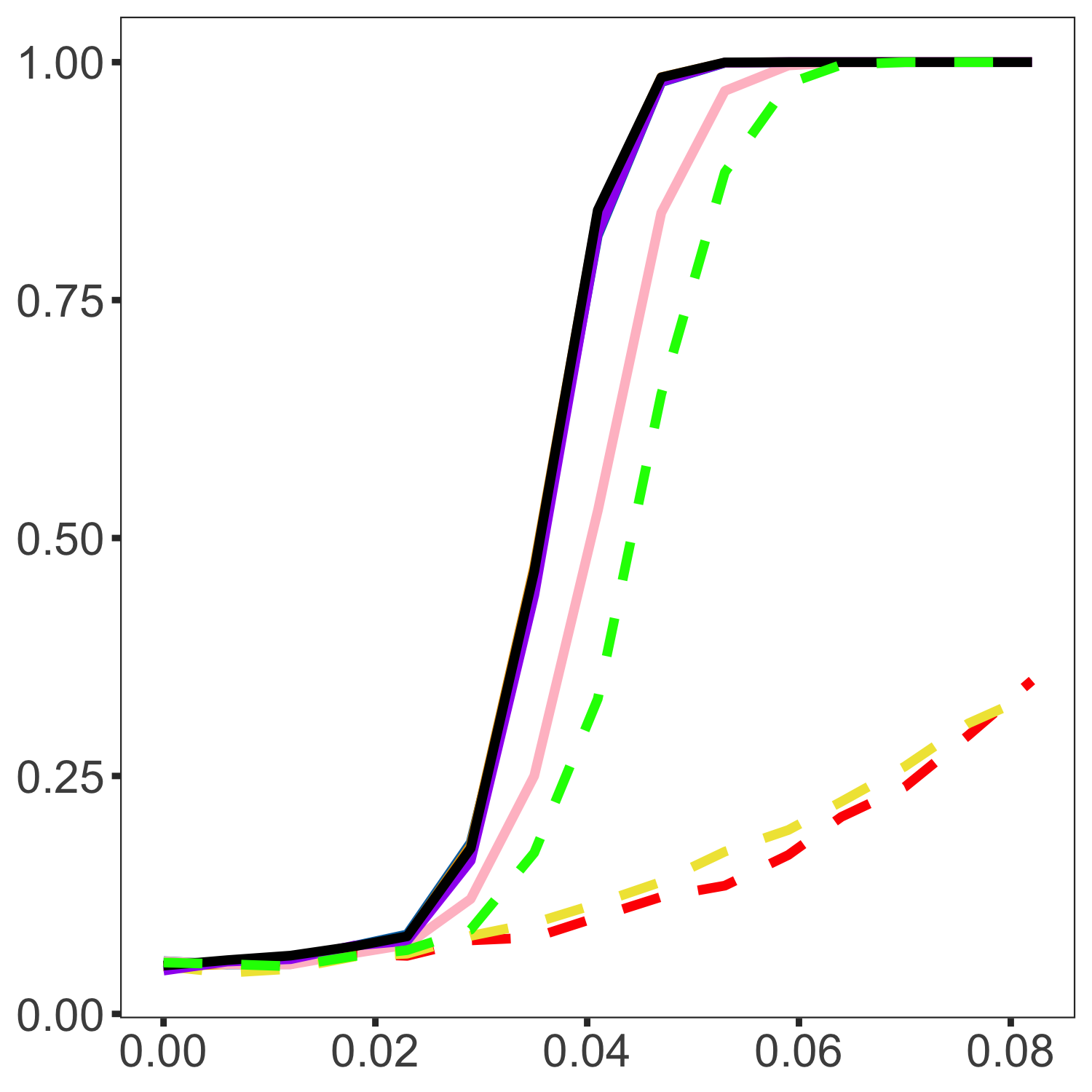}
    \end{subfigure}
     \begin{subfigure}[t]{0.23\textwidth}
        \centering
        \includegraphics[width=\linewidth, height=0.7\linewidth]{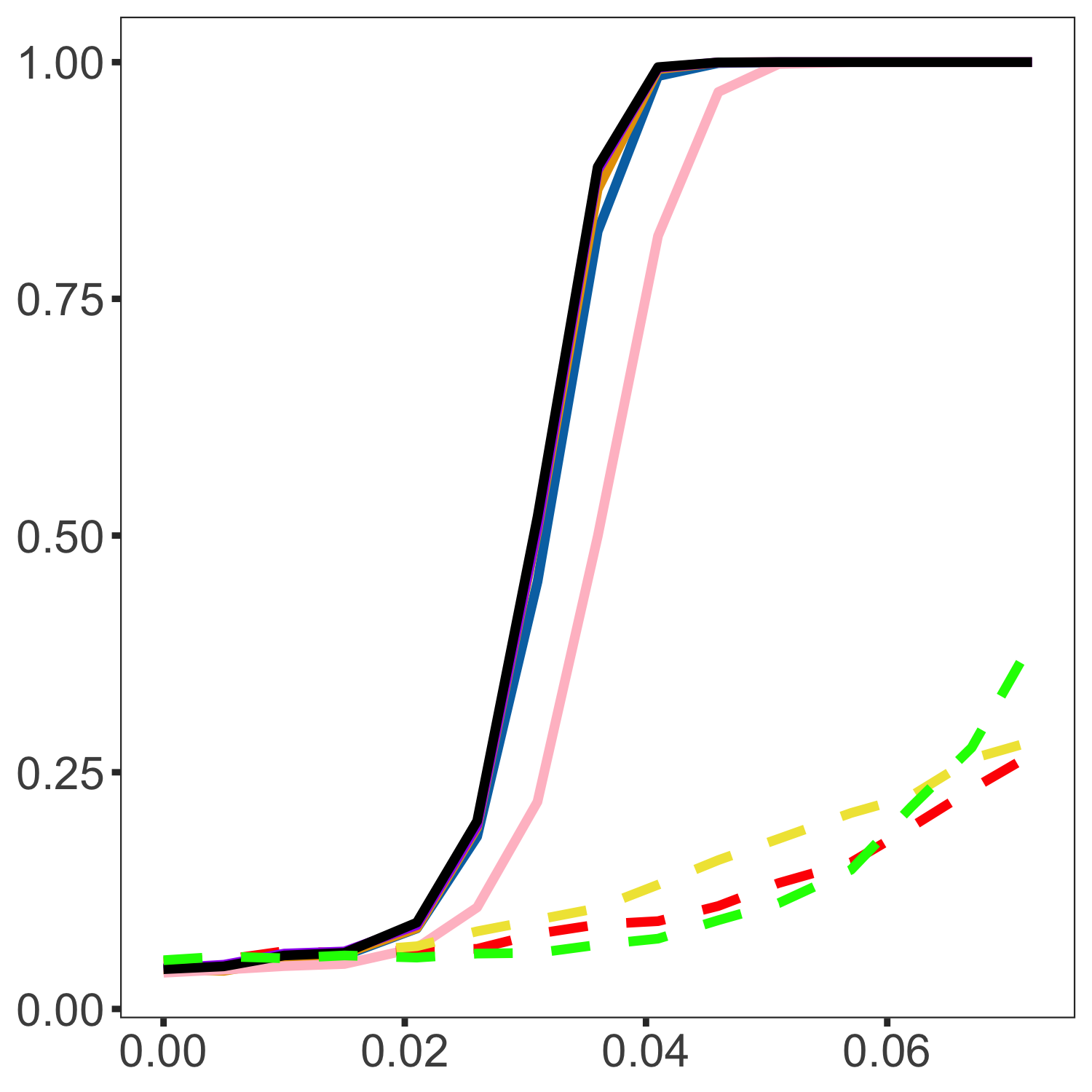}
    \end{subfigure}
    \begin{subfigure}[t]{0.23\textwidth}
        \centering
        \includegraphics[width=\linewidth, height=0.7\linewidth]{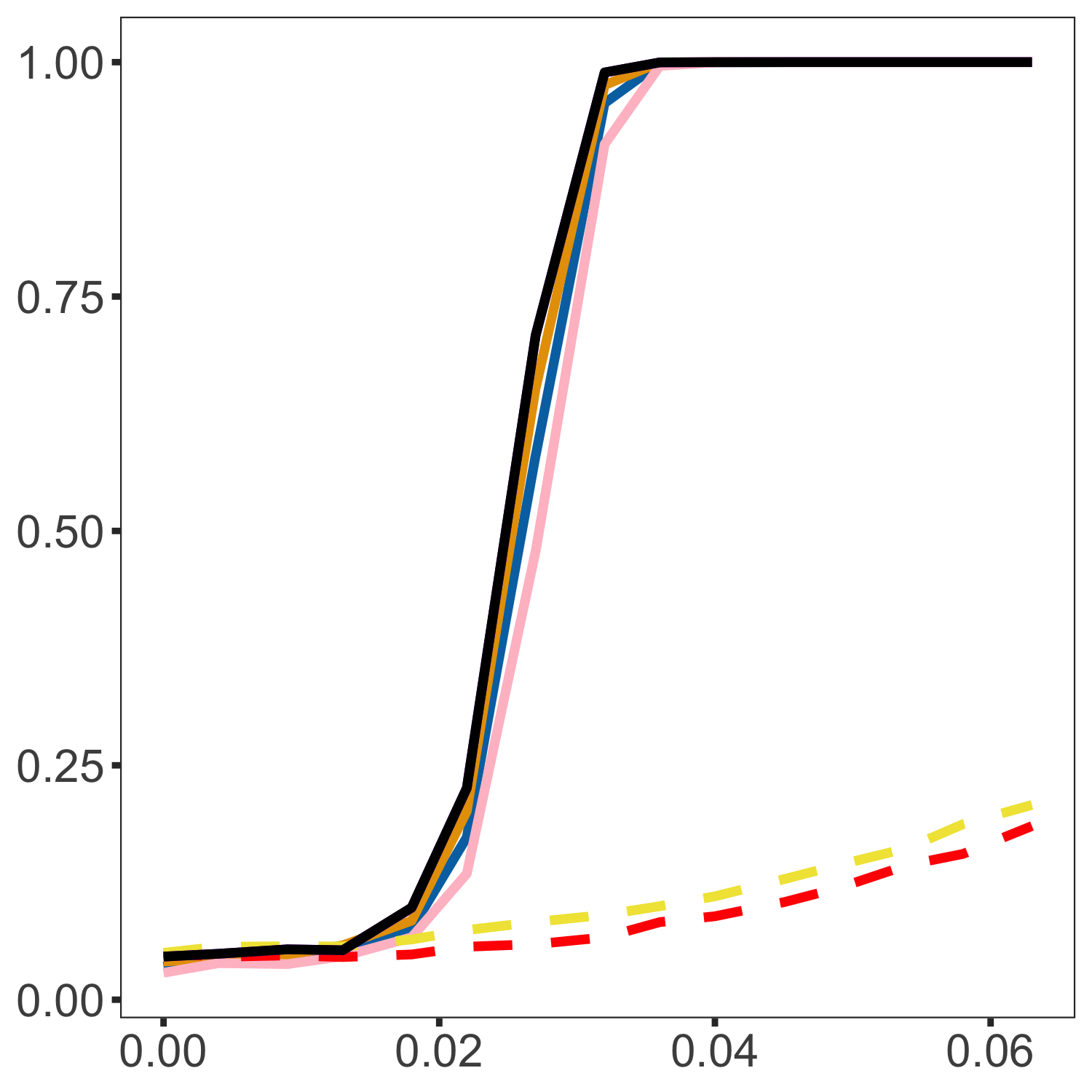}
    \end{subfigure}
    \vfill
    \begin{subfigure}[t]{0.23\textwidth}
        \centering
         \includegraphics[width=\linewidth, height=0.7\linewidth]{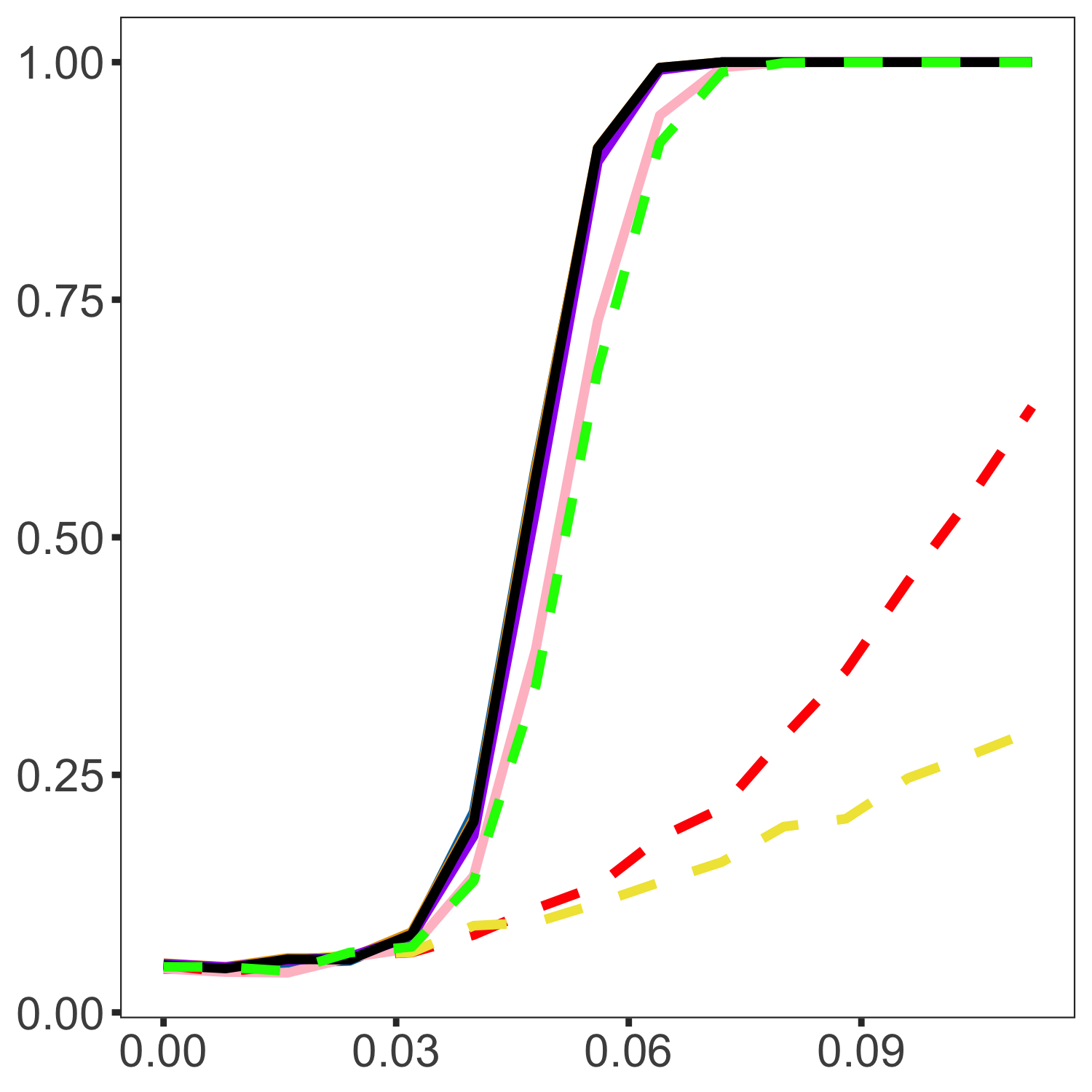}
    \end{subfigure}%
    \begin{subfigure}[t]{0.23\textwidth}
        \centering
        \includegraphics[width=\linewidth, height=0.7\linewidth]{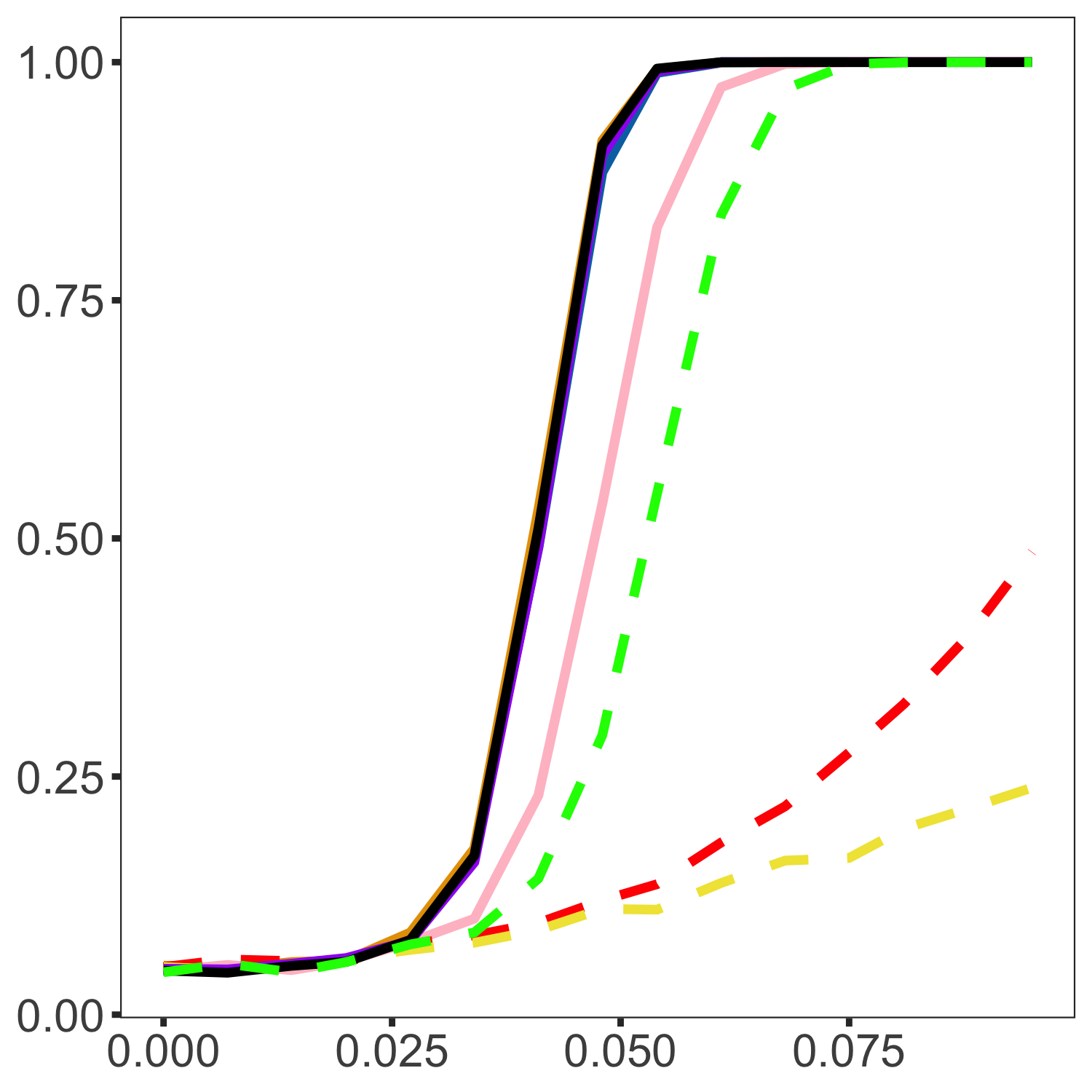}
    \end{subfigure}
     \begin{subfigure}[t]{0.23\textwidth}
        \centering
        \includegraphics[width=\linewidth, height=0.7\linewidth]{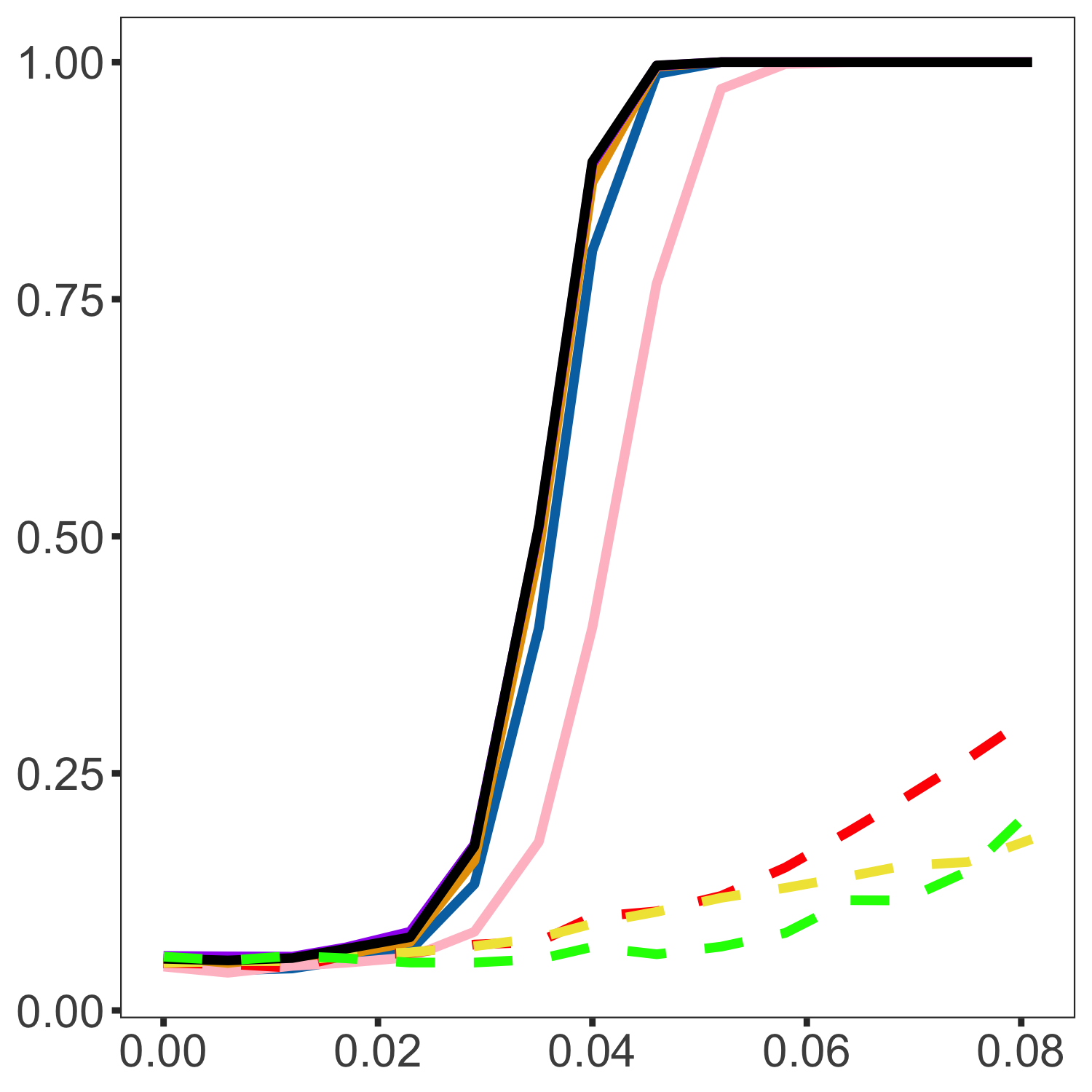}
    \end{subfigure}
    \begin{subfigure}[t]{0.23\textwidth}
        \centering
        \includegraphics[width=\linewidth, height=0.7\linewidth]{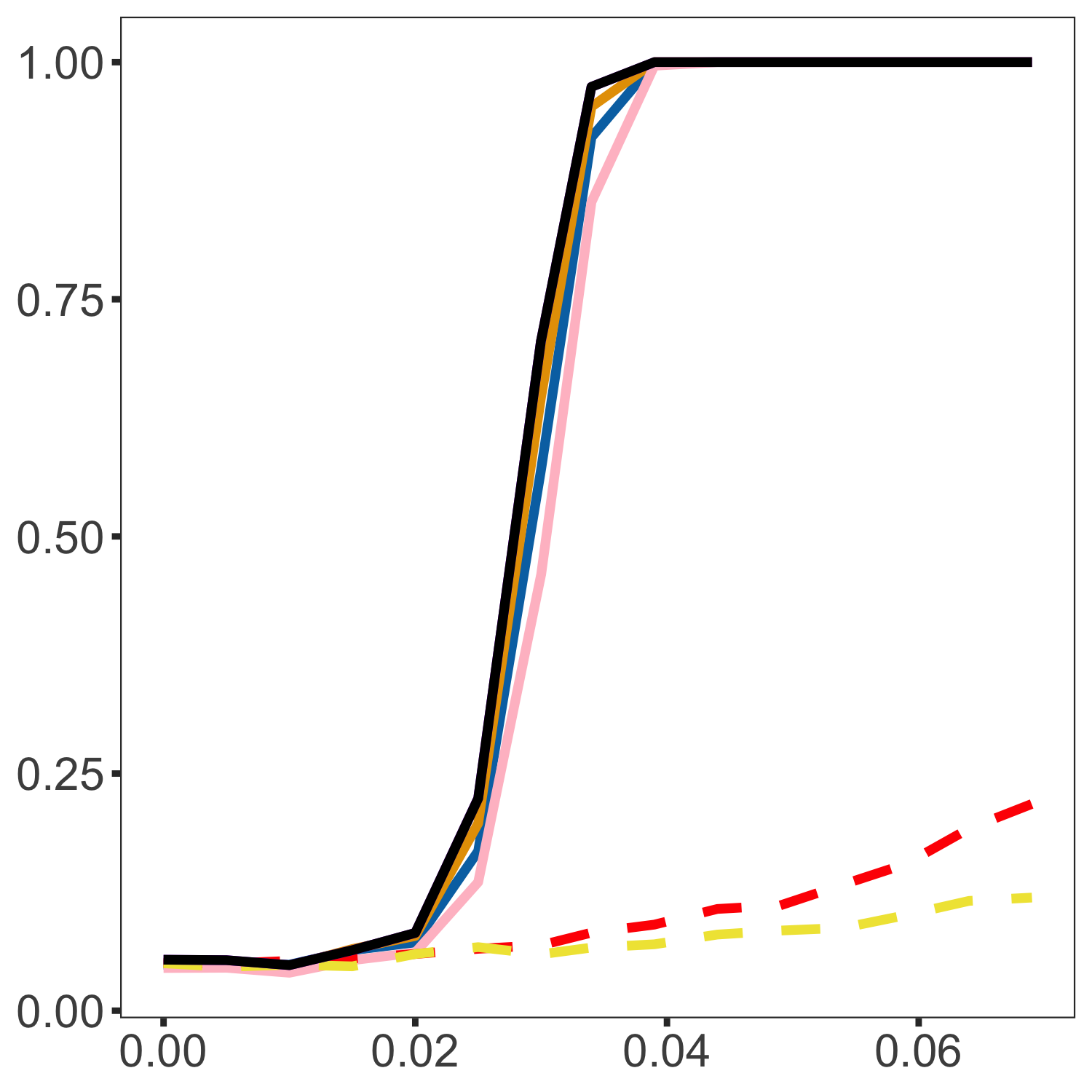}
    \end{subfigure}
    \vfill
    \begin{subfigure}[t]{0.23\textwidth}
        \centering
         \includegraphics[width=\linewidth, height=0.7\linewidth]{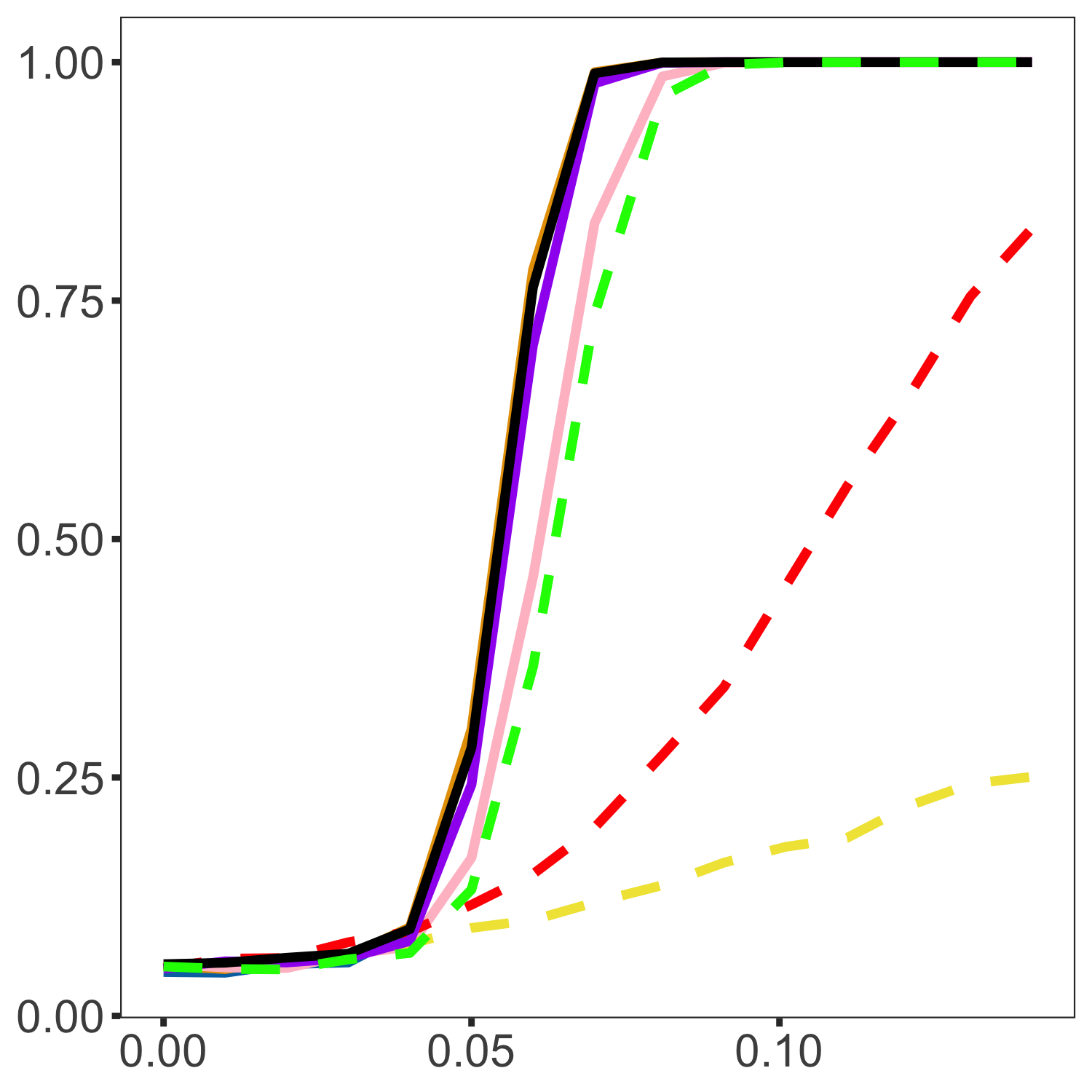}
    \end{subfigure}%
    \begin{subfigure}[t]{0.23\textwidth}
        \centering
        \includegraphics[width=\linewidth, height=0.7\linewidth]{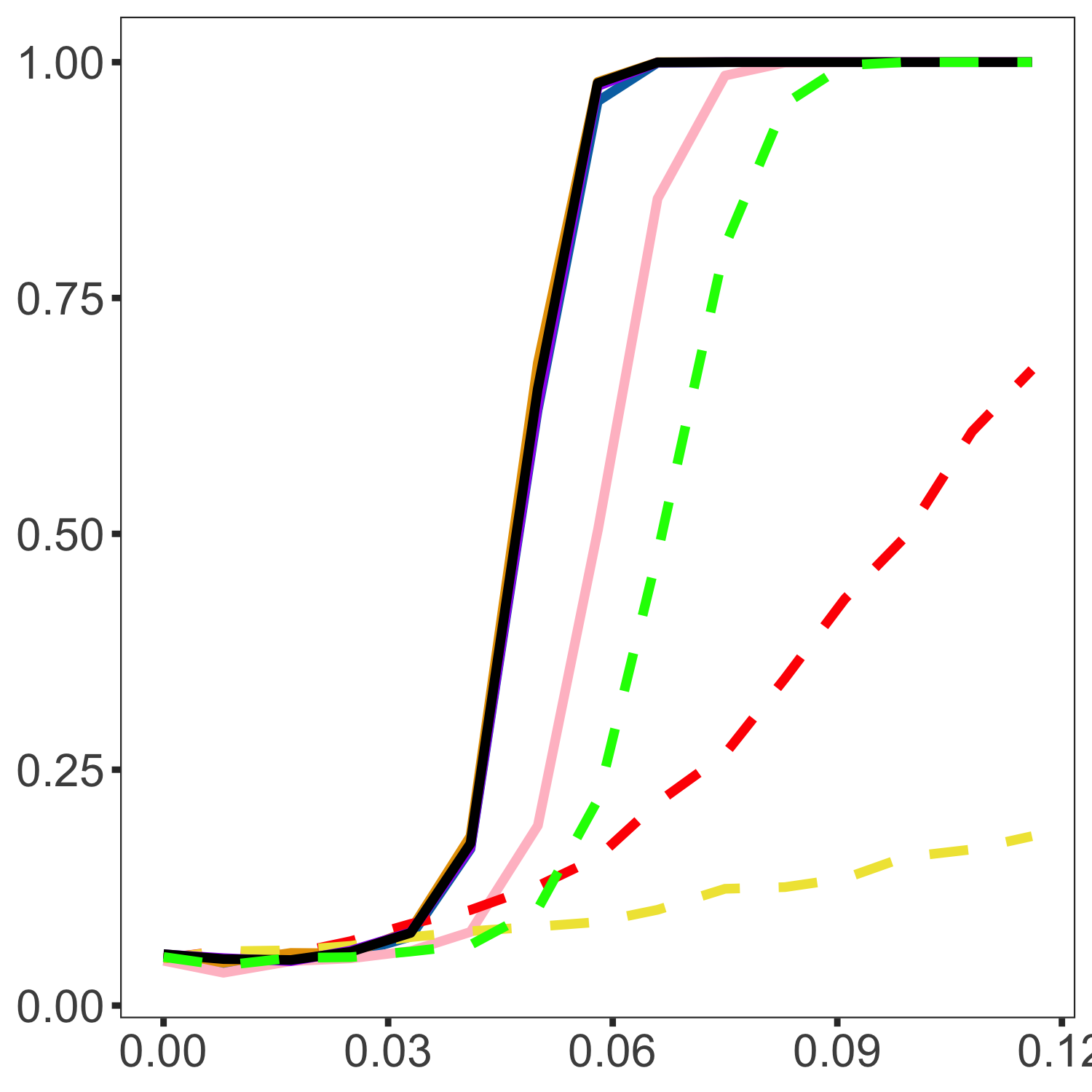}
    \end{subfigure}
     \begin{subfigure}[t]{0.23\textwidth}
        \centering
        \includegraphics[width=\linewidth, height=0.7\linewidth]{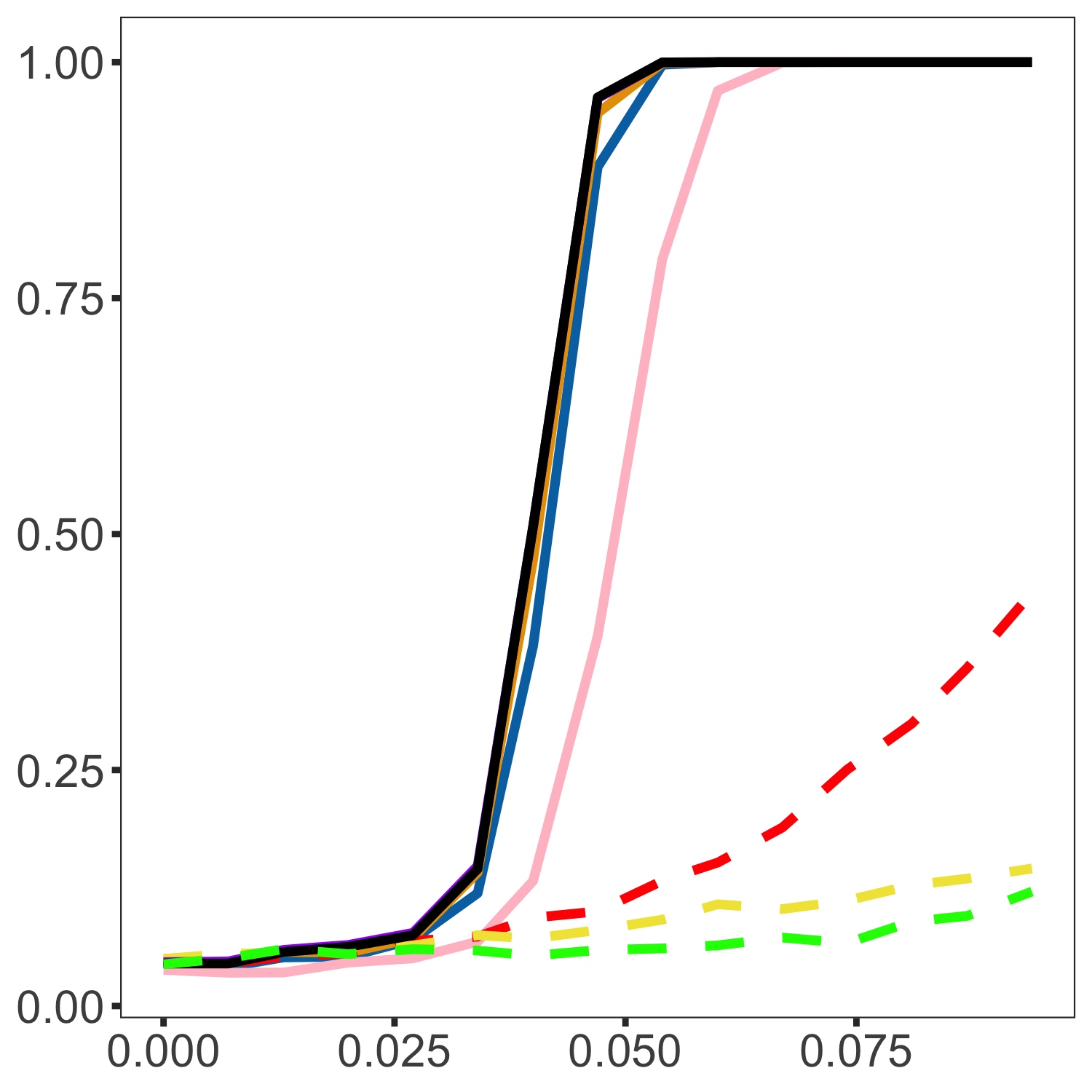}
    \end{subfigure}
    \begin{subfigure}[t]{0.23\textwidth}
        \centering
        \includegraphics[width=\linewidth, height=0.7\linewidth]{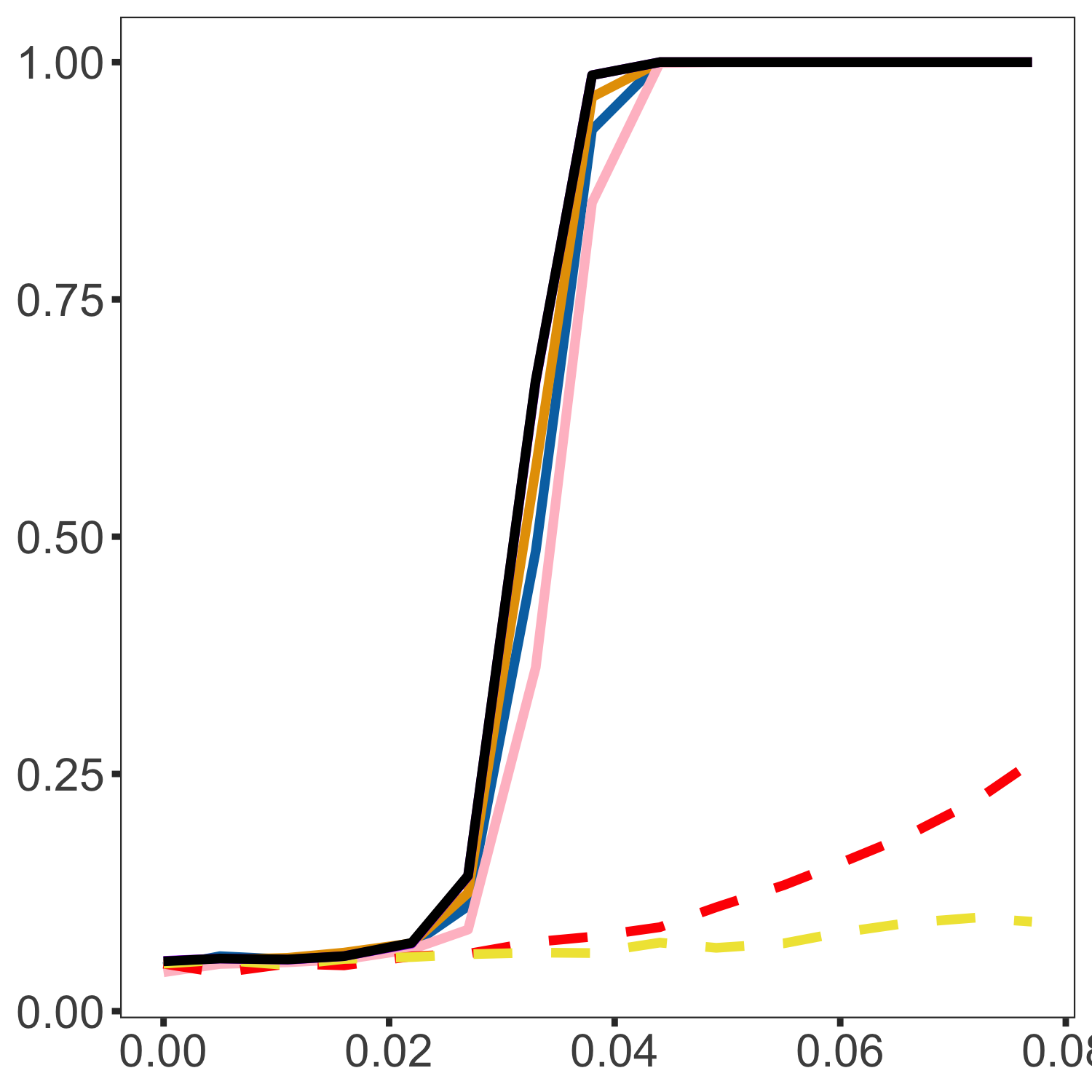}
    \end{subfigure}
   \caption{Size-adjusted empirical power when \(\Sigma\) is AR-ACF. Columns (left to right) correspond to \(\hat{\gamma}_2 = 0.3, 0.5, 0.9, 2\); rows (top to bottom) correspond to \(n_1 = 50, 100, 250\). Solid curves: blue (\(\lambda = 0.5\)), orange (\(\lambda = 1\)), black (\(\lambda=\hat{\lambda}_{I_p}\)), purple (\(\lambda=\hat{\lambda}_{\Sigma_p}\)), and pink (\(\lambda=\hat{\lambda}_*\)). Dashed curves: red (Proj-LRT), yellow (Ridge-LRT), and green (\cite{han2016tracy}, \(\lambda=0\)), the latter available only when \(p<n_1+n_2\).}
    \label{fig:emp_power_Sigma3}
\end{figure}

\begin{figure}[htbp]
    \centering
    \begin{subfigure}[t]{0.23\textwidth}
        \centering
         \includegraphics[width=\linewidth, height=0.7\linewidth]{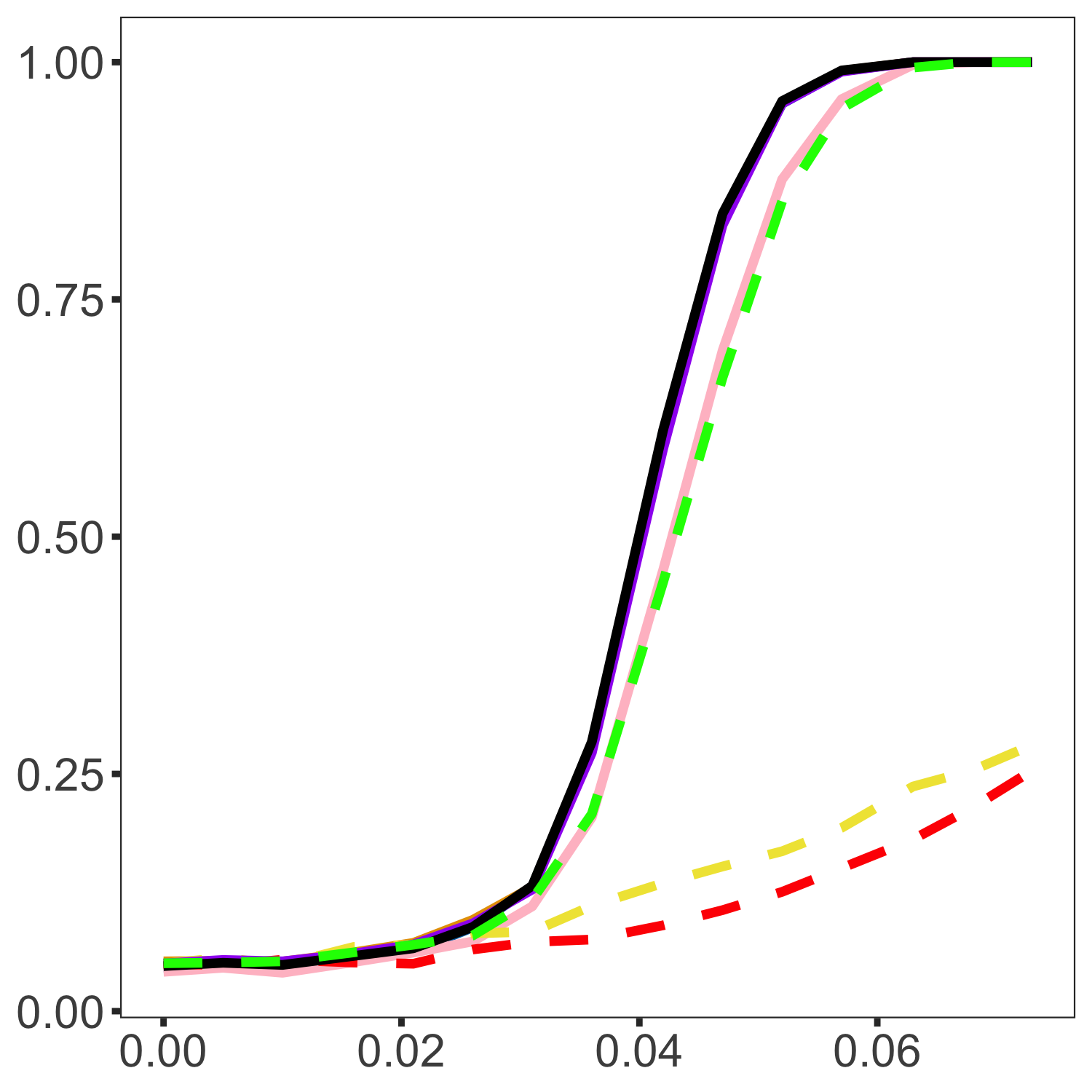}
    \end{subfigure}%
    \begin{subfigure}[t]{0.23\textwidth}
        \centering
        \includegraphics[width=\linewidth, height=0.7\linewidth]{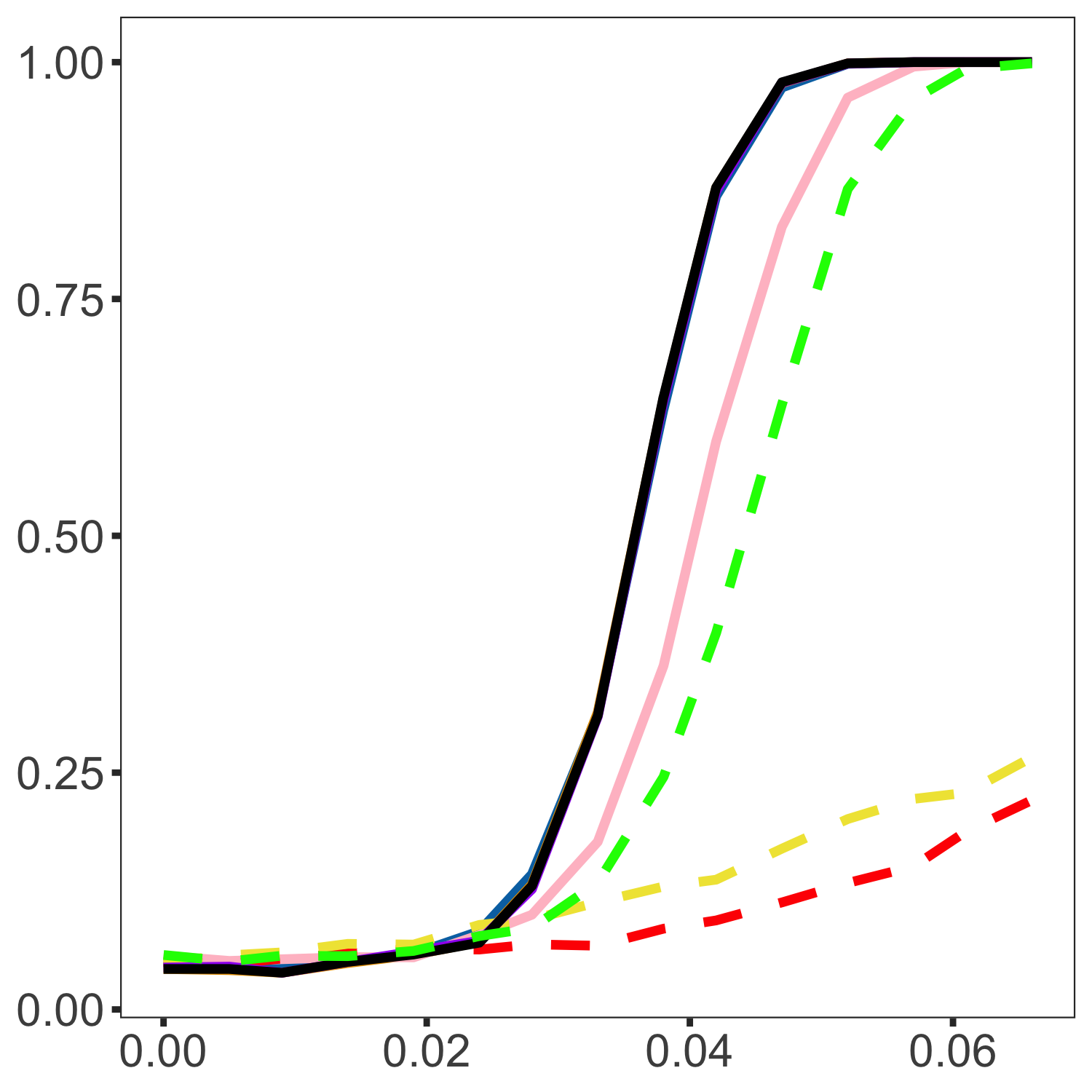}
    \end{subfigure}
     \begin{subfigure}[t]{0.23\textwidth}
        \centering
        \includegraphics[width=\linewidth, height=0.7\linewidth]{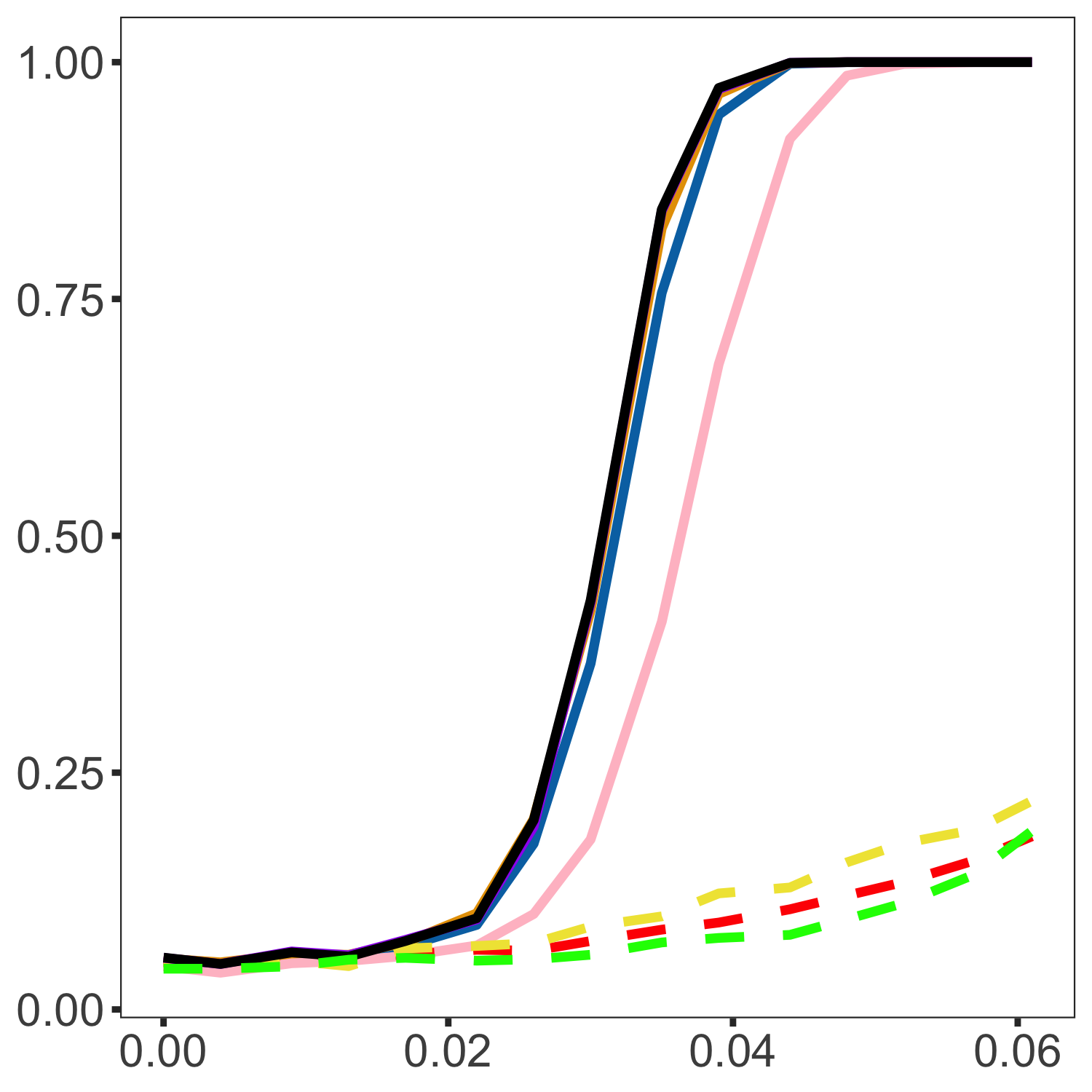}
    \end{subfigure}
    \begin{subfigure}[t]{0.23\textwidth}
        \centering
        \includegraphics[width=\linewidth, height=0.7\linewidth]{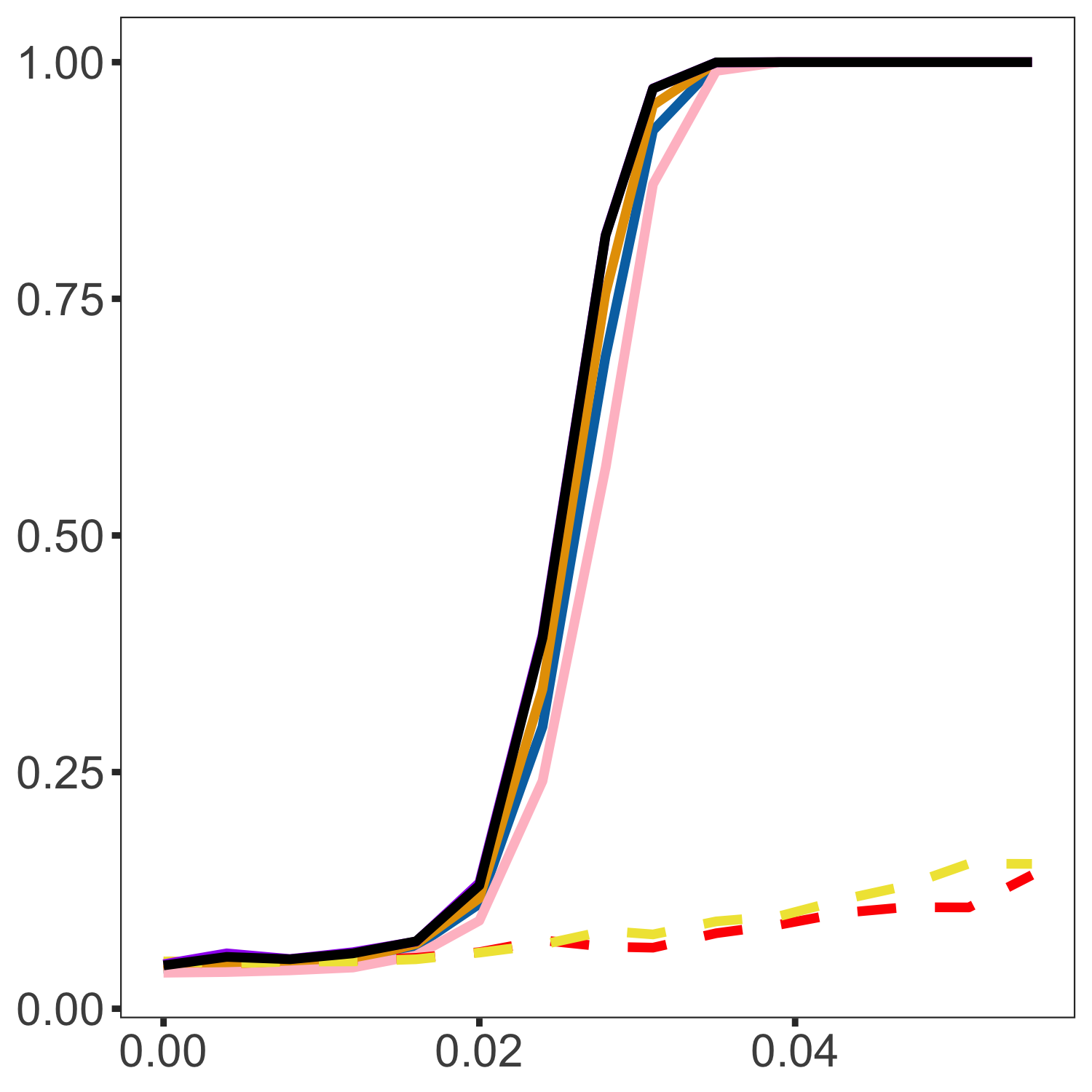}
    \end{subfigure}
    \vfill
    \begin{subfigure}[t]{0.23\textwidth}
        \centering
         \includegraphics[width=\linewidth, height=0.7\linewidth]{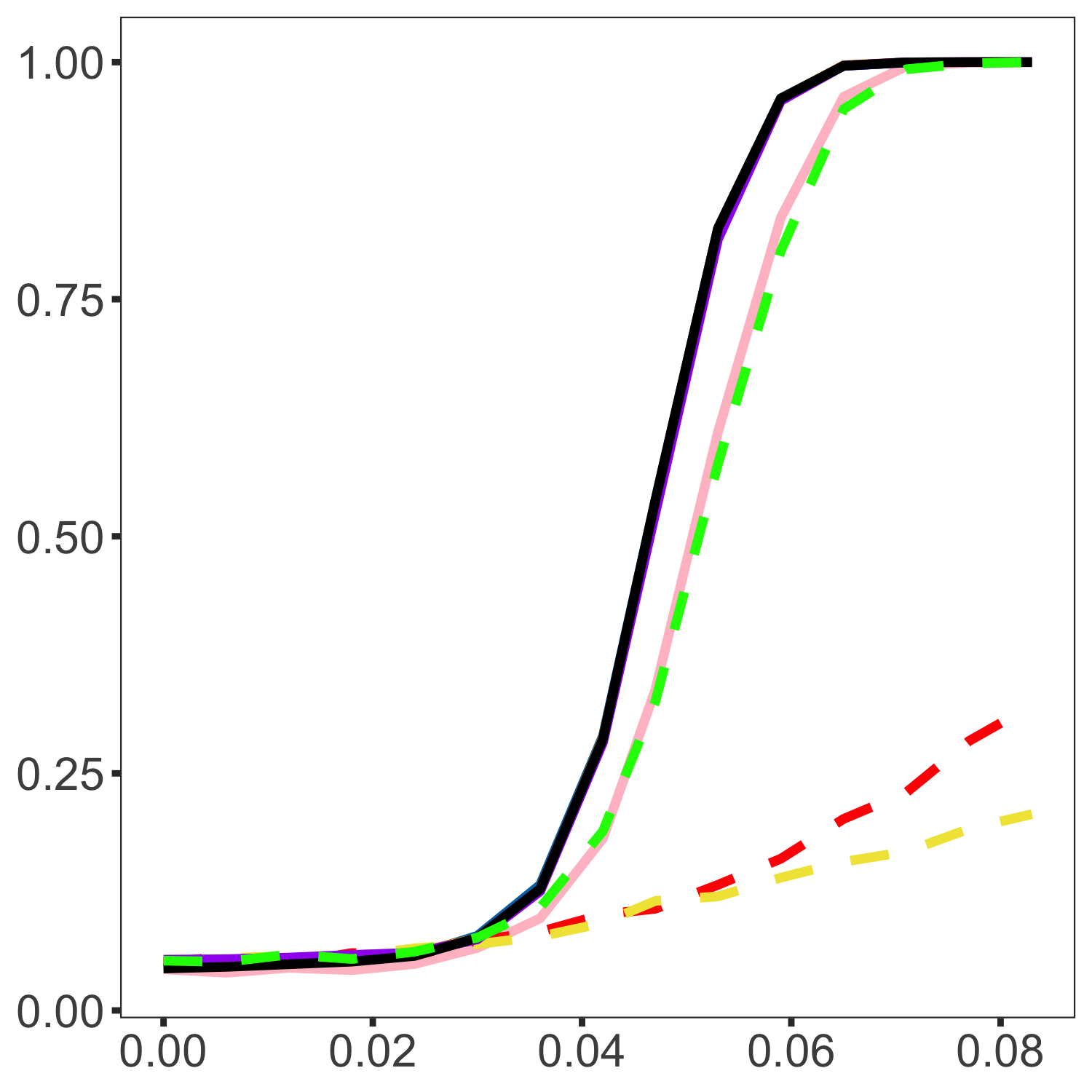}
    \end{subfigure}%
    \begin{subfigure}[t]{0.23\textwidth}
        \centering
        \includegraphics[width=\linewidth, height=0.7\linewidth]{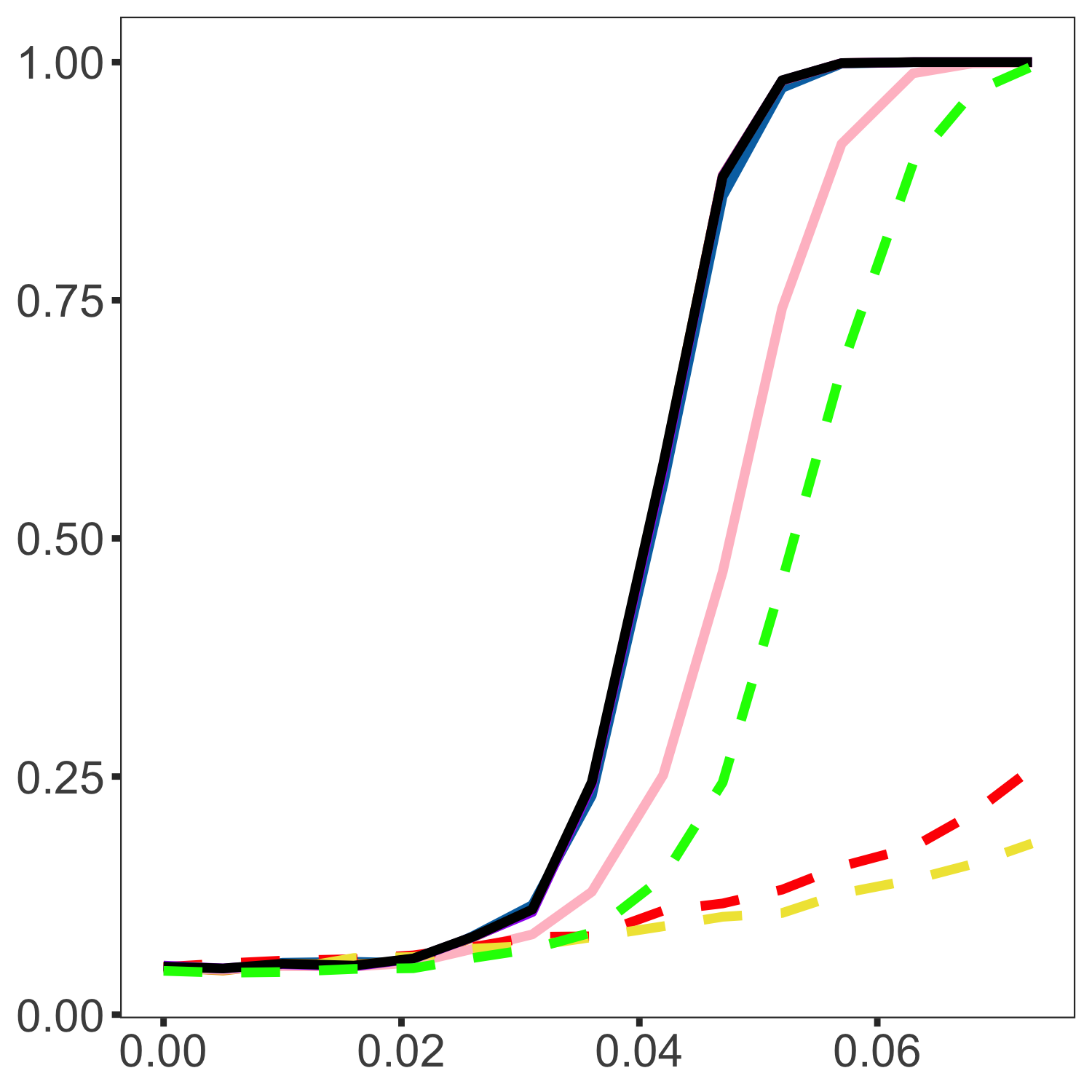}
    \end{subfigure}
     \begin{subfigure}[t]{0.23\textwidth}
        \centering
        \includegraphics[width=\linewidth, height=0.7\linewidth]{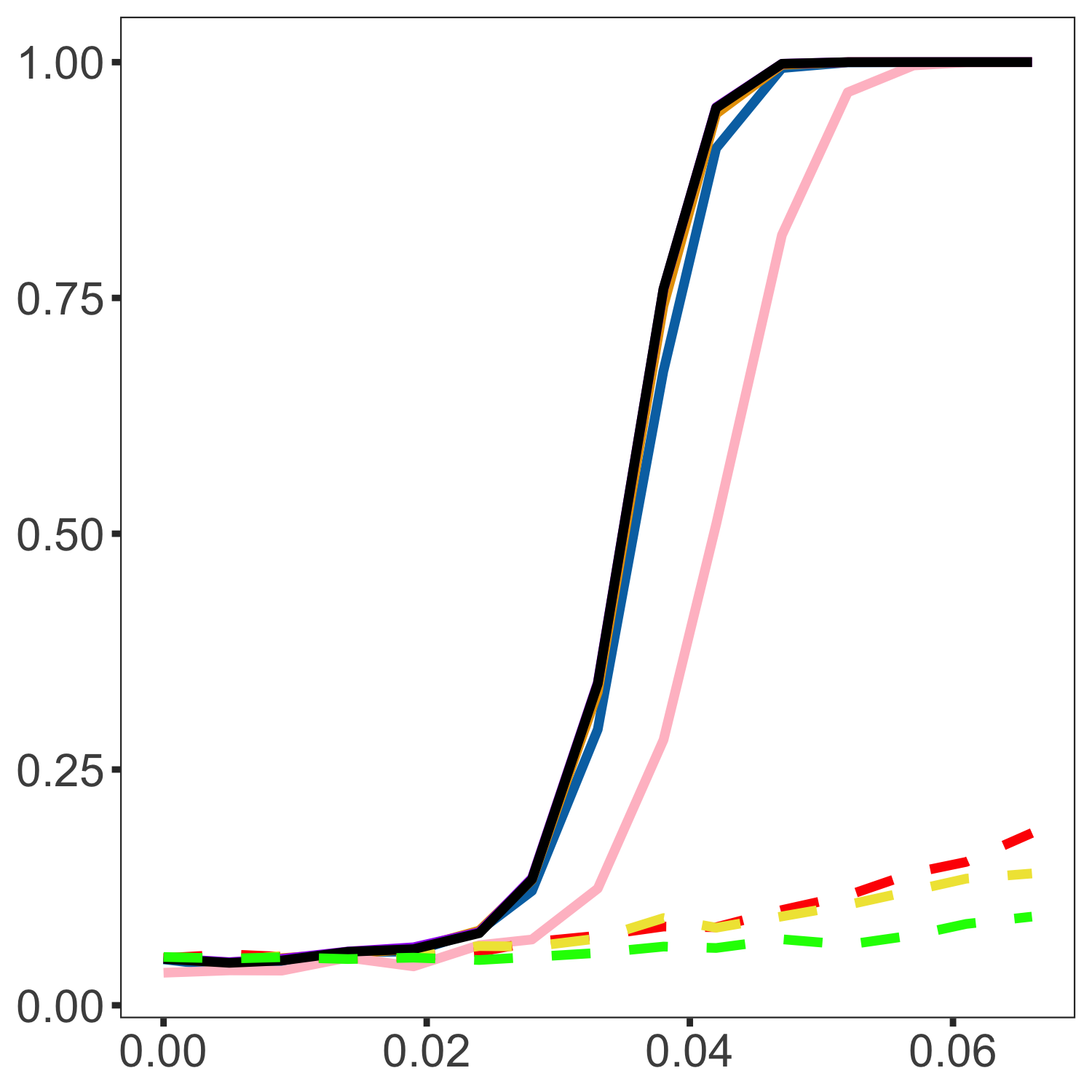}
    \end{subfigure}
    \begin{subfigure}[t]{0.23\textwidth}
        \centering
        \includegraphics[width=\linewidth, height=0.7\linewidth]{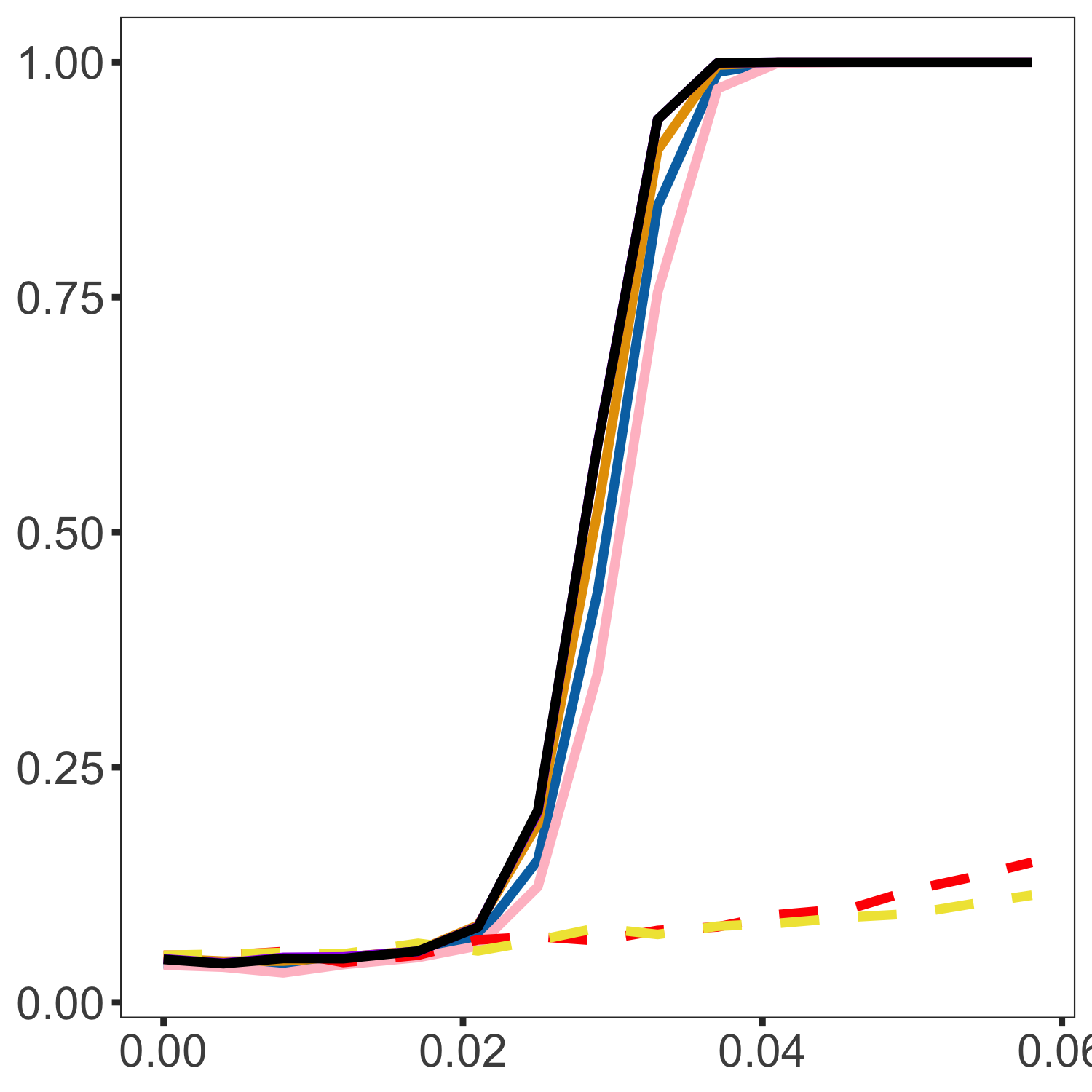}
    \end{subfigure}
    \vfill
    \begin{subfigure}[t]{0.23\textwidth}
        \centering
         \includegraphics[width=\linewidth, height=0.7\linewidth]{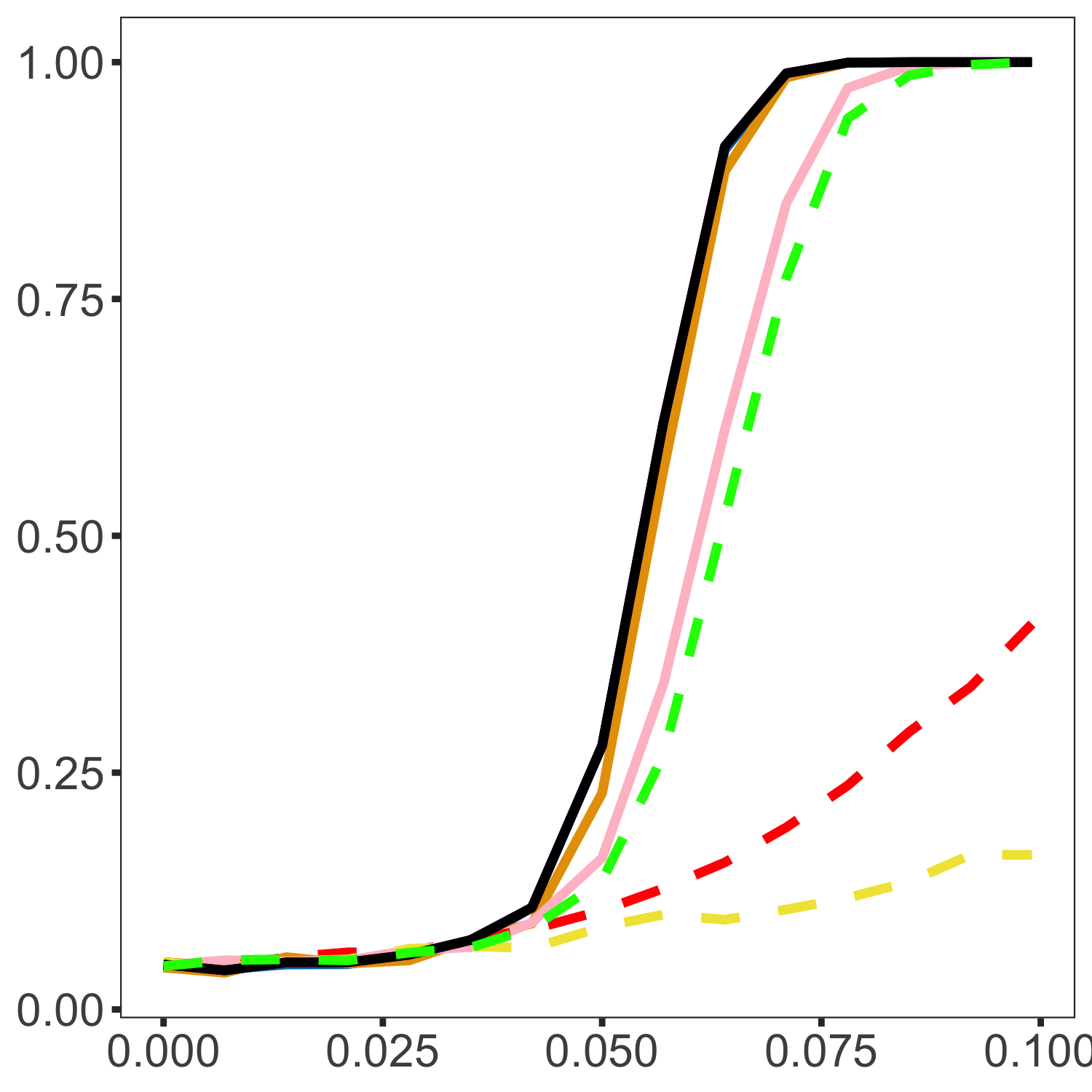}
    \end{subfigure}%
    \begin{subfigure}[t]{0.23\textwidth}
        \centering
        \includegraphics[width=\linewidth, height=0.7\linewidth]{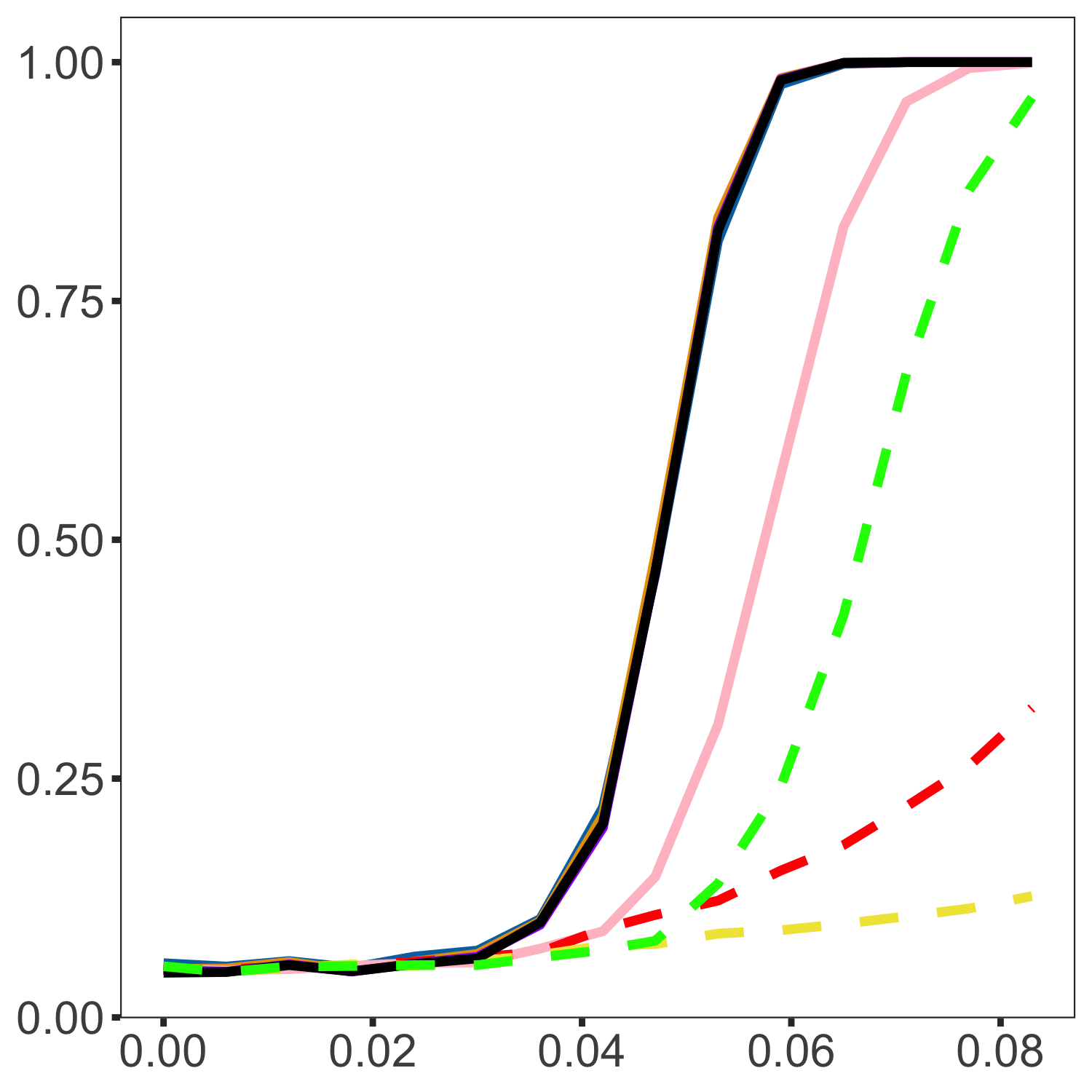}
    \end{subfigure}
     \begin{subfigure}[t]{0.23\textwidth}
        \centering
        \includegraphics[width=\linewidth, height=0.7\linewidth]{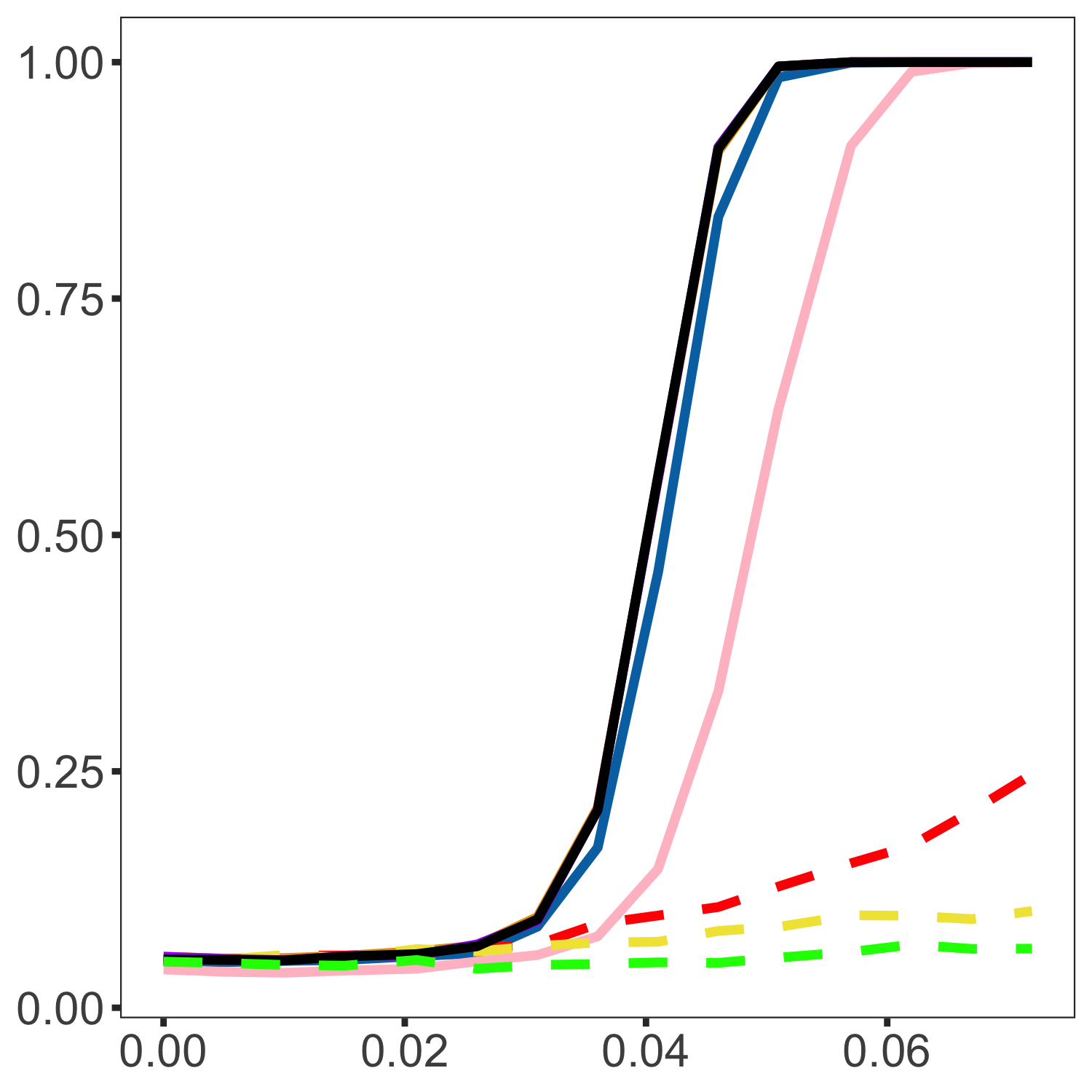}
    \end{subfigure}
    \begin{subfigure}[t]{0.23\textwidth}
        \centering
        \includegraphics[width=\linewidth, height=0.7\linewidth]{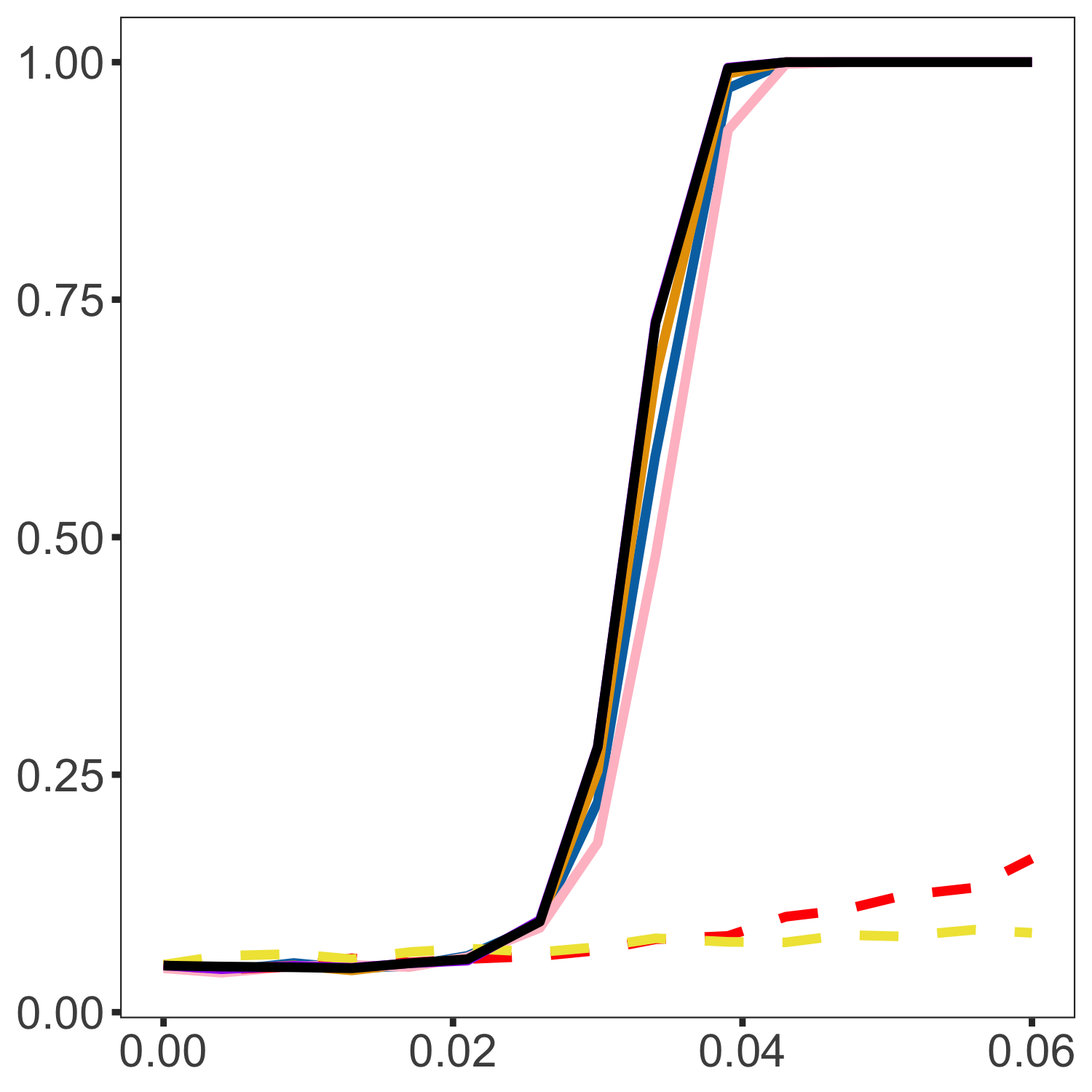}
    \end{subfigure}
   \caption{Size-adjusted empirical power when \(\Sigma\) is Factor. Columns (left to right) correspond to \(\hat{\gamma}_2 = 0.3, 0.5, 0.9, 2\); rows (top to bottom) correspond to \(n_1 = 50, 100, 250\). Solid curves: blue (\(\lambda = 0.5\)), orange (\(\lambda = 1\)), black (\(\lambda=\hat{\lambda}_{I_p}\)), purple (\(\lambda=\hat{\lambda}_{\Sigma_p}\)), and pink (\(\lambda=\hat{\lambda}_*\)). Dashed curves: red (Proj-LRT), yellow (Ridge-LRT), and green (\cite{han2016tracy}, \(\lambda=0\)), the latter available only when \(p<n_1+n_2\).}
    \label{fig:emp_power_Sigma4}
\end{figure}

\subsection{Results under high-rank alternatives}\label{subsec:full_rank_simulation}

In this section, we examine the performance of the proposed methods under a class of high-rank alternatives and compare them with likelihood-based approaches, namely Ridge-LRT and Proj-LRT.

Specifically, while all other model configurations remain the same as those in Section~\ref{sec:simulation_settings}, we consider the following specification of $B$ under the alternative:
\begin{itemize}
    \item[(vi)] $X$ is an $n_1 \times (n_1+n_2)$ matrix whose entries are i.i.d.\ $N(0,1)$. The coefficient matrix $B$ ($p\times n_1$) is assumed to have rank $(n_1/2)$, where its first $(n_1/2)$ columns are drawn from $N(0,\zeta^2 D)$ and the remaining columns are set to zero. Here $D=\operatorname{diag}(d_1^2,\dots,d_p^2)$, where $d_j$ are evenly spaced between $2$ and $10$, and $D$ is further normalized so that $\tr(D)=p$. The parameter $\zeta$ controls the signal strength. We test $H_0:B=0$ against $H_a:B\neq0$, that is, $C=I_{n_1}$.
\end{itemize}

The empirical power of all competing methods under this configuration is shown in Figures~\ref{fig:FR_emp_power_Sigma1}--\ref{fig:FR_emp_power_Sigma4}. The results indicate that, although less efficient than likelihood-based approaches as expected, the largest-root tests still retain considerable power under high-rank alternatives.

\begin{figure}[htbp]
    \centering
    \begin{subfigure}[t]{0.23\textwidth}
        \centering
         \includegraphics[width=\linewidth, height=0.7\linewidth]{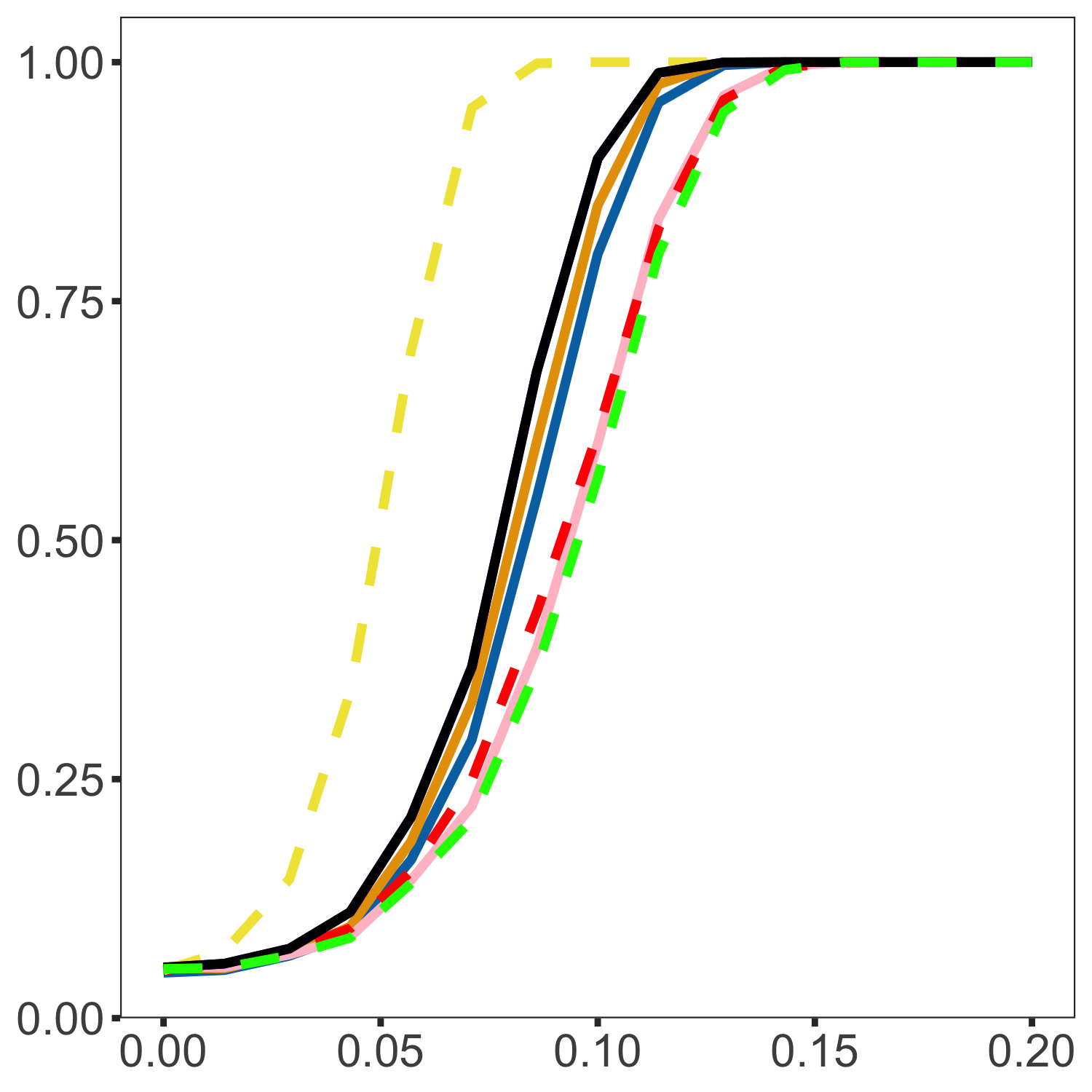}
    \end{subfigure}%
    \begin{subfigure}[t]{0.23\textwidth}
        \centering
        \includegraphics[width=\linewidth, height=0.7\linewidth]{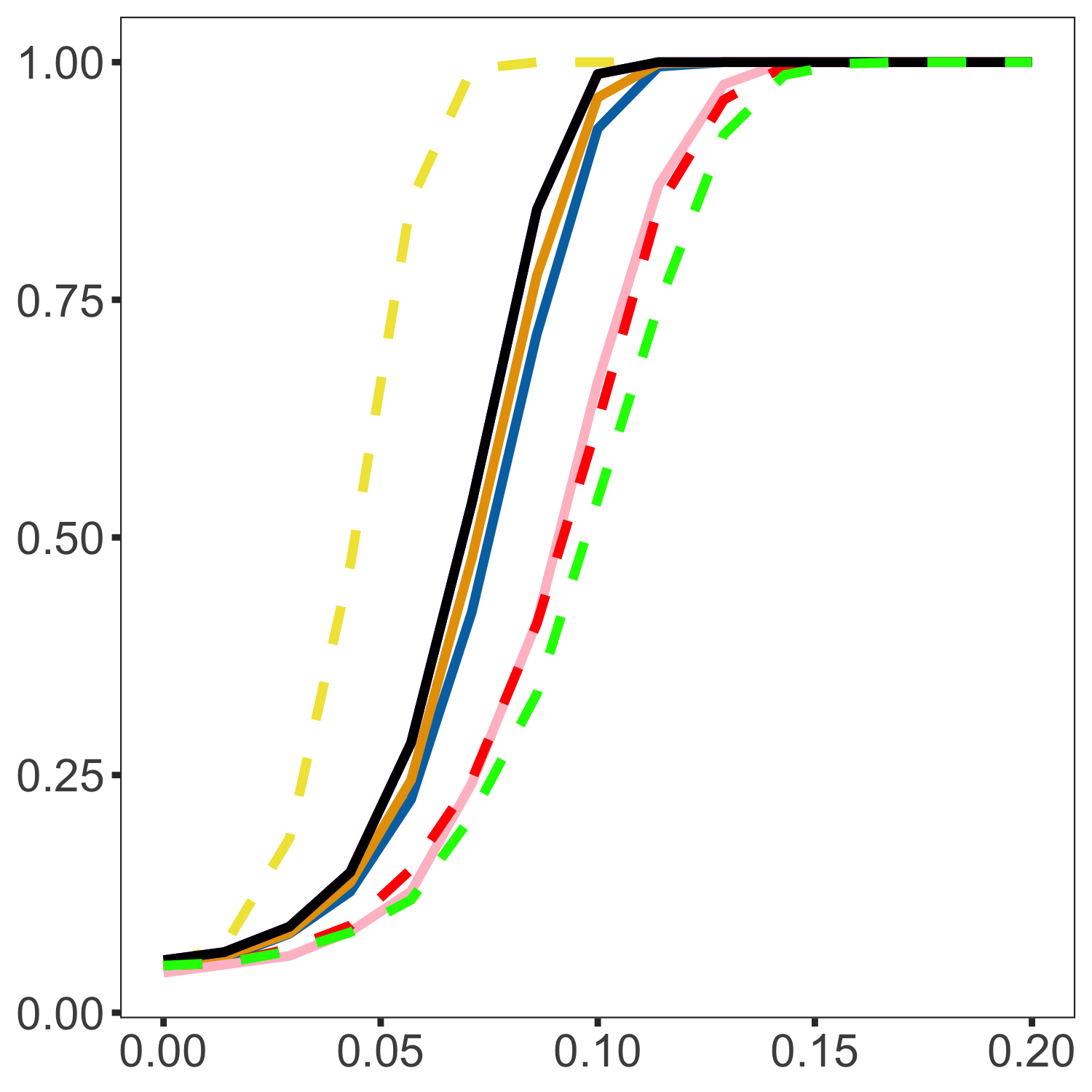}
    \end{subfigure}
     \begin{subfigure}[t]{0.23\textwidth}
        \centering
        \includegraphics[width=\linewidth, height=0.7\linewidth]{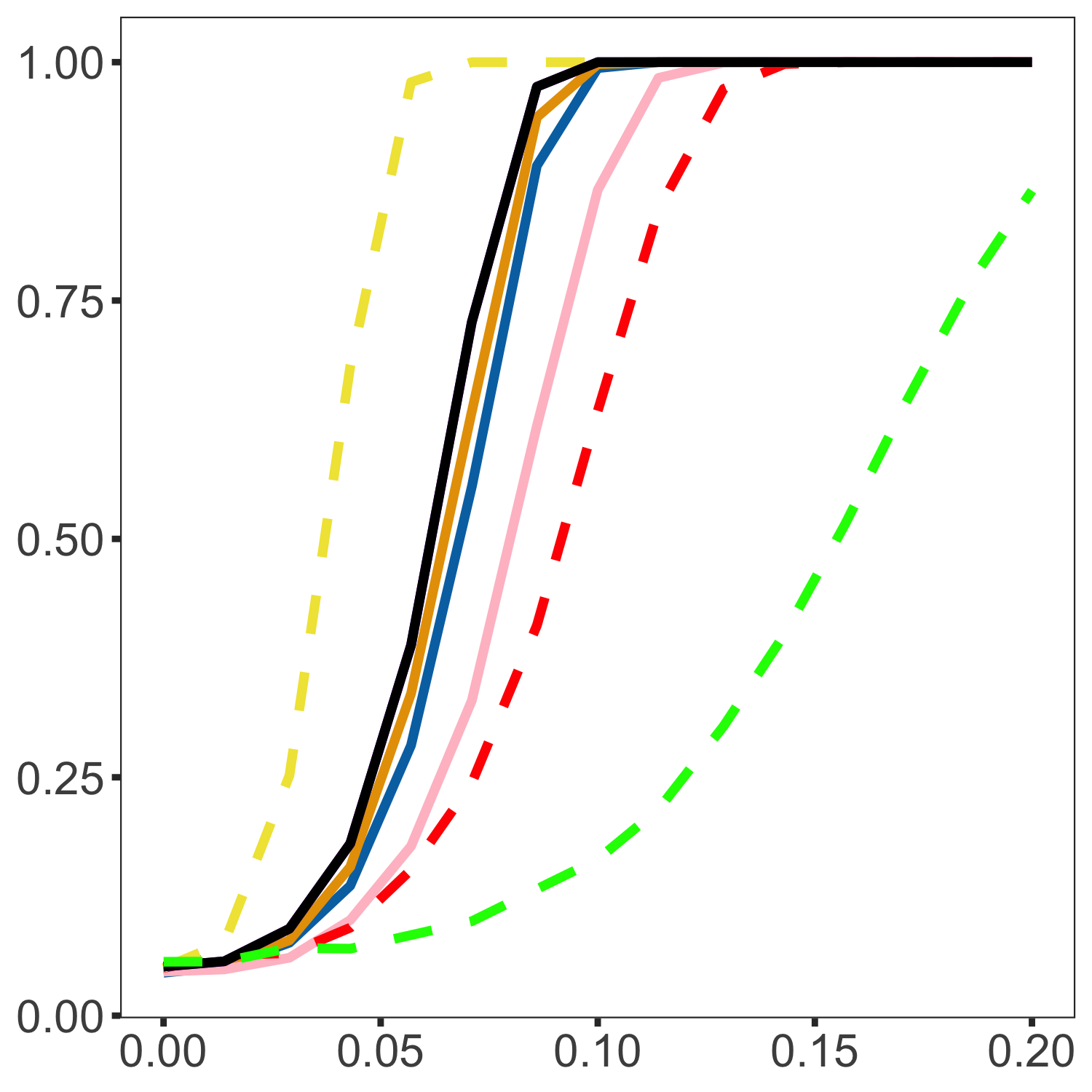}
    \end{subfigure}
    \begin{subfigure}[t]{0.23\textwidth}
        \centering
        \includegraphics[width=\linewidth, height=0.7\linewidth]{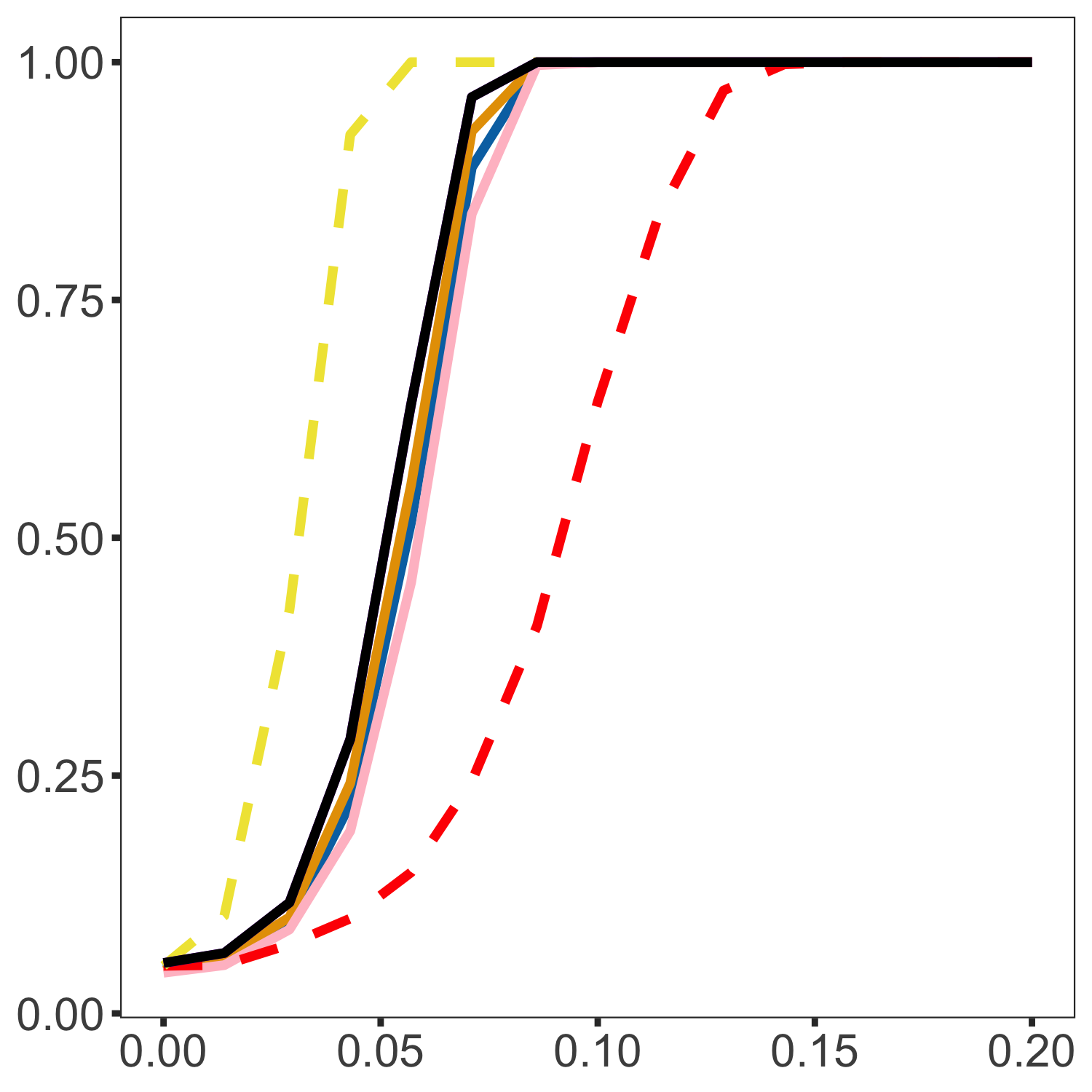}
    \end{subfigure}
    \vfill
    \begin{subfigure}[t]{0.23\textwidth}
        \centering
         \includegraphics[width=\linewidth, height=0.7\linewidth]{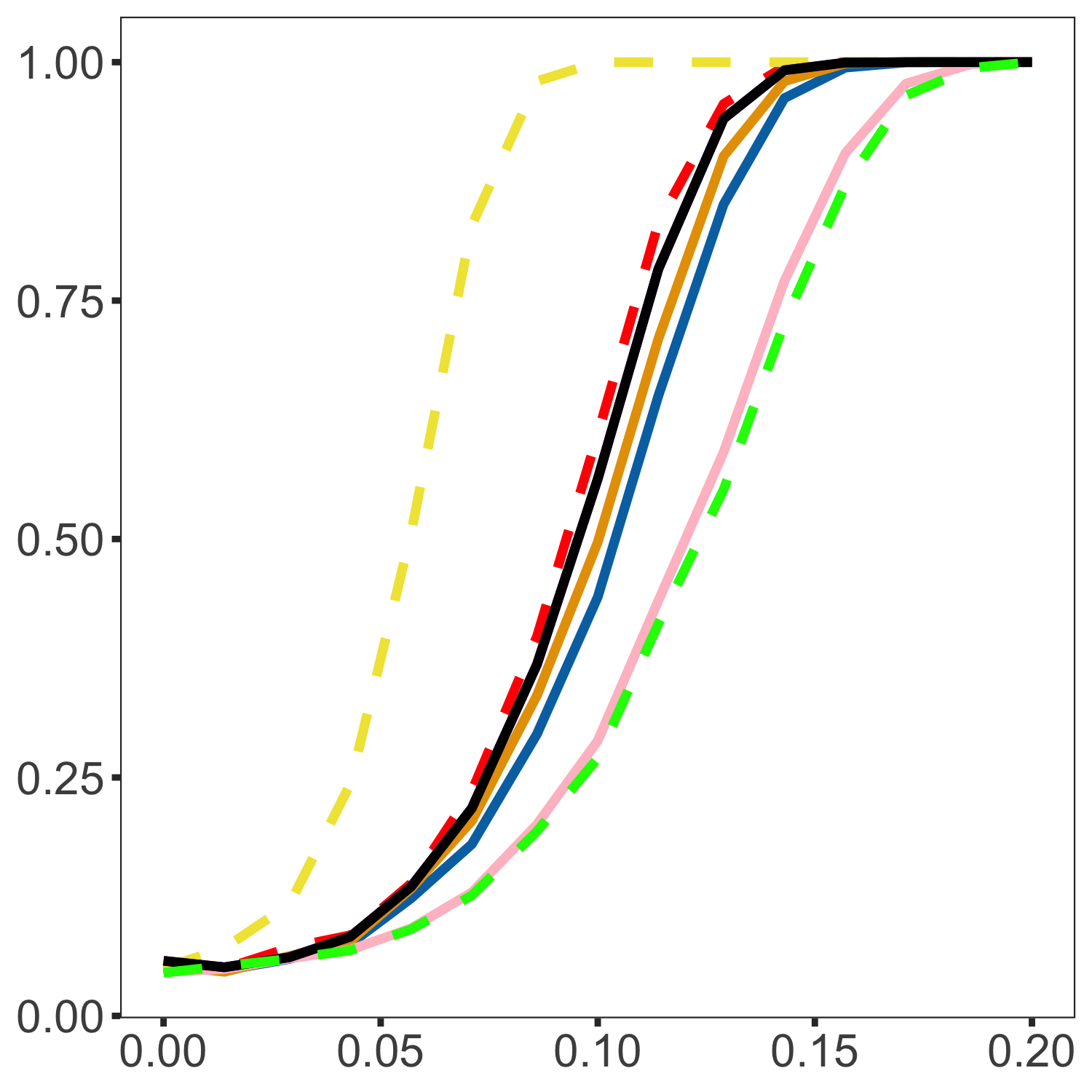}
    \end{subfigure}%
    \begin{subfigure}[t]{0.23\textwidth}
        \centering
        \includegraphics[width=\linewidth, height=0.7\linewidth]{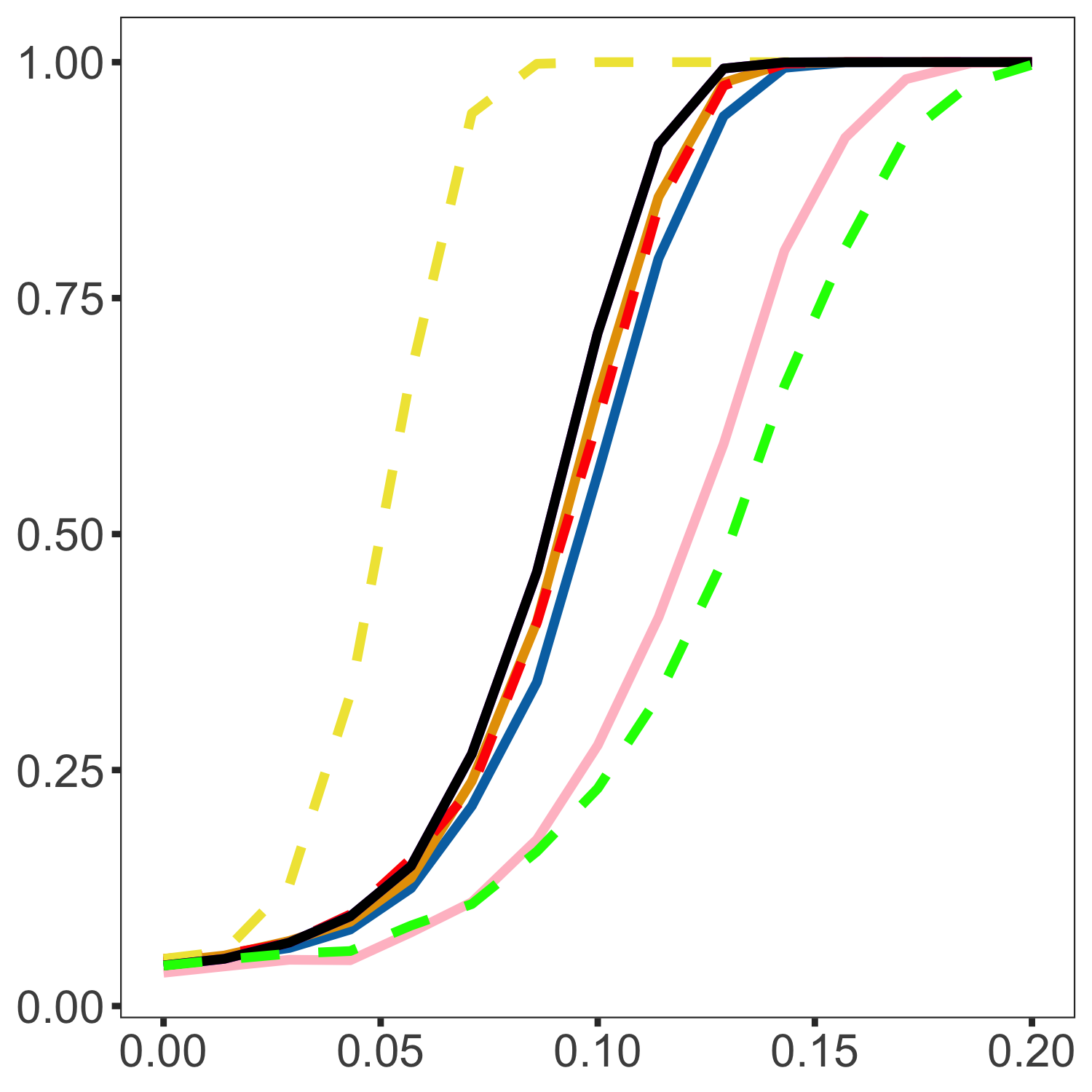}
    \end{subfigure}
     \begin{subfigure}[t]{0.23\textwidth}
        \centering
        \includegraphics[width=\linewidth, height=0.7\linewidth]{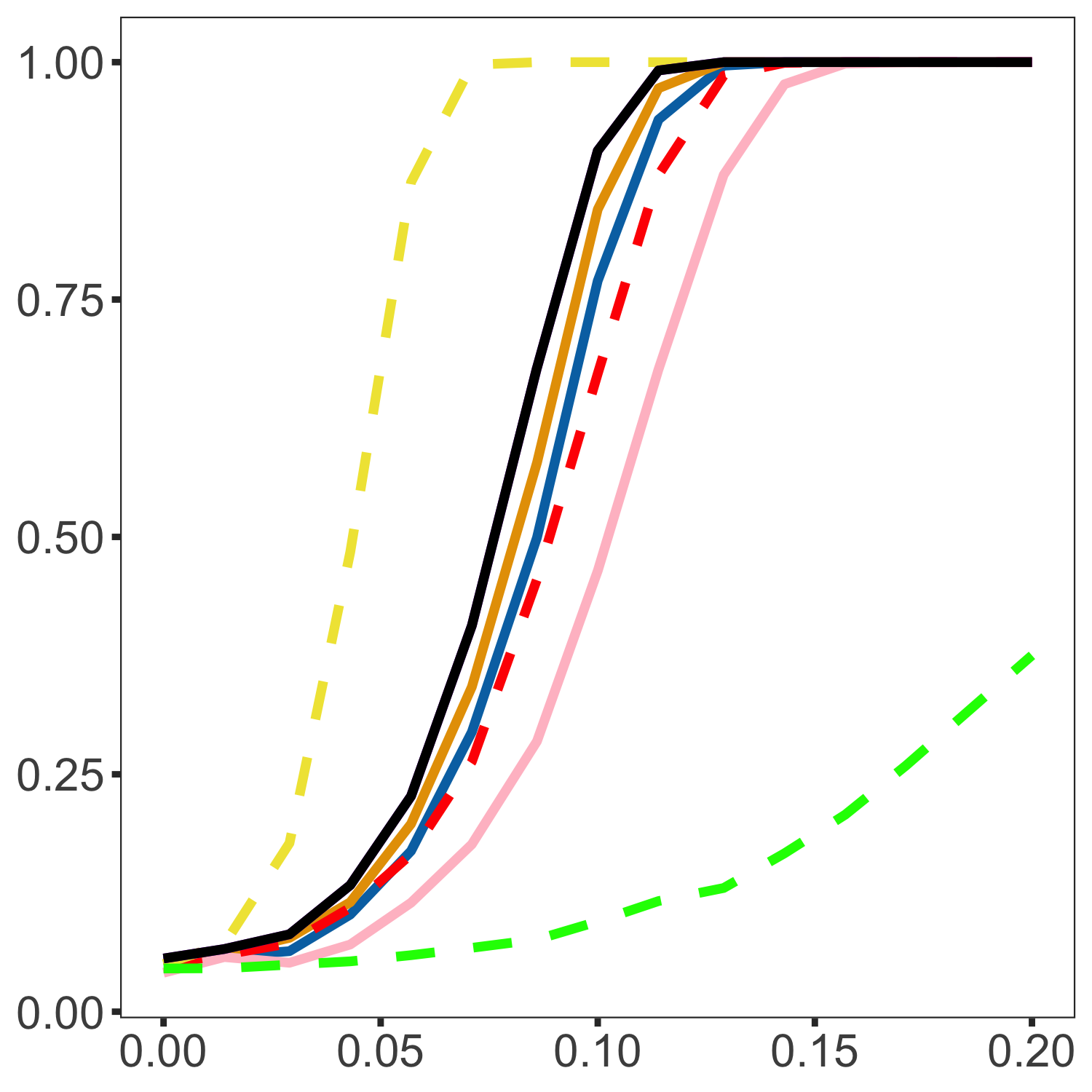}
    \end{subfigure}
    \begin{subfigure}[t]{0.23\textwidth}
        \centering
        \includegraphics[width=\linewidth, height=0.7\linewidth]{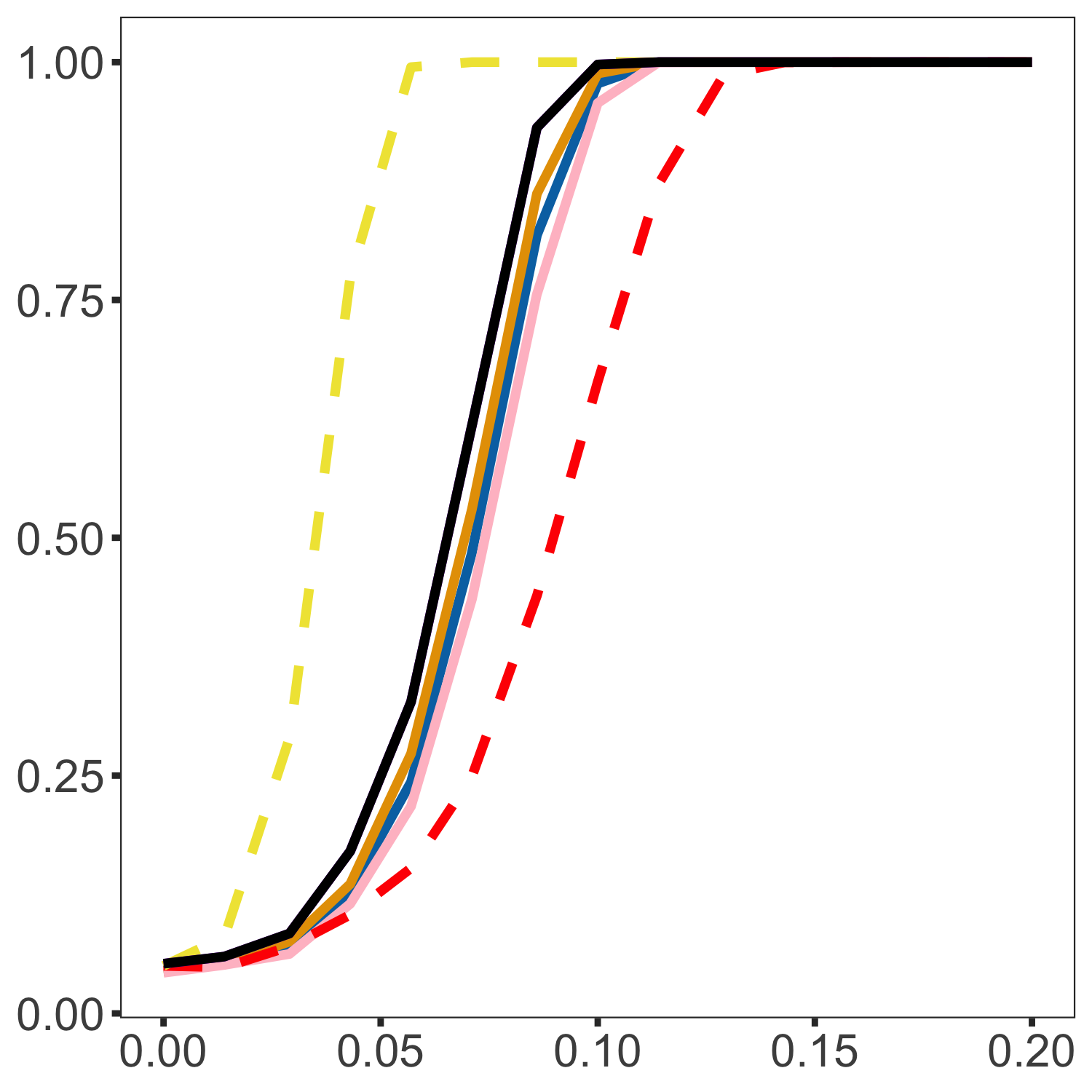}
    \end{subfigure}
    \vfill
    \begin{subfigure}[t]{0.23\textwidth}
        \centering
         \includegraphics[width=\linewidth, height=0.7\linewidth]{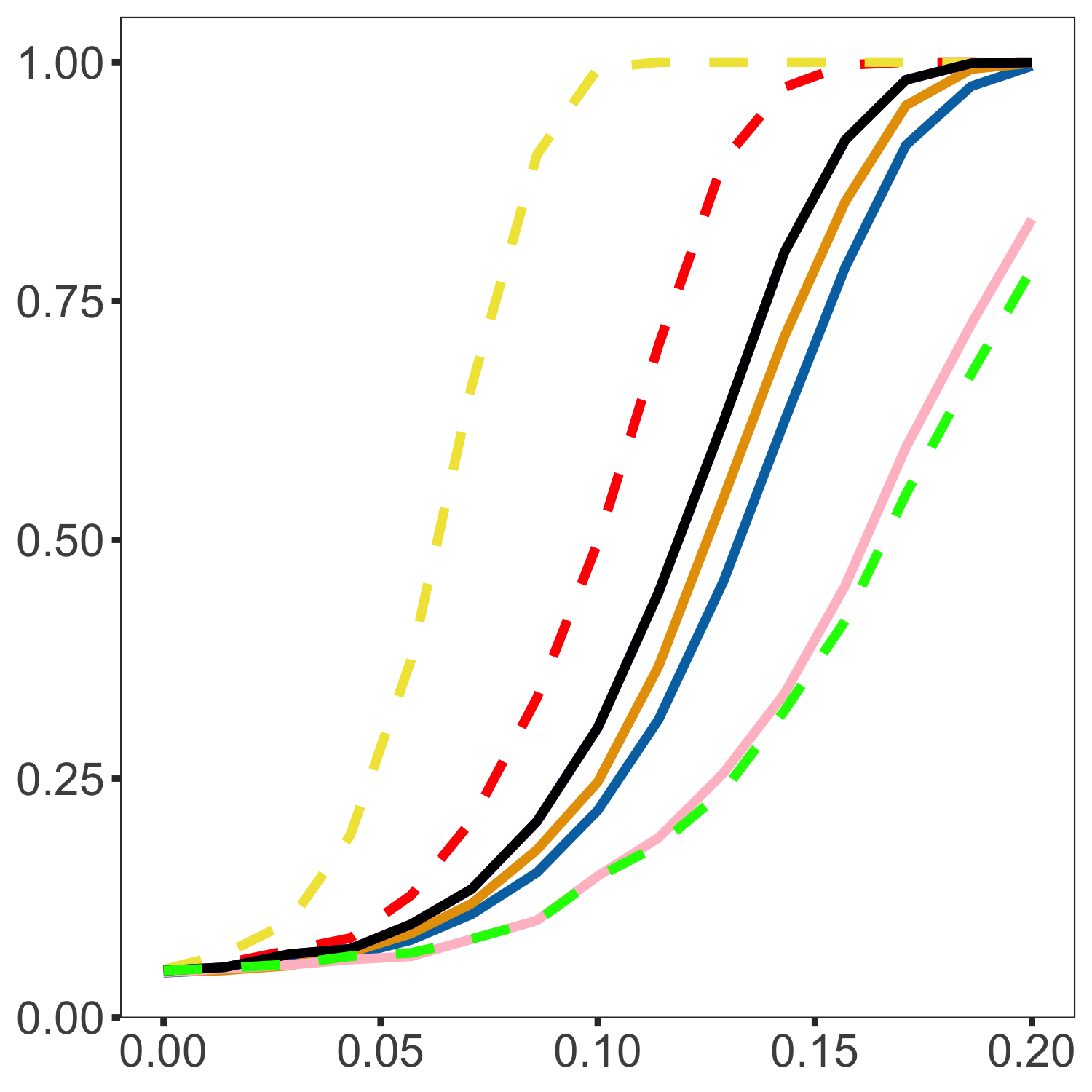}
    \end{subfigure}%
    \begin{subfigure}[t]{0.23\textwidth}
        \centering
        \includegraphics[width=\linewidth, height=0.7\linewidth]{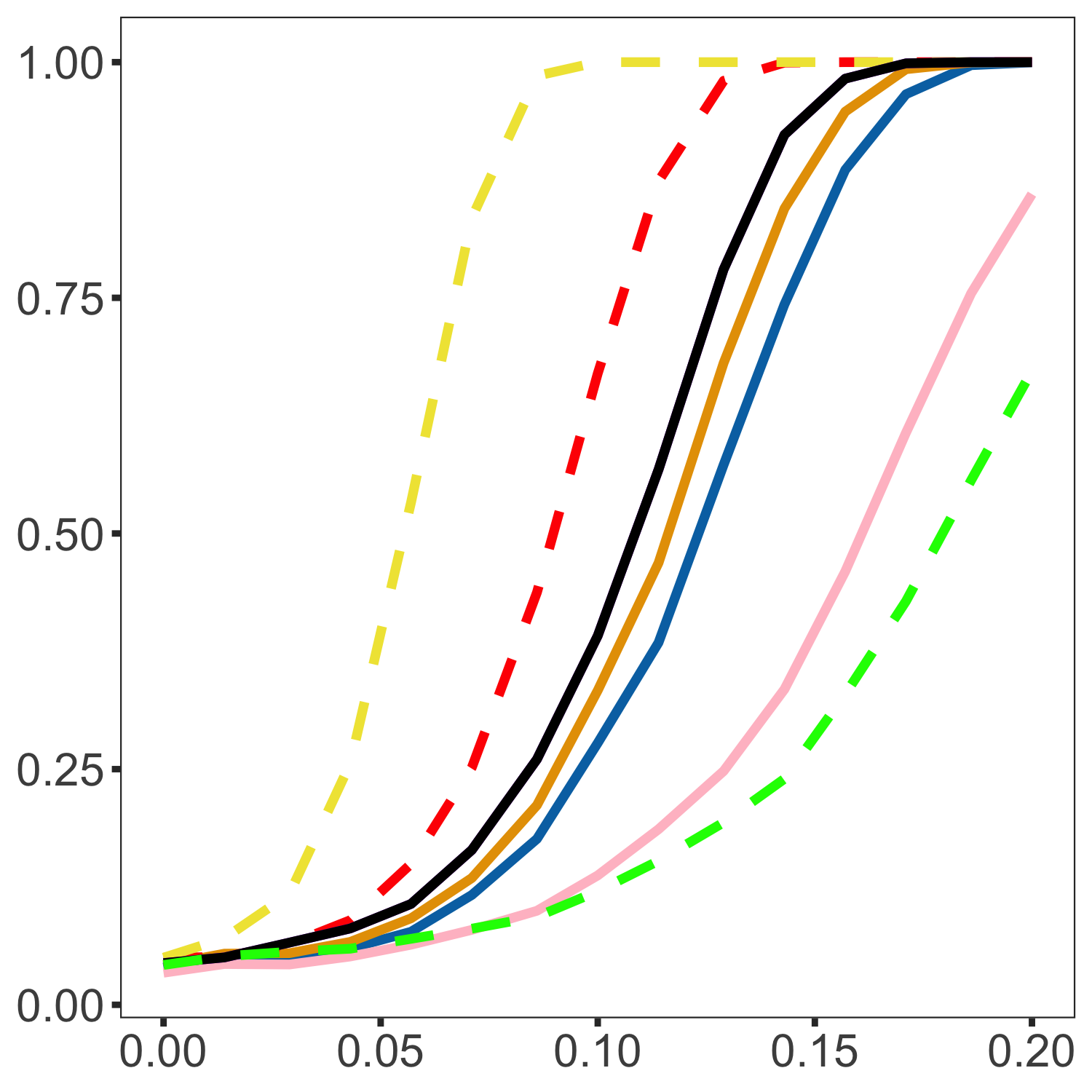}
    \end{subfigure}
     \begin{subfigure}[t]{0.23\textwidth}
        \centering
        \includegraphics[width=\linewidth, height=0.7\linewidth]{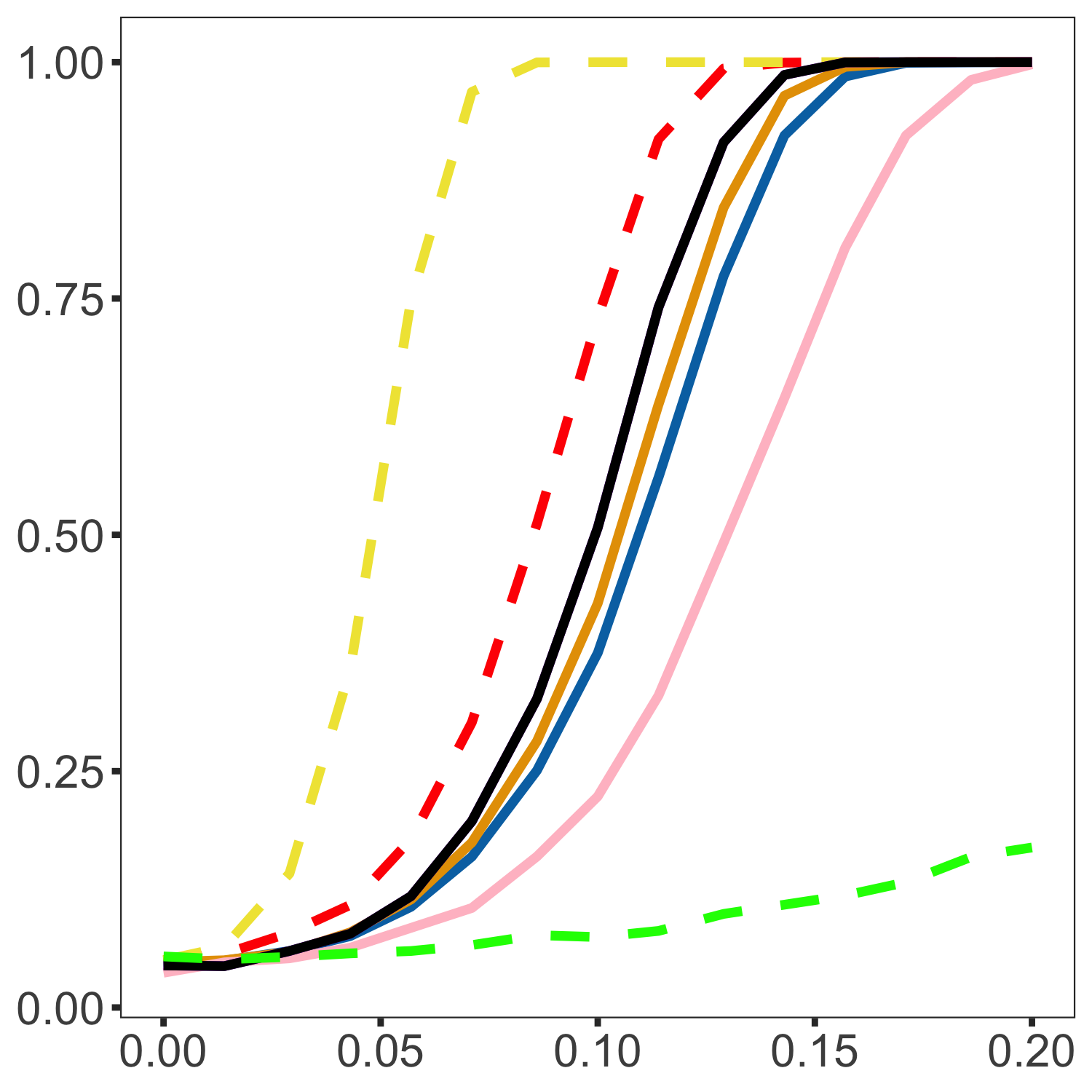}
    \end{subfigure}
    \begin{subfigure}[t]{0.23\textwidth}
        \centering
        \includegraphics[width=\linewidth, height=0.7\linewidth]{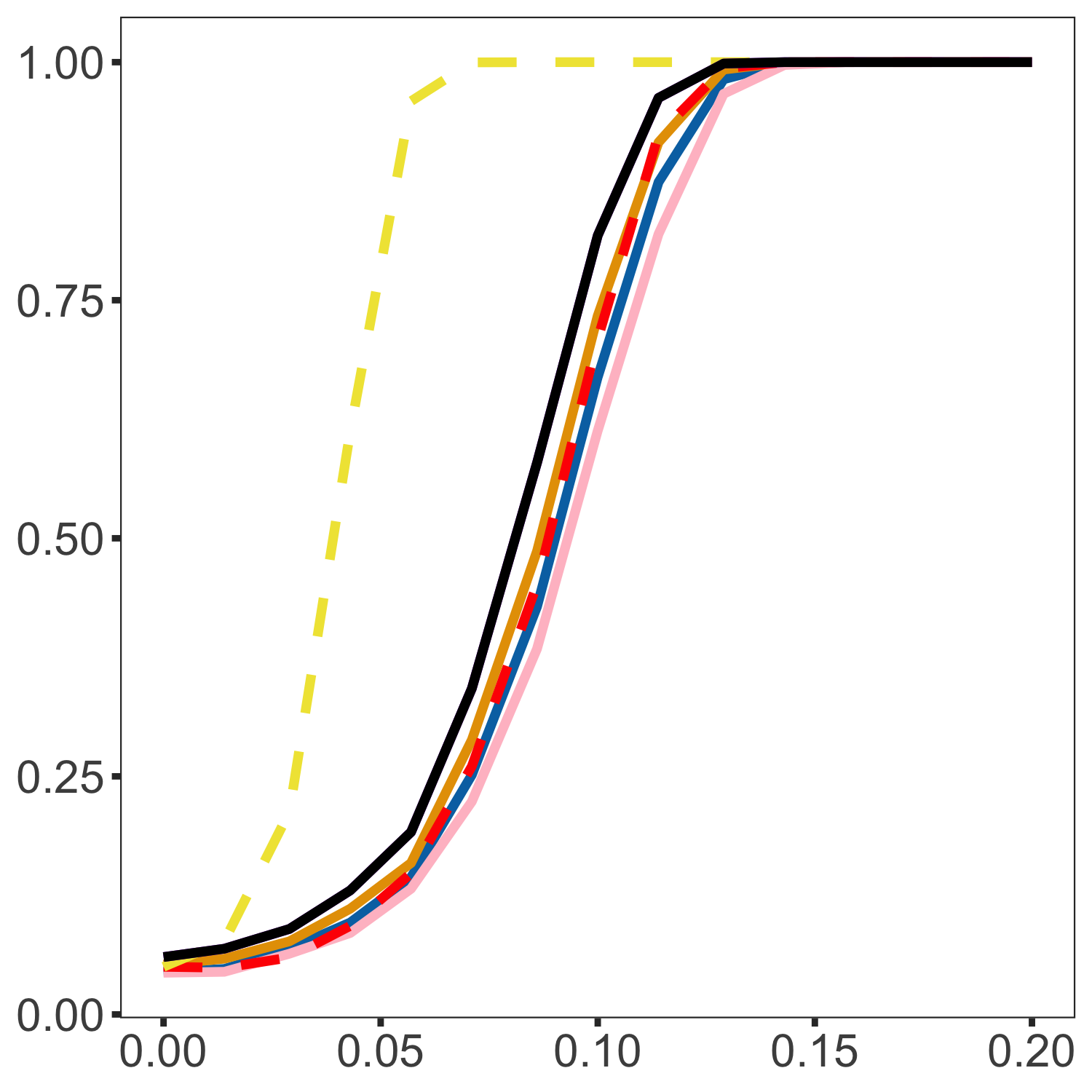}
    \end{subfigure}
   \caption{Size-adjusted empirical power under high-rank alternatives when \(\Sigma\) is Identity. Columns (left to right) correspond to \(\hat{\gamma}_2 = 0.3, 0.5, 0.9, 2\); rows (top to bottom) correspond to \(n_1 = 50, 100, 250\). Solid curves: blue (\(\lambda = 0.5\)), orange (\(\lambda = 1\)), black (\(\lambda=\hat{\lambda}_{I_p}\)), purple (\(\lambda=\hat{\lambda}_{\Sigma_p}\)), and pink (\(\lambda=\hat{\lambda}_*\)). Dashed curves: red (Proj-LRT), yellow (Ridge-LRT), and green (\cite{han2016tracy}, \(\lambda=0\)), the latter available only when \(p<n_1+n_2\).}
    \label{fig:FR_emp_power_Sigma1}
\end{figure}

\begin{figure}[htbp]
    \centering
    \begin{subfigure}[t]{0.23\textwidth}
        \centering
         \includegraphics[width=\linewidth, height=0.7\linewidth]{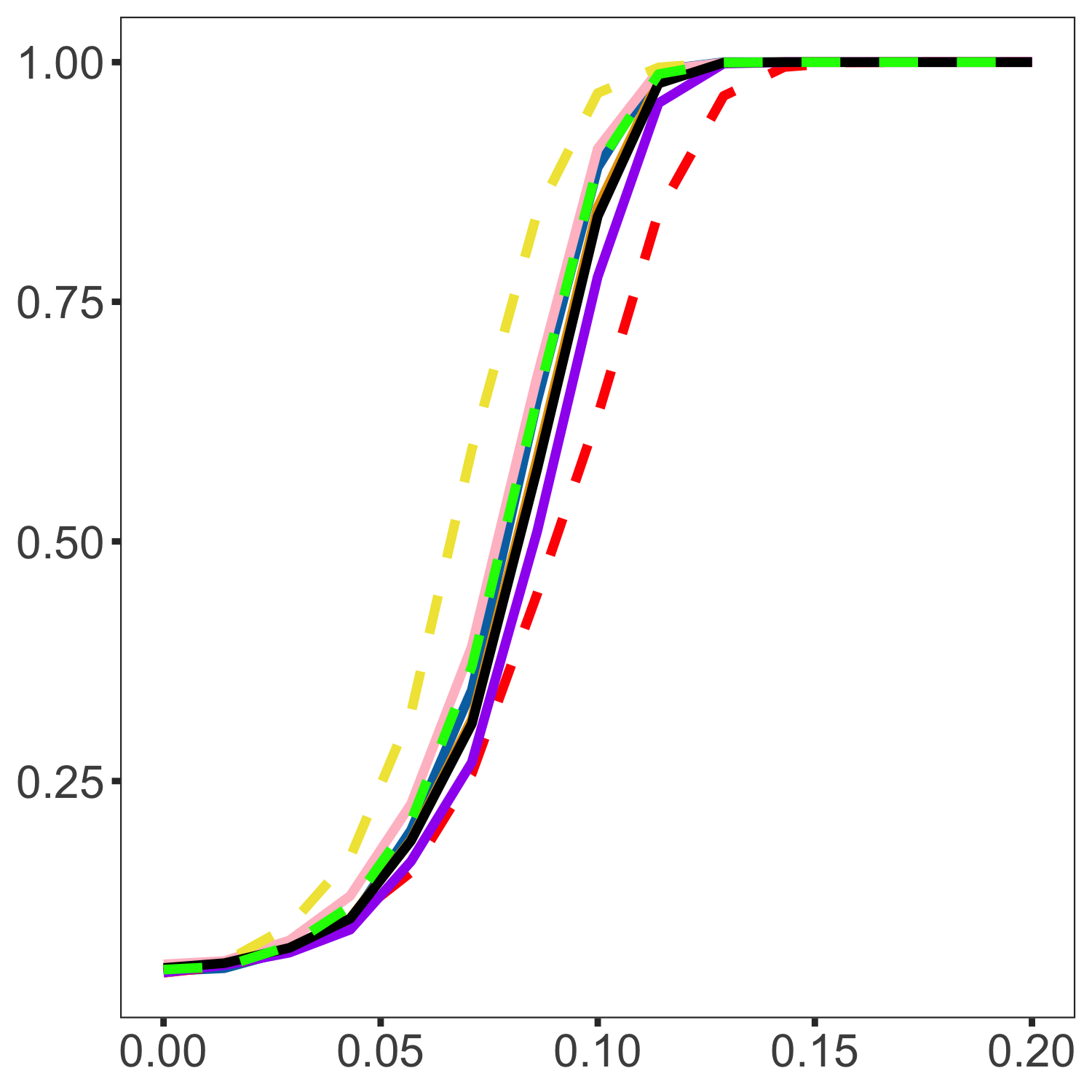}
    \end{subfigure}%
    \begin{subfigure}[t]{0.23\textwidth}
        \centering
        \includegraphics[width=\linewidth, height=0.7\linewidth]{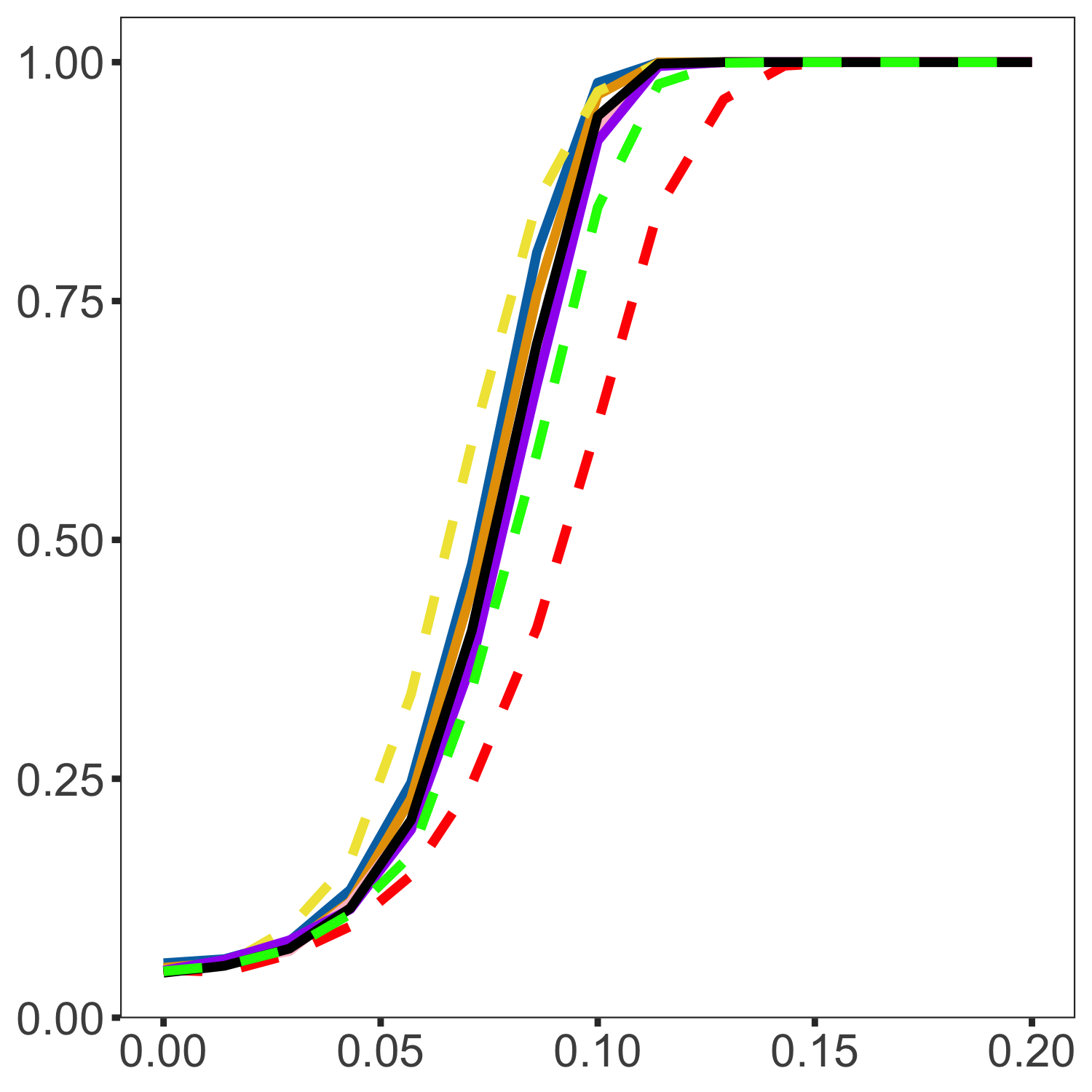}
    \end{subfigure}
     \begin{subfigure}[t]{0.23\textwidth}
        \centering
        \includegraphics[width=\linewidth, height=0.7\linewidth]{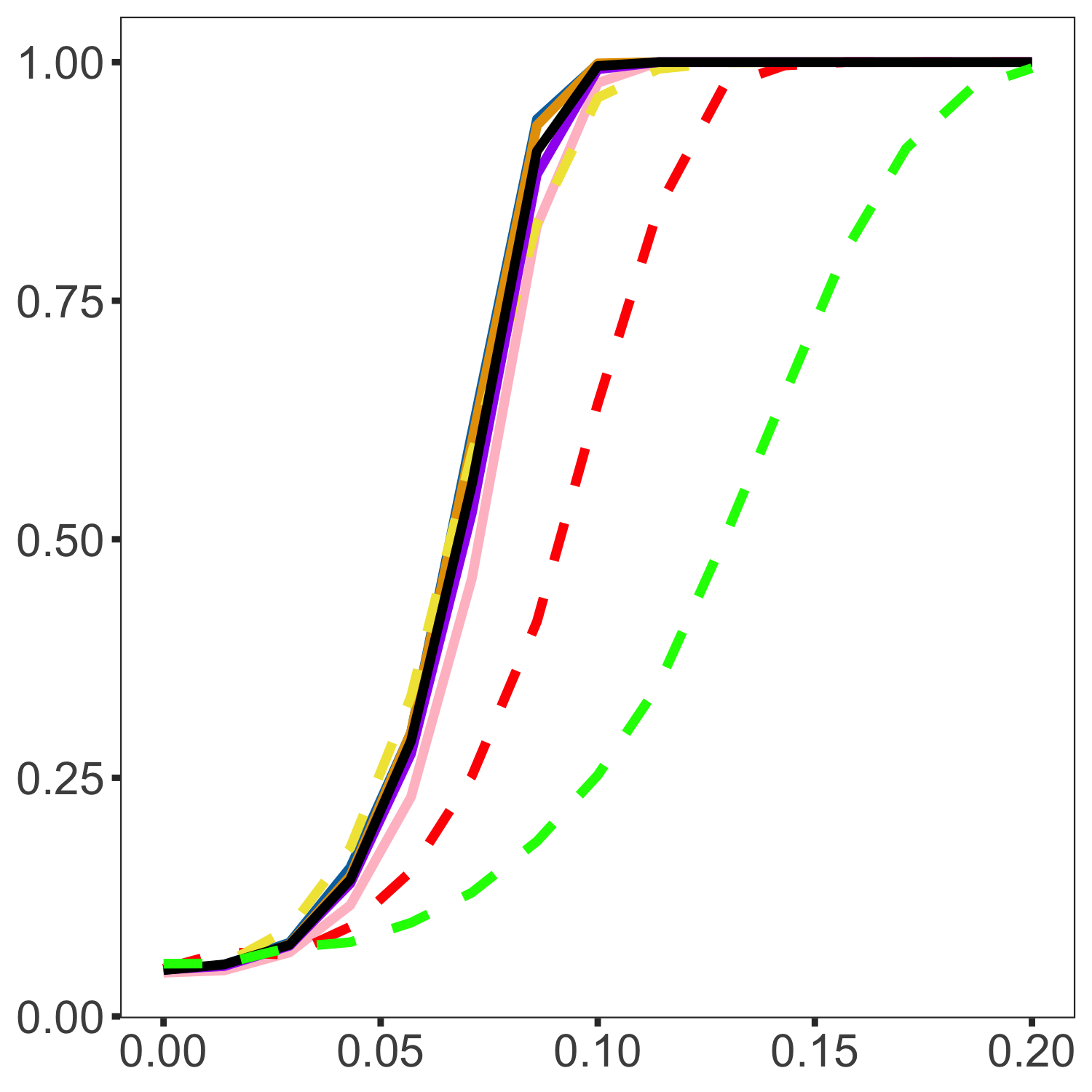}
    \end{subfigure}
    \begin{subfigure}[t]{0.23\textwidth}
        \centering
        \includegraphics[width=\linewidth, height=0.7\linewidth]{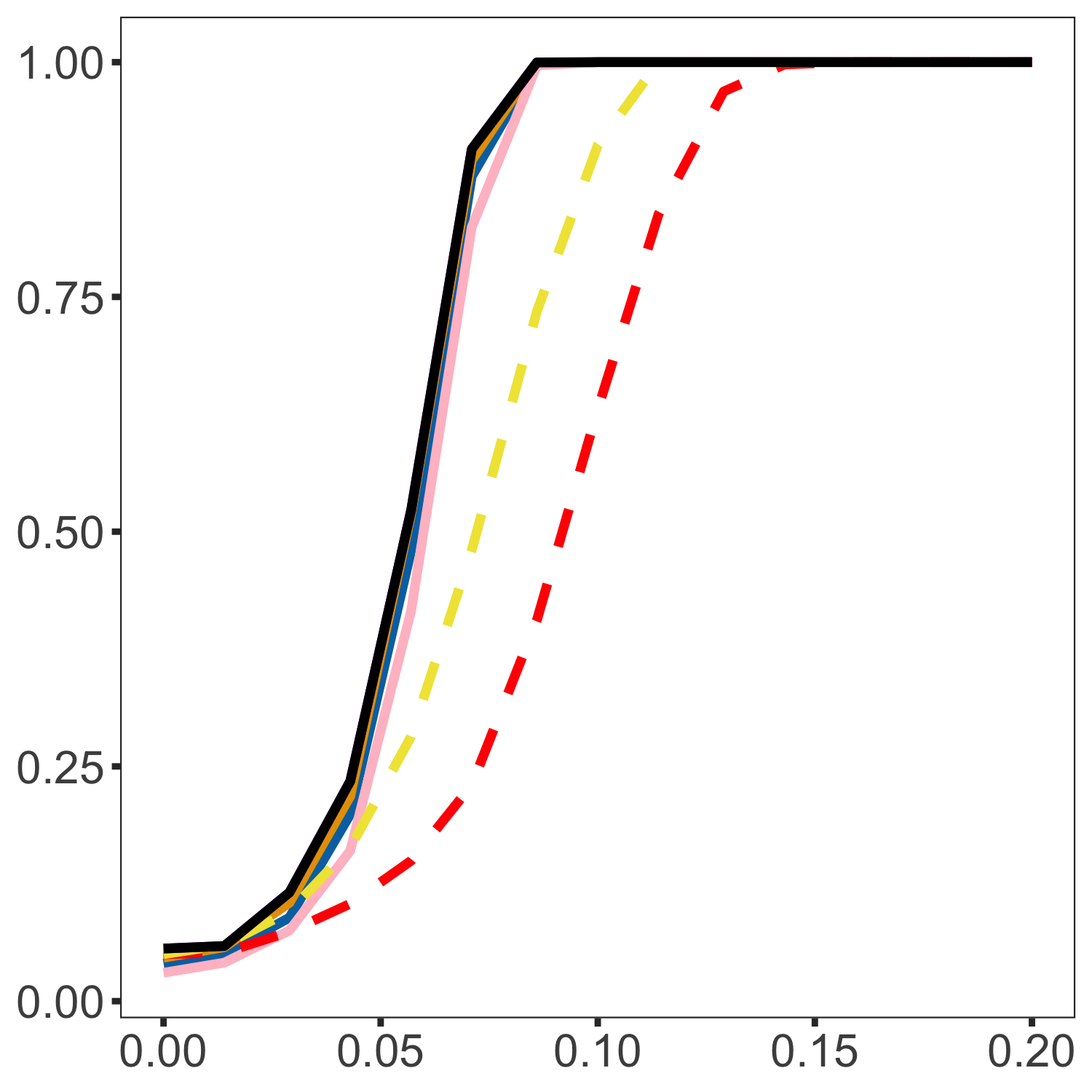}
    \end{subfigure}
    \vfill
    \begin{subfigure}[t]{0.23\textwidth}
        \centering
         \includegraphics[width=\linewidth, height=0.7\linewidth]{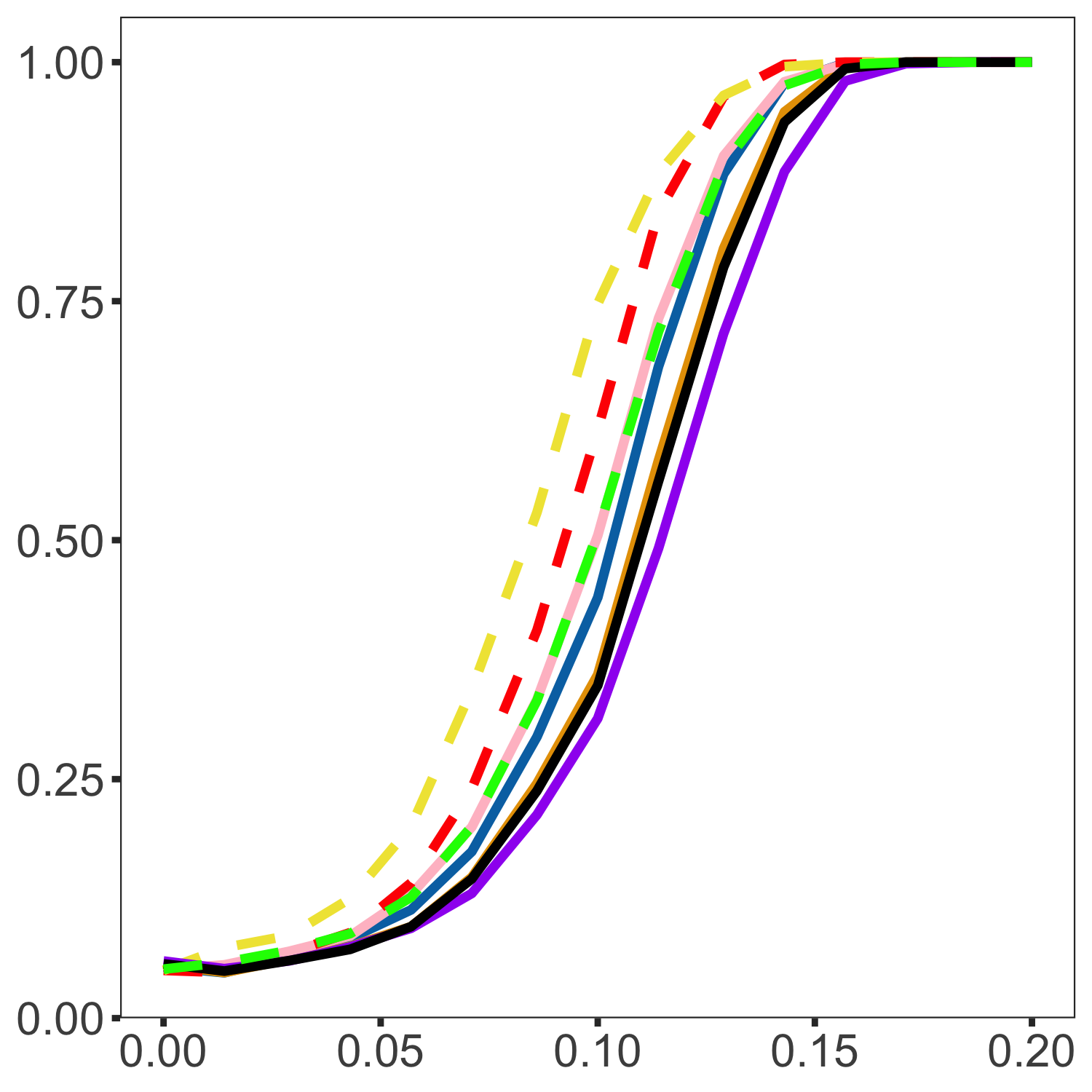}
    \end{subfigure}%
    \begin{subfigure}[t]{0.23\textwidth}
        \centering
        \includegraphics[width=\linewidth, height=0.7\linewidth]{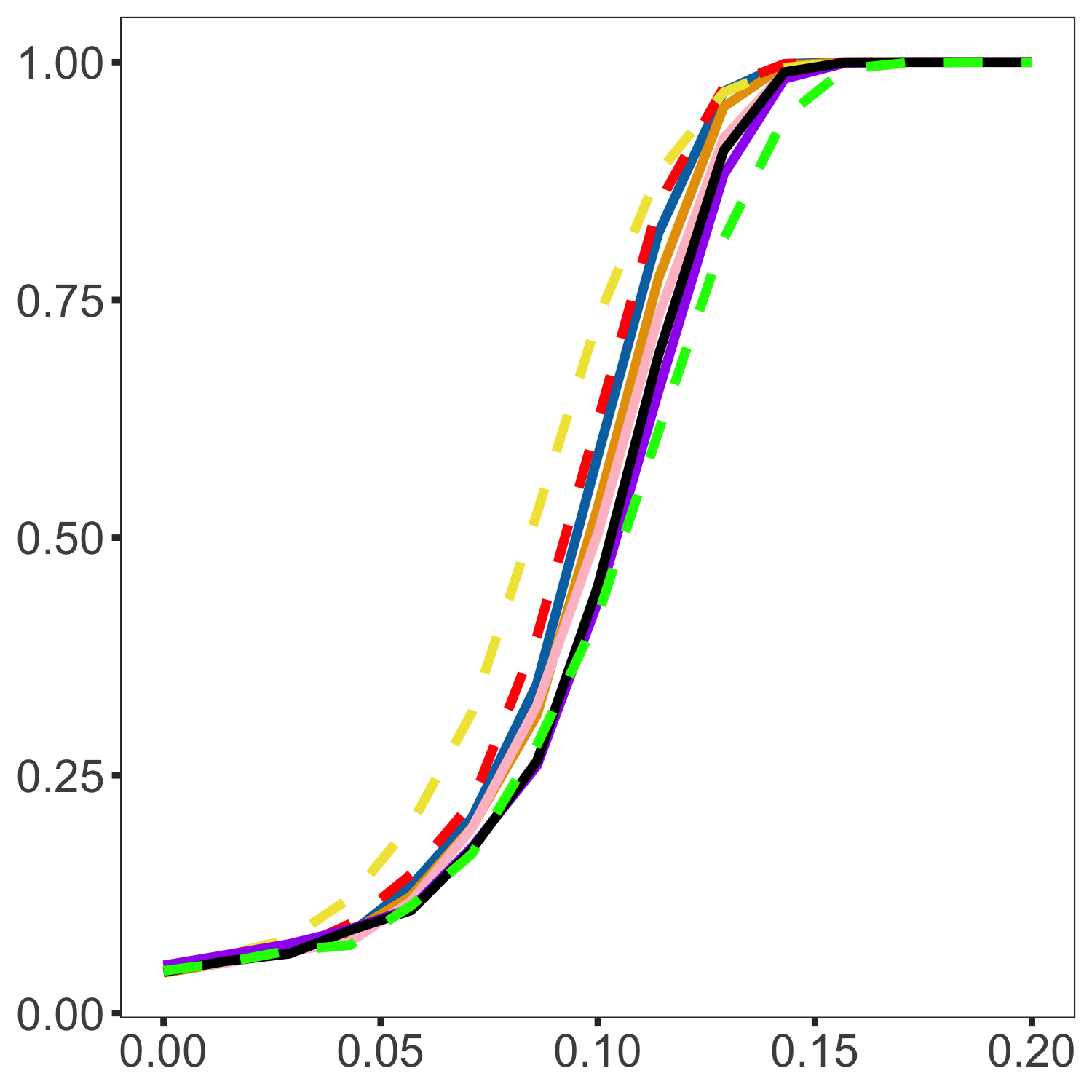}
    \end{subfigure}
     \begin{subfigure}[t]{0.23\textwidth}
        \centering
        \includegraphics[width=\linewidth, height=0.7\linewidth]{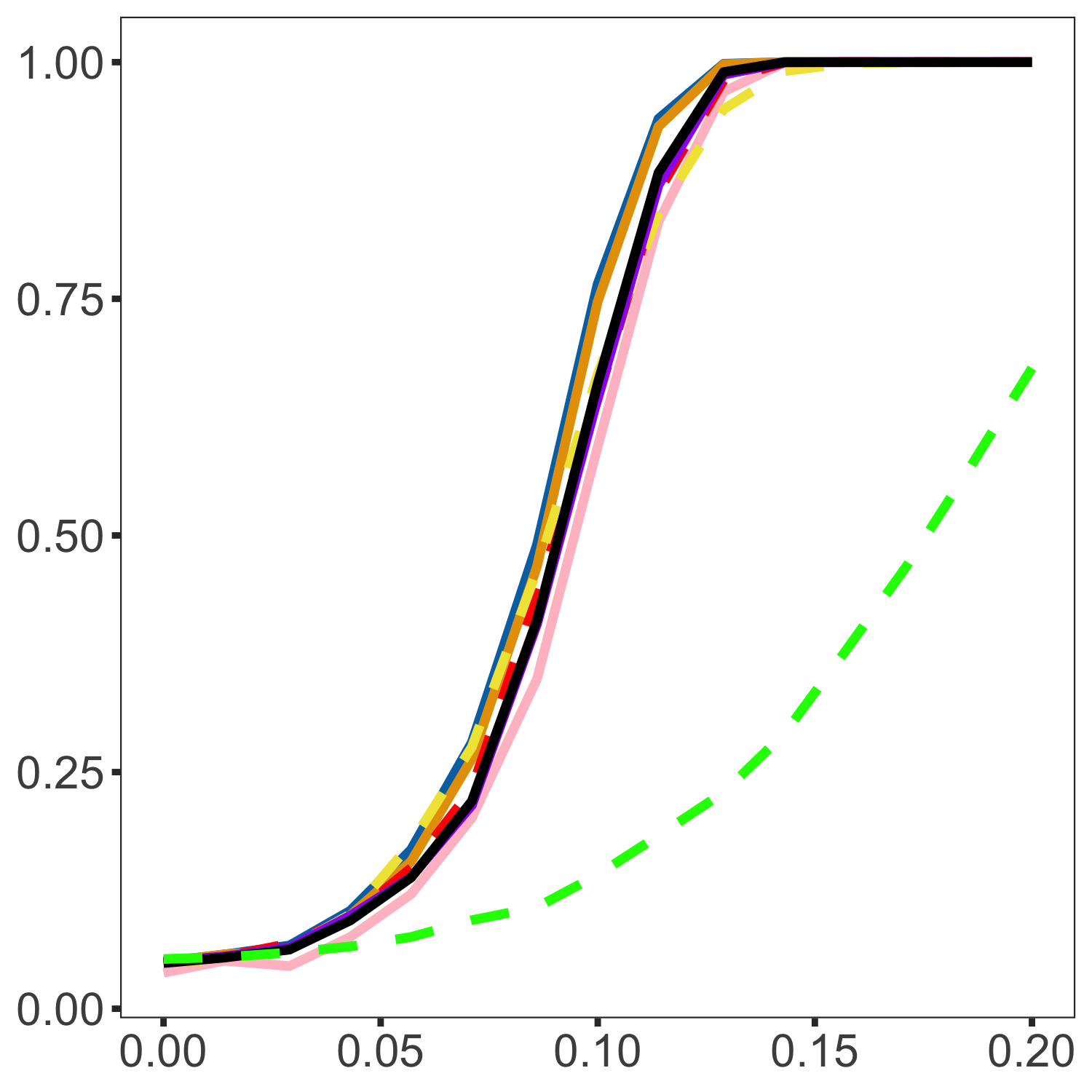}
    \end{subfigure}
    \begin{subfigure}[t]{0.23\textwidth}
        \centering
        \includegraphics[width=\linewidth, height=0.7\linewidth]{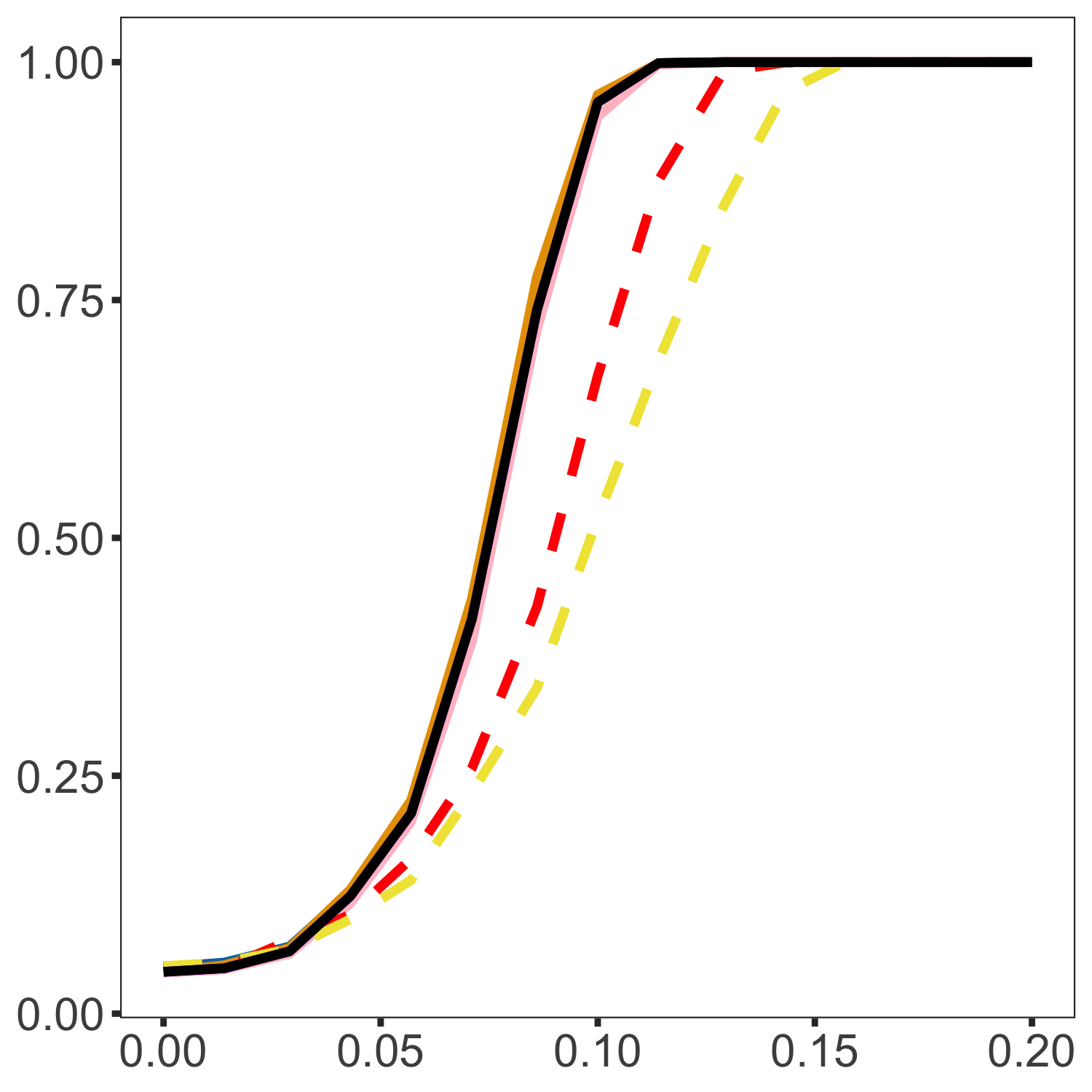}
    \end{subfigure}
    \vfill
    \begin{subfigure}[t]{0.23\textwidth}
        \centering
         \includegraphics[width=\linewidth, height=0.7\linewidth]{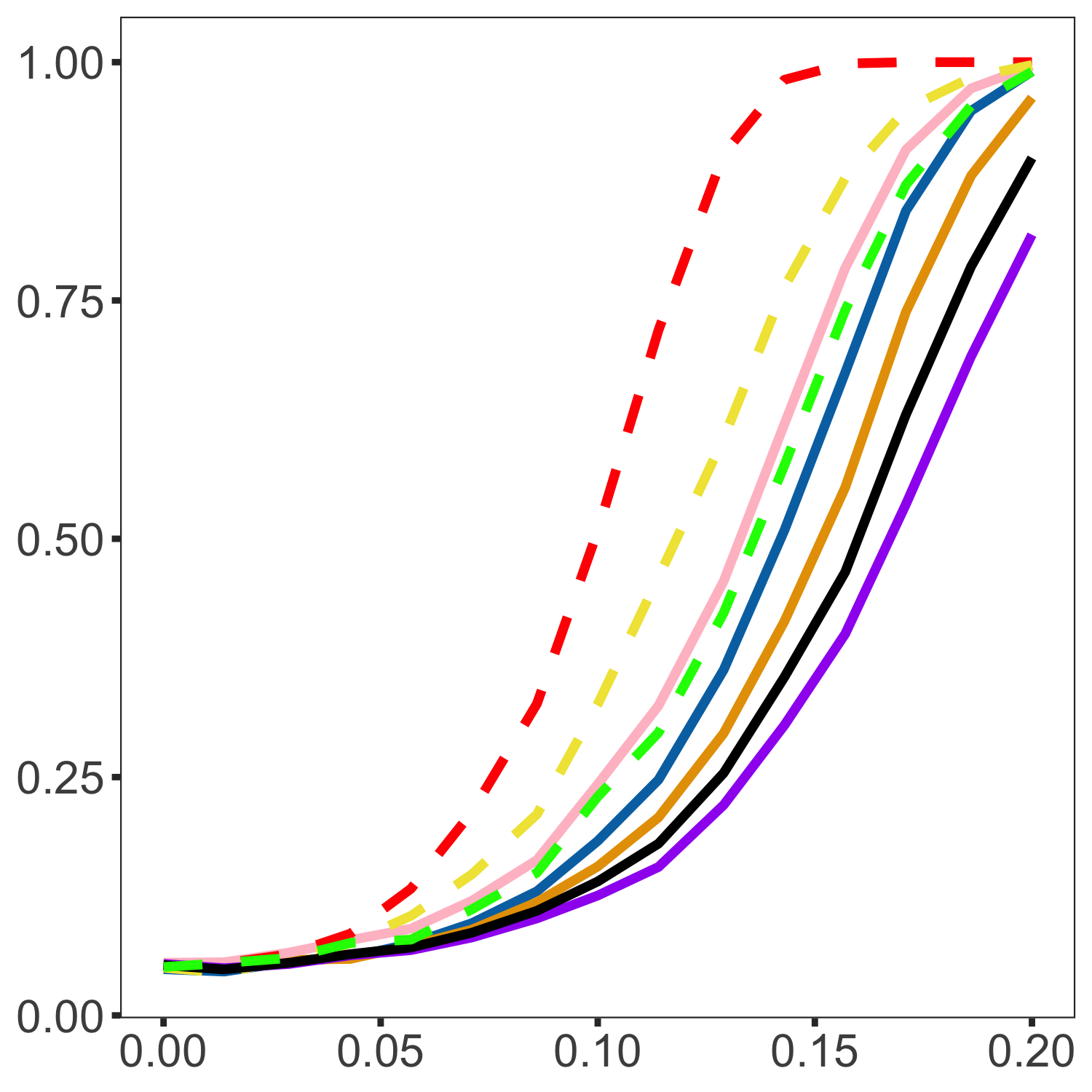}
    \end{subfigure}%
    \begin{subfigure}[t]{0.23\textwidth}
        \centering
        \includegraphics[width=\linewidth, height=0.7\linewidth]{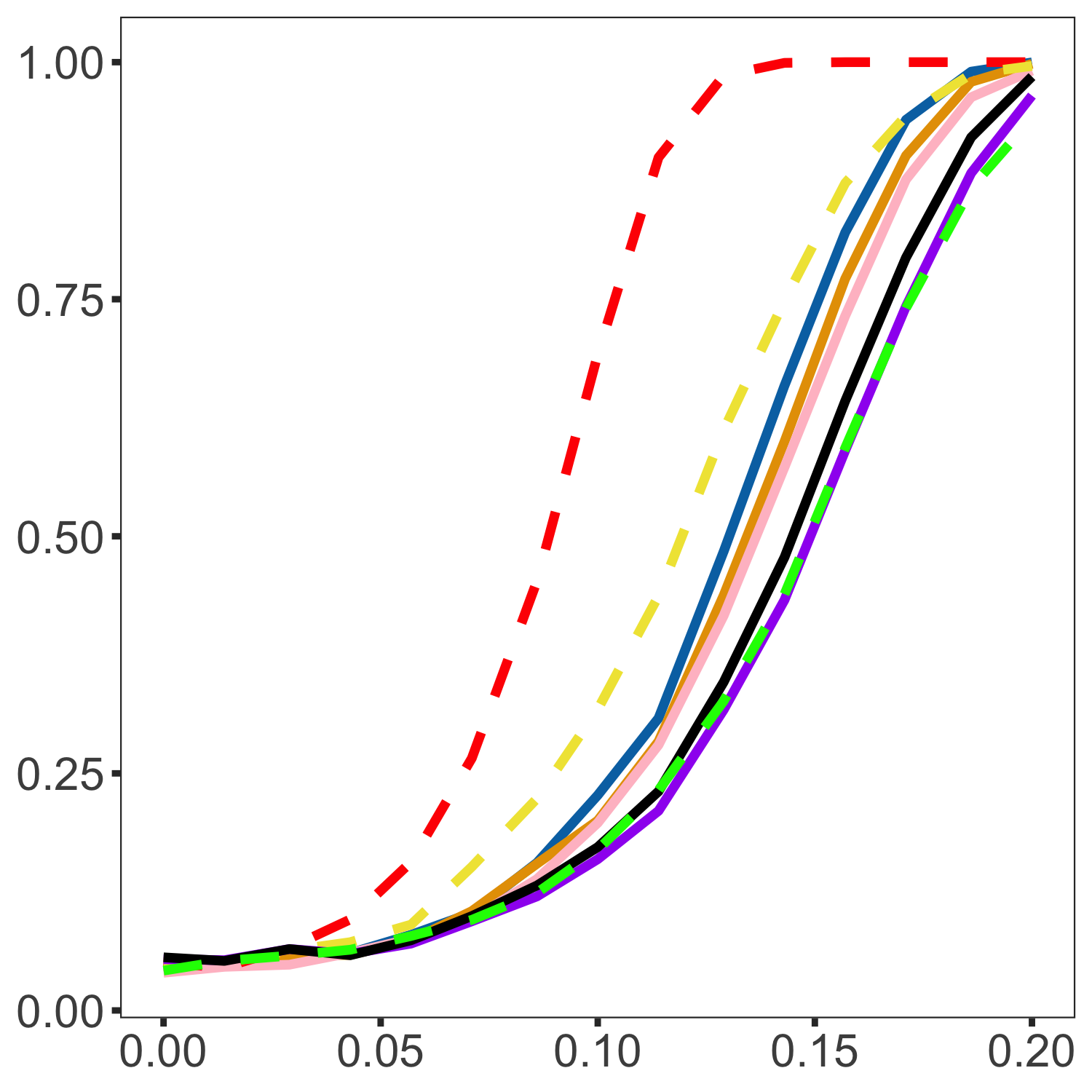}
    \end{subfigure}
     \begin{subfigure}[t]{0.23\textwidth}
        \centering
        \includegraphics[width=\linewidth, height=0.7\linewidth]{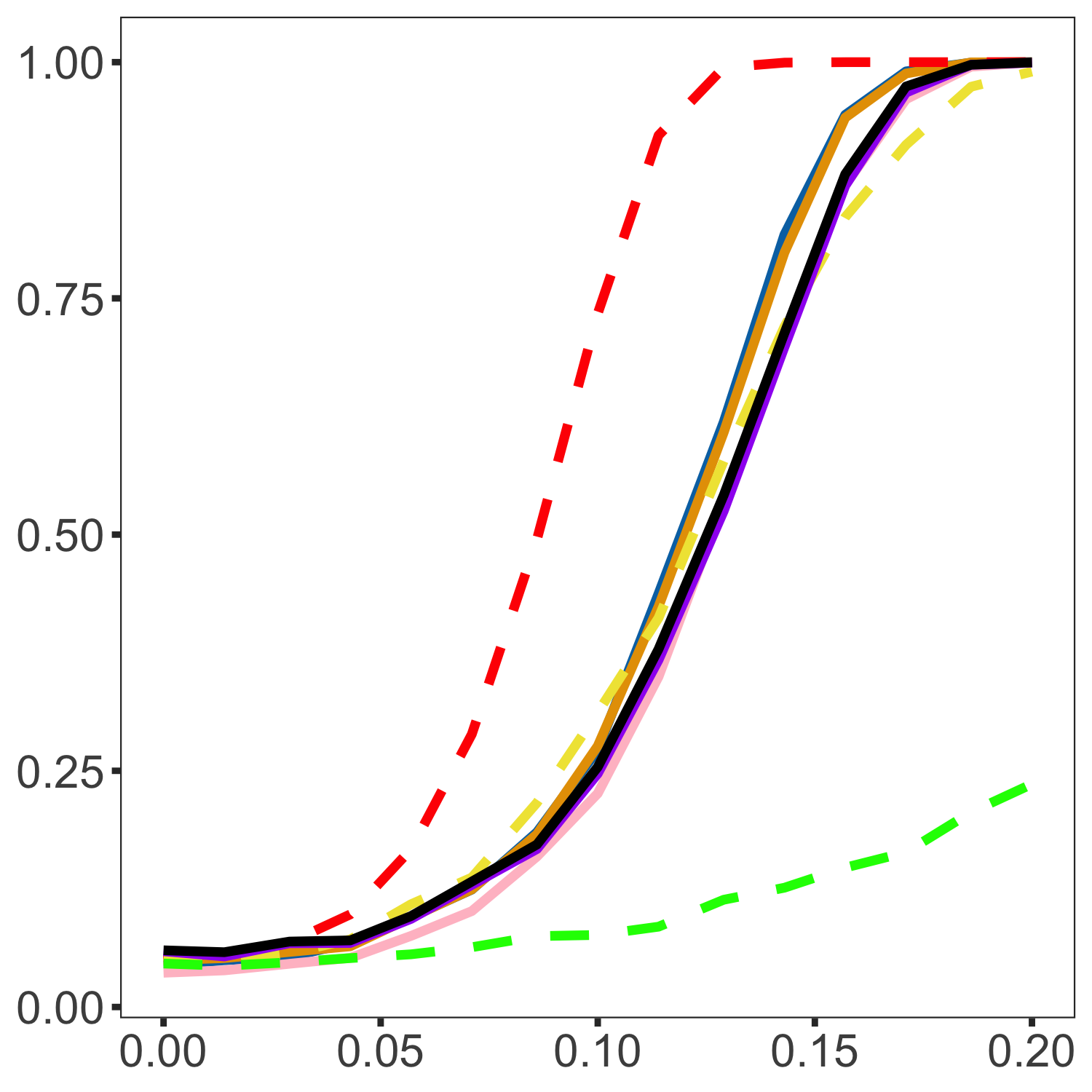}
    \end{subfigure}
    \begin{subfigure}[t]{0.23\textwidth}
        \centering
        \includegraphics[width=\linewidth, height=0.7\linewidth]{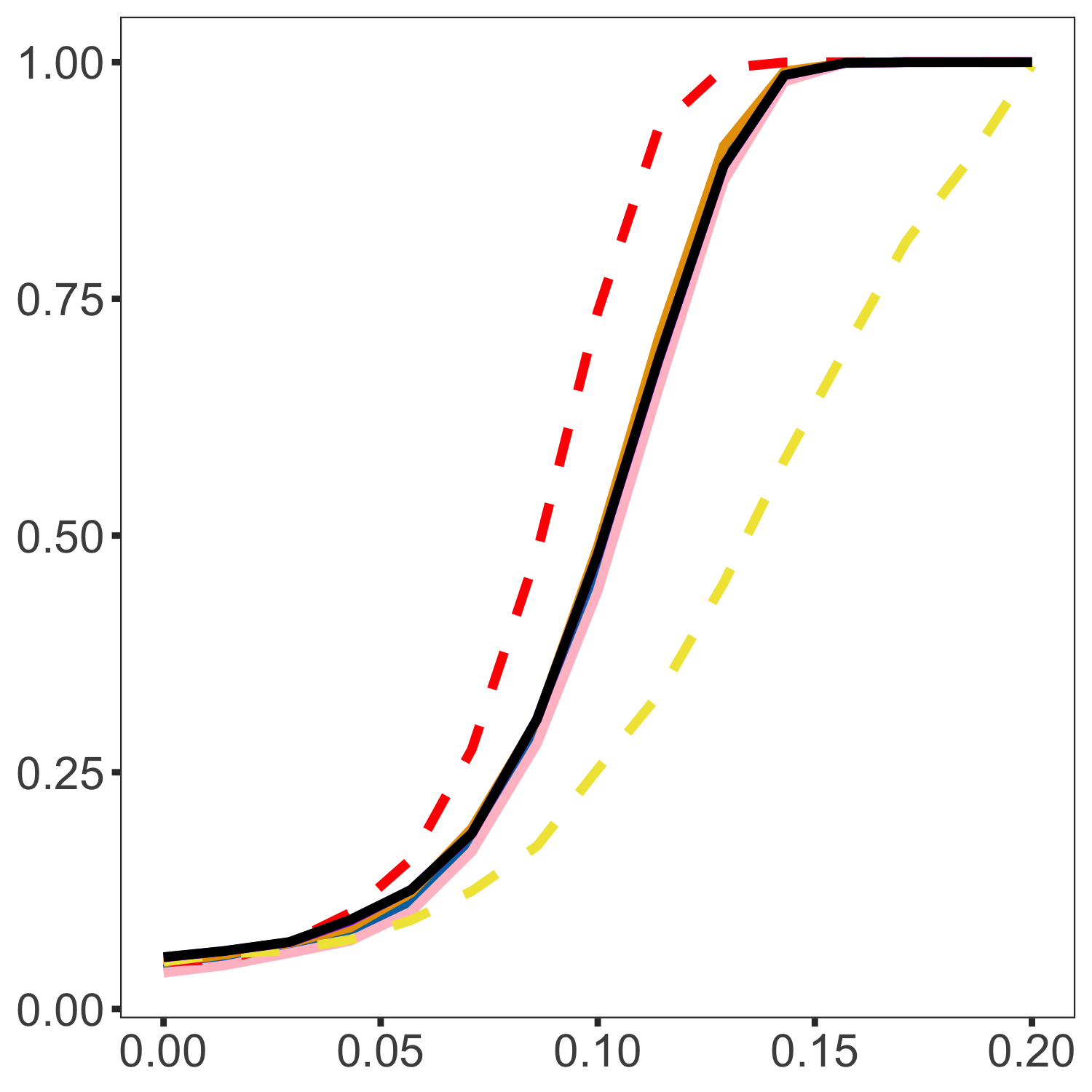}
    \end{subfigure}
   \caption{Size-adjusted empirical power under high-rank alternatives when \(\Sigma\) is Poly-Decay. Columns (left to right) correspond to \(\hat{\gamma}_2 = 0.3, 0.5, 0.9, 2\); rows (top to bottom) correspond to \(n_1 = 50, 100, 250\). Solid curves: blue (\(\lambda = 0.5\)), orange (\(\lambda = 1\)), black (\(\lambda=\hat{\lambda}_{I_p}\)), purple (\(\lambda=\hat{\lambda}_{\Sigma_p}\)), and pink (\(\lambda=\hat{\lambda}_*\)). Dashed curves: red (Proj-LRT), yellow (Ridge-LRT), and green (\cite{han2016tracy}, \(\lambda=0\)), the latter available only when \(p<n_1+n_2\).}
    \label{fig:FR_emp_power_Sigma2}
\end{figure}

\begin{figure}[htbp]
    \centering
    \begin{subfigure}[t]{0.23\textwidth}
        \centering
         \includegraphics[width=\linewidth, height=0.7\linewidth]{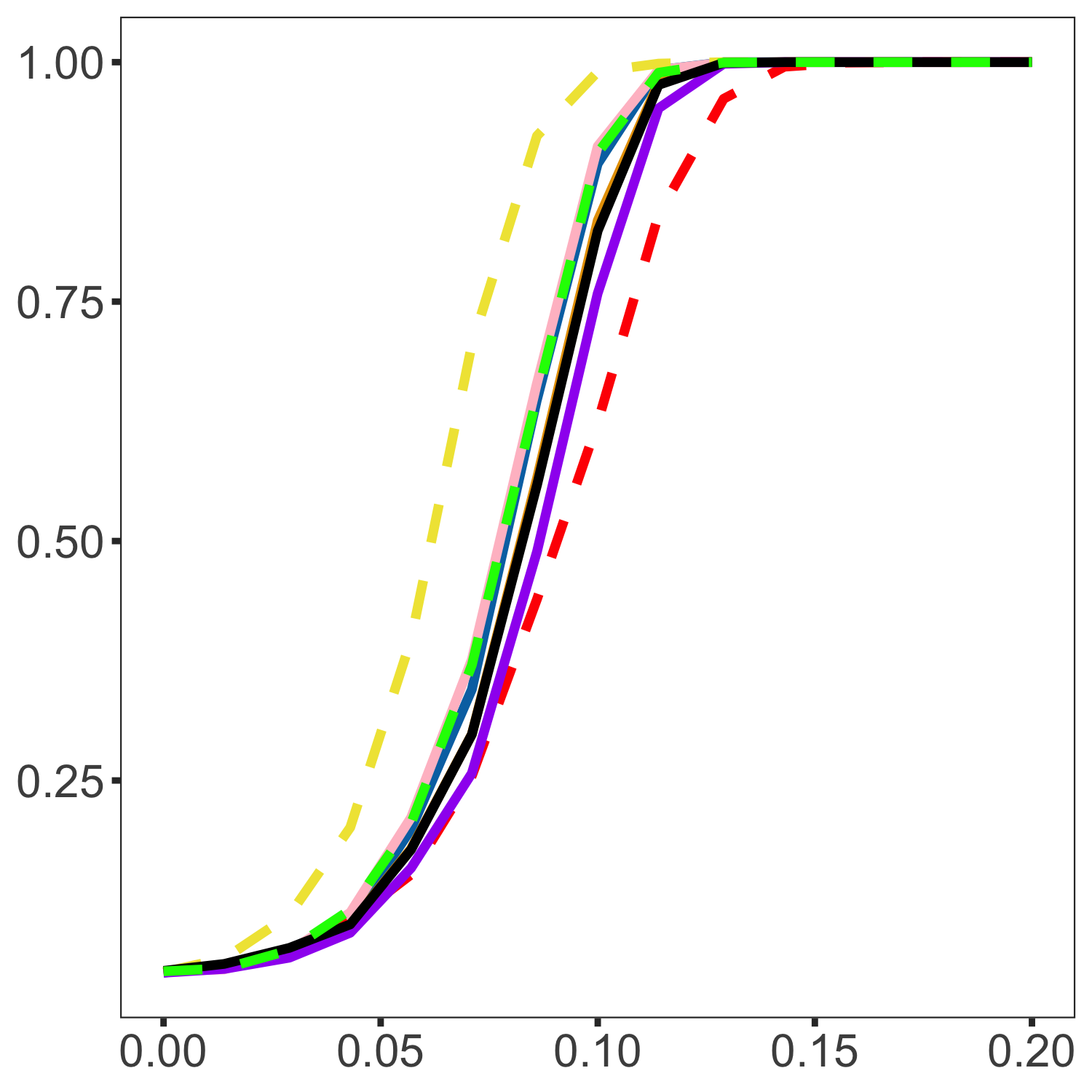}
    \end{subfigure}%
    \begin{subfigure}[t]{0.23\textwidth}
        \centering
        \includegraphics[width=\linewidth, height=0.7\linewidth]{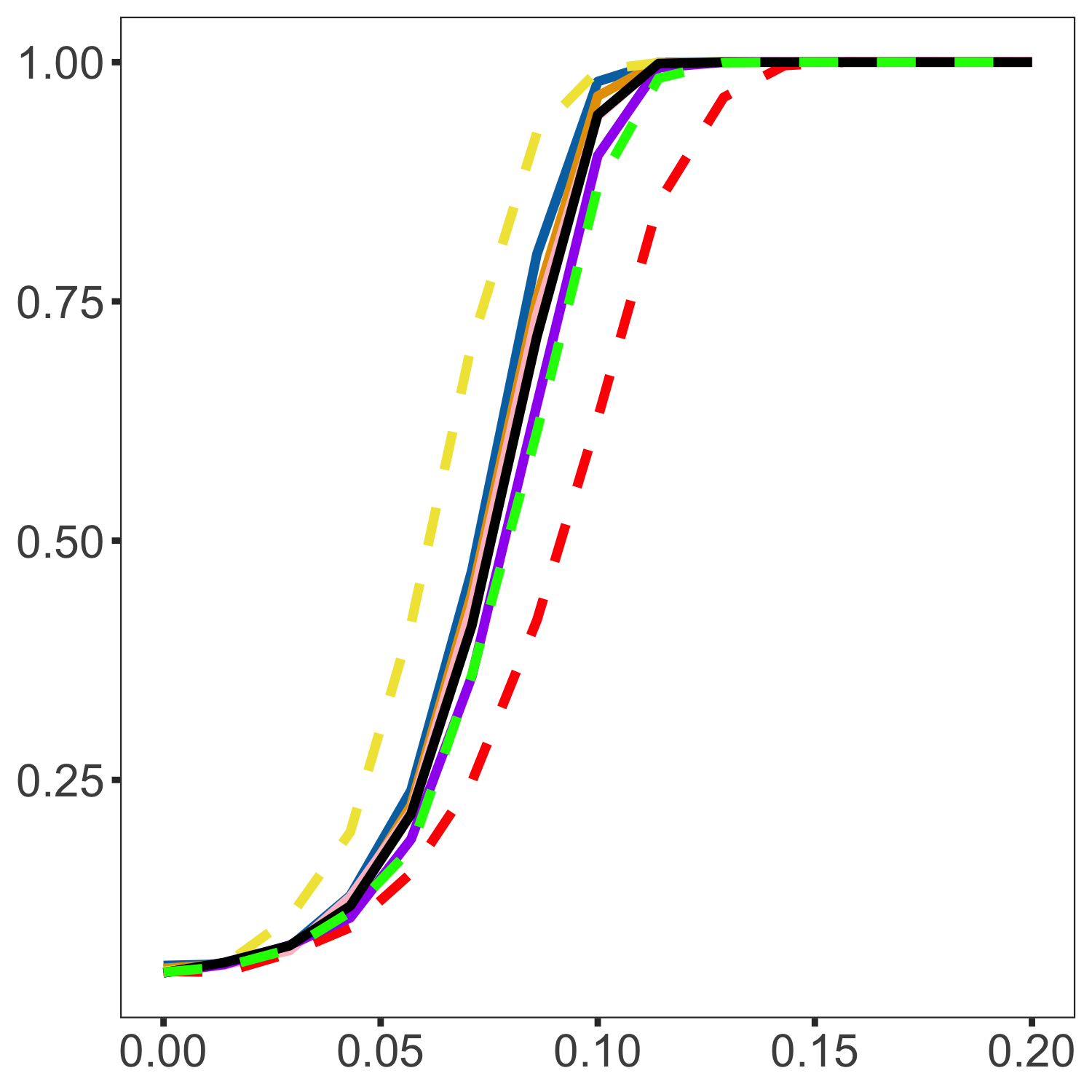}
    \end{subfigure}
     \begin{subfigure}[t]{0.23\textwidth}
        \centering
        \includegraphics[width=\linewidth, height=0.7\linewidth]{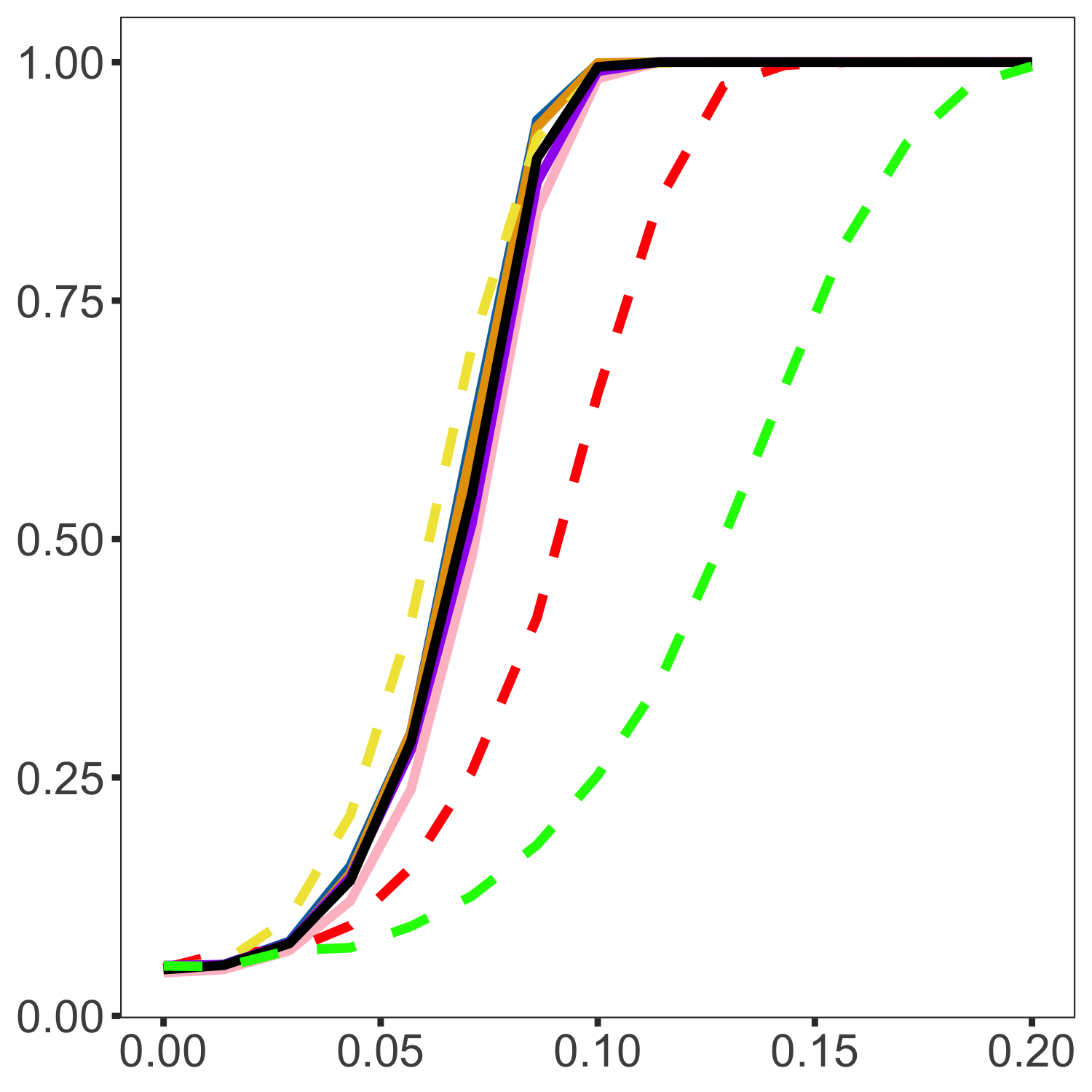}
    \end{subfigure}
    \begin{subfigure}[t]{0.23\textwidth}
        \centering
        \includegraphics[width=\linewidth, height=0.7\linewidth]{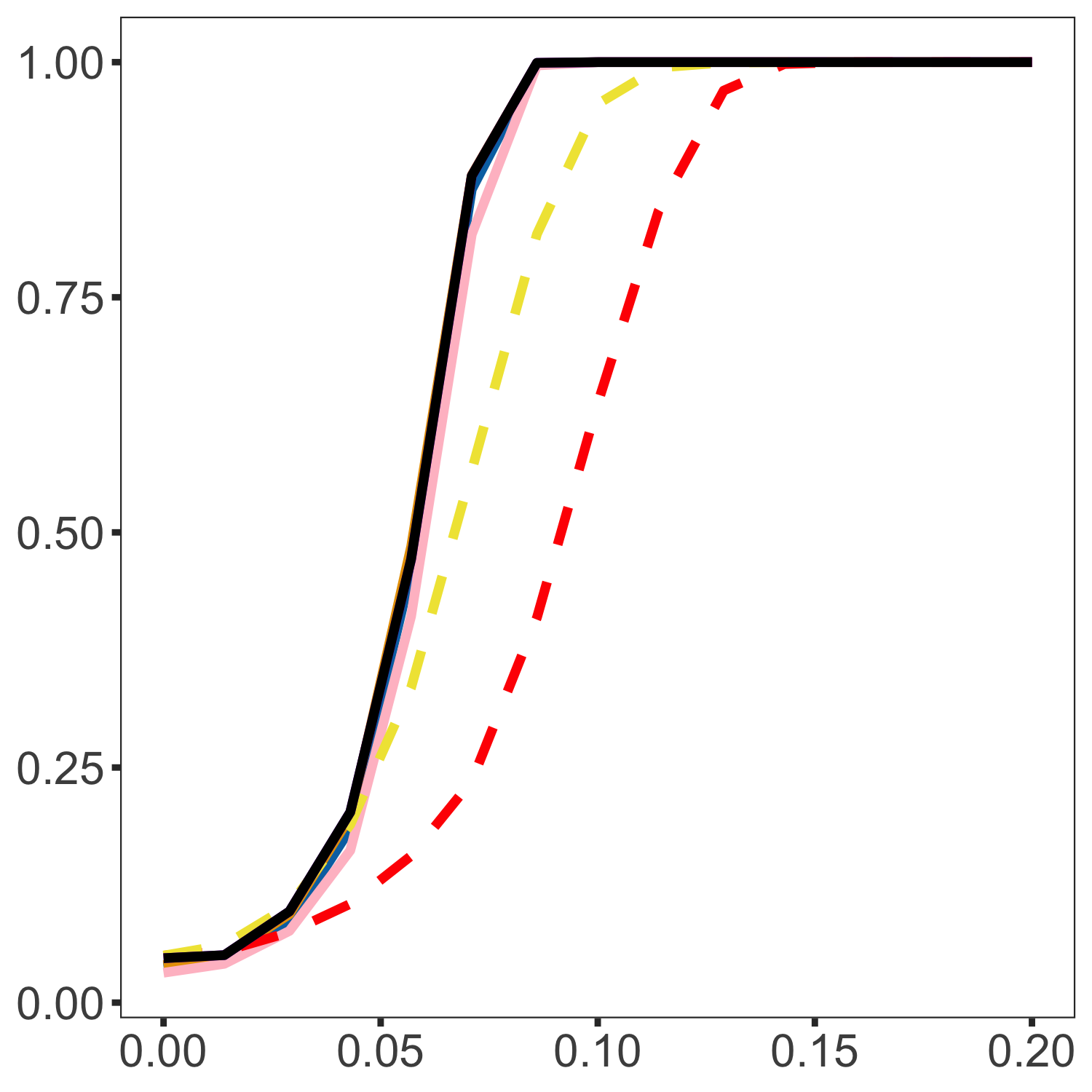}
    \end{subfigure}
    \vfill
    \begin{subfigure}[t]{0.23\textwidth}
        \centering
         \includegraphics[width=\linewidth, height=0.7\linewidth]{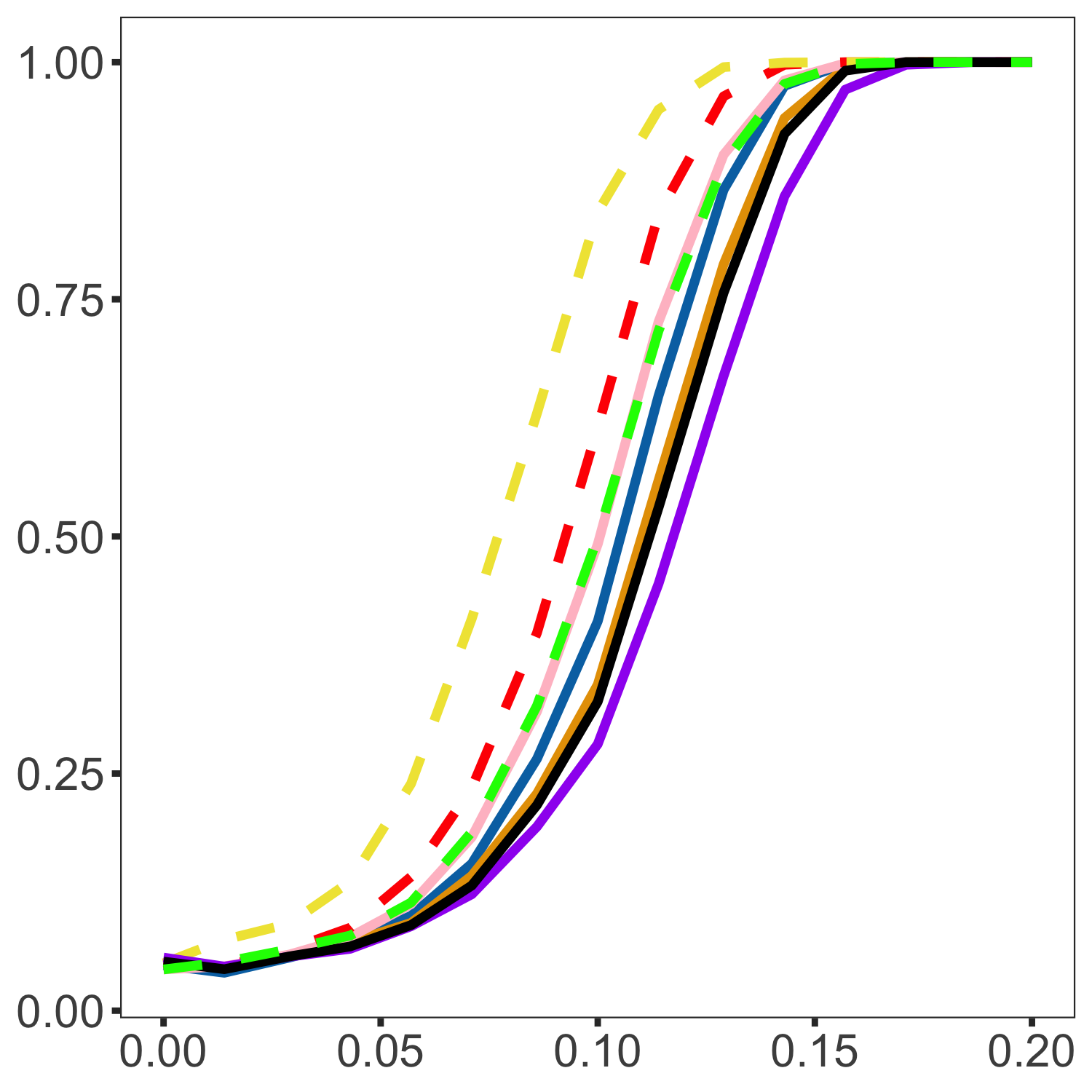}
    \end{subfigure}%
    \begin{subfigure}[t]{0.23\textwidth}
        \centering
        \includegraphics[width=\linewidth, height=0.7\linewidth]{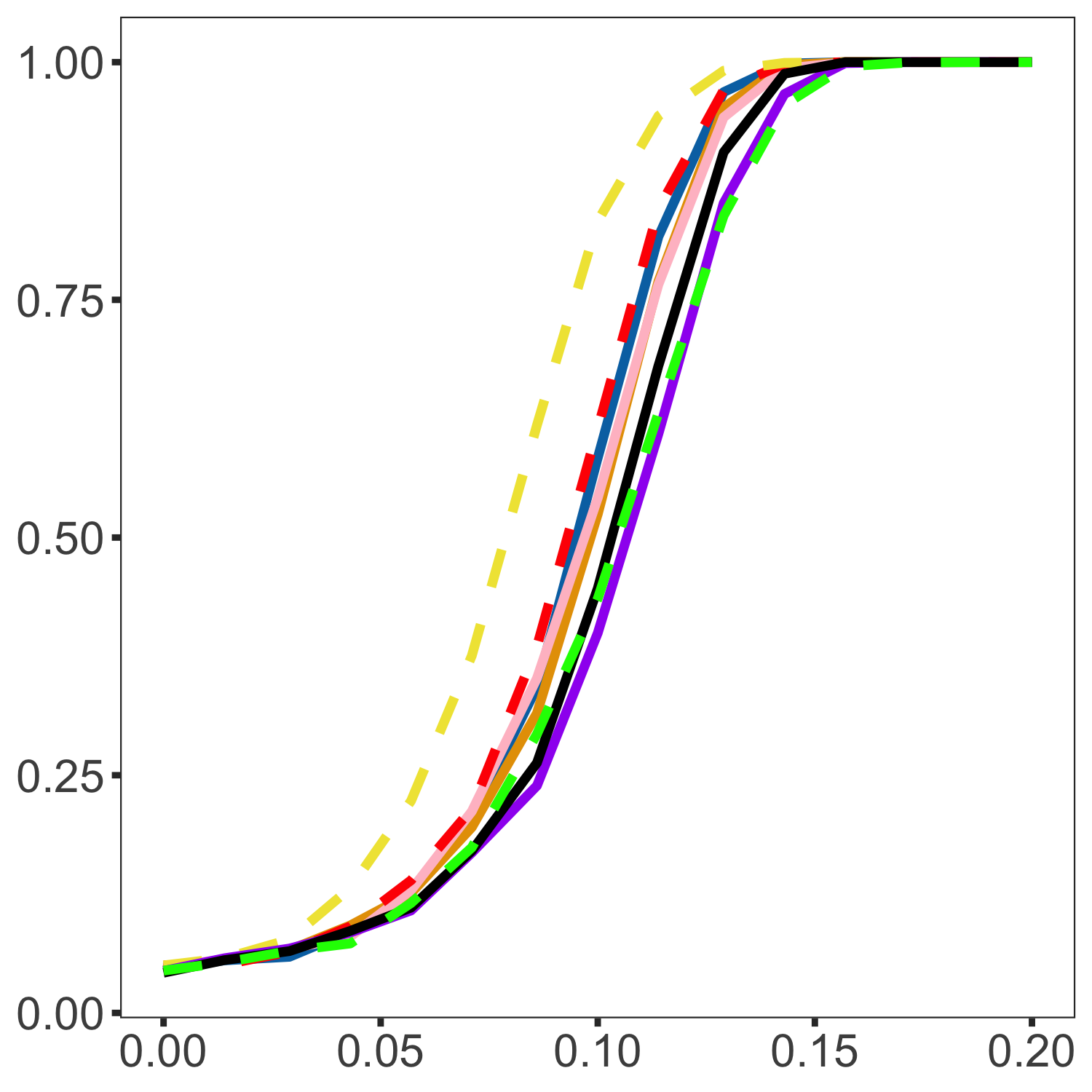}
    \end{subfigure}
     \begin{subfigure}[t]{0.23\textwidth}
        \centering
        \includegraphics[width=\linewidth, height=0.7\linewidth]{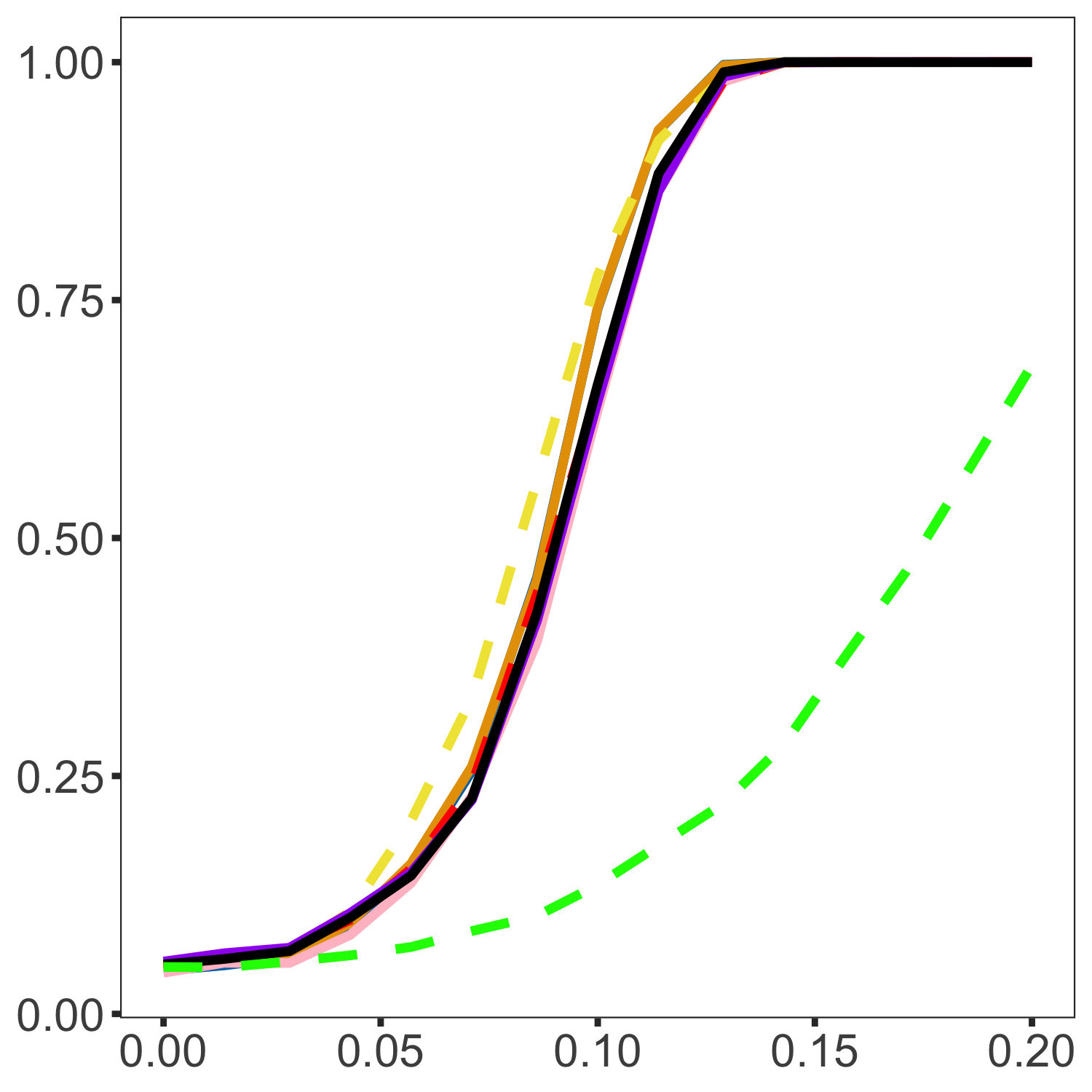}
    \end{subfigure}
    \begin{subfigure}[t]{0.23\textwidth}
        \centering
        \includegraphics[width=\linewidth, height=0.7\linewidth]{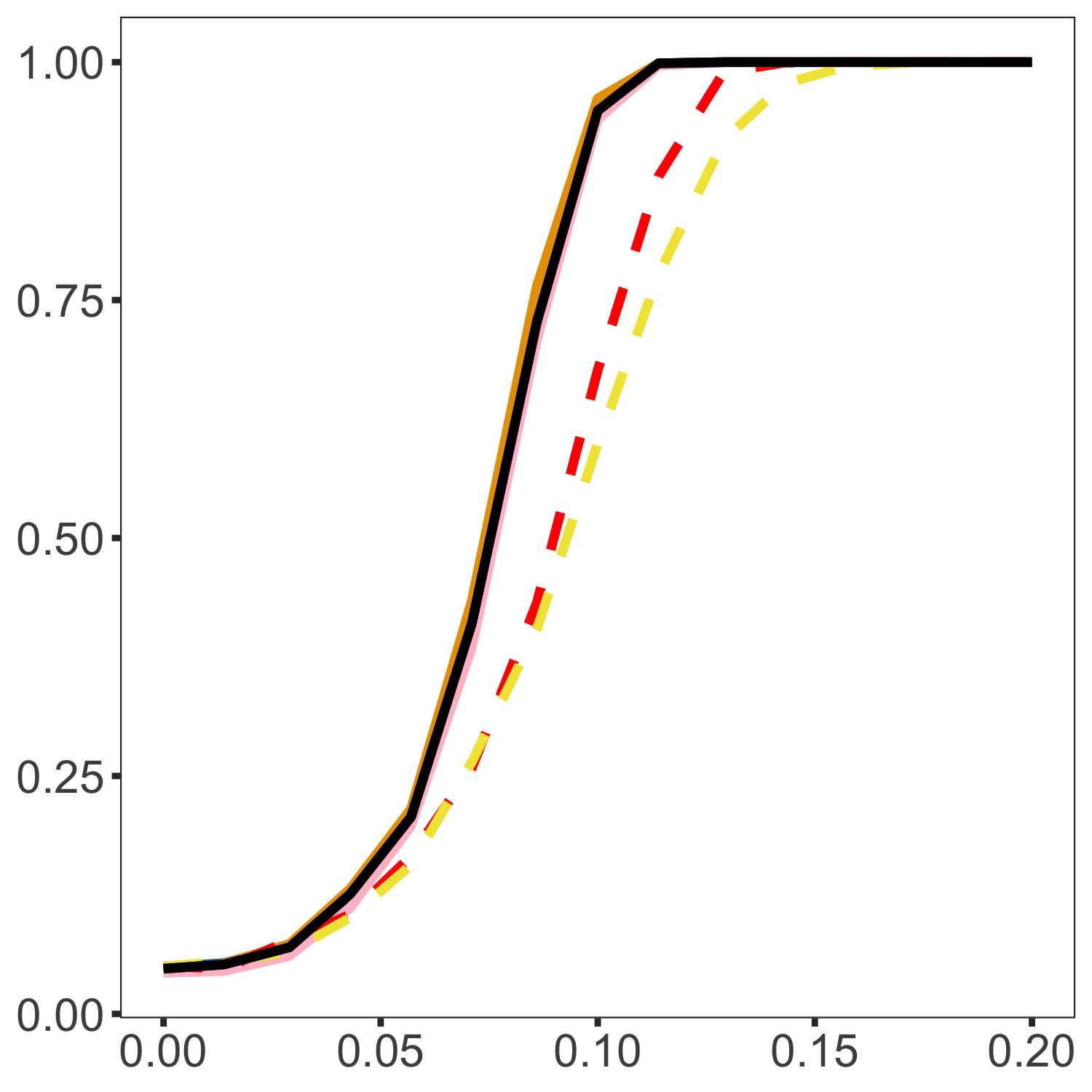}
    \end{subfigure}
    \vfill
    \begin{subfigure}[t]{0.23\textwidth}
        \centering
         \includegraphics[width=\linewidth, height=0.7\linewidth]{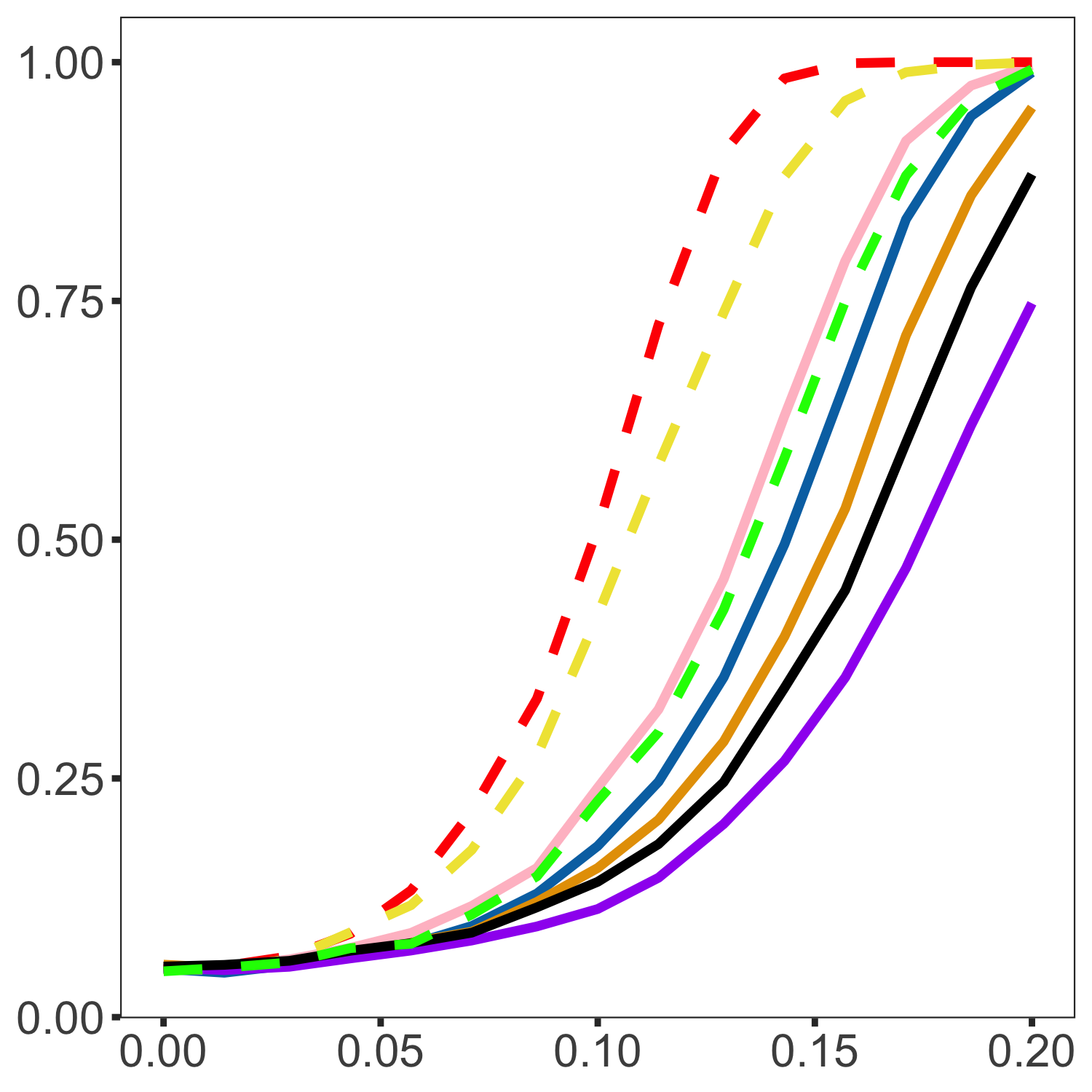}
    \end{subfigure}%
    \begin{subfigure}[t]{0.23\textwidth}
        \centering
        \includegraphics[width=\linewidth, height=0.7\linewidth]{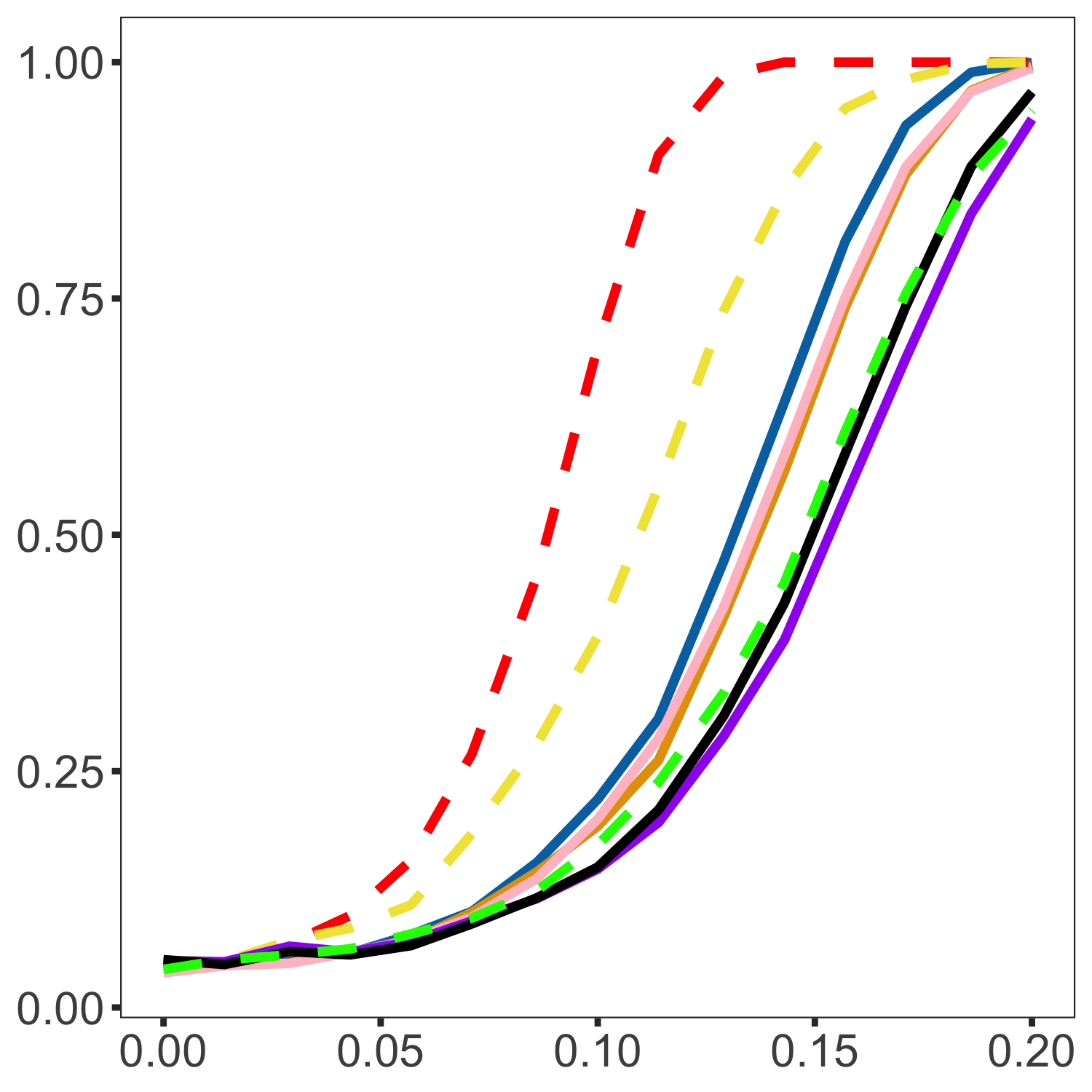}
    \end{subfigure}
     \begin{subfigure}[t]{0.23\textwidth}
        \centering
        \includegraphics[width=\linewidth, height=0.7\linewidth]{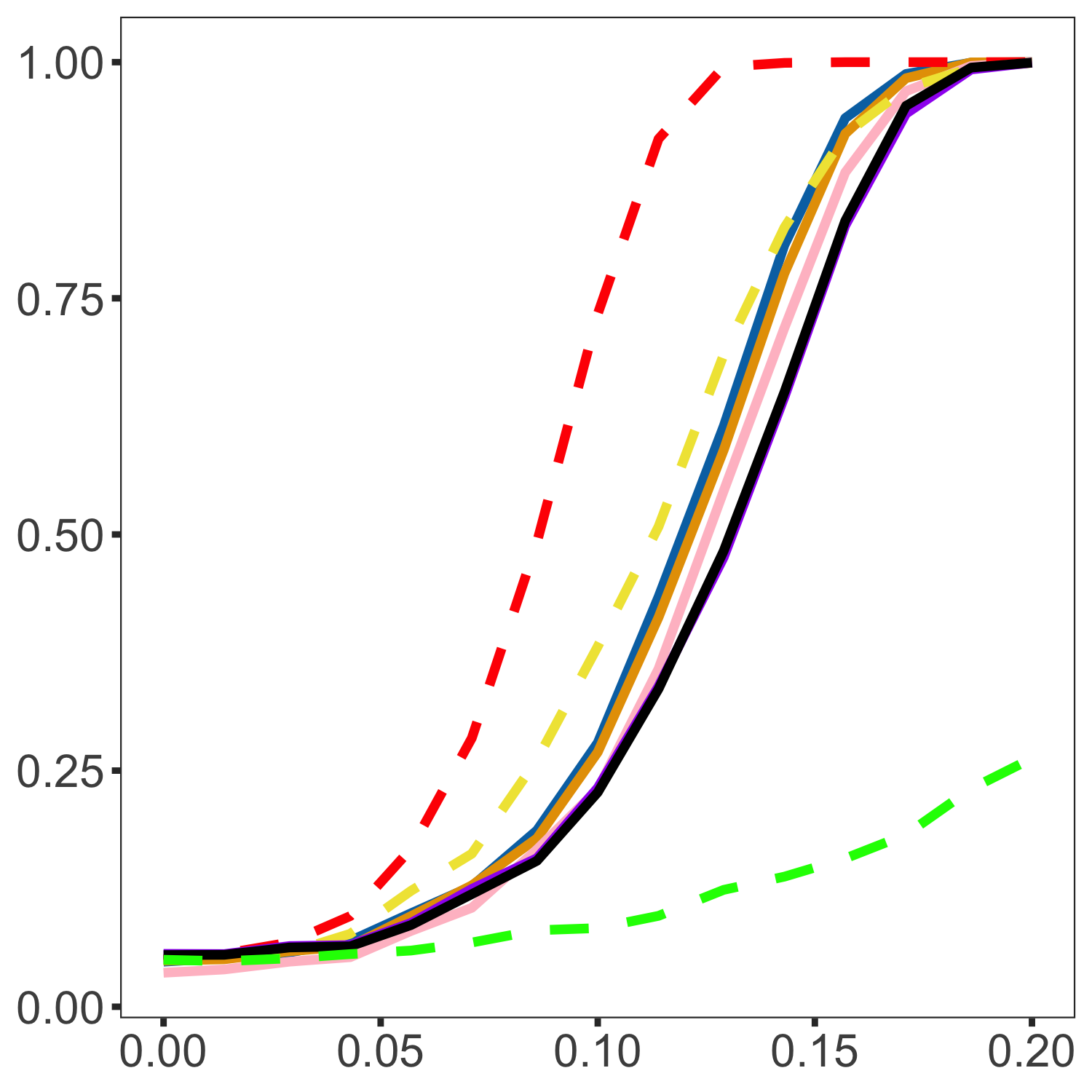}
    \end{subfigure}
    \begin{subfigure}[t]{0.23\textwidth}
        \centering
        \includegraphics[width=\linewidth, height=0.7\linewidth]{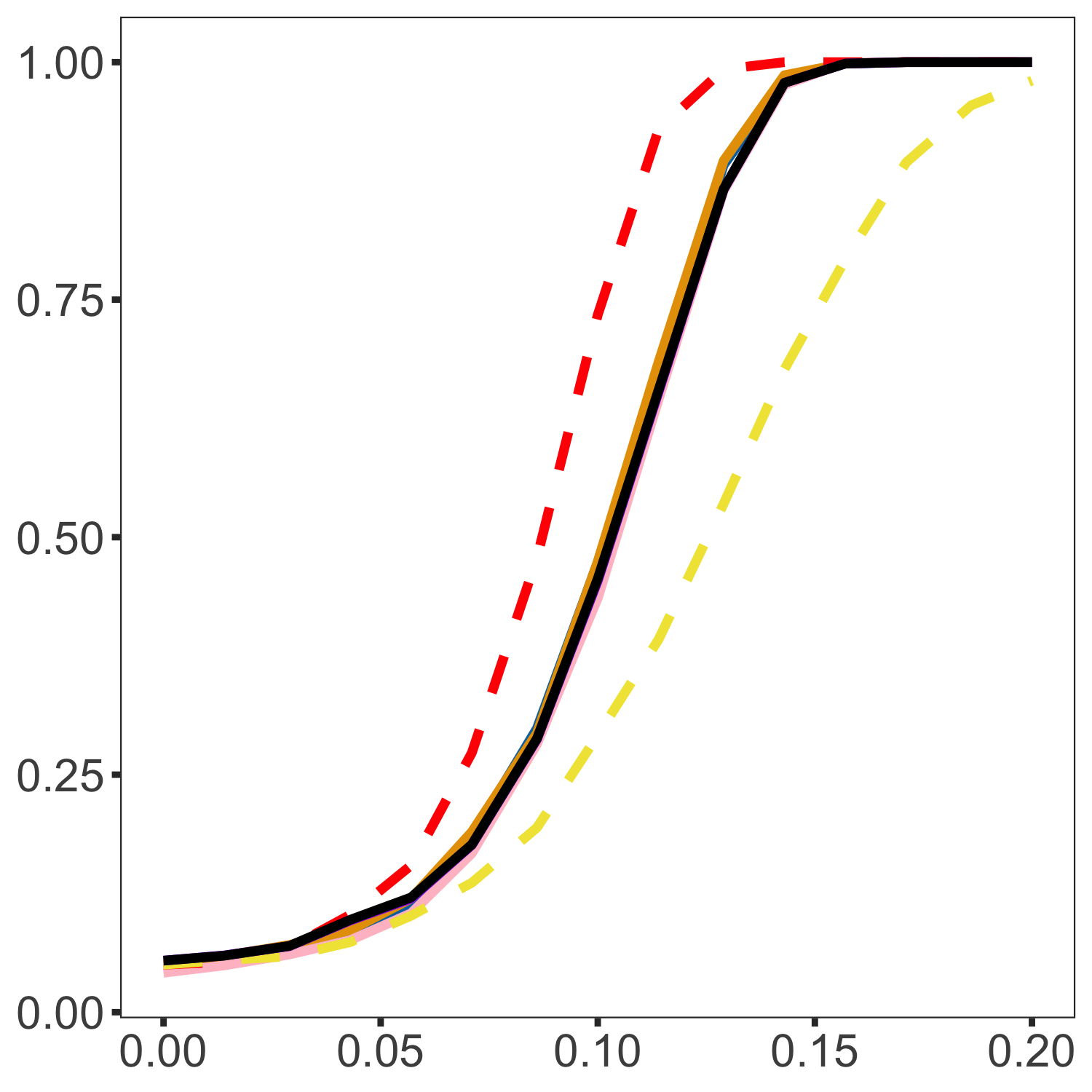}
    \end{subfigure}
   \caption{Size-adjusted empirical power under high-rank alternatives when \(\Sigma\) is AR-ACF. Columns (left to right) correspond to \(\hat{\gamma}_2 = 0.3, 0.5, 0.9, 2\); rows (top to bottom) correspond to \(n_1 = 50, 100, 250\). Solid curves: blue (\(\lambda = 0.5\)), orange (\(\lambda = 1\)), black (\(\lambda=\hat{\lambda}_{I_p}\)), purple (\(\lambda=\hat{\lambda}_{\Sigma_p}\)), and pink (\(\lambda=\hat{\lambda}_*\)). Dashed curves: red (Proj-LRT), yellow (Ridge-LRT), and green (\cite{han2016tracy}, \(\lambda=0\)), the latter available only when \(p<n_1+n_2\).}
    \label{fig:FR_emp_power_Sigma3}
\end{figure}

\begin{figure}[htbp]
    \centering
    \begin{subfigure}[t]{0.23\textwidth}
        \centering
         \includegraphics[width=\linewidth, height=0.7\linewidth]{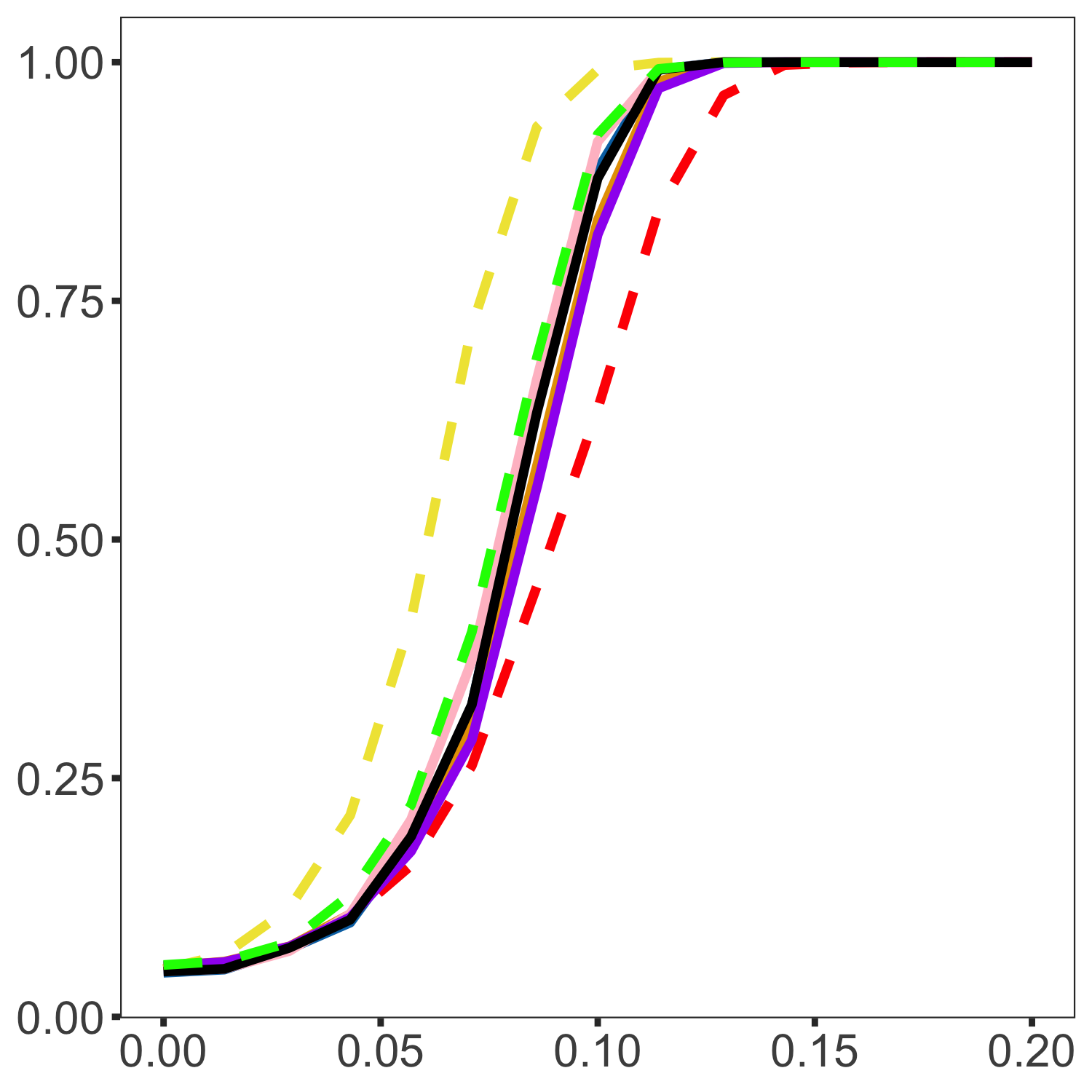}
    \end{subfigure}%
    \begin{subfigure}[t]{0.23\textwidth}
        \centering
        \includegraphics[width=\linewidth, height=0.7\linewidth]{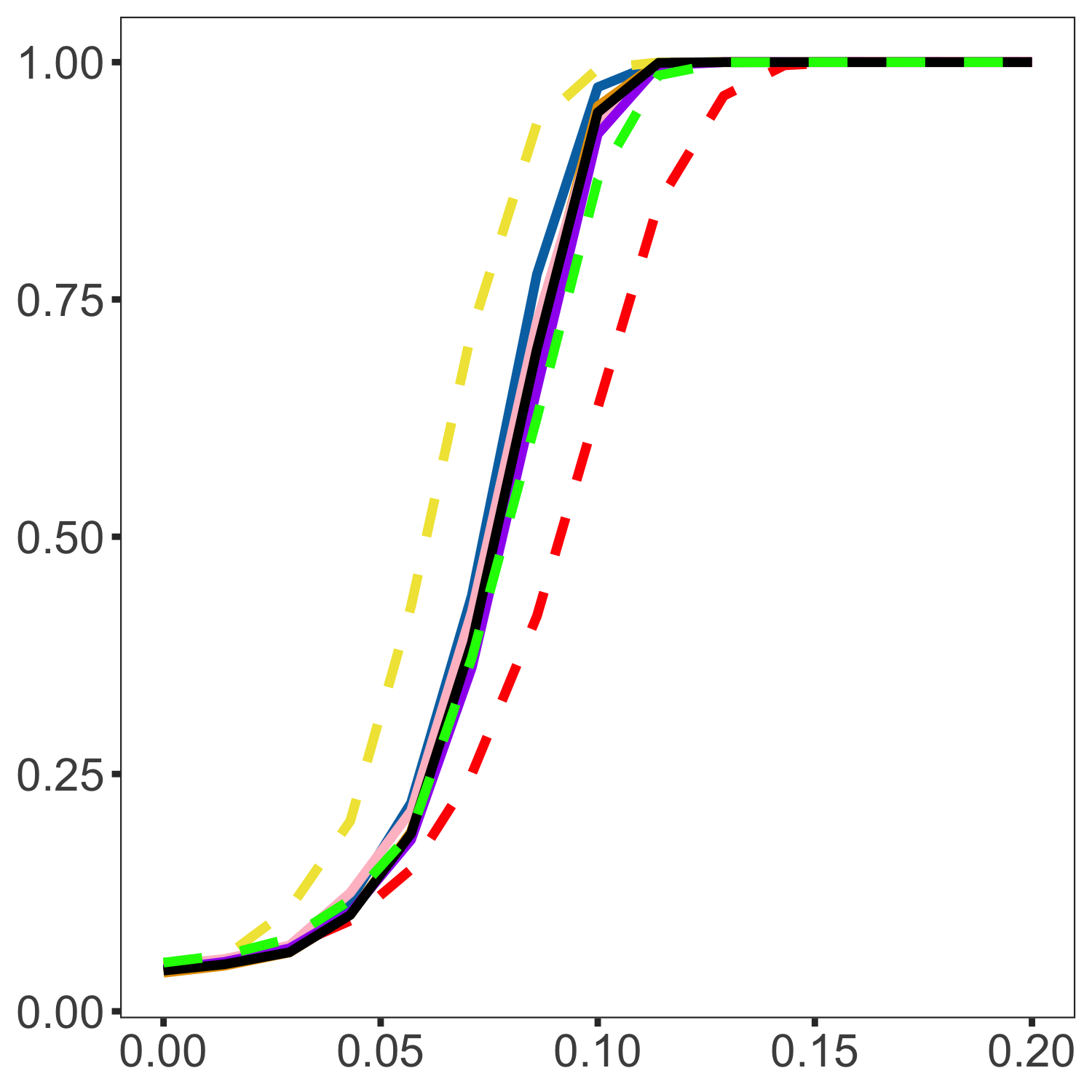}
    \end{subfigure}
     \begin{subfigure}[t]{0.23\textwidth}
        \centering
        \includegraphics[width=\linewidth, height=0.7\linewidth]{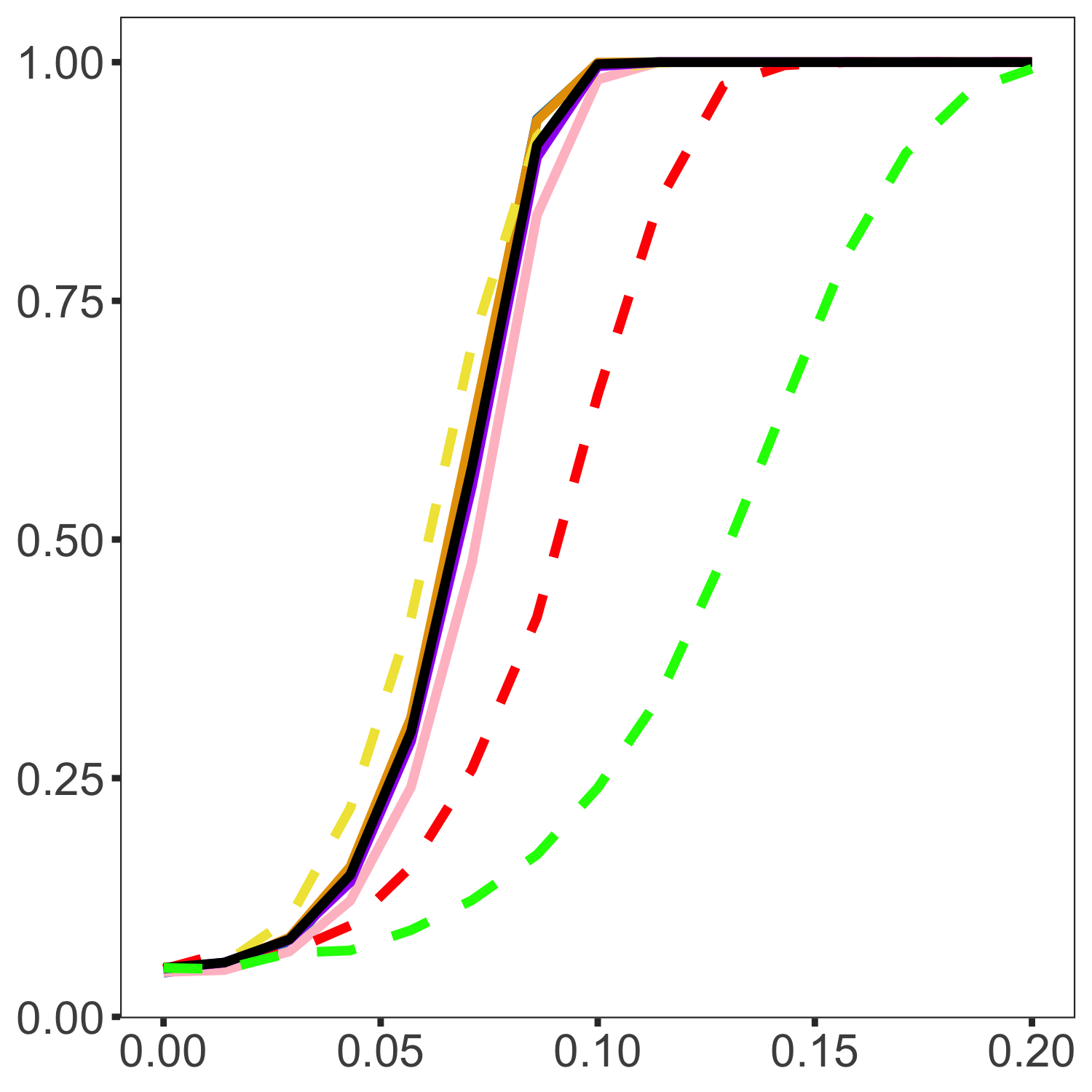}
    \end{subfigure}
    \begin{subfigure}[t]{0.23\textwidth}
        \centering
        \includegraphics[width=\linewidth, height=0.7\linewidth]{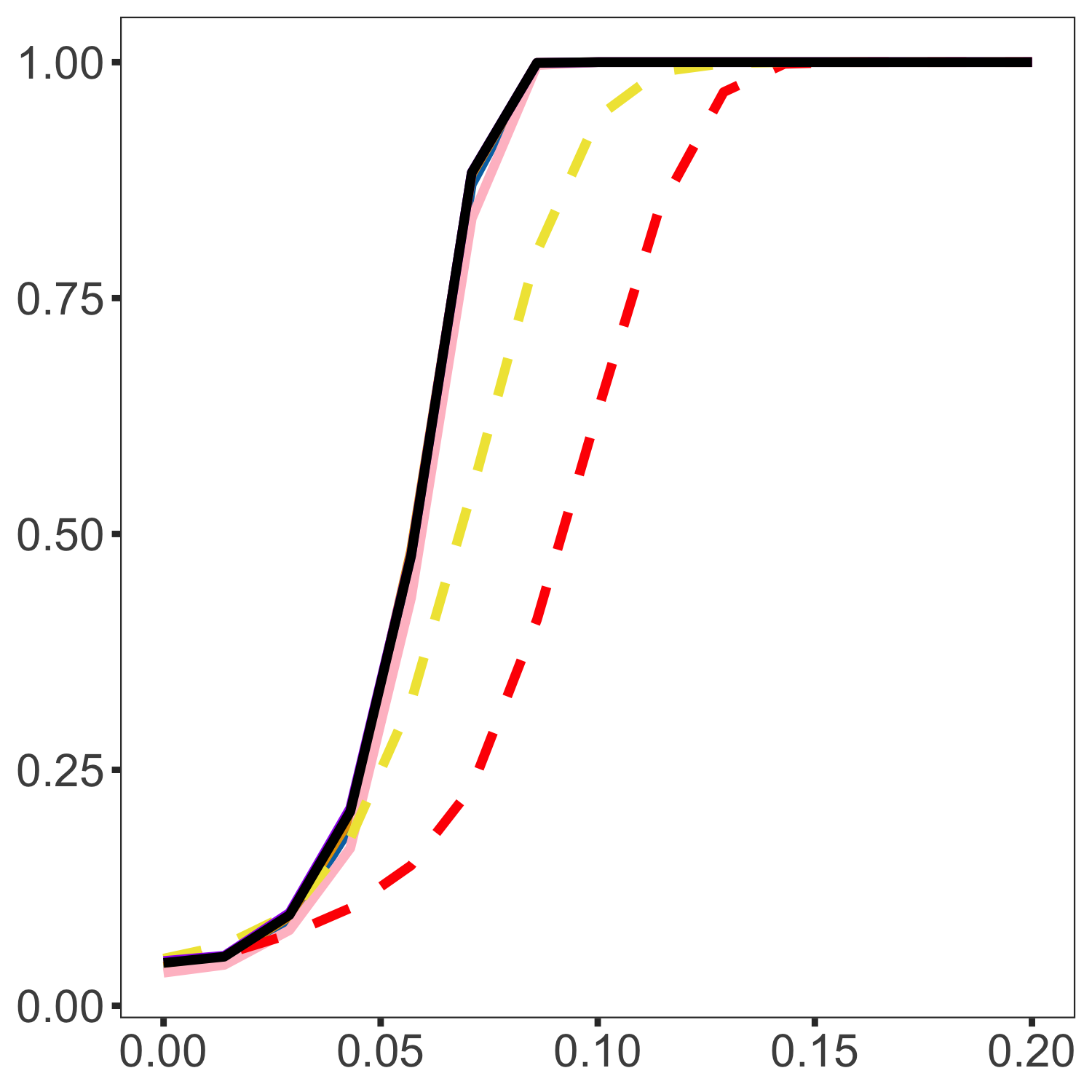}
    \end{subfigure}
    \vfill
    \begin{subfigure}[t]{0.23\textwidth}
        \centering
         \includegraphics[width=\linewidth, height=0.7\linewidth]{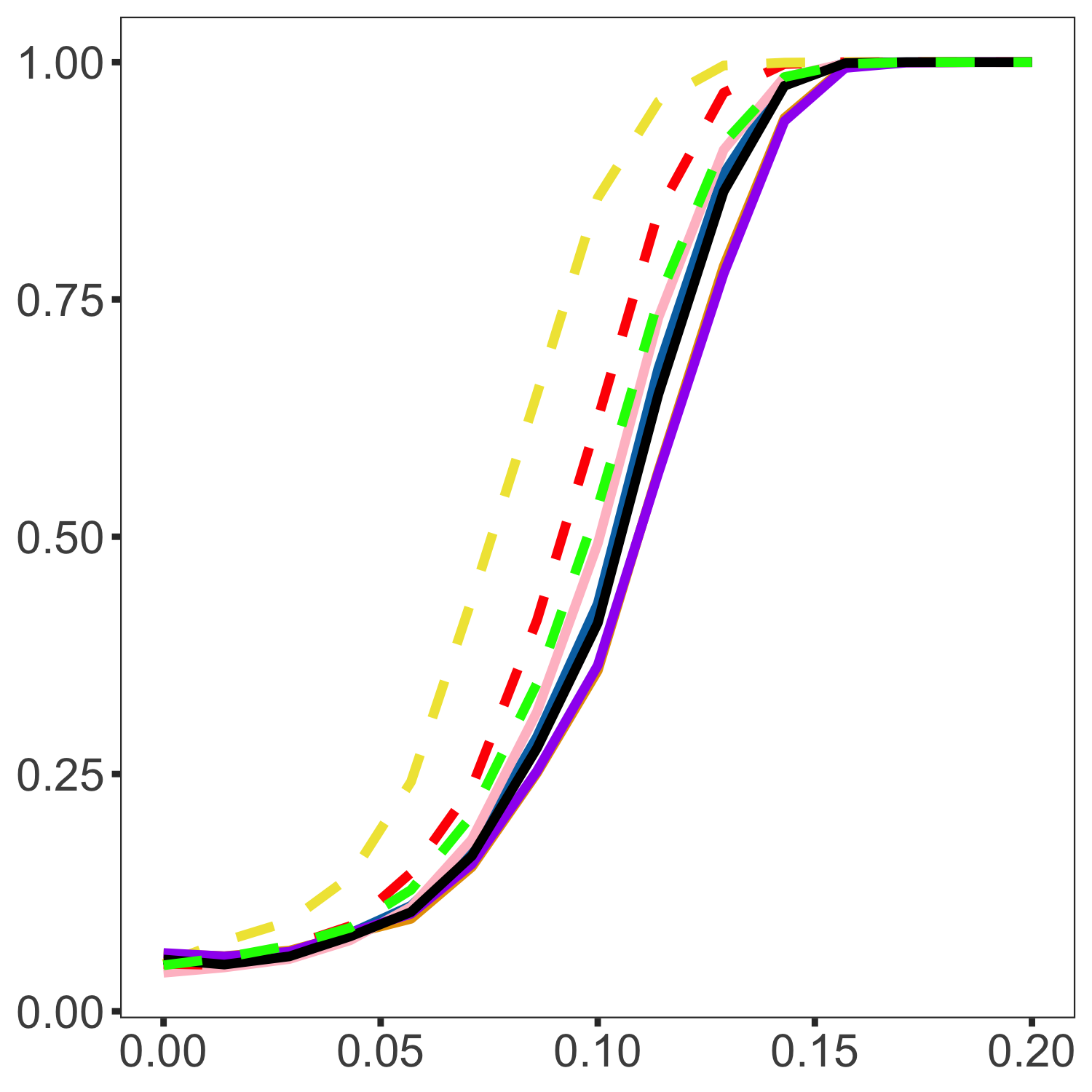}
    \end{subfigure}%
    \begin{subfigure}[t]{0.23\textwidth}
        \centering
        \includegraphics[width=\linewidth, height=0.7\linewidth]{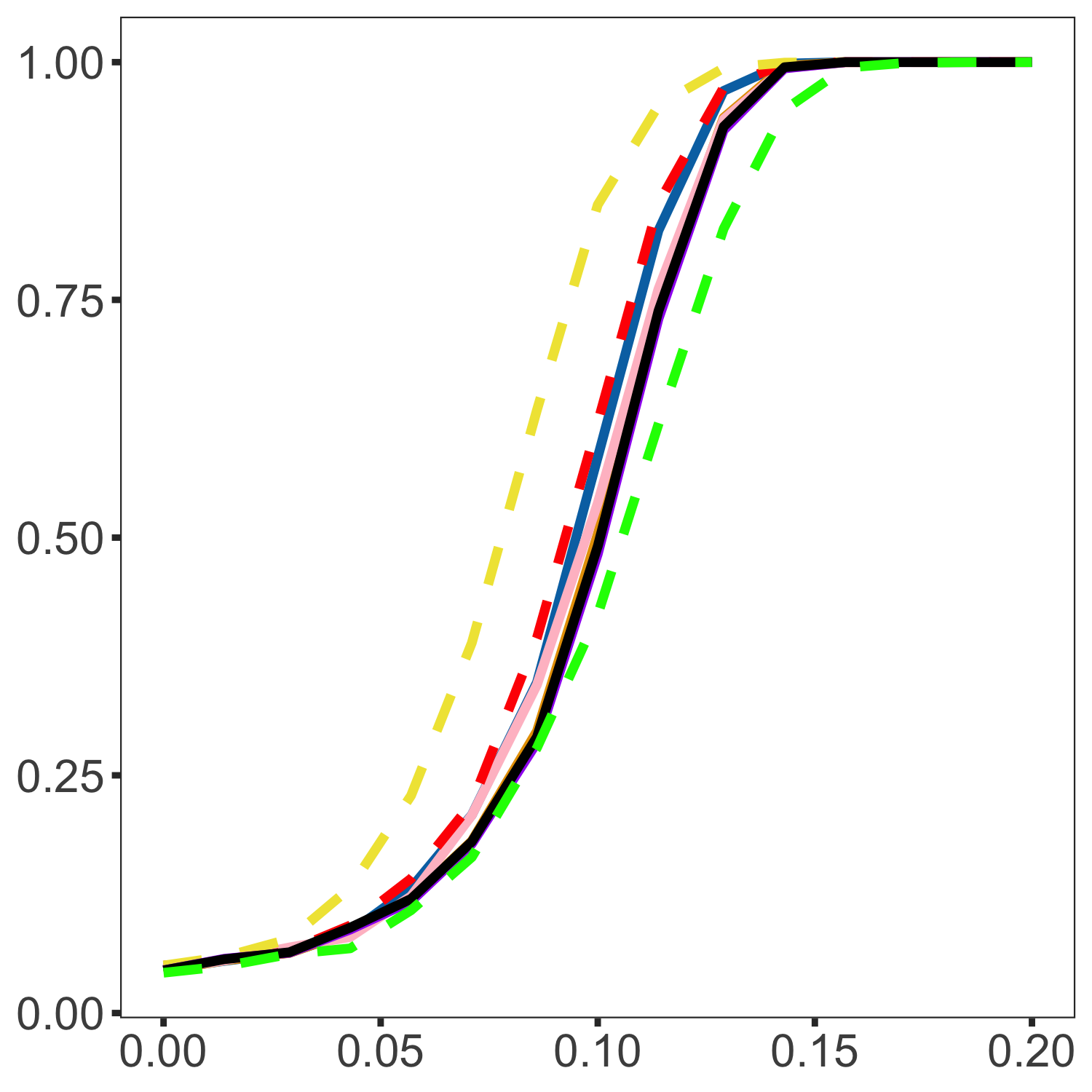}
    \end{subfigure}
     \begin{subfigure}[t]{0.23\textwidth}
        \centering
        \includegraphics[width=\linewidth, height=0.7\linewidth]{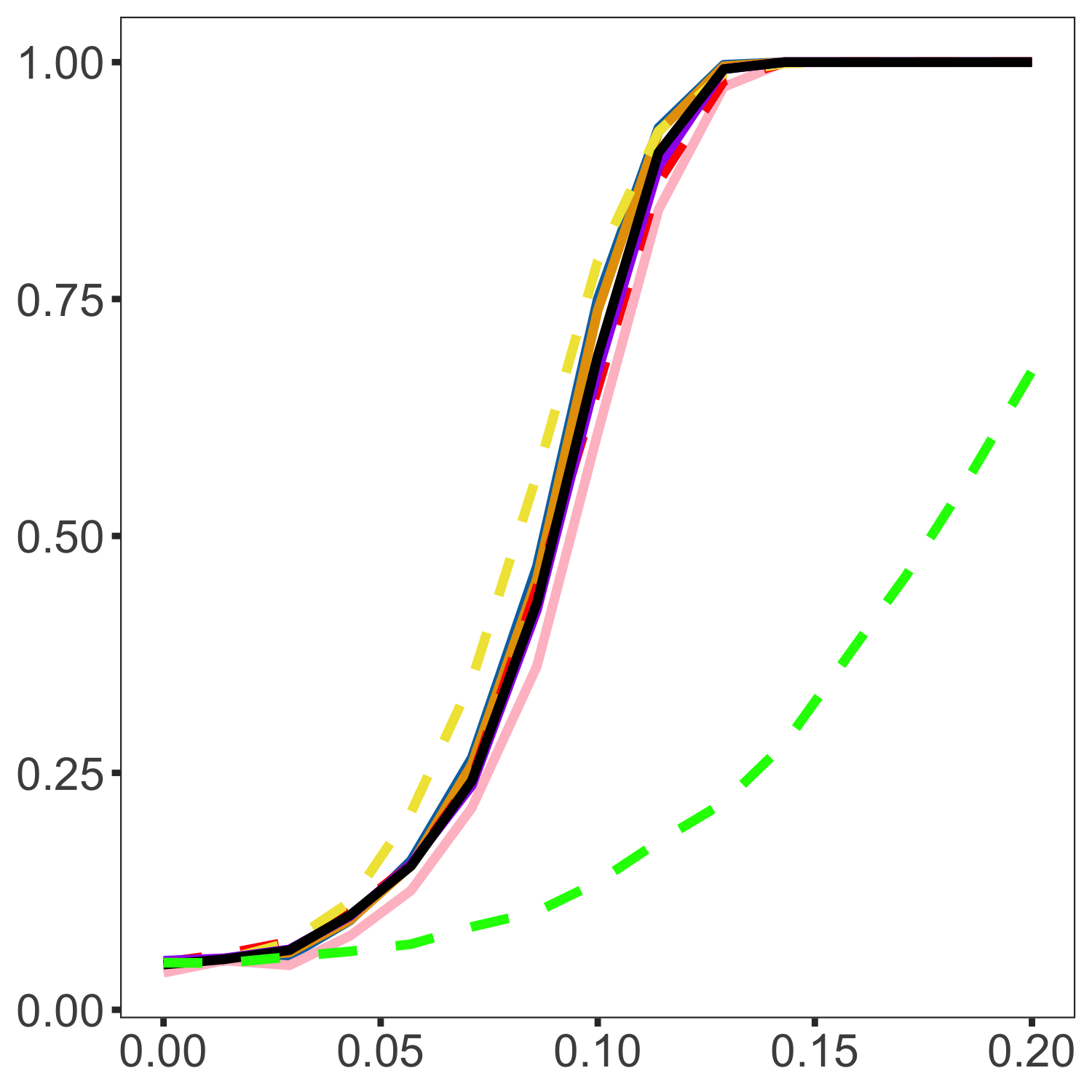}
    \end{subfigure}
    \begin{subfigure}[t]{0.23\textwidth}
        \centering
        \includegraphics[width=\linewidth, height=0.7\linewidth]{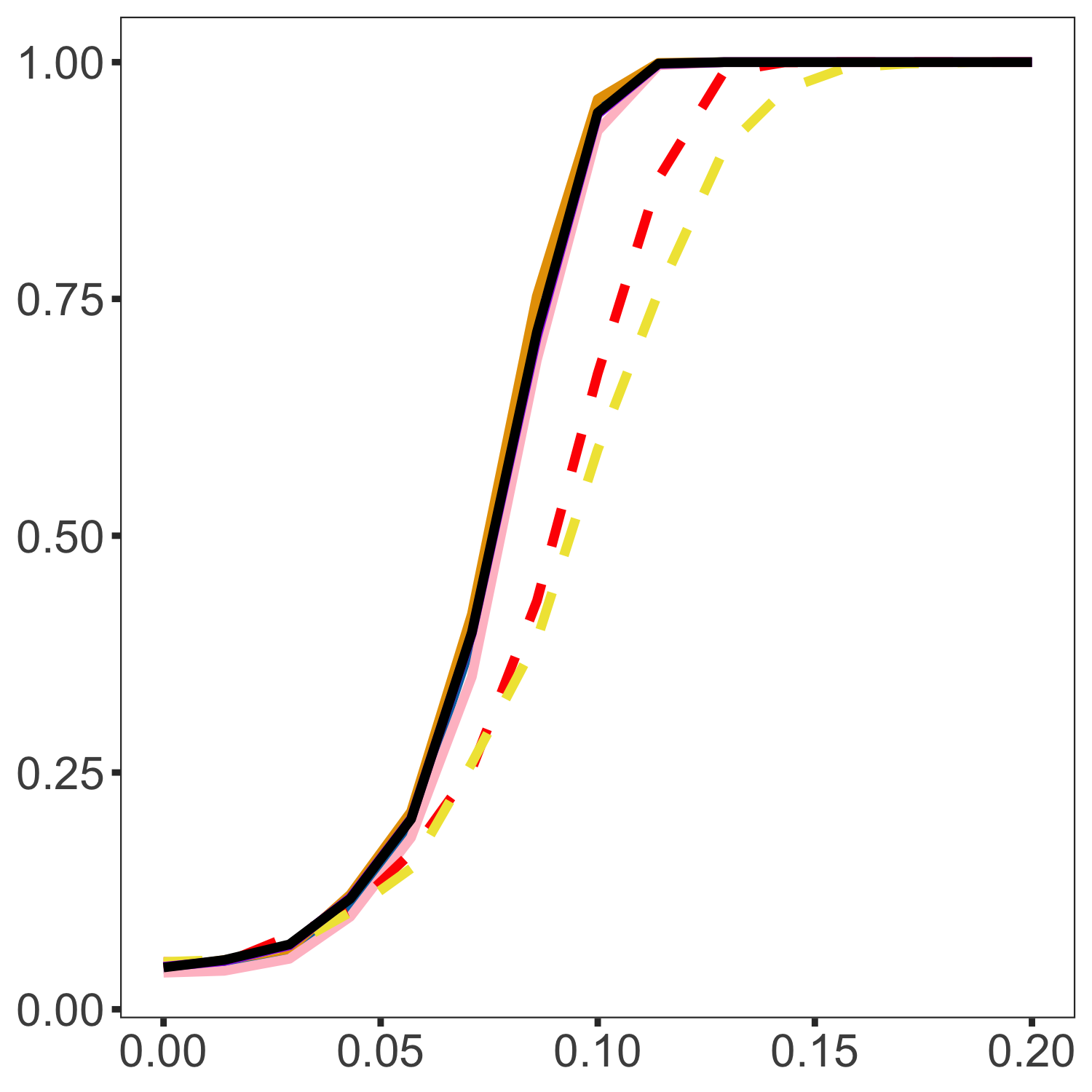}
    \end{subfigure}
    \vfill
    \begin{subfigure}[t]{0.23\textwidth}
        \centering
         \includegraphics[width=\linewidth, height=0.7\linewidth]{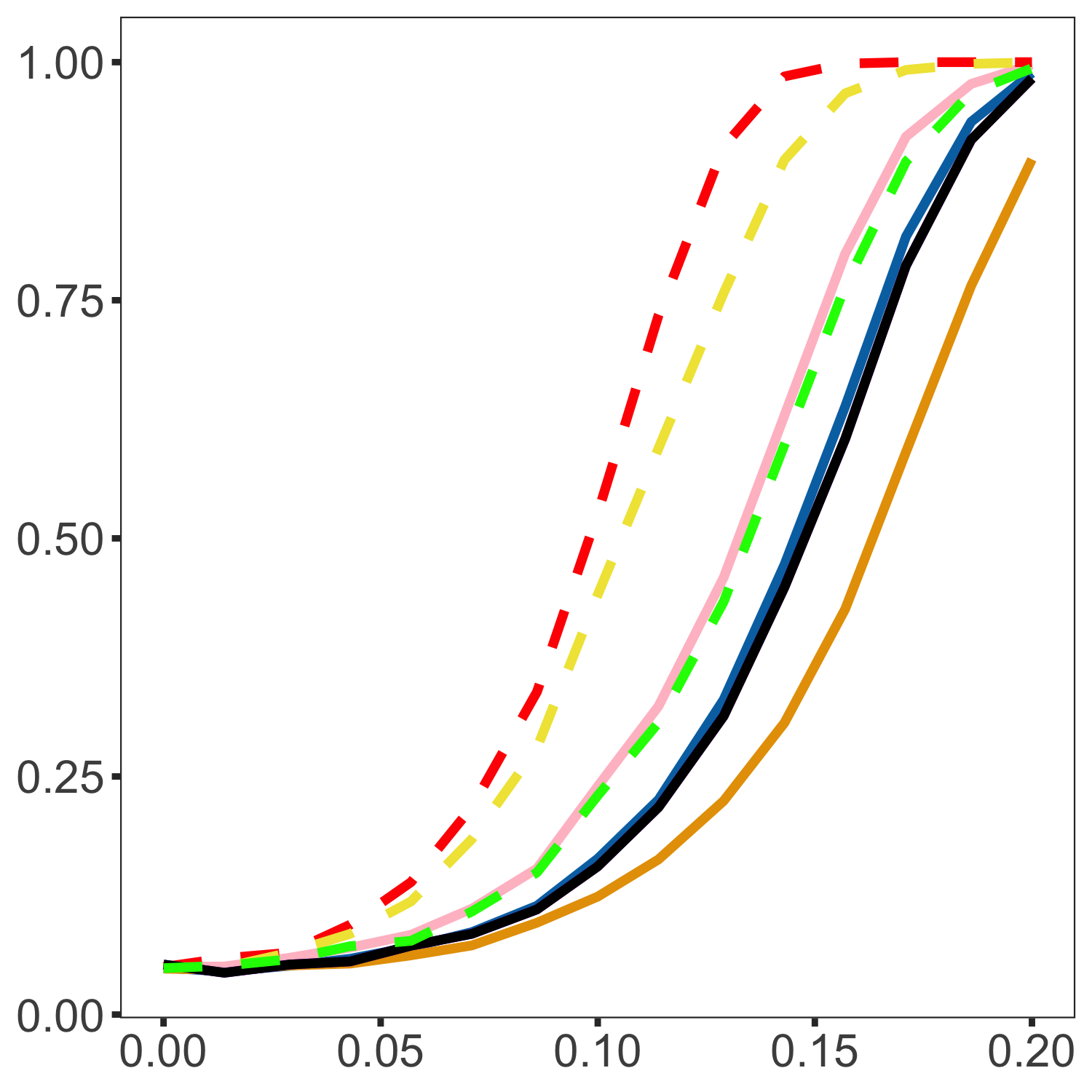}
    \end{subfigure}%
    \begin{subfigure}[t]{0.23\textwidth}
        \centering
        \includegraphics[width=\linewidth, height=0.7\linewidth]{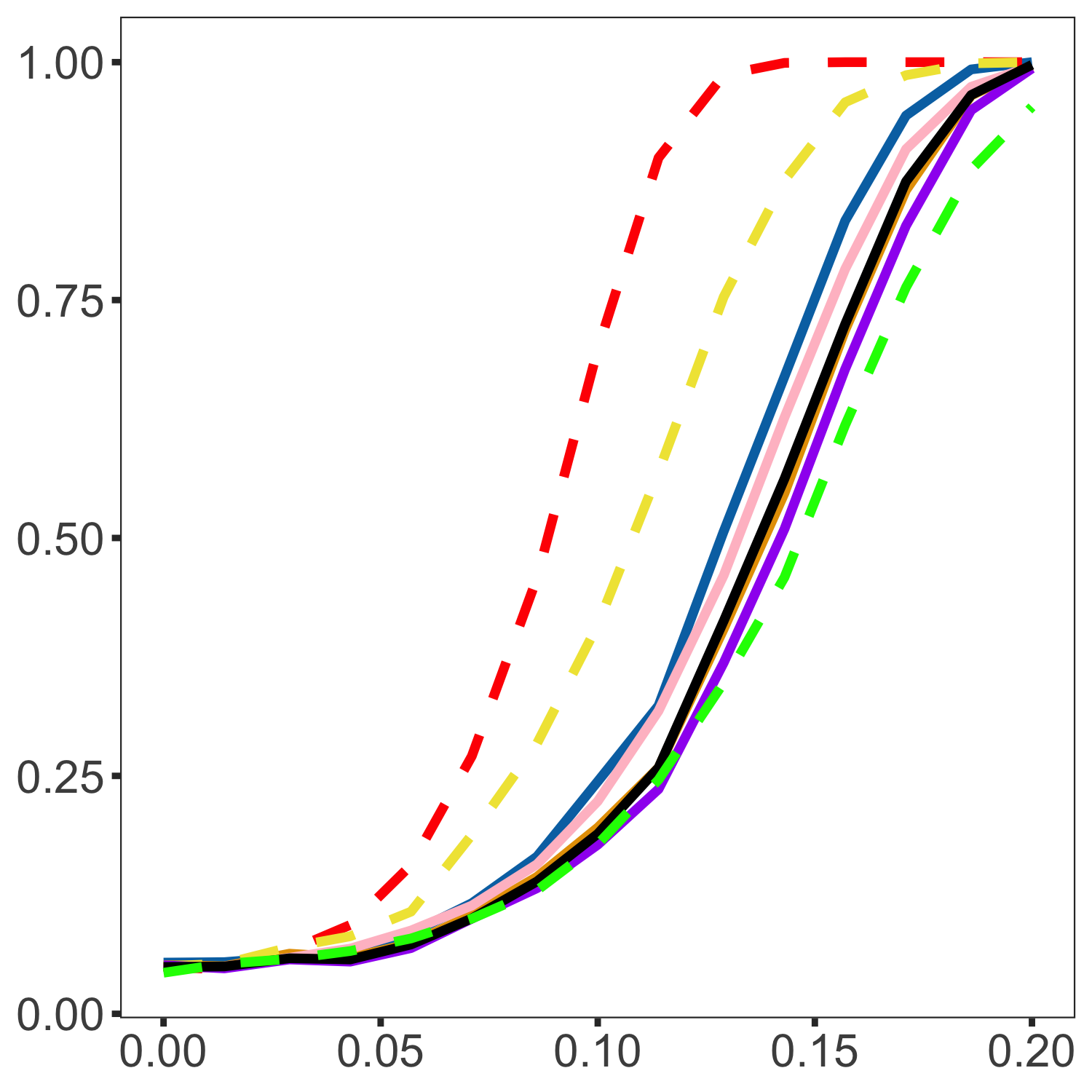}
    \end{subfigure}
     \begin{subfigure}[t]{0.23\textwidth}
        \centering
        \includegraphics[width=\linewidth, height=0.7\linewidth]{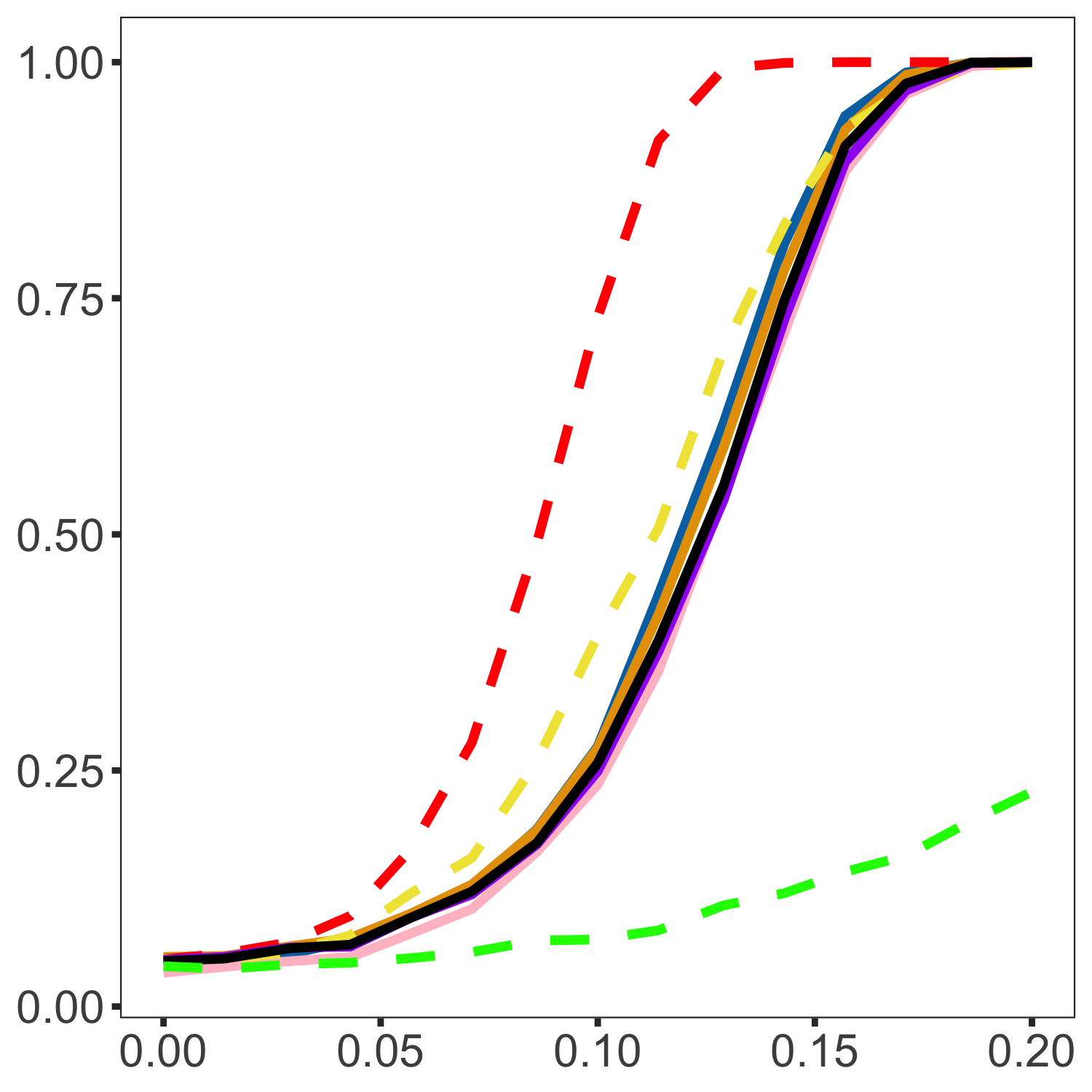}
    \end{subfigure}
    \begin{subfigure}[t]{0.23\textwidth}
        \centering
        \includegraphics[width=\linewidth, height=0.7\linewidth]{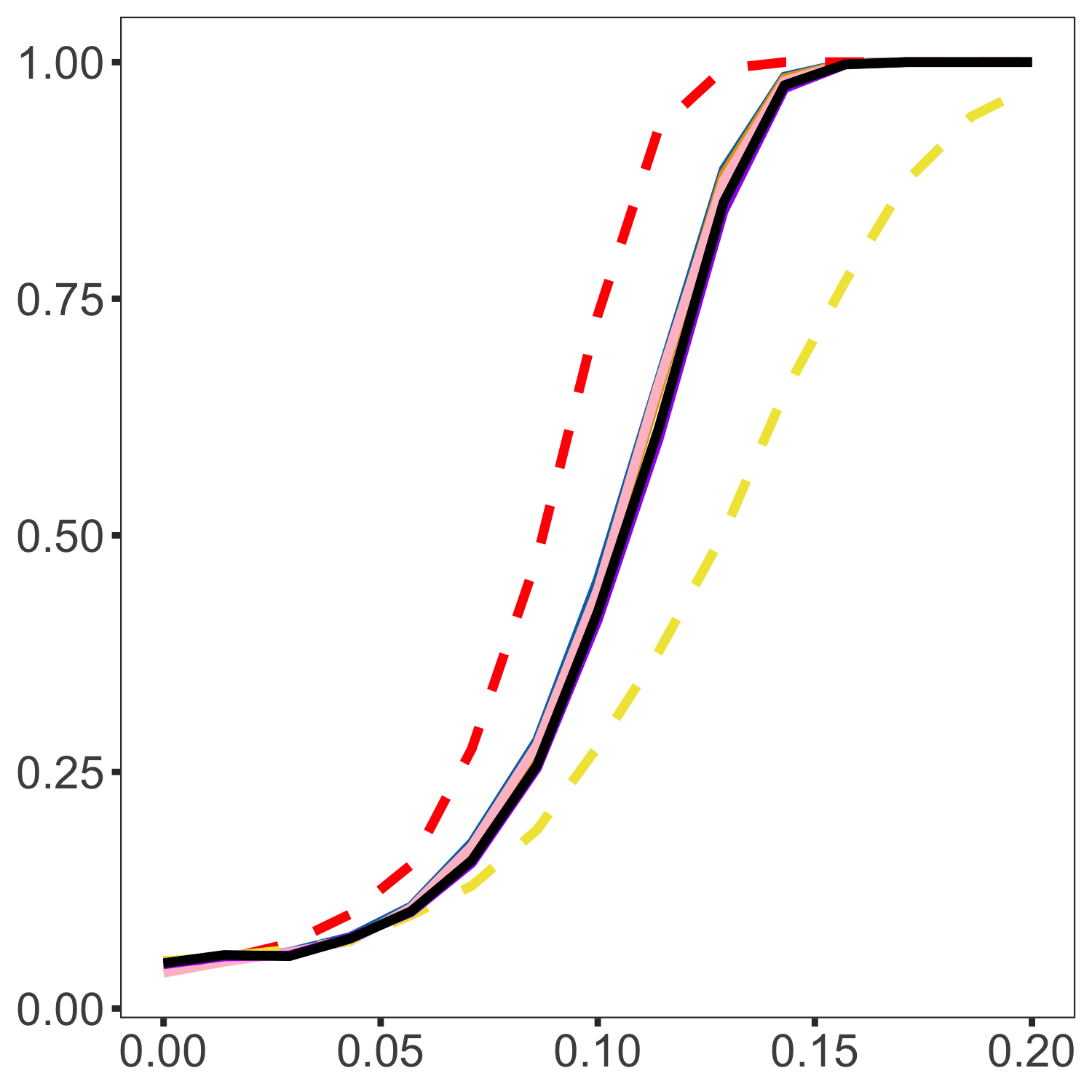}
    \end{subfigure}
   \caption{Size-adjusted empirical power under high-rank alternatives when \(\Sigma\) is Factor. Columns (left to right) correspond to \(\hat{\gamma}_2 = 0.3, 0.5, 0.9, 2\); rows (top to bottom) correspond to \(n_1 = 50, 100, 250\). Solid curves: blue (\(\lambda = 0.5\)), orange (\(\lambda = 1\)), black (\(\lambda=\hat{\lambda}_{I_p}\)), purple (\(\lambda=\hat{\lambda}_{\Sigma_p}\)), and pink (\(\lambda=\hat{\lambda}_*\)). Dashed curves: red (Proj-LRT), yellow (Ridge-LRT), and green (\cite{han2016tracy}, \(\lambda=0\)), the latter available only when \(p<n_1+n_2\).}
    \label{fig:FR_emp_power_Sigma4}
\end{figure}

\newpage
\clearpage 

\section{Basic Definitions and Additional Properties of the M-P Equations}\label{sec:basic_definitions}

In the remaining sections, we present the proofs of all technical results stated in the manuscript. We begin in this section by introducing the necessary notations and definitions. We then provide additional results concerning the properties of the Mar\v{c}enko–Pastur equation \eqref{eq:MP_W2} and equation \eqref{eq:MP_G}. Throughout this section, we assume Conditions~\ref{enum:high_dimensional_regime}–\ref{enum:regular_edge}.

\subsection{Basic definitions}\label{subsec:basic_definitions}
Throughout, we use $a$, $a'$, $c$, $c'$, $C$, and $C'$ to denote generic constants whose values may change from line to line. Unless stated otherwise, these constants do not depend on $p$. 
We use $E = E(z)$ and $\eta = \eta(z)$ to denote the real and imaginary parts of a complex number $z$, respectively. The argument $z$ is frequently omitted from expressions when no confusion may arise.

\begin{definition}[Matrix norms]\label{def:matrix_norm}
Let $A= (a_{ij})$ be a complex matrix. We define the following norms 
\[ \| A \| = \max_{\|x\| =1} \|x^*Ax\| , \quad \|A\|_\infty  = \max_{i,j}|a_{ij}|, \quad \|A\|_F = \sqrt{\tr (AA^*)},\]
where $A^*$ is the conjugate transpose of a matrix $A$. If $A$ is a vector, we write $|A| = \|A\|$ for simplicity. 
\end{definition}

The following notion of a high-probability bound has been used in a number of works on random matrix theory. It provides a simple way of systematizing and making precise statements of the form ``$\xi$ is bounded with high probability by $\zeta$ up to small powers of $n$''. 
\begin{definition}[Stochastic domination]\label{def:stochastic_domination}~
\begin{itemize}
\item[(i)] Consider two families of nonnegative random variables
$$
\xi=\left(\xi^{(n)}(u): m \in \mathbb{N}, u \in U^{(n)}\right), \quad \zeta=\left(\zeta^{(n)}(u): n \in \mathbb{N}, u \in U^{(n)}\right),
$$
where $U^{(n)}$ is a possibly $n$-dependent parameter set. We say that $\xi$ is stochastically dominated by $\zeta$, uniformly in $u$, if for all (small) $\varepsilon>0$ and (large) $D>0$ we have
$$
\sup _{u \in U^{(n)}} \mathbb{P}\left[\xi^{(n)}(u)>n^{\varepsilon} \zeta^{(n)}(u)\right] \leqslant n^{-D}
$$
for large enough $n \geqslant n_0(\varepsilon, D)$. Throughout this paper the stochastic domination will always be uniform in all parameters (such as matrix indices, deterministic vectors, and spectral parameters $z$) that are not explicitly fixed. 
If $\xi$ is stochastically dominated by $\zeta$, uniformly in $u$, we use the notation $\xi \prec \zeta$. Moreover, if for some complex family $\xi$ we have $|\xi| \prec \zeta$ we also write $\xi=O_{\prec}(\zeta)$.
\item[(ii)] We extend the definition of $O_{\prec}(\cdot)$ to matrices in the weak operator sense as follows. Let $A$ be a family of complex square random matrices and $\zeta$ a family of nonnegative random variables. Then we use $A=O_{\prec}(\zeta)$ to mean $|\langle\mathbf{v}, A \mathbf{w}\rangle| \prec \zeta|\mathbf{v}||\mathbf{w}|$ uniformly for all deterministic vectors $\mathbf{v}$ and $\mathbf{w}$.
\item[(iii)] If there exists a positive constant $C$ such that $\xi \leq C\zeta$, then we write $\xi \lesssim \zeta$. If further there exists a positive constant $C'$ such that $\zeta \leq C'\xi$, then we write $\xi \asymp \zeta$. 
\end{itemize}
\end{definition}

\begin{definition}\label{def:high_probability}
We say that an event $\Lambda$ holds with high probability if for any large positive constant $D$, there exists $n_0(D)$ such that 
\[\mP(\Lambda^\complement)\leq n^{-D}, ~\mbox{ for any }n\geq n_0(D).\]
\end{definition}

\subsection{Additional properties of the Mar\v{c}enko-Pastur Equation}\label{subsec:additional_properties_of_the_mar_v_c}
We present additional properties of the solution to the Mar\v{c}enko-Pastur Equation \eqref{eq:MP_W2} evaluated at the probability measure $F$ being the distribution of $1/X$ when $X\sim \calG_{p\lambda}$ and $\gamma = \hat{\gamma}_1$. That is, we study the properties of $q(z) = q_\lambda(z)$ being the unique solution in $\mathbb{C}^+$ to  
\begin{equation}\label{eq:def_q_z}
 z =  - \frac{1}{q} + \hat{\gamma}_1 \int\frac{d\calG_{p\lambda}(\tau)}{ \tau + q}.
\end{equation}
Following foundational results in RMT, $q(z)$ is the Stieltjes transform of a probability measure, denoted by $\calF_{q}$. It is further known that $\calF_q$ is compactly supported in $(0, \infty)$. It is understood that $q(z)$ and $\calF_q$ depend on $p$ and $\lambda$.

Recall the definition of $\phi_{p\lambda}(\iz)$: $\mathbb{C}^+\to\mathbb{C}^+$ as the Stieltjes transform of $\calG_{p\lambda}$ (See the statement right after Eq. \eqref{eq:MP_G}). That is,
\[ \phi_{p\lambda}(\iz) = \int \frac{d\calG_{p\lambda}(\tau)}{\tau - \iz}, \quad \iz \in \mathbb{C}^+. \]
For convenience, we extend the definition of $\phi_{p\lambda}(\iz)$ to $\iz \in \mathbb{C}^- = \{E + \i\eta: \eta <0 \}$ by setting $\phi_{p\lambda}(\iz) = \overline{\phi_{p\lambda}(\overline{\iz})}$ if $\eta <0$. By doing so, we can concisely write \eqref{eq:def_q_z} as 
\[ z = -\frac{1}{q(z)} + \hat{\gamma}_1 \phi_{p\lambda}(-q(z)). \]

{We assume that Conditions \ref{enum:high_dimensional_regime}--\ref{enum:regular_edge} hold. In the remainder of this section, the parameter $\lambda$ is fixed. We suppress the explicit dependence of many quantities on $\lambda$ and $p$ for notational simplicity. Unless stated otherwise, quantities introduced in this section may depend on $p$ and $\lambda$, even when this dependence is not made explicit. 

The properties of the solution to \eqref{eq:MP_W2} have been studied in \citet{knowles2017anisotropic}. The application of the main theorems in \citet{knowles2017anisotropic} requires verification of their key structural assumptions. In particular, we need to verify that Definition~2.7 of \citet{knowles2017anisotropic} is satisfied at the rightmost support edge of the distribution $\calF_q$.

Recall that $s(x) = s_{p\lambda}(x)$ denotes the extension of $\phi_{p\lambda}(z)$ to the real line. Define the function
\[
f(q) = -\frac{1}{q} + \hat{\gamma}_1 s(-q), \qquad q \in (-\rho, 0).
\]
In the technical proof, we shall also consider  $f(q)$ evaluated at complex-valued $q$. 

Recall also the definitions of $\beta = \beta_{p\lambda}$ and $\Theta_1 = \Theta_{1p}(\lambda)$ given in Theorem~\ref{thm:main}. Then, $-\beta$ is the unique critical point of $f$ in $(-\rho, 0)$, since
\[
f'(-\beta) = \frac{1}{\beta^2} - \hat{\gamma}_1 s'(\beta) = 0.
\]
Moreover,
\[
\Theta_1 = f(-\beta) = \frac{1}{\beta} + \hat{\gamma}_1 s(\beta).
\]
Following well-known results in RMT, $\Theta_1$ is the rightmost support edge of $\calF_q$; see e.g. \cite{silverstein1995analysis}.

To verify Definition 2.7 of \cite{knowles2017anisotropic} at the edge $\Theta_1$, we only need to show that there exists a small constant $c$
\[ \inf_{\tau \in \supp(\calG_{p\lambda})}|\beta - \tau| = \rho -\beta \geq c, \]
for all sufficiently large $p$. 

When Part~(a) of Definition~\ref{def:regular_edge} holds, the density of $\calG_{p\lambda}$ exhibits a square-root behavior near $\rho$. Namely, we can find constants $C_1>0$, $C_2>0$ and $\epsilon>0$ such that for all sufficiently large $p$,
\[
C_1 \sqrt{x - \rho} \leq f_{\calG}(x) \leq C_2 \sqrt{x - \rho}, 
\qquad \text{for all } x \in (\rho, \rho + \epsilon).
\]
The results are shown in Theorem 5.1 of \cite{li2024analysis}. 

On the other hand, when Part~(b) of Definition~\ref{def:regular_edge} holds, Theorem~3.2 of \cite{li2024analysis} implies that $\rho$ is an isolated point mass of $\calG_{p\lambda}$, satisfying
\[
\calG_{\lambda}(\{\rho\}) \geq F^{\Sigma_p}(\{\sigma_{1p}\}) - \frac{n_2}{p}.
\]
Note that the right hand side is bounded away from $0$ under our assumptions.  

Under both cases, as $x\uparrow \rho$,
\[s'(x) = \int\frac{d\calG_{p\lambda}(\tau)}{\tau-x} \to \infty.\]
Indeed, it is straightforward to verify that for any large constant $C$, we can find a small constant $c>0$ such that 
\[ s'(x) \geq C, \text{ if }x\in (\rho - c,~ \rho), \]
for all sufficiently large $p$. Therefore, for any fixed ${\gamma}$, we can find a small constant $c>0$ such that, 
\[ x^2 s'(x) > 1/{\gamma}, ~~\text{if  }x\in (\rho- c, ~\rho).\]
Combining these results, we conclude that when $F^{\Sigma_p}$ is regular, there exists a constant $c>0$ such that
\[
\rho - \beta \geq c
\]
for all sufficiently large $p$. Consequently, the requirement in Definition 2.7 of \cite{knowles2017anisotropic} is satisfied at the edge $\Theta_1$.

\vspace{\baselineskip}

We next present a combination of the main results in \citet{silverstein1995analysis}, \citet{bao2015}, and \citet{knowles2017anisotropic}, together with references therein. The proof is therefore omitted. 

\begin{lemma}\label{lemma:properties_q_z}
Fix any $\lambda>0$. The following results are known for $\calF_q$ and $q(z)$. 
\begin{itemize}

\item[(1)] Fix any small constant $a>0$. Then, there exists a constant $C>0$, depending on $a$, such that 
\[C^{-1} \leq |q(z)| \leq C, \quad \mbox{and}\quad  \Im q(z) \geq C^{-1} \eta, \]
for all  $z \in \mathbb{C}^+$ satisfying $a \leq |z|\leq a^{-1}$. 
\item[(2)]  There exists sufficiently small constants $c>0$ and $c'>0$ such that for all $\iz \in \{ E+\i\eta ~:~ |E-\Theta_1|\leq c,~~ 0<\eta<c^{-1}\}$ 
\[ \inf_{\tau \in S_{\calG} }  |\tau + q(z)| \geq c', \]
where $S_{\calG}$ is the support of $\calG_{p\lambda}$. It indicates that the equation \eqref{eq:def_q_z} is non-singular when $z$ is near $\Theta_1$.

\item[(3)]  There exists an open ball of $\Theta_1$ (in $\mathbb{C}$), say $\frakN_{\Theta_1}$ such that when $\iz = E+ \i\eta \in \frakN_{\Theta_1}\cap \mathbb{C}^+$, 
\[\Im q(z)  \asymp \begin{cases} \displaystyle\sqrt{\kappa + \eta} & \mbox{if }E <\Theta_1, \\  \displaystyle\frac{\eta}{\sqrt{\kappa+\eta}} & \mbox{ if } E> \Theta_1, \end{cases}\]
where $\kappa = |E-\Theta_1|$.
\end{itemize}
\end{lemma}

Next, we are interested in the behavior of Eq. \eqref{eq:def_q_z} under a small perturbation of $\iz$, provided that $z$ is near $\Theta_1$. For this purpose, we consider the following domains of $\iz$.  For the positive constants $a$ and $a'$, we define 
\begin{equation}\label{eq:def_D_an}
D = D(a, a', n_1) \coloneqq \{z = E+ \i \eta \in \mathbb{C}^+~ :~ |E-\Theta_1| \leq a,~ n_1^{-1+a'} \leq \eta \leq 1/a'\}.
\end{equation}
In the subsequent analysis, we shall always choose $a$ and $a'$ to be such that Result (3) of Lemma \ref{lemma:properties_q_z} holds for all $\iz \in D$. 

The following lemma indicates that if a small perturbation is added to $\iz$ when $\iz \in {D}$, Eq. \eqref{eq:def_q_z} is stable in the sense that $q(\iz)$ will only have small changes.  The results are a combination of Definition 5.4 and Definition A.2 of \citet{knowles2017anisotropic}. Therefore, the proof is omitted. 
\begin{lemma}[Strong stability of Eq. \eqref{eq:def_q_z}] \label{lemma:stability_q_z}
Eq. \eqref{eq:def_q_z} is strongly stable in the following sense. Suppose that $\delta: {D} \to (0,\infty)$ satisfies $p^{-2} \leq \delta(\iz) \leq 1/\log(p)$ for $\iz \in {D}$ and that $\delta$ is Lipschitz continuous with Lipschitz constant $p^2$. Suppose moreover that for each fixed $\iz \in {D}$, the function $\eta\to \delta(E+\i\eta) $ is non-increasing for $\eta>0$. Suppose that $u: {D} \to \mathbb{C}^+$ is the Stieltjes transform of a compactly supported probability measure.  Let $\iz \in {D}$ and suppose that 
\[ \left| \tilde{\iz}- \iz  \right|  \leq \delta(\iz),\]
where $\tilde{z} =  -{1}/{u(\iz)} + \hat{\gamma}_2 \phi(-u(z))$. Note that $z = -{1}/{q(\iz)} + \hat{\gamma}_2 \phi(-q(z))$ and $u(\iz) = q(\tilde{\iz})$.

If $\Im \iz <1$, suppose also that 
\begin{equation}\label{eq:stability_statement}
| u - q| \leq \frac{C\delta }{\sqrt{\kappa + \eta } + \sqrt{\delta}} 
\end{equation}
holds at $z + \i p^{-5}$. Here, $\kappa = |E-\Theta_1|$. Then,  Eq. \eqref{eq:stability_statement} holds at $\iz$. 
\end{lemma}
\noindent  
It is worth mentioning that if for some $\iz_0 \in {D}$ with $\eta(\iz_0) \geq 1$, $|\tilde{\iz_0} - \iz_0 | \leq \delta(\iz_0)$ is satisfied, then Eq. \eqref{eq:stability_statement} holds for all $\iz$ in 
\[\Big\{\iz \mid  E(\iz) = E(\iz_0),~~ \eta(\iz_0) -\eta(\iz) \in  p^{-5} \mathbb{N}\Big\} \cap {D}, \]
by induction.

\subsection{Additional properties of the generalized Mar\v{c}enk-Pastur Equation \eqref{eq:MP_G}}\label{subsec:additional_properties_of_the_generalized_mar_v_c}

In this subsection,  we present additional properties of $\phi_{p\lambda}(\iz)$. The focus is on the boundedness of $\phi_{p\lambda}(z)$ around $\rho_{p\lambda}$ and the stability of \eqref{eq:MP_G} against a small perturbation of $z$. These results are analogous to those for $q(\iz)$ in Section \ref{subsec:additional_properties_of_the_mar_v_c}. 


We assume that Conditions \ref{enum:high_dimensional_regime}--\ref{enum:regular_edge} hold. In the remainder of this section, unless stated otherwise, quantities introduced in this section may depend on $p$, $z$, and $\lambda$, even when this dependence is not made explicit. 
We frequently omit the dependence on $p$, $\lambda$, and $z$ in many quantities, whenever no confusion may arise.

For the readers' convenience, recall that as in Section \ref{sec:preliminary}, $\phi(\iz)$ is the solution to Eq. \eqref{eq:MP_G} evaluated at $F = F^{\Sigma_p}$ and $\gamma = \hat{\gamma}_2$, that is 
\[ \phi(z) = \int \frac{\tau dF^{\Sigma_p}(\tau)}{ \tau \{ [1+\hat{\gamma}_2 \phi(\iz)]^{-1} - \iz\} + \lambda}.\]
Recall that, as explained in Section \ref{sec:preliminary}, the equation can be reformulated to 
\[\iz = h(\iz) + \left[ 1+ \hat{\gamma}_2 \int \frac{\tau dF^{\Sigma_p}(\tau)}{-\tau h(\iz) + \lambda} \right]^{-1},\]
where $h(\iz) = z - [1+\hat{\gamma}_2 \phi(\iz)]^{-1}$. We have the definition
\[ x_p(h) = h + \left[ 1+ \hat{\gamma}_2 \int \frac{\tau dF^{\Sigma_p}(\tau)}{-\tau h + \lambda} \right]^{-1}, \]
as the extension to the real line. In the technical proof, we shall also $x_p$ for $h$ being complex.

Further, recall that $\rho$ is defined to be the leftmost edge of the support of $\calG_{p\lambda}$ and $\calG_{p\lambda}$ is the probability measure associated with the Stieltjes transform $\phi(z)$.

Our analysis heavily relies on the main results of the companion paper \cite{li2024analysis}. Proof of results in the rest of this section, if not provided, can be found in \cite{li2024analysis}.

\begin{lemma}[Properties of $\phi(z)$ and $h(z)$: Case (a)]
\label{lemma:properties_phi_h}
Suppose that the condition in Part (a) of Definition \ref{def:regular_edge} is satisfied. Then, we have the following results. 
\begin{itemize}
    \item[(1)] Denote the unique critical point of $x_p(h)$ in $(-\infty, \lambda/\sigma_{1p})$ to be $h_0$.  Then, $\rho = x_p(h_0) > \lambda/\sigma_{1p}$. Moreover, $\calG_{p\lambda}$ is continuous at $\rho$. 
    \item[(2)] The generalized M-P equations are non-singular near $z = \rho$ in the following sense. We can find constants $c>0$ and $C>0$ such that when $|z- \rho|\leq c$ and $z\in\mathbb{C}^+$,  
    \[  \inf_{\tau \in \supp (F^{\Sigma_p})}  \Big|-\tau h(z) + \lambda\Big| > C \quad \text{, }\quad \left| 1+ \hat{\gamma}_2 \int \frac{\tau dF^{\Sigma_p}(\tau)}{-\tau h (z) +\lambda} \right| >C,\]
    \[\text{ and }\inf_{\tau \in \supp (F^{\Sigma_p})}  \Big|\tau  \{[1+ \hat{\gamma}_2 \phi(z)]^{-1} -z \} +\lambda \Big| >C. \]
    Due to the smoothness of $h(z)$ and $\phi(z)$, the results also hold if $\eta(z) = 0$ and $|z -\rho|\leq c$.
    \item[(3)] The following bounds on the derivatives of $x_p(h)$ at $h_0$ hold. We can find constants $C>0$ and $c>0$ such that 
    \[ |h_0| \asymp 1, \quad |x''_p(h_0)| \asymp 1, \quad |x'''_p(h)| \leq C, \quad \text{for }|h - h_0| \leq c.\]

    \item[(4)] There exists constants $c$, $C$ and $C'$ such that for all $|z-\rho | \leq c$ and $z\in\mathbb{C}^+$,
\[|\phi(\iz) | \leq C',\quad \Im\phi(\iz) \geq C\eta, \quad  |1+\hat{\gamma}_2 \phi(\iz)| \geq C.\]
\end{itemize}
\end{lemma}

The following results are shown using Part (3) of Lemma  \ref{lemma:properties_phi_h}. 
\begin{lemma}\label{lemma:expression_imaginary_phi_h}
    Under the conditions of Lemma \ref{lemma:properties_phi_h}, suppose that we fix $c<1$ to be sufficiently small. Then, if $z = E + \i \eta$ is such that $|z-\rho| \leq c$, we have 
\[\Im h(z)  \asymp \begin{cases} \sqrt{\kappa + \eta} & \mbox{if }E > \rho, \\[5pt] \dfrac{\eta}{\sqrt{\kappa+\eta}} & \mbox{ if } E< \rho; \end{cases}\]
\[\Im \phi(z)  \asymp \begin{cases} \sqrt{\kappa + \eta} & \mbox{if }E > \rho, \\[5pt] \dfrac{\eta}{\sqrt{\kappa+\eta}} & \mbox{ if } E< \rho; \end{cases}\]
where $\kappa = |E-\rho|$. 
\end{lemma}
\begin{proof}
Choose $c$ to be sufficiently small. Then, by Lemma \ref{lemma:properties_phi_h}, 
\[z - \rho =x_p(h(z)) - x_p(h_0)  =\frac{x''_p(h_0)}{2}(h(z) - h_0)^2 + O ( |h(z) - h_0|^3), \quad \text{for } \quad |z- \rho| \leq c.\]
It follows then, for $|E-h_0| \leq c$,
\[ E - \rho = \frac{x''_p(h_0) }{2} (h(E) - h_0)^2 \big(1+  O(|h(E) - h_0|) \big).\]
First, because $\rho$ is the leftmost edge point of $\calG_{p\lambda}$, $\Im \phi(E) = 0$ as $\phi(z)$ is the Stieltjes transform of $\calG_{p\lambda}$. It immediately follows that  $\Im (h(E))= 0$ for $E< \rho$, since $h(E) = E - [1+ \hat{\gamma}_2 \phi(E)]^{-1}$. Secondly, if $E> \rho$, taking square root on both sides of the equation, we obtain 
$\Im h(E) \asymp \sqrt{E- \rho}$. Together, we conclude the square-root behavior 
\[\Im h(E) \asymp \sqrt{ |E-\rho| } \mathbb{I}(E > \rho).\] 
Since $ z = h + (1+\hat{\gamma}_2\phi)^{-1}$, the square-root behavior also holds for $\phi(z)$ as  
\[\Im \phi(E) \asymp  \sqrt{|E-\rho|} \mathbb{I}(E > \rho).\] 
By the inversion formula of the Stieltjes transform, it follows that the density $f_{\calG}$ of $\calG_{p\lambda}$ has a square-root behavior near the edge $\rho$. Namely, we can find constants $C$, $C'$, and $\epsilon$ such that 
\[ C_1 \sqrt{E-\rho} \leq f_{\calG}(E) \leq C_2 \sqrt{E-\rho}, \quad \rho \leq  E  \leq \rho+ \epsilon.\]
Then, the described pattern of $\Im \phi(z)$ follows easily. Lastly, the pattern of $\Im h(z)$ is deduced through the relationship $ z = h + (1+\hat{\gamma}_2\phi)^{-1}$. 
\end{proof}

On the other hand, if Part (b) of Definition \ref{def:regular_edge} holds, we have the following result that follows directly from Theorem 3.2 of \citet{li2024analysis}.  
\begin{lemma}[Properties of $\phi(z)$ and $h(z)$: Case (b)] Suppose that Part (b) of Definition \ref{def:regular_edge}  holds. 
\begin{itemize}
    \item[(1)] Then, $\rho = \lambda/\sigma_{1p}$ and $\rho$ is an isolated point of $\calG_{p\lambda}$ with mass $F^{\Sigma_p}(\{\sigma_{1p}\}) - 1/\hat{\gamma}_2$. Therefore, as $\iz \to \rho$, $\Im \phi(\iz)$ diverges and $h(z)\to \rho$. 
\item[(2)] Part (2) and (4) of Lemma \ref{lemma:properties_phi_h} still hold if additionally $|z-\rho|>c'$ for some small constant $c'>0$. 
\end{itemize}
\end{lemma}

\vspace{\baselineskip}
In the rest of this section, we shall focus on the case when $\calG_{p\lambda}$ is continuous at $\rho$. For the positive constants $a$ and $a'$, we define the following domains. First, we define a domain around $\rho$ with the imaginary part at least $n_2^{-1+a'}$ as 
\[Q = Q(a, a', n_2) = \{\iz \in\mathbb{C}^+ ~:~ |E|\leq a^{-1}, \quad E \leq \rho + a, \quad n_2^{-1+a'}\leq \eta \leq 1/a'\}.\]
Secondly, we restrict the real part to be smaller than $\rho$ and define 
\[Q_- = Q_-(a,a', n_2) = \{\iz \in Q(a, a', n_2): E \leq  \rho\}.\]
Thirdly, we define a domain that is away from the support of $\calG_{p\lambda}$ as 
\[Q_{\rm away}=Q_{\rm away}(a,a',n_2) = \{\iz\in \mathbb{C}^+ ~:~  a\leq  \operatorname{dist}(\iz, \calI)\leq a^{-1}, \eta \geq n_2^{-1+a'}\},\] 
where $\calI = [\rho, ~ (1+\sqrt{\hat{\gamma}_2})^2 + \lambda/\liminf \ell_{\min}(\Sigma_p)]$. Here, we are using the fact that the rightmost edge of the support of $\calG_{p\lambda}$ is bounded by  $(1+\sqrt{\hat{\gamma}_2})^2 + \lambda/\liminf \ell_{\min}(\Sigma_p)$. Indeed, recall that 
\[\bG_\lambda = \bZ U_2U_2^T \bZ^T + \lambda \Sigma^{-1}_p.\]
It is well-known in RMT that the largest eigenvalue of $\bZ U_2U_2^T \bZ^T$ converges almost surely to $(1+\sqrt{\hat{\gamma}_2})^2$ under Condition \ref{enum:moments_conditions}. The largest eigenvalue of $\lambda \Sigma^{-1}_p$ is bounded by $\lambda/\liminf \ell_{\min}(\Sigma_p)<\infty$,  

In the subsequent analysis, we shall always choose $a$ and $a'$ to be such that Part (2) and Part (4) of Lemma \ref{lemma:properties_phi_h} hold for $\iz \in Q$. 

\vspace{\baselineskip}

The following lemma is analogous to Lemma \ref{lemma:stability_q_z}. It indicates that if a small perturbation is added to $\iz$ when $\iz \in Q$,  Eq. \eqref{eq:MP_G} is stable in the sense that $\phi(\iz)$ or $h(\iz)$ will only have limited changes. 
\begin{lemma}[Strong stability of Eq. \eqref{eq:MP_G}]\label{lemma:stablity_phi} Suppose that $\rho > \lambda/\sigma_{1p}$. Eq. \eqref{eq:MP_G} is strongly stable in $Q$ in the following sense. Suppose that $\delta: Q \to (0,\infty)$ satisfies $p^{-2} \leq \delta(\iz) \leq 1/\log(p)$ for $\iz \in Q$ and that $\delta$ is Lipschitz continuous with Lipschitz constant $p^2$. Suppose moreover that for each fixed $\iz \in Q$, the function $\eta\to \delta(E+\i\eta) $ is non-increasing for $\eta>0$. 
Suppose that $v(\iz) = \iz - [1+\hat{\gamma}_2 \varrho (\iz) ]^{-1}$,  where $\varrho(\iz)$ is the Stieltjes transform of a compactly supported probability measure.  Let $\iz \in Q$ and suppose that 
\[ \left| \tilde{\iz} - \iz  \right|  \leq \delta(\iz),\]
where 
\[ \tilde{\iz} = v(\iz) + \left[ 1+ \hat{\gamma}_2 \int \frac{\tau dF^{\Sigma_p}(\tau) }{-\tau v(\iz) + \lambda} \right]^{-1}.\]
Note that $v(\iz) = h(\tilde{\iz})$.

If $\Im \iz <1$, suppose also that 
\begin{equation}\label{eq:stability_statement2}
| v - h| \leq \frac{C\delta }{\sqrt{|E-\rho| + \eta } + \sqrt{\delta}} 
\end{equation}
holds at $z + \i p^{-5}$. Then,  Eq. \eqref{eq:stability_statement2} holds at $\iz$. 
\end{lemma}
We postpone the proof of Lemma \ref{lemma:stablity_phi} to the end of this section.

Moreover, when $\iz$ is away from the support of $\calG_{p\lambda}$. The bound can be improved. The following lemma holds for both the case $\rho>\lambda/\sigma_{1p}$ and $\rho = \lambda/\sigma_{1p}$. 
\begin{lemma}\label{lemma:stability_phi_away_case}
Fix any $\lambda>0$. Eq. \eqref{eq:MP_G} is strongly stable in $Q_{away}$ in the following sense. Suppose that $\delta: Q_{away}\to (0,\infty)$ satisfies $\delta(\iz) \leq 1/\log(p)$ for $\iz \in Q_{\rm away}$. Suppose that $v(\iz) = \iz - [1+\hat{\gamma}_2 \varrho (\iz) ]^{-1}$  where $\varrho(\iz)$ is the Stieltjes transform of a compactly supported probability measure.  Let $\iz \in Q_{away}$ and suppose that 
\[ \left| \tilde{\iz} - \iz  \right|  \leq \delta(\iz),\]
where 
\[ \tilde{\iz} = v(\iz) + \left[ 1+ \hat{\gamma}_2 \int \frac{\tau dF^{\Sigma_p}(\tau) }{-\tau v(\iz) + \lambda} \right]^{-1}.\]
  Then, 
\[ | v(\iz)-h(\iz)| \leq C \delta(\iz),\]
for any $\iz \in Q_{\rm away}$.
\end{lemma}

\begin{corollary}\label{corollary:lattice_stability_phi}
It is worth mentioning that if for some $\iz_0 \in {Q}$ with $\eta(\iz_0) \geq 1$, $|\tilde{\iz_0} - \iz_0 | \leq \delta(\iz_0)$ is satisfied, then following immediately by induction, Eq. \eqref{eq:stability_statement2} holds for all $\iz$ in 
\[\Big\{\iz \mid  E(\iz) = E(\iz_0),~~ \eta(\iz_0) -\eta(\iz) \in  p^{-5} \mathbb{N}\Big\} \cap {Q}. \] 
\end{corollary}

\begin{corollary}\label{corollary:stability_phi_cor_1}
Under the conditions of Lemma \ref{lemma:stablity_phi},  if Eq. \eqref{eq:stability_statement2} holds, then we also have 
\[ |\varrho - \phi|\leq \frac{C\delta}{\sqrt{\kappa+\eta} + \sqrt{\delta}}.\]
It is obvious since 
\[ v- h = \frac{\hat{\gamma}_2(\phi - \varrho)}{(1+\hat{\gamma}_2 \varrho)(1+\hat{\gamma}_2\phi )}  \] 
and  $|1+\hat{\gamma}_2 \phi(\iz)| \asymp 1$, $|1 +\hat{\gamma}_2 \varrho(\iz)|\asymp 1$ when $\iz \in Q$. 
\end{corollary}

\begin{corollary}\label{corollary:stability_phi_cor_2}
Using Lemma \ref{lemma:stablity_phi}, we obtain the following result. For $z\in Q$,  suppose that $\varrho(\iz)$ satisfies 
\[ \varrho(\iz) = \int \frac{\tau dF^{\Sigma_p}(\tau) }{ \tau \{  [1+\hat{\gamma}_2 \varrho(\iz) + \delta_2(\iz)]^{-1} - \iz \} + \lambda } + \delta_1(\iz).\] 
Here, $\delta_1$ and $\delta_2$  are such that $|\delta_1|$ and $|\delta_2|$ satisfy the conditions on $\delta$ in Lemma \ref{lemma:stablity_phi}. 
If $\Im\iz <1$, suppose also that 
\begin{equation}\label{eq:corollary:stability_phi_cor_2_eq}
|\varrho - \phi|\leq \frac{C|\delta_1| + C|\delta_2|}{\sqrt{\kappa+\eta} + \sqrt{|\delta_1| + |\delta_2|}}
 \end{equation}
holds at $z+ \i p^{-5}$. Then, Eq. \eqref{eq:corollary:stability_phi_cor_2_eq} holds at $\iz$.  
\end{corollary}

\begin{proof}[Proof of Lemma \ref{lemma:stablity_phi}]
Lemma~\ref{lemma:stablity_phi} can be proved by arguments very similar to those used in the proof of Lemma~\ref{lemma:stability_q_z} in Section~A.2 of \citet{knowles2017anisotropic}. For this reason, we present only the main ideas, omitting many technical details. 

Suppose first that $\Im z \geq a'$ for some constant $a'>0$. Then, from the assumption $\delta(z) \leq 1/\log(p)$, we get $\Im \tilde{z}\geq a'/2$ and therefore, due to the uniqueness of the solution $\phi(z)$ and $h(z)$ to the generalized M-P equation, we find $v(z) = h(\tilde{z})$. Hence,
\[ |v(z) - h(z)| = |h(\tilde{z}) - h(z)| \leq \int_{z}^{\tilde{z}} |h'(\zeta)||d\zeta| \leq C \delta(z).\]
Here, in the last step, we are using the fact that $h'(\zeta)$ is bounded when $\zeta$ is away from the real line.  It yields the bound 
\[|v-h| \leq C\delta \leq \frac{C\delta}{\sqrt{|E-\rho| + \eta} + \sqrt{\delta}}\]
at $z$ such that $\Im z \geq a'$.

What remains is therefore the case $|z-\rho| \leq 2a'$, which we assume for the rest of the proof. By Taylor's expansion, we can find some $v_0$ on the line segment between $v$ and $h$ that
\[ x''_p(v_0) (v-h)^2  + x'_p(h) (v-h)  =  \tilde{z} - z.\]
Suppose that 
\[ |v-h| \leq (\log p)^{-1/3}.\]
The condition will be dealt with later. Then, we claim that for small enough $a'$, we have 
\[  |x''_p(v_0) | \asymp 1 \qquad |x'_p(h)| \asymp \sqrt{ \kappa + \eta}.\]
The claim is proved using arguments similar to those in the proof of (A.16) in \cite{knowles2017anisotropic}. Intuitively, since $v_0$ is close to $h_0$ and $|x''_p(h_0)|\asymp 1$, we have $|x''_p(v_0)| \asymp 1$ since the function is smooth. Secondly, 
\[x'_p(h) = \int_{h_0}^h  x''_p(\zeta)d\zeta =\int_{h_0}^h  x''_p(h_0) + O(|h-h_0|) d\zeta = (h-h_0) x''_p(h_0) + O(|h-h_0|^2). \]
Therefore, 
\[ |x'_p(h)| \asymp |h-h_0| \asymp \sqrt{ |z-\rho|} \asymp \sqrt{\kappa + \eta}. \]
All together, we write 
\[ x''_p(v_0) (v-h)^2 + x'_p(h) (v-h) = O(\delta).\]
The remainder of the proof follows exactly the argument in the proof of Lemma~4.5 of \citet{bloemendal2014isotropic}. In particular, we explicitly solve the quadratic equation for $(v-h)$ and select the appropriate solution via a continuity argument. This selection relies on the bound
\[
|v-h| \leq \frac{C \delta}{\sqrt{|E-\rho| + \eta} + \sqrt{\delta}},
\]
which is assumed to hold at $z + \mathrm{i} p^{-5}$. 

Lastly, the previous condition $|v-h| \leq (\log p)^{-1/3}$ follows by continuity from the corresponding estimate at the neighboring point $z + \mathrm{i} p^{-5}$. We refer the reader to Section~A.2 of \citet{knowles2017anisotropic} and Lemma~4.5 of \citet{bloemendal2014isotropic} for full details.

\end{proof}

\begin{proof}[Proof of Lemma \ref{lemma:stability_phi_away_case}]
When $\iz \in Q_{\rm away}$, $|h'(\tilde{\iz})|$ is bounded. Therefore,
\[ |v(\iz) - h(\iz) | = |h(\tilde{\iz}) - h(\iz)| \leq \int_{z}^{\tilde{\iz}} |h(\zeta)| |d\zeta| \leq C \delta(\iz).\]
\end{proof}

\begin{proof}[Proof of Corollary \ref{corollary:stability_phi_cor_2}]
The result follows from Lemma \ref{lemma:stablity_phi} and Corollary \ref{corollary:stability_phi_cor_1} because the specified $\varrho$ is such that
\[  \left| \tilde{v}(\iz) + \left[ 1+ \hat{\gamma}_2 \int \frac{\tau dF^{\Sigma_p}(\tau)}{ -\tau \tilde{v}(\iz) + \lambda} \right]^{-1}  -\iz \right|  \leq C (|\delta|_1 + |\delta|_2),\]
where $\tilde{v}(\iz) =  \iz -  [1+\hat{\gamma}_2 \varrho(\iz)]^{-1}$. 
\end{proof}

\clearpage
\newpage

\section{Local law and  edge rigidity for $\bG_\lambda$}\label{sec:properties_of_bg_lambda}
Recall that
\[
\bG_\lambda = \bZ U_2 U_2^{T} \bZ^{T} + \lambda \Sigma_p.
\]
In this section, we establish a local law for the resolvent of $\bG_\lambda$ and prove eigenvalue rigidity for its smallest eigenvalue. 

Although all results remain valid for general distributions satisfying the moment conditions in Condition~\ref{enum:moments_conditions}, we focus primarily on the Gaussian setting, since the extension to non-Gaussian distributions is not required for the proof of Theorem~\ref{thm:main}. Nevertheless, such an extension, while technically more involved, can be carried out by following the arguments in Sections~7--9 of \citet{knowles2017anisotropic} in a largely verbatim manner. We discuss this extension in Theorem \ref{thm:G_lambda_extension_nonGauss}.

First of all, when $\bZ$ is normally distributed, $\bG_{\lambda}$ has the same distribution as 
\[ \bG_{\lambda} \overset{\mathrm{def}}{=\joinrel=} \tilde{\bZ} \tilde{\bZ}^T + \lambda \Sigma_p^{-1},\]
where $\tilde{\bZ}$ is $p \times n_2$ whose entries are iid $\calN(0, 1/n_2)$. We focus on this representation in the following analysis. 

Our study is divided into two regimes: (i) $\rho > \lambda / \sigma_{1p}$ and (ii) $\rho = \lambda / \sigma_{1p}$. We focus primarily on the first regime, where the analysis is more involved. The second regime is comparatively straightforward and is addressed in Section \ref{subsec:discrete_case}. 


Throughout the analysis, we consider the stochastic domination (Definition \ref{def:stochastic_domination}) in the asymptotic regime $n_2,p \to \infty$ as in Condition \ref{enum:high_dimensional_regime}, with $\lambda > 0$ fixed. Accordingly, we suppress the explicit dependence of quantities on $\lambda$, while retaining their dependence on $p$ to avoid potential confusion.

We first introduce a fundamental control parameter as 
\[\Psi(\iz)  = \sqrt{\frac{\Im\phi_p(\iz)}{n_2 \eta}} + \frac{1}{n_2 \eta}.\]
Call the resolvent of $\bG_\lambda$ 
\[R(\iz) = \left(\bG_\lambda - \iz I_p \right)^{-1}, \quad \iz \in\mathbb{C}^+.\]
Recall the following domains defined in Section \ref{subsec:additional_properties_of_the_generalized_mar_v_c} 
\begin{align*}
&Q = Q(a, a', n_2) = \{\iz \in\mathbb{C}^+ ~:~ |E|\leq a^{-1}, \quad E \leq \rho + a, \quad n_2^{-1+a'}\leq \eta \leq 1/a'\},\\
&Q_- = Q_-(a,a', n_2) = \{\iz \in Q(a, a', n_2): E \leq  \rho\},\\
&Q_{\rm away}=Q_{\rm away}(a,a',n_2) = \{\iz\in \mathbb{C}^+ ~:~  a\leq  \operatorname{dist}(\iz, \calI)\leq a^{-1}, \eta \geq n_2^{-1+a'}\}.
\end{align*}

Define the $(p +n_2)\times (p +n_2)$ linearizing block matrix 
\begin{equation}\label{eq:bJ_bK}
\bJ (\iz) = \bK^{-1}(\iz), \quad \bK(\iz)  = \left(\begin{matrix} \lambda \Sigma_p^{-1} - \iz I_p  &  \tbZ \\ \tbZ^T & - I_{n_2}\end{matrix}\right).
\end{equation}
The construction of a linearization matrix is a commonly used technique in RMT when study the local laws of various random matrices. Indeed, $\bJ$  contains the resolvent $R(\iz)$ as one of its blocks (see Lemma \ref{lemma:resolvent_identity}). Unlike $\bG_\lambda$ directly, $\bK(z)$ is linear in $\tilde{\bZ}$. The dependence of $\bK(z)$ on each entry of $\tilde{\bZ}$ is much easier to deal with.

Call 
\[ \Omega(\iz) = \left(\begin{matrix}\left(\lambda\Sigma_p^{-1} - h_p(\iz) I_p \right)^{-1}  & 0  \\  0 &  (\iz-h_p(\iz)) I_{n_2}  \end{matrix} \right).
\] 


\begin{theorem}[Strong local law]\label{thm:strong_local_law}
Suppose that Conditions \ref{enum:high_dimensional_regime}--\ref{enum:regular_edge} hold. Further suppose that the observations are Gaussian and $\rho>\lambda/\sigma_{1p}$. For arbitrary $a'>0$, there exists constants $a>0$ such that the following local laws hold.
\begin{itemize}
\item[(i)] The entrywise local law in $Q$ holds as 
\[\bJ(z) - \Omega(z) = O_\prec(\Psi(z)), \quad z \in Q.\]
\item[(ii)] The averaged local law in $Q$ holds as 
\[ \hat{\phi}_p(\iz) - \phi_p(\iz) = O_\prec\left(\frac{1}{n_2\eta}\right),  \quad z \in Q, \]
where $\hat{\phi}_p(\iz) = p^{-1} \tr R(\iz)$. 
\item[(iii)] The averaged local law in $Q_-$ holds as 
\[ \hat{\phi}_p(\iz) - \phi_p(\iz) = O_\prec\left(\frac{1}{n_2(\kappa + \eta) }\right), \]
where $\kappa = |E-\rho_p|$. 
\item[(iv)] The averaged local law in $Q_{\rm away}$ holds as 
\[ \hat{\phi}_p(\iz) - \phi_p(\iz) = O_\prec\left(\frac{1}{n_2}\right).\]
\end{itemize}
\end{theorem}

Once Theorem \ref{thm:strong_local_law} is ready, it is not hard to deduce the following rigidity of the smallest eigenvalue of $\bG_\lambda$. 
\begin{theorem}[Rigidity of the smallest eigenvalue of $\bG_\lambda$]\label{thm:rigidity_G}
Under the conditions of Theorem \ref{thm:strong_local_law}, for any fixed $\epsilon \in (0,2/3)$, with high probability, there is no eigenvalue of $\bG_\lambda$ in $(-\infty, \rho_p - n_2^{-2/3+\epsilon})$. 
\end{theorem}

The following theorem shows the convergence of \emph{linear spectral statistics} of $\bG_\lambda$. 
\begin{theorem}\label{thm:analytical_function_1n_convergence}
Suppose the conditions of Theorem \ref{thm:strong_local_law} hold. Consider the interval
\[  \calI_\epsilon = [\rho_p - \epsilon, ~ (1+\sqrt{p/n_2})^2 + \lambda/\liminf\ell_{\min}(\Sigma_p) + \epsilon],\]
where $\epsilon$ is any fixed positive constant.  
\begin{itemize}
\item[(i)] Then, with high probability, all eigenvalues of $\bG_\lambda$ are contained in $\calI_\epsilon$. 
\item[(ii)] 
Let $f$ be any function analytic on $\calI(\infty)$ with $\calI(\infty)$ being any open interval such that $\calI_\epsilon \subset \calI (\infty)$ for all sufficiently large $p$. The existence of $\calI(\infty)$ is straightforward under the conditions. 
Denote by $\frakM$ the event that all eigenvalues of $\bG_\lambda$ are contained in $\calI_\epsilon$. We have 
 \[ \mathbb{I}(\frakM) \left| \frac{1}{p} \sum_{j=1}^pf(\ell_{j}(\bG_\lambda) )  - \int f(\tau) d\calG_{p\lambda}(\tau) \right| \prec \frac{1}{n_2}. \] 
\end{itemize}
\end{theorem}

For the proof of Theorem~\ref{thm:main}, it suffices to establish the local laws and rigidity results under Gaussianity. Nevertheless, these results can be extended to the non-Gaussian setting.

\begin{theorem}\label{thm:G_lambda_extension_nonGauss}
Theorems~\ref{thm:strong_local_law}, \ref{thm:rigidity_G}, and \ref{thm:analytical_function_1n_convergence} remain valid when the Gaussianity assumption is removed.
\end{theorem}

The rest of this section is devoted to the proof of Theorem \ref{thm:strong_local_law}, Theorem \ref{thm:rigidity_G}, and Theorem \ref{thm:analytical_function_1n_convergence}. We also discuss the proof of Theorem \ref{thm:G_lambda_extension_nonGauss}. 

\subsection{Proof of Theorem \ref{thm:strong_local_law}}\label{subsec:proof_of_theorem_ref_thm_strong_local_law}

The analysis proceeds in two steps. First, we assume that $\Sigma_p$ is diagonal, with diagonal entries $\sigma_{1p} \geq \sigma_{2p} \geq \cdots \geq \sigma_{pp}$. We then extend the results to a general covariance matrix $\Sigma_p$.

Our analysis follows the framework developed in Sections~4 and~5 of \citet{knowles2017anisotropic}. Accordingly, we omit many technical details and refer the reader to the corresponding arguments therein. Throughout, we frequently suppress the argument $z$ from our notation.

\subsubsection{Basic tools}\label{ssub:basic_tools}

We make the following notation. We use  $A_{st}$ to denote the $(s,t)$-th element of a matrix $A$. The four blocks of $\bJ(\iz)$ are called $\bJ_{(11)}$, $\bJ_{(12)}$, $\bJ_{(21)}$ and $\bJ_{(22)}$. For  a set of indices, say $S\subset\{1,2,\dots, p+ n_2\}$, we define the minors $\bK^{(S)} = (\bK_{st}, s,t \notin S)$. That is, $\bK$ but with the rows and columns whose indices are specified in $S$ deleted. We also write $\bJ^{(S)} = (\bK^{(S)})^{-1}$.  The block formed by the first $p$ rows and columns of $\bJ^{(S)}$ is called $\bJ_{(11)}^{(S)}$. The block formed by the last $n_2$ rows and columns of $\bJ^{(S)}$ is called $\bJ_{(22)}^{(S)}$. Similar for $\bJ_{(12)}^{(S)}$ and $\bJ_{(21)}^{(S)}$. Similar notation is made for the matrix $\Omega$.  

For $u=1,2,\dots,p$, call the $u$-th row of $\tbZ$ to be $\tbZ_{u\cdot}$. For $i = p+1, \dots, p+n_2$, we use  $\tbZ_{\cdot i}$ to denote the $(i-p)$-th column of $\tbZ$. Moreover, if $u \in \{1,\dots, p\}$, let $e_u$ denote the standard unit vector of dimension $p$ in the coordinate direction $u$. If $i \in \{p+1, \dots, p+n_2\}$, let $e_i$ denote the standard unit vector of dimension $n_2$ in the coordinate direction $(i-p)$.

The following lemma is a variant of Lemma 4.4 in \citet{knowles2017anisotropic}. These results follow directly from the Schur complement formula and resolvent identities that have been previously derived in  \citet{knowles2017anisotropic} and the references therein. The proof is therefore omitted. 
\begin{lemma}[Resolvent Identities] Suppose that $\Sigma_p$ is diagonal with diagonal elements $\sigma_{1p},\dots, \sigma_{pp}$. 
\begin{itemize}
\item[(i)] We have 
\[\bJ = \left(\begin{matrix} \bJ_{(11)} & & & & \bJ_{(12)} \\ \bJ_{(21)} & & & & \bJ_{(22)} \end{matrix} \right)= \left(\begin{matrix} R & & & & R \tbZ \\ \tbZ^T R & & & &\tbZ^T R\tbZ - I_{n_2} \end{matrix} \right).\]

\item[(ii)] For $u \in \{1,2,\dots, p\}$, we have 
\[ \frac{1}{\bJ_{uu}} = \lambda/\sigma_{up} - z  - \tbZ_{u\cdot} \bJ_{(22)}^{(u)} \tbZ_{u\cdot}^T. \]
For $u\neq v \in \{ 1,2,\dots, p\}$, 
\[ \bJ_{uv} = - \bJ_{uu} \tbZ_{u\cdot} \bJ_{(22)}^{(u)} e_v = - \bJ_{vv} e_u^T \bJ_{(22)}^{(v)} \tbZ_{v \cdot}^T = \bJ_{uu} \bJ_{vv}^{(u)} \tbZ_{u\cdot} \bJ_{(22)}^{(uv)}\tbZ_{v\cdot}^T. \] 

\item[(iii)] For $i \in \{p+1, \dots, p+ n_2\}$, we have 
\[ \frac{1}{\bJ_{ii}} = -1 - \tbZ_{\cdot i}^T \bJ_{(11)}^{(i)} \tbZ_{\cdot i}.\]
For $i \neq j \in \{p+1, \dots, p +n_2\}$, we have 
\[ \bJ_{ij}= - \bJ_{ii} \tbZ_{\cdot i}^T \bJ_{(11)}^{(i)} e_j = - \bJ_{jj} e_i^T \bJ^{(j)}_{(11)} \tbZ_{\cdot j} = \bJ_{ii} \bJ_{jj}^{(i)} \tbZ^T_{\cdot i }\bJ_{(11)}^{(ij)}  \tbZ_{\cdot j} .\]  

\item[(iv)] For $u \in \{1, 2, \dots, p\}$ and $i \in\{p+1, \dots, p+n_2\}$, 
\[\bJ_{ui}= - \bJ_{ii} e_u^T \bJ^{(i)}_{(11)} \tbZ_{\cdot i}. \]

\item[(v)] If $i,j, k \notin S$ and $i,j \neq k$,
\[ \bJ_{ij}^{(S)} = \bJ_{ij}^{(S\cup \{k\})} + \frac{  \bJ_{ik}^{(S)} \bJ_{kj}^{(S)}}{\bJ_{kk}^{(S)}}.\] 
\end{itemize}
\label{lemma:resolvent_identity}
\end{lemma}

The following result is our fundamental tool for estimating entries of $\bJ$. 
\begin{lemma}\label{lemma:fundamental_tool}
Fix $a>0$. Then, the following estimates hold for any $z\in Q$. We have 
\[  \|\bJ\| \leq \frac{C}{\eta^3}, \quad  \quad \quad \|\partial z \bJ \|\leq \frac{C}{\eta^6}.\]
Furthermore, let $\bw\in \mathbb{R}^p$ and $\bv \in \mathbb{R}^{n_2}$. Then, we have the bounds 
\begin{align*}
&\sum_{u=1}^p |\bw^T \bJ_{(11)} e_u|^2 = \frac{\Im \bw^T \bJ_{(11)} \bw }{\eta},\\
&\sum_{i=p+1}^{p+n_2} | \bv^T \bJ_{(22)} e_i |^2 \leq \frac{C\|\tbZ\tbZ^T \|}{\eta} \Im \bv^T\bJ_{(22)} \bv + C \bv^T\bv,\\
&\sum_{u=1}^p | \bv^T \bJ_{(21)} e_u|^2 = \frac{\Im \bv^T \bJ_{(22)} \bv}{\eta},\\  
&\sum_{i=p+1}^{p+ n_2} | \bw^T \bJ_{(12)} e_i |^2\leq  C \| \tbZ\tbZ^T \| \sum_{u=1}^{p} |\bw^T \bJ_{(11)} e_u|^2. 
\end{align*}
The estimates remain true for $\bJ^{(S)}$ instead of $\bJ$ if $S\subset\{1,\dots, p \}$ or $S\subset\{p+1,\dots, p+n_2\}$. 
\end{lemma}
\begin{proof}
First, consider the bound on the operator norm of $\bJ$. 
\begin{align*}
\|\bJ\| &\leq \left\| \left(\begin{matrix} (\lambda\Sigma_p^{-1} - \iz I_p)^{-1} & 0\\ 0 & I \end{matrix}\right)  \right\| \left\| \left(\begin{matrix} I_p &  (\lambda\Sigma_p^{-1} - \iz I_p)^{-1/2}\tbZ   \\ \tbZ^T (\lambda\Sigma_p^{-1} - \iz I_p)^{-1/2} & -I_{n_2} \end{matrix}\right)^{-1}\right\| \\
&= \|J_1\| \|J_2\|, \mbox{ say}.  
\end{align*}
For $J_1$, we clearly have $\| J_1 \| \leq C/(|E - \lambda/\sigma_{1p}| + \eta)\leq C/\eta$, for $\iz \in Q$. The eigenstructure of a matrix of the form $J_2$ is analyzed in (4.12-4.15) of \citet{knowles2017anisotropic}. Based on their arguments, 
\begin{equation*}
\begin{aligned}
\|J_2\| \leq C (1 + \max_{u} \frac{|s_u| + 1}{|s_u- 1|})&\leq  C (1 + \max_u \frac{(C/\eta)\|\tbZ\tbZ^T\| + 1}{|s_u -1|} ) \\
&\leq C + (C/\eta)\|\tbZ\tbZ^T\| \|\lambda\Sigma^{-1}_p - \iz I\| \|(\tbZ\tbZ^T - (\lambda\Sigma_p^{-1} - \iz I))^{-1}\|\\
& \leq C+ \frac{C}{\eta^2},
\end{aligned}
\end{equation*}
where $s_k$'s are the eigenvalues of  $(\lambda\Sigma_p^{-1} - \iz I)^{-1/2}\tbZ \tbZ^T (\lambda\Sigma_p^{-1} - \iz I)^{-1/2}$.  The bound on $\|\bJ\|$ follows. 

The bound on the operator norm of $\partial_{\iz} \bJ$ follows from the fact that $\| \partial_{\iz} \bJ\| = \| \bJ (\partial_{\iz} \bK ) \bJ  \| \leq \frac{C}{\eta^6}$.

Next, using Result (i) of Lemma \ref{lemma:resolvent_identity}, 
\[  \sum_{u=1}^p |\bw^T \bJ_{(11)} e_u|^2  = \sum_{u=1}^p \bw^T \bJ_{(11)} e_u e_u^T \overline{\bJ}_{(11)} \bw  = \bw^T \bJ_{(11)} \overline{\bJ}_{(11)} \bw = \bw^T \frac{ \Im \bJ_{(11)} }{\eta} \bw.\] 
Moreover, 
\begin{align*}
\sum_{i=p+1}^{p+n_2}& |\bv^T \bJ_{(22)} e_i|^2  = \sum_{i=p+1}^{p+n_2} |\bv^T \tbZ^T R\tbZ e_i - \bv^T e_i|^2 \leq 2 \bv^T \tbZ^T \bJ_{(11)} \tbZ \tbZ^T \overline{\bJ}_{(11)} \tbZ\bv + 2\bv^T \bv \\
& \leq C \|\tbZ\tbZ^T\| \bv^T \tbZ^T \bJ_{(11)} \overline{\bJ}_{(11)} \tbZ \bv + C\bv^T \bv  \leq C \|\tbZ \tbZ^T \| \frac{\Im \bv^T \bJ_{(22)} \bv}{\eta} + C\bv^T \bv. 
\end{align*}
The remaining two results on $\bJ_{(12)}$ and $\bJ_{(21)}$ can be proved similarly using Result (i) of Lemma \ref{lemma:resolvent_identity}. We omit the details. 
\end{proof}

\begin{lemma}\label{lemma:bounds_on_ZZ}
Under \ref{enum:high_dimensional_regime}, $\|\tbZ\tbZ^T\| \leq  (1+\sqrt{p/n_2})^{2} + \epsilon$ with high probability, for any fixed $\epsilon >0$.  
\end{lemma}
\noindent The lemma follows from Theorem 2.10 of \citet{bloemendal2014isotropic}.

\subsubsection{Weak entrywise law} \label{ssub:weak_entrywise_law}

We first show a weak entrywise local law of $\bJ$ in $Q$. 
\begin{proposition}\label{proposition:weak_local_law}
Suppose that the assumptions of Theorem \ref{thm:strong_local_law} hold. Define 
\[ \Lambda  = \max_{1\leq s,t\leq p+n_2 } |(\bJ - \Omega)_{st}| \quad  \mbox{and} \quad \Lambda_o = \max_{s\neq t} | (\bJ-\Omega)_{st}  |.\]
Then, $\Lambda \prec (n_2 \eta)^{-1/4}$ uniformly in $\iz \in Q \cup Q_{\rm away}$. 
\end{proposition}

Define the averaged control parameter 
\[ \frakB= \frakB_p + \frakB_n, \quad \frakB_p = \left| \frac{1}{p} \sum_{i=u}^p (\bJ_{uu} - \Omega_{uu})\right|, \quad \frakB_n = \left| \frac{1}{n_2} \sum_{i=p+1}^{p+n_2} (\bJ_{ii} - \Omega_{ii})\right|. \]
Note that $\frakB_p = |\hat{\phi}_p  -\phi_p|$. We have the trivial bound $\frakB \leq C \Lambda$.

For $s\in \{1,2,\dots, p+n_2\}$, we introduce the conditional expectation
\[  \mE_s[~\cdot~] = \mE[ ~\cdot ~\mid \bK^{(s)}]. \]
Using Lemma \ref{lemma:resolvent_identity}, we get that, for $u\in \{1,\dots, p\}$, 
\begin{equation}\label{eq:self_consistent1}
\frac{1}{\bJ_{uu}} = \lambda/\sigma_{up} - z - \frac{1}{n_2} \tr \bJ_{(22)}^{(u)} - W_u, \quad \quad W_u = (1- \mE_u)(\tbZ_{u\cdot} \bJ_{(22)}^{(u)} \tbZ_{u\cdot}^T).
\end{equation}
For $i \in \{p+1, \cdots, p+n_2\}$, 
\begin{equation}\label{eq:self_consistent2}
\frac{1}{\bJ_{ii}} = -1 - \frac{1}{n_2} \tr \bJ_{(11)}^{(i)} - W_i, \quad\quad W_i = (1- \mE_i) (\tbZ_{\cdot i}^T \bJ_{(11)}^{(i)} \tbZ_{\cdot i}). 
\end{equation}

We define the $\iz$-dependent event $\Xi = \{\Lambda \leq 1/(\log n_2) \}$ and the control parameter
\[ \Psi_{\frakB} = \sqrt{\frac{\Im \phi_p + \frakB }{n_2 \eta}}.\] 
The following estimate is analogous to Lemma 5.2 of \citet{knowles2017anisotropic}. 
\begin{lemma}\label{lemma:bounded_W_off_diagonal}
Suppose that the assumptions of Theorem \ref{thm:strong_local_law} hold. Then, uniformly  for $s \in \{1,\dots, p+n_2\}$ and $\iz \in Q\cup Q_{\rm away}$, we have 
\begin{align*}
&\mathbb{I}(\Xi) (|W_s| + \Lambda_o) \prec \Psi_{\frakB},\\
&\mathbb{I}(\eta \geq 1) (|W_s| + \Lambda_o) \prec \Psi_{\frakB}.
\end{align*}
\end{lemma}
\begin{proof}[Proof of Lemma \ref{lemma:bounded_W_off_diagonal}]
The proof relies on the identities from Lemma \ref{lemma:resolvent_identity} and large deviation estimates like Lemma 3.1 of \citet{bloemendal2014isotropic}. The arguments in the proof are similar to those in Lemma 5.2 of \citet{knowles2017anisotropic}. 

From Lemma \ref{lemma:properties_phi_h}, when $\iz \in Q\cup Q_{\rm away}$, $|\Omega_{ss}| \asymp 1$. Therefore, 
\[ \mathbb{I}(\Xi) \bJ_{ss} \asymp 1. \] 
Using Result (v) of Lemma \ref{lemma:resolvent_identity} and a simple induction argument, it is not hard to conclude that 
\[ \mathbb{I}(\Xi) \bJ_{ss}^{(S)} \asymp 1,\]
for any $S \subset \{1,\dots, p+n_2\}$ and $s\notin S$, satisfying $|S|\leq C$.

Let us first estimate $\Lambda_o$. When $u\neq v \in \{1,\dots, p\}$, using Result (ii) of Lemma \ref{lemma:resolvent_identity} and a large deviation estimation (Lemma 3.1 of \citet{bloemendal2014isotropic}), 
\begin{equation}\label{eq:proof_rigidity_eq_bounds_bJ_uv}
\mathbb{I}(\Xi)  |\bJ_{uv}| \leq  \mathbb{I}(\Xi) |\bJ_{uu} \bJ_{vv}^{(u)} | \left|\sum_{i,j = p+1}^{p+n_2}  \tilde{Z}_{ui} \tilde{Z}_{vj} \bJ^{(uv)}_{i, j}    \right| \prec \mathbb{I}(\Xi) C \left( \frac{1}{n_2^2} \sum_{i,j =p+1}^{p+n_2} |\bJ_{ij}^{(uv)} |^2 \right)^{1/2}.
\end{equation}
The term in parentheses is 
\begin{align*}
&\mathbb{I}(\Xi) \frac{1}{n_2^2} \sum_{i,j=p+1}^{p+n_2} |\bJ_{ij}^{(uv)} |^2  = \mathbb{I}(\Xi) \frac{1}{n_2^2} \sum_{i,j=p+1}^{p+n_2} |e_i^T\bJ^{(uv)}_{22} e_j|^2  \\
& \prec \mathbb{I}(\Xi) \frac{1}{n_2^2\eta} \|\tilde{\bZ}\tilde{\bZ}^T\| \sum_{i=p+1}^{p+n_2}\Im \bJ_{ii}^{(uv)} + \mathbb{I}(\Xi)  \frac{1}{n_2} \\
& \prec \mathbb{I}(\Xi) \frac{1}{n_2^2\eta} \|\tilde{\bZ}\tilde{\bZ}^T\| \sum_{i=p+1}^{p+n_2}\Im \bJ_{ii} +  \mathbb{I}(\Xi) \frac{\Lambda_o^2}{n_2\eta} +  \mathbb{I}(\Xi)  \frac{1}{n_2} \\
&\prec \frac{\Im \phi_p + \frakB + \Lambda_o^2}{n_2\eta}  + \frac{1}{n_2} \prec \frac{\Im \phi_p + \frakB + \Lambda_o^2 }{n_2\eta}.
\end{align*}
Here, we are using Result (v) of Lemma \ref{lemma:resolvent_identity}, Lemma \ref{lemma:bounds_on_ZZ}  and Lemma \ref{lemma:fundamental_tool}. The last step is due to the fact that $1\prec {\Im\phi_p}/\eta$.

A similar result for $\bJ_{ij}$, $\bJ_{ui}$, and $\bJ_{iu}$, $i\in\{p+1, p+n_2\}$ and $u\in\{1, p\}$ can be obtained using Result (iii) and (iv) of Lemma \ref{lemma:resolvent_identity} and a large deviation estimation. We omit the details as the arguments are very similar.

All together, we conclude that 
\[\mathbb{I}(\Xi) \Lambda_o \prec \mathbb{I}(\Xi) \sqrt{\frac{\Im \phi_p + \frakB +\Lambda_o^2}{n_2\eta}}.\] 
Since $1/\sqrt{n_2\eta} \to 0$ when $\iz \in Q$, it is easy to deduce that $\mathbb{I}(\Xi)\Lambda_0 \prec \Psi_{\frakB}$.

An analogous argument for $|W_s|$ (see e.g. Lemma 5.2 of \citet{erdos2013local}) completes the proof of the statement in Proposition \ref{proposition:weak_local_law} when $\Xi$ holds.

As for the case when $\eta\geq 1$, we proceed similarly. For $\eta\geq 1$, we proceed as above to get $|W_s|\prec n_2^{-1/2}$, where we used that $\|\bJ\|\leq C/\eta \leq C$ by Lemma \ref{lemma:fundamental_tool}. Similarly, as in Eq. \eqref{eq:proof_rigidity_eq_bounds_bJ_uv}, we get 
\[\bJ_{uv} \leq  |\bJ_{uu} \bJ_{vv}^{(u)} | \left|\sum_{i,j = p+1}^{p+n_2}  \tilde{Z}_{ui} \tilde{Z}_{vj} \bJ^{(uv)}_{i, j}    \right|  \prec |\bJ_{uu} \bJ_{vv}^{(u)} \frac{1}{\sqrt{n_2}},\]
where in the last step we used that $|\Im\bJ^{(uv)}_{ii}| \leq C$ because $\|\bJ\|\leq C$. Moreover, $|\bJ_{uu} \bJ_{vv}^{(u)} | \leq C$. This completes the proof of Lemma \ref{lemma:bounded_W_off_diagonal}. 
\end{proof}

In the following, we aim to show that $\hat{\phi}_p$ satisfies a perturbed version of Eq. \eqref{eq:MP_G} when $\Xi$ holds or when $\eta\geq 1$. From Eq. \eqref{eq:self_consistent1}, Eq. \eqref{eq:self_consistent2}, Lemma \ref{lemma:bounded_W_off_diagonal}, Result (v) of Lemma \ref{lemma:resolvent_identity}, we get that uniformly for $u \in \{1,\dots, p\}$, $i \in \{p+1,\dots, p+n_2\}$, $z\in Q\cup Q_{\rm away}$, 
\begin{equation}\label{eq:self_consistent1_approximate}
\mathbb{I}(\Xi) \frac{1}{\bJ_{uu}} = \mathbb{I}(\Xi) \left(\lambda/\sigma_{up} - z  - \frac{1}{n_2} \sum_{i=p+1}^{p+n_2} \bJ_{ii} - W_u + O_\prec(\Psi_{\frakB}^2)\right),
\end{equation}
\begin{equation}\label{eq:self_consistent2_approximate}
- \mathbb{I}(\Xi) \frac{1}{\bJ_{ii}} = \mathbb{I}(\Xi)\left(1+ \frac{1}{n_2} \sum_{u=1}^p  \bJ_{uu} + W_i + O_\prec(\Psi_{\frakB}^2) \right).
\end{equation}

Invert Eq. \eqref{eq:self_consistent2_approximate}, get the Taylor's expansion of the right-hand side to the order of $W_i^2$, and average the equation over $i=p+1,\dots, p+n_2$. We obtain 
\begin{equation}\label{eq:expression_bJ_22_trace}
 -\mathbb{I}(\Xi) \frac{1}{n_2}  \sum_{i=p+1}^{p+n_2} \bJ_{ii}  = \mathbb{I}(\Xi) \frac{1}{1 + \frac{1}{n_2} \sum_{u=1}^{p} \bJ_{uu}} + \mathbb{I}(\Xi) \delta_1,
 \end{equation}
where the residual $\delta_1$ is 
\[\mathbb{I}(\Xi)\delta_1 = \mathbb{I}(\Xi) [W]_n + \mathbb{I}(\Xi)\frac{1}{n_2}\sum_{i=p+1}^{p+n_2} O_{\prec}( |W_i|^2) + \mathbb{I}(\Xi) O_\prec(\Psi_{\frakB}^2) = \mathbb{I}(\Xi)[W]_n + O_\prec(\Psi_\frakB^2),\]
\[ [W]_n  = \frac{1}{(1 + \frac{1}{n_2} \sum_{u=1}^{p} \bJ_{uu})^2} \frac{1}{n_2} \sum_{i=p+1}^{p+n_2} W_i.\]
Here, we are using Lemma \ref{lemma:bounded_W_off_diagonal}. 

Similarly, from Eq. \eqref{eq:self_consistent1_approximate},
\[ \mathbb{I}(\Xi) \frac{1}{p} \sum_{u=1}^p \bJ_{uu}  = \mathbb{I}(\Xi) \frac{1}{p} \sum_{u=1}^p  \frac{1}{\lambda/\sigma_{up}  -z -\frac{1}{n_2} \sum_{i=p+1}^{p+n_2} \bJ_{ii}} + \delta_2,\]
where 
\[\mathbb{I}(\Xi) \delta_2 = \mathbb{I}(\Xi) [W]_p + \mathbb{I}(\Xi) \frac{1}{p} \sum_{u=1}^{p} O_\prec(|W_u|^2) + \mathbb{I}(\Xi)O_\prec(\Psi_\frakB^2) = \mathbb{I}(\Xi) [W]_p + O_\prec(\Psi_\frakB^2),\]
\[[W]_p = \frac{1}{p} \sum_{u=1}^p  \frac{-1}{(\lambda/\sigma_{up}  -z -\frac{1}{n_2} \sum_{i=p+1}^{p+n_2} \bJ_{ii})^2} W_u. \]

All together, we obtain
\[ \mathbb{I}(\Xi) \hat{\phi}_p = \mathbb{I}(\Xi)  \int \frac{dF^{\Sigma_p}(\tau)}{\lambda/\tau - z - \displaystyle\frac{1}{1+ (p/n_2)\hat{\phi}_p}  + \mathbb{I}(\Xi) \delta_1 } + \mathbb{I}(\Xi) \delta_2, \]
which is a perturbed version of Eq. \eqref{eq:MP_G}. Note that Lemma \ref{lemma:bounded_W_off_diagonal} guarantees that $\mathbb{I}(\Xi)\delta_1\prec \Psi_\frakB$ and $\mathbb{I}(\Xi)\delta_2\prec \Psi_\frakB$. This estimate is actually not tight. We shall later improve the estimate to $O_\prec (\Psi_\frakB^2)$. 

On the other hand, it is not hard to identify that the derivation actually only relies on the fact that $\mathbb{I}(\Xi)(|W|_s + \Lambda_o) \prec \Psi_\frakB$ from Lemma \ref{lemma:bounded_W_off_diagonal}. Since the same control holds when $\eta \geq 1$. We obtain that 
\[ \mathbb{I}(\eta\geq1) \hat{\phi}_p = \mathbb{I}(\eta\geq 1)  \int \frac{dF^{\Sigma_p}(\tau)}{\lambda/\tau - z - \displaystyle\frac{1}{1+ (p/n_2)\hat{\phi}_p}  +\mathbb{I}(\eta\geq 1)  \delta_1 } + \mathbb{I}(\eta\geq1)\delta_2,\]
and $\mathbb{I}(\eta\geq1)\delta_1 \prec \Psi_\frakB$, $\mathbb{I}(\Xi) \delta_2 \prec \Psi_\frakB$.

We apply Lemma \ref{lemma:stablity_phi}, Lemma \ref{lemma:stability_phi_away_case}, and Corollary \ref{corollary:stability_phi_cor_2} to control the difference between $\hat{\phi}_p$ and $\phi_p$ when $\eta \geq 1$. In particular, we have the following results. 
\begin{lemma}\label{lemma:weak_local_eta_1}
Suppose that the conditions of Theorem \ref{thm:rigidity_G} hold. Then, we have $\Lambda \prec n_2^{-1/2}$ uniformly in $\iz \in Q\cup Q_{\rm away}$ satisfying $\eta \geq1$.  
\end{lemma}
\begin{proof}
Applying Lemma \ref{lemma:bounded_W_off_diagonal}, $\Lambda_o \prec \Psi_{\frakB} \prec n_2^{-1/2}$. Here, we are using the fact that $\Im \phi_p  + \frakB = O(1)$, when $\eta \geq 1$. 
We only need to show that the diagonal elements are such that 
\[ \bJ_{uu} -  \frac{1}{\lambda/\sigma_{up} - h_p} \prec n_2^{-1/2}, \quad u =1,2,\dots, p,\]
\[ - \bJ_{ii} - \frac{1}{1+p/n_2 \phi_p } \prec n_2^{-1/2}, \quad i = p+1,p+2,\dots,p+n_2.\]

Applying Lemma \ref{lemma:stability_phi_away_case} or Corollary \ref{corollary:stability_phi_cor_2}, when $\eta \geq 1$, 
\[ |\hat{\phi}_p - \phi_p| \leq {C(|\delta_1|+|\delta_2|)}\prec \Psi_{\frakB} \prec n_2^{-1/2}.\] 
Then, by Eq. \eqref{eq:self_consistent2_approximate}
\[ - \frac{1}{\bJ_{ii}} = 1+ \phi_p + (\hat{\phi}_p - \phi_p) + W_i + O_\prec(\Psi_\frakB^2)  = 1 + \phi_p + O_\prec(n_2^{-1/2}). \]
Therefore, $ - \bJ_{ii} - \frac{1}{1+\phi_p}  \prec n_2^{-1/2}$. The estimate of $\bJ_{uu}$ is obtained similarly. 
\end{proof}

To show Proposition \ref{proposition:weak_local_law}, we apply the stochastic continuity argument, exactly as in Section 4 of \citet{bloemendal2014isotropic}. The argument allows us to propagate the smallness of $\Lambda$ from $\eta \geq 1$ to the whole domain $Q \cup Q_{\rm away}$. Note that in the domain, $\eta \geq n_2^{-1 + a'}$. We choose $\varepsilon < a'/4$ and an arbitrary $D>0$. Let 
\[L(z) = \{\iz\} \cup \{ s\in Q\cup Q_{\rm away}~:~ \Re s =\Re \iz, \quad \Im s \in [\Im \iz,1] \cap [n_2^{-5} \mathbb{N}]\}.\] 
We  introduce the random function
$$
v(z):=\max _{s \in L(z)} \Lambda(s)(n_2 \Im s)^{1 / 4}
$$
Our goal is to prove that with high probability, there is a gap in the range of $v$, i.e.
\begin{equation}\label{eq:gap_in_v}
\mathbb{P}\left(v(z) \leqslant n_2^{\varepsilon}, \quad  v(z)>n_2^{\varepsilon / 2}\right) \leqslant N^{-D+5}
\end{equation}
for all $z \in Q\cup Q_{\rm away}$ and large enough $n_2$. This equation says that with high probability the range of $v$ has a gap: it cannot take values in the interval $\left(n_2^{\varepsilon / 2}, n_2^{\varepsilon}\right]$.

 Since we are dealing with random variables, one has to keep track of the probabilities of exceptional events. To that end, we only work on a discrete set of values of $\eta$, which allows us to control the exceptional probabilities by a simple union bound.

Next, we prove Eq. \eqref{eq:gap_in_v}. First, Lemma \ref{lemma:bounded_W_off_diagonal} yields $\mathbb{I}(\Xi)(|\delta_1| + |\delta_2|)\prec \Psi_{\frakB} \leq (n_2\eta)^{-1/2}$.  Since $\left\{v(z) \leqslant n_2^{\varepsilon}\right\} \subset \Xi(z) \cap \Xi(s)$ for all $z \in Q\cup Q_{\rm away}$ and $s \in L(z)$, we find that for all $z \in Q\cup Q_{\rm away}$ and $s \in L(z)$ that
$$
\mathbb{P}\left(v(z) \leqslant n_2^{\varepsilon},~~ (|\delta_1(s)| + |\delta_2(s)|)(n_2 \Im s)^{1/2}> n_2^{\varepsilon / 2}\right) \leqslant n_2^{-D}
$$
for large enough $n_2$ (independent of $z$ and $w$ ). Using a union bound, we therefore get
$$
\mathbb{P}\left(v(z) \leqslant n_2^{\varepsilon},~~ \max _{s \in L(z)} (|\delta_1 |+ |\delta_2|)(n_2 \Im s)^{1/2}> n_2^{\varepsilon / 2}\right) \leqslant n_2^{-D+5}
$$
Next, applying Lemma \ref{lemma:stablity_phi} and Lemma \ref{lemma:stability_phi_away_case} with the perturbation  $n_2^{\varepsilon / 2}(n_2 \Im s)^{-1 / 2}$, we get
$$
\mathbb{P}\left(v(z) \leqslant n_2^{\varepsilon},~~ \max _{s \in L(z)} \frakB_p(s) (n_2 \Im s)^{1 / 4}> n_2^{\varepsilon / 4}\right) \leqslant n_2^{-D+5}.
$$
Moreover, since $- \bJ_{ii} - (1+(p/n_2)\phi_p)^{-1} = O_\prec(\frakB_p) + O_\prec(\Psi_\frakB)$, we have 
$$
\mathbb{P}\left(v(z) \leqslant n_2^{\varepsilon},~~ \max _{s \in L(z)} \frakB(s) (n_2 \Im s)^{1 / 4}>  2n_2^{\varepsilon / 4}\right) \leqslant n_2^{-D+5}.
$$

Then, Eq. \eqref{eq:gap_in_v} follows from the bound
\begin{align*}
\mathbb{P} &\left( v(z) \leqslant n_2^{\varepsilon}, v(s) \geq n_2^{\varepsilon/2} \right) \\
&\leq \mathbb{P} \left( v(z) \leqslant n_2^{\varepsilon}, \max_{s\in L(z)} \Lambda(s) (n_2\Im s)^{1/4} \geq n_2^{\varepsilon/2} \right) \\
&\leq \mathbb{P} \left( v(z) \leqslant n_2^{\varepsilon}, \max_{s\in L(z)} \Psi_{\frakB} (n_2\Im s)^{1/4} \geq n_2^{\varepsilon/2} \right)\\
&\leq \mathbb{P} \left( v(z) \leqslant n_2^{\varepsilon}, \max_{s\in L(z)} \frakB (n\Im s)^{-1} (n_2\Im s)^{1/2} \geq (1/2) n_2^{\varepsilon} \right) + n^{-D+5} \\
&\leq \mathbb{P} \left( v(z) \leqslant n_2^{\varepsilon}, \max_{s\in L(z)} \frakB (n_2\Im s)^{1/4} \geq n_2^{\varepsilon}(n_2\Im s)^{3/4}/2 \right)+ n^{-D+5}\\
&\leq \mathbb{P} \left( v(z) \leqslant n_2^{\varepsilon}, \max_{s\in L(z)} \frakB (n_2\Im s)^{1/4} \geq n_2^{\varepsilon}(n_2)^{3a'/4}/2 \right)+ n^{-D+5}\\
&\leq 2n_2^{-D+5}. 
\end{align*}

Here, we are using the fact that for all $s \in L(\iz)$
\[\frac{\Im \phi_p(s)}{n_2\Im s} (n_2\Im s)^{1/2}  = \frac{\Im \phi_p(s)}{\sqrt{n_2\Im s}} \prec 1.\] 
Therefore,
\[\mathbb{P} \left( v(z) \leqslant n_2^{\varepsilon}, \max_{s\in L(z)} \Im \phi_p(s) (n\Im s)^{-1} (n_2\Im s)^{1/2} \geq (1/2) n_2^{\varepsilon} \right)\leq n_2^{-D+5}. \]

We conclude the proof of Proposition \ref{proposition:weak_local_law} by combining Eq. \eqref{eq:gap_in_v} and Lemma \eqref{lemma:weak_local_eta_1} with a continuity argument. We choose a lattice $\Delta \subset Q\cup Q_{\rm away}$ such that $|\Delta| \leqslant n_2^{10}$ and for each $z \in Q\cup Q_{\rm away}$ there exists a $s \in \Delta$ satisfying $|z-s| \leqslant N^{-4}$. Then, Eq. \eqref{eq:gap_in_v} combined with a union bound yields
$$
\mathbb{P}\left(\exists s \in \Delta: v(s) \in\left(n_2^{\varepsilon / 2}, n_2^{\varepsilon}\right]\right) \leqslant n_2^{-D+15}
$$
From the definitions of $\Lambda$ and $Q\cup Q_{\rm away}$,  we find that $v$ is Lipschitz continuous, with Lipschitz constant $n_2^2$.
$$
\mathbb{P}\left(\exists z \in Q\cup Q_{\rm away}~:~ v(z) \in\left(2 n_2^{\varepsilon / 2}, n_2^{\varepsilon} \right]\right) \leqslant n_2^{-D+15}
$$
By Lemma \ref{lemma:weak_local_eta_1}, we have
$$
\mathbb{P}\left(v(z)>n_2^{\varepsilon / 2}\right) \leqslant n_2^{-D+15}
$$
for any $\iz \in Q\cup Q_{\rm away}$ with $\eta \geq1$. Therefore, 
\[\mathbb{P}(\max_{\iz \in Q\cup Q_{\rm away}}  v(z) > 2n_2^{\epsilon/2}  ) \leq 2n_2^{-D+15}.\]
It completes the proof of Proposition \ref{proposition:weak_local_law}.

\subsubsection{The averaged local law}\label{ssub:the_averaged_local_law}

We use Proposition \ref{proposition:weak_local_law} and improved estimates for the averaged quantities $[W]_p$ and $[W]_n$ to deduce the averaged local laws in $Q$, $Q_-$, and $Q_{\rm away}$,  as in Theorem \ref{thm:strong_local_law}. Indeed, we shall show that 
\[ \mathbb{I}(\Xi) [W]_p \prec \Psi_\frakB^2 \quad \mbox{ and }\quad \mathbb{I}(\Xi)  [W]_n \prec \Psi_\frakB^2.\]
If so,  the estimate of $\delta_1$ and $\delta_2$ is improved to $\mathbb{I}(\Xi) |\delta_1| \prec \Psi_\frakB^2$ and $\mathbb{I}(\Xi) |\delta_2| \prec \Psi_\frakB^2$. 

\begin{lemma}[Fluctuation averaging]\label{lemma:fluctiation_averaging}
Suppose that the assumptions of Theorem \ref{thm:strong_local_law} hold. Suppose that $\Upsilon$ is a positive, $n_2$-dependent, deterministic function on $Q\cup Q_{\rm away}$ satisfying $n_2^{-1 / 2} \leqslant \Upsilon \leqslant n_2^{-c}$ for some constant $c>0$. Suppose moreover that $\Lambda \prec n_2^{-c}$ and $\Lambda_o \prec \Upsilon$ on $Q\cup Q_{\rm away}$. Then on $Q \cup Q_{\rm away}$ we have

$$
\frac{1}{p} \sum_{u =1}^p  \frac{1}{(\lambda/\sigma_{up}  -z -\frac{1}{n_2} \sum_{i=p+1}^{p+n_2} \bJ_{ii})^2} \left(1-\mathbb{E}_u\right) \frac{1}{\bJ_{uu}}=O_{\prec}\left(\Upsilon^2\right)
$$
and
$$
\frac{1}{n_2} \sum_{i = p+1}^{p+n_2} \left(1-\mathbb{E}_i\right) \frac{1}{\bJ_{ii}}=O_{\prec}\left(\Upsilon^2\right).
$$
\end{lemma}
Th results are an extension of Lemma 5.6 of \citet{knowles2017anisotropic}, Lemma 4.9 of \citet{bloemendal2014isotropic}, Theorem 4.7 of \citet{erdos2013local}. They can be proved using exactly the same method. We therefore omit the details. The readers may refer to the proof of Theorem 4.7 of \citet{erdos2013local} for details.

Note that  Proposition \ref{proposition:weak_local_law} implies that $1- \mathbb{I}(\Xi) \prec 0$, that is $\Xi$ holds with high probability. In the following analysis, we shall always assume $\Xi$ holds. Combining with 
Lemma \ref{lemma:bounded_W_off_diagonal}, we have 
\[\Lambda \prec \frac{1}{(n_2\eta)^{1/4}}, \quad \Lambda_o \prec \Psi_{\frakB}, \]
on $Q \cup Q_{\rm away}$.  Therefore, the condition of Lemma \ref{lemma:fluctiation_averaging} is satisfied and we obtain
\[ |[W]_p | + |[W]_n| \prec \Psi_\frakB^2,\]
on $Q \cup Q_{\rm away}$. Recall that $\delta_1 = [W]_n + O_\prec(\Psi^2_\frakB)$ and $\delta_2 = [W]_p + O_\prec(\Psi^2_\frakB)$. We immediately get 
\[ (|\delta_1| + |\delta_2|) \prec \Psi_\frakB^2,\]
uniformly in $Q \cup Q_{\rm away}$.

The following lemma holds immediately from the definition of $\Psi_\frakB$.
\begin{lemma}
Let $c >0$ and suppose that we have $\frakB \prec  (n_2\eta)^{-c}$ uniformly in $z \in Q\cup Q_{\rm away}$. Then, we have 
\[  |\delta_1| + |\delta_2| \prec \frac{\Im \phi_p + (n_2\eta)^{-c}}{n_2 \eta}\]
uniformly in $z\in Q\cup Q_{\rm away}$. 
\end{lemma}

We first consider the averaged local law in $Q$.  Suppose $\frakB\prec (n_2\eta)^{-c}$. Then, uniformly in $\iz \in Q$, 
\[ |\delta_1| + |\delta_2| \prec \frac{\Im \phi_p + (n_2\eta)^{-c} }{n_2\eta} .\]
We invoke Lemma \ref{lemma:stablity_phi} to get
\begin{equation}\label{eq:precise_estimates_B}
\begin{aligned}
\frakB_p =& |\hat{\phi}_p  - \phi_p| \prec \frac{\displaystyle\frac{\Im \phi_p + (n_2\eta)^{-c}}{n_2\eta}}{\sqrt{\kappa+\eta} + \sqrt{\displaystyle\frac{\Im \phi_p + (n_2\eta)^{-c}}{n_2\eta}}}\leq  \frac{\Im \phi_p}{n_2\eta} \frac{1}{\sqrt{\kappa + \eta}} + \sqrt{\frac{(n_2\eta)^{-c}}{n_2\eta}} \\
& \leq \frac{C}{n_2\eta} + (n_2\eta)^{-1/2-c/2} \prec (n_2\eta)^{-1/2-c/2}.
\end{aligned}
\end{equation} 
In the fourth step, we are using the expression of $\Im \phi_p$ as shown in Lemma \ref{lemma:expression_imaginary_phi_h}.  As for $\frakB_n$, plugging the estimate of $\hat{\phi}_p$ back to Eq. \eqref{eq:expression_bJ_22_trace}, we also have 
\[\frakB_n \prec \frakB_p + \Psi_{\frakB}^2 \prec (n_2\eta)^{-1/2-c/2} + \frac{\Im \phi_p + (n_2\eta)^{-c} }{n_2\eta} \prec (n_2\eta)^{-1/2-c/2} . \] 
All together, we have shown that 
\begin{equation}\label{eq:iteration_frakB}
\frakB\prec (n_2\eta)^{-c} \quad  \quad \Longrightarrow \quad \quad   \frakB \prec (n_2\eta)^{-1/2-c/2}.
\end{equation}
From Proposition \ref{proposition:weak_local_law}, we know that $\frakB\leq C\Lambda \prec (n_2\eta)^{-1/4}$. Therefore, for any $\varepsilon >0$, we iterate Eq. \eqref{eq:iteration_frakB} finite number of times to get,
\[ \frakB \prec (n_2\eta)^{-1+\varepsilon}.\]
 It completes the proof of the averaged local law in $Q$ in Theorem \ref{thm:strong_local_law} under the additional assumption that $\Sigma_p$ is diagonal.

Next, we consider the averaged local law in $Q_-$. The follows exactly from the same arguments but with the specific form of $\Im \phi_p$ when $ E \leq \rho - n^{-2/3+a'}$. When $E$ is restricted to $E\leq \rho - n_2^{-2/3+a'}$ and $\eta \geq n_2^{-2/3}$, we have 
\[ \Im \phi_p \asymp \frac{\eta}{\sqrt{\kappa+\eta}},\]
where $\kappa = \rho -E$. Assume that $\frakB\prec (n_2\eta)^{-c}$. When $\iz \in Q_-$, Eq. \eqref{eq:precise_estimates_B} can be improved to 
\begin{align*} 
\frakB_p &\prec \frac{\Im\phi_p}{n_2\eta \sqrt{\kappa+\eta}} + \frac{1}{(n_2\eta)^{c}} \frac{1}{n_2\eta \sqrt{\kappa+\eta}} = \frac{1}{n_2(\kappa+\eta)}+ \frac{1}{(n_2\eta)^{c}} \frac{1}{n_2\eta \sqrt{\kappa+\eta}} \\ 
&\prec \frac{1}{n_2(\kappa+\eta)} + \frac{1}{(n_2\eta)^{c+ a'/2}  },
\end{align*}

 plugging the estimate of $\hat{\phi}_p$ back to Eq. \eqref{eq:expression_bJ_22_trace}, we also have 
\[\frakB_n \prec \frakB_p + \frac{\Im \phi_p + (n_2\eta)^{-c}}{n_2 \eta} \prec  \frac{1}{n_2(\kappa+\eta)} + \frac{1}{(n_2\eta)^{c+ a'/2}}.\]
All together, we get 
\[  \frakB\prec \frac{1}{(n_2\eta)^c}  \quad \quad \Longrightarrow \quad \quad \frakB \prec \frac{1}{n_2(\kappa+\eta)} + \frac{1}{(n_2\eta)^{c+ a'/2}}.\] 
Again, from Proposition \ref{proposition:weak_local_law}, $\frakB \prec (n_2\eta)^{-1/4}$. We iterate the equation finite number of times to get,
\[ \frakB \prec \frac{1}{n_2(\kappa+\eta)},\]
uniformly on $Q_-$. It completes the proof of the averaged local law in $Q_-$ in Theorem \ref{thm:strong_local_law} under the additional assumption that $\Sigma_p$ is diagonal.

Lastly, we consider the averaged local law in $Q_{\rm away}$. It follows from similar arguments to the previous two cases. On $Q_{\rm away}$, we invoke Lemma \ref{lemma:stability_phi_away_case} to control $|\hat{\phi}_p - \phi_p|$. Assuming  $\frakB \prec (n_2\eta)^{-c}$, then 
\[ \frakB_p =   |\hat{\phi}_p - \phi_p| \leq C |\delta_1| + C|\delta_2| \prec \frac{\Im \phi_p + (n_2\eta)^{-c}}{n_2\eta} \prec \frac{1}{n_2} + (n_2\eta)^{-c-1}  \prec \frac{1}{n_2}  + (n_2\eta)^{-c-1}.  \]
Here, we are using the fact that $\Im \phi_p/\eta \asymp 1$ when $\iz$ is away from the support of $\calG_{p\lambda}$. The rest arguments are very similar to those in the previous two cases. We omit further details.

We completed the proof of the averaged local laws when $\Sigma_p$ is diagonal. The diagonal assumption will be removed in Section \ref{ssub:extension_0o_non_diagonal_sigma_p_}.

\subsubsection{The entrywise local law}\label{ssub:the_entrywise_local_law}
It remains to prove the entrywise local law. Using Eq. \eqref{eq:self_consistent1}, Result (v) of Lemma \ref{lemma:resolvent_identity}, Lemma \ref{lemma:bounded_W_off_diagonal} and the averaged local law, we obtain 
\[ \bJ_{uu} - \frac{1}{\lambda/\sigma_{up} - \iz +{\phi}_p} \prec \Psi_{\frakB} \prec \Psi,\]
uniformly for $u=1,\dots, p$ and $\iz\in Q$. 
Similarly, using Eq. \eqref{eq:self_consistent2},
\[ \bJ_{ii} + \frac{1}{1+ (p/n_2) \phi_p} \prec \Psi,\]
uniformly for $i=p+1,\dots, p+n_2$ oand $\iz\in Q$. 
Moreover, since $\Lambda_o \prec \Psi_{\frakB} \prec \Psi$, the entrywise local law follows. 

\subsubsection{Extension to non-diagonal $\Sigma_p$}\label{ssub:extension_0o_non_diagonal_sigma_p_}

While the previous analysis assumes $\Sigma_p$ is diagonal, in this section, we show that Theorem \ref{thm:strong_local_law} and Theorem \ref{thm:rigidity_G} hold when the population covariance matrix is non-diagonal. In order to utilize the notation in the previous sections, we denote the population covariance as $\tilde{\Sigma}_p$ and continue to use $\Sigma_p$ to denote the diagonal matrix containing the eigenvalues of $\tilde{\Sigma}_p$. The eigendecomposition of $\tilde{\Sigma}_p$ is denoted as 
\[ \tilde{\Sigma}_p = H \Sigma_p H^T,\]
where $H$ is the eigenvector matrix of $\tilde{\Sigma}_p$. In this section, we show that if Theorem \ref{thm:strong_local_law} holds with $\Sigma_p$, it also holds with $\tilde{\Sigma}_p$. All definitions in Section \ref{ssub:basic_tools} are kept. 

Define  
\[ \tilde{\bK} = \left(\begin{matrix} H & 0 \\ 0 & I_{n_2} \end{matrix}\right) \bK \left(\begin{matrix} H^T & 0 \\ 0 & I_{n_2} \end{matrix}\right),\]
\[ \tilde{\bJ} = \left(\begin{matrix} H & 0 \\ 0 & I_{n_2} \end{matrix}\right) \bJ \left(\begin{matrix} H^T & 0 \\ 0 & I_{n_2} \end{matrix}\right),\]
\[ \tilde{\Omega} = \left(\begin{matrix} H & 0 \\ 0 & I_{n_2} \end{matrix}\right) \Omega \left(\begin{matrix} H^T & 0 \\ 0 & I_{n_2} \end{matrix}\right).\]
Then, $\tilde{\bK}$, $\tilde{\bK}$ and $\tilde{\Omega}$ are equal to $\bK$, $\bJ$ and $\Omega$ when $\Sigma_p$ is replaced by $\tilde{\Sigma}_p$.

The averaged local laws in $Q$, $Q_-$ and $Q_{\rm away}$  still hold when $\bJ$ and $\Omega$ are replaced by $\tilde{\bJ}$ and $\tilde{\Omega}$, since 
\[\frac{1}{p} \sum_{u=1}^{p}\tilde{\bJ}_{uu} = \frac{1}{p}\tr(\tilde{\bJ}_{(11)}) = \frac{1}{p}\tr (\bJ_{(11)} H^TH) = \hat{\phi}_p.\] 
It remains to show the entrywise local law. We shall actually show 
\[ H (\bJ_{(11)} - \Omega_{(11)}) H^T = O_\prec(\Psi),\]
uniformly in $Q$. In components, this reads
\[\left| \sum_{u,v=1}^{p} H_{iu} \left(\bJ_{uv} - \mathbb{I}(u=v)\frac{1}{\lambda/\sigma_{up} - h_p } \right) H_{vj} \right| \prec \Psi, \]
for all $i,j \in \{1,2, \dots, p\}$. 

The proof is based on the polynomialization method developed in Section 5 of \citet{bloemendal2014isotropic}.  The argument is very similar to that of \citet{bloemendal2014isotropic}, and we only outline the differences.

By the assumption 
\[\bJ_{(11)}-  (\lambda \Sigma_p^{-1} - h_p I_p)^{-1} =   O_\prec (\Psi)\] 
and orthogonality of $H$, we have
\begin{align*}
&\sum_{u,v=1}^{p} H_{iu} \left(\bJ_{uv} - \mathbb{I}(u=v)\frac{1}{\lambda/\sigma_{up} - h_p } \right) H_{vj} \\
&= \sum_{u=1}^{p} H_{iu} \left(\bJ_{uu} - \frac{1}{\lambda/\sigma_{up} - h_p } \right) H_{uj}     +\sum_{u \neq v} H_{iu} \bJ_{uv} H_{vj}=O_{\prec}(\Psi)+\mathcal{Z}
\end{align*}
where we defined $\mathcal{Z}:=\sum_{u \neq v}  H_{i u} \bJ_{uv} H_{vj}$. We need to prove that $|\mathcal{Z}| \prec \Psi$, which, following Section 5 of \citet{bloemendal2014isotropic} we do by estimating the moment $\mathbb{E}|\mathcal{Z}|^k$ for fixed $k \in 2 \mathbb{N}$. The argument from Section 5 of \citet{bloemendal2014isotropic} may be taken over with minor changes.

Use Result (ii) and (v) of Lemma \ref{lemma:resolvent_identity}, 
$$
\sum_{u \neq v \notin S}  H_{i u} \bJ_{uv}^{(S)} H_{vj} = \sum_{u \neq v \notin S}  H_{iu} H_{vj} \bJ_{uu}^{(S)} \bJ_{vv}^{(S \cup \{u\})} \left(\tbZ_{u\cdot} \bJ^{(S\cup\{uv\})}_{(22)} \tbZ^T_{v\cdot}\right)
$$
for any $S\subset \{1,2,\dots, p\}$ and
\begin{align*}
\bJ_{uu}^{(S)}&= \frac{1}{\lambda/\sigma_{up} - \iz - \tbZ_{u\cdot} \bJ_{(22)}^{(S\cup\{u\})}\tbZ_{u\cdot}^T} \\
& = \sum_{\ell=0}^{L-1}  \frac{1}{ (\lambda/\sigma_{up}  -\iz + [1+(p/n_2)\phi_p]^{-1})^{\ell+1} } \left(\left( \tbZ_{u\cdot} \bJ^{(S\cup\{u\})}_{(22)} \tbZ_{u\cdot}^T\right)-\frac{1}{ 1+(p/n_2)\phi_p}\right)^{\ell}\\ 
&\phantom{sssss} +O_{\prec}\left(\Psi^L\right).
\end{align*} 
We omit further details.

\subsection{Proof of Theorem \ref{thm:rigidity_G}}\label{sub:proof_of_theorem_ref_thm_rigidity_g}

The rigidity of the smallest eigenvalue of $\bG_\lambda$ can be proved using the averaged local law in $Q_-$ as in Theorem \ref{thm:strong_local_law}. In the following, $\epsilon \in (0, 2/3)$ is fixed. Since $\bG_\lambda = \tbZ\tbZ^T + \lambda \Sigma_p^{-1}$ and $\tbZ\tbZ^T$ is non-negative definite, all eigenvalues of $\bG_{\lambda}$ are therefore at least $\lambda/\ell_{\max}(\Sigma_p)$.
It then suffices to show that there is no eigenvalue in the interval 
\[[\lambda/\ell_{\max}(\Sigma_p), \rho_p - n_2^{-2/3 + \epsilon}).\]

Recall the definition of $Q(a,a',n_2)$ and $Q_-(a,a',n_2)$ in Section \ref{subsec:additional_properties_of_the_generalized_mar_v_c}.  We consider $a$ and $a'$ to be sufficiently small so that 
\[\Big\{\iz ~:~ E \in [\lambda/\ell_{\max}(\Sigma_p),~ \rho_p - n_2^{-2/3 + \epsilon}],~~  n_2^{-2/3}\leq  \eta \leq 1\Big\} \subset Q_-.\]
Consider the set 
\[Q^* = Q_-\cap \Big\{ \iz \mid \eta =  n_2^{-1/2 - \epsilon/4} |E- \rho|^{1/4} \Big\}.\] 

To show that  there is no eigenvalue in $[\lambda/\ell_{\max}(\Sigma_p), \rho -n_2^{-2/3+\epsilon})$, we only  need to show that uniformly in $Q^*$
\[ \Im \hat{\phi}_p(\iz) \prec n_2^{-\epsilon/2} (n_2\eta)^{-1}. \]
It is because that $\hat{\phi}_p$ is the Stieltjes transform of $\bG_{\lambda}$. When $\iz = \ell(\bG_\lambda) +\i\eta$ where $\ell(\bG_\lambda)$ is an eigenvalue of $\bG_{\lambda}$, 
\[\Im\hat{\phi}_p (\iz) \asymp \frac{1}{n_2\eta}.\]

Recall that in Lemma \ref{lemma:expression_imaginary_phi_h} for $\iz \in Q^*$, 
\[\Im \phi_p \asymp \frac{\eta}{\sqrt{|E-\rho| + \eta} }. \]
We conclude that it suffices to prove 
\[ |\hat{\phi}_p - \phi_p |\prec  \frac{1}{n_2(|E-\rho| + \eta)},\]
which is ensured by the averaged local law in $Q_-$ in Theorem \ref{thm:strong_local_law}.

\subsection{Proof of Theorem \ref{thm:analytical_function_1n_convergence}}

Theorem \ref{thm:rigidity_G} indicates that the smallest eigenvalue of $\bG_\lambda$ is larger than $\rho -\epsilon$ with high probability. Lemma \ref{lemma:bounds_on_ZZ} indicates that the largest eigenvalue of $\bG_\lambda$ is smaller than 
\[x_{\rm right} = (1+\sqrt{\hat{\gamma}_2})^2 + \lambda/ \liminf \ell_{\min}(\Sigma_p) + \epsilon,\]
with high probability.  Therefore, the statement in (i) of Theorem \ref{thm:analytical_function_1n_convergence} holds. 

The convergence of $p^{-1} \sum_{j=1}^p f(\ell_j(\bG_\lambda))$  can be proved using the strategy of \citet{bai2004clt}. We consider a contour in the complex plane that encloses the interval 
\[ [\rho - \epsilon, ~(1+\sqrt{\hat{\gamma}_2})^2 + \lambda/ \liminf \ell_{\min}(\Sigma_p) + \epsilon].\]
Without loss of generality, we choose the rectangle, denoted by $\mathcal{R}$, with the four vertices 
\[ \iz_1 = \rho - 2\epsilon + \epsilon'\i, ~~\iz_2 = \rho - 2\epsilon - \epsilon'\i, ~~\iz_3 = x_{\rm right} + \epsilon + \epsilon'\i, ~~ \iz_4 = x_{\rm right} + \epsilon - \epsilon'\i.\]
In the following, we choose $\epsilon$ and $\epsilon'$ to be sufficiently small so that the function $f$ is analytic on the rectangle. 

Clearly, with high probability, the rectangle encloses all eigenvalues of $\bG_\lambda$. Notably, if it is, we have
\[\frac{1}{p}\sum_{j=1}^p f(\ell_j(\bG_\lambda)) =  \int f(\tau) dF^{\bG_\lambda}(\tau)  =  \frac{1}{2\pi \i} \oint_{\mathcal{R}} f(\iz) \hat{\phi}_p(\iz)d \iz.\]
Here, the contour integral is taken over the rectangle in the positive direction in the complex plane. In the integral, we extend the definition of $\hat{\phi}_p$ to $\mathbb{C}^-$ by setting $\hat{\phi}_p(\iz) = \overline{\hat{\phi}_p(\overline{\iz})}$ if $\eta(\iz) <0$. 
To show the convergence of  the linear spectral statistics, we only need to show 
\[ \mathbb{I}(\frakM)\oint_{\mathcal{R}}  |\hat{\phi}_p(z) - \phi_p(z)| |dz| \prec \frac{1}{n_2}.\]

For any $c \in (0,1)$, we define 
\[\hat{\mathcal{R}}(c) = \mathcal{R} \cap \{\iz : \eta \geq n_2^{-1+c}\}.\] 
Then, 
\[  \mathbb{I}(\frakM) \oint_{\mathcal{R} \setminus \hat{\mathcal{R}}(c)}  |\hat{\phi}_p(\iz) | |d\iz| \leq C n_2^{-1+c}.\]
Here, we are using the fact that when $\frakM$ holds, a deterministic bound on $|\hat{\phi}_p(z)|$ holds uniformly for $z\in \mathcal{R}\setminus\hat{\mathcal{R}}(c)$. 

Therefore, to show 
\[ \mathbb{I}(\frakM) \left| \frac{1}{p}\sum_{j=1}^p f(\ell_j(\bG_\lambda))   -  \int f(\tau) d\calG_{p\lambda}(\tau)\right| \prec \frac{1}{n_2}.\]
We only need to show that for an arbitrary $c \in (0,1)$,
\[ \mathbb{I}(\frakM) \oint_{\hat{\mathcal{R}}(c) } |\hat{\phi}_p(z) - \phi_p(z)| |dz| \prec \frac{1}{n_2},\]It is a direct consequence of the averaged local law on $Q_{\rm away}$ in Theorem \ref{thm:strong_local_law}. It completes the proof.




\subsection{Proof of Theorem \ref{thm:G_lambda_extension_nonGauss}} \label{subsec:local_law_G_nonGaussian}

Although the extension of the local laws to non-Gaussian settings is technically more demanding, a general framework for this purpose has been developed in \citet{knowles2017anisotropic}. Theorem~\ref{thm:G_lambda_extension_nonGauss} can be established by following the arguments in Sections~7--9 of \citet{knowles2017anisotropic} in a largely verbatim manner. Accordingly, in this section, we only describe the similarities and differences, and highlight the modifications required to adapt the proofs in \citet{knowles2017anisotropic} to our setting.

In \citet{knowles2017anisotropic}, concentration of the resolvent matrix
\[
{\bJ}_{KY}(z) = {\bK}_{KY}^{-1}(z), 
\qquad 
{\bK}_{KY}(z) =
\begin{pmatrix}
-\Sigma^{-1} & \frac{1}{\sqrt{n_2}} \bZ \\[7pt]
\frac{1}{\sqrt{n_2}}\bZ^{T} & -z I_{n_2}
\end{pmatrix}
\]
around its deterministic equivalent
\[
{\Omega}_{KY}(z) =
\begin{pmatrix}
-\Sigma_p(1+m_p(z)\Sigma_p)^{-1} & 0\\
0 & m_p(z)I_{n_2}
\end{pmatrix}
\]
is established. Here, $m_p(z)$ is defined analogously to $q(z)$ in Section~\ref{subsec:additional_properties_of_the_mar_v_c}, with $\calG_{p\lambda}$ and $\hat{\gamma}_1$ replaced by $F^{\Sigma_p}$ and $\hat{\gamma}_2$, respectively. In particular, \citet{knowles2017anisotropic} prove the local law
\[
{\bJ}_{KY}(z) - {\Omega}_{KY}(z) = O_\prec\!\left({\Psi}_{KY}(z)\right),
\]
for $z$ in a neighborhood of the support edge of the measure associated with $m_p(z)$, where the control parameter is
\[
{\Psi}_{KY}(z) = \sqrt{\frac{\Im m_p(z)}{n_2 \eta}} + \frac{1}{n_2\eta}.
\]

For any deterministic unit vectors $v$ and $w$, define the inner product function
\[
F_{KY}(\bZ) = F_{KY}(z,v,w) = |v^T\!\left(\bJ_{KY}(z) - \Omega_{KY}(z)\right)w|,
\]
where the resolvent is constructed from the observations $\bZ$. It suffices to show that
\[
F_{KY}(\bZ) \prec \Psi_{KY}(z)
\]
whenever $z$ lies near the edge of the measure associated with $m_p(z)$, for any fixed sequences of unit vectors $v$ and $w$, and for any $\bZ$ satisfying Condition~\ref{enum:moments_conditions}. 

The argument proceeds by first establishing the result under Gaussian observations. Let $\bZ^0$ denote a matrix with i.i.d.\ $N(0,1)$ entries. It then suffices to control the difference
\[
\mE F_{KY}(\bZ) -\mE F_{KY}(\bZ^0),
\]
where $\bZ$ satisfies Condition~\ref{enum:moments_conditions}. To this end, \citet{knowles2017anisotropic} employ an interpolation method in which a family of distributions $\{f_\theta : \theta \in [0,1]\}$ is introduced such that $f_0$ is $N(0,1)$ and $f_1$ is the distribution of the entries of $\bZ$ (see Definition~7.6 of \citet{knowles2017anisotropic} for details). Let $\bZ^\theta$ denote the matrix whose entries follow the distribution $f_\theta$, and denote by $\bJ_{KY}(\bZ^\theta)$, $\bK_{KY}(\bZ^\theta)$, $F_{KY}(\bZ^{\theta})$ the corresponding resolvent, linearization matrix and the inner-product function. 
To control $\mE F_{KY}(\bZ^1) -\mE F_{KY}(\bZ^0)$, one analyzes the change in $\bJ_{KY}(\bZ^\theta)$ induced by replacing a single entry of $\bZ^{\theta}$. Specifically, let $\bZ_{(uv)}^{\theta,s}$ denote the matrix obtained from $\bZ^{\theta}$ by replacing its $(u,v)$-th entry with the value $s$. It turns out that, to control   $\mE F_{KY}(\bZ^1) -\mE F_{KY}(\bZ^0)$, we only need to control 
\[ \sum_{(u,v)} \mE F_{KY}^k (\bZ_{(uv)}^{\theta, z_{uv}^{1}}) - \mE F_{KY}^k (\bZ_{(uv)}^{\theta, z_{uv}^{0}}),\]
for arbitrary power $k\in2\mathbb{N}$, $\theta\in[0,1]$, and $(u,v)$. See Lemmas 7.9 and 7.10 of \cite{knowles2017anisotropic} for details. The analysis relies on explicitly quantifying the difference in $\bJ_{KY}(\bZ^{\theta})$ induced by the change  of the $(u,v)$-th entry of $\bZ^{\theta}$. Since $A^{-1} - B^{-1} = - A^{-1} (A-B) B^{-1}$, we have 
\[ \bJ_{KY}(\bZ_{(uv)}^{\theta, z_{uv}^{1}}) - \bJ_{KY}(\bZ_{(uv)}^{\theta, z_{uv}^{0}}) =  - \bJ_{KY}(\bZ_{(uv)}^{\theta, z_{uv}^{1}}) [\bK_{KY}(\bZ_{(uv)}^{\theta, z_{uv}^{1}})  - \bK_{KY}(\bZ_{(uv)}^{\theta, z_{uv}^{0}}) ] \bJ_{KY}(\bZ_{(uv)}^{\theta, z_{uv}^{0}}).\]
\begin{equation}\label{eq:diff_K_entry_change}
\bK_{KY}(\bZ_{(uv)}^{\theta, z_{uv}^{1}})  - \bK_{KY}(\bZ_{(uv)}^{\theta, z_{uv}^{0}})  = -(z_{uv}^{1} - z_{uv}^{0}) [ e_{(p+n_2)u} e_{n_2v}^T \Delta^T +\Delta e_{n_2v} e_{(p+n_2)u}^T].
\end{equation}
Here, $e_{pi}$ is the standard unit vector of dimension $p$ in the coordinate direction $i$ and 
\[\Delta = (0_{n_2\times p},~~~~ n_2^{-1/2} I_{n_2})^T.\] 

Arguments in Sections 7--9 of \cite{knowles2017anisotropic} are mainly devoted to the study of the above quantities. 

\medskip
Under our context, we consider the concentration of the resolvent matrix 
\[ \bJ(z) = \bK^{-1}(z), \qquad \bK(z) =  \begin{pmatrix}
\lambda \Sigma^{-1}-  zI_p &  \bZ U_2 \\[7pt]
U_2^T\bZ^{T} & - I_{n_2}
\end{pmatrix} \]
to the deterministic equivalent 
\[ \Omega(z) = \left(\begin{matrix}\left(\lambda\Sigma_p^{-1} - h_p(\iz) I_p \right)^{-1}  & 0  \\  0 &  (\iz-h_p(\iz)) I_{n_2}  \end{matrix} \right),\]
when $z$ is near the edge of the support of the probability distribution associated with $\phi_p(z)$ (specifically, $z\in Q(a,a', n_2)$). 

Similarly, we define the inner product function
\[
F(\bZ) = F(z,v,w) = |v^T\!\left(\bJ(z) - \Omega(z)\right)w|,
\]
It suffices to show that
\[
F(\bZ) \prec \Psi(z)
\]
for any fixed sequences of unit vectors $v$ and $w$, and for any $\bZ$ satisfying Condition~\ref{enum:moments_conditions}.

Similar to the framework in \cite{knowles2017anisotropic}, the argument proceeds by first establishing the result under Gaussian observations, which is completed in Theorem \ref{thm:strong_local_law}. The results are then extended to non-Gaussianity by  controlling the difference
\[
\mE F(\bZ) -\mE F(\bZ^0).
\]
Following the framework in \cite{knowles2017anisotropic}, we only need to control 
\[ \sum_{uv} \mE F^k (\bZ_{(uv)}^{\theta, z_{uv}^{1}}) - \mE F^k (\bZ_{(uv)}^{\theta, z_{uv}^{0}}),\]
for arbitrary power $k\in2\mathbb{N}$, $\theta\in[0,1]$, and $(u,v)$. While $\bJ$ differs from $\bJ_{KY}$ in all four blocks, we have similar stochastic domination bounds on 
\[ \|\bJ (z) \| \quad \text{and} \quad \| \bJ (z) - \Omega(z) \|,\]
as that on $\|\bJ_{KY}(z)\|$ and $\|\bJ_{KY}(z) - \Omega_{KY}(z)\|$. Moreover, analogously to those for $\bJ_{KY}$, 
\[ \bJ(\bZ_{(uv)}^{\theta, z_{uv}^{1}}) - \bJ(\bZ_{(uv)}^{\theta, z_{uv}^{0}}) =  - \bJ(\bZ_{(uv)}^{\theta, z_{uv}^{1}}) [\bK(\bZ_{(uv)}^{\theta, z_{uv}^{1}})  - \bK(\bZ_{(uv)}^{\theta, z_{uv}^{0}}) ] \bJ(\bZ_{(uv)}^{\theta, z_{uv}^{0}}).\]
\begin{equation}\label{eq:diff_K_entry_change}
\bK(\bZ_{(uv)}^{\theta, z_{uv}^{1}})  - \bK(\bZ_{(uv)}^{\theta, z_{uv}^{0}})  = -(z_{uv}^{1} - z_{uv}^{0}) [ e_{(p+n_2)u} e_{n_0v}^T \Delta^T_{\rm new} +\Delta_{\rm new} e_{n_0v} e_{(p+n_2)u}^T],
\end{equation}
with the updated matrix
\[\Delta_{\rm new} = (0_{n_0\times p},~~~~ U_2)^T.\] 
Here, we only need to substitute $U_2$ for $I_{n_2}$. 

It is not hard to verify that all bounds in Sections 7--9 of \cite{knowles2017anisotropic} remain valid under this updated definition of $\Delta_{\rm new}$. Therefore, their arguments apply in a largely verbatim manner, with the only necessary change being the substitution of $\Delta_{\rm new}$. This change effectively corresponds to replacing the data matrix ``$X$'' in \cite{knowles2017anisotropic} by $\bZ U_2$. Similar substitution of $\Delta_{\rm new}$ is conducted in  \cite{han2016tracy} and \cite{han2018unified} (See (9.30) and (9.31) of \cite{han2018unified}). We omit further details.

\subsection{Discrete case}\label{subsec:discrete_case}
While the previous subsections, we focus on Case (a), namely when $\rho > \lambda/\sigma_{1p}$,  in this section, we consider the case when $\rho = \lambda/\sigma_{1p}$. 
We show analogous results as Theorem \ref{thm:strong_local_law},  Theorem \ref{thm:rigidity_G} and Theorem \ref{thm:analytical_function_1n_convergence}. 

\begin{theorem}\label{thm:rigidity_discrete_case}
Suppose that \ref{enum:high_dimensional_regime}--\ref{enum:regular_edge}. Assume that $\rho = \lambda/\sigma_{1p}$. Then, with probability $1$, the smallest eigenvalue of $\bG_{\lambda}$ is $\lambda/\ell_{\max}(\Sigma_p)$. 
\end{theorem}

\begin{theorem}\label{thm:strong_local_law_discrete_case}
Suppose the conditions in Theorem \ref{thm:rigidity_discrete_case} hold. Then, uniformly in $Q_{\rm away}$
\[ \hat{\phi}_p(\iz)  -{\phi}_p(\iz)  = O_\prec(\frac{1}{n_2 }).\]
\end{theorem}
However, it is worth noting that the local law in $Q(a,a', n_2)$ as in Theorem \ref{thm:strong_local_law} no longer holds, since $\Im \phi_p(z)$ diverges as $z \to \rho = \lambda/\sigma_{1p}$.

\begin{theorem}\label{thm:analytical_function_converg_discrete}
Suppose the conditions in Theorem \ref{thm:rigidity_discrete_case} hold.  The results in Theorem \ref{thm:analytical_function_1n_convergence} still hold. 

\end{theorem}

\begin{proof}[Proof of Theorem \ref{thm:rigidity_discrete_case}]
Under Definition \ref{def:regular_edge}, for all sufficiently large $p$, we have 
\[p F^{\Sigma_p}(\{ \ell_{\max} (\Sigma_p)\}) > n_2.\]
Note that $pF^{\Sigma_p}(\{ \ell_{\max} (\Sigma_p)\})$ is the rank of the eigen-subspace associated with the smallest eigenvalue of $\lambda \Sigma_p^{-1}$, that is $\lambda\ell^{-1}_{\max}(\Sigma_p)$. The rank of $\tbZ\tbZ^T$ is $n_2$ with probability $1$, since the condition implies that $p > n_2$. Therefore, using Weyl's inequality, we obtain that there are at least one eigenvalue of $\tbZ\tbZ^T$ less than or equal to $\lambda/\ell_{\max}(\Sigma_p)$. On the other hand, $\ell_{\min} (\tbZ\tbZ^T + \lambda\Sigma_p^{-1}) \geq \lambda/\ell_{\max}(\Sigma_p)$. Therefore, with probability 1, the smallest eigenvalue of $\bG_\lambda$  is $\lambda/\ell_{\max}(\Sigma_p)$.
\end{proof}

\begin{proof}[Proof of Theorem \ref{thm:strong_local_law_discrete_case}]
The proof of Theorem~\ref{thm:strong_local_law_discrete_case} follows the same arguments as those used in the proof of Theorem~\ref{thm:strong_local_law}. The only modification is to replace the bound in Lemma~\ref{lemma:stablity_phi} with the stronger bound provided in Lemma~\ref{lemma:stability_phi_away_case}. We therefore omit the details.
\end{proof}

\begin{proof}
The proof of Theorem~\ref{thm:analytical_function_converg_discrete} follows the same arguments as those used in the proof of Theorem~\ref{thm:analytical_function_1n_convergence}. The bound on
\[
|\hat{\phi}_p - \phi_p|
\]
in the final step is now ensured by Theorem~\ref{thm:strong_local_law_discrete_case}, rather than by Theorem~\ref{thm:strong_local_law}.
.  
\end{proof}

\newpage \clearpage

\section{Proof of Theorem~\ref{thm:main}} \label{sec:proof_theorem_main}

Recall the proof outline presented in the main text. We begin by establishing the result under the Gaussian assumption and subsequently extend it to the non-Gaussian setting. Step~1 has already been completed in Section~\ref{sec:properties_of_bg_lambda}. In this section, we establish the remaining steps.

\subsection{Step 2: Tracy--Widom Limit under Gaussianity}\label{subsec:step2}

Under Gaussianity, the matrices $\bZ U_1$ and $\bZ U_2$ are independent. It is therefore equivalent to consider a matrix of the form
\[
\tilde{\bZ}_1^T \bG_{\lambda}^{-1} \tilde{\bZ}_1, 
\qquad 
\bG_{\lambda} = \tilde{\bZ}_2 \tilde{\bZ}_2^T + \lambda \Sigma_p^{-1},
\]
where $\tilde{\bZ}_1$ is a $p \times n_1$ matrix with i.i.d.\ $N(0,1/n_1)$ entries, $\tilde{\bZ}_2$ is a $p \times n_2$ matrix with i.i.d.\ $N(0,1/n_2)$ entries, and $\tilde{\bZ}_1$ and $\tilde{\bZ}_2$ are independent.

The main strategy of the proof is as follows. Conditioning on $\bG_\lambda$, we apply the main theorem of \cite{lee2016tracy}, treating  $\bG_\lambda^{-1}$ as the population covariance matrix (corresponding to $\Sigma$ in their notation). To justify this application, it remains to verify that the regularity condition in Assumption~2.2 of \citet{knowles2017anisotropic} holds with high probability. 

For the reader’s convenience, we restate here the key quantities introduced in Section~\ref{subsec:additional_properties_of_the_mar_v_c}. Recall the function
\[ f(q) = -\frac{1}{q} + \hat{\gamma}_1 s(-q) = -\frac{1}{q} + \hat{\gamma}_1 \int \frac{d\calG_{p\lambda}(\tau)}{\tau+q}.\]
Then, $-\beta$ is the critical point of $f$ in $(-\rho, 0)$. Namely,
\[0 = f'(-\beta) = \frac{1}{\beta^2} - \hat{\gamma}_1 s'(\beta) = \frac{1}{\beta^2} - \hat{\gamma}_1 \int\frac{d\calG_{p\lambda}(\tau)}{(\tau - \beta)^2}.\]
Also, the centering parameter $\Theta_1 = \Theta_{1p}$ is 
\[ \Theta_1 = f(-\beta) = \frac{1}{\beta} +\hat{\gamma}_1 s(\beta) = \frac{1}{\beta} +  \hat{\gamma}_1 \int \frac{d\calG_{p\lambda}(\tau)}{\tau-\beta}.\]
Moreover, recall the definition of the scaling parameter $\Theta_{2} = \Theta_{2p}$ as 
\[ {\Theta}_{2}  =\left[ \big(\frac{p}{n_1}\big)^3 \int \frac{d\calG_{p\lambda}(\tau)}{(\tau - {\beta})^{3} }  + \frac{(p/n_1)^2}{{\beta}^3}  \right]^{1/3}. \]
The following local domain close to $\Theta_1$ is considered. 
\[ {D} = {D}(a,a',n_1) \coloneqq \{ \iz = E+\i \eta \in \mathbb{C}^+ ~:~ |E - {\Theta}_{1}| \leq a, ~~ n_1^{-1 + a'} \leq \eta \leq 1/a'\}.\]

Results in Section~\ref{sec:properties_of_bg_lambda} indicate that the LSS of $\bG_{\lambda}$ converges to their deterministic counterparts based on $\calG_{p\lambda}$ and the smallest eigenvalue of $\bG_{\lambda}$ converges to $\rho$. We replace $\calG_{p\lambda}$ with the stochastic ESD $F^{\bG_{\lambda}}$ and define the following counterparts of $f(\cdot)$, $\beta$, and $\Theta_1$. 

Denote the eigenvalues of $\bG_{\lambda}$ to be $g_1\geq g_2 \geq \cdots \geq g_p \geq \lambda/\sigma_{1p} >0$. Replacing $\calG_{p\lambda}$ with $F^{\bG_{\lambda}}$, we define  
\[ \hat{f}(q) = -\frac{1}{q} + \hat{\gamma}_1 \int \frac{dF^{\bG_{\lambda}}(\tau)}{\tau+q} = -\frac{1}{q} + \frac{1}{n_1} \sum_{j=1}^p \frac{1}{g_j+q}.\]
Then, there exists a unique solution to the equation $z = \hat{f}(q)$ for $z\in\mathbb{C}^+$, we denote it to be $\hat{q}(z)$. Namely,
\begin{equation}\label{eq:def_hat_q}
z  = -\frac{1}{\hat{q}} + \hat{\gamma}_1 \int \frac{dF^{\bG_{\lambda}}(\tau)}{\tau+ \hat{q} }.
\end{equation}
The unique critical point of $\hat{f}(q)$ in $(-g_p, 0)$ is denoted as $-\hat{\beta}$. Namely,
\[ 0= \hat{f}'(-\hat{\beta}) = \frac{1}{\hat\beta^2}  -  \frac{1}{n_1} \sum_{j=1}^p \frac{1}{(g_j-\hat{\beta})^2}.\]
It is straightforward to verify that $\hat{\beta}$ exists and is unique for any fixed sequence $g_j$'s. Let
\[\hat{\Theta}_1 = \hat{f}(-\hat\beta) = \frac{1}{\hat{\beta}} + \frac{1}{n_1} \sum_{j=1}^p \frac{1}{g_j -\hat{\beta}}. \]
Lastly, call 
\[ \hat{\Theta}_{2}  =\left[ ({\frac{p}{n_1})^3} \int \frac{dF^{\bG_\lambda}(\tau)}{(\tau - \hat{\beta})^{3} }  + \frac{(p/n_1)^2}{\hat{\beta}^3}  \right]^{1/3}. \]

\medskip
The following results are derived from those in Section \ref{sec:properties_of_bg_lambda}.
\begin{theorem}\label{thm:convergence_hatbeta_hatTheta}
    Suppose that Conditions \ref{enum:high_dimensional_regime}-- \ref{enum:regular_edge} hold. Additionally, suppose that the observations are Gaussian. 
    \begin{itemize}
        \item[(i)] There exists a universal constant $c>0$ such that $|\hat{\beta} -g_p|>c$ holds with high probability.
        \item[(ii)]  $| \hat{\Theta}_1 - \Theta_1 | \prec 1/n_2$ and $| \hat{\Theta}_2 - \Theta_2 | \prec 1/n_2$   
        \item[(iii)] For arbitrary $a'>0$, there exists $a>0$ such that 
        \[ |\hat{q}(z) - q(z)| \prec \frac{1}{n_1\eta},\qquad z \in D(a,a', n_1).\]
    \end{itemize}
\end{theorem}

Before presenting the proof of Theorem~\ref{thm:convergence_hatbeta_hatTheta}, we first explain how Theorem~\ref{thm:convergence_hatbeta_hatTheta} implies that Theorem~\ref{thm:main} holds under Gaussianity.

Fix $\bG_\lambda$ and suppose that $|\hat{\beta} - g_p| > c$ holds for some small constant $c$. Then, the condition in Assumption 2.2 of \cite{lee2016tracy} is satisfied. Then, applying the main theorem of \cite{lee2016tracy}, we obtain that, for any $t\in\mathbb{R}$,
\[ \lim_{p \to \infty}  \mP \left( \left\{\frac{p^{2/3}}{\hat{\Theta}_2} \Big(\ell_{\max}(\tilde{\bF}_\lambda) - \hat{\Theta}_1 \Big)  \leq t \right\} \cap \left\{|\hat{\beta}- g_p|>c \right\} \mid \bG_{\lambda}   \right) \longrightarrow \TW_1(t),\]
where $\TW_1(t)$ is the distribution function of the Tracy-Widom distribution of type 1. Now, integrating both sides with respect to the distribution of $\bG_{\lambda}$, using the dominated convergence theorem, we obtain that, for any $t\in\mathbb{R}$,
\[ \lim_{p \to \infty}  \mP \left( \left\{\frac{p^{2/3}}{\hat{\Theta}_2} \Big(\ell_{\max}(\tilde{\bF}_\lambda) - \hat{\Theta}_1 \Big)  \leq t \right\} \cap \left\{|\hat{\beta}- g_p|>c \right\}\right) \longrightarrow \TW_1(t).\]

Using the results of Theorem \ref{thm:convergence_hatbeta_hatTheta}, we can select the constant $c$ to be sufficiently small so that $\mP(|\hat{\beta} - g_p |>c) \to 1$. Also, by the convergence of $\hat{\Theta}_1$ and $\hat{\Theta}_2$ and Slutsky's theorem, the weak convergence still holds when $\hat{\Theta}_1$ and $\hat{\Theta}_2$ are replaced by $\Theta_1$ and $\Theta_2$. It completes the proof of Theorem \ref{thm:main} under Gaussianity.

\medskip
The rest of this subsection is devoted to the proof of Theorem \ref{thm:convergence_hatbeta_hatTheta}. 
\begin{proof}[Proof of Theorem \ref{thm:convergence_hatbeta_hatTheta}]
We separately consider two cases: $\rho > \lambda/\sigma_{1p}$ and $\rho = \lambda/\sigma_{1p}$.  For both cases, since $s'(x) \to \infty$ as $x\to \rho$, we have that $\beta <\rho-c$ for some universal constant $c>0$ and all sufficiently large $p$. We first consider the case when $\rho > \lambda/\sigma_{1p}$. The second case is rather straightforward, and we address it at the end. 

Theorem \ref{thm:rigidity_G} in Section \ref{sec:properties_of_bg_lambda} indicates that  
\[|g_p  - \rho | \prec n_2^{-2/3}.\]
In order to show $|\hat{\beta} - g_p| \geq c$ with high probability for some constant $c$, it suffices to show that 
\[|\hat{\beta} - \beta| \prec \frac{1}{n_2}.\]

Recall the functions $f(q)$ and $\hat{f}(q)$. We extend the functions to complex-valued $q$. In particular, for $q \in \mathbb{C}^+ \cup (-\rho, 0)$, we consider $f'(q)$, $f''(q)$, $\hat{f}'(q)$, and $\hat{f}''(q)$
\begin{align*}
f'(q) &= \frac{1}{q^2} - \hat{\gamma}_1 \int \frac{d\calG_{p\lambda}(\tau)}{(\tau + q)^2}, \\
f''(q) &= -\frac{2}{q^3} + 2\hat{\gamma}_1 \int \frac{d\calG_{p\lambda}(\tau)}{(\tau+q)^3},\\ 
\hat{f}'(q)& = \frac{1}{q^2} - \frac{1}{n_1} \sum_{j=1}^p \frac{1}{(g_p + q)^2},\\
\hat{f}''(q) &= -\frac{2}{q^3} + \frac{2}{n_1} \sum \frac{1}{ (g_j + q)^3}.
\end{align*}
Since $\beta$ is away from $\rho$, there exists a domain 
\[\frakD = \{ \iz ~:~  |E +\beta| \leq C, ~~0\leq \eta \leq C'\}, \quad \text{for some } C,C' >0\] 
such that  
\[\left| f''(q)\right| \asymp 1, \quad q \in \frakD.\]
In the following, we shall always select $C$ and $C'$ to be such that the condition is satisfied. 

\medskip
We notice that in $(-g_p, 0 )$, $\hat{f}''(q) >0$. Therefore, there exists at most one critical point of $\hat{f}(q)$ on $(-g_p, 0)$. To show that $|\hat{\beta} -\beta|\prec n_2^{-1}$, we only need to show that for arbitrary small $\epsilon>0$, there exists a root to $\hat{f}(q)=0$ in $(-\beta- n_1^{-1+\epsilon}, -\beta+ n_1^{-1+ \epsilon})$ with high probability.

Clearly, due to the boundedness of $f''(q)$ in $\frakD$, we can find constants $C_1>0$ and $C_2>0$  such that 
\[  C_1 n_2^{-1+\epsilon}  \leq \Re f'(-\beta + n_2^{-1 + \epsilon} + \i n_2^{-1 + \epsilon/2} )  \leq  C_2   n_2^{-1+\epsilon},\]
\[  - C_2  n_2^{-1+\epsilon}  \leq \Re f'(-\beta - n_2^{-1 + \epsilon} + \i n_2^{-1 + \epsilon/2} )  \leq  -C_1n_2^{-1+\epsilon}, \]
\[\Im f'(-\beta \pm n_2^{-1 + \epsilon} + \i n_2^{-1 + \epsilon/2} )  \asymp  n_2^{-1+\epsilon/2}.\]

For any fixed $q\in \frakD$,  we view 
\[ f(q) = -\frac{1}{q} + \hat{\gamma}_1 \int  \frac{1}{\tau +q} d\calG_{p\lambda}(\tau) = \int  (-1/q) \tau + \frac{\hat{\gamma}_1}{\tau +q } d\calG_{p\lambda}(\tau),\]
as a LSS with the analytical function being $(-1/q) \tau + \hat{\gamma}_1 (\tau+q)^{-1}$. The condition of  Theorem \ref{thm:analytical_function_1n_convergence} is satisfied. 
Applying Theorem \ref{thm:analytical_function_1n_convergence}, we obtain that 
\[ \mathbb{I}(g_p \geq (\rho+ \beta)/2)\left|\hat{f}(-\beta \pm n_2^{-1+\epsilon} + \i n_2^{-1 + \epsilon/2}) - f(-\beta \pm n_2^{-1+\epsilon} + \i n_2^{-1 + \epsilon/2})\right| \prec \frac{1}{n_2}.\]
Here, we also use the uniformity of the stochastic domination across all LSS $f(q)$ when $q\in \frakD$. It is easy to verify by checking the arguments in the proof of Theorem \ref{thm:analytical_function_1n_convergence}.  Indeed, the domination is uniform for all $\{f_q, q \in \frakD\}$ if a universal bound on $\sup_q |1/q|$ and $\sup_q\sup_{\tau \in (\rho-a, \infty)}(\tau+q)^{-1}$ when $q\in \frakD$ exists for some sufficiently small $a>0$.  

Moreover, since when $g_p \geq (\rho+\beta/2)$ and $q\in \frakD$, $|\hat{f}''(q)| \asymp 1$.
We conclude that there exists constants $C_1>0$ and $C_2 >0$ such that with high probability
\[ C_1 n_2^{-1+\epsilon} \leq \hat{f}'(-\beta + n_2^{-1 + \epsilon}) \leq C_2 n_2^{-1+\epsilon},\]
\[ -C_2 n_2^{-1+\epsilon} \leq \hat{f}'(-\beta - n_2^{-1 + \epsilon}) \leq -C_1 n_2^{-1+\epsilon}.\]
It implies that there exists a critical point of $\hat{f}'(q)$ in between $-\beta \pm n_2^{-1+\epsilon}$ with high probability. It completes the proof of Result (i).

\medskip

To show Result (ii),  first observe that when $g_p \geq (\rho+\beta)/2$, 
\[ \hat{\Theta}_{1} - \Theta_1  =  \hat{f}(-\hat{\beta}) - f(-\beta) = (\hat{f}(-\hat{\beta}) - \hat{f}(-\beta)) + (\hat{f}(-\beta) - f(-\beta)).\]
The first term is such that 
\[ \hat{f}(-\hat{\beta}) - \hat{f}(-\beta) = O_{\prec} (\hat{\beta} - \beta)  = O_{\prec}(\frac{1}{n_2}). \]
Here, we are using the fact that $\hat{f}'(q)$ is bounded uniformly when $q\in \frakD$ with high probability. For the second term, we again treat $\hat{f}(-\beta)$ as a LSS of $\bG_{\lambda}$. Using Theorem \ref{thm:analytical_function_1n_convergence},
\[\mathbb{I}( g_p \geq (\rho+\beta)/2) |\hat{f}(-\beta) - f(-\beta)| \prec n_2^{-1}.
\]
It follows that $|\hat{\Theta}_{1} - \Theta_1|  \prec 1/n_2$. 

Moreover, we can verify that $| \hat{\Theta}_2 - \Theta_2 | \prec 1/n_2$ following similar results. We omit details.

\medskip

Lastly, we show Result (iii).
Recall that 
\[ \iz = -\frac{1}{\hat{q}(\iz)} + (p/n_1) \int \frac{dF^{\bG_\lambda}(\tau)}{\tau + \hat{q}(\iz)}.\]
Fix $\iz_0 = E+ \i\eta$ such that $\eta\geq 1$. Applying Lemma \ref{lemma:properties_q_z}, we obtain that 
\[ \Im \hat{q}(\iz_0) \geq C,\]
for some universal constant $C>0$ independent of $\bG_\lambda$. 
Therefore, $-\hat{q}(\iz_0)$ is away from the support of $\calG_{p\lambda}$. Applying the local laws of $\bG_{\lambda}$ away from the support as in Theorem \ref{thm:strong_local_law}, we obtain that 
\[ \left|\int \frac{dF^{\bG_\lambda}(\tau)}{\tau + \hat{q}(\iz_0)} - \int\frac{ d\calG_{p\lambda}(\tau) }{\tau + \hat{q}(\iz_0) }   \right|\prec \frac{1}{n_2}.\]
It follows then
\[ \iz_0 + O_\prec(n_1^{-1}) =  -\frac{1}{\hat{q}(\iz_0)} + (p/n_1) \int \frac{d\calG_{p\lambda}(\tau)}{\tau + \hat{q}(\iz_0)}.\]
Due to the strong stability of the Mar\v{c}enko-Pastur equation (Lemma \ref{lemma:stability_q_z}), we obtain 
\[ |\hat{q}(\iz) - q(\iz)| \prec \frac{1}{n_2}.\]

To propagate the convergence from $\eta\geq 1$ to all  $z\in D$, we then use a stochastic continuity argument and Lemma \ref{lemma:stability_q_z} to propagate the smallness of $|\hat{q}_p(\iz) - q(\iz)|$. The argument is very similar to those in the proof of Proposition \ref{proposition:weak_local_law}. We therefore omit the details. All together, we conclude that 
\[|\hat{q}(\iz) - q(\iz)| \prec \frac{1}{n_2\eta}. \]

\medskip

As for the case when $\rho = \lambda/\sigma_{1p}$, from Theorem \ref{thm:rigidity_discrete_case}, Result (i) clearly holds. Results (ii) and (iii) follow from the same arguments as in the case of $\rho = \lambda/\sigma_{1p}$. We only need to replace the arguments in Theorem \ref{thm:analytical_function_1n_convergence} with  that in Theorem \ref{thm:analytical_function_converg_discrete}. 

 \end{proof}

\subsection{Step 3: Green Function Comparison Theorem and Linearization Matrix}\label{subsec:step3}

While the asymptotic Tracy--Widom distribution under Gaussianity has been established in Steps~1 and~2, the goal of the present step is to extend the result to non-Gaussian distributions satisfying Condition~\ref{enum:moments_conditions}. This phenomenon is known as edge universality in the RMT literature, meaning that the asymptotic distribution of edge eigenvalues is typically unaffected by the specific distribution of the matrix entries, provided that suitable moment conditions are satisfied.

Edge universality is typically established via a Green function comparison theorem. Define the resolvent (Green function) of the matrix $\tilde{\bF}_{\lambda}$ by
\[
\bL(z) = (U_1^{T} \bZ^{T} \bG_{\lambda}^{-1} \bZ U_1 - z I_{n_1})^{-1}, 
\qquad z \in \mathbb{C}^+ .
\]
We consider two settings: one in which $\bZ$ has a general distribution satisfying Condition~\ref{enum:moments_conditions}, and one in which $\bZ$ is Gaussian. To distinguish the latter from the former, we denote the Gaussian matrix by $\bZ^0$. Accordingly, the resolvent corresponding to $\bZ$ is denoted by $\bL_{\bZ}$, while the resolvent corresponding to $\bZ^0$ is denoted by $\bL_{\bZ^0}$.
The averaged trace of $\bL_{\bZ}(z)$ and $\bL_{\bZ^0}$ is taken as 
\[ \overline{L}_{\bZ}(z) = \frac{1}{n_1} \tr (\bL_{\bZ}(z)) \quad \text{and} \quad \overline{L}_{\bZ^0}(z) = \frac{1}{n_1} \tr (\bL_{\bZ^0} (z)).\]
Moreover, the transformed F matrix corresponding to $\bZ^0$ is denoted by $\tilde{\bF}^0_\lambda$. 

\begin{theorem}[Green function comparison]\label{thm:green_function_comparison}
Let $\epsilon>0$ and set $\eta = p^{-2/3-\epsilon}$. Let $E_1, E_2 \in \mathbb{R}$ satisfy $E_1 < E_2$ and
\[
|E_1 - \Theta_1|,\; |E_2 - \Theta_1| \leq p^{-2/3+\epsilon}.
\]
Let $K:\mathbb{R}\rightarrow \mathbb{R}$ be a smooth function such that
\[
\max_x |K^{(l)}(x)| \leq C, \qquad l=1,2,3,4,5,
\]
for some constant $C>0$. Then there exists a constant $c>0$ such that, for all sufficiently large $p$ and sufficiently small $\epsilon$, 
\[
\left|
\mathbb{E} \, K\!\left(p \int_{E_1}^{E_2} \Im \overline{L}_{\bZ}(x+i \eta)\, dx\right)
-
\mathbb{E} \, K\!\left(p \int_{E_1}^{E_2} \Im \overline{L}_{\bZ^0}(x+i \eta)\, dx\right)
\right| 
\leq p^{-c}.
\]
\end{theorem}
The theorem is a standard result in RMT; See, for example, Theorem~6.3 of \citet{erdHos2012rigidity}.  To connect Theorem \ref{thm:green_function_comparison} to the edge universality of $\tilde{\bF}_\lambda$, we need the following results. 
\begin{lemma}[Rigidity of largest eigenvalue]\label{lemma:rough_order_largest_root}
    Under Conditions \ref{enum:high_dimensional_regime}--\ref{enum:regular_edge}, we have
    \[ \ell_{\max}(\tilde{\bF}_\lambda) - \Theta_1  = O_\prec(p^{-2/3}). \]
\end{lemma}
Given Lemma \ref{lemma:rough_order_largest_root}, we can locate to the domain of $E$ such that $|E-\Theta_1|\prec p^{-2/3}$. Fix $E^* \prec p^{-2/3}$ such that it suffices to consider $\ell_{\max}(\tilde{\bF}_\lambda) \leq \Theta_1 + E^*$. Choose $|E - \Theta_1| \prec p^{-2/3}$, $\eta  = p^{-2/3- 9\epsilon}$ and $l = \frac{1}{2} p^{-2/3-\epsilon}$. Then, for some sufficiently small constant $\epsilon>0$ and sufficiently large constant $M$, there exists a constant $p_0$ depending on $\epsilon$ and $M$, such that 
\begin{align}
& \mathbb{E} K\left(\frac{n_1}{\pi} \int_{E-l}^{\Theta_1+E^*} \Im \overline{L}_{\bZ}(x+i \eta) d x\right) \leq \mathbb{P}\left( \ell_{\max}(\tilde{\bF}_\lambda) \leq E\right)\nonumber\\
&\phantom{ssssss} \leq \mathbb{E} K\left(\frac{n_1}{\pi} \int_{E+l}^{\Theta_1+E^*} \Im \overline{L}_{\bZ}(x+i \eta) d x\right)+n_1^{-M}, \label{eq:bound_probability_by_ST_1}\\
&\mathbb{E} K\left(\frac{n_1}{\pi} \int_{E-l}^{\Theta_1+E^*} \Im \overline{L}_{\bZ^0}(x+i \eta) d x\right) \leq \mathbb{P}\left( \ell_{\max}(\tilde{\bF}^0_\lambda) \leq E\right) \nonumber\\
&\phantom{ssssss} \leq \mathbb{E} K\left(\frac{n_1}{\pi} \int_{E+l}^{\Theta_1+E^*} \Im \overline{L}_{\bZ^0}(x+i \eta) d x\right)+n_1^{-M}, \label{eq:bound_probability_by_ST_2}
\end{align}
whenever $p>p_0$ and $K$ is a smooth cutoff function satisfying the condition in Theorem \ref{thm:green_function_comparison}. We omit the details in the function $K$ and the proof of Eq. \eqref{eq:bound_probability_by_ST_1} and Eq. \eqref{eq:bound_probability_by_ST_2}  because it is a standard procedure in RMT. Similar works include Corollary 5.1 of   Lemma 6.1 of \citet{erdHos2012rigidity},  \cite{BaoPanZhou2013LocalEdge}, and Lemma 4.1 of \cite{PillaiYin2014}. 

Following Eq. \eqref{eq:bound_probability_by_ST_1} and Eq. \eqref{eq:bound_probability_by_ST_2},  setting $E$ in Eq. \eqref{eq:bound_probability_by_ST_2} as $E\pm 2l$,  we immediately get 
\[ \mP(\ell_{\max}(\tilde{\bF}_\lambda) \leq E) \geq \mP(\ell_{\max}(\tilde{\bF}^0_\lambda) \leq E -2l) - n_1^{-M},\]
\[ \mP(\ell_{\max}(\tilde{\bF}_\lambda) \leq E) \leq \mP(\ell_{\max}(\tilde{\bF}^0_\lambda) \leq E + 2l) + n_1^{-M}.\]
Since $l = O(p^{-2/3 -\epsilon})$, we have then
\[  \mP(\ell_{\max}(\tilde{\bF}^0_\lambda) \leq E \pm 2l) - \mP(\ell_{\max}(\tilde{\bF}^0_\lambda) \leq E ) \to 0.\]
We then conclude that 
\[ \mP(\ell_{\max}(\tilde{\bF}_\lambda) \leq E)  - \mP (\ell_{\max}(\tilde{\bF}^0_\lambda) \leq  E) \to 0.\]
Since we already established the asymptotic Tracy-Widom distribution of $\ell_{\max}(\tilde{\bF}^0_\lambda)$ as stated in Theorem \ref{thm:main}, $\ell_{\max}(\tilde{\bF}_\lambda)$ has the same weak limit. It will complete the proof of Theorem \ref{thm:main}.

\medskip
It remains to prove Lemma~\ref{lemma:rough_order_largest_root} and Theorem~\ref{thm:green_function_comparison}. To this end, we analyze the behavior of the resolvent $\bL(z)$ in a neighborhood of $\Theta_1$. The objective is to show that the resolvent $\bL_{\bZ}(z)$ has the same asymptotic limit as that of $\bL_{\bZ^0}(z)$. However, it is difficult to work directly with $\bL(z)$ due to its complicated dependence on the individual entries of $\bZ$.

Following the strategy of \citet{han2016tracy} and \citet{han2018unified}, we instead express the resolvent $\bL(z)$ as a submatrix of the inverse of a linearization matrix.
In particular, define the following $3\times 3$ block matrix 
\begin{equation}\label{eq:bH}
\bH_\lambda(z) = \bH_\lambda(z,\bZ) =\left( \begin{matrix} 
- z I_{n_1} &  U_1^T \bZ^T & 0 \\[10pt]
 \bZ U_1 & -\lambda \Sigma_p^{-1} & \bZ U_2 \\[10pt]
0 & U_2^T \bZ^T & I_{n_2} 
\end{matrix} \right).
\end{equation}
Note that $\bH_\lambda(z)$ is linear in $\bZ$. Using the Schur complement formula, it can be verified easily that the upper-left block of the matrix $\bH^{-1}_\lambda(z)$ is the resolvent $\bL(z)$. Therefore, it suffices to study the properties of $\bH^{-1}_\lambda(\iz)$. This linearization technique is widely used in RMT for studying the universality of matrices with complex structure.

It is important to note that when $\lambda = 0$, the matrix $\bH_0(z)$ reduces to the linearization matrix (also denoted by $\bH$) used in \citet{han2016tracy} and \citet{han2018unified}. In other words, our linearization matrix differs from that in \citet{han2016tracy} and \citet{han2018unified} only in a single deterministic block (the middle block). Owing to this structural similarity, the arguments developed in \citet{han2018unified} can be reused with only minor modifications, which substantially simplifies our proof.

The explicit expression of $\bH^{-1}(z)$ can be obtained using the standard block matrix inversion formula
\[
\begin{pmatrix}
\mathbf{K} & \mathbf{B} \\
\mathbf{C} & \mathbf{D}
\end{pmatrix}^{-1}
=
\begin{pmatrix}
0 & 0 \\
0 & \mathbf{D}^{-1}
\end{pmatrix}
+
\begin{pmatrix}
\mathbf{I} \\
-\mathbf{D}^{-1}\mathbf{C}
\end{pmatrix}
\left(\mathbf{K}-\mathbf{B}\mathbf{D}^{-1}\mathbf{C}\right)^{-1}
\begin{pmatrix}
\mathbf{I} ~&~\mathbf{B}\mathbf{D}^{-1}
\end{pmatrix}.
\]

As an illustration, the three diagonal blocks of $\bH_\lambda^{-1}$ are given by
\begin{align*}
(\bH^{-1}_\lambda)_{11} &= \bL(z) = (U_1^{T}\bZ^{T}\bG_\lambda^{-1}\bZ U_1 - z I_{n_1})^{-1},\\[4pt]
(\bH^{-1}_\lambda)_{22} &= (-\bG_{\lambda} + z^{-1}\bZ U_1 U_1^{T}\bZ^{T})^{-1},\\[4pt]
(\bH^{-1}_\lambda)_{33} &= I_{n_2} + U_2^{T}\bZ^{T}(\bH^{-1}_\lambda)_{22}\bZ U_2.
\end{align*}
We omit the explicit expressions for the off-diagonal blocks.

Importantly, when $\bZ U_2$ is regarded as fixed, all blocks of $\bH^{-1}$ can be expressed as linear combinations of the block matrices appearing in equation~(4.3) of \citet{knowles2017anisotropic}, provided that the matrix
\[
\bG^{-1}_\lambda = (\bZ U_2 U_2^{T}\bZ^{T} + \lambda \Sigma_p^{-1})^{-1}
\]
is taken as the population covariance matrix ``$\Sigma$'' in the formulation of \citet{knowles2017anisotropic}. For example, among the diagonal blocks, $(\bH^{-1}_\lambda)_{11}$ corresponds to the last block in (4.3) of \citet{knowles2017anisotropic} (denoted there by $G_n$), while $(\bH^{-1}_\lambda)_{22}$ corresponds to the first block in (4.3) (denoted there by $G_M$). Finally, $(\bH^{-1}_\lambda)_{33}$ can be written as $I_{n_2} + U_2^{T}\bZ^{T}G_M \bZ U_2$.

While the formal proof of the local laws for $\bH^{-1}_\lambda$ is presented in the next section, we briefly describe the conjectured limiting behavior of the blocks of $\bH^{-1}_\lambda$. Assume that the observations are Gaussian, so that $\bZ U_1$ and $\bZ U_2$ are independent. Conditional on $\bZ U_2$, we may therefore apply the deterministic equivalents of $G_M$ and $G_N$ established in \citet{knowles2017anisotropic}. In particular, based on the form of the deterministic equivalents given in equation~(3.5) of \citet{knowles2017anisotropic}, we obtain
\begin{align*}
(\bH^{-1}_\lambda)_{11} &\approx \hat{q}_{\lambda}(z)\, I_{n_1},\\[4pt]
(\bH^{-1}_\lambda)_{22} &\approx -\, \bG^{-1}_{\lambda} \bigl(1+ \hat{q}_{\lambda}(z)\bG_\lambda^{-1}\bigr)^{-1}
= -\, (\bG_\lambda + \hat{q}_{\lambda}(z) I_p)^{-1},\\[4pt]
(\bH^{-1}_\lambda)_{33} &\approx I_{n_2} - U_2^{T} \bZ^{T} (\bG_\lambda + \hat{q}_{\lambda}(z) I_p)^{-1} \bZ U_2,\\[4pt]
(\bH^{-1}_\lambda)_{ij} & \approx 0 ,\quad \text{if } i\neq j. 
\end{align*}
Here, $\hat{q}_{\lambda}(z)$ is defined in \eqref{eq:def_hat_q}. 

Moreover, $(\bG_\lambda + \hat{q}_{\lambda}(z) I_p)^{-1}$ is precisely the resolvent of $\bG_\lambda$ evaluated at $-\hat{q}_{\lambda}(z)$. Equivalently, it is the first block $\bJ_{(11)}$ of $\bJ$ evaluated in $-\hat{q}_\lambda(z)$  (see Eq.~\eqref{eq:bJ_bK} and Lemma~\ref{lemma:resolvent_identity}). In addition, the matrix
\[
I_{n_2} - U_2^{T} \bZ^{T} (\bG_\lambda + \hat{q}_{\lambda}(z) I_p)^{-1} \bZ U_2
\]
can be written as $-\bJ_{(22)}(-\hat{q}_{\lambda}(z))$, where $\bJ_{(22)}$ denotes the lower-right block of the matrix $\bJ$. Combining the deterministic equivalent $\Omega(z)$ of $\bJ(z)$ established in Section~\ref{sec:properties_of_bg_lambda} with the concentration of $\hat{q}_{\lambda}(z)$ to $q_\lambda(z)$, we obtain
\begin{align*}
(\bH^{-1}_\lambda)_{22} 
&\approx -\, (\bG_\lambda + \hat{q}_{\lambda}(z) I_p)^{-1} 
\approx \big( \lambda \Sigma_p^{-1} - h_\lambda(-q_\lambda(z)) I_p \big)^{-1},\\[6pt]
(\bH^{-1}_\lambda)_{33} 
&\approx I_{n_2} - U_2^{T} \bZ^{T} (\bG_\lambda + \hat{q}_{\lambda}(z) I_p)^{-1} \bZ U_2 
\approx \big(h_\lambda(-q_\lambda(z)) - z \big) I_{n_2}.
\end{align*}

We define the following two block diagonal matrices
\begin{align*}
    \widehat{\Pi}_\lambda(z) =\widehat{\Pi} (z, \bZ, \lambda) = \left(
    \begin{matrix}
        \hat{q}_{\lambda}(z) I_{n_1} & 0 & 0\\
        0 & (\bG_\lambda + \hat{q}_{\lambda}(z) I_p)^{-1} & 0\\
        0 & 0 & I_{n_2} - U_2^{T} \bZ^{T} (\bG_\lambda + \hat{q}_{\lambda}(z) I_p)^{-1} \bZ U_2
    \end{matrix}
    \right).
\end{align*}

\begin{align*}
    \Pi_\lambda(z) = \Pi(z, \lambda)  =\begin{pmatrix}
        q(z) I_{n_1} & 0 & 0 \\
        0 & \big( \lambda \Sigma_p^{-1} - h_\lambda(-q_\lambda(z)) I_p \big)^{-1} & 0\\
        0 & & \big(h_\lambda(-q_\lambda(z)) - z \big) I_{n_2}
    \end{pmatrix}.
\end{align*}

It is worth noting that when $\lambda=0$, the matrix $\widehat{\Pi}_0(z)$ coincides with the ``limit'' matrix $\Pi(z)$ in \citet{han2018unified} (see equation~(8.5) therein). In particular, for $\lambda=0$, \citet{han2016tracy} and \citet{han2018unified} established that
\[
\bH^{-1}_0(z) - \widehat{\Pi}_0(z) 
\prec 
\sqrt{\frac{\Im q_0(z)}{n_1\eta}} + \frac{1}{n_1\eta},
\]
whenever $z$ lies in a neighborhood of $\Theta_1(0)$. In the next section, we establish the corresponding concentration result for $\lambda>0$, where $\widehat{\Pi}_\lambda(z)$ is replaced by the deterministic equivalent matrix $\Pi_\lambda(z)$.

\subsection{Step 4: Local Laws for the Linearization Matrix}\label{subsec:step4}

In this section, we establish a local law for the behavior of $\bH^{-1}_\lambda(z)$ when $z$ is close to $\Theta_1$. As before, we frequently suppress the explicit dependence of many quantities on $\lambda$ and $p$ whenever no ambiguity arises.


Recall the domain
\[
D = D(a,a',n_1) = \Big\{ z = E+\mathrm{i}\eta \in \mathbb{C}^+ : |E-\Theta_1| \leq a,\; n_1^{-1+a'} \leq \eta \leq 1/a' \Big\}
\]
from Section~\ref{subsec:additional_properties_of_the_generalized_mar_v_c}. We define the control parameter 
\[ \Phi(z)  =  \Phi(z,\lambda) =  \sqrt{\frac{\Im q_\lambda(z) }{n_1\eta} } + \frac{1}{n_1\eta}.\]



Again, we proceed by first assume Gaussianity. The following theorem follows from a combination of Theorem 3.14 in \cite{knowles2017anisotropic} and \ref{thm:convergence_hatbeta_hatTheta} in Section \ref{subsec:step2}.
\begin{theorem}
    \label{thm:local_laws_bH}
    Suppose that Conditions \ref{enum:high_dimensional_regime}--\ref{enum:regular_edge} hold. Additionally, assume that the observations are Gaussian.  Fix any $\lambda>0$. 
For any $a' >0$,  there exist constants $a>0$ such that the following results hold. 
\begin{itemize}
    \item[(i)] (Entrywise local law)  We have
\[ \bH^{-1}_{\lambda}(\iz) - \widehat{\Pi}_{\lambda}(z) = O_\prec(\Phi(\iz)),\]
uniformly in ${D}(a, a', n_1)$. 
\item[(ii)] (Averaged local law) We have 
\[| \overline{L}(\iz) - {q}_\lambda(\iz)| \prec \frac{1}{n_1 \eta},\]
uniformly in $D(a,a',n_1)$. Here $\overline{L}(z)$ is the averaged trace of $\bL_{\lambda}(z) = (\bH^{-1}_\lambda(z) )_{11}$. 

\item[(iii)] (Local law for $\widehat{\Pi}_\lambda(z)$) Moreover,
\[  \widehat{\Pi}_{\lambda}(z) - \Pi_\lambda(z) = O_\prec(\frac{1}{n_2\eta}),\]
uniformly in ${D}(a, a', n_1)$. It immediately follows that part (i) still holds if $\widehat{\Pi}_\lambda(z)$ is replaced by $\Pi_\lambda(z)$.
\end{itemize}
\end{theorem}

Before proceeding to the proof of Theorem~\ref{thm:local_laws_bH}, as a connection to the literature, \cite{han2016tracy} and \cite{han2018unified} established part (i) and (ii) of Theorem \ref{thm:local_laws_bH} when $\lambda =0$ under both Gaussian and Condition \ref{enum:moments_conditions}.  

\begin{proof}[Proof of Theorem~\ref{thm:local_laws_bH}]
By part~(i) of Theorem~\ref{thm:convergence_hatbeta_hatTheta}, there exists a small constant $c>0$ such that $\mathbb{I}(|\hat{\beta}-g_p|>c)=1$ with high probability. Consequently, when $\bG_{\lambda}^{-1}$ is regarded as the ``population covariance matrix'', the regularity condition in Definition~2.7 of \citet{knowles2017anisotropic} is satisfied with high probability.

As explained in Section~\ref{subsec:step3}, under this regularity condition, each block of $\bH^{-1}$ can be expressed as a linear combination of the blocks of the linearization matrix $G$ studied in \citet{knowles2017anisotropic} (see equation~(4.3) therein). Applying the entrywise local laws from Theorem~3.14 of \citet{knowles2017anisotropic}, we obtain that for any $a'>0$ there exists $2a>0$ such that
\[
\mathbb{I}(|\hat{\beta}-g_p|>c)\,
\big|\bH^{-1}(z)-\widehat{\Pi}(z)\big|
\,\Big|\, \bZ U_2
\prec 
\sqrt{\frac{\Im \hat{q}(z)}{n_1\eta}}
+\frac{1}{n_1\eta},
\]
uniformly for $z\in \hat{D}(2a,a',n_1)$. Here, $\hat{D}(2a,a',n_1)$ denotes the $\bZ U_2$–dependent counterpart of $D(2a,a',n_1)$ obtained by replacing $\Theta_1$ with $\hat{\Theta}_1$. The notation $\xi\,|\,\bZ U_2 \prec \zeta$ means that $\xi$ is stochastically dominated by $\zeta$ (Definition~\ref{def:stochastic_domination}) with respect to the conditional distribution of $\bZ$ given $\bZ U_2$.

Next, by part~(ii) of Theorem~\ref{thm:convergence_hatbeta_hatTheta}, for any $a>0$ we have $D(a,a',n_1)\subset \hat{D}(2a,a',n_1)$ with high probability. Moreover, part~(iii) of Theorem~\ref{thm:convergence_hatbeta_hatTheta} implies that $|\hat{q}(z)-q(z)|\prec 1/(n_1\eta)$. Combining these results, 
\[
\big|\bH^{-1}(z)-\widehat{\Pi}(z)\big|
\,\Big|\, \bZ U_2
\prec 
\sqrt{\frac{\Im q(z)}{n_1\eta}}
+\frac{1}{n_1\eta},
\]
uniformly for $z\in D(a,a',n_1)$. Finally, integrating with respect to the distribution of $\bZ U_2$, we conclude that
\[
\big|\bH^{-1}(z)-\widehat{\Pi}(z)\big|
\prec 
\sqrt{\frac{\Im q(z)}{n_1\eta}}
+\frac{1}{n_1\eta},
\]
uniformly in $D(a,a',n_1)$.

The averaged local law in part (ii) follows directly from the averaged version of Theorem~3.14 of \citet{knowles2017anisotropic} together with analogous arguments.

To show part (iii), we use Theorem \ref{thm:strong_local_law}. We only need to consider the second and third diagonal blocks, since the first block is such that $|\hat{q}(z) - q(z)|\prec 1/(n_2\eta)$ (note that $n_1\asymp  n_2$).  As explained in Section \ref{subsec:step3},  these two blocks in $\widehat{\Pi}(z)$ are linear combinations of the blocks in $\bJ(-\hat{q}(z))$.  Accordingly, the two diagonal blocks in $\Pi(z)$ are the same linear combinations of the blocks in $\Omega(-q(z))$. Given Theorem \ref{thm:strong_local_law}, to show the concentration of $\widehat{\Pi}(z)$ around $\Pi(z)$, it suffices to verify that $-\hat{q}(z)$ is away from the support of $\calG_{p\lambda}$ with high probability. By part (2) of Lemma \ref{lemma:properties_q_z}, there exits a constant $c'>0$ such that 
\[\inf_{\tau \in \supp(\calG_{p\lambda})} |\tau + q(z)|\geq c'\] 
for all $z\in D(a,a',n_1)$. Therefore, $\inf_{\tau \in \supp(\calG_{p\lambda})} | \tau + \hat{q}(z) |\geq c'$ with high probability.  It completes the proof.
\end{proof}

The last step in the proof of Theorem \ref{thm:main} is to extend the results in Theorem \ref{thm:local_laws_bH} to non-Gaussianity. We aim to establish the following results.
\begin{theorem}
    \label{thm:local_laws_bH_nonGauss}
    Suppose that Conditions \ref{enum:high_dimensional_regime}--\ref{enum:regular_edge} hold. Fix any $\lambda>0$. For any $a' >0$,  there exist constants $a>0$ such that the following results hold. 
\begin{itemize}
    \item[(i)] (Entrywise local law)  We have
\[ \bH^{-1}_{\lambda}(\iz) - \Pi_{\lambda}(z) = O_\prec(\Phi(\iz)),\]
uniformly in ${D}(a, a', n_1)$. 
\item[(ii)] (Averaged local law) We have 
\[| \overline{L}(\iz) - {q}_\lambda(\iz)| \prec \frac{1}{n_1 \eta},\]
uniformly in $D(a,a',n_1)$. Here $\overline{L}(z)$ is the averaged trace of $\bL_{\lambda}(z) = (\bH^{-1}_\lambda(z) )_{11}$. 
\end{itemize}
\end{theorem}

Once Theorem~\ref{thm:local_laws_bH_nonGauss} is established, it is standard in RMT to deduce Lemma~\ref{lemma:rough_order_largest_root} (edge rigidity) and Theorem~\ref{thm:green_function_comparison} (the Green function comparison theorem). Consequently, edge universality for $\tilde{\bF}_{\lambda}$ follows, completing the proof of Theorem~\ref{thm:main}. In particular, given Theorem~\ref{thm:local_laws_bH}, the proof of Lemma~\ref{lemma:rough_order_largest_root} is very similar to that of Theorem~\ref{thm:rigidity_G}. Such arguments are standard in the literature; see, for example, Lemma~4 of \citet{han2016tracy} and Lemma~1 of \citet{han2018unified}. We therefore omit the details.

Given Theorem~\ref{thm:local_laws_bH}, Theorem \ref{thm:local_laws_bH_nonGauss}, and Lemma~\ref{lemma:rough_order_largest_root}, it is likewise routine in RMT to derive the bound required for the Green function comparison theorem. Similar arguments can be found in Section~6 of \citet{erdHos2012rigidity} and Section~12 of \citet{han2018unified}. In particular, owing to the close similarity between our linearization matrix $\bH(z)$ and that in \citet{han2018unified}, the arguments in Section~12 of \citet{han2018unified} can be applied essentially verbatim. We therefore omit the details.

\begin{proof}[Proof of Theorem \ref{thm:local_laws_bH_nonGauss}]

The proof of Theorem~\ref{thm:local_laws_bH_nonGauss} is technically considerably more demanding. Nevertheless, closely related analyses have been carried out in Sections~9.1, 9.2, and~10 of \citet{han2016tracy} and \citet{han2018unified}, by adapting the arguments originally developed in \citet{knowles2017anisotropic}. In particular, we have already employed these arguments in Section~\ref{subsec:local_law_G_nonGaussian} to establish the local laws for $\bG_\lambda$ under non-Gaussianity.

Recall that our linearization matrix $\bH_\lambda$ coincides with that of \citet{han2018unified} in the special case $\lambda=0$. Owing to this structural similarity, most of the arguments in \citet{han2018unified} can be directly adapted to our setting. We therefore focus only on describing the similarities and differences, and on identifying the modifications required to tailor those arguments to the present context. With these modifications, Theorem \ref{thm:local_laws_bH_nonGauss} is proved using arguments in Sections 9.1, 9.2 and 10 of \cite{han2018unified} in a verbatim manner.

For any deterministic unit vector $v$ and $w$, any observation matrix $\bZ$ satisfying \ref{enum:moments_conditions}, and $z\in D$, define the inner product function 
\[ F_\lambda (\bZ) = F_\lambda(\bZ, z, v, w)  = |v^T (\bH^{-1}_\lambda (\bZ, z) - \Pi_\lambda(\bZ, z))w|. \]
To show Theorem \ref{thm:local_laws_bH_nonGauss}, it suffices to show that $F_{\lambda}(\bZ) \prec \Phi$ for $z\in D$. While this argument is established in Theorem \ref{thm:local_laws_bH} under Gaussian observations, say $\bZ^0$, we only need to control the difference 
\[ \mE F_\lambda(\bZ) - \mE F_\lambda ( \bZ^0). \]

To this end, \citet{knowles2017anisotropic} employ an interpolation method in which a family of distributions $\{f_\theta : \theta \in [0,1]\}$ is introduced such that $f_0$ is $N(0,1)$ and $f_1$ is the distribution of the entries of $\bZ$ (see Definition~7.6 of \citet{knowles2017anisotropic} for details). Let $\bZ^\theta$ denote the matrix whose entries follow the distribution $f_\theta$. 

To control $\mE F_{\lambda}(\bZ^1) - \mE F_{\lambda}(\bZ^0)$, one analyzes the change in $\bH^{-1}_{\lambda}(\bZ^\theta)$ induced by replacing a single entry of $\bZ^{\theta}$. Specifically, let $\bZ_{(uv)}^{\theta,s}$ denote the matrix obtained from $\bZ^{\theta}$ by replacing its $(u,v)$-th entry with the value $s$. It turns out that, to control $\mE F_{\lambda}(\bZ^1) - \mE F_{\lambda}(\bZ^0)$, we only need to control 
\[ \sum_{uv} \mE F_{\lambda}^k (\bZ_{(uv)}^{\theta, z_{uv}^{1}}) - \mE F_{\lambda}^k (\bZ_{(uv)}^{\theta, z_{uv}^{0}}),\]
for arbitrary power $k\in2\mathbb{N}$, $\theta\in[0,1]$, and $(u,v)$. See Lemmas 7.9 and 7.10 of \cite{knowles2017anisotropic} for details. The analysis relies on explicitly quantifying the difference in $\bH^{-1}_{\lambda}(\bZ^{\theta})$ induced by the change  of the $(u,v)$-th entry of $\bZ^{\theta}$. Since $A^{-1} - B^{-1} = - A^{-1} (A-B) B^{-1}$, we have 
\[ \bH^{-1}_{\lambda}(\bZ_{(uv)}^{\theta, z_{uv}^{1}}) - \bH^{-1}_{\lambda}(\bZ_{(uv)}^{\theta, z_{uv}^{0}}) =  - \bH^{-1}_{\lambda}(\bZ_{(uv)}^{\theta, z_{uv}^{1}}) [\bH_\lambda(\bZ_{(uv)}^{\theta, z_{uv}^{1}})  - \bH_\lambda(\bZ_{(uv)}^{\theta, z_{uv}^{0}}) ] \bH^{-1}_\lambda(\bZ_{(uv)}^{\theta, z_{uv}^{0}}).\]
\begin{equation}\label{eq:diff_bH_entry_change}
\bH_{\lambda}(\bZ_{(uv)}^{\theta, z_{uv}^{1}})  - \bH_{\lambda}(\bZ_{(uv)}^{\theta, z_{uv}^{0}})  = -(z_{uv}^{1} - z_{uv}^{0}) [ e_{(n_1+ p+n_2)u} e_{n_0v}^T \Delta^T +\Delta e_{n_0v} e_{(n_1+ p+n_2)u}^T].
\end{equation}
Here, $e_{pi}$ is the standard unit vector of dimension $p$ in the coordinate direction $i$ and 
\[\Delta = \begin{pmatrix}
    U_1^T \\[4pt]
    0_{p\times n_0} \\[4pt]
    U_2^T
\end{pmatrix}.\] 
Arguments in Sections 9.1, 9.2 and 10 of \cite{han2018unified} are devoted to the study of the above quantities when $\lambda=0$. 

The similarities between our setting and that of \citet{han2018unified} can be summarized as follows. First, the bounds on the norms $\|\bH_\lambda(z)\|$, $\|\bH^{-1}_\lambda(z)\|$, and $\|\partial_z\bH_\lambda^{-1}(z)\|$ by the imaginary part of $z$ are unaffected by the presence of the parameter $\lambda$. In particular, the estimates in (9.16)--(9.19) of \citet{han2018unified} remain valid for all $\lambda>0$. Second, bounds of the form $|v^{T}(\bH^{-1}_0(\bZ^\theta)-\widehat{\Pi}_0(z))w|$, such as those established in Lemma~8 of \citet{han2018unified}, continue to hold for $|v^{T}(\bH^{-1}_\lambda(\bZ^\theta)-\Pi_\lambda(z))w|$. Third, the analysis of the perturbation induced by modifying a single entry of $\bZ^\theta$ can be carried out in exactly the same manner as in \citet{han2018unified}. In particular, equation~\eqref{eq:diff_bH_entry_change} does not depend on $\lambda$, since in the difference $\bH_{\lambda}(\bZ_{(uv)}^{\theta,z_{uv}^{1}}) - \bH_{\lambda}(\bZ_{(uv)}^{\theta,z_{uv}^{0}})$ the middle block cancels.

The main difference between our setting and that of \citet{han2018unified} lies in the treatment of the deterministic equivalent matrix. In our context, concentration is established with respect to the deterministic equivalent $\Pi_\lambda(z)$, whereas \citet{han2018unified} employ a $\bZ U_2$–dependent equivalent $\widehat{\Pi}_0(z)$. This distinction leads to a simpler analysis in our case. Indeed, when individual entries of $\bZ$ are modified, the matrix $\widehat{\Pi}_0(z)$ changes accordingly, while $\Pi_\lambda(z)$ remains unaffected. Consequently, all arguments in \citet{han2018unified} that pertain to the analysis of $\widehat{\Pi}_0(z)$ can be omitted in our setting. These include, for example, the definition of the groups $g^{(j)}$ in Section~9.1.2, the second term on the right-hand side of equation~(9.42), the quantities $\mathcal{H}_{sti}$ and $\mathcal{H}_{stu}$ appearing between equations~(9.48) and~(9.49), equation~(9.53), as well as the discussion from the fourth paragraph of Section~9.2 up to equation~(9.60), among others.

With these modifications in place, the remaining arguments of \citet{han2018unified} can be applied in a verbatim manner. We therefore omit further details.
\end{proof}

\newpage
\clearpage
\section{Proof of Theorem \ref{thm:consistency_estimators}, Lemma \ref{lemma:consistency_condition1}, and Lemma \ref{lemma:consistency_condition3}}\label{sec:proof_of_theorem_ref_thm_consistency_estimators}

First, consider the proof of Theorem \ref{thm:consistency_estimators}. Let $F^{\Sigma_p}_{(p)}$ be the discretization of $F^{\Sigma_p}$ onto the points in $R_p$. Specifically, we assume that $F^{\Sigma_p}_{(p)}$ is the distribution that places masses $F^{\Sigma_p}(( \sigma_{p,j-1}, \sigma_{pj}])$ at the point $\sigma_{pj}$, $j =1,2,\dots$, and $\sigma_{p0} = -\infty$. Since  $\mbox{size}(R_p) = o(p^{-2/3})$, clearly 
\[ \mathcal{D}_W(F^{\Sigma_p}_{(p)}, F^{\Sigma_p}) = o(p^{-2/3}).\]

Recall the errors in Algorithm \ref{algo:linear_program} as 
\[e_{ij} =  \frac{\hat{Q}_j(z_i)}{|\hat{Q}_j(z_i)|} -\frac{1}{|\hat{Q}_j(z_i)|} \sum_{k=1}^K \frac{\sigma_k^j w_k}{(\sigma_k \lambda \hat{\varphi}(z_i) +\lambda)^j}.\]
By the convergence of $\hat{\varphi}(z)$, uniformly for $y \in J_p$, $\hat{\varphi}(y) - \varphi(y) = o_{P}(n^{-2/3})$ and $\hat{\varphi}'(y) - \varphi'(y) = o_{P}(n^{-2/3})$ . Therefore, together with Mar\v{c}enko-Pastur equation \eqref{eq:MP_W2}, 
\[ \sup_{y \in J_p }\left|\frac{\hat{Q}(y)}{|\hat{Q}(y)|} -\frac{1}{|\hat{Q}(y)|}  \int \frac{\tau dF^{\Sigma_p}(\tau) }{\tau \lambda \hat{\varphi}(y) +\lambda}    \right|  = o_{P}(p^{-2/3}). \] 
It follows then 
\[ \sup_{y \in J_p }\left|\frac{\hat{Q}(y)}{|\hat{Q}(y)|} -\frac{1}{|\hat{Q}(y)|}  \int \frac{\tau dF^{\Sigma_p}_{(p)}(\tau) }{\tau \lambda \hat{\varphi}(y) +\lambda}    \right|  = o_{P}(p^{-2/3}). \] 
Then, we can conclude that the estimated measure $\hat{F}_p$ is such that
\[  \sup_{y\in J_p} \left|\int \frac{\tau d\hat{F}_p(\tau) }{\tau {\varphi}(y) + 1}  -   \int \frac{\tau dF^{\Sigma_p}(\tau) }{\tau {\varphi}(y) + 1} \right|= o_{P}(p^{-2/3}).\]
Under the assumption $\{\varphi(y), y\in J_p \}$ is dense in $D(\nu)$. Therefore,  
\begin{equation}\label{eq:proof_theorem_consistency_eq1}
\sup_{z \in D(\nu)} \left|\int \frac{\tau d\hat{F}_p(\tau) }{\tau z + 1}  -   \int \frac{\tau dF^{\Sigma_p}(\tau) }{\tau z + 1} \right| = o_{P}(p^{-2/3}). 
\end{equation}
Similarly, we can get
\begin{equation}\label{eq:proof_theorem_consistency_eq2}
\sup_{z \in D(\nu)} \left|\int \frac{\tau^2 d\hat{F}_p(\tau) }{(\tau z + 1)^2}  -   \int \frac{\tau^2 dF^{\Sigma_p}(\tau) }{(\tau z + 1)^2} \right| = o_{P}(p^{-2/3}). 
\end{equation}

Notice that the functions are analytic in $z$. Consider the power series expansion 
\begin{align*}
\int \frac{\tau d\hat{F}_p(\tau) }{\tau z + 1}  -   \int \frac{\tau dF^{\Sigma_p}(\tau) }{\tau z + 1} & = \sum_{m=0}^\infty (-1)^m  z^m \left(\int \tau^{m+1} (d\hat{F}_p(\tau) - dF^{\Sigma_p}(\tau)) \right) \\ &=\sum_{m=0}^\infty (-1)^m  z^m (\mu_{m+1}(\hat{F}_p) - \mu_{m+1} (F^{\Sigma_p})).
\end{align*}
Here, $\mu_{m}(\cdot)$ is the $m$th moment of a distribution.  Using Cauchy's estimates on the power series coefficients, we have 
\[ \delta_p =  \sup_{m\geq 1} p^{2/3} \nu^m |\mu_{m}(\hat{F}_p) - \mu_{m} (F^{\Sigma_p})| = o_P(1).\]
Then, it follows that the characteristic functions of $\hat{F}_p$ and $F^{\Sigma_p}$ are such that for $t\in \mathbb{R}$, 
\[\left| \frac{\mE_{\hat{F}_p} e^{i tX} -  \mE_{F^{\Sigma_p}} e^{i tX}}{t} \right| \leq C \nu^{-1}\delta_p p^{-2/3} e^{t/\nu}, \]
\[\left|\frac{d}{dt}  \frac{\mE_{\hat{F}_p} e^{i tX} -  \mE_{F^{\Sigma_p}} e^{i tX}}{t}  \right| \leq C\nu^{-2} \delta_p p^{-2/3} e^{t/\nu}. \]
Following from Corollary 8.3 of \cite{bobkov2016proximity}, for all $T\geq 3$
\begin{align*}
\calD_W(\hat{F}_p, F^{\Sigma_p}) \leq &\left(\int_{-T}^T  \left| \frac{\mE_{\hat{F}_p} e^{i tX} -  \mE_{F^{\Sigma_p}} e^{i tX}}{t} \right|^2 dt\right)^{1/2} + \left(\left|\frac{d}{dt}  \frac{\mE_{\hat{F}_p} e^{i tX} -  \mE_{F^{\Sigma_p}} e^{i tX}}{t}  \right|^2 dt\right)^{1/2} \\
&+ \frac{C}{T}.
\end{align*}
It follows then 
\[\calD_W(\hat{F}_p, F^{\Sigma_p}) \leq \frac{C(1+\nu)}{\nu^2} p^{-2/3} \exp(T)\delta_p + \frac{C}{T}.\]
Note that setting $T= \log (p^{1/3})$, we have 
\[\calD_W(\hat{F}_p, F^{\Sigma_p}) = o_P(1). \]

Next, we consider the scenario when $J_p$ is such that $\{\varphi(y): ~ y\in J_p\}$ eventually covers $[-h_\beta/\lambda - \varepsilon, -h_\beta/\lambda + \varepsilon]$.  Again, due to the fact that uniformly for  $y\in J_p$,  $\hat{\varphi}(y_p) -\varphi(y_p) = o_{P}(p^{-2/3})$, $\hat{\varphi}'(y_p) -\varphi'(y_p) = o_{P}(p^{-2/3})$,  Eq. \eqref{eq:proof_theorem_consistency_eq1} and Eq. \eqref{eq:proof_theorem_consistency_eq2}, 
\[\sup_{h \in [h_\beta -\lambda \varepsilon, h_\beta + \lambda \varepsilon]  } \left|\int \frac{\tau d\hat{F}_p(\tau)}{-h\tau  + \lambda}  - \int \frac{\tau d{F}^{\Sigma_p}(\tau)}{-h\tau  + \lambda}  \right|  = o_P(p^{-2/3}).\]

\[\sup_{h \in [h_\beta -\lambda \varepsilon, h_\beta + \lambda \varepsilon] } \left|\int \frac{\tau^2 d\hat{F}_p(\tau)}{(-h\tau  + \lambda)^2}  - \int \frac{\tau^2 d{F}^{\Sigma_p}(\tau)}{(-h\tau  + \lambda)^2}  \right|  = o_P(p^{-2/3}).\]
It indicates that uniformly in  $h \in [h_\beta -\lambda \varepsilon, h_\beta + \lambda \varepsilon]$, 
\[ \hat{\calH}_1(h) - \calH_1(h)  = o_P(p^{-2/3}), \quad   \hat{\calH}_2(h) - \calH_2(h)  = o_P(p^{-2/3}) \quad \mbox{and} \quad \hat{\calH}_3(h) - \calH_3(h) = o_P(1). \]
Here, the bound on $\hat\calH_3(h)-  \calH_3(h)$ is due to $\calD_W(\hat{F}_p, F^{\Sigma_p})$. As $x$,  $s(x)$, $s'(x)$, and $s''(x)$ are smooth in $\calH_j(h)$, $j=1,2,3$ as shown in Theorem \ref{thm:determine_s}, we get 
\[\hat{\beta} - \beta = o_P(p^{-2/3}), \quad \hat{s}(\hat\beta) - s(\beta) = o_P(p^{-2/3}),\] 
\[\hat{s}'(\hat\beta) - s'(\beta) = o_P(p^{-2/3}), \quad  \hat{s}''(\hat\beta) - s''(\beta) = o_P(1).\]
Consequently, $\hat{\Theta}_1 -\Theta_1 =  o_P(p^{-2/3})$ and $\hat{\Theta}_2 - \Theta_2 = o_P(1)$.

As for Lemma \ref{lemma:consistency_condition1}, recall that $\varphi(z)$ is the Stieltjes transform of $\calW$. Then,
\[ \varphi(z) = \int\frac{d\calW }{\tau - z}.\]
It follows that $\{\varphi(z): z \in \mathbb{C}^+ ~\mbox{and}~   \operatorname{dist}(z, \operatorname{supp}(\calW)) \geq c \}$ contains the half disk $\calD(\nu)$ for some $\nu >0$. We can then select $J_p$ such that $\varphi(y)$ is dense in $\calD(\nu)$ when $y\in J_p$.

As for Lemma \ref{lemma:consistency_condition3}, as $\hat{\gamma}_1 \to \infty$, $\beta \to 0$ and $h_\beta/\lambda \to -\varphi(-\lambda)$.  Therefore, select $J_p$ to be a grid containing a subsequence dense in $[-3/2\lambda, -\lambda/2]$. Then, $\{\varphi(y): y\in J_p\}$ will eventually covers a neighborhood of $-h_\beta/\lambda$. 


\section{Proof of Lemma \ref{lemma:alternatives_deterministic}, Corollary \ref{corollary:power_consistency_deterministic}, Lemma \ref{lemma:alternatives_prob} and Corollary \ref{corollary:power_consistency_prob}}
\label{sec:proof_power}

First of all, under $H_a: BC \neq 0$, we can decompose $\bF_\lambda$ as
\begin{align*} 
\bF_\lambda &=  \bF_\lambda^{(0)}  + \frac{1}{n_1} BC (C^T(XX^T)^{-1}C)^{-1} C^T B^T (\bW_2+\lambda I_p)^{-1} \\
&+ \frac{1}{n_1} \left(BX P_1 \bZ^T \Sigma^{1/2}_p + \Sigma_p^{1/2} \bZ P_1X^T B^T\right)(\bW_2+\lambda I_p)^{-1},
\end{align*} 
where $\bF_{\lambda}^{(0)}$ is the regularized $F$-matrix under the null hypothesis $H_0: BC = 0$. 

Therefore, under \textbf{DA}, 
\begin{align*}
\frac{1}{p^s} \bF_{\lambda} =& \frac{1}{p^s} \bF_{\lambda}^{(0)} + d_pq_pq_p^T(\bW_2+\lambda I_p)^{-1}  \\
&+   \frac{1}{n_1p^s} \left(BX P_1 \bZ^T \Sigma^{1/2}_p + \Sigma_p^{1/2} \bZ P_1X^T B^T\right)(\bW_2+\lambda I_p)^{-1}.
\end{align*}
Since $\ell_{\max} (\bF_{\lambda}^{(0)})  = O_P(1)$ as indicated by Theorem \ref{thm:main} and  
\[ \frac{1}{p^{s/2} n_1} \| BX P_1\|_2 \| \bZ^T\| \| \Sigma_p^{1/2}\| \|(\bW_2+ \lambda I_p)^{-1} \| = O_P(1), \]
we have 
\[ \frac{1}{p^s} \ell_{\max} (\bF_{\lambda}) = d_p q_p^T (\bW_2+\lambda I_p)^{-1}  q_p  + o_P(1). \]

Using the main results of \citet{karoui2011geometric}, we can show that 
\[  q_p^T (\bW_2+\lambda I_p)^{-1}  q_p = q_p^T(\lambda \varphi(-\lambda)\Sigma_p + \lambda I_p)^{-1}q_p + o_P(1).\]
The result of Lemma \ref{lemma:alternatives_deterministic} follows.  

As for Corollary \ref{corollary:power_consistency_deterministic}, under the assumed conditions, $\tilde{\ell}(\lambda) \to \infty$. Therefore, 
\[ \mP\Big( \tilde{\ell}(\lambda) > \TW_1(1-\alpha)\mid BC  \Big)\longrightarrow 1.\]

Under \textbf{PA}, similarly, we get $p^{-s}\ell_{\max}(\bF_{\lambda})$ is dominated by the following term
\[ \frac{1}{p^s} \ell_{\max} (\bF_{\lambda}) = \frac{1}{p} \nu^T D^{1/2} (\bW_2+\lambda I_p)^{-1} D^{1/2} \nu + o_P(1). \]
Using Lemma 2.7 of \citet{bai1998no}, conditional on $\bW_2$,
\[\frac{1}{p} \nu^T D^{1/2} (\bW_2+\lambda I_p)^{-1} D^{1/2} \nu  = \frac{1}{p} \tr\Big[(\bW_2 + \lambda I_p)^{-1}D\Big] + o_P(1).\]
Here, the residual $o_P(1)$ is with respect to $\mP_{BC}$ and the convergence is uniform with respect to $\bW_2$ since $\|(\bW_2 + \lambda I_p)^{-1}\|_2\leq 1/\lambda$ for any $\bW_2$. 

Again, using the main results of \citet{karoui2011geometric},
\[ \frac{1}{p} \tr\Big[(\bW_2 + \lambda I_p)^{-1}D\Big] = \frac{1}{p}\tr\Big[(\lambda \varphi(-\lambda)\Sigma_p + \lambda I_p)^{-1}D \Big] + o_P(1).\]
The results of Lemma \ref{lemma:alternatives_prob} follow. 

As for Corollary \ref{corollary:power_consistency_prob}, under the assumed conditions, $\tilde{\ell}(\lambda) \to \infty$. Therefore, 
\[ \mP\Big( \tilde{\ell}(\lambda) > \TW_1(1-\alpha)\mid BC  \Big) \stackrel{\mP_{BC}}{\longrightarrow}1.\]

\section{Potential relaxation of technical conditions}\label{sec:relaxtion_conditions}
The technical assumptions imposed in this work may be further relaxed. In this section, we briefly discuss possible directions for such relaxation.

First, in Condition~\ref{enum:high_dimensional_regime}, convergence of \(\hat{\gamma}_1 = p/n_1\) and \(\hat{\gamma}_2 = p/n_2\) to \(\gamma_1\) and \(\gamma_2\) is assumed. This requirement may be relaxed by replacing convergence with mere boundedness of \(\hat{\gamma}_1\) and \(\hat{\gamma}_2\) in the sense that 
\begin{itemize}
    \item[\textbf{Boundedness of aspect ratios.}] There exist constants $c,C>0$ such that $c< \hat{\gamma}_1 < C$ and $c<\hat{\gamma}_2<C$, for all sufficiently large $p$. 
\end{itemize}
Such a formulation has been adopted in the RMT literature; see, for example, \cite{ding2024eigenvector}. In fact, under this weaker assumption, no modifications to the arguments in Sections~\ref{sec:properties_of_bg_lambda}, \ref{sec:proof_theorem_main}, \ref{sec:proof_of_theorem_ref_thm_consistency_estimators}, and~\ref{sec:proof_power} would be required.

Second, on the distributional side, Condition~\ref{enum:moments_conditions} requires the existence of all moments of \(z_{ij}\). For edge universality of \(\bG_{\lambda}\) and \(\tilde{\bF}_\lambda\), this assumption could potentially be weakened to a fourth-moment–type condition. Indeed, \cite{DingYang2018AAP} showed that universality for the sample covariance matrix \(n_0^{-1}\bZ\bZ^T\) holds if and only if
\[
\lim_{t\to\infty} t^{4}\,\mathbb{P}(|z_{ij}|\ge t)=0,
\]
which is slightly weaker than the standard finite fourth-moment condition. In the present setting, we conjecture that universality for
\[
\bG_\lambda = \bZ U_1 U_1^T \bZ^T + \lambda \Sigma_p^{-1}
\]
and for \(\tilde{\bF}_\lambda\) could likewise be established under this tail condition. Proving such an extension would require substantial additional technical work and is therefore left for future research.

Third, Conditions~\ref{enum:wasserstein_convergence}, \ref{enum:edge_stability}, and \ref{enum:regular_edge} may be replaced by the following alternative set of assumptions. Recall that
\[
x_p(h)=h+\left[1+\frac{p}{n_2}\int\frac{\tau\, dF^{\Sigma_p}(\tau)}{\lambda-\tau h}\right]^{-1},
\]
and
\[
f(q)=-\frac{1}{q}+\frac{p}{n_1}\int\frac{d\calG_{p\lambda}(\tau)}{\tau+q}
=-\frac{1}{q}+\hat{\gamma}_1\,\phi_{p\lambda}(-q).
\]

\begin{itemize}
\item[\textbf{Regularity of the population spectrum.}]
There exists a constant \(c>0\) such that, for all sufficiently large \(p\), the following conditions (i) and (ii) hold:
\begin{enumerate}
\item[(i)] One of the following alternatives is satisfied:
\begin{itemize}
\item[(a)] \(|\lambda-\sigma_{1p}h_0|\ge c\), where \(h_0\) is the unique solution of \(x_p'(h_0)=0\) in \((-\infty,\lambda/\sigma_{1p})\). In this case, define \(\rho_{p\lambda}=x_p(h_0)\).
\item[(b)] \((p/n_2)\,F^{\Sigma_p}(\{\sigma_{1p}\})>1+c\). In this case, define \(\rho_{p\lambda}=\lambda/\sigma_{1p}\).
\end{itemize}
\item[(ii)] \(|\beta-\rho_{p\lambda}|\ge c\), where \(\beta\) satisfies \(f'(-\beta)=0\) for \(\beta\in(-\infty,\rho_{p\lambda})\).
\end{enumerate}
\end{itemize}

Conditions of this form are used in studies of edge universality for random matrices in the literature; see, for example, Assumption~2.5 of \cite{DingYang2018AAP}.

\clearpage
\section{Extension to more general linear hypotheses}\label{sec:extension_to_more_general_hypotheses}

In this section, we discuss how the proposed method can be extended from the linear hypothesis $H_0: BC = 0$ to the more general form
\begin{equation}\label{eq:generalized_LH}
H_0: ABC = \Gamma \quad \text{against} \quad H_a: ABC \neq \Gamma,
\end{equation}
where $A$ is a given $p_0 \times p$ matrix of rank $p_0$, $C$ is a given $m \times n_1$ matrix of rank $n_1$, and $\Gamma$ is a given $p_0 \times n_1$ matrix.

The proposed method can be adapted to accommodate this hypothesis. Recall the linear model
\[
\bY = BX + \Sigma_p^{1/2}\bZ .
\]
Define the transformed response matrix
\[
\tilde{\bY} = A\bY .
\]
Then
\[
\tilde{\bY} = \tilde{B}X + \tilde{\Sigma}_p^{1/2}\bZ ,
\]
where $\tilde{B} = AB$ and $\tilde{\Sigma}_p = A\Sigma_p A^T$. Under this transformation, the hypothesis \eqref{eq:generalized_LH} is equivalent to
\[
H_0: \tilde{B}C = \Gamma \quad \text{against} \quad H_a: \tilde{B}C \neq \Gamma .
\]

Define the matrices
\begin{align*}
\bW_1^* &= \frac{1}{n_1}
\left[\tilde{\bY}X^T(XX^T)^{-1}C - \Gamma \right]
\left[C^T(XX^T)^{-1}C\right]^{-1}
\left[C^T(XX^T)^{-1}X\tilde{\bY}^T - \Gamma^T\right], \\
\bW_2^* &= \frac{1}{n_2}
\tilde{\bY}\left[I - X^T(XX^T)^{-1}X\right]\tilde{\bY}^T,
\end{align*}
with $n_1$ being the rank of $C$ and $n_2 = n_0 - m$ being the effective sample size. 
The proposed procedure can then be applied to the regularized $F$-matrix
\[
\bW_1^*(\bW_2^* + \lambda I_{p_0})^{-1}.
\]
Here, it is understood that $p$ and $\Sigma_p$ are replaced by $p_0$ and $\tilde{\Sigma}_p$. The theoretical results continue to hold provided that $(p_0, n_1, n_2, \tilde{\Sigma}_p)$ satisfy Conditions~\ref{enum:high_dimensional_regime}--\ref{enum:regular_edge}.

\end{document}